\documentclass[10pt, reqno]{amsart}
\usepackage{amsthm,amsfonts,amssymb,euscript,cite, mathrsfs}
\usepackage{graphics}
\usepackage{xcolor} 
\usepackage{hyperref}

\usepackage{graphicx}
\usepackage{color}
\usepackage{transparent}

\usepackage{fullpage} 

\usepackage{seqsplit}

\definecolor{green}{rgb}{0,0.8,0} 

\newtheorem{theorem}{Theorem}[section]
\newtheorem{corollary}[theorem]{Corollary}
\newtheorem{lemma}[theorem]{Lemma}
\newtheorem{proposition}[theorem]{Proposition}
\theoremstyle{definition}
\newtheorem{definition}[theorem]{Definition}

\theoremstyle{remark}
\newtheorem{remark}[theorem]{Remark}
\theoremstyle{conjecture}

\numberwithin{equation}{section}
\newcommand{\nrm}[1]{\Vert#1\Vert}
\newcommand{\abs}[1]{\vert#1\vert}
\newcommand{\brk}[1]{\langle#1\rangle}
\newcommand{\set}[1]{\{#1\}}

\newcommand{\ep}{\epsilon}
\def\beaa{\begin{eqnarray*}}
\def\eeaa{\end{eqnarray*}}
\def\bea{\begin{eqnarray}}
\def\eea{\end{eqnarray}}
\def\be{\begin{equation}}
\def\ee{\end{equation}}

\newcommand{\ud}{\mathrm{d}}
\newcommand{\rd}{\partial}
\newcommand{\nb}{\nabla}

\newcommand{\alp}{\alpha}
\newcommand{\bt}{\beta}
\newcommand{\gmm}{\gamma}
\newcommand{\Gmm}{\Gamma}
\newcommand{\dlt}{\delta}
\newcommand{\Dlt}{\Delta}
\newcommand{\eps}{\epsilon}
\newcommand{\veps}{\varepsilon}
\newcommand{\kpp}{\kappa}
\newcommand{\lmb}{\lambda}
\newcommand{\Lmb}{\Lambda}
\newcommand{\sgm}{\sigma}
\newcommand{\Sgm}{\Sigma}

\newcommand{\Tht}{\Theta}
\newcommand{\vtht}{\vartheta}
\newcommand{\omg}{\omega}

\newcommand{\Omg}{\Omega}
\newcommand{\ups}{\upsilon}

\newcommand{\bfe}{{\bf e}}

\newcommand{\bfU}{{\bf U}}
\newcommand{\bfV}{{\bf V}}

\newcommand{\bbR}{\mathbb R}
\newcommand{\bbS}{\mathbb S}

\newcommand{\bbZ}{\mathbb Z}

\newcommand{\calA}{\mathcal A}
\newcommand{\calB}{\mathcal B}
\newcommand{\calC}{\mathcal C}
\newcommand{\calD}{\mathcal D}
\newcommand{\calE}{\mathcal E}

\newcommand{\calH}{\mathcal H}
\newcommand{\calI}{\mathcal I}

\newcommand{\calM}{\mathcal M}
\newcommand{\calN}{\mathcal N}
\newcommand{\calO}{\mathcal O}

\newcommand{\calQ}{\mathcal Q}
\newcommand{\calR}{\mathcal R}

\newcommand{\calX}{\mathcal X}

\newcommand{\PD}{\calQ}

\newcommand{\uC}{\underline{C}}

\newcommand{\dur}{\nu}
\newcommand{\dvr}{\lmb}

\newcommand{\f}{\frac}
\newcommand{\de}{\delta}
\newcommand{\nab}{\nabla}

\newcommand{\CH}{\calC \calH^{+}}		
\newcommand{\EH}{\calH^{+}}
\newcommand{\pfstep}[1]{\vspace{.5em} {\it \noindent #1.}}

\newcommand{\e}{\bfe}							

\newcommand{\NI}{\calI^{+}}						

\newcommand{\M}{\calM}							
\newcommand{\gbg}{\overline{g}}					
\newcommand{\Fbg}{\overline{F}}					
\newcommand{\ubg}{\overline{u}}					
\newcommand{\Ubg}{\overline{U}}					
\newcommand{\bfUbg}{\overline{\bfU}}				
\newcommand{\vbg}{\overline{v}}					
\newcommand{\Vbg}{\overline{V}}					
\newcommand{\bfVbg}{\overline{\bfV}}				
\newcommand{\Omgbg}{\overline{\Omg}}				
\newcommand{\phibg}{\overline{\phi}}				
\newcommand{\psibg}{\overline{\psi}}				
\newcommand{\dvrbg}{\overline{\dvr}}				
\newcommand{\durbg}{\overline{\dur}}				
\newcommand{\rbg}{\overline{r}}					
\newcommand{\varpibg}{\overline{\varpi}}				
\newcommand{\mubg}{\overline{\mu}}				
\newcommand{\kppbg}{\overline{\kpp}}					
\newcommand{\gmmbg}{\overline{\gmm}}					
\newcommand{\ebg}{\overline{\e}}					
\newcommand{\Phibg}{\overline{\Phi}}				

\newcommand{\fbg}{\overline{f}}
\newcommand{\hbg}{\overline{h}}
\newcommand{\ellbg}{\overline{\ell}}
\newcommand{\dphibg}{\overline{\dot{\phi}}}

\newcommand{\phidf}{\widetilde{\phi}}				
\newcommand{\psidf}{\widetilde{\psi}}				
\newcommand{\dvrdf}{\widetilde{\dvr}}				
\newcommand{\durdf}{\widetilde{\dur}}				
\newcommand{\rdf}{\widetilde{r}}					
\newcommand{\varpidf}{\widetilde{\varpi}}				
\newcommand{\Phidf}{\widetilde{\Phi}}				

\newcommand{\dvrd}{\underline{\dvr}}				
\newcommand{\durd}{\underline{\dur}}				
\newcommand{\alpd}{\underline{\alp}}				
\newcommand{\btd}{\underline{\bt}}					

\newcommand{\ub}{\underline{u}}
\newcommand{\Vb}{\underline{V}}

\newcommand{\Int}{\mathscr{B}}
\newcommand{\Ext}{\mathscr{E}}

\newcommand{\Abs}[1]{\left\vert #1 \right\vert}

\begin{document}

\title[]{Strong cosmic censorship in spherical symmetry for two-ended asymptotically flat initial data II. The exterior of the black hole region}
\author{Jonathan Luk}
\address{Department of Mathematics, Stanford University, Palo Alto, CA, USA}
\email{jluk@stanford.edu}

\author{Sung-Jin Oh}
\address{Korea Institute for Advanced Study, Seoul, Korea}
\email{sjoh@kias.re.kr}


\begin{abstract}
This is the second and last paper of a two-part series in which we prove the $C^2$-formulation of the strong cosmic censorship conjecture for the Einstein--Maxwell--(real)--scalar--field system in spherical symmetry for two-ended asymptotically flat data. In the first paper, we showed that the maximal globally hyperbolic future development of an admissible asymptotially flat Cauchy initial data set is $C^2$-future-inextendible provided that an $L^2$-averaged (inverse) polynomial lower bound for the derivative of the scalar field holds along each horizon. In this paper, we show that this lower bound is indeed satisfied for solutions arising from a generic set of Cauchy initial data. Roughly speaking, the generic set is open with respect to a (weighted) $C^1$ topology and is dense with respect to a (weighted) $C^\infty$ topology. The proof of the theorem is based on extensions of the ideas in our previous work on the linear instability of Reissner--Nordstr\"om Cauchy horizon, as well as a new large data asymptotic stability result which gives good decay estimates for the difference of the radiation fields for small perturbations of an arbitrary solution.
\end{abstract}

\maketitle

\tableofcontents

\section{Introduction}\label{sec.intro}

This is the second (and last) paper of our series on proof of the strong cosmic censorship conjecture 
for the Einstein--Maxwell--(real)--scalar--field system in spherical symmetry for two-ended asymptotically flat initial data. More precisely, we study a quadruple $(\mathcal M,g,\phi,F)$, where $\mathcal M$ is a $4$-dimensional manifold, $g$ is a Lorentzian metric on $\mathcal M$, $\phi$ is a real-valued function on $\mathcal M$ and $F$ is a $2$-form on $\mathcal M$. The system of equations is given by
\begin{equation}\label{EMSFS}
\begin{cases}
Ric_{\mu\nu}-\f12 g_{\mu\nu} R=2(T^{(sf)}_{\mu\nu}+T^{(em)}_{\mu\nu}),\\
T^{(sf)}_{\mu\nu}=\rd_\mu\phi\rd_\nu\phi-\f 12 g_{\mu\nu} (g^{-1})^{\alp\beta}\rd_\alp\phi\rd_{\beta}\phi,\\
T^{(em)}_{\mu\nu}=(g^{-1})^{\alp\bt}F_{\mu\alp}F_{\nu\bt}-\f 14 g_{\mu\nu}(g^{-1})^{\alp\bt}(g^{-1})^{\gamma\sigma}F_{\alp\gamma}F_{\bt\sigma},
\end{cases}
\end{equation}
where $\phi$ and $F$ satisfy
$$\Box_g\phi=0,\,\quad dF=0,\,\quad (g^{-1})^{\alpha\mu}\nab_\alpha F_{\mu\nu}=0.$$
Here, $\Box_g$ and $\nab$ respectively denote the Laplace--Beltrami operator and the Levi--Civita connection associated to the metric $g$. We consider the class of solutions that arise from spherically symmetric data, so that the solutions are themselves spherically symmetric.

One of the most puzzling problems in general relativity can already be seen in the Reissner--Nordstr\"om spacetimes, which constitute an explicit two-parameter family of spherically symmetric, static solutions to \eqref{EMSFS}. When the charge does not vanish and the spacetime is subextremal, the maximal globally hyperbolic future developments of Reissner--Nordstr\"om data are extendible non-uniquely (!) as smooth solutions to \eqref{EMSFS}. This breakdown of uniqueness challenges the deterministic nature of Einstein's theory. A proposed resolution by Penrose, known under the name of the strong cosmic censorship conjecture, is that this phenomenon does \emph{not} occur for solutions to \eqref{EMSFS} arising from \emph{generic} initial data. Together with our companion paper \cite{LO.interior}, we complete in this paper the proof of this conjecture when the initial data are spherically symmetric and have two asymptotically flat ends. We state our theorem roughly as follows and refer the reader to \cite{LO.interior} for a more precise statement, in particular for the definition of the relevant topologies. For simplicity, we will only state the theorem for initial data with smooth and compactly supported initial scalar field: this is in fact the most difficult case; we will again refer the reader to \cite{LO.interior} for versions of the theorem in the larger class of initial data with polynomially decaying initial scalar field.
\begin{theorem}[$C^2$ formulation of strong cosmic censorship conjecture, rough version]\label{main.theorem.intro}
There exists a subset $\mathcal G$ of admissible smooth two-ended asymptotically flat spherically symmetric initial data with compactly supported scalar field such that the maximal globally hyperbolic future development of any element of $\mathcal G$ is $C^2$-future-inextendible. The set $\mathcal G$ satisfies the following genericity conditions:
\begin{enumerate}
\item The generic set $\mathcal G$ is open in a weighted $C^1$ topology,
\item The non-generic set $\mathcal G^c$ has co-dimension at least $1$ in a weighted $C^\infty$ topology. In particular, $\mathcal G$ is dense in a weighted $C^\infty$ topology.
\end{enumerate}
\end{theorem}

Our theorem builds on many remarkable works in the past three decades. In the chargeless ({\bf e}=0) case, a stronger result of $C^0$-future-inextendibility\footnote{It can be shown at least that it is $C^0$-future-inextendible within the class of spherical symmetric metrics. Without the symmetry assumption on the extension, only $C^2$-future-inextendibility is known \cite{Kommemi}, although the methods in \cite{Sbie.C0} are probably relevant for understanding $C^0$-future-inextendibility.} \footnote{It is also for the reason of Christodoulou's theorem that we did not specify that the charge is non-zero in the rough statement of Theorem~\ref{main.theorem.intro}, even though this paper mainly concerns that case.} follows as a consequence of the seminal series of works of Christodoulou on the weak cosmic censorship conjecture \cite{Chr.instab, Christodoulou:1991yfa}. In the case where the charge ${\bf e}\neq 0$, the situation is clearly more subtle as one at least has to understand the instability of Reissner--Nordstr\"om, which as we mentioned above is smoothly extendible. This problem was first studied in the works of Dafermos \cite{D1, D2}, who showed that in a neighborhood of Reissner--Nordstr\"om, the solutions exhibit both stability and instability features. From the point of view of the strong cosmic censorship conjecture, the most important (and perhaps somewhat surprising) consequence of the celebrated works of Dafermos \cite{D2} and Dafermos--Rodnianski \cite{DRPL} is that when ${\bf e}\neq 0$, $C^0$-future-inextendibility is \emph{always false}! Thus, in the formulation of the strong cosmic censorship conjecture, future-inextendibility can only hold in terms of a higher regularity. Indeed, in \cite{LO.interior} and this paper, we show\footnote{See also a remarkable conditional result of Dafermos \cite{D2}, which is to be discussed after the statement of Theorem~\ref{thm.main.thispaper.intro}.} in the main theorem that solutions arising from generic data are $C^2$-future-inextendible\footnote{For the purpose of our proof, $C^2$ is a useful regularity class as it allows us to discuss geodesics and curvature. On the other hand, from a PDE point of view, it is perhaps more natural to ask the question of extendibility in $W^{1,2}_{loc}$, which is at the threshold where one can define weak solutions. See discussions in \cite{Chr, D3, LO.interior}.}. We refer the reader to our companion paper \cite{LO.interior} and the references therein for further discussions on the motivation and history of the problem, as well as the contributions from the physics literature.

By results in \cite{Kommemi, D3}, the maximal globally hyperbolic future development of an arbitrary admissible initial data set is depicted by one of the two Penrose diagrams\footnote{For an introduction to Penrose diagrams, see \cite[Appendix~C]{DRPL}.} in Figure~\ref{fig:Kommemi.intro} (for the notation, we refer the reader to Theorem~\ref{thm:kommemi}). In particular, it can be divided into the black hole exterior region\footnote{Note that the black hole exterior region has two components $\Ext_{1}$ and $\Ext_{2}$ corresponding to the two asymptotically flat ends of the initial hypersurface.} $\Ext = \Ext_{1} \cup \Ext_{2}$ and the black hole interior region $\Int$. This paper is mainly focused on the black hole exterior region $\Ext$. In our companion paper \cite{LO.interior}, where we mainly treated the black hole interior region $\Int$, we proved the strong cosmic censorship conjecture \emph{assuming} that there exists a generic set of admissible initial data such that a polynomial lower bound on an $L^2$-averaged quantity holds along each of the event horizons. The results of \cite{LO.interior} can be rephrased and summarized as follows:

\begin{theorem}[Main result in \cite{LO.interior}, rough version]\label{main.result.interior}
Let $(\mathcal M, g, \phi, F)$ be the maximal globally hyperbolic future development of an admissible $2$-ended asymptotically flat smooth spherically symmetric Cauchy initial data set with smooth and compactly supported initial scalar field and with non-vanishing charge. Then $(\mathcal M, g)$ is depicted by one of the two Penrose diagrams in Figure~\ref{fig:Kommemi.intro}, and the solution approaches the event horizons of (potentially two different) Reissner--Nordstr\"om spacetime along $\EH_1$ and $\EH_2$. Assume that along $\EH_1$, with respect to an Eddington-Finkelstein type coordinate\footnote{This can be understood as requiring the normalization $\f{\rd_v r}{1-\mu}=1$ on $\EH_1$, cf. Theorem~\ref{thm:blowup}.} $v$, the following lower bound holds for some\footnote{$\alp'$ is restricted to this range when the initial scalar field is compactly supported, but not in the general case. For details, see \cite{LO.interior}.} $\alp'\in [3, 10)$:
\begin{equation}\label{interior.condition}
\int_1^\infty v^{\alp'}(\rd_v\phi\restriction_{\EH_1})^2(v)\, dv =\infty.
\end{equation}
Moreover, assume that a similar lower bound as \eqref{interior.condition} holds on $\EH_2$ (with potentially different $E$ and $\alp'$).
Then $(\mathcal M, g)$ is future-inextendible with a $C^2$ Lorentzian metric.
\end{theorem}

\begin{figure}[h]
\begin{center}
\def\svgwidth{400px}
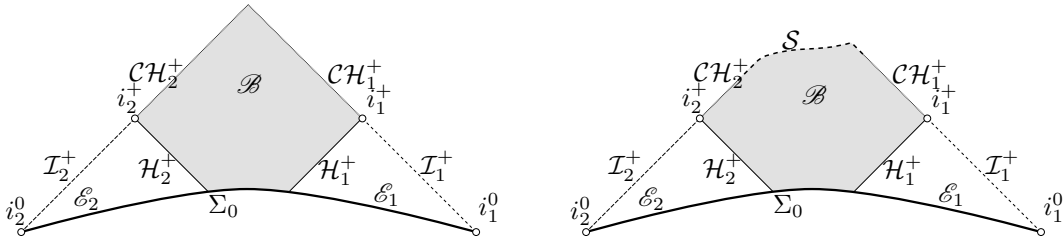 
\caption{The two possible Penrose diagrams of the maximal globally hyperbolic future development of an admissible initial data set.} \label{fig:Kommemi.intro}
\end{center}
\end{figure}

Given the above result, the remaining task, which is achieved in this paper, is to show that for a generic (in the sense of Theorem~\ref{main.theorem.intro}) set of initial data $\mathcal G$, \eqref{interior.condition} holds on each of the event horizons. In rough terms, our results in this paper can be summarized as follows:
\begin{theorem}[Main result in this paper, rough version]\label{thm.main.thispaper.intro}
There exists a subset $\mathcal G$ of future admissible smooth two-ended asymptotically flat spherically symmetric initial data with compactly supported scalar field such that the lower bound \eqref{interior.condition} holds on each of the event horizons for every\footnote{We remark that Price's heuristic (cf.~the linear case in Theorem~\ref{thm.AAG} below) suggests that \eqref{interior.condition} may hold for $\alp'=7$. However, this end-point case does not seem to follow from our arguments.} $\alp'>7$. Moreover, the set $\mathcal G$ satisfies the genericity assumptions (1) and (2) in Theorem~\ref{main.theorem.intro}.
\end{theorem}

We will in fact not prove Theorem~\ref{thm.main.thispaper.intro} as stated, but will prove Theorems~\ref{thm:blowup}, \ref{thm:L-stability} and \ref{thm:instability} in Section~\ref{sec:main-thm}, which are stronger and together imply Theorem~\ref{thm.main.thispaper.intro}. These results are precisely what is assumed in \cite{LO.interior} to complete the proof of (more general versions of) Theorem~\ref{main.theorem.intro}. We again refer the reader to \cite{LO.interior} for details.

Prior to Theorem~\ref{main.result.interior}, a slightly weaker conditional result was known in the seminal work of Dafermos \cite{D2}. Namely, instead of the averaged lower bound \eqref{interior.condition}, \cite{D2} requires a \emph{pointwise} lower bound along the event horizon. While it is still an open problem to construct a \emph{single} regular solution obeying this pointwise lower bound, this pointwise lower bound conjecturally holds for solutions arising from generic initial data. This conjecture is in particular supported by the very recent result \cite{AAG} on a similar generic lower bound for solutions to the linear wave equation on fixed Reissner--Nordstr\"om background, which is obtained by first proving very precise upper bound estimates, see discussions in Section~\ref{sec.linear}. On the other hand, in this paper, we make use of the fact that the required lower bound in Theorem~\ref{main.result.interior} is much less precise than that in \cite{D2}, and it can be proved even without having precise upper bound estimates of the scalar field\footnote{In addition to this, when compared to \cite{D2}, Theorem~\ref{main.result.interior} also has the advantage that the $L^2$-type lower bound conjecturally holds for a larger class of models and the ideas in establishing Theorem~\ref{main.result.interior} may therefore be applicable in other settings. We refer the reader to \cite{LO.interior} for further discussions.}.

Theorem~\ref{thm.main.thispaper.intro} is proved by identifying a constant\footnote{The reader may notice that in the precise statements of the theorems, say, Theorem~\ref{thm:blowup}, the relevant quantity is called $\mathfrak L_{(\omg_0)\infty}$ instead of $\mathfrak L$. Nevertheless, according to \eqref{eq:adm-id-limits} and \eqref{eq:Linfty-def}, $\mathfrak L_{(\omg_0)\infty}=\mathfrak L$ when the scalar field is initially compactly supported.} $\mathfrak L$ associated to each connected component of future null infinity so that the non-vanishing of this constant implies the lower bound \eqref{interior.condition}. It is then shown that $\mathfrak L$ is generically non-vanishing. The constant $\mathfrak L$, which was introduced in \cite{LO1} for a different class of spacetimes, can be computed as an integral along null infinity involving the radiation field of the scalar field and the Bondi mass. Intuitively, it measures the leading order contribution of the incoming scalar field radiation along future null infinity. We note that $\mathfrak L$ cannot be computed directly from the initial data alone, but it nevertheless turns out to be a useful conceptual tool for proving Theorem~\ref{thm.main.thispaper.intro}. This is particularly the case for proving that the generic condition, characterized by $\mathfrak L\neq 0$, is stable: while in general it is difficult to directly prove that an asymptotic lower bound is stable under small perturbations, the identification of $\mathfrak L$ reduces this to the much more familiar task of proving upper bounds for the perturbations of the radiation field. 

We emphasize that in the above scheme, the lower bound along the event horizon is completely generated by back-scattering near the asymptotically flat end. In particular, the lower bound does \underline{not} occur due to some polynomial tail of the initial data. Rather, it is generated by the interaction between the scalar field and the geometry of the spacetime.

The proof of generic lower bounds using the quantity $\mathfrak L$ can be viewed as a nonlinear generalization of \cite{LO.instab}, where similar generic lower bounds were proven for the solutions to the \emph{linear} scalar wave equation on a \emph{fixed} subextremal Reissner--Nordstr\"om spacetime; see Section~\ref{sec.intro.LO.linear}. In order to apply the ideas in \cite{LO.instab}, we need an additional \emph{stability} result. In particular, the stability result quantifies the ``effect'' of small perturbations of a particular solution, which is important for showing that the generic set of initial data is open and dense in appropriate topologies. We will describe these ideas further in Section~\ref{sec.ingredients}, where we discuss the main ingredients of the proof.

The remainder of the introduction is organized as follow: In Section~\ref{sec.ingredients}, we will discuss some ingredients of the proof. In Sections~\ref{sec.previous.works} and \ref{sec.previous.lower.bound}, we will review some previous works related to mathematical analysis of the exterior regions of the black holes. In Section~\ref{sec.previous.works}, we discuss works related to the stability of the exterior region. In Section~\ref{sec.previous.lower.bound}, we discuss works on the lower bounds of solutions to the wave equation in the black hole exterior. Finally, in Section~\ref{outline}, we end the introduction with an outline of the remainder of the paper.

\subsection{Ingredients of the proof}\label{sec.ingredients}

In this subsection, we give a high-level overview of the structure of the proof of Theorem~\ref{thm.main.thispaper.intro}. First, our proof relies heavily on the following known analytic results:
\begin{enumerate}
\item (Dafermos--Rodnianski Price's law \cite{DRPL}) In a seminal work, Dafermos--Rodnianski \cite{DRPL} showed that for \emph{any} (potentially large) spherically symmetric solutions to \eqref{EMSFS} such that the event horizon asymptotes to a subextremal\footnote{Recall that this holds in our setting due to an observation of Kommemi, see Proposition~\ref{prop.subextremality}.} event horizon, then in an appropriately normalized double null coordinate system, the scalar field $\phi$, together with its first derivatives, decay with a quantitative inverse polynomial rate (see Section~\ref{sec:bg} for a precise statement). One can infer from this, as we show in this paper, that in an appropriately normalized gauge, appropriate geometric quantities also approach the corresponding values of a Reissner--Nordstr\"om spacetime with a quantitative rate.
\item (Linear instability \cite{LO.instab}) In \cite{LO.instab}, we proved that for the linear wave equation $\Box_{g_{RN}}\phi=0$ on a fixed Reissner--Nordstr\"om background, solutions arising from \emph{generic} smooth and compactly supported initial data fail to be in $W^{1,2}_{loc}$ at the Cauchy horizon. One key step of the proof in \cite{LO.instab} was to establish an $L^2$-averaged lower bound on the event horizon similar to that in \eqref{interior.condition} (see Theorem~\ref{linear.thm} for a precise statement of this linear result).
\end{enumerate}

The rough program is therefore to use the Price's law decay result to establish that at least after a sufficiently long time, the spacetime is sufficiently close to Reissner--Nordstr\"om and therefore one can apply (a suitable modification) of the linear instability argument in \cite{LO.instab} to show that the lower bound required in Theorem~\ref{main.result.interior} holds for solutions arising from generic data. 

The argument in \cite{LO.instab} is based on analysis of the quantity $\mathfrak L$ and it consists of two parts: namely, (1) show that $\mathfrak L\neq 0$ implies that the lower bound on the event horizon is achieved, and (2) show that $\mathfrak L\neq 0$ is a generic condition. In the nonlinear setting, using the Price's law decay (with appropriate extensions), part (1) can indeed be carried out. That is, it can be shown that $\mathfrak L\neq 0$ implies the desired lower bound. Moreover, this is achieved with a similar proof as in \cite{LO.instab} except that now we also have to handle some decaying error term. More precisely, we use a contradiction argument to show that if the lower bound on the event horizon is not satisfied, then one can propagate decay estimates which are sufficiently strong to contradict $\mathfrak L\neq 0$.

However, the genericity of $\mathfrak L\neq 0$ is much more subtle in the nonlinear setting. To begin with, in the linear setting, genericity amounts to proving (1) $\mathfrak L$ is bounded in terms of some initial norm of the scalar field, and (2) the existence of a \emph{single} solution such that $\mathfrak L\neq 0$. This immediately implies that the generic set is open\footnote{In fact, this statement is so obvious in the linear setting that it was not explicitly stated in \cite{LO.instab}!} and the non-generic set has co-dimension at least $1$. In the nonlinear setting, however, one at least needs to understand the perturbations of an arbitrary large data solution.

It is quite fortunate that \underline{both} the openness of the generic set and the nontrivial co-dimensionality of the non-generic set turn out to rely on the same main technical result, which is a stability statement for the quantity $\mathfrak L$ for perturbations of any (potentially large data) solution. To prove this stability result, we need to compare two nearby solutions and prove that their difference remains small and decays sufficiently fast. For this, it is important to carefully choose a future-normalized gauge (for each solution) in order to control the scalar field, the geometric quantities and their differences. To control the difference of the scalar fields, we rely on many ideas developed in recent years on linear wave equation on black hole spacetimes, in particular the Dafermos--Rodnianski $r^p$-weighted estimates \cite{DRNM}; see Section~\ref{sec.linear} for further discussions.

It is obvious that the stability result implies the openness of $\mathfrak L\neq 0$. However, it is a bit more subtle to see how the stability result is useful for showing that $\mathfrak L=0$ is unstable. A simpler form of the key idea is already used in \cite{LO.instab} to construct a solution with $\mathfrak L\neq 0$. Roughly speaking, given a solution with $\mathfrak L=0$, one constructs a one-parameter family of perturbations (parametrized by $\ep$) of such that the perturbation of the scalar field is compactly supported near the asymptotically flat end. The initial perturbation can be chosen such that 
\begin{enumerate}
\item it gives an $\ep$-perturbation of $\mathfrak L$ near spatial infinity, to the past of an outgoing null hypersurface $\{u=U_{pert}\}$;
\item the perturbation of the spacetime is essentially supported to the past of $\{u=U_{pert}\}$ in the sense that the deviation of the solution with the background on $\{u=U_{pert}\}$ is of size $\ll \ep$.
\end{enumerate}
See Figure~\ref{fig:L-inst-key} in Section~\ref{sec:instability} for a depiction of the relevant regions. The key point is that this can be achieved just by analyzing the region to the past of $\{u=U_{pert}\}$, without having to track the global dynamics of the perturbated spacetime. Nevertheless, once this is achieved, one can rely on the stability result of $\mathfrak L$ to conclude that the contribution to $\mathfrak L$ to the future of $\{u=U_{pert}\}$ is much smaller than the main term to the past of $\{u=U_{pert}\}$. Thus, in the perturbed spacetime, $\mathfrak L\sim \ep \neq 0$.

Let us recap by outlining the main steps of the proof. We include here a pointer to the more precise statements of these steps, which can be found in Section~\ref{sec:main-thm}, after the necessary notations are introduced.

\begin{enumerate}
\item \pfstep{Step 1, Theorem~\ref{thm:blowup}: Relating the condition on the the event horizon to $\mathfrak L\neq 0$} To derive this result, we prove the contrapositive: We assume that the $L^2$-averaged lower bound fails and this can be viewed as an upper bound which we can then propagated all the way up to null infinity and show that we must have $\mathfrak L=0$. This argument, which was introduced in \cite{LO.instab}, is very robust and the decay estimates in \cite{DRPL} is more than sufficient to show that this also applies in the nonlinear setting.
\item \pfstep{Step 2, Theorem~\ref{thm:L-stability}: Stability theorem in the exterior region} This step is the technical heart of this paper. We prove that for every fixed solution arising from admissible data, $\epsilon$-size perturbations of the initial data lead to $C\epsilon$-size perturbations (with appropriate weights) of the radiation field along future null infinity, with $C$ \emph{depending only on the fixed solution}. 
\item \pfstep{Step 3: Openness of $\mathfrak L\neq 0$} Once the stability theorem is proved, the openness of $\mathfrak L\neq 0$ follows as an immediate corollary.
\item \pfstep{Step 4, Theorem~\ref{thm:instability}: Instability of $\mathfrak L=0$} To show the instability of the condition $\mathfrak L=0$, we construct a one-parameter family of perturbations of the background data such that the difference of the scalar fields behaves in a similar manner as for the linear case \cite{LO.instab}. As described in the paragraph above, we construct perturbations so that the main term near spatial infinity can be well-controlled and such that the remaining contribution can be shown to be small error terms using Step 2 above.
\end{enumerate}

\subsection{Previous works on the stability of black hole exteriors}\label{sec.previous.works}

We will not review all the relevant literature regarding the general study of the strong cosmic censorship conjecture for asymptotically flat initial data in spherical symmetry or the instability of Cauchy horizons. For our particular problem at hand, the most relevant previous results by Dafermos \cite{D1, D2, D3}, Dafermos--Rodnianski \cite{DRPL} and Kommemi \cite{KomThe, Kommemi} have already been discussed in other parts of this paper. For a further discussion of the relevant literature, we refer the reader to our companion paper \cite{LO.interior}.

Instead, in this and the next subsections, we will review some works related to the analysis of the black hole \emph{exterior} region. In this subsection, we in particular focus on stability results, in the spherically symmetric or in the linear setting. We will also take the opportunity to point out the relevance of these works to our present paper.

\subsubsection{Stability of black hole exterior regions in spherically symmetric problems}

In spherical symmetry, the study of black hole exterior regions is much simpler than outside spherical symmetry since the area radius function is strictly positive and the (monotonic) Hawking mass is coercive and subcritical. In particular, for a large class of physical matter models including \eqref{EMSFS}, Dafermos showed that in the presence of a trapped surface, null infinity is complete \cite{DafTrapped}. Roughly speaking, this implies that even for large data, as long as the presence of a trapped surface is guaranteed, the solution has a ``regularly behaved'' exterior region. Using moreover the monotonicity of the Hawking mass, it can be easily shown that the explicit Schwarzschild, Reissner--Nordstr\"om black hole exteriors are orbitally stable. Similar results also hold for the vacuum equations in $(4+1)$-dimensions under triaxial Bianchi-IX symmetry \cite{DafHol}.

Even in spherical symmetry, the question of asymptotic stability is considerably harder than that of orbital stability. In the setting\footnote{See also \cite{Hol.biaxial} for a closely related problem on the asymptotic stability of the Schwarzschild---Tangherlini family of solutions for the vacuum equations in $(4+1)$-dimensions under triaxial Bianchi-IX symmetry.} of \eqref{EMSFS}, since (by a generalization of Birkhoff's theorem) all spherically symmetric asymptotically flat solutions in the absence of the scalar field are given by Reissner--Nordstr\"om, asymptotic stability of Reissner--Nordstr\"om can be captured by the decay rate of the scalar field. This was achieved in the remarkable work of Dafermos--Rodnianski \cite{DRPL}, who proved inverse polynomial pointwise decay upper bound estimates for the scalar field, which are consistent with the decay rate suggested by the heuristic study of Price \cite{Price}. Even more remarkably, the result in \cite{DRPL} holds for \emph{large data} solutions as long as a trapped surface is present! See Section~\ref{subsec:DR-full} for the precise statement of the result in \cite{DRPL} and further discussions of its consequences. As mentioned before, in the context of this paper, the result in \cite{DRPL} is an important starting point of the analysis, since it gives very strong asymptotic control of any solution arising from admissible initial data. 

Finally, let us note that in the analysis of this paper, we also prove an asymptotic stability result, for perturbations around \emph{any} fixed (potentially large data) solution. We show that the difference of the scalar fields between the perturbed solution and the background solution decays. However, in contrast to \cite{DRPL}, in our setting there are no gauge-invariant ways to measure the difference between two solutions with non-vanishing scalar fields, and the proof therefore requires a careful choice of gauges normalized in the future for each solution. In fact, one can easily convince oneself that for some unwise choices of coordinate system, the stability result no longer holds.

\subsubsection{Boundedness and decay estimates for solutions to the linear wave equation in black hole exterior regions}\label{sec.linear}
 
Another line of works which is relevant to our approach in this paper concerns the study of the linear scalar wave equation
\begin{equation}\label{linear.wave}
\Box_g\phi=0
\end{equation}
in the black hole exterior region with a fixed background metric. This problem has received a lot of attention in the past decade, see \cite{Are11a, Are12, AB, BSt, DRS, DRK, DRL, DRNM, DRSR, DHR, MMTT, MTT, Ta, TT} and the references therein for a sample of results. A culmination of these results is a theorem of Dafermos--Rodnianski--Shlapentokh-Rothman, which proves that for all subextremal Kerr spacetimes, solutions to \eqref{linear.wave} (arising from sufficiently regular initial data) obey both energy decay and pointwise decay estimates.

A reason that the study of \eqref{linear.wave} is relevant to our paper is that (as mentioned in Section~\ref{sec.ingredients}) a crucial ingredient of the proof is to show that small perturbations of \emph{any} (potentially large) data lead to globally small perturbations of the solution. In order to prove such a result, one needs to control the long time behavior of the difference of two solutions. The main contribution in the equations for the difference comes from the linear terms, which can be controlled using methods from the study of \eqref{linear.wave}.

A particularly relevant technique is the Dafermos--Rodnianski $r^p$-weighted estimates \cite{DRNM} (see also recent extensions by Moschidis \cite{Mos.rp}). This is a physical space technique of proving decay estimates on a very general class of asymptotically flat spacetimes, and it has been previously implemented for nonlinear applications \cite{Yang}. This method plays a crucial role in one of the steps in the stability argument. We will defer a more detailed discussion to Section~\ref{sec:L-stability}.

\subsection{Previous works on lower bounds for solutions to the linear wave equation in black hole exteriors}\label{sec.previous.lower.bound}

Perhaps even more relevant to the present paper than the stability results in the previous subsection are theorems on lower bounds for solutions to linear wave equations. Those results should be thought of as the analogues of the lower bound \eqref{interior.condition} (which is proven in Theorem~\ref{thm.main.thispaper.intro}) but in a linear setting. There are in the literature three different approaches to obtain such lower bounds, which we will review in the next few subsubsections. This in particular includes the work \cite{LO.instab}, to be described in Section~\ref{sec.intro.LO.linear}, which forms part of basis of the proof of Theorem~\ref{thm.main.thispaper.intro}. After reviewing the relevant literature, in Section~\ref{linear.lower.bound.comments}, we will comment on the choice in this paper to attack the nonlinear problem using the approach of \cite{LO.instab}.

\subsubsection{Lower bound via time-translation symmetry}

The simplest result that can be proven regarding lower bounds for solutions to linear wave equations would be one such that \emph{an inverse polynomial lower bound is imposed in the initial data}. Though this is not explicitly available in the literature, using the methods of Dafermos--Shlapentokh-Rothman \cite{DafShl} (see also McNamara \cite{McN}), it can be proven by using the time-translation symmetry of the spacetime and the scattering theory developed in \cite{DRSR} that 

\begin{theorem}[Dafermos--Shlapentokh-Rothman]\label{thm.McN}
Generic smooth and spherically symmetric initial data with a (fixed but otherwise arbitrarily fast) inverse polynomial rate to the linear wave equation $\Box_{g_{RN}}\phi=0$ on a fixed Reissner--Nordstr\"om spacetime with parameters satisfying $0<|{\bf e}|<M$ give rise to solutions $\phi$ which obey the following lower bound along the event horizon:
$$\int_{\EH} v^{\alp} (\rd_{v} \phi)^{2} \, \ud v = \infty$$
for some sufficiently large $\alp$, where $v$ is the Eddington--Finkelstein coordinate defined by $v=\f 12(t+r^*)$ with $r^*=r+(M+\frac{2M^2-{\bf e}^2}{2\sqrt{M^2-{\bf e}^2}})\log (r-r_+) +(M-\frac{2M^2-{\bf e}^2}{2\sqrt{M^2-{\bf e}^2}})\log (r-r_-)$.
\end{theorem}

The idea underlying this result is very simple: Consider a smooth and compactly supported initial data set on past null infinity such that the solution is non-trivial on the future event horizon. Such an initial data set exists by considerations in \cite{DRSR}. If this solution does not already obey the estimate in Theorem~\ref{thm.McN}, then one can sum a series of solutions arising from time-translating and rescaling the initial data on past null infinity. This new solution will then satisfy the estimate in Theorem~\ref{thm.McN}. Notice that this construction requires the possibility that the data on past null infinity only decay polynomially, which as a consequence implies that the data on a Cauchy hypersurface also only decay polynomially. 

The advantage of this approach, as we see above, is that it requires very little information about the specific form of the metric. Thus, the result generalizes easily to Kerr and Kerr-Newman spacetimes. On the other hand, the solution that is constructed has a very special profile. Though in the linear setting, this is already sufficient to deduce genericity of solutions with a lower bound on the event horizon, it is difficult even to study stability properties of these solutions.

\subsubsection{Lower bound via a contradiction argument in the exterior region}\label{sec.intro.LO.linear}

A second approach of establishing lower bounds for solutions to the linear wave equation on Reissner--Nordstr\"om was given as a part of our previous work \cite{LO.instab} on the linear instability of the Reissner--Nordstr\"om Cauchy horizon. We will only state the result for spherically symmetric solutions (in particular in view of their special relevance to the present paper), but it of course follows easily that the following bound holds for the spherical mean of any (potentially non-spherically symmetric) solutions.

\begin{theorem}[Luk--Oh \cite{LO.instab}]\label{linear.thm}
Generic smooth, compactly supported and spherically symmetric initial data to the linear wave equation $\Box_{g_{RN}}\phi=0$ on a fixed Reissner--Nordstr\"om spacetime with parameters satisfying $0<|{\bf e}|<M$ give rise to solutions $\phi$ which obey the following lower bound along the event horizon:
$$\int_{\EH} v^{\alp} (\rd_{v} \phi)^{2} \, \ud v = \infty$$
for any $\alp>7$, where $v$ is as in Theorem~\ref{thm.McN}.
\end{theorem}

As discussed in Section~\ref{sec.ingredients}, Theorem~\ref{linear.thm} is proven by identifying a quantity $\mathfrak L$ at infinity which is generically nonvanishing and such that $\mathfrak L\neq 0$ implies the desired lower bound. In addition to the fact that Theorem~\ref{linear.thm} now proves a lower bound even for compactly supported initial data, the approach of Theorem~\ref{linear.thm} has the additional advantage that $\mathfrak L\neq 0$ is a \emph{stable} property.

The proof of Theorem~\ref{linear.thm} proceeds by a contradiction argument: it is shown that if $\int_{\EH} v^{\alp} (\rd_{v} \phi)^{2} \, \ud v$ is finite, then one can obtain strong enough decay estimates in the exterior region which are inconsistent with $\mathfrak L\neq 0$. As we have already indicated earlier, this theorem (or more precisely its proof) will play an important role in the present paper. This is in part due to the robustness of the proof, see Section~\ref{linear.lower.bound.comments}. On the other hand, partly because the proof is based on a contradiction argument, the result is relatively weak. For instance, it does not give any lower bounds away from the event horizon.

\subsubsection{Lower bound via precise late-time asymptotics}

A third approach of obtaining lower bound has very recently been introduced, which moreover gives a good description of the leading order asymptotics. In particular, it is shown that a pointwise lower bound holds true everywhere in the black hole exterior:
\begin{theorem}[Angelopoulos--Aretakis--Gajic \cite{AAG}]\label{thm.AAG}
Generic smooth and compactly supported initial data to the linear wave equation $\Box_{g_{RN}}\phi=0$ on a fixed Reissner--Nordstr\"om spacetime with parameters satisfying $0<|{\bf e}|<M$ give rise to solutions $\phi$ which obey the following \emph{pointwise} lower bound
$$|\rd_t^k \phi|(u,v) \geq c_k v^{-3-k} \qquad (c_{k} > 0) $$
in the region $\{r_+\leq r\leq R\}$ for any finite but fixed $R>r_+$ and $k \geq 1$, where $v$ and $r_+$ are as in Theorem~\ref{thm.McN}.
\end{theorem}
The bounds in Theorem~\ref{thm.AAG} imply a fortiori that $\rd_v\phi$ (on the event horizon $\rd_v$ is parallel to $\rd_t$) obeys the estimates on the event horizon in Theorems~\ref{thm.McN} and \ref{linear.thm}. Combining with the ideas in \cite{D2}, the results in \cite{AAG} also give an alternative proof of the linear instability result of the Reissner--Nordstr\"om Cauchy horizon in \cite{LO.instab}. This result combines sharp upper bounds of the solutions together with a clever application of the conservation of the Newman--Penrose quantities. In particular, it requires very precise upper bounds of the solutions, which are in turn derived using an extension of the Dafermos--Rodnianski $r^p$-weighted estimates. 

\subsubsection{Comments on the nonlinear problem}\label{linear.lower.bound.comments}

In view of the three philosophically distinct approaches to obtaining lower bounds, one can in principle attempt to obtain instability results for the nonlinear problem using any of these approaches. In this subsection, we explain our choice of using ideas in \cite{LO.instab} to tackle the nonlinear problem.

First, simply because Theorem~\ref{thm.McN} constructs solutions to the linear wave equation with an $L^2$-averaged lower bound, Theorem~\ref{thm.McN} could in principle\footnote{Here, let us suppress the additional complication that \eqref{interior.condition} is in fact needed for a specific range of $\alp$.} yield a result on the existence of perturbations of Reissner--Nordstr\"om data such that \eqref{interior.condition} holds. This then implies the \emph{instability} of Reissner--Nordstr\"om Cauchy horizon. However, it is not immediately clear from the approach of Theorem~\ref{thm.McN} whether the $L^2$-average polynomial lower bound is a stable property. In particular, it would be more difficult to use this to prove a \emph{genericity} statement or even to obtain an instability result for (potentially existing) large data solutions with a $C^2$-regular Cauchy horizon.

In contrast, at least in the linear setting, the other two approaches both prove that the lower bound is a stable property. The approach of Theorem~\ref{thm.AAG} moreover has the advantage that the type of pointwise lower bound for $|\rd_t\phi|$ along the event horizon has a long history in the study of black hole interiors. In fact, that type of lower bound has long been assumed in the study of the instability of the Cauchy horizon, beginning with the pioneering works \cite{Hiscock, PI1, PI2} on the Einstein--null dust model. For \eqref{EMSFS} in spherical symmetry, Dafermos showed in \cite{D2} that a pointwise lower bound for $|\rd_t\phi|$ along the event horizon as in Theorem~\ref{thm.AAG} imply that the Hawking mass blows up at the Cauchy horizon and thus the spacetime is $C^2$-future-inextendible. It would therefore be an interesting problem to prove this pointwise lower bound in the nonlinear setting for generic Cauchy data, completing the program of Dafermos.

However, even in the linear setting, the proof of Theorem~\ref{thm.AAG} requires first proving precise upper bounds. This in particular include the sharp Price's law bounds \cite{Price} for the derivatives of the scalar field, and these are stronger than the known estimates for the nonlinear problem. Moreover, the method in \cite{AAG} requires quite heavily on the staticity of the spacetime metric, which would require some extensions in order to be applicable to the nonlinear problem.

In contrast, the proof of Theorem~\ref{linear.thm} is very robust, and has the advantage that the $L^2$-averaged lower bound on the event horizon can be obtained \emph{without proving sharp decay upper bound estimates}. To see this, recall from Section~\ref{sec.ingredients} that we need to show that (1) $\mathfrak L\neq 0$ implies the desired $L^2$-averaged lower bound and (2) $\mathfrak L\neq 0$ is generic. For (1), the ideas in \cite{LO.interior} in fact need very little geometric information, see discussions in Section~\ref{sec:blowup}. In particular, the decay results in \cite{DRPL} are more than sufficient for this purpose. While the proof of (2) is considerably more involved, it should be noted that statements about $\mathfrak L$ are morally at the level of the $L^1$ norm of the radiation field along future null infinity, which is again weaker than expected sharp decay of the radiation field\footnote{More precisely, it is known \cite{DRPL} that the radiation field $\Phi$ (cf. Definition~\ref{def.L}) decays like $|\Phi|(u) = O(u^{-2})$, while in our proof we only need to show that under small perturbations, the difference of the radiation fields decays with a rate $= O(\eps u^{-\bt})$ for some $\bt > 1$ where $\eps$ is the size of the initial data difference.}.

\subsection{Outline of the paper}\label{outline}
The remainder of the paper will be organized as follows:
\begin{itemize}
\item {\bf Section~\ref{sec.SS}.} We first introduce the notions of spherically symmetric solutions to the Einstein--Maxwell--(real)--scalar--field system. We then write down the symmetry-reduced equation and discuss the initial value problem in spherical symmetry.
\item {\bf Section~\ref{sec.review}.} We review various notions introduced in \cite{LO.interior}. These include the class of data we consider and some preliminary results regarding the maximal globally hyperbolic future development.
\item {\bf Section~\ref{sec:main-thm}.} We give precise statements of the three main theorems we prove in this paper. This can be viewed as the precise formulation of Theorem~\ref{thm.main.thispaper.intro} earlier in the introduction.
\item {\bf Section~\ref{sec:bg}.} We give (decay) estimates in the exterior region in the maximal globally hyperbolic future development arising from an admissible initial data set. These estimates rely heavily on the Price's law decay estimates of Dafermos--Rodnianski \cite{DRPL} (which we will recall).
\item {\bf Section~\ref{sec:blowup}.} This section contains the proof Theorem~\ref{thm:blowup}, which is the statement that if\footnote{$\mathfrak L_{(\omg_0)\infty}$ was defined in \cite{LO.interior} and the definition will be recalled in Definition~\ref{def.L} in Section~\ref{sec.review}. To relate this to the discussion so far, note that $\mathfrak L_{(\omg_0)\infty}=\mathfrak L$ for initial data with compactly supported scalar field} $\mathfrak L_{(\omg_0)\infty}\neq 0$, then a desired $L^2$-averaged lower bound holds along the event horizon.
\item {\bf Section~\ref{sec:extr}.} In this section, we begin the proof of the stability theorem (Theorem~\ref{thm:L-stability}, cf. Step~2 in Section~\ref{sec.ingredients}). We focus here first on three slightly ``easier'' regions, namely a compact region (which is treated using Cauchy stability), a neighborhood of spatial infinity (cf. Figure~\ref{fig:extr-st-cauchy} in Section~\ref{subsec:extr-st}), and a neighborhood of null infinity with finite retarded time range (cf. Figure~\ref{fig:extr-st} in Section~\ref{subsec:extr-st}).
\item {\bf Section~\ref{sec:L-stability}.} This is the technical heart of the paper, in which we treat the most difficult region for the stability theorem (Theorem~\ref{thm:L-stability}).
\item {\bf Section~\ref{sec:instability}.} In this final section, we prove Theorem~\ref{thm:instability}, which is the instability statement for $\mathfrak L=0$ (cf. Step~4 in Section~\ref{sec.ingredients}).
\item {\bf Appendix~A.} To assist the reader, an index of frequently used symbols is given in the appendix.
\end{itemize}

\subsection*{Acknowledgments} The authors thank Mihalis Dafermos for many stimulating discussions and for helpful comments on the manuscript. Much of this work was carried out while J. Luk was at Cambridge University and S.-J. Oh was at UC Berkeley. S.-J. Oh thanks Cambridge University for hospitality during several visits. The authors also thank the Chinese University of Hong Kong for hospitality while some of this work was pursued. 

J. Luk is supported in part by a Terman Fellowship. S.-J. Oh was supported by the Miller Research Fellowship from the Miller Institute, UC Berkeley and the TJ Park Science Fellowship from the POSCO TJ Park Foundation.

\section{Einstein--Maxwell--(real)--scalar--field system in spherical symmetry}\label{sec.SS}
\subsection{Equations in double null coordinates in spherical symmetry}
We begin with a precise definition of a spherically symmetric solution to the Einstein--Maxwell--(real)--scalar--field system.
\begin{definition}[Spherically symmetric solutions]\label{def.SS}
Let $(\mathcal M,g,\phi,F)$ be a suitably regular solution to the Einstein--Maxwell--(real)--scalar--field system \eqref{EMSFS}. We say that $(\mathcal M,g,\phi,F)$ is \emph{spherically symmetric} if the following properties hold:
\begin{enumerate}
\item The symmetry group $SO(3)$ acts on $(\mathcal M,g)$ by isometry with spacelike orbits.
\item The metric $g$ on $\mathcal M$ is given by
\begin{equation}\label{SS.metric.1}
g=g_{\mathcal Q}+r^2 d\sigma_{\mathbb S^2},
\end{equation}
where
\begin{equation}\label{SS.metric.2}
g_{\mathcal Q}=-\f{\Omg^2 }{2}(du\otimes dv+dv\otimes du)
\end{equation}
is a Lorentzian metric on the $2$-dimensional manifold $\mathcal Q=\mathcal M/SO(3)$ and $r$ is defined to be the area radius function of the group orbit, i.e.,
$$r=\sqrt{\f{\mbox{Area}({\boldsymbol \pi}^{-1}(p))}{4\pi}},$$
for every $p\in \mathcal Q$, where ${\boldsymbol \pi}$ is natural projection ${\boldsymbol \pi}:\mathcal M\to \mathcal Q$ taking a point to the group orbit it belongs to. Here, as in the introduction, $d\sigma_{\mathbb S^2}$ denotes the standard round metric on $\mathbb S^2$ with radius $1$.
\item The function $\phi$ at a point $x$ depends only on ${\boldsymbol \pi}(x)$, i.e., for $p\in \mathcal Q$ and $x,y\in {\boldsymbol \pi}^{-1}(p)$, it holds that $\phi(x)=\phi(y)$.
\item The Maxwell field $F$ is invariant under pullback by the action (by isometry) of $SO(3)$ on $\calM$.
Moreover, there exists ${\bf e}:\mathcal Q\to \mathbb R$ such that
$$F=\f{\bfe}{2({\boldsymbol \pi}^* r)^2}{\boldsymbol \pi}^*(\Omg^2\,du\wedge dv).$$
\end{enumerate}
\end{definition}

\begin{remark}[Gauge freedom]
For the metric $g_{\mathcal Q}$ taking the form \eqref{SS.metric.2}, there remains a gauge freedom in choosing the null coordinates. More precisely, one can change variables
\begin{equation*}
	(u, v) \mapsto (u', v') = (u'(u), v'(v))
\end{equation*}
where $u'$ [resp. $v'$] is a strictly monotonic function of $u$ [resp. $u'$]. It will be important in this paper to exploit this freedom and perform the analysis in appropriately chosen gauges. In particular, we will consider more than one system of double null coordinate system in this paper. 
\end{remark}

Given a double null coordinate system, it will be convenient to use the following notation:
\begin{definition}[Null hypersurfaces $C_u$ and $\uC_v$]\label{def:C.uC}
Given a double null coordinate system $(u,v)$ in $\calQ$, we denote
$$C_u=\{(u',v')\in \calQ: u'=u \},\quad \uC_v=\{(u',v')\in \calQ: v'=v \}.$$
\end{definition}

Given the above setup, we now describe the Einstein--Maxwell--(real)--scalar--field system in spherical symmetry, written in terms of the variables and coordinate system introduced in Definition~\ref{def.SS}. First, it is easy to verify that as a consequence of the Maxwell's equations,
\begin{equation}
{\bf e}\mbox{ is a \underline{constant}.}
\end{equation}
The Einstein--Maxwell--(real)--scalar--field system in spherical symmetry then reduces to three wave equations for $r, \phi, \Omg$:
\begin{equation} \label{eq:EMSF-wave}
\left\{
\begin{aligned}
\rd_{u} \rd_{v} r = & - \frac{\Omg^{2}}{4 r} - \frac{\rd_{u} r \rd_{v} r}{r} + \frac{\Omg^{2} \e^{2}}{4 r^{3}}, \\
\rd_{u} \rd_{v} \phi = & - \frac{\rd_{v} r \rd_{u} \phi}{r} - \frac{\rd_{u} r \rd_{v} \phi}{r}, \\
\rd_{u} \rd_{v} \log \Omg = & - \rd_{u} \phi \rd_{v} \phi - \frac{\Omg^{2} \e^{2}}{2 r^{4}} + \frac{\Omg^{2}}{4 r^{2}} + \frac{\rd_{u} r \rd_{v} r}{r^{2}},
\end{aligned}
\right.
\end{equation}
complemented with the Raychaudhuri equations:
\begin{equation} \label{eq:EMSF-ray-orig}
\left\{
\begin{aligned}
	\rd_{v} \left( \frac{\rd_{v} r}{\Omg^{2}}\right) =& - \frac{r (\rd_{v} \phi)^{2}}{\Omg^{2}} ,\\
	\rd_{u} \left( \frac{\rd_{u} r}{\Omg^{2}}\right) =& - \frac{r (\rd_{u} \phi)^{2}}{\Omg^{2}} .
\end{aligned}
\right.
\end{equation}
We introduce the Hawking mass
\begin{equation}\label{eq:mass.def}
	m = \frac{r}{2} (1 - \nb_{a} r \nb^{a} r) = \frac{r}{2} \left( 1 + \frac{4 \rd_{u} r \rd_{v} r}{\Omg^{2}} \right)
\end{equation}
and define also $\mu = \frac{2m}{r}$. Then \eqref{eq:EMSF-ray-orig} can be rewritten as
\begin{equation} \label{eq:EMSF-ray}
\left\{
\begin{aligned} 
	\rd_{v} \left(\frac{\rd_{u} r}{1-\mu} \right) =& \frac{r}{\rd_{v} r} (\rd_{v} \phi)^{2} \frac{\rd_{u} r}{1-\mu}, \\
	\rd_{u} \left(\frac{\rd_{v} r}{1-\mu} \right) =& \frac{r}{\rd_{u} r} (\rd_{u} \phi)^{2} \frac{\rd_{v} r}{1-\mu}.
\end{aligned}
\right.
\end{equation}
Define the modified mass
\begin{equation*}
	\varpi = m + \frac{\e^{2}}{2 r}.
\end{equation*}
Then
\begin{equation} \label{eq:EMSF-r-phi-m}
\left\{
\begin{aligned}
\rd_{u} \rd_{v} r = & \frac{2(\varpi - \frac{\e^{2}}{r})}{r^{2}} \frac{\rd_{u} r \rd_{v} r}{1-\mu}, \\
\rd_{u} \rd_{v} \phi = & - \frac{\rd_{v} r \rd_{u} \phi}{r} - \frac{\rd_{u} r \rd_{v} \phi}{r}, \\
	\rd_{v} \varpi =& \frac{1}{2} \frac{1-\mu}{\rd_{v} r} r^{2} (\rd_{v} \phi)^{2}, \\
	\rd_{u} \varpi =& \frac{1}{2} \frac{1-\mu}{\rd_{u} r} r^{2} (\rd_{u} \phi)^{2}.
\end{aligned}
\right.
\end{equation}
Moreover, it is easy to check that as long as $1-\mu\neq 0$, the system of equations consisting of \eqref{eq:EMSF-ray} and \eqref{eq:EMSF-r-phi-m} is equivalent to the system consisting of \eqref{eq:EMSF-wave} and \eqref{eq:EMSF-ray-orig}.

\subsection{Cauchy problem and characteristic initial value problem}\label{sec:Cauchy}

To solve \eqref{EMSFS} in spherical symmetry (see Definition~\ref{def.SS}), we will consider both the Cauchy problem and the characteristic initial value problem. We begin with the Cauchy problem formulation of the Einstein--Maxwell--(real)--scalar--field system in spherical symmetry. An initial data set in this setting is defined as follows.
\begin{definition}[Cauchy data]\label{def.Cauchy.data}
A \emph{Cauchy initial data set} for the Einstein--Maxwell--(real)--scalar--field system in spherical symmetry consists of a curve $\Sgm_{0}$ (without boundary), a collection of six real-valued functions $(r, f, h, \ell, \phi, \dot{\phi})$ on $\Sgm_{0}$ and a real number ${\bf e}\neq 0$. We require $r\in C^2(\Sgm_0;\mathbb R)$, $f,h,\phi\in C^1(\Sgm_0;\mathbb R)$ and $\ell,\dot{\phi}\in C^0(\Sgm_0;\mathbb R)$. Moreover, $f, r$ are required to be strictly positive everywhere on $\Sgm_0$. For $\Sgm_0$ parametrized\footnote{At this point, we allow $\rho$ to have either finite or infinite range. We will require $\Sgm_0 =\mathbb R$ later in Definition~\ref{def:adm-data}.} by $\rho$, the collection of functions together with ${\bf e}$ give rise to \emph{geometric data} consisting of the following:

\begin{enumerate}
\item The initial hypersurface $\underline{\Sgm_{0}} = \Sgm_{0} \times \bbS^{2}$ is endowed with the intrinsic Riemannian metric
\begin{equation*}
\hat{g} = f^{2}(\rho) \, \ud \rho^2+ r^{2}(\rho) \, \ud \sgm_{\bbS^{2}}.
\end{equation*}

\item The symmetric $2$-tensor $\hat{k}$ on the initial hypersurface $\underline{\Sgm_{0}}$ (which will be the second fundamental form of the solution) given by
\begin{equation*}
	\hat{k} = h(\rho) \, \ud \rho^{2} + \ell(\rho) \, \ud \sgm_{\bbS^{2}}.
\end{equation*}

\item The initial data on $\underline{\Sgm_{0}}$ for the matter fields\footnote{We abuse notation slightly here, where $\phi$ is used to both denote the scalar field in the spacetime and its restriction to the initial slice $\underline{\Sgm_{0}}$.}
$$(\phi, \underline{n} \phi)\restriction_{\underline{\Sgm_{0}}}=(\phi, \dot{\phi}),\quad F(\underline{n}, \rd_\rho)\restriction_{\underline{\Sgm_{0}}}=\f{{\bf e}f}{r^2},$$
where $\underline{n}$ denotes the unique future-directed unit normal to $\underline{\Sgm_0}$ in $\mathcal M$.
\end{enumerate}
Moreover, the following constraint equations are satisfied: 
\begin{gather}
R_{\hat{g}}-|\hat{k}|^2_{\hat{g}}+(\mbox{tr}_{\hat{g}}\hat{k})^2=4T(\underline{n}, \underline{n})=2\dot{\phi}^2+\f{2}{f^2}(\rd_\rho\phi)^2+\f{2{\bf e}^2}{r^4},\label{ham.con}\\
\quad(\mbox{div}_{\hat{g}} \hat{k})_\rho-\rd_{\rho}(\mbox{tr}_{\hat{g}} \hat{k})=2T(\underline{n}, \rd_\rho)=2\dot{\phi}(\rd_{\rho}\phi).\label{mom.con}
\end{gather}

Here, $T = T^{(sf)}_{\mu\nu}+T^{(em)}_{\mu\nu}$ (cf.~\eqref{EMSFS}) and $R_{\hat{g}}$ is the scalar curvature of $\hat{g}$.
\end{definition}

We recall also a lemma from \cite{LO.interior}, which gives a relation between the Cauchy initial data defined in Definition~\ref{def.Cauchy.data} and the initial data for $(r, \phi, \log \Omg)$ in a double null coordinate system.
\begin{lemma} \label{lem:cauchy-to-char}
Consider a parametrization $\rho \mapsto \Sgm_{0}(\rho)$ of the initial curve $\Sgm_{0}$, and consider a double null coordinate system $(u, v)$ on $\calM$ normalized by the conditions
\begin{equation*}
\frac{\ud u}{\ud \rho} = -1, \quad 
\frac{\ud v}{\ud \rho} = 1 \quad \hbox{ on } \Sgm_{0}.
\end{equation*}
Then the following identities hold on $\Sgm_{0}$: 
\begin{gather*}
	\rd_{u} \restriction_{\Sgm_{0}} = \frac{1}{2} (-\rd_{\rho} + f n), \quad 
	\rd_{v} \restriction_{\Sgm_{0}} = \frac{1}{2} (\rd_{\rho} + f n), \\
	\rd_{u} r \restriction_{\Sgm_{0}} = - \frac{1}{2} \rd_{\rho} r + \frac{f}{2 r} \ell, \quad
	\rd_{v} r \restriction_{\Sgm_{0}} = \frac{1}{2} \rd_{\rho} r + \frac{f}{2 r} \ell, \\
	\rd_{u} \phi \restriction_{\Sgm_{0}} = \frac{1}{2} ( - \rd_{\rho} \phi + f \dot{\phi}), \quad
	\rd_{v} \phi \restriction_{\Sgm_{0}} = \frac{1}{2} ( \rd_{\rho} \phi + f \dot{\phi}), \\
	\rd_{u} \log \Omg \restriction_{\Sgm_{0}} = \frac{1}{2f} (- \rd_{\rho} f + h), \quad
	\rd_{v} \log \Omg \restriction_{\Sgm_{0}} = \frac{1}{2f} (\rd_{\rho} f + h),\\
	\Omg \restriction_{\Sgm_{0}}=f,
\end{gather*}
where $n$ denotes the unique future-directed unit normal to $\Sgm_0$ in $\mathcal Q$.
\end{lemma}

Alternatively, the evolution problem for \eqref{EMSFS} in spherical symmetry can be phrased as a characteristic initial value problem. We make precise below the notion of characteristic initial data in Definition~\ref{def.CID}. Although we will not be stating our main theorems in terms of a characteristic initial value problem, this point of view will nonetheless be useful for isolating certain subregions of the spacetime in the proof of the main theorems. 

\begin{definition}[Characteristic initial data]\label{def.CID}
A \emph{characteristic initial data set} for the Einstein--Maxwell--(real)--scalar--field system in spherical symmetry consists of two transversally intersecting null curves $\uC_{in}$ and $C_{out}$, parametrized by $\uC_{in}=\{(u, v_1): u\in [u_1, u_2 )\}$ and $C_{out}=\{(u_1, v): v\in [v_1, v_2 )\}$, where $u_1, v_1\in \mathbb R$ and $u_2, v_2\in \mathbb R\cup\{+\infty\}$ with $u_1<u_2$, $v_1<v_2$, together with
\begin{itemize}
\item a constant $\e$,
\item $(r,\phi,\Omg)\in C^2(\uC_{in})\times C^1(\uC_{in})\times C^1(\uC_{in})$ with $r>0$ and $\Omg>0$, and
\item $(r,\phi,\Omg)\in C^2(C_{out})\times C^1(C_{out})\times C^1(C_{out})$ with $r>0$ and $\Omg>0$
\end{itemize}
such that
\begin{itemize}
\item the values of $(r,\phi,\Omg)$ at $(u_1, v_1)$ coincide,
\item the Raychaudhuri equations \eqref{eq:EMSF-ray-orig} are satisfied on $\uC_{in}$ and $C_{out}$.
\end{itemize}
\end{definition}

Both the Cauchy problem and the characteristic initial value problem are \emph{locally well-posed}. In spherical symmetry, local well-posedness results hold at a level of regularity consistent with $\phi$ being $C^1$. Such results are well-known and we refer the readers to \cite{LO.interior} for precise statements.

\section{Review of geometric setup in \cite{LO.interior}}\label{sec.review}

\subsection{Admissible initial data and the distances on the space of initial data sets}\label{sec.review.basic}

In this subsection, we recall from \cite{LO.interior} the class of data we consider and the various topologies on the space of initial data sets.
\begin{definition} [Admissible Cauchy initial data] \label{def:adm-data}
Let $\omg_{0} > 2$. An \emph{$\omg_{0}$-future-admissible spherically symmetric $2$-ended asymptotically flat Cauchy initial data set with non-vanishing charge} (in short, an $\omg_{0}$-admissible initial data set) is a Cauchy initial data set $\Tht = (r, f, h, \ell, \phi, \dot{\phi}, \e)$ on $\Sgm_{0} = \bbR$ satisfying the following properties.
\begin{enumerate}
\item $\phi,\, f \in C^2(\Sigma_0;\mathbb R)$, $\dot{\phi},\, h \in C^1(\Sigma_0;\mathbb R)$ (i.e., are more regular than required in Definition~\ref{def.Cauchy.data}).
\item The following \emph{asymptotic flatness conditions} hold as $\rho \to \pm\infty$ (i.e., towards each end):
\begin{equation} \label{eq:adm-id-af}
\begin{gathered}
	f(\rho) - 1 = O_{2}(\abs{\rho}^{-1}),\quad
	h(\rho) = O_{1}(\abs{\rho}^{-2}), \\
	r(\rho) - \abs{\rho} = O_{2}(\log \abs{\rho}), \quad
	\ell(\rho) = O_{1}(1)
\end{gathered}
\end{equation}
Here the notation $O_{i}(\abs{\rho}^{-n})$ denotes that the function on the LHS is $O(|\rho|^{-n})$ and the $j$-th derivative is $O(|\rho|^{-n-j})$ for all $j\leq i$. In the case $n = 0$, we simply write $O_{i}(1) = O_{i}(| \rho |^{0})$. The notation $O_{i}(\log \abs{\rho})$ is defined similarly.
\item 
The following asymptotic flatness conditions hold for the scalar field: As $\rho \to \pm \infty$,
\begin{equation} \label{eq:adm-id-phi}
	\phi (\rho) = O_{2}(\abs{\rho}^{-\omg_{0}}), \quad 
	\dot{\phi}(\rho) = O_{1}(\abs{\rho}^{-\omg_{0}-1}). 
\end{equation}
Furthermore, the following limits exist:
\begin{equation} \label{eq:adm-id-limits}
\begin{aligned}
	\mathfrak{L}_{(\omg_{0}) 0}[\Tht] := & \lim_{\rho \to \infty} \frac{r^{\min\set{\omg_{0}, 3}}}{\rd_{\rho} r + \frac{f\ell}{r}} \left(\rd_{\rho} (r \phi) + \frac{f \ell}{r} \phi + f r \dot{\phi} \right), \\ 
	\mathfrak{L}_{(\omg_{0}) 0}'[\Tht] := & \lim_{\rho \to - \infty} \frac{r^{\min\set{\omg_{0}, 3}}}{-\rd_{\rho} r + \frac{f\ell}{r}} \left(-\rd_{\rho} (r \phi) + \frac{f \ell}{r} \phi + f r \dot{\phi} \right).
\end{aligned}
\end{equation}
\item $\e\neq 0$.
\item The following \emph{future admissibility condition} holds: There exist $\rho_1< \rho_2$ such that 
\begin{equation} \label{eq:adm-id-adm}
	\left(- \rd_{\rho} r + \frac{f \ell}{r}  \right) (\rho) < 0 \hbox{ for all } \rho \geq \rho_{1}, \quad \hbox{ and } \quad
	\left(\rd_{\rho} r + \frac{f \ell}{r}\right) (\rho) < 0 \hbox{ for all } \rho \leq \rho_{2}.
\end{equation}
\end{enumerate}
We denote the set of all $\omg_{0}$-admissible initial data sets by $\calA \calI \calD(\omg_{0})$.
\end{definition}

Next, we introduce a family of weighted-$C^{k}$-type distances $d_{k, \omg}$ on the class $\mathcal{AID}(\omg_{0})$.
\begin{definition} [Distances $d_{k,\omg}$ on $\calA \calI \calD(\omg_{0})$]\label{def:adm-top}
Given any positive integer $k$ and real numbers $\omg, \omg_{0} > 2$, we define the distance $d_{k, \omg}$ on the set $\calA \calI \calD(\omg_{0})$ of $\omg_{0}$-admissible initial data with two asymptotic ends (cf. Definition~\ref{def:adm-data}) as follows (we allow\footnote{This in particular happens when the initial data are not $k$-times differentiable.} $d_{k, \omg} (\Tht, \overline{\Tht}) = \infty$): 
\begin{equation} \label{eq:d-k-omg}
\begin{aligned}
	d_{k, \omg} (\Tht, \overline{\Tht}) := 
	& \nrm{\brk{\rho} \log (f/\fbg)(\rho)}_{C^{0}} 
	+ \sum_{i = 1}^{k} \left(  \nrm{\brk{\rho}^{1+i} \rd_{\rho}^{i} \log (f /\fbg)(\rho)}_{C^{0}} 
					+ \nrm{\brk{\rho}^{1+i} \rd_{\rho}^{i-1} (h - \hbg)(\rho)}_{C^{0}} \right) \\
	& + \nrm{\log^{-1}(1+\brk{\rho}) (r - \rbg) (\rho)}_{C^{0}} 
	+ \sum_{i=1}^{k} \left(  \nrm{\brk{\rho}^{i} \rd_{\rho}^{i} (r - \rbg) (\rho)}_{C^{0}} 
					+ \nrm{\brk{\rho}^{i} \rd_{\rho}^{i-1} (f\ell - \fbg \ellbg)(\rho)}_{C^{0}} \right) \\
	& + \nrm{\brk{\rho}^{\omg} (\phi - \phibg)(\rho)}_{C^{0}} 
	+ \sum_{i=1}^{k} \left(  \nrm{\brk{\rho}^{\omg+i} \rd_{\rho}^{i} (\phi - \phibg)}_{C^{0}} 
	+ \nrm{\brk{\rho}^{\omg+i} \rd_{\rho}^{i-1} (f \dot{\phi} - \fbg \dphibg)(\rho)}_{C^{0}} \right) \\
	& + |{\bf e}- \overline{\bf e}|												.
\end{aligned}\end{equation}
Here, $\brk{\rho} = (1 + \rho^{2})^{1/2}$ and $\rho_{\pm} = \max\set{0, \pm \rho}$.
\end{definition}
With the help of these distances, we define subclasses of $\mathcal{AID}(\omg_{0})$ with higher regularities and $r$-weights as follows.
\begin{definition}[$C^k_{\omg}$ initial data]\label{Ck.def}
For $k\in \mathbb N$ with $k\geq 2$ and $\omg, \omg_0>2$, we say that $\Theta\in \calA \calI \calD(\omg_{0})$ is $C^k_\omg$ if 
$d_{k,\omg}(\Theta,\Theta_{RN, M,\e})<\infty$ for some $\Theta_{RN, M,\e}$ which is an admissible smooth Cauchy initial data set for a fixed Reissner--Nordstr\"om solution with parameter $0<|\e|< M$ such that outside a compact set $[-R,R]$, $(r,f,h,\ell,\phi,\dot{\phi},\e)=(|\rho|,(1-\f{2M}{|\rho|}+\f{\e^2}{\rho^2})^{-\f 12},0,0,0,0,\e)$.
\end{definition}

\subsection{Preliminary results on the maximal globally hyperbolic future development}\label{sec:MGHFD.prelim}
By \cite{Kommemi, D3}, we have the following preliminary characterization of the future boundary of the maximal globally hyperbolic future development of an admissible Cauchy initial data set.
\begin{theorem} \label{thm:kommemi}
Let $(\calM, g, \phi, F)$ be the maximal globally hyperbolic future development of an admissible Cauchy initial data set (with arbitrary $\omg_{0}>2$, cf. Definition~\ref{def:adm-data}), and denote by $(\calQ = \calM / SO(3), g_{\calQ})$ the quotient Lorentzian manifold. Then the following statements hold:
\begin{enumerate}
\item $(\mathcal Q,g_{\mathcal Q})$ can be conformally embedded into a bounded subset of $\mathbb R^{1+1}$. 
\item Let $\mathcal Q^+$ be the closure of $\mathcal Q$ with respect to the topology induced by the conformal embedding described in part (1). Then the boundary\footnote{We abuse notation slightly to name the image of $\mathcal Q$ under the conformal embedding also as $\mathcal Q$. We will similarly do this for subsets of $\mathcal Q$, such as $\Sigma_{0}$.} of $\mathcal Q$ in $\mathcal Q^+$ has the following components: 
\begin{enumerate}
\item The initial hypersurface $\Sigma_{0}$.
\item Spatial infinities $i^0_1$ and $i^0_2$ which are the end-points of $\Sigma_{0}$ in $\mathcal Q^+$, with the convention that $i^0_1$ is the end-point with $\rho\to \infty$ and $i^0_2$ is the end-point with $\rho\to -\infty$.
\item Two connected components of null infinity, denoted by $\mathcal I^+_1$ and $\mathcal I^+_2$ respectively, each of which is an open null segment\footnote{The fact that it is open and that $r$ does not diverge to $\infty$ along $\EH_1$ and $\EH_2$ (see Definition~\ref{def.EH} below), follows from \cite{DafTrapped}.}, defined as the part of the boundary such that the $r$ diverges to $\infty$ along a transversal null curve towards $\mathcal I^+_1$ and $\NI_2$. 
\item Timelike infinities $i^+_1$ and $i^+_2$, which are defined to be future end-points of $\mathcal I^+_1$ and $\mathcal I^+_2$ respectively.
\item The Cauchy horizons\footnote{In the general setting of \cite{Kommemi}, $\CH_1$ and $\CH_2$ may be empty. Nevertheless, it is non-empty in our setting thanks to the result of Dafermos \cite{D2}, see also discussions in \cite{LO.interior}.} $\CH_1$ and $\CH_2$, which are defined to be half open\footnote{Both $\CH_1$ and $\CH_2$ are chosen to include their future endpoints. Therefore, in the case where $\mathcal S$ is empty (or contains only a single point of $\mathcal Q$), by our convention the bifurcation sphere is part of both $\CH_1$ and $\CH_2$.} null segments emanating from future null infinities $\mathcal I^+_1$ and $\mathcal I^+_2$ respectively such that the area-radius function $r$ extends continuously to $\CH_1\cup \CH_2$ and is strictly positive except possibly at the future endpoints of $\CH_1$ or $\CH_2$.
\item A (possibly empty) achronal set\footnote{In \cite{Kommemi}, Kommemi further distinguishes the sets for which $r$ extends to $0$ into null segments emanating from the endpoints of $\CH_1$ or $\CH_2$ and another piece which does not intersect any null rays emanating from future null infinity. We do not need such distinction here and will simply consider one achronal set $\mathcal S$ on which $r$ extends to $0$.} $\mathcal S$ which is defined to be the subset of the boundary on which $r$ extends continuously to $0$.
\end{enumerate}
Moreover, $\mathcal Q^+$ can be given by the Penrose diagram in Figure~\ref{fig:Kommemi}.
\end{enumerate}
\end{theorem}

\begin{figure}[h]
\begin{center}
\def\svgwidth{300px}
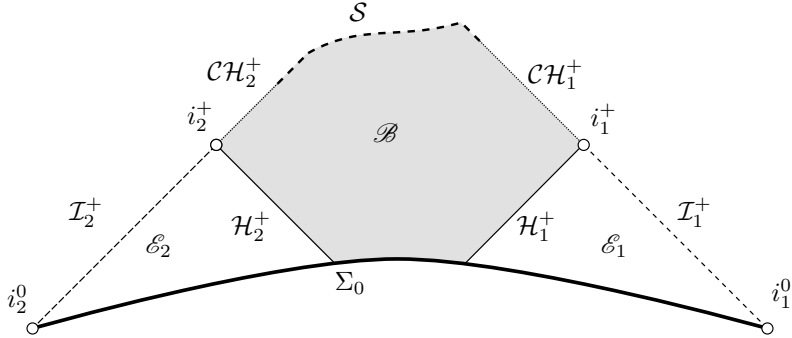 
\caption{Boundary of the maximal globally hyperbolic future development} \label{fig:Kommemi}
\end{center}
\end{figure}

Unless otherwise specified, from now on, as in Theorem~\ref{thm:kommemi}, we will use $(\calM, g, \phi, F)$ to denote the maximal globally hyperbolic future development of an admissible initial data set and $\calQ$ to denote the quotient manifold.

The above result in particular allows us to define the exterior region $\Ext = \Ext_{1} \cup \Ext_{2}$ of the spacetime, which is what we will focus on in the rest of the paper. Before that, we first define the event horizons $\EH_1$ and $\EH_2$. For the rest of this section, for concreteness, we use the double null coordinate system $(U, V)$ normalized as in Lemma~\ref{lem:cauchy-to-char}. We use the notation $\underline{V}(U)$ [resp. $\underline{U}(V)$] for the $V$- [resp. $U$-]coordinate of the point in the intersection $C_{U} \cap \Sgm_{0}$ [resp. $\uC_{V} \cap \Sgm_{0}$] (which is clearly unique).
\begin{definition}[Event horizons]\label{def.EH}
Given the maximal globally hyperbolic future development of an admissible Cauchy initial data set (with arbitrary $\omg_{0} > 2$, cf. Definition~\ref{def:adm-data}), define the \emph{event horizon} $\EH_1:=\{(U,V): U=U_{\EH_1},\, V \geq \underline{V}(U_{\EH_1}) \}$, where $U_{\EH_1}:=\sup \{U: \lim_{V \to \infty} r(U', V)=\infty \mbox{ for all }U'\leq U\}$, $(U_{\EH_1},\underline{V}(U_{\EH_1}))\in \Sigma_0$.

We also define the event horizon $\EH_2$ (and $V_{\EH_2}$) in a completely analogous manner, by switching the roles of $U$ and $V$.
\end{definition}

\begin{definition}[Exterior and interior regions]
Given the maximal globally hyperbolic future development of an admissible Cauchy initial data set (with arbitrary $\omg_{0}>2$, cf. Definition~\ref{def:adm-data}), define the \emph{exterior region} $\Ext$ to be $\Ext = \Ext_{1} \cup \Ext_{2}$, where
\begin{equation*}
\Ext_{1} = \set{(U, V) : U \leq U_{\EH_{1}}, \, V \geq \underline{V}(U)}, \quad
\Ext_{2} = \set{(U, V) : V \leq V_{\EH_{2}}, \, U \geq \underline{U}(V)}.
\end{equation*}
Note that each of the two sets $\Ext_{1}, \Ext_{2}$ is a connected component of the exterior region $\Ext$.
On the other hand, the black hole interior region $\Int$ is defined to be the complement of $\Ext$.
\end{definition}

\textbf{In the remainder of the paper, unless otherwise stated\footnote{Notice that since the event horizon is teleologically defined, i.e., it requires information about the maximal globally hyperbolic future development, when comparing two spacetimes and their initial data sets, one necessarily considers larger sets than the exterior regions themselves. We will return to this point in Section~\ref{sec:L-stability}, and point out carefully the domain that we study.}, we will focus on the connected component $\Ext_{1}$ of the exterior region. Because of this, we will also omit the subscript ${ }_1$ in $\EH_1$, $\NI_1$, $i^+_1$ and $i^0_1$.} By Proposition~\ref{prop.subextremality} below, this region is also characterized as the subset of the exterior region such that $\rd_V r (U,V)\geq 0$, $\rd_U r (U,V) < 0.$

We recall some simple geometric facts in the exterior region. We will only state the proposition in one connected component of the exterior region, a symmetric statement holds for the other connected component (with $u$ and $v$ interchanged). All of these statements (and their proofs) can be found in the appendix of \cite{LO.interior}. The subextremality statement (3) below was proven by Kommemi \cite{KomThe}. 
\begin{proposition}\label{prop.subextremality}
The following statements hold.
\begin{enumerate}
\item Let $(U,V) \in \Ext_{1} = \{(U,V) : U \leq U_{\EH},\, V \geq \underline{V}(U) \}$. Then
\begin{equation}\label{no.trapped.in.1}
\rd_V r (U,V)\geq 0,\quad 
\rd_U r (U, V) < 0.
\end{equation}
\item Define
$$r_{\EH}=\sup_{\EH} r,\quad \varpi_{\EH}=\sup_{\EH} \varpi.$$
Then 
\begin{equation}\label{r.varpi.limits}
r_{\EH}=\lim_{V \to\infty} r(U_{\EH}, V),\quad \varpi_{\EH}=\lim_{V \to\infty} \varpi(U_{\EH}, V).
\end{equation}
Moreover,
\begin{equation}\label{r.varpi.poly}
r_{\EH}=\varpi_{\EH} +\sqrt{\varpi_{\EH}^2-\e^2}.
\end{equation}
\item The following strict inequality holds:
\begin{equation}\label{eq:subextremality}
\varpi_{\EH}>|{\bf e}|.
\end{equation}
\end{enumerate}
\end{proposition}

\subsection{Definition of $\mathfrak L$ and $\mathfrak L_{(\omg_0)\infty}$.}\label{sec:L.Li}

Finally, in this subsection, we define the quantities $\mathfrak L$ and $\mathfrak{L}_{(\omg_0)\infty}$, which will play important roles in the rest of the paper; see for instance the statements of the theorems in Section~\ref{sec:main-thm}. The two quantities $\mathfrak L$ and $\mathfrak{L}_{(\omg_0)\infty}$ will be defined on $\NI_1$. Following the convention introduced in the previous subsection, we will omit the subscript ${ }_1$.

\begin{definition}[Definition of $\mathfrak{L}_{(\omg_{0}) \infty}$]\label{def.L}
Let $\omg_0>2$. Consider the maximal globally hyperbolic future development of an $\omg_0$-admissible initial data set. Introduce
\begin{align} \label{eq:L-def}
	\mathfrak{L} :=& \int_{\NI} 2 M(u) \Phi(u) \Gmm(u) \, \ud u,
\end{align}
where for each $u$ such that $\sup_{C_{u}} r = \infty$ (which is what we mean by $u \in \NI$), we define
\begin{equation*}
M(u) = \lim_{r \to \infty} \varpi(u, \cdot), \quad
\Phi(u) = \lim_{r \to \infty} r \phi(u, \cdot), \quad 
\Gmm(u) = \lim_{r \to \infty} \frac{\rd_{u} r}{1- \mu}(u, \cdot),
\end{equation*}
using $u$ and $r$ as the coordinate system\footnote{Such a choice of coordinates is justified by the Raychaudhuri equation \eqref{eq:EMSF-ray-orig}, which ensures that $\rd_{v} r \geq 0$ on a curve $C_{u}$ such that $\sup_{C_{u}} r = \infty$.}.
Define
\begin{equation} \label{eq:Linfty-def}
\mathfrak{L}_{(\omg_{0}) \infty} := \begin{cases}\mathfrak{L}_{(\omg_{0}) 0} + \mathfrak{L}\quad &\mbox{if }\omg_{0}\geq 3\\
\mathfrak{L}_{(\omg_{0}) 0}\quad &\mbox{if } \omg_{0} \in (2,3)
\end{cases}
,
\end{equation}
where $\mathfrak{L}_{(\omg_{0}) 0} = \mathfrak{L}_{(\omg_{0}) 0}[\Tht]$ depends only on the initial data and is defined as in \eqref{eq:adm-id-limits}. 
\end{definition}

Corresponding to $\Ext_2$, we analogously define $\mathfrak{L}'$ and $\mathfrak{L}_{(\omg_{0}) \infty}'$ (switching the roles of $u$ and $v$). 

\begin{remark} 
Of course, at this point it is not a priori clear whether $\mathfrak{L}$ and $\mathfrak{L}_{(\omg_{0}) \infty}$ (as well as $\mathfrak{L}'$ and $\mathfrak{L}_{(\omg_{0}) \infty}'$ on the other end) are well-defined. In our setting, these definitions are justified by the Dafermos--Rodnianski Price's law theorem (cf. Section~\ref{sec:bg}).
\end{remark}

\section{Precise statements of the main theorems}  \label{sec:main-thm}

We state in this section the main theorems that are to be proven in this paper. All of these theorems are stated in \cite{LO.interior} and were used to obtain a proof of strong cosmic censorship.

As mentioned in the previous section, we will focus our attention on the connected component $\Ext_{1}$ of the exterior region such that $\rd_V r>0$ and $\rd_U r<0$. All the results in this section of course hold in the other connected component of the exterior region after the following replacements: $\mathfrak{L}_{(\omg_{0}) 0}\leftrightarrow \mathfrak{L}_{(\omg_{0}) 0}'$, $\mathfrak{L}\leftrightarrow \mathfrak{L}'$, $\mathfrak{L}_{(\omg_{0}) \infty}\leftrightarrow \mathfrak{L}_{(\omg_{0}) \infty}'$, $U \leftrightarrow V$, $u\leftrightarrow v$, $\rho \leftrightarrow -\rho$.

Our first main theorem states that nonvanishing of $\mathfrak{L}_{(\omg_{0}) \infty}$ implies some integrated lower bound for the incoming radiation $\rd_{v} \phi$ along $\EH$. 

\begin{theorem} \label{thm:blowup}
For $\omg_{0} > 2$, let $\Tht = (r, f, h, \ell, \phi, \dot{\phi}, \e)$ be an $\omg_{0}$-admissible data set, and let $(\calM, g, \phi, F)$ be the corresponding maximal globally hyperbolic future development. Suppose that
\begin{equation*}
	\mathfrak{L}_{(\omg_{0}) \infty} \neq 0.
\end{equation*}
Then for an advanced null coordinate $v$ such that
\begin{equation*}
	C^{-1} < \inf_{\EH} \frac{\rd_{v} r}{1-\mu} \leq \sup_{\EH} \frac{\rd_{v} r}{1-\mu} < C
\end{equation*}
for some $C > 0$, we have
\begin{equation*}
	\int_{\EH} v^{\alp} (\rd_{v} \phi)^{2} \, \ud v = \infty
\end{equation*}
for every $\alp > \min\{ 2 \omg_{0} + 1 , 7 \}$. 
\end{theorem}

The proof of this theorem will be given in Section~\ref{sec:blowup}.

Our next theorem asserts stability of the quantity $\mathfrak{L}$, which is the dynamically defined part of $\mathfrak{L}_{(\omg_{0}) \infty}$ in the case $\omg_{0} \geq 3$ (observe that $\mathfrak{L}_{(\omg_{0}) 0}$ is determined by the initial data). 
\begin{theorem} \label{thm:L-stability}
Fix $\omg_0>2$. Let $\Tht = (r, f, h, \ell, \phi, \dot{\phi}, \e)$ and $\overline{\Tht} = (\rbg, \fbg, \hbg, \ellbg, \phibg, \dphibg, \ebg)$ be $\omg_0$-admissible data sets (cf. Definition~\ref{def:adm-data}) such that $d_{1, \omg_0}^{+}(\Tht, \overline{\Tht}) < \eps$, where\footnote{Note that $d^{+}_{1, \omg}$ is essentially $d_{1,\omg}$, except that it has weights only near one asymptotically flat end; see Definition~\ref{def:adm-top}.}
\begin{equation*}
\begin{aligned}
	d_{1, \omg_0}^{+} (\Tht, \overline{\Tht}) := 
	& \nrm{\brk{\rho_{+}} \log (f / \fbg)(\rho)}_{C^{0}} 
	+ \nrm{\brk{\rho_{+}}^{2} \rd_{\rho} \log (f / \fbg)(\rho)}_{C^{0}} 
	+ \nrm{\brk{\rho_{+}}^{2}  (h - \hbg) (\rho)}_{C^{0}} \\
	& + \nrm{\log^{-1}(1+\brk{\rho_{+}}) (r - \rbg)(\rho)}_{C^{0}} 
	+ \nrm{\brk{\rho_{+}} \rd_{\rho} (r - \rbg)(\rho)}_{C^{0}} 
	+ \nrm{(f \ell - \fbg \ellbg)(\rho)}_{C^{0}}  \\
	& + \nrm{\brk{\rho_{+}}^{\omg_0} (\phi - \phibg)(\rho)}_{C^{0}} 
	+ \nrm{\brk{\rho_{+}}^{\omg_0+1} \rd_{\rho} (\phi- \phibg)(\rho)}_{C^{0}} 
	+ \nrm{\brk{\rho_{+}}^{\omg_0+1} (f \dot{\phi} - \fbg \dphibg)(\rho)}_{C^{0}}  + |{\bf e} - \ebg|.  
\end{aligned}
\end{equation*}
Here, $\brk{\rho_{+}} := (1 + \rho_{+}^{2})^{1/2}$ and $\rho_{+} := \max\set{0, \rho}$.

Then, for $\mathfrak L:=\mathfrak L[\Tht]$ and $\overline{\mathfrak L}:=\mathfrak L[\overline{\Tht}]$, there exists a constant $C_{\overline{\Tht}}$, which depends only on $\overline{\Tht}$, such that
\begin{equation*}
	\Abs{\mathfrak{L} - \overline{\mathfrak{L}}} \leq C_{\overline{\Tht}} \, \eps.
\end{equation*}
\end{theorem}
Section~\ref{sec:L-stability} will be devoted to the proof of Theorem~\ref{thm:L-stability}. 

Our final theorem concerns instability of the condition $\mathfrak{L}_{(\omg_{0}) \infty} = 0$. 
\begin{theorem} \label{thm:instability}
For $\omg_{0} \geq 3$, let $\overline{\Tht} = (\rbg, \fbg, \hbg, \ellbg, \phibg, \dphibg, \ebg)$ be an $\omg_{0}$-admissible data set. 
Suppose that
\begin{equation*}
	\mathfrak{L}_{(\omg_{0}) \infty} [\overline{\Tht}] = 0.
\end{equation*}
Then for some $\eps_{\ast} = \eps_{\ast}(\overline{\Tht}) > 0$, there exists a one-parameter family $(\Tht_{\eps})_{\eps \in (-\eps_{\ast}, \eps_{\ast})}$ of $\omg_{0}$-admissible initial data sets such that
\begin{itemize}
\item $\Tht_{0} = \overline{\Tht}$, 
\item $\mathfrak{L}_{(\omg_{0}), \infty}[\Tht_{\eps}] \neq 0$ for all $\eps \in (-\eps_{\ast}, \eps_{\ast})\setminus\{0\}$,
\item if $\overline{\Theta}\in C^k_{\omg_0}$ for $k\in\mathbb N$, $k\geq 2$, then $\eps \mapsto \Tht_{\eps}$ is continuous with respect to $d_{k, \omg}$ for all $\omg > 2$, 
\item $\mathfrak{L}_{(\omg_{0}) 0}[\Tht_{\eps}] = \mathfrak{L}_{(\omg_{0}) 0}[\overline{\Tht}]$ and $\mathfrak{L}_{(\omg_{0}) 0}'[\Tht_{\eps}] = \mathfrak{L}_{(\omg_{0}) 0}'[\overline{\Tht}]$ for all $\eps \in (-\eps_{\ast}, \eps_{\ast})$.
\end{itemize}
In fact, there exist $\bar{\rho}_{2} > \bar{\rho}_{1} \gg 1$ such that 
\begin{itemize}
\item $\Tht_{\eps} = \overline{\Tht}$ in $\set{\rho \in \Sgm_{0} : \rho < \bar{\rho}_{1}}$ for all $\eps \in (-\eps_{\ast}, \eps_{\ast})$,
\item denoting $\Tht_\eps = (r_\eps, f_\eps, h_\eps, \ell_\eps, \phi_\eps, \dot{\phi}_\eps, \e_\eps)$, it holds that
$$\overline{\phi}=\phi_\eps,\quad\fbg\overline{\dot{\phi}}=f_\eps\dot{\phi}_\eps $$
in $\set{\rho \in \Sgm_{0} : \rho > \bar{\rho}_{2}}$ for all $\eps \in (-\eps_{\ast}, \eps_{\ast})$.
\end{itemize}
\end{theorem}
This theorem will be proved in Section~\ref{sec:instability}.

\section{Consequences of Price's law decay} \label{sec:bg}

In this section, we derive decay estimates in the exterior region of the maximal globally hyperbolic future development of an \emph{arbitrary} admissible Cauchy data set. The estimates are based on the remarkable theorem of Dafermos--Rodnianski \cite{DRPL} on Price's law, which provides quantitative decay rates (towards timelike infinity) of the scalar field in the maximal globally hyperbolic future development of the characteristic initial value problem. In order to apply the Dafermos--Rodnianski theorem, we need to first prove decay estimates near spacelike infinity to reduce a Cauchy problem to a characteristic initial value problem. After stating the results in \cite{DRPL}, we also derive various ramifications of Price's law, including in particular decay estimates for the geometric quantities (in an appropriate gauge) and for higher derivatives of the scalar field, which will play a fundamental role in the remainder of the paper. In the process, we introduce the null coordinates $U, V$ (referred to as \emph{initial-data-normalized}) and $u, v$ (referred to as \emph{future-normalized}), and discuss the different purposes they serve in the subsequent analysis.

This section is structured as follows:
\begin{itemize}
\item {\bf Section~\ref{subsec:bg-coords}.} We begin with the introduction of the $(U,V)$, $(u,V)$ and $(u,v)$ coordinate systems. $U$ and $V$ are normalized with respect to the (Cauchy) initial data, whereas $u$ and $v$ are normalized with respect to the future null infinity and the event horizon respectively.
\item {\bf Section~\ref{subsec:bg-large-r}.} We prove some bounds in a neighborhood of spacelike infinity that are necessary to justify the application of the results in \cite{DRPL} in our setting. We work in the \emph{initial-data-normalized null coordinates} $(U, V)$.
\item {\bf Section~\ref{subsec:DR-full}.} We cite the theorem of Dafermos--Rodnianski (Theorem~\ref{thm:DR-full}) and discuss some immediate consequences of it (Corollaries~\ref{cor:DR-small-r}, \ref{cor:DR-large-r} and \ref{cor:DR-final}). To measure the decay towards timelike infinity in a large $r$ region, the \emph{future-normalized (retarded) null coordinate} $u$ is used.
\item {\bf Section~\ref{subsec:bg-uv}.} We prove bounds on fundamental geometric quantities using Price's law (Proposition~\ref{prop:bg-uv}). In order that the geometric quantities decay, we will use the \emph{future-normalized (retarded and advanced) null coordinates} $(u,v)$.
\end{itemize}

\subsection{Initial-data- and future-normalized null coordinates} \label{subsec:bg-coords}
Let\footnote{{Notice that all results obviously still hold for $\omg_{0}>3$ \underline{as long as} we replace $\omg_0$ by $\min\{\omg_0,3\}$ appropriately in the decay estimates.}} $2 < \omg_{0} \leq 3$. \textbf{In this section, we fix an $\omg_{0}$-admissible Cauchy data set $\Tht$ on $\Sgm_{0}$}. We consider its maximal globally hyperbolic future development $(\calM, g, \phi, F)$, as well as the projection of this development onto $\mathcal Q$. Let us emphasize that in this section, $\omg_0$ will be considered fixed; we will also use the parameters $\omg$ and $\omg'$, which will be compared to the fixed $\omg_0$.

In the remainder of this short subsection, we introduce and discuss various null coordinates used throughout this paper. 

\subsubsection*{Initial-data-normalized coordinates}
We say that a double null coordinate system $(U, V)$ on $\PD$ is \emph{initial-data-normalized} if the following normalization conditions hold on $\Sgm_{0}$:
\begin{equation} \label{eq:coord-UV}
	\frac{\ud}{\ud \rho} U \restriction_{\Sgm_{0}} = - 1, \quad 
	\frac{\ud}{\ud \rho} V \restriction_{\Sgm_{0}} = 1.
\end{equation}
These are precisely the null coordinates considered in Lemma~\ref{lem:cauchy-to-char}. 

Henceforth, we abuse the notation a bit and consider $\PD$ and $\Sgm_{0}$ as subsets of $(U, V) \in \bbR^{1+1}$. Given $(U, V) \in \PD$, let $\underline{U}(V)$ and $\underline{V}(U)$ be defined by $(\underline{U}(V), V) \in \Sgm_{0}$ and $(U, \underline{V}(U)) \in \Sgm_{0}$, respectively (i.e., $C_{U} \cap \Sgm_{0} = \set{(U, \underline{V}(U))}$ and $\uC_{V} \cap \Sgm_{0} = \set{(\underline{U}(V), V)}$, where $C_U$, $\uC_{V}$ are as in Definition~\ref{def:C.uC}). We also introduce the notation $\underline{\ups} (U) = r(U, \underline{V}(U))$ (i.e., $\underline{\ups}(U)$ is the value of $r$ at the point $C_{U} \cap \Sgm_{0}$). 

Although easy to define globally on $\PD$, the initial-data-normalized coordinates $(U, V)$ are ill-suited for describing the asymptotic behavior of solutions near timelike infinity (which exists for each asymptotically flat end). For this reason, when analyzing a single solution, it is often more advantageous to use future-normalized coordinates (introduced just below) adapted to the particular asymptotically flat end. Nevertheless, the initial-data-normalized coordinates are indispensible when we consider the difference of two solutions, since the difference of data (which is coordinate dependent!) is originally given in these coordinates. 

\subsubsection*{Future-normalized coordinates}
By Theorem~\ref{thm:kommemi}, there exists a connected component of the exterior region  corresponding to each asymptotically flat end, bounded to the future by null infinity $\NI$ and an event horizon $\EH$. As mentioned before, we focus on the end where $\rho \to \infty$ along $\Sgm_{0}$. In this context, we introduce the \emph{future-normalized retarded null coordinate} $u = u(U)$, so that the following condition holds on $\NI$:
\begin{equation}  \label{eq:coord-u-bg-0}
	\lim_{V \to \infty} \frac{-\rd_{u} r}{1-\mu}(u, V) = 1. 
\end{equation}
As a consequence of future completeness of null infinity $\NI$ (see \cite[\S~5.4]{DafTrapped} for the precise definition) and the normalization \eqref{eq:coord-u-bg-0}, we have $u_{\EH} = \lim_{U \to U_{\EH}^{-}} u(U) = \infty$, i.e., the event horizon $\EH$ now lies in the ideal curve $\set{(u, V) : u = \infty}$; we omit the simple proof of this statement. In particular, the $(u,V)$ coordinate system only covers the (connected component of) the exterior region. Moreover, null infinity corresponds to the ideal curve $\set{(u, V) : V = \infty}$.

The coordinate system $(u, V)$ is suitable for describing the decay of the scalar field near timelike infinity. Indeed, it is essentially the coordinate system used in Dafermos--Rodnanski \cite{DRPL}; see Theorem~\ref{thm:DR-full} below.

Finally, we also introduce the \emph{future-normalized advanced null coordinate} $v = v(V)$, so that the following condition holds on $\EH$:
\begin{equation} \label{eq:coord-v-bg-0}
	\frac{\rd_{v} r}{1-\mu}(u=\infty, v) = 1.
\end{equation}
Here, in addition to $u$, which is adapted to null infinity $\NI$, the advanced null coordinate $v$ is adapted to the event horizon $\EH$. When restricted to an exact Reissner--Nordstr\"om solution, note that $(u, v)$ coincides with the Eddington--Finkelstein coordinates. Accordingly, in  Proposition~\ref{prop:bg-uv} below, we show that the geometric quantities in the coordinate system $(u, v)$ asymptotes (in a suitable sense) to the corresponding quantities in the Reissner--Nordstr\"om spacetime with the parameters $(\varpi_{\EH}, \e)$.

When working in the coordinate systems $(u, V)$ or $(u, v)$, we again abuse notation and regard the exterior of the black hole in $\PD$ and $\Sgm_{0}$ as subsets of $(u, V) \in \bbR^{1+1}$ or $(u, v) \in \bbR^{1+1}$, respectively. In the coordinates $(u, V)$, we define $\underline{u}(V)$, $\underline{V}(u)$ and $\underline{\ups}(u)$ similarly as before.

\begin{remark}  \label{rem:coord-trans}
In each of the coordinate systems $(U, V)$, $(u, V)$ and $(u, v)$, note that there remains the freedom of translating the coordinates, i.e., $(U, V) \mapsto (U + U_{0}, V  +V_{0})$ etc.
\end{remark}

\subsection{Preliminary bounds near spacelike infinity} \label{subsec:bg-large-r}

From this subsection onwards, we will prove decay estimates in the exterior region of a fixed solution. To achieve this, we will in particular apply \cite{DRPL}. For this purpose, it is convenient to introduce the following quantities: the ADM mass, which is useful as it gives a priori control of the modified mass $\varpi$, and $\psi$, which is more convenient for deriving decay estimates near (spatial and null) infinity.

\begin{definition}\label{def.ADM.psi}
\begin{enumerate}
\item (Definition of ADM mass) Define the ADM mass as the following limit on the initial hypersurface:
$$\varpi_i:=\lim_{\rho\to \infty} \varpi\restriction_{\Sigma_0}(\rho).$$
Although it will not be explicitly used in this paper, we also define the ADM mass on the other asymptotically flat end by $\varpi_i':=\lim_{\rho\to -\infty} \varpi\restriction_{\Sigma_0}(\rho).$
\item (Definition of $\psi$) Define $\psi:=r\phi$.
\end{enumerate}
\end{definition}

\begin{remark}[Finiteness of $\varpi_i$]
By \eqref{eq:adm-id-af}, $\varpi_i$ is a well-defined limit and is finite. It can moreover be shown that $\varpi_i>0$.
\end{remark}

In preparation for the statement of the theorem of Dafermos--Rodnianski, we first prove some bounds for $r$, $\psi$ and their derivatives, which hold in a neighborhood of spacelike infinity, see Figure~\ref{fig:DR-full}. In this proposition, we use the initial-data normalized coordinates $(U,V)$ for concreteness, but we note that the estimates are in fact coordinate-invariant. In what follows, we will suppress the explicit dependence on $(U, V)$, $(u,V)$, etc. (e.g. in LHS of \eqref{eq:bg-large-r:rphi}, \eqref{eq:bg-large-r:dvrphi} and similarly Theorem~\ref{thm:DR-full}, Corollary~\ref{cor:DR-small-r}, \ref{cor:DR-large-r}, \ref{cor:DR-final} and Proposition~\ref{prop:bg-uv}), if there is no danger of confusion.
\begin{proposition}[Estimates near spatial infinity] \label{prop:bg-large-r}
Given $\eta_{i^{0}} > 0$, we define
\begin{equation*}
	\Sgm_{0, i^{0}} = \set{\rho \in \Sgm_{0} : \forall \rho' \geq \rho, \ r \restriction_{\Sgm_{0}}(\rho') \geq \eta_{i^{0}}^{-1} \max \set{\varpi_{i}, \e}, \ \rd_{U} r \restriction_{\Sgm_{0}}(\rho') < 0, \ \ \rd_{V} r \restriction_{\Sgm_{0}}(\rho') > 0}.
\end{equation*}
Let $\underline{\calN}$ be the future domain of dependence of $\Sgm_{0, i^{0}}$ in $\PD$. If $\eta_{i^{0}}$ is sufficiently small depending on $\omg_{0}-2$ and $\Tht$, then there exists $B_{i^{0}} > 0$, which depends on $\Tht \restriction_{\Sgm_{0, i^{0}}}$, such that the following bounds hold on $\underline{\calN}$:
\begin{align} 
	\frac{1}{2} \leq \frac{\rd_{V} r(U, V)}{\rd_{V} r(\underline{U}(V), V)} \leq & 2, \label{eq:bg-large-r:dvr}  \\
	\frac{1}{2} \leq \frac{-\rd_{U} r(U, V)}{-\rd_{U} r(U, \underline{V}(U))} \leq & 2, \label{eq:bg-large-r:dur}  \\
\Abs{\psi}
	\leq & \underline{\ups}^{-\omg_{0}+1} B_{i^{0}} ,			\label{eq:bg-large-r:rphi} \\
\Abs{\frac{1}{\rd_{V} r} \rd_{V} \psi }
	\leq & r^{-\omg_{0}} B_{i^{0}} .			\label{eq:bg-large-r:dvrphi}
\end{align}
Moreover, for $(U, \underline{V}(U)) \in \Sgm_{0} \cap \underline{\calN}$, we have
\begin{equation} \label{eq:bg-large-r:duups} 
	B_{i^{0}}^{-1} \leq \frac{-\rd_{U} \underline{\ups}(U)}{-\rd_{U} r(U, \underline{V}(U))} \leq B_{i^{0}}.
\end{equation}
\end{proposition}

After the proof, {\bf we fix $\eta_{i^{0}} > 0$ and thus $\underline{\calN}$ so that Proposition~\ref{prop:bg-large-r} applies.}

\begin{proof}
In this proof, we abbreviate $R = \eta_{i^{0}}^{-1} \max \set{\varpi_{i}, \e}$.
Note that all (maximally-extended) past-directed null curves in $\underline{\calN}$ intersect $\Sgm_{0}$, i.e., if $(U, V) \in \underline{\calN}$, then $\set{(U', V) \in \PD : \underline{U}(V) \leq U' \leq U}, \set{(U, V') \in \PD : \underline{V}(U) \leq V' \leq V} \subseteq \underline{\calN}$. Moreover, by the monotonicity of $\varpi$ and $r$, in $\underline{\calN}$ we have
\begin{equation} \label{eq:bg-large-r:pf-me}
	\frac{\sup_{\underline{\calN}} \varpi}{r} + \frac{\abs{\e}}{r} \leq \eta_{i^{0}}.
\end{equation}
As a consequence, in $\underline{\calN}$ we have
\begin{equation} \label{eq:bg-large-r:pf-mu} 
	1 - 2 \eta_{i^{0}}
	\leq
	1 - \mu 
	\leq 1 + \eta_{i^{0}}^{2} . 
\end{equation}

\pfstep{Step~1: Bounds for $\rd_{V} r$, $\rd_{U} r$ and $\rd_{U} \underline{\ups}$}
Note that $\log \rd_{V} r$ obeys the equation
\begin{equation*}
	\rd_{U} \log \rd_{V} r = \frac{\rd_{U} \rd_{V} r}{\rd_{V} r} = \frac{2 (\varpi - \frac{\e^{2}}{r})}{r^{2}} \frac{\rd_{U} r}{1-\mu}.
\end{equation*}
By \eqref{eq:bg-large-r:pf-mu}, we have
\begin{equation} \label{eq:bg-large-r:pf0}
	\Abs{\frac{2 (\varpi - \frac{\e^{2}}{r})}{r^{2}} \frac{\rd_{U} r}{1-\mu}}
	\leq 2 (1-2\eta_{\NI})^{-1} \frac{\varpi_{i}}{r^{2}} (-\rd_{U} r).
\end{equation}
Integrating this bound over $[\underline{U}(V), U]$ (where we make the change of variables $U \mapsto r$, $\rd_{U} r \ud U \mapsto \ud r$), applying \eqref{eq:bg-large-r:pf-me} and taking $\eta_{i^{0}}$ sufficiently small, \eqref{eq:bg-large-r:dvr} follows. The bound \eqref{eq:bg-large-r:dur} follows in a similar manner; we omit the details. Finally, for \eqref{eq:bg-large-r:duups}, observe that (using Lemma~\ref{lem:cauchy-to-char})
\begin{align*}
\frac{- \rd_{U} r(U, \underline{V}(U))}{-\rd_{U} r(U, \underline{V}(U)) + \rd_{V} r(U, \underline{V}(U)) (-\frac{\ud }{\ud U} \underline{V})(U)} 
= \frac{\frac{1}{2} \rd_{\rho} r - \frac{f}{2 r} \ell}{\frac{1}{2} \rd_{\rho} r - \frac{f}{2 r} \ell + \frac{1}{2} \rd_{\rho} r + \frac{f}{2 r} \ell} 
= \frac{\rd_{\rho} r - \frac{f}{r} \ell}{2 \rd_{\rho} r} (\rho(U)),
\end{align*}
where $\rho(U)$ denotes the inverse of $\rho \mapsto U(\rho)$. By hypothesis, there exists some $B \geq 1$ depending on $\Tht \restriction_{\Sgm_{0, i^{0}}}$ such that
\begin{equation} \label{eq:bg-large-r:pf:duups}
B^{-1} \leq \Abs{\frac{-\rd_{U} r(U, \underline{V}(U))}{-\rd_{U} \underline{\ups}(U)}} \leq B,
\end{equation}
which proves \eqref{eq:bg-large-r:duups}.

\pfstep{Step~2: Bounds for $\psi$ and $\rd_{V} \psi$}
To prove \eqref{eq:bg-large-r:rphi} and \eqref{eq:bg-large-r:dvrphi}, we proceed by a bootstrap argument. Suppose that 
\begin{equation} \label{eq:bg-large-r:rphi-btstrp}
	\abs{\psi} \leq 2 \underline{\ups}^{-\omg_{0}+1} B_{i^{0}} 
\end{equation}
holds on $\underline{\calN} \cap \set{U \leq U_{f}}$ for some $B_{i^{0}}$ and $U_{f}$. We claim that, for sufficiently large $B_{i^{0}}$ and small enough $\eta_{i^{0}}$, \eqref{eq:bg-large-r:rphi} (which improves \eqref{eq:bg-large-r:rphi-btstrp}) and \eqref{eq:bg-large-r:dvrphi} hold on $\underline{\calN} \cap \set{U \leq U_{f}}$. Then by a standard continuous induction argument,  \eqref{eq:bg-large-r:rphi} and \eqref{eq:bg-large-r:dvrphi} on the whole region $\underline{\calN}$ would follow.

To prove the claim, we use the equation
\begin{equation} \label{eq:dUdvpsi}
	\rd_{U} \left(\frac{1}{\rd_{V} r} \rd_{V}  \psi \right) 
	= - \frac{\rd_{U} \rd_{V} r}{\rd_{V} r} \frac{1}{\rd_{V} r} \rd_{V} \psi + \frac{\rd_{U} \rd_{V} r}{r \rd_{V} r} \psi.
\end{equation}
By \eqref{eq:bg-large-r:pf0}, $\int_{\underline{U}(V)}^{U} \frac{\rd_{U} \rd_{V} r}{\rd_{V} r} (U', V) \, \ud U' \leq C$ for some universal $C>0$. Thus by Gr\"onwall's inequality, we obtain
\begin{equation} \label{eq:bg-large-r:pf1}
\begin{aligned}
	r^{\omg_{0}} \Abs{\frac{1}{\rd_{V} r} \rd_{V} \psi} (U, V)
	\leq & C  r^{\omg_{0}}(U, V) \Abs{\frac{1}{\rd_{V} r} \rd_{V} \psi} (\underline{U}(V), V)  \\
		& +  C B_{i^{0}}  \frac{\varpi_{i}}{r^{3-\omg_{0}}(U, V)} \sup_{U' \in [\underline{U}(V), U]} \Abs{\frac{-\rd_{U} r(U', V)}{-\rd_{U} \underline{\ups}(U')}}  \int_{\underline{U}(V)}^{U} \frac{-\rd_{U} \underline{\ups}(U')}{\underline{\ups}^{\omg_{0}-1}(U')} \, \ud U'.
\end{aligned}
\end{equation}
Since $r(U, V) \leq r(\underline{U}(V), V)$, the first term on the RHS is bounded by some constant $B$, which depends on $\Tht \restriction_{\Sgm_{0, i^{0}}}$. For the second term, by \eqref{eq:bg-large-r:dur} and \eqref{eq:bg-large-r:pf:duups} we have
\begin{equation*}
\Abs{\frac{-\rd_{U} r(U, V)}{-\rd_{U} \underline{\ups}(U)}} \leq 2 B,
\end{equation*}
where we enlarge $B$ (depending on $\Tht \restriction_{\Sgm_{0, i^{0}}}$) if necessary. Since $R = \underline{\ups}(U) \leq r(U, V)$, the second term on the RHS of \eqref{eq:bg-large-r:pf1} is bounded by
\begin{equation*}
	\frac{CB}{\omg_{0}-2} \frac{\varpi_{i}}{R} B_{i^{0}} \leq\frac{CB}{\omg_{0}-2} \eta_{i^{0}} B_{i^{0}}.
\end{equation*}
As a consequence, we obtain
\begin{equation*}
	\Abs{\rd_{V} \psi} \leq r^{-\omg_{0}} \rd_{V} r \left(B + \frac{C B}{\omg_{0}-2} \eta_{i^{0}} B_{i^{0}} \right).
\end{equation*}
Hence \eqref{eq:bg-large-r:dvrphi} holds in the bootstrap domain $\underline{\calN} \cap \set{U \leq U_{f}}$.
Integrating $\rd_{V} \psi$ along $C_{U}$, we also obtain
\begin{equation*}
	\abs{\psi} (U, V) \leq \abs{\psi}(U, \underline{V}(U)) + \underline{\ups}^{-\omg_{0}+1} \left(B + \frac{C B}{\omg_{0}-2} \eta_{i^{0}} B_{i^{0}} \right) \leq \underline{\ups}^{-\omg_{0}+1} \left(2 B + \frac{C B}{\omg_{0}-2} \eta_{i^{0}} B_{i^{0}} \right)
\end{equation*}
by enlarging $B$ depending on $\Tht \restriction_{\Sgm_{0, i^{0}}}$ if necessary. Choosing $B_{i^{0}} \geq 4B$ and $\eta_{i^{0}}$ small enough, we obtain \eqref{eq:bg-large-r:rphi}. This closes the bootstrap, and thus completes the proof of \eqref{eq:bg-large-r:rphi} and \eqref{eq:bg-large-r:dvrphi}.  \qedhere
\end{proof}

\subsection{Price's law decay: A theorem of Dafermos--Rodnianski and its corollaries} \label{subsec:DR-full}
In the coordinates $(u, V)$, consider a characteristic rectangle $\calX_{0} = \set{(u, V) : u \geq u_{\calX_{0}}, \ V \geq V_{\calX_{0}}}$ with past null boundaries $\uC_{in} = \uC_{V_{\calX_{0}}} \cap \calX_{0}$ (incoming) and $C_{out} = C_{u_{\calX_{0}}} \cap \calX_{0}$ (outgoing), such that $C_{out} \subseteq \underline{\calN}$, see Figure~\ref{fig:DR-full}. Note that $\calX_0$ is by definition a subset of the connected component $\Ext_{1}$ of the exterior region with $\rd_V r \geq 0$ and $\rd_U r<0$. By translating the coordinates (see Remark~\ref{rem:coord-trans}), we may assume that $u_{\calX_{0}} = 1$ and $V_{\calX_{0}} = 1$.

\begin{figure}[h]
\begin{center}
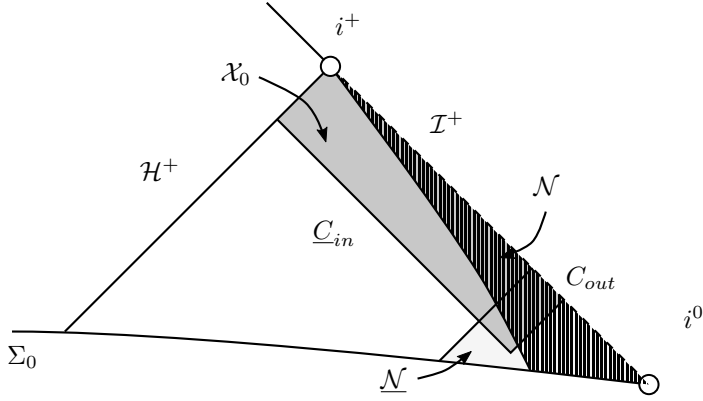 
\caption{The characteristic rectangle $\calX_{0}$ in Theorem~\ref{thm:DR-full}} \label{fig:DR-full}
\end{center}
\end{figure}

Thanks to Proposition~\ref{prop:bg-large-r}, we can apply the results of Dafermos--Rodnianski \cite{DRPL} to the solution on $\calX_{0}$. Recall from the beginning of the section that we have taken $2<\omg_0\leq 3$: note that for the theorem below, the decay rate does not improve even if $\omg_0$ is taken to be larger than $3$. We also remind the reader that, by future completeness of $\NI$ and Proposition~\ref{prop:bg-large-r}, $u$ takes values in the full real line $(-\infty, \infty)$.
\begin{theorem}[Dafermos--Rodnianski Price's law decay theorem] \label{thm:DR-full}
On the characteristic rectangle $\calX_{0} \subseteq \PD$  in the coordinates $(u, V)$ defined as above, the following conclusions hold.
\begin{enumerate}
\item There exists $B_{0} > 0$, which depends on $\Tht$ and the solution\footnote{Of course the solution in $\calX_{0}$ ultimately depends only on $\Tht$. We use this language to emphasize that $B_{0}$ depends \underline{not} only on the norms of $\Tht$, but depends on its particular profile.} in $\calX_{0}$, such that
\begin{align} 
	B_{0}^{-1} \leq \frac{\rd_{V} r}{1-\mu}  \leq & B_{0} \quad \hbox{ in } \calX_{0}. \label{eq:DR:kpp}
\end{align}
\item For any $\eta > 0$, there exist constants $c_{\omg_{0} - \eta}, B_{\omg_{0} - \eta} > 0$, which depend on $\omg_{0} - \eta$, $\Tht$ and the solution in $\calX_{0}$, such that the following bounds hold for $(u, V) \in \calX_{0} \cap \set{r \leq c_{\omg_{0} - \eta} V}$: 
\begin{align}
	\abs{\phi} \leq & \left( V^{-(\omg_{0} - \eta)} r^{-1} + V^{-\omg_{0}} \right) B_{\omg_{0} - \eta}  , \label{eq:DR:phi-med-r} \\
	r \abs{\rd_{V} \phi} \leq & \left( V^{-(\omg_{0} - \eta)} r^{-1} + V^{-\omg_{0}} \right) B_{\omg_{0} - \eta}, \label{eq:DR:dvphi-med-r}\\
	\abs{\rd_{V} (r \phi)} \leq & \left( V^{-(\omg_{0} - \eta)} r^{-1} + V^{-\omg_{0}} \right) B_{\omg_{0} - \eta}. \label{eq:DR:dvrphi-med-r}
\end{align}
\item Moreover, there exists a constant $B_{\eta} > 0$, which depends on $\eta$, $\Tht$ and the solution in $\calX_{0}$, such that the following bounds hold for $(u, V) \in \calX_{0} \cap \set{r \geq c_{\omg_{0} - \eta} V}$:
\begin{align}
	r \abs{\rd_{V} \phi} \leq &  V^{-1} u^{-(\omg_{0} - \eta) + 1} B_{\eta}, 	\label{eq:DR:dvphi-huge-r} \\
	r \abs{\rd_{u} \phi} \leq & u^{-(\omg_{0} - \eta)} B_{\eta} ,		\label{eq:DR:duphi-huge-r} \\
	r \abs{\phi} \leq & u^{-\omg_{0} + 1} B_{\eta}.				\label{eq:DR:phi-huge-r}
\end{align}
\end{enumerate}
\end{theorem}
\begin{proof}
Observe that all the assumptions in \cite[Section~2]{DRPL} are satisfied on the characteristic rectangle $\calX_{0}$; we highlight here two of the assumptions. ${\bf \Dlt'}$, which is the uniform boundedness of $r^2\left|\f{\rd_V\phi}{\rd_V r}\right|$ and $r^{\omg_0}\left|\f{\rd_V\psi}{\rd_V r}\right|$ on $C_{out}$, follows from Proposition~\ref{prop:bg-large-r} (since $C_{out} \subseteq \underline{\calN}$). ${\bf \Sgm T'}$, which is the condition of subextremality of the event horizon, is a consequence of \eqref{eq:subextremality}. Therefore, the results in \cite{DRPL} are applicable;
\eqref{eq:DR:kpp}, \eqref{eq:DR:phi-med-r}--\eqref{eq:DR:dvrphi-med-r} and \eqref{eq:DR:dvphi-huge-r}--\eqref{eq:DR:phi-huge-r} correspond to Proposition~5.3, Theorem~9.4 and Theorem~9.5 in \cite{DRPL}, respectively. We note that these results in \cite{DRPL} are stated in the $(u, \bfV)$ coordinates, where $\bfV$ is defined by the normalization condition $\frac{\rd_{\bfV} r}{1-\mu} = 1$ on $C_{out}$. Nevertheless, by \eqref{eq:bg-large-r:dvr}, we clearly have $c (\bfV -1) \leq V -1 \leq C (\bfV -1)$ for some $c, C$ that depend on $\Tht$. \qedhere
\end{proof}

As a corollary of Theorem~\ref{thm:DR-full}, we may establish the following decay estimates for $\phi$, $\frac{1}{\rd_{u} r}\rd_{u} \phi$ and $(\frac{1}{\rd_{u} r} \rd_{u} )^{2} \phi$ near the event horizon.
\begin{corollary}[Decay of $\phi$ and its $\frac{1}{\rd_{u} r}\rd_u$ derivatives in bounded $r$ region] \label{cor:DR-small-r}
Let $\Lmb \geq 1$ and $\eta > 0$. In the region $\calX_{0} \cap \set{r \leq \Lmb r_{\EH}}$, there exists $B_{\Lmb} >  0 $, which depends on $\Lmb$, $\omg_{0}-\eta$ and the solution in $\calX_{0}$, so that we have, in the $(u,V)$ coordinate system
\begin{align} 
	\Abs{\phi} \leq B_{\Lmb} V^{-(\omg_{0}-\eta)}, \label{eq:DR:phi-small-r} \\
	\Abs{\frac{1}{\rd_{u} r} \rd_{u} \phi} \leq B_{\Lmb} V^{-(\omg_{0}-\eta)}, \label{eq:DR:duphi-small-r} \\
	\Abs{\left(\frac{1}{\rd_{u} r} \rd_{u}\right)^{2} \phi} \leq B_{\Lmb} V^{-(\omg_{0}-\eta)},\label{eq:DR:duduphi-small-r} \\
	\abs{\varpi - \varpi_{\EH}} \leq B_{\Lmb} V^{- 2(\omg_{0} - \eta) + 1}, \label{eq:DR:varpi-small-r}
\end{align}
where $\varpi_{\EH}$ is as in Proposition~\ref{prop.subextremality}.
\end{corollary}

The proof of Corollary~\ref{cor:DR-small-r} makes use of the celebrated red-shift effect along $\EH$. For this purpose, we need the following  elementary integration lemma. 
\begin{lemma} \label{lem:exp-int}
Let $\alp \in \bbR$, $b > 0$ and $0 < v_{0} \leq v_{1}$. Then we have
\begin{equation} \label{eq:exp-int}
	\int_{v_{0}}^{v_{1}} v^{\alp} e^{b (v - v_{1})} \, \ud v \leq C_{v_{0}, b, \alp} \, v_{1}^{\alp}.
\end{equation}
\end{lemma}
\begin{proof}
We divide both sides of \eqref{eq:exp-int} by $v_{1}^{\alp}$, and make the change of variables $s = v / v_{1}$. Then splitting the integral into $\int_{v_{0}/v_{1}}^{1/2}(\cdots)$ and $\int_{1/2}^{1}(\cdots)$, we estimate:
\begin{align*}
	v_{1}^{-\alp} \int_{v_{0}}^{v_{1}} v^{\alp} e^{b (v - v_{1})} \, \ud v
	= & \int_{v_{0}/v_{1}}^{1} s^{\alp} e^{b v_{1} (s - 1)} v_{1} \, \ud s  \\
	\leq & e^{-\frac{b}{2} v_{1}} v_{1} \left(\sup_{s \in [v_{0} / v_{1}, 1/2]} s^{\alp} \right)\int_{v_{0}/v_{1}}^{1/2} \, \ud s
		+ \left( \sup_{s \in [1/2, 1]} s^{\alp} \right) \int_{1/2}^{1} e^{b v_{1}(s - 1)} v_{1} \, \ud s \\
	\leq & 2^{-1} e^{-\frac{b}{2} v_{1}} v_{1} \left(\sup_{s \in [v_{0} / v_{1}, 1/2]} s^{\alp} \right)
		+ b^{-1} \left( \sup_{s \in [1/2, 1]} s^{\alp} \right) .
\end{align*}
Clearly, each term on the last line is uniformly bounded in $v_{1}$, with the bound depending only on $\alp, b, v_{0}$. \qedhere
\end{proof}

\begin{proof}[Proof of Corollary~\ref{cor:DR-small-r}]
We proceed in several steps.

\pfstep{Step~1: Proofs of \eqref{eq:DR:phi-small-r} and \eqref{eq:DR:duphi-small-r}}
Since $\inf_{\calX_{0}} r > 0$, and only the region where $V$ is large is relevant, \eqref{eq:DR:phi-small-r} follows from \eqref{eq:DR:phi-med-r}. Moreover, to prove \eqref{eq:DR:duphi-small-r}, it suffices to prove that 
\begin{equation} \label{eq:DR:duphi-small-r-goal}
	\Abs{\frac{r}{\rd_{u} r} \rd_{u} \phi}(u, V) \leq B_{\Lmb} V^{-(\omg_{0}-\eta)} 
\end{equation}
in a region of the form $\calX_{0} \cap \set{r \leq \Lmb r_{\EH}} \cap \set{V \geq V_{0}}$. Note that $\frac{r}{\rd_{u} r} \rd_{u} \phi$ obeys the equation 
\begin{align}
	\rd_{V} \left(\frac{r}{\rd_{u} r} \rd_{u} \phi \right)
=& 	- \rd_{V} \phi
	- \frac{2 (\varpi - \frac{\e^{2}}{r})}{r^{2}} \frac{\rd_{V} r}{1-\mu} \left( \frac{r}{\rd_{u} r} \rd_{u} \phi \right) . \label{eq:dVrduphi}
\end{align}
We claim that, for some $V_{0} \geq 1$ (depending on the solution in $\calX_{0}$) and $b_{\Lmb} > 0$ (depending on $\varpi_{f}, \e, \Lmb$ and $B_{0}$), we have
\begin{equation} \label{eq:DR-small-r:redshift}
	r \geq \frac{1}{2} r_{\EH}, \qquad 
	\frac{2 (\varpi-\frac{\e^{2}}{r})}{r^{2}} \frac{\rd_{V} r}{1-\mu} \geq b_{\Lmb} \quad \hbox{ in } \calX_{0} \cap \set{r \leq \Lmb r_{\EH}} \cap \set{V \geq V_{0}} .
\end{equation}
The second one-sided bound captures the red-shift effect along the event horizon $\EH$. We defer its proof until Step~4, and proceed with the proof of \eqref{eq:DR:duphi-small-r-goal}.

To use the method of an integrating factor, we rewrite \eqref{eq:dVrduphi} as
\begin{equation*}
	\rd_{V}\left( e^{I} \frac{r}{\rd_{u} r} \rd_{u} \phi \right)(u,V) = - e^{I} \rd_{V} \phi(u, V) , \quad \hbox{ where } \rd_{V} I = \frac{2 (\varpi-\frac{\e^{2}}{r})}{r^{2}} \frac{\rd_{V} r}{1-\mu}.
\end{equation*}
Note that if $(u, V)$ belongs to $\calX_{0} \cap \set{r \leq \Lmb r_{\EH}}$, then so does $(u, V')$ for $1 \leq V' \leq V$. Using \eqref{eq:DR-small-r:redshift} to bound $- (I(u, V) - I(u, V_{0})) \leq - b_{\Lmb} (V - V_{0})$, we then obtain
\begin{align*}
	\Abs{\frac{r}{\rd_{u} r} \rd_{u} \phi(u, V)}
	= & \Abs{e^{-(I(u, V) - I(u, V_{0}))} \frac{r}{\rd_{u} r} \rd_{u} \phi(u, V_{0}) - \int_{V_{0}}^{V} e^{- (I(u, V) - I(u, V'))} \rd_{V} \phi(u, V') \, \ud V'} \\
	\leq & e^{-b_{\Lmb}(V - V_{0})} \sup_{\uC_{V_{0}} \cap \calX_{0} \cap \set{r \leq \Lmb r_{\EH}}} \Abs{\frac{r}{\rd_{u} r} \rd_{u} \phi}
		 + \int_{V_{0}}^{V} e^{- b_{\Lmb}(V - V')} \abs{\rd_{V} \phi}(u, V') \, \ud V' .
\end{align*}
The first term decays exponentially in $V$, and is thus already acceptable. Taking $V_{0}$ larger (depending on $c_{\omg_{0} - \eta}$) if necessary, we may apply \eqref{eq:DR:dvphi-med-r} to obtain the bound 
\begin{equation} \label{eq:DR-small-r:dvphi}
\abs{\rd_{V} \phi}(u, V') \leq C_{r_{\EH}} B_{\omg_{0}-\eta} (V')^{-(\omg_{0} - \eta)}  \quad \hbox{ in } \calX_{0} \cap \set{r \leq \Lmb r_{\EH}} \cap \set{V \geq V_{0}} .
\end{equation}
Then by Lemma~\ref{lem:exp-int}, we have
\begin{equation*}
\int_{V_{0}}^{V} e^{- b_{\Lmb}(V - V')} \abs{\rd_{V} \phi}(u, V') \, \ud V' \leq C_{V_{0}, b_{\Lmb}, \omg_{0}-\eta, r_{\EH}} B_{\omg_{0} - \eta} V^{-(\omg_{0} - \eta)},
\end{equation*}
which is also acceptable for the proof of \eqref{eq:DR:duphi-small-r-goal}.

\pfstep{Step~2: Proof of \eqref{eq:DR:duduphi-small-r}}
Next, we establish \eqref{eq:DR:duduphi-small-r} assuming \eqref{eq:DR-small-r:redshift}. 
Let $V_{0}$ and $B_{\Lmb}$ be chosen as in Step~1.
Thanks to \eqref{eq:DR:duphi-small-r} and the fact that $\inf_{\calX_{0}} r > 0$, it suffices to show that 
\begin{equation} \label{eq:DR:duduphi-small-r-goal}
	\Abs{\frac{1}{\rd_{u} r} \rd_{u} \left( \frac{r}{\rd_{u} r} \rd_{u} \phi\right)} (u, V)
	\leq B_{\Lmb}' V^{-(\omg_{0} - \eta)}
\end{equation}
in $\calX_{0} \cap \set{r \leq \Lmb r_{\EH}} \cap \set{V \geq V_{0}}$, for some $B_{\Lmb}'$ depending on $r_{\EH}$, $\varpi_{i}$, $\e$, $b_{\Lmb}$ and $B_{\Lmb}$ as in \eqref{eq:DR:duphi-small-r}.

Taking $\frac{1}{\rd_{u} r} \rd_{u}$ of \eqref{eq:dVrduphi}, we obtain
\begin{align*}
	\rd_{V} \left(\frac{1}{\rd_{u} r} \rd_{u} \left(\frac{r}{\rd_{u} r} \rd_{u} \phi\right) \right)
=& 	\frac{1}{r} \rd_{V} \phi
	+ \left( \frac{\rd_{V} r}{r^{2}} - \frac{1}{\rd_{u} r} \rd_{u} \left(\frac{2 (\varpi - \frac{\e^{2}}{r})}{r^{2}} \frac{\rd_{V} r}{1-\mu} \right) \right) 
		\left( \frac{r}{\rd_{u} r} \rd_{u} \phi \right) \\
	& - \frac{4 (\varpi - \frac{\e^{2}}{r})}{r^{2}} \frac{\rd_{V} r}{1-\mu} \left(\frac{1}{\rd_{u} r} \rd_{u} \left(\frac{r}{\rd_{u} r} \rd_{u} \phi\right) \right) .
\end{align*}
On the RHS, note the presence of a linear term, whose coefficient is strictly negative thanks again to \eqref{eq:DR-small-r:redshift}. Therefore, proceeding as in Step~1, it suffices to show that the first two terms on the RHS decay at the rate $V^{-(\omg_{0}-\eta)}$. For the first term, it is a clear consequence of $r \geq \frac{1}{2} r_{\EH}$ and \eqref{eq:DR-small-r:dvphi}. For the second term, we use \eqref{eq:DR:duphi-small-r} and the bound
\begin{equation} \label{eq:DR:duduphi-small-r:claim}
\Abs{ \frac{\rd_{V} r}{r^{2}} - \frac{1}{\rd_{u} r} \rd_{u} \left(\frac{2 (\varpi - \frac{\e^{2}}{r})}{r^{2}} \frac{\rd_{V} r}{1-\mu} \right) } (u, V') \leq C_{r_{\EH}, \varpi_{i}, \e, B_{\Lmb}} .
\end{equation}
To prove \eqref{eq:DR:duduphi-small-r:claim}, we rewrite the expression inside the absolute value on the LHS as
\begin{align*}
\frac{\rd_{V} r}{(1-\mu) r^{2}} (1-\mu)  
+ \frac{4 \varpi - 6 \frac{\e^{2}}{r}}{r^{3}} \frac{\rd_{V} r}{1-\mu} - \frac{2}{r^{2}} \frac{\rd_{V} r}{1-\mu} \frac{1}{\rd_{u} r} \rd_{u} \varpi 
- \frac{2 (\varpi - \frac{\e^{2}}{r})}{r^{2}} \frac{1}{\rd_{u} r} \rd_{u} \left( \frac{\rd_{V} r}{1-\mu} \right).
\end{align*}
Clearly, the first two terms are uniformly bounded in $\calX_{0} \cap \set{r \leq \Lmb r_{\EH}} \cap \set{V \geq V_{0}}$ thanks to \eqref{eq:DR:kpp}, $\varpi \leq \varpi_{i}$ and $r \geq \frac{1}{2} r_{\EH}$. For the last two terms, we use
\begin{align*}
\frac{1}{\rd_{u} r} \rd_{u} \varpi
= \frac{1}{2} (1-\mu) r^{2} \left( \frac{1}{\rd_{u} r} \rd_{u} \phi \right)^{2}, \qquad
\frac{1}{\rd_{u} r} \rd_{u} \left( \frac{\rd_{V} r}{1-\mu} \right)
=  \frac{\rd_{V} r}{1-\mu} r \left( \frac{1}{\rd_{u} r} \rd_{u} \phi \right)^{2}, 
\end{align*}
and the bounds \eqref{eq:DR:kpp}, \eqref{eq:DR:duphi-small-r}, $\varpi \leq \varpi_{i}$ and $r \in [\frac{1}{2} r_{\EH}, \Lambda r_{\EH}]$.

\pfstep{Step~3: Proof of \eqref{eq:DR:varpi-small-r}}
First, we work on $\set{(\infty, V) : V \geq 1}$, which lies on $\EH$. By Theorem~\ref{thm:DR-full}, we have
\begin{equation*}
	\abs{\varpi(\infty, V) - \varpi_{\EH}}
	\leq \int_{V}^{\infty} \frac{1}{2} \frac{1-\mu}{\rd_{V} r} \left(r \rd_{V} \phi \right)^{2} (\infty, V') \, \ud V'
	\leq B'' V^{-2 (\omg_{0} - \eta) + 1},
\end{equation*}
for some $B'' > 0$ depending on $\inf_{\calX_{0}} r$, $B_{0}$ and $B_{\omg_{0} - \eta}$. Next, by \eqref{eq:DR:duphi-small-r}, we have
\begin{align*}
	\abs{\varpi(u, V) - \varpi(\infty, V)}
	\leq & \int_{u}^{\infty} \frac{1}{2} (1-\mu) r^{2} (- \rd_{u} r) \left(\frac{1}{(-\rd_{u} r)} \rd_{u} \phi\right)^{2} (u', V) \, \ud u' \\
	\leq & C \sup_{\calX_{0}} (1-\mu) (\Lmb r_{\EH})^{3} B_{\Lmb}^{2} V^{-2 (\omg_{0} - \eta)},
\end{align*}
where $B_{\Lmb}$ is from \eqref{eq:DR:duphi-small-r}. Combining the previous two bounds and readjusting $B_{\Lmb}$, \eqref{eq:DR:varpi-small-r} follows.

\pfstep{Step~4: Proof of \eqref{eq:DR-small-r:redshift}}
To complete the proof, it remains to find $V_{0} \geq 1$ and $b_{\Lmb} > 0$ such that \eqref{eq:DR-small-r:redshift} holds. By the relation $r_{\EH} = \varpi_{\EH} + \sqrt{\varpi_{\EH}^{2} - \e^{2}}$ in \eqref{r.varpi.poly} and the subextremality of $\EH$ ($\abs{\e} < \varpi_{\EH}$, see \eqref{eq:subextremality}), we have
\begin{equation*}
	\varpi_{\EH}  - \frac{\e^{2}}{r_{\EH}}
	= r_{\EH}^{-1}\left( \varpi_{\EH} + \sqrt{\varpi_{\EH}^{2} - \e^{2}} \right) \sqrt{\varpi_{\EH}^{2} - \e^{2}}
	= \sqrt{\varpi_{\EH}^{2} - \e^{2}} > 0.
\end{equation*}
Recall from \eqref{r.varpi.limits} that $r_{\EH} = \lim_{V \to \infty} r(\infty, V)$ and $\varpi_{\EH} = \lim_{V \to \infty} \varpi(\infty, V)$. Thus, there exists $V_{0} \geq 1$ such that
\begin{equation*}
	r(\infty, V) \geq \frac{1}{2} r_{\EH}, \quad
	\left( \varpi - \frac{\e^{2}}{r}\right)(\infty, V) \geq \frac{1}{2} \sqrt{\varpi_{\EH}^{2} - \e^{2}}, \qquad
	V \geq V_{0}.
\end{equation*}
By monotonicity of $r$ and $\varpi$ and \eqref{eq:DR:kpp}, we see that \eqref{eq:DR-small-r:redshift} holds with $b_{\Lmb} = \frac{1}{\Lmb^{2} r_{\EH}^{2}} \sqrt{\varpi_{\EH}^{2} - \e^{2}} B_{0}^{-1}$. \qedhere
\end{proof}

Next, as another corollary of Theorem~\ref{thm:DR-full}, we propagate $r$-weighted bounds for $\rd_{V} \psi$ and $\rd_{V}^{2} \psi$ to a neighborhood of null infinity. We refer the readers to Figure~\ref{fig:DR-full} for a depiction of the region $\mathcal N$ below.
\begin{corollary}[Decay in a full neighborhood of null infinity] \label{cor:DR-large-r}
Given $\eta_{\NI} > 0$, we define
\begin{equation*}
	\calN = \set{(u, V) \in \PD : \forall V' \geq V, \ r(u, V') \geq \eta_{\NI}^{-1} \max \set{ \varpi_{i}, \abs{\e} },\ \rd_{u} r (u, V') < 0, \ \rd_{V} r (u, V') > 0}.
\end{equation*}
If $\eta_{\NI}$ is sufficiently small depending on $\calX_{0}$, $\eta_{i^{0}}$, there exists $B_{\eta_{\NI}} >  0 $, which depends on $\eta_{\NI}$, $\varpi_{i}$, $\e$, $\omg_{0}$, $\eta$, $c_{\omg_{0}-\eta}$, $B_{\omg_{0}-\eta}$, $B_{\eta}$ and $B_{\Lmb}$ (where $\Lmb$ depends on $\eta_{\NI}$, $\varpi_{i}$, $\abs{\e}$), so that the following bounds in the $(u,V)$ coordinates hold in $\calN$:
\begin{align} 
	\frac{1}{2} \leq \frac{\rd_{V} r(u, V)}{\rd_{V} r (\underline{u}(V), V)}\leq & 2 , \label{eq:DR:dvr-large-r} \\
	\frac{1}{2} \leq -\rd_{u} r \leq & 2 , \label{eq:DR:dur-large-r} \\
	\f 12\leq (1-\mu) \leq & 2, \label{eq:DR:1-mu-large-r} \\
	\Abs{\frac{1-\mu}{\rd_{V} r} \rd_{V} \psi} \leq & r^{-\omg_{0}} B_{\eta_{\NI}} , \label{eq:DR:dvrphi-large-r} \\
	\Abs{\frac{1}{\rd_{u} r} \rd_{u} \psi} \leq & (V_{R}(u)^{-(\omg_{0}-\eta)} + u^{-(\omg_{0} - \eta)}) B_{\eta_{\NI}} \quad \hbox{ for } u \geq 1, \label{eq:DR:durphi-large-r} \\
	\Abs{\left(\frac{1-\mu}{\rd_{V} r} \rd_{V}\right)^{2} \psi} \leq & r^{-\omg_{0}-1} B_{\eta_{\NI}}.\label{eq:DR:dvdvrphi-large-r}
\end{align}
Here, $\underline{u}(V)$ is defined by $(\underline{u}(V), V) \in \Sgm_{0}$, and $(u, V_{R}(u))$ is the point in $C_u\cap \set{r = \eta_{\NI}^{-1} \max\set{\varpi_{i}, \abs{\e}}}$.
\end{corollary}

After the proof, {\bf we fix $\eta_{\NI} > 0$ and thus $\calN$ so that $\eta_{\NI} \leq \eta_{i^{0}}$ and Corollary~\ref{cor:DR-large-r} applies.}

\begin{proof}
As in the proof of Proposition~\ref{prop:bg-large-r}, we abbreviate $R = \eta_{\NI}^{-1} \max \set{\varpi_{i}, \e}$. Choosing $\eta_{\NI} > 0$ sufficiently small (depending on $\calX_{0}$ and $\eta_{i^{0}}$), we may assume\footnote{Recall that $\calX_{0}$ is chosen such that the initial outgoing curve lies in $\underline{\calN}$.} that $\calN \subset \calX_{0} \cup \underline{\calN}$. Note that all (maximally-past-extended) incoming null curves in $\calN$ meet $\Sgm_{0}$ in the past direction, i.e., if $(u, V) \in \calN$, then $\set{(u', V) \in \PD : \underline{u}(V) \leq u' \leq u} \subseteq \calN$. Moreover, the following bounds hold in $\calN$:
\begin{align} 
	\frac{\sup_{\calN} \varpi}{r} + \frac{\abs{\e}}{r} \leq & \eta_{\NI}, \label{eq:DR-large-r:pf-me} \\
	1 - 2 \eta_{\NI} \leq 1 - \mu \leq & 1 + \eta_{\NI}^{2} . \label{eq:DR-large-r:pf-mu}
\end{align}

Choosing $\eta_{\NI} > 0$ sufficiently small, \eqref{eq:DR:dvr-large-r} follows from exactly the same proof as \eqref{eq:bg-large-r:dvr}. Next, \eqref{eq:DR:dur-large-r} is proved in a similar manner as \eqref{eq:bg-large-r:dur}, except we use $\lim_{V \to \infty} \log (-\rd_{u} r) (u, V) = \lim_{V \to \infty} \log \frac{-\rd_{u} r}{1-\mu} (u, V)=0$ and integrate from $(u, \infty)$ to $(u, V)$. The bound \eqref{eq:DR:1-mu-large-r} follows immediately from \eqref{eq:DR-large-r:pf-mu} after choosing $\eta_{\NI}$ to be sufficiently small. For the proof of \eqref{eq:DR:dvrphi-large-r}, in view of \eqref{eq:bg-large-r:dvrphi} in $\underline{\calN} \cap \calX_{0}$, it suffices to consider the case $(u, V) \in \calN \cap \calX_{0}$ (i.e., $u \geq 1$). In this region, we have 
\begin{equation} \label{eq:DR-large-r:pf-psi}
\abs{\psi} (u, V) \leq \left\{
\begin{array}{cl}
	V^{-\omg_{0}+1} (1+c_{\omg_{0} - \eta}) B_{\omg_{0} - \eta} & \hbox{ for } u \geq 1,\ r(u, V) \leq c_{\omg_{0} - \eta} V \\
	u^{-\omg_{0}+1} B_{\eta} & \hbox{ for } u \geq 1,\ r(u, V) \geq c_{\omg_{0} - \eta} V 
\end{array} \right.
\end{equation}
by \eqref{eq:DR:phi-med-r} and \eqref{eq:DR:phi-huge-r}. Proceeding as in the proof of \eqref{eq:bg-large-r:dvrphi}, we estimate
\begin{align*}
	r^{\omg_{0}} \Abs{\frac{1}{\rd_{V} r} \rd_{V} \psi}(u, V)
	\leq & C r^{\omg_{0}} \Abs{\frac{1}{\rd_{V} r} \rd_{V} \psi}(1, V) 
	 + C \frac{\varpi_{i}}{R^{3-\omg_{0}}} \int_{1}^{\max\set{u, 1}} \abs{\psi(u', V)} (-\rd_{u} r)(u', V)\, \ud u' .
\end{align*}
The first term on the RHS is bounded by $C B_{i^{0}}$ in view of \eqref{eq:bg-large-r:dvrphi}, whereas the integral in the second term may be bounded by $C ((1+c_{\omg_0-\eta}) c_{\omg_{0}-\eta}^{\omg_{0}-1} R^{-\omg_{0}+2} B_{\omg_{0} - \eta} + \f{B_{\eta}}{\omg_0-2})$ using \eqref{eq:DR:dur-large-r} and \eqref{eq:DR-large-r:pf-psi}. Using \eqref{eq:DR-large-r:pf-mu} to bound $1-\mu$, we obtain \eqref{eq:DR:dvrphi-large-r}. Similarly, to prove \eqref{eq:DR:durphi-large-r}, we use the equation
\begin{equation*}
	\rd_{V} \left( \frac{1}{\rd_{u} r} \rd_{u} \psi\right)
	= - \frac{\rd_{u} \rd_{V} r}{\rd_{u} r} \frac{1}{\rd_{u} r} \rd_{u} \psi + \frac{\rd_{u} \rd_{V} r}{r \rd_{V} r} \psi
\end{equation*}
Given $(u, V) \in \calN$, let $V_{R}(u)$ be defined by $r (u, V_{R}(u)) = R$. Observe that $\int_{V_{R}(u)}^{V} \Abs{\frac{\rd_{V} \rd_{u} r}{\rd_{u} r}}(u, V') \, \ud V' \leq C$ by \eqref{eq:DR:dur-large-r}. Therefore, integrating the preceding equation from $(u, V_{R}(u))$ to $(u, V)$, then applying \eqref{eq:DR-large-r:pf-me}, \eqref{eq:DR-large-r:pf-mu} and Gr\"onwall's inequality, we obtain
\begin{align*}
	\Abs{\frac{1}{\rd_{u} r} \rd_{u} \psi}(u, V)
	\leq & C \Abs{\frac{1}{\rd_{u} r} \rd_{u} \psi}(u, V_{R}(u)) + C  \int_{V_{R}(u)}^{V} \frac{2 (\varpi - \frac{\e^{2}}{r})}{r^{3}} \frac{\rd_{V} r}{1-\mu} \abs{\psi} (u, V') \, \ud V' \\
	\leq & C \Abs{\frac{1}{\rd_{u} r} \rd_{u} \psi}(u, V_{R}(u)) + C \eta_{\NI} \sup_{V' \in [V_{R}(u), \infty)} r^{-1} \abs{\psi} (u, V') .
\end{align*}
On the last line, the first term is acceptable in view of \eqref{eq:DR:phi-small-r} and \eqref{eq:DR:duphi-small-r}. For the second term, we use the bound
\begin{equation*}
	r^{-1} \abs{\psi(u, V)} = \abs{\phi(u, V)} \leq (V^{-(\omg_{0}-\eta)} r^{-1} + V^{-\omg_{0}}) B_{\omg_{0}-\eta} + u^{-(\omg_{0}-1)} V^{-1} c_{\omg_{0} - \eta}^{-1}B_{\eta},
\end{equation*}
which follows by combining \eqref{eq:DR:phi-med-r} and \eqref{eq:DR:phi-huge-r}. Then using the simple bound $u^{-(\omg_{0} - 1)} V^{-1} \leq C_{\omg_{0}} (u^{- \omg_{0}} + V^{- \omg_{0}})$, we see that this term is acceptable, too.

It remains to prove \eqref{eq:DR:dvdvrphi-large-r}. By \eqref{eq:DR:dvrphi-large-r}, \eqref{eq:DR-large-r:pf-me}, \eqref{eq:DR-large-r:pf-mu} and
\begin{equation*}
	\frac{1-\mu}{\rd_{V} r} \rd_{V}(1-\mu)
	= (1-\mu) \left(\frac{2 \varpi}{r^{2}} - \frac{2 \e^{2}}{r^{3}}\right) - \frac{1}{r} \left( \frac{1-\mu}{\rd_{V} r} r \rd_{V} \phi \right)^{2} ,
\end{equation*}
it suffices to bound $r^{\omg_{0} + 1} \Abs{\left(\frac{1}{\rd_{V} r} \rd_{V}\right)^{2} \psi}$. We use the equation
\begin{align*}
	\rd_{u} \left(\frac{1}{\rd_{V} r} \rd_{V} \right)^{2} \psi 
	=& \left(\frac{1}{\rd_{V} r} \rd_{V} \right) \rd_{u} \left(\frac{1}{\rd_{V} r} \rd_{V} \right) \psi 
		- \frac{\rd_{u} \rd_{V} r} {\rd_{V} r}\left(\frac{1}{\rd_{V} r} \rd_{V} \right)^{2} \psi \\
	= &\left(\frac{1}{r \rd_{V} r} \rd_{V} \left( \frac{2 (\varpi - \frac{\e^{2}}{r}) } {r^{2}} \frac{\rd_{u} r}{1-\mu} \right) - \frac{2 (\varpi - \frac{\e^{2}}{r}) } {r^{4}} \frac{\rd_{u} r}{1-\mu} \right) \left(- \frac{r}{\rd_{V} r} \rd_{V} \psi +  \psi \right) \\
		& - \frac{4 (\varpi - \frac{\e^{2}}{r}) } {r^{2}} \frac{\rd_{u} r}{1-\mu} \left(\frac{1}{\rd_{V} r} \rd_{V} \right)^{2} \psi .
\end{align*}
By \eqref{eq:DR-large-r:pf-me} and \eqref{eq:DR-large-r:pf-mu}, we have $\int_{\underline{u}(V)}^{u} \Abs{\frac{4 (\varpi - \frac{\e^{2}}{r}) } {r^{2}} \frac{\rd_{u} r}{1-\mu}} (u', V) \, \ud u' \leq C$. Moreover, by \eqref{eq:bg-large-r:rphi}, \eqref{eq:DR:dvrphi-large-r}, \eqref{eq:DR-large-r:pf-me}, \eqref{eq:DR-large-r:pf-mu} and the equations
\begin{align*}
\frac{1}{\rd_{V} r} \rd_{V} \varpi
= \frac{1}{2} (1-\mu) \left( \frac{1}{\rd_{V} r} \rd_{V} \psi - \frac{1}{r} \psi \right)^{2}, \qquad
\frac{1}{\rd_{V} r} \rd_{V} \left( \frac{\rd_{u} r}{1-\mu} \right)
=  \frac{\rd_{u} r}{1-\mu} r^{-1} \left( \frac{1}{\rd_{V} r} \rd_{V} \psi - \frac{1}{r} \psi \right)^{2}, 
\end{align*}
we have
\begin{equation*}
	\int_{\underline{u}(V)}^{u} r^{\omg_{0}+1} \Abs{\left(\frac{1}{r \rd_{V} r} \rd_{V} \left( \frac{2 (\varpi - \frac{\e^{2}}{r}) } {r^{2}} \frac{\rd_{u} r}{1-\mu} \right) - \frac{2 (\varpi - \frac{\e^{2}}{r}) } {r^{4}} \frac{\rd_{u} r}{1-\mu} \right) \left(- \frac{r}{\rd_{V} r} \rd_{V} \psi +  \psi \right)}(u', V) \, \ud u' \leq B',
\end{equation*}
where $B'$ depends on $\eta_{\NI}$, $R$, $\omg_{0}$, $c_{\omg_{0} - \eta}$, $B_{i^{0}}$, $B_{\omg_{0} - \eta}$, $B_{\eta}$ and $B_{\eta_\NI}$ (from \eqref{eq:DR:dvrphi-large-r} that we just established). Therefore, by Gr\"onwall's inequality and monotonicity of $r$ (i.e., $\rd_{u} r \leq 0$), we have
\begin{align*}
	r^{\omg_{0}+1} \Abs{\left( \frac{1}{\rd_{V} r} \rd_{V} \right)^{2} \psi}(u, V)
	\leq & C r^{\omg_{0}+1} \Abs{\left( \frac{1}{\rd_{V} r} \rd_{V} \right)^{2} \psi}(\underline{u}(V), V) 
	+ C B' ,
\end{align*}
which implies \eqref{eq:DR:dvdvrphi-large-r}. \qedhere
\end{proof}

\begin{remark} \label{rem:DR-large-r:dUr}
Although stated in the $(u, V)$ coordinates, note that \eqref{eq:DR:dvr-large-r}, \eqref{eq:DR:dvrphi-large-r} and \eqref{eq:DR:dvdvrphi-large-r} are \emph{invariant} under coordinate transforms of the form $(u, V) \mapsto (U(u), V)$. 
The bounds \eqref{eq:DR:dur-large-r} and \eqref{eq:DR:durphi-large-r} are specific to the coordinate system $(u, V)$. However, in the case of \eqref{eq:DR:dur-large-r}, there is a qualitative analogue in the $(U, V)$ coordinates. Namely, given some $U_{1} \leq U_{\EH}$, observe (by a straightforward compactness argument and Proposition~\ref{prop:bg-large-r}) that the Jacobian $\frac{\ud u(U)}{\ud U}$ is bounded from below on $\set{(U, V) \in \PD : U \leq U_{1}}$ by a positive constant. Therefore, \eqref{eq:DR:dur-large-r} implies
\begin{equation} \label{eq:DR-large-r:dUr}
	B_{cpt}^{-1} \leq - \rd_{U} r \leq B_{cpt} \quad \hbox{ in } \calN \cap \set{(U, V) \in \PD : U \leq U_{1}}
\end{equation}
where $B_{cpt} > 0$ depends on the background solution on $\set{(U, V) \in \PD : U \leq U_{1}}$.
\end{remark} 
After the proof, {\bf we fix $\eta_{\NI} > 0$ and thus $\calN$ so that $\eta_{\NI} \leq \eta_{i^{0}}$ and Corollary~\ref{cor:DR-large-r} applies.}

As an almost immediate consequence of Corollaries~\ref{cor:DR-small-r} and \ref{cor:DR-large-r}, we show that the limits of the mass on the event horizon and null infinity coincide.
\begin{corollary}[$\varpi_{\EH}=\varpi_f$]\label{varpi.same.limit}
Let $\varpi_{\EH}$ be as in Proposition~\ref{prop.subextremality} and\footnote{Notice that $\varpi_f$ is well-defined by the monotonicity of $\varpi$.} $\varpi_f=\lim_{u\to \infty} \lim_{V\to\infty} \varpi(u,V)$. Then
$$\varpi_{\EH}=\varpi_f.$$
\end{corollary}
\begin{proof}
Take a sequence $\{(u_n, V_n)\}_{n=1}^\infty$ with $u_n \to \infty$, $V_n \to \infty$, $r(u_n, V_n)=\eta_{\NI}^{-1}\max\{\varpi_i,\e\}$ and $(u_n, V_n)\in \calX_{0}$. In particular, Theorem~\ref{thm:DR-full} and Corollary~\ref{cor:DR-large-r} are applicable in $\{(u_n, V):V\geq V_n \}$; meanwhile Corollary~\ref{cor:DR-small-r} is applicable in $\{(u,V_n): u\geq u_n \}$ for $\Lambda =\f{\eta_{\NI}^{-1}\max\{\varpi_i,\e\}}{r_\EH}$. By Lemma~\ref{lem:cauchy-to-char}, \eqref{eq:adm-id-af}, \eqref{eq:DR:phi-med-r}, \eqref{eq:DR:phi-huge-r}, \eqref{eq:DR:dvr-large-r} and \eqref{eq:DR:1-mu-large-r},
\begin{equation*}
\begin{split}
\left|\varpi(u_n,V_n)-\varpi(u_n,\infty)\right|\leq &\int_{V_n}^\infty \f 12 \f{1-\mu}{\rd_V r} (r\rd_V\phi)^2 (u_n, V')\, \ud V'\\
\leq &CB_{\omg_0-\eta}^2 \int_{V_n}^\infty(V')^{-\min\{2(\omg_0-\eta),4\}}\,\ud V'\leq CB_{\omg_0-\eta}^2 V_n^{-\min\{2(\omg_0-\eta),4\}+1}.
\end{split}
\end{equation*}
Combining this with \eqref{eq:DR:varpi-small-r} (to control $\varpi(u_n,V_n)-\varpi_{\EH}$), we have
$$\left|\varpi(u_n,\infty)-\varpi_{\EH}\right|\leq CB_{\omg_0-\eta}^2 V_{n}^{-\min\{2(\omg_0-\eta),4\}+1}+B_{\Lambda}V_{n}^{2(\omg_0-\eta)+1},$$
where $\Lambda =\f{\eta_{\NI}^{-1}\max\{\varpi_n,\e\}}{r_\EH}$. Taking $n \to \infty$ and using the definition of $\varpi_f$ yields the conclusion. \qedhere
\end{proof}

Putting together Theorem~\ref{thm:DR-full} and Corollaries~\ref{cor:DR-small-r} and \ref{cor:DR-large-r}, we restate the Price's law decay rates in a form that will be convenient for later applications.
\begin{corollary} \label{cor:DR-final}
For any $2 \leq \omg' < \omg_{0}$, there exists $B_{\omg'} > 0$, which depends on $\omg_{0}$, $\omg'$ and the solution in $\calX_{0}$, such that the following bounds in $(u,V)$ coordinates hold.
\begin{align}
	\Abs{\phi} \leq & \left\{
\begin{array}{ll}
 V^{-\omg'} B_{\omg'} & \hbox{ in } \calX_{0} \cap \set{r \leq 30 r_{\EH}} , \\
 u^{-(\omg'-1)} \min \set{ u^{-1}, r^{-1}} B_{\omg'}  & \hbox{ in } \calX_{0} \cap \set{r \geq 10 r_{\EH}} ,
\end{array}\right.  		\label{eq:DR-final:phi} \\
	\Abs{\frac{1}{\rd_{u} r}\rd_{u} \phi} \leq & \left\{
\begin{array}{ll}
V^{-\omg'} B_{\omg'}  & \hbox{ in } \calX_{0} \cap \set{r \leq 30 r_{\EH}} , \\
u^{-\omg'} r^{-1} B_{\omg'}   & \hbox{ in } \calX_{0} \cap \set{r \geq 10 r_{\EH}} ,
\end{array}\right. 	\label{eq:DR-final:duphi} \\
	\Abs{\frac{1-\mu}{\rd_{V} r}\rd_{V} \phi} \leq & \left\{
\begin{array}{ll}
V^{-\omg'}  B_{\omg'} & \hbox{ in } \calX_{0} \cap \set{r \leq 30 r_{\EH}} , \\
u^{-(\omg'-1)} r^{-1}  \min \set{ u^{-1}, r^{-1}} B_{\omg'}  & \hbox{ in } \calX_{0} \cap \set{r \geq 10 r_{\EH}} ,
\end{array}\right. 	\label{eq:DR-final:dvphi} \\
	\Abs{\frac{1-\mu}{\rd_{V} r}\rd_{V} (r \phi)} \leq & 
\min \set{V^{-\omg'}, r^{-\omg'}} B_{\omg'}	 \qquad \hbox{ in } \calX_{0}	. \label{eq:DR-final:dvrphi}
\end{align}
Moreover, let $v = v(V)$ be the future-normalized null coordinate defined by \eqref{eq:coord-v-bg-0} and $v(1) = 1$. We have, for $v,\, V\geq 1$,
\begin{equation} \label{eq:DR-final:vV}
	B_{0}^{-1} (V-1) \leq v-1 \leq B_{0} (V -1).
\end{equation}
\end{corollary}

\begin{proof}
Observe that \eqref{eq:DR-final:phi}--\eqref{eq:DR-final:dvphi} in the case $r \leq 30 r_{\EH}$ follow from Theorem~\ref{thm:DR-full} and Corollary~\ref{cor:DR-small-r} with $\omg' = \omg_{0} - \eta$ and $\Lmb = 30$. For the proof of these bounds in the case $r \geq 10 r_{\EH}$, we claim that
\begin{equation} \label{eq:DR-final:uV}
	u \leq C V \quad \hbox{ in } \calX_{0} \cap \set{r \geq 10 r_{\EH}}
\end{equation}
for some $C > 0$ that depends on the solution in $\calX_{0}$. 

Assuming \eqref{eq:DR-final:uV} for the moment, \eqref{eq:DR-final:phi} in the case $r \geq 10 r_{\EH}$ follows immediately from Theorem~\ref{thm:DR-full}. 
Next, for \eqref{eq:DR-final:duphi}, given any $\Lmb \geq 1$, combining \eqref{eq:DR:duphi-small-r} and \eqref{eq:DR-final:uV} handles the case $r \leq \Lmb r_{\EH}$. On the other hand, by \eqref{eq:DR:durphi-large-r}, \eqref{eq:DR-final:phi} and
\begin{equation*}
	r \Abs{\frac{1}{\rd_{u} r} \rd_{u} \phi}
	\leq \Abs{\frac{1}{\rd_{u} r} \rd_{u} (r \phi)} + \abs{\phi},
\end{equation*}
we obtain \eqref{eq:DR-final:duphi} in $\calX_{0} \cap \calN$. Thus, choosing $\Lmb$ large enough, the full statement of \eqref{eq:DR-final:duphi} follows.
For \eqref{eq:DR-final:dvphi}, the case $u \geq r$ can be handled using \eqref{eq:DR-final:uV} and Theorem~\ref{thm:DR-full}. In the alternative case $r \geq u$, we note that
\begin{equation*}
	r \Abs{\frac{1-\mu}{\rd_{V} r} \rd_{V} \phi}
	\leq \Abs{\frac{1-\mu}{\rd_{V} r} \rd_{V} (r \phi)} + (1 - \mu)\Abs{\phi}.
\end{equation*}
Thus, in combination with \eqref{eq:DR-final:phi}, \eqref{eq:DR-final:dvphi} follows from \eqref{eq:DR-final:dvrphi}, which in turn is a straightforward consequence of \eqref{eq:DR:kpp}, \eqref{eq:DR:dvrphi-med-r} and \eqref{eq:DR:dvrphi-large-r}. Finally, \eqref{eq:DR-final:vV} follows from \eqref{eq:DR:kpp} and the relation
\begin{equation*}
	\frac{\ud v}{\ud V}(V) =  \frac{\rd_{V} r}{1-\mu}(\infty, V),
\end{equation*}
which is equivalent to \eqref{eq:coord-v-bg-0}.

To conclude the proof, it only remains to verify \eqref{eq:DR-final:uV}. By \eqref{eq:DR:varpi-small-r} and the fact that $|\e| < \varpi_{\EH} \leq r_{\EH} = \varpi_{\EH} + \sqrt{\varpi_{\EH}^{2} - \e^{2}}$ (by \eqref{r.varpi.poly} and \eqref{eq:subextremality}), for any $\Lmb \geq 10$ we can find $V_{\Lmb} \geq 1$ large enough so that 
\begin{equation*}
1-\mu \geq \frac{1}{2} \quad \hbox{ in } \calX_{0} \cap \set{10 r_{\EH} \leq r \leq \Lmb r_{\EH}} \cap \set{V \geq V_{\Lmb}}.
\end{equation*}
Fixing $\Lmb$ to be a large enough number, we may combine this with Corollary~\ref{cor:DR-large-r} (more precisely, \eqref{eq:DR-large-r:pf-mu} from its proof) to conclude that $1-\mu \geq \frac{1}{2}$ holds in $\calX_{0} \cap \set{r \geq 10 r_{\EH}} \cap \set{V \geq V_{\Lmb}}$. Then integrating the equation \begin{equation*}
\rd_{V} \log (-\rd_{u} r) = \frac{2 (\varpi - \frac{\e^{2}}{r})}{r^{2} (1-\mu)} \rd_{V} r
\end{equation*}
from $(u, \infty)$ to $(u, V)$, we obtain
\begin{equation*}
	C^{-1} \leq - \rd_{u} r \leq C \quad \hbox{ in } \calX_{0} \cap \set{r \geq 10 r_{\EH}} \cap \set{V \geq V_{\Lmb}}
\end{equation*}
for some $C > 0$ depending on the solution in $\calX_{0}$. For $(u, V) \in \calX_{0} \cap \set{r \geq 10 r_{\EH}} \cap \set{V \geq V_{\Lmb}}$, it follows that
\begin{equation*}
	u -1 \leq C \int_{1}^{u} (-\rd_{u} r) (u', V) \, \ud u' = C (r(1, V) - r(u, V)) \leq C (r(1, V) - 10 r_{\EH}).
\end{equation*}
On the other hand,
\begin{equation*}
	r(1, V) = r(1, 1) + \int_{1}^{V} \rd_{V} r(1, V') \, \ud V' \leq r(1, 1) + (V - 1) \sup_{C_{out}} \rd_{V} r,
\end{equation*}
where $\sup_{C_{out}} \rd_{V} r < \infty$ by Proposition~\ref{prop:bg-large-r}. Putting together these inequalities, we have \eqref{eq:DR-final:uV} for $V\geq V_{\Lmb}$. The remaining region is compact so that \eqref{eq:DR-final:uV} is immediate. \qedhere
\end{proof}
\subsection{Geometric bounds in the future-normalized coordinates} \label{subsec:bg-uv}
Let $\calX_{0} = \set{(u, V) : u \geq 1, \ V \geq 1}$ be the characteristic rectangle introduced in Section~\ref{subsec:DR-full}. Consider the future-normalized advanced null coordinate $v = v(V)$, where $v(1) = 1$ and \eqref{eq:coord-v-bg-0} holds along $\EH$.
In the coordinate system $(u, v)$, we introduce the following set of notation for geometric quantities:
\begin{align*}
	\dvr = \rd_{v} r, \quad \dur = \rd_{u} r, \quad
	\kpp = \frac{\rd_{v} r}{1-\mu}, \quad \gmm = \frac{\rd_{u} r}{1-\mu}.
\end{align*}
In the following proposition we show that, near timelike infinity, these geometric quantities obey similar bounds as the corresponding quantities in the Reissner--Nordstr\"om spacetime in the Eddington--Finkelstein coordinates. 

\begin{proposition} \label{prop:bg-uv}
Let $2 < \omg \leq 3$, $\dlt_{1} > 0$ and $\Lmb \geq 1$. Let $\calX_{1} = \set{(u, v) \in \calX_{0} : u \geq u_{1}, \ v \geq v_{1}}$ be a characteristic rectangle, on which the following bounds in $(u,v)$ coordinates hold:
\begin{align}
	r \geq & (1 - \dlt_{1}) r_{\EH} \quad \hbox{ in } \calX_{1}, \label{eq:bg-uv:r-bnd} \\
	\Abs{\phi} \leq & \left\{
\begin{array}{ll}
\brk{\frac{v-v_{1}}{r_{\EH}}}^{-\omg} \dlt_{1} & \hbox{ in } \calX_{1} \cap \set{r \leq 30 r_{\EH}} , \\
\brk{\frac{u-u_{1}}{r_{\EH}}}^{-(\omg-1)} \min \set{ \brk{\frac{u-u_{1}}{r_{\EH}}}^{-1}, (\frac{r}{r_{\EH}})^{-1}} \dlt_{1}  & \hbox{ in } \calX_{1} \cap \set{r \geq 10 r_{\EH}} ,
\end{array}\right. 		\label{eq:bg-uv:phi} \\
	\Abs{\frac{1}{\dur}\rd_{u} \phi} \leq & \left\{
\begin{array}{ll}
\brk{\frac{v-v_{1}}{r_{\EH}}}^{-\omg}  r_{\EH}^{-1} \dlt_{1} & \hbox{ in } \calX_{1} \cap \set{r \leq 30 r_{\EH}} , \\
\brk{\frac{u-u_{1}}{r_{\EH}}}^{-\omg} (\frac{r}{r_{\EH}})^{-1} \Lmb^{-1/2} r_{\EH}^{-1} \dlt_{1}  & \hbox{ in } \calX_{1} \cap \set{r \geq 10 r_{\EH}} ,
\end{array}\right. 	\label{eq:bg-uv:duphi} \\
	\Abs{\frac{1}{\kpp}\rd_{v} \phi} \leq & \left\{
\begin{array}{ll}
\brk{\frac{v-v_{1}}{r_{\EH}}}^{-\omg}  r_{\EH}^{-1} \dlt_{1} & \hbox{ in } \calX_{1} \cap \set{r \leq 30 r_{\EH}} , \\
\brk{\frac{u-u_{1}}{r_{\EH}}}^{-(\omg-1)} (\frac{r}{r_{\EH}})^{-1}  \min \set{ \brk{\frac{u-u_{1}}{r_{\EH}}}^{-1}, (\frac{r}{r_{\EH}})^{-1}} r_{\EH}^{-1} \dlt_{1}  & \hbox{ in } \calX_{1} \cap \set{r \geq 10 r_{\EH}} ,
\end{array}\right. 	\label{eq:bg-uv:dvphi}
\end{align}
where $\brk{\cdot} = (\cdot) + 1$.

For sufficiently small $\dlt_{1}$ depending on $\varpi_{f}^{-1} \abs{\e}$, there exists a universal constant $C > 0$ so that the following bounds hold.
\begin{enumerate}
\item In $\calX_{1}$, we have
\begin{align} 
	\sup_{(u, v), (u', v') \in \calX_{1}} \abs{\varpi(u, v) - \varpi(u', v')} 
	\leq & C \dlt_{1}^{2} r_{\EH},	\label{eq:bg-uv:m-bnd} \\
	1 \leq \kpp \leq 2 \quad \hbox{ in } \calX_{1}. \label{eq:bg-uv:kpp-bnd} 
\end{align}
\item Away from $\EH$, we have
\begin{align} 
	\frac{1}{2} \leq 1 - \mu \leq 2  \quad \hbox{ in } \calX_{1} \cap \set{r \geq 10 r_{\EH}} , 	\label{eq:bg-uv:mu-large-r} \\
	\frac{1}{4} \leq -\dur \leq 4  \quad \hbox{ in } \calX_{1} \cap \set{r \geq 10 r_{\EH}},	\label{eq:bg-uv:dur-large-r} \\
	\frac{1}{4} \leq \dvr \leq 4  \quad \hbox{ in } \calX_{1} \cap \set{r \geq 10 r_{\EH}}. 	\label{eq:bg-uv:dvr-large-r}
\end{align}
Moreover, let $(u_{30}, v_{30})$ be the intersection between $\gmm_{30 r_{\EH}} = \set{r = 30 r_{\EH}}$ and the past null boundary $\rd \calX_{1}$ of $\calX_{1}$ (i.e., $\set{(u_{30}, v_{30})} = \gmm_{30 r_{\EH}} \cap \rd \calX_{1}$). Then for $C_{30} = u_{30} - v_{30}$, we have
\begin{equation} \label{eq:bg-uv:u-v}
	\sup_{\gmm_{30} \cap \calX_{1}} \abs{u - v - C_{30}} \leq C \dlt_{1}.
\end{equation}

\item Near $\EH$, we have
\begin{align} 
	0 \leq 1 - \mu \leq &\, 1 \quad \hbox{ in } \calX_{1} \cap \set{r \leq 30 r_{\EH}}, \label{eq:bg-uv:mu-small-r} \\
	0 \leq \dvr \leq & \,2  \quad \hbox{ in } \calX_{1} \cap \set{r \leq 30 r_{\EH}}.	\label{eq:bg-uv:dvr-small-r}
\end{align}
Moreover, there exist $b_{1}, b_{1}', B_{1} > 0$, which depend on $|\e| / \varpi_{f}$, so that
\begin{equation} \label{eq:bg-uv:dur-small-r}
	B_{1}^{-1} e^{- b'_{1} r_{\EH}^{-1} (u - v - C_{30})} \leq -\dur \leq B_{1} e^{- b_{1} r_{\EH}^{-1} (u - v - C_{30})}  \quad \hbox{ in } \calX_{1} \cap \set{r \leq 30 r_{\EH}}.
\end{equation}
\item Finally, the following decay estimates hold for $\log (-\gmm)$ and $\log \kpp$:

\begin{align} 
	\abs{\log (-\gmm)} \leq & C \brk{\tfrac{u - u_{1}}{r_{\EH}}}^{-2 \omg + 2} (\tfrac{r}{r_{\EH}} )^{-2} \dlt_{1}^{2} \quad \hbox{ in } \calX_{1} \cap \set{r \geq 10 r_{\EH}},	\label{eq:bg-uv:gmm-decay} \\
	\abs{\log \kpp} \leq & \left\{
	\begin{array}{cl}
	C_{\varpi_{f}^{-1} \abs{\e}} \brk{\tfrac{u - C_{30}}{r_{\EH}}}^{-2 \omg} \dlt_{1}^{2} & \hbox{ in } \calX_{1} \cap \set{r \leq 30 r_{\EH}}	, \\
	C \left( \brk{\tfrac{v - v_{1}}{r_{\EH}}}^{-2 \omg} + \brk{\tfrac{u - u_{1}}{r_{\EH}}}^{-2 \omg + 1} \Lmb^{-1} \right) \dlt_{1}^{2} & \hbox{ in } \calX_{1} \cap \set{r \geq 10 r_{\EH}}	,
	\end{array}\right. 	\label{eq:bg-uv:kpp-decay} 
\end{align}
where $C_{\varpi_{f}^{-1} \abs{\e}} > 0$ depends on $\varpi_{f}^{-1} \abs{\e}$.
\end{enumerate}
\end{proposition}

\begin{remark} \label{rem:bg-uv-dim0}
Note that Proposition~\ref{prop:bg-uv} is formulated so that the assumptions and the estimates are dimensionless, i.e., they are invariant under the scaling $(\Omg, r, \phi, \e) \mapsto (a \Omg, a r, \phi, a \e)$ for any constant $a > 0$. Moreover, they are invariant under coordinate translations $(u, v) \mapsto (u+u_{0}, v + v_{0})$.
\end{remark}
\begin{remark} 
Let $\dlt_{1} > 0$. As a quick consequence of Corollary~\ref{cor:DR-final} (where we take $\omg' = \omg$), we may find $\calX_{1} \subseteq \calX_{0}$ such that the assumptions \eqref{eq:bg-uv:phi}--\eqref{eq:bg-uv:dvphi} hold with $\Lmb = 1$ for any $\omg < \omg_{0}$; the conclusions of Proposition~\ref{prop:bg-uv} in this case will be useful in Section~\ref{sec:blowup}. On the other hand, the freedom of having an extra parameter $\Lmb$ will be useful in Section~\ref{sec:L-stability}; see Definition~\ref{def:dlt-adm} and Proposition~\ref{prop:bg-geom}.
\end{remark}

\begin{proof}
By Remark~\ref{rem:bg-uv-dim0}, we may assume that $r_{\EH} = 1$ by scaling. By translating the coordinates (Remark~\ref{rem:coord-trans}), we also set\footnote{Notice that this convention in the proof is \underline{different} from that in the earlier parts of this section, where we have instead set $u_{\calX_0}=v_{\calX_0}=1$.} $u_{1} = v_{1} = 1$, so that $\brk{\tfrac{u - u_{1}}{r_{\EH}}} = u$ and $\brk{\tfrac{v - v_{1}}{r_{\EH}}} = v$. 

\pfstep{Step~1: Bounds on $\varpi$ and $\mu$}
In this step, we prove \eqref{eq:bg-uv:m-bnd}, \eqref{eq:bg-uv:mu-large-r} and the upper bound in \eqref{eq:bg-uv:mu-small-r}, which concern $\varpi$ and $\mu$.

We begin with \eqref{eq:bg-uv:m-bnd}. By the monotonicity properties of $\varpi$, note that $\inf_{\calX_{1}} \varpi = \varpi(\infty, 1)$ and $\sup_{\calX_{1}} \varpi = \varpi(1, \infty)$. By Corollary~\ref{varpi.same.limit}, $\varpi_{f} = \sup_{\EH} \varpi = \inf_{\NI} \varpi$. Thus, to prove \eqref{eq:bg-uv:m-bnd}, it suffices to show
\begin{equation*}
	\varpi_{f} - \varpi(\infty, 1) \leq C \dlt_{1}^{2}, \quad
	\varpi(1, \infty) - \varpi_{f} \leq C \dlt_{1}^{2}.
\end{equation*}
These bounds follow from \eqref{eq:bg-uv:duphi}, \eqref{eq:bg-uv:dvphi} and the equations
\begin{align*}
\varpi_{f} - \varpi(\infty, 1)
= \int_{1}^{\infty} \frac{1}{2} \kpp r^{2} \left( \frac{1}{\kpp} \rd_{v} \phi \right)^{2} (\infty, v) \, \ud v, \quad
\varpi(1, \infty) - \varpi_{f}
= \int_{1}^{\infty} \frac{1}{2} \left( \frac{r}{\dur} \rd_{u} \phi \right)^{2} (u, \infty) \, \ud u,
\end{align*}
which are straightforward to justify.

By \eqref{r.varpi.poly} and \eqref{eq:subextremality}, $|\e| < \varpi_{f} \leq r_{\EH} = \varpi_{f} + \sqrt{\varpi_{f}^{2} - \e^{2}}$. Therefore, \eqref{eq:bg-uv:mu-large-r} and the upper bound in \eqref{eq:bg-uv:mu-small-r} easily follow from \eqref{eq:bg-uv:r-bnd} and \eqref{eq:bg-uv:m-bnd} with $\dlt_{1} > 0$ sufficiently small (as a universal constant).

\pfstep{Step~2: Bounds away from $\EH$}
In this step, we establish the bounds \eqref{eq:bg-uv:dur-large-r}, \eqref{eq:bg-uv:dvr-large-r} and \eqref{eq:bg-uv:u-v}, which take place in the region $\calX_{1} \cap \set{r \geq 10}$ away from $\EH$. In the process, we also prove \eqref{eq:bg-uv:kpp-bnd} and \eqref{eq:bg-uv:gmm-decay}.

It will be convenient at this point to define $u_{10}(v)$ and $u_{30}(v)$ so that 
$$r(u_{10}(v),v)=10,\quad r(u_{30}(v),v)=30.$$
When such a $u_{10}(v)$ [resp. $u_{30}(v)$] does not exist, we just set $u_{10}(v) = u_{1}$ [resp. $u_{10}(v) = u_{1}$].

We claim that
\begin{align} 
	\abs{\log (-\gmm)} \leq & C u^{-2 \omg + 2} r^{-2} \dlt_{1}^{2} \quad \hbox{ in } \calX_{1} \cap \set{r \geq 10}, 	\label{eq:bg-uv:gmm-large-r} \\
	\abs{\log \kpp} \leq & C \dlt_{1}^{2} \quad \hbox{ in } \calX_{1}. \label{eq:bg-uv:kpp-bnd-pf}
\end{align}
Note that \eqref{eq:bg-uv:gmm-large-r} is precisely \eqref{eq:bg-uv:gmm-decay}, whereas \eqref{eq:bg-uv:kpp-bnd-pf}, together with $\kappa\restriction_{\EH}=1$ and the fact that $\rd_u\kappa\leq 0$, imply \eqref{eq:bg-uv:kpp-bnd}. Since $\dur = (1-\mu) \gmm$ and $\dvr = (1-\mu) \kpp$, note that \eqref{eq:bg-uv:gmm-large-r} and \eqref{eq:bg-uv:kpp-bnd}, in combination with \eqref{eq:bg-uv:mu-large-r}, imply \eqref{eq:bg-uv:dur-large-r} and \eqref{eq:bg-uv:dvr-large-r}, respectively.

We begin with \eqref{eq:bg-uv:gmm-large-r}. Recall that, by our coordinate normalization, $-\gmm = 1$ on $\NI$. Moreover,
\begin{equation*}
	\rd_{v} \log (-\gmm) = \frac{r}{\dvr} (\rd_{v} \phi)^{2}.
\end{equation*}
Given $(u, v) \in \calX_{1} \cap \set{r \geq 10}$, integrating the previous equation from $(u, \infty)$ to $(u, v)$, then applying \eqref{eq:bg-uv:dvphi} and \eqref{eq:bg-uv:mu-large-r}, we obtain
\begin{align*}
	\abs{\log (-\gmm)}(u, v)
	\leq & \int_{v}^{\infty} (1-\mu)^{-2} r^{-3} \rd_{v} r \left(\frac{r^{2}}{\kpp} \rd_{v} \phi \right)^{2}(u, v') \, \ud v'  \\
	\leq & C u^{-2 \omg + 2} \dlt_{1}^{2} \int_{v}^{\infty} r^{-3} \rd_{v} r(u, v') \, \ud v' 
	\leq  C u^{-2 \omg + 2} r^{-2} \dlt_{1}^{2}.
\end{align*}
Next, we establish \eqref{eq:bg-uv:kpp-bnd-pf}. 
Here we recall the coordinate normalization $\kpp = 1$ on $\EH$, as well as the equation 
\begin{equation*}
	\rd_{u} \log \kpp = \frac{r}{\dur} (\rd_{u} \phi)^{2}.
\end{equation*}
We divide into two cases: $(u, v) \in \calX_{1} \cap \set{r \leq 30}$ or  $(u, v) \in \calX_{1} \cap \set{r \geq 10}$. In the former case, we integrate $\rd_{u} \log \kpp$ from $(\infty, v)$ to $(u, v)$ and estimate using \eqref{eq:bg-uv:duphi} as follows:
\begin{align*}
	\abs{\log \kpp}(u, v)
	\leq \int_{u}^{\infty} r (-\rd_{u} r) \Abs{\frac{1}{\dur} \rd_{u} \phi}^{2} (u', v) \, \ud u'
	\leq C \dlt_{1}^{2} \int_{u_{30}(v)}^{\infty} r (-\rd_{u} r) (u', v) \, \ud u' 
	\leq C \dlt_{1}^{2}.
\end{align*}
Next, when $(u, v) \in \calX_{1} \cap \set{r \geq 10}$, we integrate $\rd_{u} \log \kpp$ from $(u_{10}(v), v)$ to $(u, v)$. Then we have
\begin{equation*}
	\abs{\log \kpp}(u, v)
	\leq \abs{\log \kpp}(u_{10}(v), v)
	+ \int_{u}^{u_{10}(v)} r^{-1} (-\rd_{u} r) \left(\frac{r}{\dur} \rd_{u} \phi \right)^{2}(u', v) \, \ud u'.
\end{equation*}
By the previous case, the first term is bounded by $C \dlt_{1}^{2}$. For the second term, we use $r(u', v) \geq 10$, \eqref{eq:bg-uv:duphi} and \eqref{eq:bg-uv:dur-large-r} to bound
\begin{align*}
\int_{u}^{u_{10}(v)} r^{-1} (-\rd_{u} r) \left(\frac{r}{\dur} \rd_{u} \phi \right)^{2}(u', v) \, \ud u'
\leq & C \dlt_{1}^{2} \int_{1}^{\infty} (u')^{-2 \omg} \, \ud u' 
\leq C \dlt_{1}^{2}.
\end{align*}
This proves \eqref{eq:bg-uv:kpp-bnd-pf}, and thus \eqref{eq:bg-uv:dvr-large-r}.

Finally, we prove \eqref{eq:bg-uv:u-v}. Consider the vector field 
\begin{equation*}
T = \frac{1}{\kpp} \rd_{v} - \frac{1}{\gmm} \rd_{u}
\end{equation*}
which is future-pointing and obeys $T r = 0$; in particular, $T$ is tangent to $\gmm_{30}$. Let $t \mapsto \gmm_{30}(t)$ be the parametrization of $\gmm_{30}$ so that $\gmm_{30}(0) = (u_{30}, v_{30}) \in \rd \calX_{1}$ and $\dot{\gmm}_{30}(t) = T \restriction_{\gmm_{30}(t)}$. Let $C_{30} = u_{30} - v_{30}$. Then $(u - v - C_{30})(\gmm_{30}(0)) = 0$. Moreover, at $\gmm_{30}(t)$ we have
\begin{equation*}
	\Abs{\frac{\ud}{\ud t} (u - v - C_{30})}
	= \abs{T (u - v - C_{30})}
	= \Abs{\frac{1}{\gmm} + \frac{1}{\kpp}}
	\leq \Abs{1 - \frac{1}{(-\gmm)}} + \Abs{1 - \frac{1}{\kpp}}
	\leq C (\abs{\log (-\gmm)} + \abs{\log \kpp}) ,
\end{equation*}
where in the last inequality, we used \eqref{eq:bg-uv:mu-large-r}, \eqref{eq:bg-uv:dur-large-r}, \eqref{eq:bg-uv:dvr-large-r} and the simple bound $\abs{e^{\vtht} - 1} \leq \abs{\vtht} e^{\abs{\vtht}}$. Thus 
\begin{align*}
	(u - v - C_{30})(\gmm_{30}(t))
	\leq & C \int_{0}^{t} (\abs{\log (-\gmm)} + \abs{\log \kpp})(\gmm_{30}(t')) \, \ud t' \\
	\leq & C \left(\sup_{\gmm_{30}} \abs{T u}^{-1} \right) \int_{u_{30}}^{u(\gmm_{30}(t))} \abs{\log (-\gmm)} (u', v_{30}(u')) \, \ud u' \\
	& + C \left(\sup_{\gmm_{30}} \abs{T v}^{-1} \right) \int_{v_{30}}^{v(\gmm_{30}(t))} \abs{\log \kpp} (u_{30}(v'), v') \, \ud v'. 
\end{align*}
Thanks to \eqref{eq:bg-uv:mu-large-r}--\eqref{eq:bg-uv:dur-large-r}, $\abs{T u}^{-1} = -\gmm$ and $\abs{T v}^{-1} = \kpp$ are uniformly bounded on $\gmm_{30}$; hence it remains to bound the integrals of $\abs{\log (-\gmm)}$ and $\abs{\log \kpp}$. For the former, using the equation for $\rd_{v} \log(-\gmm)$, we have
\begin{align*}
	\int_{1}^{\infty} \abs{\log (-\gmm)} (u', v_{30}(u')) \, \ud u' 
	\leq & \int_{1}^{\infty} \int_{u_{30}(v')}^{\infty} (1-\mu)^{-2} r \dvr \Abs{\frac{1}{\kpp} \rd_{v} \phi}^{2}(u', v') \, \ud v' \, \ud u' \\
	\leq & \frac{1}{4} \int_{1}^{\infty} (u')^{-2 \omg +1} \dlt_{1}^{2} \int_{u_{30}(v')}^{\infty} r^{-2} \rd_{v} r (u', v') \, \ud v' \, \ud u' .
\end{align*}
Since $r(u, v_{30}(u)) = 30$, note that $\int_{v_{30}(u')}^{\infty} r^{-2}(\rd_{v} r)(u', v') \, \ud v' \leq 30^{-1}$ for every $u' \geq 1$. Hence the last line is bounded by $C \dlt_{1}^{2}$, which is acceptable. Similarly, using the equation for $\rd_{u} \log \kpp$, we have
\begin{align*}
	\int_{1}^{\infty} \abs{\log \kpp} (u_{30}(v'), v') \, \ud v' 
	\leq & \int_{1}^{\infty} \int_{u_{30}(v')}^{\infty} r (-\dur) \Abs{\frac{1}{\dur} \rd_{u} \phi}^{2}(u', v') \, \ud u' \, \ud v' \\
	\leq & \int_{1}^{\infty} (v')^{-2\omg} \dlt_{1}^{2} \int_{u_{30}(v')}^{\infty} r (-\rd_{u} r) (u', v') \, \ud u' \, \ud v' .
\end{align*}
Since $r(u_{30}(v), v) = 30$, we have $\int_{u_{30}(v')}^{\infty} r(-\rd_{u} r)(u', v') \, \ud u' \leq \frac{1}{2} 30^{2}$ for every $v' \geq 1$. Thus the last line is bounded by $C \dlt_{1}^{2}$, which completes the proof of \eqref{eq:bg-uv:u-v}.

\pfstep{Step~3: Bounds near $\EH$}
Recall that the upper bound in \eqref{eq:bg-uv:mu-small-r} was established in Step~1, and as a consequence, the upper bound in \eqref{eq:bg-uv:dvr-small-r} follows from \eqref{eq:bg-uv:kpp-bnd} proved in Step~2. On the other hand, the lower bound in \eqref{eq:bg-uv:dvr-small-r} is immediate from \eqref{no.trapped.in.1}, and similarly as before, this together with \eqref{eq:bg-uv:kpp-bnd} imply the lower bound in \eqref{eq:bg-uv:mu-small-r}. Therefore, near $\EH$, it only remains to show \eqref{eq:bg-uv:dur-small-r}.
Recall the equation
\begin{equation*}
	\rd_{v} \log (-\dur) = \frac{2 (\varpi - \frac{\e^{2}}{r})}{r^{2}} \kpp.
\end{equation*}
On the one hand, by a similar computation as in the proof of \eqref{eq:DR-small-r:redshift} in Corollary~\ref{cor:DR-small-r} using \eqref{eq:bg-uv:r-bnd}, \eqref{eq:bg-uv:m-bnd} and \eqref{eq:bg-uv:kpp-bnd}, we have
\begin{equation*}
\frac{2 (\varpi - \frac{\e^{2}}{r})}{r^{2}} \kpp \geq \frac{1}{30^{2}}\sqrt{\varpi_{f}^{2} - \e^{2}} =: b_{1} > 0 \quad \hbox{ in } \calX_{1} \cap \set{r \leq 30}.
\end{equation*}
On the other hand, by \eqref{eq:bg-uv:r-bnd}, \eqref{eq:bg-uv:m-bnd} and \eqref{eq:bg-uv:kpp-bnd} we also have
\begin{equation*}
\frac{2 (\varpi - \frac{\e^{2}}{r})}{r^{2}} \kpp \leq b_{1}' \quad \hbox{ in } \calX_{1} \cap \set{r \leq 30},
\end{equation*}
for some universal constant $b_{1}' \geq b_{1}$. Integrating the equation for $\rd_{v} \log (-\dur)$, we obtain
\begin{align*}
	(-\dur)(u, v) = (-\dur)(u, v_{30}(u)) \exp \left( - \int_{v}^{v_{30}(u)} \frac{2 (\varpi - \frac{\e^{2}}{r})}{r^{2}} \kpp (u, v') \, \ud v' \right),
\end{align*}
where $v_{30}(u)$ is defined by $r(u, v_{30}(u)) = 30$. By \eqref{eq:bg-uv:dur-large-r}, we have $\frac{1}{4} \leq -\dur(u, v_{30}(u)) \leq 4$. Moreover, by \eqref{eq:bg-uv:u-v}, we have
\begin{equation*}
	u - v - C_{30} -1 \leq v_{30}(u) - v \leq u - v - C_{30} + 1.
\end{equation*}
Putting together these bounds, \eqref{eq:bg-uv:dur-small-r} follows.

\pfstep{Step~4: Proof of \eqref{eq:bg-uv:kpp-decay}}
Finally we prove \eqref{eq:bg-uv:kpp-decay}, which strengthens the integral bound for $\log \kpp$ used in the proof of \eqref{eq:bg-uv:u-v}.

As in the proof of \eqref{eq:bg-uv:kpp-bnd-pf}, we integrate $\rd_{u} \log \kpp$ from $(\infty, v)$ to $(u, v)$. We divide into two cases: $r(u, v) \leq 30$ and $r(u, v) \geq 10$. In the former case, we use \eqref{eq:bg-uv:duphi} and \eqref{eq:bg-uv:dur-small-r} to estimate
\begin{align*}
	\abs{\log \kpp} (u, v) 
	\leq \int_{u}^{\infty} r (-\dur) \Abs{\frac{1}{\dur} \rd_{u} \phi}^{2}(u', v) \, \ud u' 
\leq 30 B_{1} b_{1}^{-1} \dlt_{1}^{2} e^{- b_{1}(u - v - C_{30})} v^{-2 \omg} .
\end{align*}
If $u - C_{30} \leq 2 v$, then $v^{-2 \omg} \leq  C (u - C_{30})^{-2 \omg}$. On the other hand, if $u - C_{30} \geq 2 v$, then we have the exponential decay $e^{-\frac{1}{2} b_{1} (u - C_{30})}$, which is better than $(u - C_{30})^{-2 \omg}$.

When $r(u, v) \geq 10$, we estimate by \eqref{eq:bg-uv:duphi}
\begin{align*}
	\abs{\log \kpp} (u, v) 
\leq \int_{u_{10}(v)}^{\infty} r (-\dur) \Abs{\frac{1}{\dur} \rd_{u} \phi}^{2}(u', v) \, \ud u' 
	+ \int_{u}^{u_{10}(v)} r (-\dur) \Abs{\frac{1}{\dur} \rd_{u} \phi}^{2}(u', v) \, \ud u' .
\end{align*}
For the first term, we estimate
\begin{align*}
\int_{u_{10}(v)}^{\infty} r (-\dur) \Abs{\frac{1}{\dur} \rd_{u} \phi}^{2}(u', v) \, \ud u' 
\leq v^{-2 \omg} \dlt_{1}^{2} \int_{u_{10}(v)}^{\infty} r (- \rd_{u} r) (u', v) \, \ud u'
\leq C v^{-2 \omg} \dlt_{1}^{2}.
\end{align*}
On the other hand, by \eqref{eq:bg-uv:duphi} and \eqref{eq:bg-uv:dur-large-r} we have 

\begin{align*}
\int_{u}^{u_{10}(v)} r (-\dur) \Abs{\frac{1}{\dur} \rd_{u} \phi}^{2}(u', v) \, \ud u' 
\leq \Lmb^{-1} \dlt_{1}^{2} \int_{u}^{u_{10}(v)}  r^{-1} (- \rd_{u} r) (u', v) (u')^{-2 \omg} \, \ud u'
\leq C u^{-2 \omg + 1} \Lmb^{-1} \dlt_{1}^{2},
\end{align*}
which is acceptable. \qedhere
\end{proof}

\section{Lower bound on the event horizon: Proof of Theorem~\ref{thm:blowup}} \label{sec:blowup}
\subsection{Ideas of the proof and the beginning of the proof}\label{sec:blowup:ideas}

In this section, we prove Theorem~\ref{thm:blowup}. The main idea of the proof is to apply the strategy in \cite{LO.instab} to the nonlinear setting. Let us first recall the strategy in \cite{LO.instab}, which is based on proving the contrapositive. Assume that the quantity we wish to show to be infinite is instead finite (see \eqref{main.contradiction}). We will then show that that $\mathfrak L_{(\omg_0)\infty}=0$. The proof consists of four steps:
\begin{enumerate}
\item \pfstep{Step~1} Show that in a region sufficiently close to the event horizon, decay estimates for $\rd_v\phi$ and $\rd_u\phi$ hold. In particular, we obtain integral estimates on the curve $\{r=(1+\de_r) r_{\EH}\}$ restricted to sufficiently large $v$:
$$\int_{v_1}^\infty (v')^{\alp}\left((\rd_v\phi\restriction_{\{r=(1+{\de_r}) r_{\EH}\}})^2(v')+(\rd_u\phi\restriction_{\{r=(1+{\de_r}) r_{\EH}\}})^2(v')\right)\, \ud v' \leq C,$$
for $v_1$ large and ${\de_r}$ small. This is achieved using the red-shift estimates of Dafermos--Rodnianski \cite{DRPL, DRS}.
\item \pfstep{Step~2} Show that up to any finite $\{r=\Lambda r_{\EH}\}$ curve, $\rd_v\phi$ obeys an integrated estimate similar to that in Step~1 and $\phi$ obeys the pointwise bound $|\phi|(u,v)\leq C v^{-\f{\alp}{2}+\f 12}$, with constants depending on $\Lambda$.
\item \pfstep{Step~3} Show that away from the event horizon, say for $r(u,v) \geq 2r_{\EH}$, $\phi$ obeys the uniform bound $|\phi|(u,v) \leq C u^{-3}$. (The difference between this estimate and that in the previous step is that the constant in this step is \emph{uniform} while the constant in the previous step depends on the choice of the $\{r=\Lambda r_{\EH}\}$ curve.)
\item \pfstep{Step~4} Using Step~3, show that if $\mathfrak L_{(\omg_0)\infty}\neq 0$, then $|\rd_v(r\phi)|$ obeys the following \emph{lower bound} on the $\{r=R_{\calI^+}\}$ curve for sufficiently large $R_{\calI^+}$:
$$\left|\f{\rd_v(r\phi)}{\rd_v r}\restriction_{\{r=R_{\calI^+}\}}\right| \geq C^{-1} \left|\mathfrak L_{(\omg_0)\infty}\right|v^{-\min\{\omg_0,3\}}.$$
Hence if $\mathfrak L_{(\omg_0)\infty}\neq 0$, this leads to a contradiction with the conclusion of Step~2 (noting that $\alp>\min\{2\omg_0+1, 7\}$).
\end{enumerate}

All of the steps above are very robust: not only do they not require the spacetime to be exactly Reissner--Nordstr\"om, in fact they only rely very little on the exact geometric properties. Namely, the argument essentially only requires that the spacetime metric is spherically symmetric, and that the spacetime can be divided into the following regions:
\begin{enumerate}
\item A red-shift region near the event horizon, where $\rd_v\log (-\rd_{u} r)>0$ and bounded uniformly away from $0$.
\item A region of bounded geometry, with uniformly bounded $r$, $\rd_{v} r$ and $\rd_{u} r$.
\item An asymptotic region, where the metric is close to that of Minkowski, and that the scalar field decays with a rate that is almost as fast as that of Price's law.
\end{enumerate}
The estimates in Section~\ref{sec:bg} based on \cite{DRPL} are more than sufficient to justify such a partition of the spacetime with the desired properties even in the nonlinear setting.

\textbf{We now begin the proof of Theorem~\ref{thm:blowup}.} Let $\omg_0>2$. From now on, \textbf{fix an $\omg_0$-admissible initial data set.} Consider the connected component $\Ext_{1}$ of the exterior region of its maximal globally hyperbolic future development (cf. Figure~\ref{fig:Kommemi}).
Since it will be clearer for the argument in Sections~\ref{subsec:blowup-const-r} and \ref{sec.contra.conclude} to consider a general $\omg_0>2$, in this section, as opposed to Section~\ref{sec:bg}, \textbf{we will \underline{not} assume $\omg_0\leq 3$}. We can of course still apply the results in Section~\ref{sec:bg} as long as we replace $\omg_0$ by $\min\{\omg_0,3\}$ at appropriate places.

Let $A_{contra}\in [0,\infty)$ be a constant for the contradiction argument and introduce the following:

\textbf{Main contradiction assumption:} Consider the future-normalized coordinate system $(u,v)$ of the exterior region defined by \eqref{eq:coord-u-bg-0} and \eqref{eq:coord-v-bg-0} with $u_{\calX_0}=1$, $v_{\calX_0}(=V_{\calX_0})=1$ where $\calX_0$ is a region as in Section~\ref{subsec:DR-full} such that $C_{out}\subseteq \underline{\calN}$. Assume, for the sake of contradiction, that there exists $v_0\geq 1$ such that the following estimate holds
\begin{equation}\label{main.contradiction}
\int_{v_0}^\infty \lim_{u\to \infty} v^{\alp}(\rd_v\phi)^2(u,v)\, \ud v \leq A_{contra}
\end{equation}
for some $\alp>\min\{2\omg_0+1, 7\}$.

To carry out the proof, we perform Steps 1-4 in the linear setting \cite{LO.instab} as described in the beginning of this subsection. A rough outline of the rest of the section is as follows:
\begin{itemize}
\item \textbf{Section~\ref{subsec:blowup-pf}.} Steps~1 and 2 are carried out in this subsection, in Propositions~\ref{red.shift.prop} and \ref{contra.large.r} respectively.
\item \textbf{Section~\ref{sec.contra.unif}.} Step~3 is carried out in Proposition~\ref{prop.contra.unif}.
\item \textbf{Section~\ref{subsec:blowup-const-r}.} Step~4 is carried out in Proposition~\ref{contra.lower}.
\item \textbf{Section~\ref{sec.contra.conclude}}. Finally, we combine the conclusions of Steps~2 and 4 to conclude the proof.
\end{itemize}

Let us note that in different parts of the argument, it will be convenient to shift between the $(u,v)$ and the $(u,V)$ coordinate systems. This will be clearly specified at different steps below.

\subsection{Upper bound on a constant-$r$ curve}	\label{subsec:blowup-pf}

\textbf{In this subsection, we will work in the $(u,v)$ coordinate system}, i.e., the future-normalized coordinate system introduced in Section~\ref{subsec:bg-coords} with the ``initial conditions'' $u_{\calX_0}=1$, $v_{\calX_0}(=V_{\calX_0})=1$ (cf. Section~\ref{subsec:DR-full}). The advantage is that we can apply the bounds in Proposition~\ref{prop:bg-uv} when restricting to a region with sufficiently large $u$ and $v$.

\begin{proposition}[Red-shift estimates]\label{red.shift.prop}
Suppose the contradiction assumption \eqref{main.contradiction} holds. Then for some $\de_r>0$ sufficiently small and for some $v_1$ sufficiently large (both depending on the solution), there exists $B_{contra}>0$ depending on $A_{contra}$, $\alp$ and the solution such that the following estimates hold:
\begin{equation}\label{red.shift.prop.dphi}
\int_{v_1}^\infty v^{\alp}\left((\rd_v\phi\restriction_{r=(1+\de_r)r_{\EH}})^2+(\rd_u\phi\restriction_{r=(1+\de_r)r_{\EH}})^2 \right)(v) \, \ud v <B_{contra}^2,
\end{equation}
and 
\begin{equation}\label{red.shift.prop.phi}
\left|\phi\restriction_{r=(1+\de_r)r_{\EH}}\right|\leq B_{contra} v^{-\f{\alp}{2}+\f 12}.
\end{equation}
\end{proposition}
\begin{proof}
\pfstep{Step~0: Restriction to a large-$v$ region}
Clearly, the result only concerns a large-$v$ region. Moreover, when restricted to the region $\{r\leq (1+\de_r) r_{\EH}\}\cap \{v\geq v_1\}$, the $u$-value must also be large since $\{r=(1+\de_r) r_{\EH}\}$ is timelike. Therefore, for the rest of this proof, we can assume $v_1$ to be sufficiently large so that Proposition~\ref{prop:bg-uv} is applicable.

Let $\de_r>0$ be the small constant in the statement of this proposition (which is to be chosen later). By the above discussions, we can choose $v_1\geq v_0$ to be sufficiently large such that in the region of interest
\begin{equation}\label{red.shift.r.bounds}
(1-\de_r)r_{\EH}\leq r \leq (1+\de_r)r_{\EH}.
\end{equation}
We will assume these bounds for the rest of this proof.

\pfstep{Step~1: First estimate for the wave equation}
By the wave equation for $\phi$ and $r$ in \eqref{eq:EMSF-r-phi-m}, we derive
$$\rd_v\left(\f{r\rd_u\phi}{(-\nu)^{\f 12}}\right)=-\f{2 (\varpi-\f{\e^2}{r})}{2 r^2}\f{\rd_v r}{1-\mu}\left(\f{r\rd_u\phi}{(-\nu)^{\f 12}}\right)+(-\nu)^{\f 12}\rd_v\phi.$$
By \eqref{eq:DR-small-r:redshift}, for 
$$\Lambda_{\de_r}=1+\de_r,$$ there exists $v_1\geq 1$ and $\tilde{b}>0$ such that
$$\f{2 (\varpi-\f{\e^2}{r})}{2r^2}\f{\rd_v r}{1-\mu}\geq \tilde{b}\quad \mbox{ in }\{r\leq \Lambda_{\de_r} r_{\EH}\}\cap\{v\geq v_1\}.$$
Moreover, $\tilde{b}$ can be chosen to be independent of ${\de_r}$ as long as ${\de_r}<1$.

This implies that for $v\geq v_1$ and $r\leq \Lambda_{\de_r} r_{\EH}$, we have
\begin{equation}\label{red.shift.wave.1}
\begin{split}\
\f 12\rd_v\left(v^{\alp}\left(\f{r\rd_u\phi}{(-\nu)^{\f 12}}\right)^2\right)
\leq &\left(-\tilde{b} v^{\alp}+\f{\alp}{2}v^{\alp-1}\right)\left(\f{r\rd_u\phi}{(-\nu)^{\f 12}}\right)^2+v^{\alp} r\rd_u\phi\rd_v\phi.
\end{split}
\end{equation}
For $(u, v) \in \set{v \geq v_{1}, \, r \leq \Lmb_{\dlt_{r}} r_{\EH}}$, we introduce the notation so that $u_{\Lambda_{\de_r}}(v)$ and $v_{\Lambda_{\de_r}}(u)$ are defined by the following relations:
$$r(u_{\Lambda_{\de_r}}(v), v)=r(u, v_{\Lambda_{\de_r}}(u))=\Lambda_{\de_r} r_{\EH}.$$
Integrating \eqref{red.shift.wave.1} in $v$ from $(u,v_1)$ to $(u, v)$ for all $u\in [u_{\Lambda_{\de_r}}(v_1),\infty)$ and $v\in [v_1, v_{\Lambda_{\de_r}}(u)]$, and taking supremum over $v$, we obtain the following estimate for every $u\in [u_{\Lambda_{\de_r}}(v_1),\infty)$:
\begin{equation}\label{red.shift.wave.1.1}
\begin{split}
&\sup_{v\in [v_1, v_{\Lambda_{\de_r}}(u)]}\f {v^{\alp}}2 \left(\f{r\rd_u\phi}{(-\nu)^{\f 12}}\right)^2(u,v)+\int_{v_1}^{v_{\Lambda_{\de_r}}(u)} \left(\tilde{b} v^{\alp}-\f{\alp}{2}v^{\alp-1}\right)\left(\f{r\rd_u\phi}{(-\nu)^{\f 12}}\right)^2(u,v)\,\ud v\\
\leq & v_1^{\alp}\left(\f{r\rd_u\phi}{(-\nu)^{\f 12}}\right)^2(u,v_1)+2\int_{v_1}^{v_{\Lambda_{\de_r}}(u)} v^{\alp} \left|r\rd_u\phi\rd_v\phi\right|(u,v)\,\ud v.
\end{split}
\end{equation}
Notice that we can choose $v_1$ sufficiently large (depending on $\tilde{b}$ and $\alp$) so that for $v\geq v_1$,
\begin{equation}\label{red.shift.wave.1.2}
\tilde{b} v^{\alp}-\f{\alp}{2}v^{\alp-1}\geq \f {\tilde{b}}4 v^{\alp} .
\end{equation}
Integrating \eqref{red.shift.wave.1.1} in $u$ for $u\in [u_{\Lambda_{\de_r}}(v_1),\infty)$ and using \eqref{red.shift.wave.1.2}, we obtain
\begin{equation}\label{red.shift.wave.1.3}
\begin{split}
&\int_{u_{\Lambda_{\de_r}}(v_1)}^\infty\sup_{v\in [v_1, v_{\Lambda_{\de_r}}(u)]}\f {v^{\alp}}2 \left(\f{r\rd_u\phi}{(-\nu)^{\f 12}}\right)^2(u,v)\, \ud u\\
&+\f {\tilde{b}}4\int_{u_{\Lambda_{\de_r}}(v_1)}^\infty\int_{v_1}^{v_{\Lambda_{\de_r}}(u)}  v^{\alp}\left(\f{r\rd_u\phi}{(-\nu)^{\f 12}}\right)^2(u,v)\,\ud v\, \ud u\\
\leq &  \underbrace{\int_{u_{\Lambda_{\de_r}}(v_1)}^\infty v_1^{\alp}\left(\f{r\rd_u\phi}{(-\nu)^{\f 12}}\right)^2(u,v_1)\, \ud u}_{=:I}+ 2\underbrace{\int_{u_{\Lambda_{\de_r}}(v_1)}^\infty\int_{v_1}^{v_{\Lambda_{\de_r}}(u)} v^{\alp} \left|r\rd_u\phi\rd_v\phi\right|(u,v)\,\ud v\, \ud u}_{=:II}.
\end{split}
\end{equation}
To bound $I$ in \eqref{red.shift.wave.1.3}, notice that since $r\leq \Lambda_{{\de_r}} r_{\EH}\leq 30 r_{\EH}$, we can apply \eqref{eq:DR-final:duphi} (with $\omg'=2$) and \eqref{eq:DR-final:vV} to show that for some $C=C(B_0, B_{\omg'=2})$ (which may change from line to line), 
\begin{equation}\label{red.shift.wave.1.4}
\begin{split}
I\leq C v_1^{\alp-4} \int_{(1-{\de_r})r_\EH}^{(1+{\de_r})r_\EH} r^2\, dr\leq  C v_1^{\alp-4}\left( (1+\de_r)^3 r_{\EH}^3-(1-{\de_r})^3 r_{\EH}^3\right) \leq C v_1^{\alp-4}{\de_r} r_{\EH}^3
\end{split}
\end{equation}
for ${\de_r}>0$ sufficiently small.

For $II$, we use the Cauchy--Schwarz inequality, Young's inequality and H\"older's inequality to obtain the following estimate for some universal constant $C>0$:
\begin{equation}\label{red.shift.wave.1.5}
\begin{split}
II\leq &\underbrace{\f {\tilde{b}}8\int_{u_{\Lambda_{\de_r}}(v_1)}^\infty\int_{v_1}^{v_{\Lambda_{\de_r}}(u)}  v^{\alp}\left(\f{r\rd_u\phi}{(-\nu)^{\f 12}}\right)^2(u,v)\,\ud v\, \ud u}_{=:III} \\
&+\underbrace{C \tilde{b}^{-1}\left(\int_{v_1}^{\infty} \sup_{u\in [u_{\Lambda_{\de_r}}(v),\infty)} v^{\alp}(\rd_v\phi)^2 (u,v)\, \ud v\right) \left(\sup_{v\in [v_1,\infty)}\int_{u_{\Lambda_{\de_r}}(v)}^\infty (-\nu) (u,v)\, \ud u\right)}_{=:IV}.
\end{split}
\end{equation}
For the term $IV$ in \eqref{red.shift.wave.1.5}, note that
\begin{equation}\label{red.shift.wave.1.6}
\begin{split}
\sup_{v\in [v_1,\infty)}\int_{u_{\Lambda_{\de_r}}(v)}^\infty (-\nu) (u,v)\, \ud u\leq \int_{(1-{\de_r})r_{\EH}}^{(1+{\de_r})r_{\EH}} \, \ud r\leq 2{\de_r} r_{\EH}.
\end{split}
\end{equation}
On the other hand, for the term $III$, since the constant $\f{\tilde{b}}{8}$ in front of the integral is sufficiently small, we can absorb this term by the second term on the LHS of \eqref{red.shift.wave.1.3}. Hence, combining \eqref{red.shift.wave.1.3}, \eqref{red.shift.wave.1.4}, \eqref{red.shift.wave.1.5}, \eqref{red.shift.wave.1.6} and the above observation, we obtain 
\begin{equation}\label{red.shift.wave.1.7}
\begin{split}
&\int_{u_{\Lambda_{\de_r}}(v_1)}^\infty\sup_{v\in [v_1, v_{\Lambda_{\de_r}}(u)]}\f {v^{\alp}}2 \left(\f{r\rd_u\phi}{(-\nu)^{\f 12}}\right)^2(u,v)\, \ud u
+\f {\tilde{b}}8 \int_{u_{\Lambda_{\de_r}}(v_1)}^\infty\int_{v_1}^{v_{\Lambda_{\de_r}}(u)}  v^{\alp}\left(\f{r\rd_u\phi}{(-\nu)^{\f 12}}\right)^2(u,v)\,\ud v\, \ud u\\
\leq & C v_1^{\alp-4}{\de_r} r_{\EH}^3+ 2C \tilde{b}^{-1} {\de_r} r_{\EH}\left(\int_{v_1}^{\infty} \sup_{u\in [u_{\Lambda_{\de_r}}(v),\infty)} v^{\alp}(\rd_v\phi)^2 (u,v)\, \ud v\right).
\end{split}
\end{equation}

\pfstep{Step~2: Second estimate for the wave equation}
On the other hand, we can also write the wave equation for $\phi$ in \eqref{eq:EMSF-wave} (in a way useful to estimate $\rd_v\phi$) as follows:
$$\rd_u(r\rd_v\phi)=-\lambda \rd_u\phi.$$
This implies
\begin{equation}\label{red.shift.wave.2}
\f 12\rd_u \left(v^{\alp}(r\rd_v\phi)^2\right)=-v^{\alp}\lambda r\rd_u\phi\rd_v\phi.
\end{equation}
Integrating in $u$ for $u\in [u_{\Lambda_{\de_r}}(v),\infty)$, $v\in [v_1,\infty)$, and taking supremum over $u\in [u_{\Lambda_{\de_r}}(v),\infty)$, we obtain the following for every $v\in [v_1,\infty)$:
\begin{equation}\label{red.shift.wave.2.1}
\begin{split}
\sup_{u\in [u_{\Lambda_{\de_r}}(v),\infty)}\f 12 v^{\alp}(r\rd_v\phi)^2(u,v)
\leq & \f 12 v^{\alp}\lim_{u\to \infty}(r\rd_v\phi)^2(u,v)+\int_{u_{\Lambda_{\de_r}}(v)}^\infty v^{\alp}\left|\lambda r\rd_u\phi\rd_v\phi\right|(u,v) \, \ud u.
\end{split}
\end{equation}
Integrating in $v$ for $v\in [v_1,\infty)$, we obtain
\begin{equation}\label{red.shift.wave.2.2}
\begin{split}
&\f 12\int_{v_1}^\infty \sup_{u\in [u_{\Lambda_{\de_r}}(v),\infty)} v^{\alp}(r\rd_v\phi)^2(u,v)\, \ud v\\
\leq & \underbrace{\f 12\int_{v_1}^\infty  v^{\alp}\lim_{u\to \infty}(r\rd_v\phi)^2(u,v)\, \ud v}_{=:V}+\underbrace{\int_{v_1}^\infty\int_{u_{\Lambda_{\de_r}}(v)}^\infty v^{\alp}\left|\lambda r\rd_u\phi\rd_v\phi\right|(u,v) \, \ud u\, \ud v}_{=:VI}.
\end{split}
\end{equation}
By the assumption \eqref{main.contradiction} and the bounds \eqref{red.shift.r.bounds} on $r$, for ${\de_r}$ sufficiently small, we have
\begin{equation}\label{red.shift.wave.2.3}
V\leq \f 12 (1+{\de_r})^2 r_{\EH}^2 A_{contra}\leq r_{\EH}^2 A_{contra}.
\end{equation}
For the term $VI$, we use \eqref{eq:bg-uv:dvr-small-r} to show that $VI\leq II$ (for $II$ in \eqref{red.shift.wave.1.3}). Hence using \eqref{red.shift.wave.1.5} and \eqref{red.shift.wave.1.6}, we have
\begin{equation}\label{red.shift.wave.2.4}
\begin{split}
VI\leq & 2C \tilde{b}^{-1}{\de_r} r_{\EH} \int_{v_1}^\infty \sup_{u\in [u_{\Lambda_{\de_r}}(v),\infty)} v^{\alp}(r\rd_v\phi)^2(u,v)\, \ud v+ \f{\tilde{b}}{8} \int_{v_1}^\infty\int_{u_{\Lambda_{\de_r}}(v)}^\infty v^{\alp} \left(\f{r\rd_u\phi}{(-\nu)^{\f 12}}\right)^2(u,v) \, \ud u\, \ud v.
\end{split}
\end{equation}
Combining \eqref{red.shift.wave.2.2}, \eqref{red.shift.wave.2.3} and \eqref{red.shift.wave.2.4} and absorbing $2C\tilde{b}^{-1} {\de_r} r_{\EH} \int_{v_1}^\infty \sup_{u\in [u_{\Lambda_{\de_r}}(v),\infty)} v^{\alp}(r\rd_v\phi)^2(u,v)\, \ud v$ to the LHS (which is possible for ${\de_r}$ sufficiently small so that $2C\tilde{b}^{-1}{\de_r} r_{\EH}\leq \f 14$), we obtain
\begin{equation}\label{red.shift.wave.2.5}
\begin{split}
\f 14\int_{v_1}^\infty \sup_{u\in [u_{\Lambda_{\de_r}}(v),\infty)} v^{\alp}(r\rd_v\phi)^2(u,v)\, \ud v
\leq & r_{\EH}^2 A_{contra}+ \f{\tilde{b}}{8}\int_{v_1}^\infty\int_{u_{\Lambda_{\de_r}}(v)}^\infty v^{\alp} \left(\f{r\rd_u\phi}{(-\nu)^{\f 12}}\right)^2(u,v) \, \ud u\, \ud v.
\end{split}
\end{equation}

\pfstep{Step~3: Putting the two estimates together}
Adding \eqref{red.shift.wave.1.7} and \eqref{red.shift.wave.2.5}, we obtain
\begin{equation}\label{red.shift.wave.combined.1}
\begin{split}
&\int_{u_{\Lambda_{\de_r}}(v_1)}^\infty\sup_{v\in [v_1, v_{\Lambda_{\de_r}}(u)]}\f {v^{\alp}}2 \left(\f{r\rd_u\phi}{(-\nu)^{\f 12}}\right)^2(u,v)\, \ud u+ (\f 14-2C\tilde{b}^{-1} \de_r r_{\EH})\int_{v_1}^\infty \sup_{u\in [u_{\Lambda_{\de_r}}(v),\infty)} v^{\alp}(r\rd_v\phi)^2(u,v)\, \ud v\\
\leq &r_{\EH}^2 A_{contra}+C v_1^{\alp-4}{\de_r} r_{\EH}^3.
\end{split}
\end{equation}
Fixing $v_1$, the RHS is hence a constant depending on $A_{contra}$, $\alp$ and the solution.

Choosing ${\de_r}$ sufficiently small and using the estimate for $(-\nu)$ in \eqref{eq:bg-uv:dur-small-r}, \eqref{red.shift.wave.combined.1} implies that when restricted to $\{r=(1+{\de_r})r_{\EH}\}$, we have the following bound:
\begin{equation}\label{red.shift.wave.combined.2}
\begin{split}
&\int_{u_{\Lambda_{\de_r}}(v_1)}^\infty v^{\alp} \left(\rd_u\phi\restriction_{r=(1+{\de_r})r_{\EH}}\right)^2(u)\, \ud u+ \int_{v_1}^\infty  v^{\alp}\left(\rd_v\phi\restriction_{r=(1+{\de_r})r_{\EH}}\right)^2(v)\, \ud v \leq B_{contra}^2,
\end{split}
\end{equation}
for some $B_{contra}>0$ depending on $A_{contra}$, $\alp$ and the solution.
To get from \eqref{red.shift.wave.combined.2} to \eqref{red.shift.prop.dphi}, it suffices to compare $\ud u$ and $\ud v$ on $\{r=(1+{\de_r})r_{\EH}\}$. For this, we note that the vector field the vector field $T=\f{1}{\kappa}\rd_v-\f{1}{\gamma}\rd_u$ is tangent to $\{r=(1+{\de_r})r_{\EH}\}$ and hence on $\{r=(1+{\de_r})r_{\EH}\}$, 
$$\ud u= \f{|Tu|}{|Tv|}\, \ud v= -\f{\kappa}{\gamma} \, \ud v.$$
Using the estimates for $\gamma$ and $\kappa$ in \eqref{eq:bg-uv:gmm-decay} and \eqref{eq:bg-uv:kpp-decay} respectively, and taking $B_{contra}$ larger if necessary, we obtain \eqref{red.shift.prop.dphi}.

\pfstep{Step~4: Estimates for $|\phi|$}
To obtain the estimate \eqref{red.shift.prop.phi} for $\phi$, simply notice again that the vector field $T=\f{1}{\kappa}\rd_v-\f{1}{\gamma}\rd_u$ is tangent to $\{r=(1+{\de_r})r_{\EH}\}$ and that if we parametrize $\{r=(1+{\de_r})r_{\EH}\}$ by $t$ so that $Tt=1$, we have $\ud t=\f{|T t|}{|T v| }\, \ud v= \kappa \, \ud v$. By \eqref{eq:bg-uv:kpp-decay}, we have that for some universal $C>0$,
\begin{equation*}
\begin{split}
&\left|\phi\restriction_{r=(1+{\de_r})r_{\EH}}\right|(v) \leq C\int_{v}^\infty \left|T\phi\restriction_{r=(1+{\de_r})r_{\EH}}\right|(v') \, \ud v'\\
\leq &C\left(\int_{v}^\infty (v')^{-\alp}\, \ud v'\right)^{\f 12}\left(\int_{v}^\infty (v')^{\alp}\left((\rd_v\phi\restriction_{r=(1+{\de_r})r_{\EH}})^2+(\rd_u\phi\restriction_{r=(1+{\de_r})r_{\EH}})^2\right)(v') \, \ud v'\right)^{\f 12}\leq B_{contra} v^{-\f{\alp}{2}+\f 12},
\end{split}
\end{equation*}
after taking $B_{contra}$ larger if necessary. \qedhere

\end{proof}

The above proposition gives an estimate up to $\{r=(1+{\de_r})r_{\EH}\}$. In order to obtain a contradiction, we will need to propagate this estimate up to any $\{r=\Lambda r_{\EH}\}$ curve, where $\Lambda$ is fixed but otherwise arbitrarily large. This is achieved in the following proposition in which we perform a ``sideway'' energy estimate.

\begin{proposition}[Estimates up to any large but finite $r$]\label{contra.large.r}
Suppose the contradiction assumption \eqref{main.contradiction} holds. Let ${\de_r}>0$ be as in the conclusion of Proposition~\ref{red.shift.prop} and $\Lambda>1+{\de_r}$ be arbitrary. Then there exist $u_1\geq 1$ (depending on ${\de_r}$, $\Lambda$ and the solution) and $B'_{contra}>0$ (depending on\footnote{Here, $B_{contra}$ is as in Proposition~\ref{red.shift.prop}.} $B_{contra}$, ${\de_r}$, $\Lambda$ and the solution) such that the following integrals\footnote{Here, it is understood that the curves $\{(u,v): r(u,v)=\Lambda' r_{\EH}\}$ are parametrized by their $u$ values. In particular, when restricted to $\{(u,v): r(u,v)=\Lambda' r_{\EH}\}$, $v$ can be viewed as a function of $u$.} along the curves $\{(u,v): r(u,v)=\Lambda' r_{\EH}\}$ are uniformly bounded for $\Lambda'\in [(1+{\de_r}), \Lambda]$:
\begin{equation}\label{E.contra.prop}
\begin{split}
\int_{u_1}^\infty \left(v\restriction_{r=\Lambda' r_{\EH}}\right)^{\alp}\left((\rd_u\phi\restriction_{r=\Lambda' r_{\EH}})^2+(\rd_v\phi\restriction_{r=\Lambda' r_{\EH}} )^2\right)(u)\, \ud u &\leq (B'_{contra})^2.
\end{split}
\end{equation}
Moreover, in the region $\{u\geq u_1\}\cap\{r \in [(1+{\de_r})r_{\EH}, \Lambda r_{\EH}]\}$, $\phi(u,v)$ satisfies the following estimate:
\begin{equation}\label{pointwise.contra}
|\phi|(u,v)\leq B'_{contra} v^{-\f{\alp}{2}+\f 12}.
\end{equation}
\end{proposition}
\begin{proof}
The main idea is to carry out the corresponding the argument in the linear setting \cite{LO.instab}, which is simply a Gr\"onwall argument for the quantity on the LHS of \eqref{E.contra.prop}. The main new ingredient is that this now needs to be coupled with some estimates for the geometric quantities.

\pfstep{Step~1: Preliminary estimate for the geometric quantities} We first need a strictly positive lower bound on $\lambda$ in the region of interest. Note that this does \underline{not} follow from \eqref{eq:bg-uv:dvr-small-r} as such an estimate only holds away from the horizon. 

As long as $r\in [r_{\EH}, \Lambda r_{\EH}]$, since $r_{\EH}>\varpi_{\EH}$, we have
$$\f{d}{dr} \left(\f{2\varpi_{\EH}}{r}-\f{\e^2}{r^2}\right)=-\f{2\varpi_{\EH}}{r^2}+\f{2\e^2}{r^3}\leq -\f{2}{\Lambda^2 r_{\EH}^2}\left(\varpi_{\EH}-\f{\e^2}{\varpi_{\EH}}\right).$$
Hence, for $r\in [(1+{\de_r}) r_{\EH},\Lambda r_{\EH}]$, by the mean value theorem applied on the interval $[r_{\EH}, r]$,
$$\f{2\varpi_{\EH}}{r_{\EH}}-\f{2\varpi_{\EH}}{r}-\f{\e^2}{r^2_{\EH}}+\f{\e^2}{r^2}\geq \f{2}{\Lambda^2 r_{\EH}^2}\left(\varpi_{\EH}-\f{\e^2}{\varpi_{\EH}}\right) (r-r_{\EH})\geq \f{2}{\Lambda^2 r_{\EH}^2}\left(\varpi_{\EH}-\f{\e^2}{\varpi_{\EH}}\right) {\de_r} r_{\EH}.$$
On the other hand, given $\de_1>0$, we have by \eqref{eq:bg-uv:m-bnd} that for $v_1'$ sufficiently large
$$\left|\varpi(u,v)-\varpi_{\EH}\right|\leq C(\de_1)^2 r_{\EH}\quad\mbox{ in }\{v\geq v_1'\}\cap\{r\in [(1+{\de_r})r_{\EH},\Lambda r_{\EH}]\}.$$
Therefore, using \eqref{r.varpi.poly}, we have 
$$1-\mu=\f{2\varpi_{\EH}}{r_{\EH}}-\f{2\varpi}{r}-\f{\e^2}{r^2_{\EH}}+\f{\e^2}{r^2}\geq \f{2}{\Lambda^2 r_{\EH}^2}\left(\varpi_{\EH}-\f{\e^2}{\varpi_{\EH}}\right) {\de_r} r_{\EH} -\f{C(\de_1)^2 r_{\EH}}{r}.$$
By choosing $v_1'$ sufficiently large (depending on ${\de_r}$, $\varpi_{\EH}$, $\e$ and $\Lambda$), we have $\de_1$ sufficiently small so that 
$$1-\mu\geq \f{1}{\Lambda^2 r_{\EH}^2}\left(\varpi_{\EH}-\f{\e^2}{\varpi_{\EH}^2}\right) {\de_r} r_{\EH}\quad\mbox{ in }\{v\geq v_1'\}\cap\{r\in [(1+{\de_r})r_{\EH},\Lambda r_{\EH}].$$
In particular, by \eqref{eq:bg-uv:kpp-decay}, this implies the following lower bound on $\lambda$ 
\begin{equation}\label{E.contra.dvr.lower}
\lambda\geq c_{contra}\quad\mbox{ in }\{v\geq v_1'\}\cap\{r\in [(1+{\de_r})r_{\EH},\Lambda r_{\EH}]\}
\end{equation}
for some $c_{contra}>0$ depending on ${\de_r}$, $\varpi_{\EH}$, $\e$ and $\Lambda$. 

\textbf{We now choose $u_1$} such that Proposition~\ref{prop:bg-uv} applies in $\{u\geq u_1\}\cap\{r\in [(1+{\de_r})r_{\EH},\Lambda r_{\EH}]\}$ and such that
$$u_1\geq \max\{1, u_{(1+{\de_r})r_{\EH}}(v_1),\, u_{(1+{\de_r})r_{\EH}}(v_1')\},$$ 
where $u_{(1+{\de_r})r_{\EH}}(v)$ is defined by $r(u_{(1+{\de_r})r_{\EH}}(v), v)=(1+{\de_r})r_{\EH}$. In particular, \eqref{E.contra.dvr.lower} holds in $\{u\geq u_1\}\cap\{r\in [(1+{\de_r})r_{\EH},\Lambda r_{\EH}]\}$.

\pfstep{Step~2: Proof of \eqref{E.contra.prop}} 
For every $\Lambda'\in [(1+{\de_r}), \Lambda]$ and $v_\infty$ sufficiently large, define
$$E_{contra}(\Lambda'; v_\infty):=\int_{u_1}^{u_{\Lambda' r_{\EH}}(v_\infty)} \left(v\restriction_{r=\Lambda' r_{\EH}}\right)^{\alp}\left((\rd_u\phi\restriction_{r=\Lambda' r_{\EH}})^2+(\rd_v\phi\restriction_{r=\Lambda' r_{\EH}} )^2\right)(u)\, \ud u,$$
where $r(u_{\Lambda' r_{\EH}}(v_\infty), v_{\infty})=\Lambda' r_{\EH}$.

Consider the following consequences of the wave equation:
$$\f 12 \rd_u\left(v^{\alp}\left(\rd_v\phi\right)^2\right)=-\f{v^{\alp}\lambda\rd_u\phi\rd_v\phi}{r}-\f{v^{\alp}\nu(\rd_v\phi)^2}{r}$$
and
$$\f 12 \rd_v\left(v^{\alp}\left(\rd_u\phi\right)^2\right)=\f{\alp}{2}v^{\alp-1}\left(\rd_u\phi\right)^2 -\f{v^{\alp}\nu\rd_u\phi\rd_v\phi}{r}-\f{v^{\alp}\lambda(\rd_u\phi)^2}{r}.$$
Summing the above identities and integrating this with respect to $\ud u\,\ud v$ in the region $\{(u,v): (1+{\de_r})\leq \f{r(u,v)}{r_{\EH}}\leq \Lambda',\, u \geq u_1,\, v\leq v_\infty\}$ for all $\Lambda'\in [(1+{\de_r}),\Lambda]$, we obtain by \eqref{eq:bg-uv:dur-large-r}, \eqref{eq:bg-uv:dvr-large-r}, \eqref{eq:bg-uv:dvr-small-r} and \eqref{eq:bg-uv:dur-small-r} that for some $C>0$ depending on $\de_r$, $\Lambda$, $\e$ and $\varpi_{\EH}$, 
\begin{equation}\label{E.contra.main}
E_{contra}(\Lambda';v_\infty) \leq CE_{contra}(1+\de_r;v_\infty)+C\iint_{\{(u,v): (1+{\de_r})\leq \f{r(u,v)}{r_{\EH}}\leq \Lambda',\, u\geq u_1,\, v\leq v_\infty\}} v^{\alp} (\abs{\rd_{v} \phi}^{2} + \abs{\rd_{u} \phi}^{2}) \, \ud u\, \ud v.
\end{equation}
Notice that \eqref{E.contra.main} holds because the boundary terms on $\{u=u_1\}$ and $\{v=v_\infty\}$ have favorable signs. 
Note also that direct integration of $\rd_{u} \left( v^{\alp} (\rd_{v} \phi)^{2}\right)$ gives a boundary integral on $\set{r = \Lmb' r_{\EH}}$ with line element $\ud v$, but we can easily switch to $\ud u = - \frac{\gmm}{\kpp} \ud v$ as in Step~3 of Proposition~\ref{red.shift.prop}.

It is convenient for the purpose of this proof to consider the $(r,u)$ coordinate system since we need an estimate that is integrated in $u$ for each curve with fixed $r$. The volume element $\ud u \ud v$ transforms as follows:
\begin{equation}\label{contra.vol.form}
\ud u\, \ud v= \f{1}{\lambda} \, \ud r\, \ud u.
\end{equation}
Using \eqref{E.contra.dvr.lower} and \eqref{contra.vol.form} we can change to the $(r,u)$ coordinate system in \eqref{E.contra.main} so that for some constant $C>0$ (depending on $c_{contra}$) we have
\begin{equation*}
E_{contra}(\Lambda';v_\infty) \leq CE_{contra}(1+{\de_r};v_\infty)+C\int_{(1+{\de_r})r_{\EH}}^{\Lambda' r_{\EH}} E_{contra}\left( \f{r}{r_{\EH}};v_\infty \right) \, \ud r.
\end{equation*}
By \eqref{red.shift.prop.dphi} in Proposition~\ref{red.shift.prop}, $E_{contra}(1+{\de_r};v_\infty)$ is uniformly bounded for all $v_\infty$. Gr\"onwall's inequality then implies that $E_{contra}(\Lambda';v_\infty)$ is uniformly bounded for all $v_\infty$. Therefore, \eqref{E.contra.prop} follows.

\pfstep{Step~3: Proof of \eqref{pointwise.contra}} Finally, \eqref{pointwise.contra} can be obtained using \eqref{E.contra.prop} and the argument as in Step~4 of Proposition~\ref{red.shift.prop}. \qedhere
\end{proof}

\subsection{Uniform decay bounds up to $r=\infty$}\label{sec.contra.unif}

In this subsection, we need to prove another upper bound for $\phi$, as stated below in Proposition~\ref{prop.contra.unif}. Notice that along any constant-$r$ curve, the decay rate in Proposition~\ref{prop.contra.unif} is worse than that in Proposition~\ref{contra.large.r}. Nevertheless, the key point of Proposition~\ref{prop.contra.unif} is that the constant in the estimate is independent of $r$ as $r\to \infty$. 

\textbf{In this subsection, we will use the $(u,V)$ coordinate system,} where $u$ is the future-normalized coordinate and $V$ is the initial-data normalized coordinate (see Section~\ref{subsec:bg-coords}), translated so that $u_{\calX_0}=V_{\calX_0}=1$ (see Section~\ref{subsec:DR-full}). The advantage of using the $(u,V)$ coordinate system is that we can apply the estimates in Corollary~\ref{cor:DR-large-r} in the whole of region $\mathcal N$, whereas the estimates in $(u,v)$ coordinate system are only applicable in an asymptotic region near timelike infinity. In order to prove Proposition~\ref{prop.contra.unif}, we need to compare the $u$, $V$ and $r(u,V)$ values in a large-$r$ region. This is given by the following lemma:
\begin{lemma}\label{lem.uv-est}
There exists $R_{\NI,0}$ sufficiently large such that if $r(u,V)\geq R_{\NI,0}$, the following upper and lower bounds for $r$ hold in the $(u,V)$ coordinate system:
\begin{align*}
r(u,V)\leq & R_{\NI,0}+2(V - \Vb_{R_{\NI,0}})-\f 12(u-\ub_{R_{\NI,0}})_++2(u-\ub_{R_{\NI,0}})_-, \\
r(u, V)\geq &R_{\NI,0}+\f 18 (V - \Vb_{R_{\NI,0}})-2(u-\ub_{R_{\NI,0}})_++\f 12(u-\ub_{R_{\NI,0}})_-,
\end{align*}
where $\ub_{R_{\NI,0}}<1$ and $\Vb_{R_{\NI,0}} \geq 1$ are, respectively, the $u$ and $V$ values of the intersection of $\{r=R_{\NI,0}\}$ and the initial hypersurface $\Sigma_0$. Here, we used the notation $(\cdot)_{\pm} = \max \set{\pm (\cdot), 0}$.
\end{lemma}
\begin{proof}
Without loss of generality, we can assume that $\{(u,V):r(u,V)\geq R_{\NI,0}\}\subset \mathcal N$, where $\calN$ is as in Corollary~\ref{cor:DR-large-r}.

By \eqref{eq:adm-id-af} and \eqref{eq:coord-UV}, we have on $\Sigma_0$ that\footnote{We use again the notation, introduced earlier that $\ub(V)$ is chosen such that $(\ub(V),V)\in \Sigma_0$.}
\begin{equation}\label{rdVr.initial}
\f 12-CV^{-1}\leq (\rd_V r)(\ub(V),V)\leq \f 12 +CV^{-1}
\end{equation}
for some solution-dependent $C>0$ when $V$ is sufficiently large. As a consequence, using \eqref{eq:DR:dvr-large-r}, we have
\begin{equation}\label{rdVr.large.r}
\f 18\leq (\rd_V r)(u,V)\leq 2\quad\mbox{in } \{(u,V):r(u,V)\geq R_{\NI,0}\}.
\end{equation}

Let $(\ub_{R_{\NI,0}}, \Vb_{R_{\NI,0}})$ be the point in the intersection of $\{r=R_{\NI,0}\}$ and the initial hypersurface $\Sigma_0$. First integrating along the curve $\{u=\ub_{R_{\NI,0}}\}$ from the initial hypersurface, and then integrating along a constant-$V$ curve and using \eqref{rdVr.large.r} and \eqref{eq:DR:dur-large-r}, we obtain the desired upper and lower bounds for $r(u, V)$.
Finally, we note that by taking $R_{\NI,0}$ larger is necessary, we can assume that $\ub_{R_{\NI,0}}<1$ and $\Vb_{R_{\NI,0}} \geq 1$. \qedhere
\end{proof}

We now turn to the main estimate of this subsection. Let us emphasize again that while $B_{contra}'$ in Proposition~\ref{contra.large.r} depends on $\Lambda$, the constant $C_{contra}$ in the following lemma is uniform in the whole region.
\begin{proposition}\label{prop.contra.unif}
Suppose the contradiction assumption \eqref{main.contradiction} holds.
There exists a constant $C_{contra}>0$ depending on $A_{contra}$ and the solution such that 
$$|\phi|(u,V)\leq C_{contra} u^{-\min\{\omg_0,3\}},$$
whenever $r(u,V)\geq 2 r_{\EH}$ and $u\geq 1$.
\end{proposition}
\begin{proof}
By Proposition~\ref{contra.large.r}, the conclusion of the proposition holds for $2 r_{\EH}\leq r(u,V)\leq \Lambda r_{\EH}$ for any fixed (but otherwise arbitrary) $\Lambda>1$. It therefore suffices to prove the proposition for $\{(u,V):r(u,V)\geq \Lambda r_{\EH}\}$ for some large $\Lambda>1$ to be chosen later. We can assume in particular that \textbf{$\Lambda$ is chosen large enough so that Lemma~\ref{lem.uv-est} can be applied} when $r(u,V)\geq R_{\NI,0}$.

It will be convenient in the proof to also restrict to 
\begin{equation}\label{contra.unif.large.V}
V - \Vb_{R_{\NI,0}}\geq 128 R_{\NI,0}, \quad V \geq 2 \Vb_{R_{\NI,0}}
\end{equation}
which can be guaranteed by taking $\Lambda$ large. \textbf{We will make this assumption for the rest of the proof.}

To obtain the desired estimate, we will simultaneously bound $\phi$ and $\rd_V(r\phi)$. First, notice that, for $V_{\Lambda r_{\EH}}(u)$ being the unique value of $V$ such that $r(u,V_{\Lambda r_{\EH}}(u))=\Lambda r_{\EH}$, we have
\begin{equation}\label{contra.unif.phi.est}
\begin{split}
|\phi|(u,V)\leq &|\phi|(u, V_{\Lambda r_{\EH}}(u))+\f{1}{r(u,V)}\int_{V_{\Lambda r_{\EH}}(u)}^V \f{|\rd_V(r\phi)|}{\rd_V r}\rd_V r (u,V') \ud V'\\
\leq &|\phi|(u, V_{\Lambda r_{\EH}}(u))+\f{r(u,V)-\Lambda r_{\EH}}{r(u,V)} \sup_{V'\in [V_{\Lambda r_{\EH}}(u),V]}\f{|\rd_V(r\phi)|}{\rd_V r}(u,V')\\
\leq & C_{\Lambda} B'_{contra}(\Lambda) u^{-\f{\alp}{2}+\f 12}+\sup_{V'\in [V_{\Lambda r_{\EH}}(u),V]}\f{|\rd_V(r\phi)|}{\rd_V r}(u,V'),
\end{split}
\end{equation}
where we have used \eqref{pointwise.contra} in the last line and the fact that $u$ and $v$ are comparable up to a constant depending on $\Lambda$. (This follows from $Tr=\left(\f{1}{\kappa}\rd_v-\f{1}{\gamma}\rd_u\right)r=0$, \eqref{eq:bg-uv:gmm-decay} and \eqref{eq:bg-uv:kpp-decay}). We used the notation $C_\Lambda$, $B'_{contra}(\Lambda)$ to emphasize the dependence of the constants on $\Lambda$. 

On the other hand, consider the wave equation
$$\rd_u\left(\f{\rd_V(r\phi)}{\rd_V r}\right) = -\f{\rd_V(r\phi)}{(\rd_V r)^2} \frac{2(\varpi - \frac{\e^{2}}{r})}{r^{2}} \frac{\rd_{u} r \rd_{V} r}{1-\mu} + \frac{2(\varpi - \frac{\e^{2}}{r})}{r^{2}} \frac{\rd_{u} r}{1-\mu}\phi.$$
Integrating along a constant-$V$ curve from $(1,V)$ to $(u,V)$, we obtain
\begin{equation}\label{rdvrphi.diff.est.0}
\begin{split}
&V^{\min\{\omg_0,3\}}\left(\f{\rd_V(r\phi)}{\rd_V r}\right)(u,V)
-\underbrace{V^{\min\{\omg_0, 3\}}\left(\f{\rd_V(r\phi)}{\rd_V r}\right)(1,V) e^{-\int_{1}^u \frac{2(\varpi - \frac{\e^{2}}{r})}{r^{2}} \frac{\rd_{u} r }{1-\mu}(u',V) \, \ud u'}}_{:=I(u,V)}\\
= & \underbrace{e^{-\int_{1}^u \frac{2(\varpi - \frac{\e^{2}}{r})}{r^{2}} \frac{\rd_{u} r }{1-\mu}(u',V) \, \ud u'}\int_{1}^u e^{\int_{1}^{u'} \frac{2(\varpi - \frac{\e^{2}}{r})}{r^{2}} \frac{\rd_{u} r }{1-\mu}(u'',V) \, \ud u''} V^{\min\{\omg_0,3\}}\frac{2(\varpi - \frac{\e^{2}}{r})}{r^{2}} \frac{(\rd_{u} r)}{1-\mu}\phi(u',V) \, \ud u'}_{=:II(u,V)}.
\end{split}
\end{equation}
For $\Lambda$ sufficiently large, the integral in the exponential factors appearing in \eqref{rdvrphi.diff.est.0} is bounded above as follows:
\begin{equation}\label{exp.factor}
\begin{split}
&\int_{1}^u \frac{2(\varpi - \frac{\e^{2}}{r})}{r^{2}} \frac{|\rd_{u} r| }{1-\mu}(u',v) \, \ud u'
\leq  \int_{r(u,V)}^{\infty} \frac{2(\varpi_{i} - \frac{\e^{2}}{r'})}{(r')^{2}} \frac{1 }{(1-\f{2\varpi_{i}}{r'}+\f{{\bf e}^2}{r'^2})} \, dr'\leq \f{4\varpi_i}{\Lambda r_{\EH}}\leq \log \f{4}{3}.
\end{split}
\end{equation}
This implies that
\begin{equation}\label{contra.upper.I.est}
\begin{split}
|I(u,V)|\leq \f 43 V^{\min\{\omg_0, 3\}}\left|\f{\rd_V(r\phi)}{\rd_V r}\right|(1,V)\leq \f 83 \f{V^{\min\{\omg_0, 3\}}}{r^{\min\set{\omg_0, 3}}(1,V)} B_{\eta_{\NI}},
\end{split}
\end{equation}
by \eqref{eq:DR:1-mu-large-r} and \eqref{eq:DR:dvrphi-large-r} (recall that $\omg_{0}$ in Section~\ref{sec:bg} corresponds to $\min\set{\omg_{0}, 3}$ here). 

For $II(u,V)$, first note that we can choose $\Lambda$ sufficiently large so that elementary estimates together with \eqref{exp.factor} imply that 
\begin{equation}\label{contra.upper.II.est}
|II(u,V)|\leq 4 \int_{1}^u V^{\min\{\omg_0,3\}}\frac{\varpi_i}{r^{2}} (-\rd_{u} r)|\phi| (u',V) \, \ud u'.
\end{equation}
Let $u_*(V)$ be defined as the unique $u$ value such that 
\begin{equation}\label{u*.def}
r(u_*(V),V)-R_{\NI,0}=\frac {1}{16} (V - \Vb_{R_{\NI,0}}).
\end{equation}
Divide the integral on the RHS of \eqref{contra.upper.II.est} into\footnote{We make this division only in the case $u> u_*(v)$. If $u\leq u_*(v)$, then obviously $II(u,V)$ can be bounded by the first term alone.} $\int_{1}^{u_*(V)}+\int_{u_*(V)}^u$. By \eqref{u*.def} and the monotonicity of $r$, $r(u,V)\geq \f {1}{16}(V - \Vb_{R_{\NI,0}}) \geq \f V {32}$ for $u\leq u_*(V)$. Hence, using \eqref{eq:DR-final:phi}, for $\omg'\in (2,\min\{\omg_0,3\})$, we have
\begin{equation}\label{contra.upper.II.est.1}
\begin{split}
4 \int_{1}^{u_*(V)} V^{\min\{\omg_0,3\}}\frac{\varpi_i}{r^{2}} (-\rd_{u} r)|\phi| (u',V) \, \ud u' \leq &4 B_{\omg'}\int_{1}^{u_*(V)} V^{\min\{\omg_0,3\}}\left(\frac{\varpi_i}{r^{3}}(-\rd_{u} r)(u',V)\right) (u')^{-\omg'+1} \, \ud u'\\
\leq & 8\cdot (32)^3 B_{\omg'} \varpi_i \int_{1}^{u_*(V)} (u')^{-\omg'+1}\, \ud u' \leq \f{2^{18} B_{\omg'} \varpi_i }{\omg'-2}.
\end{split}
\end{equation}
To estimate the integral $\int_{u_*(V)}^u$, we will use the bound \eqref{contra.unif.phi.est} for $\phi$. For this, we need a bound for $V_{\Lambda r_{\EH}}(u_*(V))$. First, the definition \eqref{u*.def} of $u_*(V)$ and lower bound in Lemma~\ref{lem.uv-est} imply that $u_*(V)>\ub_{R_{\NI,0}}$. Then Lemma~\ref{lem.uv-est} and \eqref{contra.unif.large.V} give
\begin{equation}\label{u*.lower.bd}
\frac{V}{64} \leq \frac{1}{32} (V - \Vb_{R_{\NI,0}}) \leq (u_*(V)-\ub_{R_{\NI,0}}) \leq \f {31}8 (V - \Vb_{R_{\NI,0}}).
\end{equation}
Also using Lemma~\ref{lem.uv-est}, and the fact that $u_*(V)\geq \ub_{R_{\NI,0}}$, we obtain
$$\Lambda r_{\EH}=r(u_*(V),V_{\Lambda r_{\EH}}(u_*(V)))\leq R_{\NI,0}+2 (V_{\Lambda r_{\EH}}(u_*(V)) - \Vb_{R_{\NI,0}})-\f 12(u_*(V)-\ub_{R_{\NI,0}}).$$
This implies, using \eqref{u*.lower.bd} and \eqref{contra.unif.large.V}, that
$$V_{\Lambda r_{\EH}}(u_*(V))\geq \f 12(\Lambda r_{\EH}-R_{\NI,0})+\f 14(u_*(V)-\ub_{R_{\NI,0}}) + \Vb_{R_{\NI,0}} \geq -\f 12 R_{\NI,0}+\f 1{128} (V - \Vb_{R_{\NI,0}}) \geq \f{V}{512}.$$
This implies that if $u'\geq u_*(V)$ and $V'\geq V_{\Lambda r_{\EH}}(u')$, we have $V'\geq \f{V}{512}$. Hence, after plugging in \eqref{contra.unif.phi.est}, we obtain
\begin{equation}\label{contra.upper.II.est.2}
\begin{split}
&4 \int_{u_*(V)}^u   V^{\min\{\omg_0,3\}}\frac{\varpi_i}{r^{2}}(-\rd_{u} r)|\phi|(u',V) \, \ud u'\\
\leq & C_{\Lambda} B_{contra}'(\Lambda) V^{\min\{\omg_0,3\}}  \varpi_i \int_{u_*(V)}^u \f{(-\rd_u r)}{r^2} (u')^{-\f{\alp}{2}+\f 12}\,\ud u' \\
&+ 4\int_{u_*(V)}^u   V^{\min\{\omg_0,3\}}\frac{\varpi_i}{r^{2}}(-\rd_{u} r) \sup_{V'\in [V_{\Lambda r_{\EH}}(u'),V]}\f{|\rd_v(r\phi)|}{\rd_v r}(u',V') \, \ud u'\\
\leq & C_{\Lambda} B_{contra}'(\Lambda) V^{\min\{\omg_0,3\}} \left(u_*(V)\right)^{-\f{\alp}{2}+\f 12} \varpi_i (\Lambda r_{\EH})^{-1}\\
&+4\cdot (512)^{\min\{\omg_0,3\}} \varpi_i (\Lambda r_{\EH})^{-1}\left(\sup_{u'\in [u_*(V),u]}\sup_{V'\in [V_{\Lambda r_{\EH}}(u'),V]} \left((V')^{\min\{\omg_0,3\}}\f{|\rd_V(r\phi)|}{\rd_V r}(u',V')\right)\right).
\end{split}
\end{equation}
Choosing $\Lambda$ sufficiently large, we can arrange that $\Lambda r_{\EH}\geq R_{\NI,0}$ and moreover that
$$4\cdot (512)^{\min\{\omg_0,3\}} \varpi_i (\Lambda r_{\EH})^{-1}\leq \f 12.$$
\textbf{At this point we fix $\Lambda$.} 

Next, notice the following inclusion (within the maximal globally hyperbolic future development):
\begin{equation}\label{contra.inc}
\{(u',V'): u'\in [u_*(V),u],\,V'\in [V_{\Lambda r_{\EH}}(u'),V]\} \subset \{(u',V'): r(u', V')\geq \Lambda r_{\EH}\}.
\end{equation}
Combining the bounds \eqref{rdvrphi.diff.est.0}, \eqref{contra.upper.I.est}, \eqref{contra.upper.II.est}, \eqref{contra.upper.II.est.1} and \eqref{contra.upper.II.est.2} and using \eqref{contra.inc}, we have
\begin{equation}
\begin{split}
&\sup_{r(u,V)\geq \Lambda r_{\EH},\, u\geq 1}V^{\min\{\omg_0, 3\}}\left|\f{\rd_V(r\phi)}{\rd_V r}\right|(u,V)\\
\leq& \f 83 \f{V^{\min\{\omg_0, 3\}}}{r^{\omg_0}(1,V)} B_{\eta_{\NI}}+\f{2^{15} B_{\omg'} \varpi_i }{\omg'-2}+ C_{\Lambda} B_{contra}'(\Lambda) V^{\min\{\omg_0,3\}} \left(u_*(V)\right)^{-\f{\alp}{2}+\f 12} \varpi_i (\Lambda r_{\EH})^{-1}\\
&+ \f 12\left(\sup_{r(u',V')\geq \Lambda r_{\EH},\, u\geq 1} \left((V')^{\min\{\omg_0,3\}}\f{|\rd_V(r\phi)|}{\rd_V r}(u',V')\right)\right).
\end{split}
\end{equation}
The first three terms are uniformly bounded: For the first term, this is due to Lemma~\ref{lem.uv-est}; for the second term, boundedness is obvious; for the third term, we use 
$$\left(u_*(V)\right)^{-\f{\alp}{2}+\f 12}\leq \left(\max\{1,1-\ub_{R_{\NI,0}}\}\right)^{\f{\alp}{2}-\f 12}\left(u_*(V)-\ub_{R_{\NI,0}}\right)^{-\f{\alp}{2}+\f 12}$$ 
(which holds since $\ub_{R_{\NI,0}}<1$ and $u_*(V)\geq 1$), \eqref{u*.lower.bd} and $\alp>\min\{2\omg_0+1, 7\}$. Hence after subtracting the last term from both sides, we obtain
\begin{equation}\label{contra.unif.rdvrphi.final}
\sup_{r(u,V)\geq \Lambda r_{\EH},\, u\geq 1}V^{\min\{\omg_0, 3\}}\left|\f{\rd_V(r\phi)}{\rd_V r}\right|(u,V)\leq C_{contra}',
\end{equation}
for some $C_{contra}'>0$.

Returning to \eqref{contra.unif.phi.est} and multiplying by $u^{\min\{\omg_0,3\}}$, we have that for $u\geq 1$ and $r(u,V)\geq \Lambda r_{\EH}$,
\begin{equation}\label{contra.unif.phi.final}
\begin{split}
u^{\min\{\omg_0,3\}}|\phi|(u,V)\leq &C_{\Lambda} B_{contra}'(\Lambda) u^{-\f{\alp}{2}+\f{1}{2}+\min\{\omg_0,3\}}+u^{\min\{\omg_0,3\}}\sup_{V'\in [V_{\Lambda r_{\EH}}(u),V]}\f{|\rd_V(r\phi)|}{\rd_V r}(u,V')\\
\leq &C_{\Lambda}B_{contra}'(\Lambda) u^{-\f{\alp}{2}+\f{1}{2}+\min\{\omg_0,3\}}+\sup_{r(u,V')\geq \Lambda r_{\EH},\, u\geq 1} u^{\min\{\omg_0,3\}}\f{|\rd_V(r\phi)|}{\rd_V r}(u,V').
\end{split}
\end{equation}
To conclude, it remains to show that for $r(u,V')\geq \Lambda r_{\EH}$ and $u\geq 1$, we have $u\leq CV'$. Clearly, since $\rd_V r>0$, it suffices to show this for $r(u,V')=\Lambda r_{\EH}$. By Lemma~\ref{lem.uv-est}, as well as the fact that $\ub_{R_{\NI,0}}<1$ and $\Vb_{R_{\NI,0}} \geq 1$, if $r(u,V')=\Lambda r_{\EH}$, we have
$$u\leq \left(\max\{1, (1-\ub_{R_{\NI,0}})^{-1}\}\right)(u-u_{R_{\NI,0}}) \leq 4 (V' - \Vb_{R_{\NI,0}}) \leq 4 V'.$$
Plugging this into \eqref{contra.unif.phi.final} and using \eqref{contra.unif.rdvrphi.final}, we obtain
\begin{equation}\label{contra.unif.uv.final}
\begin{split}
u^{\min\{\omg_0,3\}}|\phi|(u,V)\leq C_{\Lambda}B_{contra}'(\Lambda) u^{-\f{\alp}{2}+\f{1}{2}+\min\{\omg_0,3\}}+4^{\min\{\omg_0,3\}} C_{contra}'.
\end{split}
\end{equation}
Finally, recalling that $\alp>\min\{2\omg_0+1, 7\}$, we obtain the conclusion. \qedhere
\end{proof}

\subsection{Lower bounds on a constant-$r$ curve} \label{subsec:blowup-const-r}

In this subsection, we turn to lower bounds for $\rd_v(r\phi)$. First, we derive a sufficiently good approximation formula for $\f{\rd_V(r\phi)}{\rd_V r}$ in the asymptotically flat region in Proposition~\ref{prop.rdvrphi.precise}. We then take appropriate limits to obtain the desired lower bound in Proposition~\ref{contra.lower}. The lower bound only holds when $\mathfrak L_{(\omg_0)\infty}\neq 0$. In the next subsection, we will conclude the proof of Theorem~\ref{thm:blowup} by showing that when $\mathfrak L_{(\omg_0)\infty}\neq 0$, this lower bound is inconsistent with the upper bound derived in Proposition~\ref{contra.large.r} and hence the contradiction assumption \eqref{main.contradiction} is false.

In Proposition~\ref{prop.rdvrphi.precise} below, \textbf{we will continue to use the $(u,V)$ coordinate system}. This will allow us to apply Corollary~\ref{cor:DR-large-r} and Lemma~\ref{lem.uv-est}.

\begin{proposition}[Precise estimate for $\rd_v(r\phi)$ near null infinity]\label{prop.rdvrphi.precise}
For every $u_{\calI^+}$ sufficiently large, define $\mathfrak L_{approx}(u_{\calI^+})(V)$ as a function of $V$ as follows: 
$$\mathfrak L_{approx}(u_{\calI^+})(V)=V^{\min\{\omg_0,3\}}\left(\f{\rd_V(r\phi)}{\rd_V r}\right)(u=\underline{u}(V),V)+ \int_{-u_{\calI^+}}^{u_{\calI^+}}  V^{\min\{\omg_0,3\}}\frac{2(\varpi - \frac{\e^{2}}{r})}{r^{2}} \frac{(\rd_{u} r)}{1-\mu}\phi(u',V) \, \ud u'.$$
Then, under the contradiction assumption \eqref{main.contradiction}, the following holds:

For every $\de_{\mathfrak L}>0$, there exists $u_{\calI^+,0}\geq 1$ sufficiently large and $R_{\calI^+}\geq R_{\NI,0}$ sufficiently large such that for every $ u_{\calI^+}\geq u_{\calI^+,0}$, there exists $V_{\calI^+}$ sufficiently large so that the following bound holds as long as $u\geq u_{\calI^+}$, $V \geq V_{\calI^+}$ and $r(u,V)\geq R_{\calI^+}$:
\begin{equation*}
\begin{split}
\left| V^{\min\{\omg_0,3\}}\left(\f{\rd_V(r\phi)}{\rd_V r}\right)(u,V)
- \mathfrak L_{approx}(u_{\calI^+})(V)\right|\leq  \de_{\mathfrak L}.
\end{split}
\end{equation*}
Moreover, all choices of constants may depend on $\omg_0$, $A_{contra}$ and the solution.
\end{proposition}
\begin{proof}
\textbf{We assume $\de_{\mathfrak L}>0$ is given and fixed.} Without loss of generality, we can assume 
\begin{equation}\label{contra.unif.large.V.2}
R_{\calI^+}\geq R_{\NI,0}. 
\end{equation}
In the proof, we allow ourselves to take $V_{\calI^+}$, $u_{\calI^+,0}$, $u_{\calI^+}$ and $R_{\calI^+}$ larger and larger at various steps until we explicitly state that we fix the choice of these constants. Notice that the statement of the proposition requires us to first fix $u_{\calI^+}$ before choosing $V_{\calI^+}$.

\pfstep{Step~1: Integrating the wave equation}
We write the wave equation for $\phi$ in the following form using \eqref{eq:EMSF-r-phi-m}:
$$\rd_u\left(\f{\rd_V(r\phi)}{\rd_V r}\right) = -\f{\rd_V(r\phi)}{(\rd_V r)^2} \frac{2(\varpi - \frac{\e^{2}}{r})}{r^{2}} \frac{\rd_{u} r \rd_{V} r}{1-\mu} + \frac{2(\varpi - \frac{\e^{2}}{r})}{r^{2}} \frac{\rd_{u} r}{1-\mu}\phi.$$
Integrating along a constant-$V$ curve in a manner similar to \eqref{rdvrphi.diff.est.0} but starting instead from $\Sigma_0$, we obtain
\begin{equation}\label{rdvrphi.diff.est}
\begin{split}
&V^{\min\{\omg_0,3\}}\left(\f{\rd_V(r\phi)}{\rd_V r}\right)(u,V)
-\underbrace{V^{\min\{\omg_0,3\}}\left(\f{\rd_V(r\phi)}{\rd_V r}\right)(u=\underline{u}(V),V)e^{-\int_{\underline{u}(V)}^u \frac{2(\varpi - \frac{\e^{2}}{r})}{r^{2}} \frac{\rd_{u} r }{1-\mu}(u',V) \, \ud u'}}_{=:I(u,V)}\\
= & \underbrace{e^{-\int_{\underline{u}(V)}^u \frac{2(\varpi - \frac{\e^{2}}{r})}{r^{2}} \frac{\rd_{u} r }{1-\mu}(u',V) \, \ud u'}\int_{\underline{u}(V)}^u e^{\int_{\underline{u}(V)}^{u'} \frac{2(\varpi - \frac{\e^{2}}{r})}{r^{2}} \frac{\rd_{u} r }{1-\mu}(u'',V) \, \ud u''} V^{\min\{\omg_0,3\}}\frac{2(\varpi - \frac{\e^{2}}{r})}{r^{2}} \frac{(\rd_{u} r)}{1-\mu}\phi(u',V) \, \ud u'}_{:=II(u,V)},
\end{split}
\end{equation} 
where we have used the notation that $\underline{u}(V)$ denotes the unique $u$ value such that $(\underline{u}(V),V)\in \Sigma_0$ (in the $(u,V)$-coordinate system).

\pfstep{Step~2: Estimates for the term $I$ in \eqref{rdvrphi.diff.est}}
Notice that the integrals in the exponential factors appearing in \eqref{rdvrphi.diff.est} is bounded above as follows:
\begin{equation}\label{bd.exp.factor}
\begin{split}
&\int_{\underline{u}(V)}^u \frac{2(\varpi - \frac{\e^{2}}{r})}{r^{2}} \frac{|\rd_{u} r| }{1-\mu}(u',V) \, \ud u'
\leq  \int_{r(u,V)}^{\infty} \frac{2(\varpi_{i} - \frac{\e^{2}}{r'})}{(r')^{2}} \frac{1 }{(1-\f{2\varpi_{i}}{r'}+\f{{\bf e}^2}{r'^2})} \, dr'\leq \f{4\varpi_i}{r(u,V)},
\end{split}
\end{equation}
for $R_{\mathcal I^+}$ sufficiently large and $r(u,V)\geq R_{\mathcal I^+}$. In particular, \eqref{bd.exp.factor} can be made arbitrarily small by choosing $R_{\mathcal I^+}$ sufficiently large.

By \eqref{eq:adm-id-limits} and Lemma~\ref{lem:cauchy-to-char}, $V^{\min\{\omg_0,3\}}\left(\f{\rd_V(r\phi)}{\rd_V r}\right)(u=\underline{u}(V),V)$ is uniformly bounded for large $V$. Hence, by \eqref{bd.exp.factor}, for $R_{\mathcal I^+}$ and $V_{\calI^+}$ sufficiently large, we have
\begin{equation}\label{I.est}
\left| I(u,V) - V^{\min\{\omg_0,3\}}\left(\f{\rd_V(r\phi)}{\rd_V r}\right)(u=\underline{u}(V),V)\right|< \f{\de_{\mathfrak L}}{5},
\end{equation}
whenever $r(u,V)\geq R_{\mathcal I^+}$ and $V \geq V_{\calI^+}$.

\pfstep{Step~3: Preliminary estimates for the term $II$ in \eqref{rdvrphi.diff.est}}
To control $II$, we divide the integral into 4 parts:
\begin{equation}\label{II.divide}
\int_{\underline{u}(V)}^u = \int_{\underline{u}(V)}^{-u_{\calI^+}} +\int_{-u_{\calI^+}}^{u_{\calI^+}} + \int_{u_{\calI^+}}^{\min\{u_*(V),u\}} + \int_{\min\{u_*(V),u\}}^u =:II_1(u,V)+II_2(u,V)+II_3(u,V)+II_4(u,V),
\end{equation}
where $u_*(V)$ is defined as in \eqref{u*.def}, i.e., $r(u_*(V),V)-R_{\NI,0}=\f 1 {16} (V - \Vb_{R_{\NI,0}})$. It follows from \eqref{u*.lower.bd} that for any fixed $u_{\calI^+}\geq 1$, one can choose $V_{\calI^+}$ sufficiently large such that $V_{\NI} \geq 2 \Vb_{R_{\NI,0}}$ and 
\begin{equation}\label{u*.V.comp}
u_*(V)\geq u_{\calI^+}
\end{equation}
for all $V\geq V_{\calI^+}$. In particular, this implies that $V - \Vb_{R_{\NI,0}} \geq \frac{V}{2}$ in the region under consideration, and also that we can indeed divide the integral as in \eqref{II.divide}.

In \eqref{II.divide}, the terms $II_1(u,V)$, $II_3(u,V)$ and $II_4(u,V)$ will be treated as error terms while $II_2(u,V)$ will be treated as a main term. To deal with the error terms, it suffices to obtain an upper bound. For that purpose, we define $II'_i$ for $i=1,3,4$ by 
$$II_i'(u,V):= 4\int   V^{\min\{\omg_0,3\}}\frac{\varpi_i}{r^{2}}(-\rd_{u} r)|\phi|(u',V) \, \ud u',$$
where the integration is over the corresponding interval for the term $II_i(u,V)$ (see \eqref{II.divide}).

By \eqref{bd.exp.factor}, if $R_{\calI^+}$ is sufficiently large, then whenever $r(u,V)\geq R_{\calI^+}$, we have (for $i=1,3,4$)
\begin{equation}\label{II.IIp}
\left|  II_i(u,V)\right|\leq II_i'(u,V).
\end{equation}

\pfstep{Step~4: Estimating the error terms $II'_1$, $II'_3$ and $II'_4$} To handle $II'_1$ and $II'_3$, we use the decay of $r|\phi|$ in $u$ (see \eqref{eq:DR-final:phi} in Corollary~\ref{cor:DR-final}) and in $\underline{\ups}$ (see \eqref{eq:bg-large-r:rphi} in Proposition~\ref{prop:bg-large-r}) respectively. For the term $II'_4$, we will use the decay of $|\phi|$ in $u$ that is proven (under the contradiction assumption) in Proposition~\ref{prop.contra.unif}. 

We first consider $II_1'$ and $II_3'$. By monotonicity of $r$, \eqref{u*.def} implies that $r(u,V)\geq \f{1}{16} (V - \Vb_{R_{\NI,0}}) \geq \frac{V}{32}$ whenever $u\leq u_*(V)$. This in particular holds in the domain of integration of $II'_1$ and $II'_3$. Moreover, choosing $u_{\calI^+,0}$ sufficiently large, the domain of integration of $II'_1$ lies completely within $\underline{\mathcal N}$. Hence, by \eqref{eq:bg-large-r:dur} and \eqref{eq:bg-large-r:rphi},
\begin{equation}\label{IIp1.est}
\begin{split}
|II'_1(u,V)|\leq &4 B_{i^0} \int_{\underline{u}(V)}^{-u_{\calI^+}} V^{\min\{\omg_0,3\}}\left(\frac{\varpi_{i}}{r^{3}}(-\rd_{u} r)\underline\ups^{-\omg_0+1}\right)(u',V)  \, \ud u'\\
\leq &4\cdot (32)^3 B_{i^0} \varpi_{i}\int_{\underline{u}(V)}^{-u_{\calI^+}} \left((-\rd_{u} r)\underline\ups^{-\omg_0+1}\right)(u',V)  \, \ud u'\\
\leq & \f{C B_{i^0} \varpi_{i} }{\omg_0-2} r^{-\omg_0+2}(-U_{\calI^+}, \underline{V}(-U_{\calI^+})),
\end{split}
\end{equation}
for some universal constant $C>0$, where\footnote{We recall the definitions for the $u$ and $U$ coordinates in Section~\ref{subsec:bg-coords}.} $U_{\calI^+}:=U(u_{\calI^+})$.

The term $II'_3$ can be estimated in a similar manner using \eqref{eq:DR-final:phi} and \eqref{eq:bg-uv:dur-large-r}. Here, we choose $u_{\calI^+,0}$ sufficiently large\footnote{Since we only consider a large-$r$ region such that Lemma~\ref{lem.uv-est} applies, choosing $u$ to be large implies also that $v$ is large and hence Proposition~\ref{prop:bg-uv} is applicable.} (and hence $u_{\calI^+}$ sufficiently large) such that \eqref{eq:bg-uv:dur-large-r} in Proposition~\ref{prop:bg-uv} holds. For $\omg'\in (2, \min\{\omg_0,3\})$,
\begin{equation}\label{IIp3.est}
\begin{split}
|II'_3(u,V)|\leq &4 B_{\omg'}\int_{u_{\calI^+}}^{u_*(V)} V^{\min\{\omg_0,3\}}\left(\frac{\varpi_i}{r^{3}}(-\rd_{u} r)(u',V)\right) (u')^{-\omg'+1} \, \ud u'\\
\leq & 4\cdot 4 \cdot (32)^3 B_{\omg'} \varpi_i \int_{u_{\calI^+}}^{u_*(V)} (u')^{-\omg'+1}\, \ud u'\\
\leq & \f{C B_{\omg'} \varpi_i }{\omg'-2} u_{\calI^+}^{-\omg'+2}.
\end{split}
\end{equation}
By \eqref{IIp1.est} and \eqref{IIp3.est}, we can then choose $u_{\calI^+,0}>0$ larger such that 
\begin{equation}\label{IIp.error.est}
|II'_1(u,v)|\leq \f{\de_{\mathfrak L}}{5},\quad |II'_3(u,v)|\leq \f{\de_{\mathfrak L}}{5}
\end{equation}
whenever $u_{\calI^+}\geq u_{\calI^+,0}$. \textbf{At this point, we fix $u_{\calI^+,0}$. We also fix some $u_{\calI^+}$ satisfying $u_{\calI^+}\geq u_{\calI^+,0}$.} 

\textbf{After fixing $u_{\calI^+}$, we fix $V_{\calI^+}$} such that \eqref{I.est}, \eqref{u*.V.comp} and $V_{\NI} \geq 2 \Vb_{R_{\NI,0}}$ hold.

We now turn to the term $II'_4(u,V)$. Note that by \eqref{u*.lower.bd}, for any $u\geq u_*(V)$, we have
$$u^{-1}\leq \left(u_*(V)\right)^{-1}\leq \max\{1,1-\ub_{R_{\NI,0}}\}\left(u_*(V)-u_{R_{\NI,0}}\right)^{-1}\leq 64\max\{1,1-\ub_{R_{\NI,0}}\}V^{-1}. $$
Hence, by Proposition~\ref{prop.contra.unif},
\begin{equation*}
\begin{split}
|II_4'(u,V)|\leq &4 \int_{\min\{u_*(V),u\}}^u   V^{\min\{\omg_0,3\}}\frac{\varpi_i}{r^{2}}(-\rd_{u} r)|\phi|(u',V) \, \ud u'\\
\leq & 4C_{contra} V^{\min\{\omg_0,3\}}  \varpi_i \int_{\min\{u_*(V),u\}}^u \f{(-\rd_u r)}{r^2} (u')^{-\min\{\omg_0,3\}}\,\ud u' \\
\leq & C_{contra}\cdot 4\cdot (64)^{\min\{\omg_0,3\}}\cdot \left(\max\{1,1-\ub_{R_{\NI,0}}\}\right)^{\min\{\omg_0,3\}}  \f{\varpi_i}{r(u,V)}.
\end{split}
\end{equation*}
Choosing $R_{\calI^+}$ to be sufficiently large, $r(u,V)\geq R_{\calI^+}$ then implies
\begin{equation}\label{IIp.error.4.est}
|II_4'(u,V)|\leq \f{\de_{\mathfrak L}}{5}.
\end{equation}

\pfstep{Step~5: Estimating the main term $II_2$}
$II_2$ is the main term in the statement of the proposition. Our goal in this step is to isolate the main contribution in $II_2$ (see the term in the first line of \eqref{II2.est}). To this end, we first note that the difference between $II_2(u,V)$ and the main contribution has an additional factor of $r$ in the integrand. Then it is easy to see that the error can be made small by choosing $R_{\calI^+}$ large.
 
We now turn to the details. First notice that by Lemma~\ref{lem.uv-est}, on any fixed compact interval (of $u$) $J=[-u_{\calI^+}, u_{\calI^+}]$, we have 
\begin{equation}\label{r.V.comp}
r(u,V)\geq C_J V
\end{equation}
for $V\geq 1$ and $u\in J$ for some constant $C_J$ depending on $J$. Therefore, after choosing $R_{\calI^+}$ to be sufficiently large, we have for some constant $C(\varpi_i, \e)$ (depending on $\varpi_i$ and $\e$) and some constant $B$ (depending on the constants in Proposition~\ref{prop:bg-large-r} and Corollary~\ref{cor:DR-final}) that
\begin{equation}\label{II2.est}
\begin{split}
&\left| II_2(u,V)- \int_{-u_{\calI^+}}^{u_{\calI^+}}  V^{\min\{\omg_0,3\}}\frac{2(\varpi - \frac{\e^{2}}{r})}{r^{2}} \frac{\rd_{u} r}{1-\mu}\phi(u',V) \, \ud u'\right|\\
\leq &C(\varpi_i, \e) \int_{-u_{\calI^+}}^{u_{\calI^+}}  V^{\min\{\omg_0,3\}}\frac{\varpi_i }{r^{3}} (-\rd_{u} r)|\phi|(u',V) \, \ud u'\\
\leq & C(\varpi_i, \e) B (C_J)^{-3} \int_{-u_{\calI^+}}^{u_{\calI^+}}  \frac{(-\rd_{u} r)}{r} (u',V) \, \ud u'\\
\leq & \frac{2 C(\varpi_i, \e) B (C_J)^{-3} u_{\calI^+}}{r(u_{\calI^+},V)} 
\leq \f{\de_{\mathfrak L}}{5},
\end{split}
\end{equation}
where in the first inequality we used \eqref{bd.exp.factor}, as well as some trivial estimates on $\f{\e^2}{r}$ and $1-\mu$; in the second inequality, we used the bounds for $|\phi|$ in Proposition~\ref{prop:bg-large-r} and Corollary~\ref{cor:DR-final} and the estimate \eqref{r.V.comp}; in the third inequality, we used \eqref{eq:DR:dur-large-r}.

\textbf{At this point, we fix $R_{\calI^+}$.}

\pfstep{Step~6: Putting everything together} By \eqref{rdvrphi.diff.est}, \eqref{I.est}, \eqref{II.divide}, \eqref{II.IIp}, \eqref{IIp.error.est}, \eqref{IIp.error.4.est} and \eqref{II2.est}, we thus obtain the desired conclusion. \qedhere

\end{proof}

We finally conclude this subsection with a lower bound for $\rd_v(r\phi)$ on the $\{r=R_{\calI^+}\}$ curve. \textbf{For this bound (and its proof), we will switch back to the future-normalized $(u,v)$ coordinate system.}

\begin{proposition}[Lower bound for $\rd_v(r\phi)$ on the $\{r=R_{\calI^+}\}$ curve]\label{contra.lower}
Suppose the contradiction assumption \eqref{main.contradiction} holds. Then there exists $R_{\calI^+}$ (sufficiently large) and $v_{\calI^+}$ sufficiently large such that the following lower bound for $\rd_v(r\phi)$ holds
\begin{equation}\label{contra.lower.goal}
v^{\min\{\omg_0,3\}}\left|\f{\rd_v(r\phi)}{\rd_v r}\restriction_{r=R_{\calI^+}}\right|(u,v)\geq \f 12 \left(\f{1}{B_0}\right)^{\min\{\omg_0,3\}} \left|\mathfrak L_{(\omg_0)\infty}\right|
\end{equation}
for $v\geq v_{\calI^+}$, where $B_0>0$ is as in \eqref{eq:DR-final:vV}. 
\end{proposition}
\begin{proof}
Without loss of generality, we can assume that $\mathfrak L_{(\omg_0)\infty}\neq 0$ for the conclusion is otherwise trivial.

\pfstep{Step~1: Simple geometric bound} As a preliminary step, we show that 
\begin{equation}\label{imp.Vr.bd}
\lim_{V\to\infty} \f{V}{r(u,V)}=2 \mbox{ uniformly for } -u_{\calI^+}\leq u\leq u_{\calI^+},
\end{equation}
where $u_{\calI^+}$ is any fixed constant. (We will choose $u_{\calI^+}$ in Step~2 below.) First, we need a slightly improved estimate over \eqref{eq:DR:dvr-large-r}, which can be proven essentially in the same manner. More precisely, using the equation for $\log (\rd_V r)$, we have
\begin{equation*}
	\left|\rd_{u} \log \rd_{V} r \right|= \left|\frac{\rd_{u} \rd_{V} r}{\rd_{V} r} \right|= \frac{2 (\varpi - \frac{\e^{2}}{r})}{r^{2}} \frac{(-\rd_{u} r)}{1-\mu}\leq 2 \left(1-\f{2\varpi_i}{r}\right)^{-1} \frac{\varpi_{i}}{r^{2}} (-\rd_{u} r).
\end{equation*}
Using the initial bound \eqref{rdVr.initial}, for $V$ sufficiently large, we have
\begin{equation*}
\left(\f 12 -\de_s'(V)\right)e^{-2\varpi_i\int^\infty_{r(u,V)} \f{1}{((r')^2-2\varpi_i r')}\, \ud r'}\leq  (\rd_V r)(u,V)\leq \left(\f 12 +\de_s'(V)\right)e^{2\varpi_i\int^\infty_{r(u,V)} \f{1}{((r')^2-2\varpi_i r')}\, \ud r'},
\end{equation*}
for some $\de_s'(V)\to 0$ as $V\to \infty$. Next, note that by Lemma~\ref{lem.uv-est}, for $u\leq u_{\calI^+}$, $2\varpi_i\int^\infty_{r(u,V)} \f{1}{((r')^2-2\varpi_i r')}\, \ud r'\to 0$ uniformly (in $u$ for $u\leq u_{\NI}$) as $V\to \infty$, which gives
\begin{equation}\label{rdVr.sharp}
\f 12 -\de_s(V)\leq  (\rd_V r)(u,V)\leq \f 12 + \de_s(V),
\end{equation}
for some $\de_s(V)\to 0$ as $V\to \infty$, when $u\leq u_{\calI^+}$. Fix some large $V_0$ such that $-u_{\calI^+}>\ub(V_0)$. Integrating \eqref{rdVr.sharp} from $\{V=V_0\}\cap\{-u_{\calI^+}\leq u\leq u_{\calI^+}\}$ to null infinity, we obtain that for $-u_{\calI^+}\leq u\leq u_{\calI^+}$ and $V\geq V_0$,
$$\inf_{-u_{\calI^+}\leq u'\leq u_{\calI^+}}r(u',V_0)+\left(\f 12-\de_s(V)\right)V\leq  r(u,V)\leq \sup_{-u_{\calI^+}\leq u'\leq u_{\calI^+}}r(u',V_0)+\left(\f 12+\de_s(V)\right)V. $$
Dividing by $V$ and noting that $r(u', V')$ is bounded above and below on the compact set $\{(u',V'):u'\leq u_{\calI^+},\, V'=V_0\}$, we have thus established \eqref{imp.Vr.bd}.

\pfstep{Step~2: Application of Proposition~\ref{prop.rdvrphi.precise}} The estimates in Corollary~\ref{cor:DR-large-r} together with the equations \eqref{eq:EMSF-r-phi-m} imply that the limit $2M(u)\Phi(u)\Gamma(u):=2\lim_{V\to \infty} (\varpi r\phi\nu)(u,V)$ is well-defined for every $u\in \mathbb R$ and it is integrable in $u$. Therefore, given $\de_{\mathfrak L}>0$, there exists $u_{\calI^+}>u_{\calI^+,0}$ sufficiently large such that 
\begin{equation}\label{L.small.tail}
\int_{-\infty}^{-u_{\calI^+}} 2M(u)\Phi(u)\Gamma(u)\, \ud u + \int^{\infty}_{u_{\calI^+}} 2M(u)\Phi(u)\Gamma(u)\, \ud u\leq 2^{-\min\{\omg_0,3\}-1} \de_{\mathfrak L}.
\end{equation}
On the other hand, since $\varpi r\phi\nu$ is uniformly bounded, and $\lim_{V\to \infty}\mu(u,V)=0$, $\lim_{V\to \infty}\f{\e^2}{r}(u,V)=0$, $\lim_{V\to\infty} \f{V}{r(u,V)}=2$ (by \eqref{rdVr.sharp}), by the dominated convergence theorem, on the bounded interval $[-u_{\calI^+}, u_{\calI^+}]$, we have
\begin{equation}\label{L.main.term}
\begin{split}
\lim_{V\to \infty} \int_{-u_{\calI^+}}^{u_{\calI^+}}  V^{\min\{\omg_0,3\}}\frac{2(\varpi - \frac{\e^{2}}{r})}{r^{2}} \frac{(\rd_{u} r)}{1-\mu}\phi(u',V) \, \ud u'
=&\begin{cases}\displaystyle 2^{3}\int_{-u_{\calI^+}}^{u_{\calI^+}} 2 M(u)\Phi(u)\Gamma(u)\, \ud u &\quad \mbox{if }\omg_0\geq 3\\
\displaystyle 0&\quad \mbox{if }\omg_0< 3.
\end{cases}
\end{split}
\end{equation}
Recalling the definitions of $\mathfrak L_{(\omg_0)\infty}$ in \eqref{eq:L-def} and \eqref{eq:Linfty-def}, and of $\mathfrak L_{approx}(u_{\calI^+})(V)$ in Proposition~\ref{prop.rdvrphi.precise}, and using \eqref{L.small.tail} and \eqref{L.main.term}, we know that for every $\de_{\mathfrak L}>0$, there exist $u_{\calI^+}$ and $V_{\calI^+}$ sufficiently large, such that whenever $V\geq V_{\calI^+}$,  
$$\left| \mathfrak L_{approx}(u_{\calI^+})(V)-2^{\min\{\omg_0,3\}}\mathfrak L_{(\omg_0)\infty}\right|\leq \de_{\mathfrak L}.$$
Using Proposition~\ref{prop.rdvrphi.precise} (and choosing $V_{\calI^+}$ larger if necessary), this implies that 
$$\left|V^{\min\{\omg_0,3\}}\left|\f{\rd_V(r\phi)}{\rd_V r}\right|(u,V)-2^{\min\{\omg_0,3\}}\mathfrak L_{(\omg_0)\infty}\right|\leq 2\de_{\mathfrak L}$$
whenever $u\geq u_{\calI^+}$, $V\geq V_{\calI^+}$ and $r\geq R_{\NI}$. Hence, for these $u$ and $V$,
$$V^{\min\{\omg_0,3\}}\left|\f{\rd_V(r\phi)}{\rd_V r}\right|(u,V)\geq 2^{\min\{\omg_0,3\}} \left|\mathfrak L_{(\omg_0)\infty}\right|-2\de_{
\mathfrak L}.$$
Since $\mathfrak L_{(\omg_0)\infty}\neq 0$, we can choose $\de_{\mathfrak L}<\f{2^{\min\{\omg_0,3\}}}{4} \left|\mathfrak L_{(\omg_0)\infty}\right|$ so that
$$V^{\min\{\omg_0,3\}}\left|\f{\rd_V(r\phi)}{\rd_V r}\right|(u,V)\geq 2^{\min\{\omg_0,3\}-1} \left|\mathfrak L_{(\omg_0)\infty}\right|$$
whenever $u\geq u_{\calI^+}$, $V\geq V_{\calI^+}$ and $R\geq R_{\calI^+}$.

\pfstep{Step 3: Changing to the $(u,v)$ coordinate system} 
To obtain the lower bound in the $(u,v)$ coordinate system, first note that $\f{\rd_v(r\phi)}{\rd_v r}=\f{\rd_V(r\phi)}{\rd_V r}$. Then by \eqref{eq:DR-final:vV}, if $V\geq 2$, we have $V\leq 2 (V-1)\leq 2 B_0(v-1)\leq 2 B_0 v$. Finally, we obtain the desired bound \eqref{contra.lower.goal} using Step~2 (and the easy observation that on $\{r=R_{\calI^+}\}$, restricting to large $v$ implies restricting to large $u$). \qedhere
\end{proof}

\subsection{Putting everything together}\label{sec.contra.conclude}

Given the estimates proven in the previous subsections, we easily obtain Theorem~\ref{thm:blowup} as a consequence.
\begin{proof}[Proof of Theorem~\ref{thm:blowup}]
Suppose $\mathfrak L_{(\omg_0)\infty}\neq 0$. Our goal is to obtain a contradiction with \eqref{main.contradiction}. Let $\tilde{v}=\max\{v_{\calI^+}, v_{R_{\calI^+}}(u_1)\}$, where $u_1$ is as in Proposition~\ref{contra.large.r}, $v_{\calI^+}$ is as in Proposition~\ref{contra.lower} and $v_{R_{\calI^+}}(u)$ is defined by $r(u, v_{R_{\calI^+}}(u))=R_{\calI^+}$, for $R_{\calI^+}$ as in Proposition~\ref{contra.lower}. By Proposition~\ref{contra.large.r}, \eqref{eq:bg-uv:dvr-large-r} (for estimating $\rd_v r$) and \eqref{eq:bg-uv:gmm-decay}, \eqref{eq:bg-uv:kpp-decay} (to control the change of variable\footnote{To see how these bounds are used for the change of variable, see for example the end of Step~3 of the proof of Proposition~\ref{red.shift.prop}.} from $\ud v$ to $\ud u$), and the fact $\alp>\min\{2\omg_0+1,\, 7\}$ we obtain
\begin{equation}\label{contra.conclude.1}
\begin{split}
&\int_{\tilde{v}}^\infty v^{2\min\{\omg_0,3\}-1} \left(\rd_v (r\phi)\restriction_{\{r=R_{\calI^+}\}}\right)^2(v)\, \ud v\\
\leq &\int_{\tilde{v}}^\infty v^{2\min\{\omg_0,3\}-1} \left(|(\rd_v r)\phi\restriction_{\{r=R_{\calI^+}\}}|^2(v)+r^2|\rd_v\phi\restriction_{\{r=R_{\calI^+}\}}|^2(v)\right)\, \ud v \\
\leq & 4(B'_{contra})^2\int_{\tilde{v}}^\infty v^{-\alp + 2\min\{\omg_0,3\}}\, \ud v + 2 R_{\calI^+}^2 \int_{u_1}^\infty \left(v\restriction_{\{r=R_{\calI^+}\}}\right)^{2\min\{\omg_0,3\}-1} \left(\rd_v\phi\restriction_{\{r=R_{\calI^+}\}}\right)^2(u)\, \ud u<\infty.
\end{split}
\end{equation}
On the other hand, by Proposition~\ref{contra.lower} and \eqref{eq:bg-uv:dvr-large-r} (for estimating $\rd_v r$), since $\mathfrak L_{(\omg_0)\infty}\neq 0$, we have
\begin{equation}\label{contra.conclude.2}
\begin{split}
\int_{\tilde{v}}^\infty v^{2\min\{\omg_0,3\}-1} \left(\rd_v (r\phi)\restriction_{\{r=R_{\calI^+}\}} \right)^2(v)\, \ud v\geq \left(\f{1}{B_0}\right)^{2\min\{\omg_0,3\}} \left|\mathfrak L_{(\omg_0)\infty}\right|^2 \int_{\tilde{v}}^\infty v^{-1}\, \ud v =\infty.
\end{split}
\end{equation}
Comparison of \eqref{contra.conclude.1} and \eqref{contra.conclude.2} clearly leads to a contradiction. Hence, under the assumption that $\mathfrak L_{(\omg_0)\infty}\neq 0$, the bound \eqref{main.contradiction} does \underline{not} hold. This concludes the proof. \qedhere
\end{proof}

\section{Cauchy and large-$r$ stability results} \label{sec:extr}
In this section, we formulate and prove preliminary stability results that serve as basic ingredients for the proofs of Theorems~\ref{thm:L-stability} and \ref{thm:instability} in Sections~\ref{sec:L-stability} and \ref{sec:instability}, respectively. Roughly speaking, these results, in combination with Price's law discussed in Section~\ref{sec:bg}, allow us to reduce the proofs of Theorems~\ref{thm:L-stability} and \ref{thm:instability} to analysis of the difference of the solutions in a ``neighborhood'' of timelike infinity (see Theorem~\ref{thm:L-st-ch} below).

In Section~\ref{subsec:cauchy-st}, we state a simple \emph{Cauchy stability} result (Proposition~\ref{prop:cauchy-st}), which is a quick consequence of the standard $C^{1}$ local well-posedness theory. In Section~\ref{subsec:extr-st}, we state and prove (what we call) \emph{large-$r$ stability results}, which concern the stability of the background solution $(\gbg, \phibg, \Fbg)$ in regions where $\rbg$ is sufficiently large (Propositions~\ref{prop:extr-st-cauchy} and \ref{prop:extr-st}).

\textbf{The setup of this section will be as follows.} Fix $\omg_0 \in (2,3]$. We will consider two $\omg_0$-admissible initial data sets $(r, f, h, \ell, \phi, \dot{\phi}, \e)$ and $(\rbg, \fbg, \hbg, \ellbg, \phibg, \dphibg, \ebg)$ which are close to each other (in a way to be made precise). We will compare subsets of the connected components of the exterior regions with $\rd_V r>0$ and $\rd_U r<0$ of the maximal globally hyperbolic future developments. We will denote the scalar field and the geometric quantities arising from the ``barred'' and ``unbarred'' initial data sets by ``barred'' and ``unbarred'' symbols respectively. In this section, the two solutions will be compared via initial-data normalized double null coordinates $(U,V)$ satisfying
\begin{equation}\label{def.U.V.Cauchy}
	\frac{\ud U}{\ud \rho} = -1, \quad \frac{\ud V}{\ud \rho} = 1 \quad \hbox{ on } \Sgm_{0}.
\end{equation}
In what follows, we will refer to the ``barred solution'' as a ``background solution'' and the ``unbarred solution'' as a ``perturbed solution'' (cf. conventions in Theorem~\ref{thm:L-stability}).

In order to compare the two solutions, it will be convenient to refer to them after reduction to spherical symmetry, i.e., we will discuss the solutions $(\Omg, r, \phi, \e)$ and $(\Omgbg, \rbg, \phibg, \ebg)$ on $\calQ$ and $\overline{\calQ}$ instead of $(\calM, g, \phi, F)$ and $(\overline{\calM}, \gbg, \phibg, \overline{F})$. In this section, we will consider various subsets of (the projection of) the maximal globally hyperbolic developments in $\calQ$ and $\overline{\calQ}$ and we will moreover identify subsets of $\calQ$ and $\overline{\calQ}$ via the initial-data-normalized coordinate system $(U,V)$ as discussed above. Notice that this is not in general possible for the whole domain $\calQ$ and $\overline{\calQ}$ as the solutions may have different $(U,V)$-range. Nevertheless, as we will show below, this can be done for the regions of interest in this section, when the initial data of the two solutions are sufficiently close to each other. 

\subsection{Cauchy stability} \label{subsec:cauchy-st}
Here we state a Cauchy stability result for the Einstein--Maxwell--(real)--scalar--field system in spherical symmetry.
\begin{proposition}[Cauchy stability] \label{prop:cauchy-st}
Let $(r, f, h, \ell, \phi, \dot{\phi}, \e)$ and $(\rbg, \fbg, \hbg, \ellbg, \phibg, \dphibg, \ebg)$ be $\omg_0$-admissible initial data sets on $\Sgm_{0}$. Let $\Sgm_{0}' = \set{ \rho \in (\rho_{1}', \rho_{2}')} \subseteq \Sigma_0$ be a subset with compact closure. Assume that $(\overline{\calD}', \Omgbg, \rbg, \phibg, \ebg)$ is a (not necessarily maximal) globally hyperbolic future development of $(\rbg, \fbg, \hbg, \ellbg, \phibg, \dphibg, \ebg)\restriction_{\Sigma_0'}$, where $\overline{\calD}'$ has compact closure in (the quotient manifold of) the maximal globally hyperbolic future development $\overline{\calQ}$ of $(\rbg, \fbg, \hbg, \ellbg, \phibg, \dphibg, \ebg)$.

Assume that for a double null coordinate system normalized by \eqref{def.U.V.Cauchy},
the following bounds hold on $\Sgm_{0}'$:
\begin{equation} \label{eq:cauchy-st}
\begin{gathered}
\abs{r - \rbg} \leq \eps_{1}, \quad
\abs{\rd_{U} r - \rd_{U} \rbg} \leq \eps_{1}, \quad 
\abs{\rd_{V} r - \rd_{V} \rbg} \leq \eps_{1}, \\
\abs{\phi - \phibg} \leq \eps_{1}, \quad
\abs{\rd_{U} \phi - \rd_{U} \phibg} \leq \eps_{1}, \quad 
\abs{\rd_{V} \phi - \rd_{V} \phibg} \leq \eps_{1}, \\
\abs{\log \Omg - \log \Omgbg} \leq \eps_{1}, \quad
\abs{\rd_{U} \log \Omg - \rd_{U} \log \Omgbg} \leq \eps_{1}, \quad 
\abs{\rd_{V} \log \Omg - \rd_{V} \log \Omgbg} \leq \eps_{1}, \\
\abs{\e - \ebg} \leq \eps_{1}.
\end{gathered}
\end{equation}
If $\eps_{1}$ is sufficiently small (depending on $(\overline{\calD}', \Omgbg, \rbg, \phibg, \ebg)$), then a globally hyperbolic future development $(\calD',\Omg, r, \phi, \e)$ of $(r, f, h, \ell, \phi, \dot{\phi}, \e)\restriction_{\Sigma_0'}$ exists with $\calD'$ a subset of (the quotient manifold of) the maximal globally hyperbolic future development $\calQ$ of $(r, f, h, \ell, \phi, \dot{\phi}, \e)$ such that the following hold:
\begin{enumerate}
\item Introducing double null coordinate systems $(U, V)$ in $\overline{\calD}'$ and $\calD'$ with the same initial values on $\Sigma_0'$ normalized by the conditions \eqref{def.U.V.Cauchy}, $\overline{\calD}'=\calD'$ as a subset of $\mathbb R^{1+1}$.
\item Moreover, there exists a constant $C = C(\overline{\calD}', \Omgbg, \rbg, \phibg, \ebg) > 0$ so that the same bounds as above hold in $\calD'$ with $\eps_{1}$ replaced by $C \eps_{1}$.
\end{enumerate}
\end{proposition}

This proposition is a straightforward corollary of the standard proof of $C^{1}$ local well-posedness of \eqref{eq:EMSF-wave} by iteration. We omit the details.

We reformulate some conclusions of Proposition~\ref{prop:cauchy-st} in a fashion that will be convenient for later applications.
\begin{corollary} \label{cor:cauchy-st:psidf-phidf}
Under the same assumptions as in Proposition~\ref{prop:cauchy-st}, the following bounds hold in $\calD'$:
\begin{align*}
	\abs{\varpi - \varpibg} \leq & C \eps_{1}, \\
	\Abs{\frac{1}{(- \rd_{U} r)} \rd_{U} \phi - \frac{1}{(- \rd_{U} \rbg)} \rd_{U} \phibg} \leq & C \eps_{1}, \\
	\Abs{\frac{1-\mu}{\rd_{V} r} \rd_{V} (r \phi) - \frac{1-\mubg}{\rd_{V} \rbg} \rd_{V} (\rbg \phibg)} \leq & C \eps_{1}.
\end{align*}
Here, $C = C(\calD', \Omgbg, \rbg, \phibg, \ebg) > 0$.
\end{corollary}

This corollary follows from the bounds in Proposition~\ref{prop:cauchy-st}, using the formulae
\begin{align*}
	\varpi = \frac{r}{2} \left(1 +4 \frac{\rd_{U} r \rd_{V} r}{\Omg^{2}} \right) + \frac{\e^{2}}{2 r}, \qquad
	\frac{1-\mu}{\rd_{V} r} = - 4 \frac{\rd_{U} r}{\Omg^{2}}.
\end{align*}
We again omit the obvious details.

\subsection{Stability in large-$r$ regions} \label{subsec:extr-st}
Given $(U, V)$, recall that $\underline{U}(V)$ and $\underline{V}(U)$ are defined by the relations $(\underline{U}(V), V) \in \Sgm_{0}$ and $(U, \underline{V}(U)) \in \Sgm_{0}$, respectively. We define $\underline{\ups} (U)$ with respect to the background solution, i.e., $\underline{\ups} (U)= \rbg(U, \underline{V}(U))$. We also remind the reader that $\psi = r \phi$ (Definition~\ref{def.ADM.psi}), and introduce the corresponding notation $\psibg = \rbg \phibg$ for the background solution.

We first state a stability result in the domain of dependence $\underline{\calR}$ of a subset of $\Sgm_{0}$ on which the background $\rbg$ is sufficiently large (see Figure~\ref{fig:extr-st-cauchy}).

\begin{figure}[h]
\begin{center}
\def\svgwidth{250px}
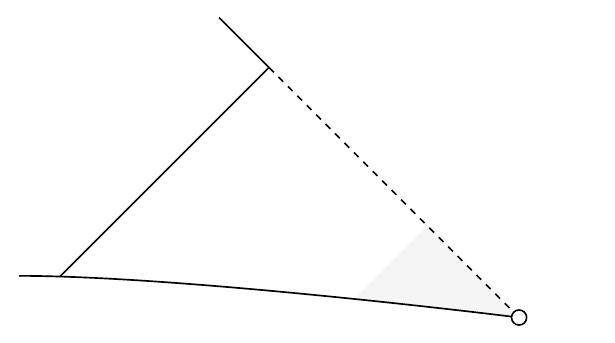 
\caption{Region $\underline{\calR}$} \label{fig:extr-st-cauchy}
\end{center}
\end{figure}

\begin{proposition}[Large-$r$ stability, Cauchy IVP] \label{prop:extr-st-cauchy}
Consider a region $\underline{\calR} \subseteq \PD$ of the form
\begin{equation*}
\underline{\calR} = \set{(U, V) \in \PD : U \leq U_{0}}
\end{equation*}
for some $U_{0} \in \bbR$ (observe that $\underline{\calR}$ is the domain of dependence of $\Sgm_{0} \cap \underline{\calR}$). In $\underline{\calR}$, assume that the background solution $(\Omgbg, \rbg, \phibg, \ebg)$ obeys 
\begin{equation*}
	\rd_{U} \rbg < 0, \quad \rd_{V} \rbg > 0,
\end{equation*}
and the following bounds for some $\rbg_{0}, \varpibg_{0}, \dlt_{\underline{\calR}} >0$ and $\underline{B} \geq 1$: 
\begin{align}
	\rbg(U_{0}, \underline{V}(U_{0})) \geq \rbg_{0}, \quad
	\varpi (U,V) \leq & \varpibg_{0}, \label{eq:extr-st-cauchy:hyp:r-m} \\
	\varpibg_{0} / \rbg_{0} \leq \dlt_{\underline{\calR}}, \quad
	\abs{\ebg} / \rbg_{0} \leq & \dlt_{\underline{\calR}}, \label{eq:extr-st-cauchy:hyp:m-e}\\
	\underline{B}^{-1} \leq -\rd_{U} \rbg(U, \underline{V}(U)) \leq & \underline{B}, \label{eq:extr-st-cauchy:hyp:dur} \\
	\underline{B}^{-1} \leq \frac{-\rd_{U} \rbg(U, \underline{V}(U))}{- \rd_{U} \underline{\ups}(U)} \leq & \underline{B}, \label{eq:extr-st-cauchy:hyp:duups}\\
	\rbg_{0}^{-1} \abs{\psibg} (U,V) \leq & (\underline{\ups} / \rbg_{0})^{- \omg_0 +1}  \underline{B}, \label{eq:extr-st-cauchy:hyp:psi}\\
	\Abs{\frac{1-\mubg}{\rd_{V} \rbg} \rd_{V} \psibg}(U,V) \leq & (\rbg / \rbg_{0})^{-\omg_0} \underline{B}. \label{eq:extr-st-cauchy:hyp:dvpsi}
\end{align}
Define the initial difference size $\eps_{2}$ on $\Sgm_{0} \cap \underline{\calR}$ by
\begin{align*}
	\eps_{2} 
= \sup_{\Sgm_{0} \cap \underline{\calR}} \bigg( &
	\rbg_{0}^{-1} \abs{\varpi - \varpibg}
	+ \rbg^{-1} \abs{r - \rbg}
	+ \Abs{\log (-\rd_{U} r) - \log (-\rd_{U} \rbg)}
	+ \Abs{\log \frac{\rd_{V} r}{1-\mu} - \log \frac{\rd_{V} \rbg}{1-\mubg}} \\
	& + (\rbg / \rbg_{0})^{\omg_0} \Abs{\frac{1-\mu}{\rd_{V} r}\rd_{V} \psi - \frac{1-\mubg}{\rd_{V} \rbg} \rd_{V} \psibg} 
	+ \rbg_{0}^{-1} \abs{\psi - \psibg} \bigg)
	+ \rbg_{0}^{-2} \int_{\Sgm_{0} \cap \underline{\calR}} \abs{\psi - \psibg} (\rd_{\rho} \rbg) \, \ud \rho 
	+ \rbg_{0}^{-1} \abs{\e - \ebg}.
\end{align*}
For sufficiently small $\dlt_{\underline{\calR}} > 0$ (depending on $\omg_0, \underline{B}$) and $\eps_{2} > 0$ (depending on $\omg_0, \underline{B}, \dlt_{\underline{\calR}}$), there exists $C = C(\omg_0, \underline{B}) > 0$ such that the following bounds hold in $\underline{\calR}$:
\begin{equation} \label{eq:extr-st-cauchy}
\begin{aligned}
	\rbg_{0}^{-1} \abs{\varpi - \varpibg} (U, V) \leq & C \eps_{2}, \\
	\rbg^{-1} \abs{r - \rbg} (U, V) \leq & C \eps_{2}, \\
	\Abs{\log (-\rd_{U} r) - \log (- \rd_{U} \rbg)} (U, V)  \leq & C \eps_{2}, \\
	\Abs{\log \frac{\rd_{V} r}{1-\mu}- \log \frac{\rd_{V} \rbg}{1-\mubg}} (U, V) \leq & C \eps_{2}, \\
	\rbg_{0}^{-1} \abs{\psi - \psibg} (U, V) \leq & \rbg_{0}^{-1} \abs{\psi - \psibg} (U, \underline{V}(U))+ C (\underline{\ups} / \rbg_{0})^{-\omg_0+1} (U) \eps_{2}, \\
	\Abs{\frac{1-\mu}{\rd_{V} r} \rd_{V} \psi - \frac{1-\mubg}{\rd_{V} \rbg} \rd_{V} \psibg} (U, V)\leq & C (\rbg / \rbg_{0})^{-\omg_0} (U, V) \eps_{2} .	
\end{aligned}
\end{equation}
Furthermore, we have
\begin{equation} \label{eq:extr-st-cauchy:L}
	\Abs{\int_{-\infty}^{-1} 2 M(U) \Phi(U) \Gmm(U) \, \ud U - \int_{-\infty}^{-1} 2 \overline{M}(U) \overline{\Phi}(U) \overline{\Gmm}(U) \, \ud U} \leq C \rbg_{0}^{3} \eps_{2}.
\end{equation}
\end{proposition}
For \eqref{eq:extr-st-cauchy:L}, we remind the reader that $M(U) = \lim_{V \to \infty} \varpi(U, V)$, $\Phi(U) = \lim_{V \to \infty} \psi(U, V)$, $\Gmm(U) = \lim_{V \to \infty} \frac{\rd_{U} r}{1-\mu}$ in the coordinates $(U, V)$ used in Proposition~\ref{prop:extr-st-cauchy}, and $\overline{M} (U)$, $\overline{\Phi}(U)$, $\overline{\Gmm}(U)$ are defined analogously.
\begin{remark} \label{rem:extr-st-cauchy}
Note that the initial data difference $\eps_{2}$ involves an \emph{integrated} norm of $\psi - \psibg$, which is weaker than the pointwise norm in $d_{1, \omg_0}^{+}(\Tht, \overline{\Tht})$, i.e., 
\begin{equation} \label{eq:extr-st-cauchy:eps2-int}
	\rbg_{0}^{-2} \int_{\Sgm_{0} \cap \underline{\calR}} \abs{\psi - \psibg} (\rd_{\rho} r) \, \ud \rho \leq C \sup_{\Sgm_{0} \cap \underline{\calR}} (\rbg / \rbg_{0})^{\omg_0-1} \rbg_{0}^{-1} \abs{\psi - \psibg}.
\end{equation}
The use of the integrated norm is motivated the instability argument in Section~\ref{sec:instability}. More precisely, we construct a perturbation $\Tht = \Tht_{\eps}$ of $\overline{\Tht}$ by placing a smooth bump function of amplitude $\eps$ in the data for $\rd_{U}(r \phi)$ around $\rho = \rho_{pert}$, where $\rho_{pert}$ is a large parameter chosen at the end of the argument. From this point of view, the RHS of \eqref{eq:extr-st-cauchy:eps2-int} is not very useful as it depends on the choice of $\rho_{pert}$. On the other hand, observe that the LHS is uniformly small with respect to $\rho_{pert}$. For more details, we refer to Lemma~\ref{lem:inst-s}.
\end{remark}

Next, we state a stability result in a characteristic rectangle $\calR$ (extending all the way to $\NI$) situated in a region where $\rbg$ is sufficiently large (see Figure~\ref{fig:extr-st}).
\begin{figure}[h]
\begin{center}
\def\svgwidth{250px}
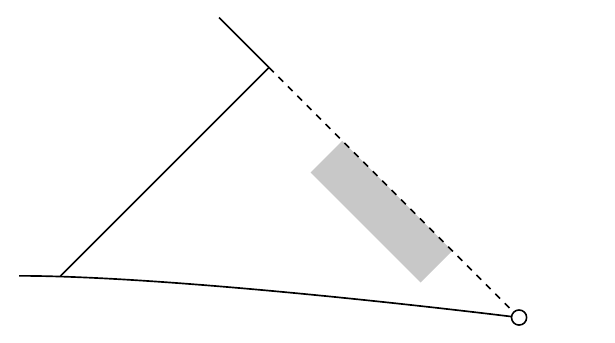 
\caption{Region $\underline{\calR}$} \label{fig:extr-st}
\end{center}
\end{figure}

\begin{proposition}[Large-$r$ stability, characteristic IVP] \label{prop:extr-st}
Given $U_{0}, U_{1}, V_{0} \in \bbR$ such that $U_{0} \leq U_{1}$, consider a characteristic rectangle $\calR \subseteq \PD$ of the form 
\begin{equation*}
	\calR = \set{(U, V) \in \PD : U_{0} \leq U \leq U_{1}, \ V \geq V_{0}}.
\end{equation*}
In $\calR$, assume that the background solution $(\Omgbg, \rbg, \phibg, \ebg)$ obeys 
\begin{equation*}
	\rd_{U} \rbg < 0, \quad \rd_{V} \rbg > 0,
\end{equation*}
and the following bounds for some $\rbg_{0}, \varpibg_{0}, \dlt_{\calR} > 0$, $B \geq 1$ and $U_{\ast} \in [U_{0}, U_{1}]$:
\begin{align}
	\rbg(U_{1}, V_{0}) \geq \rbg_{0}, \quad 
	\varpibg(U,V) \leq & \varpibg_{0}, \label{eq:extr-st:hyp:r-m} \\
	\rbg_{0}^{-1} \varpibg_{0} \leq \dlt_{\calR}, \quad
	\rbg_{0}^{-1} \abs{\ebg} \leq & \dlt_{\calR}, \label{eq:extr-st:hyp:m-e} \\
	B^{-1} \leq -\rd_{U} \rbg (U, V_{0}) \leq & B, \label{eq:extr-st:hyp:dur} \\
	B^{-1} \leq \frac{-\rd_{U} \rbg (U, V_{0})}{- \rd_{U} \underline{\ups}(U)} \leq & B \quad \hbox{ when } U < U_{\ast}, \label{eq:extr-st:hyp:duups} \\
	\rbg_{0}^{-1} \abs{\psibg}(U,V) \leq & \left\{
	\begin{array}{cl}
	B & \hbox{ when } U \geq U_{\ast} , \\
	(\underline{\ups} / \rbg_{0})^{-\omg_0+1}  B & \hbox{ when } U < U_{\ast}, \\
	\end{array} \right. \label{eq:extr-st:hyp:psi} \\
	\Abs{\frac{1-\mubg}{\rd_{V} \rbg} \rd_{V} \psibg}(U,V) \leq & (\rbg / \rbg_{0})^{-\omg_0}  B. \label{eq:extr-st:hyp:dvpsi}
\end{align}
Let $\uC_{in} = \set{(U, V_{0}) : U \geq U_{0}}$ and $C_{out} = \set{(U_{0}, V) : V \geq V_{0}}$ be the null curves constituting the past boundary of $\calR$. We define the initial difference size $\eps_{3}$ on $\uC_{in} \cup C_{out}$ as
\begin{align*}
	\eps_{3} =
	& \rbg_{0}^{-1} \abs{\e - \ebg}  + \sup_{\uC_{in}} \left(
	\rbg_{0}^{-1} \abs{\varpi - \varpibg}
	+ \rbg^{-1} \abs{r - \rbg} \right) 
	 + \sup_{\uC_{in} } \Abs{\log (-\rd_{U} r) - \log (-\rd_{U} \rbg)} \\
	& +  \sup_{\uC_{in} \cap \set{U \geq U_{\ast}}} \rbg_{0}^{-1} \abs{\psi - \psibg}
	 + \sup_{\uC_{in} \cap \set{U < U_{\ast}}} \rbg_{0}^{-1} \abs{\psi - \psibg} 
	 + \rbg_{0}^{-2} \int_{\uC_{in} \cap \set{U < U_{\ast}}} \abs{\psi - \psibg} (- \rd_{U} \underline{\ups})  \, \ud U \\
	 & + \sup_{C_{out}} \bigg( 
	 \Abs{\log \frac{\rd_{V} r}{1-\mu} - \log \frac{\rd_{V} \rbg}{1-\mubg}} 
	 + (\rbg / \rbg_{0})^{\omg_0} \Abs{\frac{1-\mu}{\rd_{V} r}\rd_{V} \psi - \frac{1-\mubg}{\rd_{V} \rbg} \rd_{V} \psibg} \bigg) .
\end{align*}
For sufficiently small $\dlt_{\calR}$ (depending on $\omg_0, B$) and $\eps_{3}$ (depending on $\omg_0, B, U_{1} - U_{\ast}, \dlt_{\calR}$),  there exists $C = C(\omg_0, B, U_{1} - U_{\ast}) > 0$ such that 
\begin{equation} \label{eq:extr-st} 
\begin{aligned}
	\rbg_{0}^{-1} \abs{\varpi - \varpibg} (U, V) \leq & C \eps_{3}, \\
	\rbg^{-1} \abs{r - \rbg} (U, V) \leq & C \eps_{3}, \\
	\Abs{\log (-\rd_{U} r) - \log (-\rd_{U} \rbg)} (U, V) \leq & C \eps_{3}, \\
	\Abs{\log \frac{\rd_{V} r}{1-\mu} - \log \frac{\rd_{V} \rbg}{1-\mubg}} (U, V) \leq & C \eps_{3}, \\
	\rbg_{0}^{-1} \abs{\psi - \psibg} (U, V) \leq & \left\{
	\begin{array}{cl}
	C \eps_{3} & \hbox{ when } U \geq U_{\ast}, \\
	\rbg_{0}^{-1}\abs{\psi - \psibg}(U, V_{0}) + C (\underline{\ups} / \rbg_{0})^{-\omg_0+1} (U)  \eps_{3} & \hbox{ when } U < U_{\ast}, \\
	\end{array} \right. \\
	\Abs{\frac{1-\mu}{\rd_{V} r} \rd_{V} \psi - \frac{1-\mubg}{\rd_{V} \rbg} \rd_{V} \psibg} (U, V) \leq & C (\rbg / \rbg_{0})^{-\omg_0} (U, V) \eps_{3} .	
\end{aligned}
\end{equation}
Furthermore, along $\NI$ we have
\begin{equation} \label{eq:extr-st:L}
	\Abs{\int_{U_{0}}^{U_{1}} 2 M(U) \Phi(U) \Gmm(U) \, \ud U - \int_{U_{0}}^{U_{1}} 2 \overline{M}(U) \overline{\Phi}(U) \overline{\Gmm}(U) \, \ud U} \leq C \rbg_{0}^{3} \eps_{3}.
\end{equation}
\end{proposition}  

Propositions~\ref{prop:extr-st-cauchy} and \ref{prop:extr-st} are proved in essentially the same manner. In what follows, we give a detailed proof of the characteristic IVP case (Proposition~\ref{prop:extr-st}), and only briefly sketch the proof of the Cauchy case (Proposition~\ref{prop:extr-st-cauchy}).

\begin{proof}[Proof of Proposition~\ref{prop:extr-st}]
Without loss of generality, we may assume that $U_{\ast} = V_{0} = 0$. Furthermore, since Proposition~\ref{prop:extr-st} is formulated so that the assumptions and the conclusions are invariant under the scaling transformation $(\Omg, r, \phi, \e) \mapsto (a \Omg, a r, \phi, a \e)$, $(\Omgbg, \rbg, \phibg, \ebg) \mapsto (a \Omgbg, a \rbg, \phibg, a \ebg)$, $(\rbg_{0}, \varpibg_{0}, U, V) \mapsto (a \rbg_{0}, a \varpibg_{0}, aU, aV)$ (for any $a > 0$), \textbf{we will also assume that $\rbg_{0} = 1$}. Notice that for $\varpi$ large, one can view this rescaling as transforming a ``large-$\rbg$ regime'' to a ``small-$\varpibg$ and $\ebg$ regime''.

\pfstep{Step~1: Bounds on the background geometry}
We begin by deriving some bounds for $\rd_{U} \rbg$ in $\calR$. For a sufficiently small $\dlt_{\calR} > 0$ (as a universal constant), we claim that
\begin{align} 
	(2B)^{-1} \leq - \rd_{U} \rbg \leq & 2 B,	\label{eq:extr-st:pf:durbg}  \\
	(2B )^{-1} (-\rd_{U} \underline{\ups}) \leq - \rd_{U} \rbg \leq & 2 B (-\rd_{U} \underline{\ups}) \quad \hbox{ when } U < 0.	\label{eq:extr-st:pf:durbg-dups}  
\end{align}
By \eqref{eq:extr-st:hyp:r-m} and \eqref{eq:extr-st:hyp:m-e}, note that
\begin{align}
	\abs{\rd_{V} \log \rd_{U} \rbg }
	= \Abs{\frac{2 (\varpibg - \frac{\ebg^{2}}{\rbg})}{\rbg^{2}} \frac{\rd_{V} \rbg}{1-\mubg}}
	\leq C \dlt_{\calR} \frac{\rd_{V} \rbg}{\rbg^{2}} \label{eq:extr-st:pf:dvlogdurbg}
\end{align}
for a sufficiently small $\dlt_{\calR}$ (as a universal constant). Integrating this bound from $(U, 0)$ to $(U, V)$, and using \eqref{eq:extr-st:hyp:dur} on $\uC_{in}$, \eqref{eq:extr-st:pf:durbg} follows. Moreover, using \eqref{eq:extr-st:hyp:duups}, we also obtain \eqref{eq:extr-st:pf:durbg-dups}.

Similarly as \eqref{eq:extr-st:pf:dvlogdurbg}, note that the following bound for $\rd_{U} \log \rd_{V} \rbg$ also holds:
\begin{equation} \label{eq:extr-st:pf:dulogdvrbg}
	\abs{\rd_{U} \log \rd_{V} \rbg }
	\leq C \dlt_{\calR} \frac{- \rd_{U} \rbg}{\rbg^{2}}.
\end{equation}

\pfstep{Step~2: Setting up the main bootstrap argument}
Our first goal is to prove \eqref{eq:extr-st}; for this purpose, we use a bootstrap argument. Let $\calR'$ be a subset of $\calR$ of the form $\calR' = \set{(U, V) \in \calR : U \leq U_{B}, \ V \leq V_{B}}$ for some $(U_{B}, V_{B}) \in \calR$. We make different bootstrap assumptions depending on whether $U < 0$ or $U \geq 0$. 

Let $A_{0}, A_{1} > 0$ be constants to be specified below, where $1 \leq A_{0} \leq A_{1}$. 
In the region $\calR' \cap \set{U < 0}$, we assume that:
\begin{align} 
	\abs{\varpi - \varpibg} (U, V) \leq & 10 A_{0} \eps_{3} , \label{eq:extr-st:btstrp:mdf-far} \\
	\Abs{\log (-\rd_{U} r) - \log (-\rd_{U} \rbg)} (U, V) \leq & 10 A_{1} \eps_{3} , \label{eq:extr-st:btstrp:durdf-far}\\
	\Abs{\log (-\rd_{V} r) - \log (-\rd_{V} \rbg)} (U, V) \leq & 10 A_{1} \eps_{3} , \label{eq:extr-st:btstrp:dvrdf-far}\\
	\rbg^{-1} \abs{r - \rbg} (U, V) \leq & 10 A_{1} \eps_{3} , \label{eq:extr-st:btstrp:rdf-far} \\
	\Abs{\frac{1}{\rd_{V} r} \rd_{V} \psi - \frac{1}{\rd_{V} \rbg} \rd_{V} \psibg} (U, V) \leq & 10 A_{1} \rbg^{-\omg_0} \eps_{3} , \label{eq:extr-st:btstrp:dvpsidf-far} \\
	\abs{\psi - \psibg} (U, V) \leq & \abs{\psi - \psibg} (U, 0) + 10 A_{1} \underline{\ups}^{-\omg_0 +1} \eps_{3}. \label{eq:extr-st:btstrp:psidf-far} 
\end{align}
In the region $\calR' \cap \set{U \geq 0}$, we assume that:
\begin{align} 
	\abs{\varpi - \varpibg} (U, V) \leq & 10 A_{0} e^{U} \eps_{3} , \label{eq:extr-st:btstrp:mdf} \\
	\Abs{\log (-\rd_{U} r) - \log (-\rd_{U} \rbg)} (U, V) \leq & 10 A_{1} e^{U} \eps_{3} , \label{eq:extr-st:btstrp:durdf}\\
	\Abs{\log (-\rd_{V} r) - \log (-\rd_{V} \rbg)} (U, V) \leq & 10 A_{1} e^{U} \eps_{3} , \label{eq:extr-st:btstrp:dvrdf}\\
	\rbg^{-1} \abs{r - \rbg} (U, V) \leq & 10 A_{1} e^{U} \eps_{3} , \label{eq:extr-st:btstrp:rdf} \\
	\Abs{\frac{1}{\rd_{V} r} \rd_{V} \psi - \frac{1}{\rd_{V} \rbg} \rd_{V} \psibg} (U, V) \leq & 10 A_{1} e^{U} \rbg^{-\omg_0} \eps_{3} , \label{eq:extr-st:btstrp:dvpsidf}  \\
	\abs{\psi - \psibg} (U, V) \leq & 10 A_{1} e^{U} \eps_{3}. \label{eq:extr-st:btstrp:psidf} 
\end{align}

\begin{remark} 
In our bootstrap argument, the factor $e^{U}$ in \eqref{eq:extr-st:btstrp:mdf}--\eqref{eq:extr-st:btstrp:psidf} essentially plays the role of the exponential factor in Gr\"onwall's inequality. In particular, this factor allows us to close the bootstrap assumptions \eqref{eq:extr-st:btstrp:mdf}--\eqref{eq:extr-st:btstrp:psidf} without taking $\dlt_{\calR}$ small compared to $U_{1} - U_{\ast}$. In fact, we may use any $e^{a U}$ with $a > 0$ provided that $\dlt_{\calR}$ and $\eps_{3}$ are small enough compared to $a$,  but we have fixed $a = 1$ for concreteness.
\end{remark}
Our goal until Step~6 is to establish the following claim: For sufficiently large $A_{0}$ and $A_{1}$ (depending on $\omg_0, B$, chosen in the order $A_{0} \to A_{1}$), and small enough $\dlt_{\calR}$ (depending on $\omg_0, B$) and $\eps_{3}$ (depending on $\omg_0, B, A_{0}, A_{1}, \dlt_{\calR}$), the bounds \eqref{eq:extr-st:btstrp:mdf-far}--\eqref{eq:extr-st:btstrp:psidf} hold in $\calR'$ with the constants $10 A_{0}$ and $10 A_{1}$ replaced by $A_{0}$ and $A_{1}$, respectively.

\pfstep{Step~3: Consequences of the bootstrap assumptions}
In the remainder of the proof, we may harmlessly assume that $\dlt_{\calR} \leq 1$. Since $\eps_{3}$ is allowed to depend on $A_{1}$, $U_{1}$, and $\dlt_{\calR}$, we may also assume that $10 A_{1} e^{U_{1}} \eps_{3} \leq \dlt_{\calR}$ and $\eps_{3} \leq \dlt_{\calR}$. Then by \eqref{eq:extr-st:hyp:r-m}, \eqref{eq:extr-st:hyp:m-e}, \eqref{eq:extr-st:btstrp:mdf-far}, \eqref{eq:extr-st:btstrp:mdf} and $\abs{\e - \ebg} \leq \eps_{3}$, we have
\begin{equation} \label{eq:extr-st:pf:m-e}
	\varpibg \leq \dlt_{\calR}, \quad
	|\ebg| \leq \dlt_{\calR}, \quad
	\varpi \leq 2 \dlt_{\calR}, \quad
	|\e| \leq 2 \dlt_{\calR}.
\end{equation}
Let $\alp \geq 0$. Using \eqref{eq:extr-st:btstrp:mdf-far}, \eqref{eq:extr-st:btstrp:mdf}, $\abs{\e^{2} - \ebg^{2}} \leq C \dlt_{\calR} \eps_{3} \leq C A_{0} \eps_{3}$ and \eqref{eq:extr-st:pf:m-e} (see also Lemma~\ref{lem:r-alp-df} below), we may show that
\begin{align} \label{eq:extr-st:pf:m-e-r-alp-df}
	\Abs{\frac{\varpi}{r^{\alp}} - \frac{\varpibg}{\rbg^{\alp}}} \leq & \frac{C_{\alp} (A_{0} + \dlt_{\calR} A_{1}) e^{U_{+}}\eps_{3}}{\rbg^{\alp}}, \\
	\Abs{\frac{\e^{2}}{r^{\alp+1}} - \frac{\ebg^{2}}{\rbg^{\alp+1}}} \leq & \frac{C_{\alp} (A_{0} + \dlt_{\calR} A_{1}) e^{U_{+}}\eps_{3}}{\rbg^{\alp+1}},
\end{align}
where $U_{+} = \max \set{U, 0}$. Taking $\eps_{3}$ smaller if necessary, we may assume that $(1+C_{1}) (A_{0} + \dlt_{\calR} A_{1}) e^{U_{1}} \eps_{3} \leq \frac{1}{100}$. It follows that
\begin{equation} \label{eq:extr-st:pf:mu}
	\frac{1}{2} \leq 1-\mu \leq 2, \quad \frac{1}{2} \leq 1-\mubg \leq 2.
\end{equation}
Then in combination with the previous bounds, we have
\begin{equation} \label{eq:extr-st:pf:mu-df}
	\Abs{\log (1-\mu) - \log (1-\mubg)} \leq C (A_{0} + \dlt_{\calR} A_{1}) e^{U_{+}}\eps_{3}.
\end{equation}
Assume furthermore that $A_{1} e^{U_{1}} \eps_{3} \leq 1$. By \eqref{eq:extr-st:btstrp:durdf-far} and \eqref{eq:extr-st:btstrp:durdf}, observe that
\begin{equation} \label{eq:dur-df}
	\max \left\{ \Abs{\frac{\rd_{U} \rbg}{\rd_{U} r} - 1}, \Abs{\frac{\rd_{U} r}{\rd_{U} \rbg} - 1} \right\}
	\leq e^{\abs{\log (-\rd_{U} r) - \log(-\rd_{U} \rbg)}} - 1
	\leq C A_{1} e^{U_{+}} \eps_{3}.
\end{equation}
Similarly, by \eqref{eq:extr-st:btstrp:dvrdf-far}, \eqref{eq:extr-st:btstrp:dvrdf} and \eqref{eq:extr-st:pf:m-e-r-alp-df}, we have
\begin{equation} \label{eq:extr-st:pf:dvr-kpp-df}
\max \left\{ \Abs{\frac{\rd_{V} \rbg}{\rd_{V} r} - 1}, \Abs{\frac{\rd_{V} r}{\rd_{V} \rbg} - 1} \right\}
\leq C A_{1} e^{U_{+}} \eps_{3}, \quad
\Abs{\frac{1-\mu}{\rd_{V} r} \frac{\rd_{V} \rbg}{1-\mubg} - 1} 
\leq C A_{1} e^{U_{+}} \eps_{3}.
\end{equation}
 
 \pfstep{Step~4: Estimate for $\varpi - \varpibg$}
Our goal in this step is to improve \eqref{eq:extr-st:btstrp:mdf-far} and \eqref{eq:extr-st:btstrp:mdf}, and therefore fix the first bootstrap constant $A_{0} \geq 1$. We write
\begin{align*}
	\rd_{V} (\varpi - \varpibg) 
	= \frac{1}{2} \frac{1-\mu}{\rd_{V} r} (\rd_{V} \psi - r^{-1} (\rd_{V} r ) \psi)^{2} - \frac{1}{2} \frac{1-\mubg}{\rd_{V} \rbg} (\rd_{V} \psibg - \rbg^{-1} (\rd_{V} \rbg ) \psibg)^{2} 
	= T'_{1} + T'_{2} + T'_{3},
\end{align*}
where
\begin{align*}
	T'_{1} =& \frac{1}{2} \frac{1-\mu}{\rd_{V} r} (\rd_{V} (\psi - \psibg) - (r^{-1} (\rd_{V} r ) \psi - \rbg^{-1} (\rd_{V} \rbg ) \psibg))^{2} , \\
	T'_{2} = & \frac{1-\mu}{\rd_{V} r} (\rd_{V} (\psi - \psibg) - (r^{-1} (\rd_{V} r ) \psi - \rbg^{-1} (\rd_{V} \rbg )  \psibg)) (\rd_{V} \psibg - \rbg^{-1} ( \rd_{V} \rbg ) \psibg) , \\
	T'_{3} = & \frac{1}{2} \left( \frac{1-\mu}{\rd_{V} r} - \frac{1-\mubg}{\rd_{V} \rbg}\right) (\rd_{V} \psibg - \rbg^{-1} (\rd_{V} \rbg ) \psibg)^{2}.
\end{align*}
We claim that, for some universal constant $C > 0$, we have
\begin{align} 
	\int_{0}^{\infty} \abs{T'_{1}}(U, V') \, \ud V' \leq & C B^{2} A_{1}^{2} e^{2 U_{+}} \eps_{3}^{2}, \label{eq:extr-st:pf:mdf-1} \\
	\int_{0}^{\infty} \abs{T'_{2}}(U, V') \, \ud V' \leq & C B \sqrt{ \dlt_{\calR}} A_{1} e^{U_{+}} \eps_{3}, \label{eq:extr-st:pf:mdf-2} \\
	\int_{0}^{\infty} \abs{T'_{3}}(U, V') \, \ud V' \leq & C \dlt_{\calR} A_{1} e^{U_{+}} \eps_{3}. \label{eq:extr-st:pf:mdf-3}
\end{align}
By hypothesis we have $\abs{\varpi - \varpibg}(U, 0) \leq \eps_{3}$. Therefore, once \eqref{eq:extr-st:pf:mdf-1}--\eqref{eq:extr-st:pf:mdf-3} are proved, \eqref{eq:extr-st:btstrp:mdf-far} and \eqref{eq:extr-st:btstrp:mdf} would follow by fixing a large universal constant $A_{0}$, and taking $\dlt_{\calR}$ and $A_{1} e^{U_{1}}  \eps_{3}$ sufficiently small depending on $A_{1} / A_{0}$.

We first prove \eqref{eq:extr-st:pf:mdf-1}. For the contribution of the term $\rd_{V} (\psi - \psibg)$, we use \eqref{eq:extr-st:hyp:dvpsi}, \eqref{eq:extr-st:btstrp:dvpsidf-far}, \eqref{eq:extr-st:btstrp:dvpsidf}, \eqref{eq:extr-st:pf:mu} and \eqref{eq:extr-st:pf:dvr-kpp-df} to estimate
\begin{align*}
& \hskip-2em
\int_{0}^{\infty} \frac{1-\mu}{\rd_{V} r}(\rd_{V} (\psi - \psibg))^{2} (U, V') \, \ud V' \\
\leq  & C \int_{0}^{\infty} \left( \frac{1}{\rd_{V} r} \rd_{V} \psi - \frac{1}{\rd_{V} \rbg} \rd_{V} \psibg + \left( 1 - \frac{\rd_{V} \rbg}{\rd_{V} r} \right) \frac{1}{\rd_{V} \rbg} \rd_{V} \psibg \right)^{2} \rd_{V} \rbg (U, V') \, \ud V' \\
\leq & C B^{2} A_{1}^{2} e^{2 U_{+}} \eps_{3}^{2}.
\end{align*}
On the other hand, for the contribution of $\psi, \psibg$, we use \eqref{eq:extr-st:hyp:psi}, \eqref{eq:extr-st:btstrp:rdf-far}, \eqref{eq:extr-st:btstrp:psidf-far}, \eqref{eq:extr-st:btstrp:rdf}, \eqref{eq:extr-st:btstrp:psidf}, \eqref{eq:extr-st:pf:mu} and \eqref{eq:extr-st:pf:dvr-kpp-df} to estimate
\begin{align*}
& \hskip-2em
\int_{0}^{\infty} \frac{1-\mu}{\rd_{V} r}(r^{-1} ( \rd_{V} r ) \psi - \rbg^{-1} ( \rd_{V} \rbg ) \psibg)^{2} (U, V') \, \ud V' \\
\leq & C \int_{0}^{\infty} \frac{\rbg^{2}}{r^{2}} \frac{\rd_{V} r}{\rd_{V} \rbg} \left(  \psi - \psibg + \left(1 - \frac{r}{\rbg} \frac{\rd_{V} \rbg}{\rd_{V} r}\right) \psibg \right)^{2} \frac{\rd_{V} \rbg}{\rbg^{2}}(U, V') \, \ud V' \\
\leq & C B^{2} A_{1}^{2} e^{2 U_{+}} \eps_{3}^{2},
\end{align*}
which completes the proof of \eqref{eq:extr-st:pf:mdf-1}.

To prove \eqref{eq:extr-st:pf:mdf-2}, we use \eqref{eq:extr-st:pf:m-e}, \eqref{eq:extr-st:pf:dvr-kpp-df} and \eqref{eq:extr-st:pf:mdf-1} (that we just established) and Cauchy--Schwarz to estimate
\begin{align*}
& \hskip-2em
\int_{0}^{\infty} \frac{1-\mu}{\rd_{V} r} \left(\rd_{V} (\psi - \psibg) - \left(r^{-1} ( \rd_{V} r ) \psi - \rbg^{-1} ( \rd_{V} \rbg ) \psibg\right)\right) \left(\rd_{V} \psibg - \rbg^{-1} ( \rd_{V} \rbg ) \psibg\right) (U, V') \, \ud V' \\ 
\leq & C \sup_{V' \in [0, \infty)} \left( \frac{1-\mu}{\rd_{V} r} \frac{\rd_{V} \rbg}{1-\mubg} (U, V') \right)^{1/2} \left( \int_{0}^{\infty} \frac{1}{2}\frac{1-\mubg}{\rd_{V} \rbg} (\rd_{V} \psibg - \rbg^{-1} ( \rd_{V} \rbg ) \psibg)^{2} (U, V') \, \ud V' \right)^{1/2} B A_{1} e^{U_{+}}  \eps_{3} \\
=& C \sup_{V' \in [0, \infty)} \left( \frac{1-\mu}{\rd_{V} r} \frac{\rd_{V} \rbg}{1-\mubg} (U, V') \right)^{1/2} \left( \lim_{V'\to \infty} \overline{\varpi}(U,V')-\overline{\varpi}(U,0) \right)^{1/2} B A_{1} e^{U_{+}}  \eps_{3} \leq  C B \sqrt{\dlt_{\calR}} A_{1} e^{U_{+}} \eps_{3}.
\end{align*}

Finally, for \eqref{eq:extr-st:pf:mdf-3}, we use \eqref{eq:extr-st:pf:m-e} and \eqref{eq:extr-st:pf:dvr-kpp-df} to estimate
\begin{align*}
\int_{0}^{\infty} \left( \frac{1-\mu}{\rd_{V} r} \frac{\rd_{V} \rbg}{1-\mubg} - 1 \right) \frac{1-\mubg}{\rd_{V} \rbg} (\rd_{V} \psibg - \rbg^{-1} ( \rd_{V} \rbg ) \psibg)^{2} (U, V') \, \ud V' 
\leq C \dlt_{\calR} A_{1} e^{U_{+}} \eps_{3}.
\end{align*}

\pfstep{Step~5: Estimate for $r - \rbg$}
In this step, we improve the bootstrap assumptions \eqref{eq:extr-st:btstrp:durdf-far}, \eqref{eq:extr-st:btstrp:dvrdf-far}, \eqref{eq:extr-st:btstrp:rdf-far}, \eqref{eq:extr-st:btstrp:durdf}, \eqref{eq:extr-st:btstrp:dvrdf} and  \eqref{eq:extr-st:btstrp:rdf}. For convenience, we introduce the abbreviations 
\begin{equation*}
\widetilde{\log \rd_{V} r} = \log \rd_{V} r - \log \rd_{V} \rbg, \quad 
\widetilde{\log \rd_{U} r} = \log \rd_{U} r - \log \rd_{U} \rbg.
\end{equation*}
We start with the bootstrap assumptions \eqref{eq:extr-st:btstrp:dvrdf-far} and \eqref{eq:extr-st:btstrp:dvrdf} concerning $\widetilde{\log \rd_{V} r}$.
On $C_{out}$, note that
\begin{align*} 
	\Abs{\widetilde{\log \rd_{V} r}} (U_{0}, V)
\leq & \Abs{\log \frac{\rd_{V} r}{1-\mu} - \log \frac{\rd_{V} \rbg}{1-\mubg}} (U_{0}, V) + \Abs{\log (1-\mu) - \log (1-\mubg)} (U_{0}, V)  \\
\leq & \eps_{3} + C (A_{0} + \dlt_{\calR} A_{1}) \eps_{3}.
\end{align*}
On the other hand, by the preliminary bounds established in the previous steps, we have
\begin{align}
	\Abs{\rd_{U} \widetilde{\log \rd_{V} r}}
	= & \Abs{\frac{2 (\varpi - \frac{\e^{2}}{r})}{r^{2}} \frac{\rd_{U} r}{1-\mu}
	- \frac{2 (\varpibg - \frac{\ebg^{2}}{\rbg})}{\rbg^{2}} \frac{\rd_{U} \rbg}{1-\mubg}} \notag \\
	\leq & C (A_{0} + \dlt_{\calR} A_{1}) e^{U_{+}} \eps_{3} \frac{- \rd_{U} \rbg}{\rbg^{2}}. \label{eq:extr-st:pf:dulogdvrdf}
\end{align}
Therefore, we obtain
\begin{align*}
	\Abs{\widetilde{\log \rd_{V} r}} (U, V) 
	\leq & \Abs{\widetilde{\log \rd_{V} r}} (U_{0}, V) + \int_{U_{0}}^{U} \Abs{\rd_{U} \widetilde{\log \rd_{V} r}} (U', V) \, \ud U' \\
	\leq & \eps_{3} + C (A_{0} + \dlt_{\calR} A_{1}) e^{U_{+}} \eps_{3}.
\end{align*}
Taking $A_{1}$ large enough (compared to $C A_{0}$) and $\dlt_{\calR}$ sufficiently small (as a universal constant), the bootstrap assumptions \eqref{eq:extr-st:btstrp:dvrdf-far} and \eqref{eq:extr-st:btstrp:dvrdf}  improve. Moreover, as a consequence of the preceding bound, we also obtain
\begin{equation*}
	\abs{r - \rbg}(U, V) 
	\leq \abs{r - \rbg}(U, 0) + \int_{0}^{V} \Abs{\frac{\rd_{V} r}{\rd_{V} \rbg} - 1} \rd_{V} \rbg (U, V') \, \ud V' 
	\leq C \left( \eps_{3} + C (A_{0} + \dlt_{\calR} A_{1}) e^{U_{+}} \eps_{3} \right) \rbg(U, V),
\end{equation*}
which also improves the bootstrap assumptions \eqref{eq:extr-st:btstrp:rdf-far} and \eqref{eq:extr-st:btstrp:rdf} for an appropriate choice of $A_{1}$ and $\dlt_{\calR}$. 

Finally, we close \eqref{eq:extr-st:btstrp:durdf-far} and \eqref{eq:extr-st:btstrp:durdf} concerning $\widetilde{\log \rd_{U} r}$, which proceeds by a similar argument. In this case, we have
\begin{equation} \label{eq:extr-st:pf:dvlogdurdf}
	\Abs{\rd_{V} \widetilde{\log \rd_{U} r}}
	\leq C (A_{0} + \dlt_{\calR} A_{1}) e^{U_{+}} \eps_{3}  \frac{\rd_{V} \rbg}{\rbg^{2}}.
\end{equation}
Therefore,
\begin{align*}
	\Abs{\widetilde{\log \rd_{U} r}} (U, V) 
	\leq & \Abs{\widetilde{\log \rd_{U} r}} (U, 0) + \int_{0}^{V} \Abs{\rd_{V} \widetilde{\log \rd_{U} r}} (U, V') \, \ud V' \\
	\leq & \eps_{3} + C (A_{0} + \dlt_{\calR} A_{1}) e^{U_{+}} \eps_{3},
\end{align*}
which improves \eqref{eq:extr-st:btstrp:durdf-far} and \eqref{eq:extr-st:btstrp:durdf} for an appropriate choice of $A_{1}$ and $\dlt_{\calR}$.

\pfstep{Step~6: Estimate for $\psi - \psibg$}
The aim of this step is to improve \eqref{eq:extr-st:btstrp:dvpsidf-far}, \eqref{eq:extr-st:btstrp:psidf-far}, \eqref{eq:extr-st:btstrp:dvpsidf} and \eqref{eq:extr-st:btstrp:psidf}. We first claim that
\begin{equation} \label{eq:extr-st:pf:dvpsidf}
	\Abs{\frac{1}{\rd_{V} r} \rd_{V} \psi - \frac{1}{\rd_{V} \rbg} \rd_{V} \psibg} \leq \left( C \eps_{3} + C_{\omg_0, B} (A_{0} + \widetilde{\dlt}_{\calR} A_{1}) e^{U_{+}} \eps_{3} \right) \rbg^{-\omg_0}.
\end{equation}
where $\widetilde{\dlt}_{\calR} (\geq \dlt_{\calR})$ is introduced below in \eqref{eq:extr-st:pf:new-dlt}. It turns out that, by choosing $\dlt_{\calR}$ and $\eps_{3}$ appropriately, $\widetilde{\dlt}_{\calR}$ can be made as small as we want depending on $\omg_0, B$. 

Assume, for the moment, that \eqref{eq:extr-st:pf:dvpsidf} has already been proved. Then \eqref{eq:extr-st:btstrp:dvpsidf-far} and \eqref{eq:extr-st:btstrp:dvpsidf} would improve by taking $A_{1}$ large enough (compared to $C_{\omg_0, B} A_{0}$) and $\widetilde{\dlt}_{\calR}$ sufficiently small (depending on $\omg_0, B$). Moreover, since we have
\begin{equation*}
	\frac{1}{\rd_{V} \rbg} \rd_{V} (\psi - \psibg)
	= \frac{\rd_{V} r}{\rd_{V} \rbg} \left( \frac{1}{\rd_{V} r} \rd_{V} \psi - \frac{1}{\rd_{V} \rbg} \rd_{V} \psibg \right) + \left(\frac{\rd_{V} r}{\rd_{V} \rbg} - 1 \right)\frac{1}{\rd_{V} \rbg} \rd_{V} \psibg,
\end{equation*}
it would follow that
\begin{align*}
	\abs{\psi - \psibg}(U, V)
	\leq & \abs{\psi - \psibg}(U, 0) + \int_{0}^{\infty} \Abs{\frac{1}{\rd_{V} \rbg} \rd_{V} (\psi - \psibg)} \rd_{V} \rbg(U, V') \, \ud V' \\
	\leq & \abs{\psi - \psibg}(U, 0) + \left( C \eps_{3} + C_{\omg_0, B} (A_{0} + \dlt_{\calR} A_{1}) e^{U_{+}} \eps_{3} \right) \rbg^{-\omg_0+1}(U, 0)
\end{align*}
Since $\rbg(U, 0) \geq \max \set{\underline{\ups}(U), 1}$, the desired improvement of \eqref{eq:extr-st:btstrp:psidf-far} and \eqref{eq:extr-st:btstrp:psidf} would follow also by taking $A_{1}$ large enough (compared to $C_{\omg_0, B} A_{0}$) and $\widetilde{\dlt}_{\calR}$ sufficiently small (depending on $\omg_0, B$).

In order to establish the claim \eqref{eq:extr-st:pf:dvpsidf}, we use the equation
\begin{align*}
	\rd_{U} \left(\frac{1}{\rd_{V} r} \rd_{V} \psi - \frac{1}{\rd_{V} \rbg} \rd_{V} \psibg \right)
	=- \rd_{U} \log \rd_{V} r \left( \frac{1}{\rd_{V} r} \rd_{V} \psi - \frac{1}{r} \psi \right)  + \rd_{U} \log \rd_{V} \rbg \left( \frac{1}{\rd_{V} \rbg} \rd_{V} \psibg - \frac{1}{\rbg} \psibg \right),
\end{align*}
which we rewrite as
\begin{align*}
	\rd_{U} \left(\frac{1}{\rd_{V} r} \rd_{V} \psi - \frac{1}{\rd_{V} \rbg} \rd_{V} \psibg \right)
	= - \rd_{U} \log \rd_{V} r \left(  \frac{1}{\rd_{V} r} \rd_{V} \psi - \frac{1}{\rd_{V} \rbg} \rd_{V} \psibg \right)
	+ T_{1} + T_{2} + T_{3} ,
\end{align*}
where
\begin{align*}
	T_{1} = & - \rd_{U} \widetilde{\log \rd_{V} r} \frac{1}{\rd_{V} \rbg} \rd_{V} \psibg , \\
	T_{2} = & \rd_{U} \log \rd_{V} r \left( \frac{1}{r} \psi - \frac{1}{\rbg} \psibg \right), \\
	T_{3} = & \rd_{U} \widetilde{\log \rd_{V} r} \frac{1}{\rbg} \psibg .
\end{align*}
By \eqref{eq:extr-st:pf:dulogdvrbg} and \eqref{eq:extr-st:pf:dulogdvrdf}, we have
\begin{equation} \label{eq:extr-st:pf:dulogdvr}
	\abs{\rd_{U} \log \rd_{V} r}
	\leq C \widetilde{\dlt}_{\calR} \rbg^{-2} (-\rd_{U}\rbg),
\end{equation}
where 
\begin{equation} \label{eq:extr-st:pf:new-dlt}
\widetilde{\dlt}_{\calR} = \dlt_{\calR} + (A_{0} + \dlt_{\calR} A_{1}) e^{U_{1}} \eps_{3} .
\end{equation}
By choosing $\dlt_{\calR}$ and $\eps_{3}$ appropriately, we may assume that $\widetilde{\dlt}_{\calR}$ is as small as we want depending on $\omg_0, B$. Since the integral of \eqref{eq:extr-st:pf:dulogdvr} over any $\uC_{V} \cap \calR$ is uniformly bounded, we may apply Gr\"onwall's inequality to estimate
\begin{equation*}
	\Abs{\frac{1}{\rd_{V} r} \rd_{V} \psi - \frac{1}{\rd_{V} \rbg} \rd_{V} \psibg}(U, V)
	\leq C \Abs{\frac{1}{\rd_{V} r} \rd_{V} \psi - \frac{1}{\rd_{V} \rbg} \rd_{V} \psibg}(U_{0}, V)
		+ C \int_{U_{0}}^{U} \abs{T_{1} + T_{2} + T_{3}}(U', V) \, \ud U'.
\end{equation*}
On $C_{out}$, note that
\begin{align*}
& \Abs{\frac{1}{\rd_{V} r} \rd_{V} \psi - \frac{1}{\rd_{V} \rbg} \rd_{V} \psibg}(U_{0}, V)\\
\leq & C \Abs{\frac{1-\mu}{\rd_{V} r} \rd_{V} \psi - \frac{1-\mubg}{\rd_{V} \rbg} \rd_{V} \psibg}(U_{0}, V) + C \Abs{\frac{1-\mu}{1-\mubg} - 1} \Abs{\frac{1-\mubg}{\rd_{V} \rbg} \rd_{V} \psibg} (U_0,V) \\
\leq & (C + C B (A_{0} + \dlt_{\calR} A_{1})) \rbg^{-\omg_0} \eps_{3},
\end{align*}
which is acceptable after taking $\dlt_{\calR} \leq B^{-1}$.
On the other hand, for $T_{1}$, $T_{2}$ and $T_{3}$ we have
\begin{align}
	\abs{T_{1}} \leq & C B (A_{0} + \dlt_{\calR} A_{1}) e^{U_{+}} \rbg^{-2 -\omg_0} (-\rd_{U} \rbg) \eps_{3}, \label{eq:extr-st:pf:dvpsi-1} \\
	\abs{T_{2}} \leq &\left\{ \begin{array}{cl}
	\displaystyle{C B \widetilde{\dlt}_{\calR} A_{1} e^{U} \rbg^{-3} (-\rd_{U} \rbg) \eps_{3}}  & \hbox{ when } U \geq 0, \\
	\displaystyle{\left(C \tilde{\dlt}_{\calR} \abs{\psi - \psibg}(U, 0) + C B \tilde{\dlt}_{\calR} A_{1} \underline{\ups}^{-\omg_0+1}(U) \eps_{3} \right) \rbg^{-3} (-\rd_{U} \rbg)  }  & \hbox{ when } U < 0,
	\end{array} \right.  \label{eq:extr-st:pf:dvpsi-2} \\
	\abs{T_{3}} \leq &\left\{ \begin{array}{cl}
	C B (A_{0} + \dlt_{\calR} A_{1}) e^{U} \rbg^{-3} (-\rd_{U} \rbg) \eps_{3}   & \hbox{ when } U \geq 0, \\
	C B (A_{0} + \dlt_{\calR} A_{1})  \underline{\ups}^{-\omg_0+1}(U)  \rbg^{-3} (-\rd_{U} \rbg) \eps_{3}   & \hbox{ when } U < 0.
	\end{array} \right.  \label{eq:extr-st:pf:dvpsi-3}
\end{align}
Indeed, \eqref{eq:extr-st:pf:dvpsi-1} follows from \eqref{eq:extr-st:hyp:dvpsi} and \eqref{eq:extr-st:pf:dulogdvrdf}. For the proof of \eqref{eq:extr-st:pf:dvpsi-2}, we apply \eqref{eq:extr-st:pf:dulogdvr} and
\begin{equation*} 
	\Abs{\frac{1}{r} \psi - \frac{1}{\rbg} \psibg}
	\leq \left\{ \begin{array}{cl}
	C B A_{1} e^{U} \rbg^{-1} \eps_{3}  & \hbox{ when } U \geq 0, \\
	\rbg^{-1} \abs{\psi - \psibg}(U, 0) + C B A_{1} \underline{\ups}^{-\omg_0+1}(U) \rbg^{-1}  \eps_{3}  & \hbox{ when } U < 0,
	\end{array} \right.
\end{equation*}
which, in turn, follows from \eqref{eq:extr-st:btstrp:rdf-far}, \eqref{eq:extr-st:btstrp:psidf-far}, \eqref{eq:extr-st:btstrp:rdf} and \eqref{eq:extr-st:btstrp:psidf}. Finally, \eqref{eq:extr-st:pf:dvpsi-3} is a consequence of \eqref{eq:extr-st:hyp:psi} and \eqref{eq:extr-st:pf:dulogdvrdf}.

Integrating \eqref{eq:extr-st:pf:dvpsi-1}, we obtain
\begin{align*}
	\int_{U_{0}}^{U} \abs{T_{1}}(U', V) \, \ud U'
	\leq C B (A_{0} + \dlt_{\calR} A_{1}) e^{U_{+}} \rbg^{-\omg_0-1} \eps_{3},
\end{align*}
which is acceptable. On the other hand, integrating \eqref{eq:extr-st:pf:dvpsi-2} and \eqref{eq:extr-st:pf:dvpsi-3} (using also $\widetilde{\dlt}_{\calR} \geq \dlt_{\calR}$), we obtain
\begin{align*}
	\int_{U_{0}}^{U} (\abs{T_{2}} + \abs{T_{3}})(U', V) \, \ud U'
	\leq & \int_{U_{0}}^{\min \set{U, 0}} (\abs{T_{2}} + \abs{T_{3}})(U', V) \, \ud U' + \int_{0}^{\max \set{U, 0}} (\abs{T_{2}} + \abs{T_{3}})(U', V) \, \ud U' \\
	\leq & C B \tilde{\dlt}_{\calR} \rbg^{-3} \int_{U_{0}}^{\min \set{U, 0}} \abs{\psi - \psibg}(U, 0) (-\rd_{U} \underline{\ups})(U) \, \ud U \\
	& + \frac{C}{\omg_0-2} B (A_{0} + \widetilde{\dlt}_{\calR} A_{1}) \rbg^{-3} \eps_{3} 
	 + C B (A_{0} + \widetilde{\dlt}_{\calR} A_{1}) e^{U_{+}} \rbg^{-3} \eps_{3}.
\end{align*}
In the last inequality, we used \eqref{eq:extr-st:pf:durbg} and \eqref{eq:extr-st:pf:durbg-dups} to estimate $-\rd_{U} \rbg$. The desired bound \eqref{eq:extr-st:pf:dvpsidf} now follows.

\pfstep{Step~7: Completion of the proof}
So far, we have closed the bootstrap assumptions in $\calR'$. Note furthermore that, by the argument so far, the bounds in \eqref{eq:extr-st} also hold (with $C = A_{1} e^{U_{1}}$) on $\calR'$. By a standard continuous induction argument, it follows that \eqref{eq:extr-st} holds in the whole rectangle $\calR$.

To conclude the proof, it remains to establish the bound \eqref{eq:extr-st:L}. As a consequence of the previous steps, for $\overline{M}(U)$, $\overline{\Phi}(U)$ and $\overline{\Gmm}(U)$ we have
\begin{align*}
	\abs{\overline{M}(U)} 
\leq &\sup_{\calR} \abs{\varpibg}
 \leq C , \\
	\abs{\Phibg(U)} 
\leq & \sup_{V \in [0, \infty)} \abs{\psibg(U, V)} 
\leq \left\{
	\begin{array}{cl}
		C B & \hbox{ when } U \geq 0, \\
		C B \underline{\ups}^{-\omg_0+1}(U) \eps_{3} & \hbox{ when } U < 0 .
	\end{array}
	\right. \\
	\abs{\overline{\Gmm}(U)} 
\leq & \sup_{\calR} \abs{\rd_{U} \rbg}  
\leq \left\{
	\begin{array}{cl}
		C B & \hbox{ when } U \geq 0, \\
		C B (- \rd_{U} \underline{\ups}(U)) & \hbox{ when } U < 0 .
	\end{array}
	\right. 
\end{align*}
On the other hand, for the corresponding differences, we have
\begin{align*}
	\abs{M(U) - \overline{M}(U)} 
\leq &\sup_{\calR} \abs{\varpi - \varpibg}
 \leq C \eps_{3}, \\
	\abs{\Phi(U) - \Phibg(U)} 
\leq & \sup_{V \in [0, \infty)} \abs{\psi(U, V) - \psibg(U, V)} 
\leq \left\{
	\begin{array}{cl}
		C \eps_{3} & \hbox{ when } U \geq 0, \\
		\abs{\psi - \psibg}(U, 0) + C \underline{\ups}^{-\omg_0+1}(U) \eps_{3} & \hbox{ when } U < 0 .
	\end{array}
	\right. \\
	\abs{\Gmm(U) - \overline{\Gmm}(U)} 
\leq & \sup_{\calR} \abs{\rd_{U} r - \rd_{U} \rbg}  
\leq \left\{
	\begin{array}{cl}
		C B \eps_{3} & \hbox{ when } U \geq 0, \\
		C B (-\rd_{U} \underline{\ups}(U)) \eps_{3} & \hbox{ when } U < 0 ,
	\end{array}
	\right. 
\end{align*}
Putting together these bounds, \eqref{eq:extr-st:L} follows. \qedhere
\end{proof}

\begin{proof}[Sketch of Proof of Proposition~\ref{prop:extr-st-cauchy}]
Without loss of generality, we may set $U_{0} = 0$ and $\rbg_{0} = 1$. Then the proof is entirely analogous to that of Proposition~\ref{prop:extr-st} in the region $U < 0$, except that the initial data lie on $\Sgm_{0} \cap \underline{\calR}$. As a consequence, all instances of $(U_{0}, V)$, $(U, V_{0})$ must be replaced by $(\underline{U}(V), V)$, $(U, \underline{V}(U))$ and $\underline{\ups}(U)$, respectively. Moreover, the roles of $B$ and $\dlt_{\calR}$ are played by $\underline{B}$ and $\dlt_{\underline{\calR}}$, respectively. We omit the straightforward details of modifying the proof. \qedhere
\end{proof}

\section{Stability of the backscattering tail: Proof of Theorem~\ref{thm:L-stability}} \label{sec:L-stability}
In this section we prove stability of the quantity $\mathfrak{L}$ (Theorem~\ref{thm:L-stability}).
It will be a consequence of \emph{asymptotic stability}, in a suitable sense, of the black hole exterior(s) of the maximal globally hyperbolic future development of an admissible initial data set. The tools we develop along the way will also play an important role in the proof of Theorem~\ref{thm:instability} in Section~\ref{sec:instability}.

\subsection{Ideas of the proof}

Recall that
\begin{equation*}
	\mathfrak{L} = \int_{\NI} 2 M(u) \Phi(u) \Gmm(u) \, \ud u
\end{equation*}
where $M(u) = \lim_{r \to \infty} \varpi(u, v)$, $\Phi(u) = \lim_{r \to \infty} r \phi(u, v)$ and $\Gmm(u) = \lim_{r \to \infty} \frac{\rd_{u} r}{1-\mu}(u, v)$. Stability of $\mathfrak{L}$ is thus a consequence of appropriate stability of each of the constituents $M(u)$, $\Phi(u)$ and $\Gmm(u)$ defined on $\NI$. In particular, since $M(u)$ and $\Gmm(u)$ are not expected to decay to zero towards timelike infinity, we need to ensure that the difference between the perturbed and the background $\Phi(u)$ \emph{decays at an integrable rate} towards timelike infinity\footnote{The analogous issue near spacelike infinity is taken care of by Proposition~\ref{prop:extr-st-cauchy}.}. 

To establish the above-mentioned stability of $M(u)$, $\Phi(u)$ and $\Gmm(u)$, we prove asymptotic stability of maximal development of admissible Cauchy data sets in the black hole exterior. Our proof is based on the following ideas:
\begin{enumerate} 
\item To begin with, we use the Cauchy and large-$r$ stability results (i.e., stability in the region close to $\NI$ where $\rbg$ is large) established in Section~\ref{sec:extr} to reduce the proof of asymptotic stability of the entire black hole exterior to showing that of a region near future timelike infinity, where the background solution is ``sufficiently small'' by Price's law (as we have seen in Section~\ref{sec:bg}); see Theorem~\ref{thm:L-st-ch}. To formalize the concept of a ``sufficiently small'' background solution, we introduce the notion of a $(\omg, \dlt, \Lmb)$-admissible solution (Definition~\ref{def:dlt-adm}). 

\item The Cauchy and large-$r$ stability results are proved simply using the initial-data-normalized coordinates $(U, V)$. However, proving stability of $\mathfrak{L}$ near future timelike infinity (Theorem~\ref{thm:L-st-ch}) requires a precise asymptotic decay bound for the difference of radiation fields $\Phi - \overline{\Phi}$ (see \eqref{eq:Phidf-decay-0} below), for which it is essential to work in a suitable coordinate system to compare the two solutions (indeed, an unwise choice of coordinates may easily destroy the desired estimates for $\Phi - \overline{\Phi}$!). 

Our choice is to use the future-normalized coordinates for the perturbed and the background solutions, which were defined in Section~\ref{subsec:bg-coords}. Under such a choice, the key goal is to show
\begin{equation} \label{eq:Phidf-decay-0}
	\abs{\Phi(u) - \Phibg(u)} \leq C u^{-\bt} \eps_{0}		\quad \hbox{ for $u$ large}, 
\end{equation}
where $\eps_{0}$ is the size of the initial data difference, and $\bt > 1$, so that the RHS is integrable in $u$. 

To be able to work in the future-normalized coordinates, we need to relate the size of the initial difference in the new coordinate systems to that in the original coordinate systems. For this purpose, we need \emph{smallness} of higher order (i.e., second order) derivatives of the background $(\omg, \dlt, \Lmb)$-admissible solution, which were derived from the Price's law theorem of Dafermos--Rodnianski in Section~\ref{sec:bg}. 

\item To efficiently handle the issue of coordinate choice in the proof of Theorem~\ref{thm:L-st-ch}, we take a modular approach. First, we establish a \emph{weak stability result} (Theorem~\ref{thm:weak-st}), which is not as strong as Theorem~\ref{thm:L-st-ch} but sufficient to achieve a nice control on the transformation from the initial-data-normalized coordinates $(U, V)$ to the future-normalized coordinates. Then, we carry out further arguments in the future-normalized coordinates to derive the desired decay of $\Phi - \Phibg$ (Propositions~\ref{prop:en-decay} and \ref{prop:psidf-decay}).

\item At the core of the weak stability result (Theorem~\ref{thm:weak-st}) lie nondegenerate energy and integrated local energy decay estimates for the difference $\phi - \phibg$ of scalar fields in the future-normalized coordinates. Through the Raychaudhuri equations \eqref{eq:EMSF-ray} and the modified mass equations in \eqref{eq:EMSF-r-phi-m}, such estimates imply pointwise and integrated bounds on the difference of geometric quantities (e.g., $\rd_{u} r - \rd_{u} \rbg$, $\rd_{v} r - \rd_{v} \rbg$) and the modified masses. These bounds, in turn, control the transformation from the initial-data-normalized coordinates to the future-normalized coordinates.

In the proof, we need to control nonlinear error terms of the form $\iint \frac{r^{3-\eta_{0}}}{\Omg^{2}} (\rd_{v} \phidf)^{2} (\rd_{u} \phidf)^{2} \, \ud u \ud v$ $(0 < \eta_{0} \ll 1)$, where $\phidf = \phi - \phibg$. To achieve this without proving additional pointwise bounds for $\rd_{u} \phidf$ and $\rd_{v} \phidf$, we use (what we call) an \emph{interaction Morawetz estimate}, which is obtained by multiplying the wave equation for $\phidf$ by $f_{0}(r) \left( \frac{1-\mu}{\dvr} \rd_{v} + \frac{1-\mu}{\dur} \rd_{u}\right)$ with an appropriate $f_{0}(r)$.

\item In Proposition~\ref{prop:en-decay}, using the nondegenerate energy and integrated local energy decay estimates in Theorem~\ref{thm:weak-st} as a starting point, we adapt the $r^{p}$-weighted energy method of Dafermos--Rodnianski \cite{DRNM} to prove decay of nondegenerate energy of $\phi - \phibg$ in an appropriate foliation. In particular, we obtain an \emph{integrable} rate of decay for $r \phi - \rbg \phibg$ along a constant-$r$ curve. Next, in Proposition~\ref{prop:psidf-decay}, we use an integration along characteristic argument (similar to \cite{LO1} and Section~\ref{subsec:blowup-const-r}) to propagate this decay to $(\Phi - \Phibg) (u)= \lim_{v \to \infty}(r \phi - \rbg \phibg)(u, v)$.
\end{enumerate}

\subsection{Organization of the section}
This section is organized as follows.
\begin{itemize}
\item {\bf Section~\ref{subsec:main-st}.} We introduce the notion of an \emph{$(\omg, \dlt, \Lmb)$-admissible background solution} (Definition~\ref{def:dlt-adm}). Then by the consequences of Price's law discussed in Section~\ref{sec:bg}, and the Cauchy/large-$r$ stability results in Section~\ref{sec:extr}, we reduce the proof of Theorem~\ref{thm:L-stability} to showing an asymptotic stability result (Theorem~\ref{thm:L-st-ch}\footnote{We note that Theorem~\ref{thm:L-st-ch} is \emph{quantitative}, in the sense that the implicit constant in the conclusion of the theorem is dependent only on certain explicit parameters of the background solution. This feature is essential for the application of the theorem in Section~\ref{sec:instability}.}) for an $(\omg, \dlt, \Lmb)$-admissible background solution. 

The remainder of this section is devoted to the proof of Theorem~\ref{thm:L-st-ch}. 
\item {\bf Section~\ref{subsec:weak-st}.} We state a \emph{weak stability} result (Theorem~\ref{thm:weak-st}), which does not imply Theorem~\ref{thm:L-st-ch}, but nevertheless gives an adequate control of the coordinate transformation $(U, V) \to (u, v), \, (\ubg, \vbg)$ and provides a starting point for the decay arguments later on.
\item {\bf Sections~\ref{subsec:btstrp}--\ref{subsec:weak-st-pf}.} These subsections comprise the proof of Theorem~\ref{thm:weak-st}. An outline of these subsections is provided at the end of Section~\ref{subsec:weak-st}, after we introduce relevant notation.
\item {\bf Section~\ref{subsec:rp-weight}.} Using Theorem~\ref{thm:weak-st} as the starting point, we employ the $r^{p}$-weighted energy method of Dafermos--Rodnianski \cite{DRNM} to prove an \emph{integrable rate} of decay (in a suitable sense) of $\phi - \phibg$ along a constant-$r$ curve in the future-normalized coordinate systems $(u, v), \, (\ubg, \vbg)$. 
\item {\bf Section~\ref{subsec:int-char}.} Here we propagate the decay of $\phi - \phibg$ along a constant-$r$ curve to that of the radiation field difference $\Phi - \Phibg$ along $\NI$ (in the future normalized coordinates) by an integration along characteristics technique similar to \cite{LO1} (see also Section~\ref{sec.contra.unif} above).
\item {\bf Section~\ref{subsec:main-st-pf}.} Finally, we put together the decay estimates proved so far to complete the proof of Theorem~\ref{thm:L-st-ch}.
\end{itemize}

\subsection{Reduction to stability of an admissible background solution} \label{subsec:main-st}
The Cauchy and large-$r$ stability results in Section~\ref{sec:extr} allow us to focus on a region near future timelike infinity of the background solution, where $\phibg$ and its derivatives are adequately small by Price's law (see Section~\ref{sec:bg}). Here we formulate a stability result in this region, which constitutes the key step of the proof of Theorem~\ref{thm:L-stability}. It is most natural and convenient (in particular, for the proof) to state this result in terms of a characteristic initial value problem.

We begin by specifying precisely the assumptions on the background solution.
\begin{definition} [$(\omg, \dlt, \Lambda)$-admissible background solution]\label{def:dlt-adm}
Let $\omg \in (2, 3]$, $\dlt > 0$ and $\Lmb > 1$. We say that a $C^{2}$ solution $(\Omgbg, \rbg, \phibg, \ebg)$ in $\widetilde{\calD}$ to the Einstein--Maxwell--(real)--scalar--field system in spherical symmetry is \emph{$(\omg, \dlt, \Lmb)$-admissible} if the following properties hold: 
\begin{enumerate}
\item $(\widetilde{\calD}, \gbg)$ is a smooth (1+1)-dimensional Lorentzian manifold with boundary (to be described below), which is a (non-maximal) globally hyperbolic future development of a characteristic initial value problem posed on $\uC_{in} \cup C_{out}$.
\begin{itemize}
\item $C_{out}$ is a past boundary component of $\widetilde{\calD}$ which is a smooth null ray; more precisely, it has a past endpoint, which we denote by $p$, but no future endpoint. It is \emph{outgoing} (i.e., $r$ is increasing towards the future), \emph{affine complete} and satisfies $\sup_{C_{out}} r = \infty$.
\item $\uC_{in}$ is a past boundary component of $\widetilde{\calD}$ which is a smooth null segment, and its past endpoint coincides with that of $C_{out}$, which we denoted by $p$ (hence $\uC_{in} \cap C_{out} = \set{p}$). It is \emph{incoming}, in the sense $r$ decreases towards the future.
\item $\widetilde{\calD}$ also has a future boundary, which is a spacelike hypersurface.
\end{itemize}
\item The event horizon $\EH$, which by definition is the null curve $\widetilde{\calD} \cap \rd (J^{-}(\NI))$, intersects the interior of $\uC_{in}$. 

\item Let $\rbg_{\EH} = \sup_{\EH} \rbg$ be the area-radius of the event horizon and $\varpibg_{\EH} := \sup_{\EH} \varpibg$; these two quantities obey the relation $\rbg_{\EH} = \varpibg_{\EH} + \sqrt{\varpibg_{\EH}^{2} - \ebg^{2}}$ (cf. \eqref{r.varpi.poly}). We assume that $\EH$ is \emph{subextremal}, i.e., $\varpibg_{\EH}^{-1} \abs{\ebg} < 1$ (cf. \eqref{eq:subextremality}).

\item On $\EH$, the area-radius function $\rbg$ obeys the bound $(1 -  \dlt) \rbg_{\EH} \leq \rbg \leq \rbg_{\EH}$. 

\item At the future endpoint of $\uC_{in}$, we have $\rbg = (1- 2\dlt) \rbg_{\EH}$. 
At the past endpoint of $\uC_{in}$, which is also the past endpoint of $C_{out}$, we have $\rbg = \Lmb \rbg_{\EH}$. The future boundary of $\widetilde{\calD}$ intersects $\uC_{in}$ at its future endpoint, i.e., at a point where $\rbg= (1-2\dlt) \rbg_{\EH}$. Moreover, on the future boundary of $\widetilde{\calD}$, $\rbg = (1-2\de) \rbg_{\EH}$.

\item In the black hole exterior $\calD := \widetilde{\calD} \cap J^{-}(\calI^{+})$, we have
\begin{equation} \label{eq:dlt-adm:m}
	\Abs{\varpibg - \varpibg_{f}} \leq \dlt^{2}.
\end{equation}

\item There exists a regular null coordinate system $(U, V)$ covering $\widetilde{\calD}$, such that $C_{out} = \set{U = 1}$, $\uC_{in} = \set{V = 1}$, $0 < \inf_{C_{out}} \frac{\rd_{V} \rbg}{1-\mubg} \leq \sup_{C_{out}} \frac{\rd_{V} \rbg}{1-\mubg} < \infty$ on $C_{out}$ and $0 < \inf_{\uC_{in}} (-\rd_{U} \rbg) \leq \sup_{\uC_{in}} (- \rd_{U} \rbg) < \infty$ on $\uC_{in}$. Moreover, $\rd_{U} \rbg < 0$ on $C_{out}$. 

\item Let $(\ubg, \vbg) = (\ubg(U), \vbg(V))$ denote the future-normalized coordinates covering $\calD:=\widetilde{\calD}\cap J^-(\NI)$ so that $\ubg = 1$ on $C_{out}$, $\frac{\rd_{\ubg} \rbg}{1 - \mubg} = -1$ along $\NI$ (cf. \eqref{eq:coord-u-bg}), $\vbg = 1$ on $\uC_{in}$ and $\frac{\rd_{\vbg} \rbg}{1 - \mubg} = 1$ along $\EH$ (cf. \eqref{eq:coord-v-bg}). The following smallness conditions hold for $\phibg$ and its (coordinate-invariant) first derivatives in $\calD$: 
\begin{align} 
	\Abs{\phibg} 
\leq &	\left\{
\begin{array}{ll}
\brk{\vbg}_{\rbg_{\EH}}^{-\omg} \dlt & \hbox{when } \rbg \leq 30 \rbg_{\EH} \\
\brk{\ubg}_{\rbg_{\EH}}^{-\omg+1} \min \set{\brk{\ubg}_{\rbg_{\EH}}^{-1}, (\rbg / \rbg_{\EH})^{-1}} \dlt  & \hbox{when } \rbg \geq 10 \rbg_{\EH} \\
\end{array}\right. 	\label{eq:dlt-adm:phi}
\\
	\rbg_{\EH} \Abs{\frac{1}{\rd_{U} \rbg} \, \rd_{U} \phibg} 
\leq &	\left\{
\begin{array}{ll}
\brk{\vbg}_{\rbg_{\EH}}^{-\omg}  \dlt & \hbox{when } \rbg \leq 30 \rbg_{\EH} \\
\brk{\ubg}_{\rbg_{\EH}}^{-\omg} (\rbg / \rbg_{\EH})^{-1} \Lmb^{-1/2} \dlt  & \hbox{when } \rbg \geq 10 \rbg_{\EH} \\
\end{array}\right. 	\label{eq:dlt-adm:duphi} \\
	\rbg_{\EH} \Abs{\frac{1-\mubg}{\rd_{V} \rbg} \, \rd_{V} \phibg} 
\leq &	\left\{
\begin{array}{ll}
\brk{\vbg}_{\rbg_{\EH}}^{-\omg} \dlt & \hbox{when } \rbg \leq 30 \rbg_{\EH} \\
\brk{\ubg}_{\rbg_{\EH}}^{-\omg+1} (\rbg / \rbg_{\EH})^{-1} \min \set{\brk{\ubg}_{\rbg_{\EH}}^{-1}, (\rbg / \rbg_{\EH})^{-1}} \dlt  & \hbox{when } \rbg \geq 10 \rbg_{\EH} \\
\end{array}\right.  \label{eq:dlt-adm:dvphi} \\
\Abs{\frac{1-\mubg}{\rd_{V} \rbg} \, \rd_{V} (\rbg \phibg)} 
\leq & \min\set{\brk{\vbg}_{\rbg_{\EH}}^{-\omg}, (\rbg / \rbg_{\EH})^{-\omg}} \dlt	\label{eq:dlt-adm:dvrphi}
\end{align}
Here $\brk{\cdot}_{\rbg_{\EH}} := \frac{(\cdot) + \rbg_{\EH}}{\rbg_{\EH}}$.

\item The following smallness conditions hold for the second derivatives of $\phibg$ on the initial characteristic hypersurface:
\begin{align} 
	\sup_{\uC_{in}} \, \rbg_{\EH}^{2}\Abs{\left(\frac{1}{\rd_{U} \rbg} \rd_{U}\right)^{2} \phibg} \leq & \Lmb^{-3/2} \, \dlt \label{eq:dlt-adm-duduphi} \\
	\sup_{C_{out}} \frac{\rbg^{\omg+1}}{\rbg_{\EH}^{\omg}} \Abs{\left(\frac{1-\mubg}{\rd_{V} \rbg} \rd_{V}\right)^{2} (\rbg \phibg)} \leq &  \dlt. \label{eq:dlt-adm-dvdvrphi}
\end{align}
\end{enumerate}
\end{definition}

The Penrose diagram depiction of $(\widetilde{\calD}, \gbg)$ is given in Figure~\ref{fig:dlt-adm} below.
\begin{figure}[h]
\begin{center}
\def\svgwidth{250px}
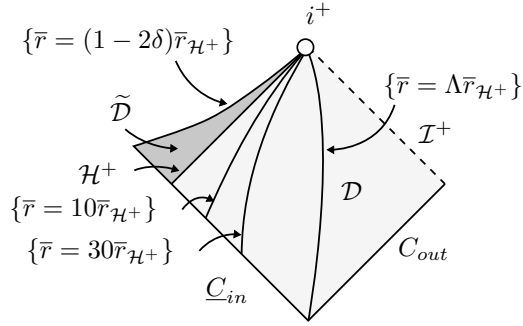 
\caption{Penrose diagram of a $(\omg, \dlt, \Lmb)$-admissible solution. In what follows, $\Lmb$ will be taken to be larger than $100$.} \label{fig:dlt-adm}
\end{center}
\end{figure}

\begin{remark} \label{rem:dlt-admissible}
Some remarks concerning our notion of $(\omg, \dlt, \Lmb)$-admissible background solution are in order.
\begin{enumerate}
\item Consider a perturbation of an $(\omg, \dlt, \Lmb)$-admissible background solution, given in terms of a perturbed characteristic initial data set on $\uC_{in} \cup C_{out}$. 
Observe that the new event horizon need not intersect $\uC_{in}$ at the same $U$-coordinate; in particular, it might happen in the black hole interior region of the background solution. This is why we consider the initial curve $\uC_{in}$ that protrudes a bit into the black hole interior region, although we are ultimately interested in stability of only the black hole exterior region.
\item On the other hand, the precise global structure of the interior region in a $(\omg, \dlt, \Lmb)$-admissible background solution is irrelevant for our stability theorem. Indeed, note that essentially all of the assumptions on $\phibg$ (see \eqref{eq:dlt-adm:phi}--\eqref{eq:dlt-adm:dvrphi}) are made only in the black hole exterior region $\calD = \widetilde{\calD} \cap J^{-}(\NI)$. The only exception is \eqref{eq:dlt-adm-duduphi}, which is also needed in the part of the initial hypersurface in $\widetilde{\calD}$. 
\item In fact, the larger domain $\widetilde{\calD}$ is only needed to compare the initial data, for instance in the statement of Theorem~\ref{thm:L-st-ch}. We only need the bound \eqref{eq:dlt-adm-duduphi} on the part of the initial hypersurface in $\widetilde{\calD}$ in order to guarantee that one can change back and forth between the initial-data-normalized and the future-normalized coordinates (cf. Section~\ref{subsec:gauge}).
\item The factors of $\rbg_{\EH}$ ensure that the bounds in Definition~\ref{def:dlt-adm} are dimensionless, i.e., invariant under the scaling $(\Omg, r, \phi, \e) \mapsto (a \Omg, a r, \phi, a \e)$ for any constant $a > 0$.
\item The dimensionless quantity $\Lmb$ can be reexpressed as $\Lmb = (\sup_{\uC_{in}} \rbg / \inf_{\uC_{in}} \rbg)$, i.e., it measures the ratio of the largest and smallest area-radii on $\uC_{in}$. In \eqref{eq:dlt-adm:duphi}, we have inserted the factor $\Lmb^{-1/2}$ for technical convenience; see \eqref{eq:dlt-tB-decay} below.
\end{enumerate}
\end{remark}

Consider a $C^{2}$ solution $(\Omg, r, \phi, \e)$ to the Einstein--Maxwell--(real)--scalar--field system in spherical symmetry on a domain $\calO$. Assume that there exist constant-$U$ curves such that $r\to \infty$ towards the future. Attach an ideal boundary $\NI \cap \overline{\calO}$ of $\calO$ which is in one-to-one correspondence to constant-$U$ curves along which $r \to \infty$ (cf. Theorem~\ref{thm:kommemi}). Suppose that $\NI \cap \overline{\calO}$ is parametrized by $U \in [U_{1}, U_{2}]$. We define
\begin{equation} \label{eq:L-calD}
	\mathfrak{L} \restriction_{\calO} = \int_{U_{1}}^{U_{2}} 2 M(U) \Phi(U) \Gmm(U) \, \ud U.
\end{equation}
where $M(U) = \lim_{r \to \infty} \varpi(U, V)$, $\Phi(U) = \lim_{r \to \infty} r \phi(U, V)$ and $\Gmm(U) = \lim_{r \to \infty} \frac{\rd_{U} r}{1- \mu}(U, V)$ as in Definition~\ref{def.L}, where the limit is taken with $U$ fixed. In fact, if $\calD$ is a subset of the maximal globally hyperbolic future development of an admissible Cauchy data set, then $\mathfrak{L} \restriction_{\calD}$ is precisely the dynamically-defined quantity $\mathfrak{L}$ in Definition~\ref{def.L}, except that the integral is restricted to $\NI \cap \overline{\calO}$.

The essential step in our proof of Theorem~\ref{thm:L-stability} is the following theorem concerning stability of $\mathfrak{L} \restriction_{\calD}$ for an $(\omg, \dlt, \Lmb)$-admissible background solution. 
\begin{theorem} [Stability of $\mathfrak{L} \restriction_{\calD}$ for an $(\omg, \dlt, \Lmb)$-admissible background solution] \label{thm:L-st-ch}
Fix $\omg \in (2, 3]$ and $\eta_{0} > 0$ so that $\eta_{0} < \min \set{2 \omg - 4, 1}$. Let $\dlt_{0}, \dlt, \Lmb$ satisfy
\begin{equation} \label{eq:L-st-ch:dlt}
	0 < \dlt \leq \Lmb^{-100 \eta_{0}} \dlt_{0}, \qquad \Lmb > 100.
\end{equation}
Assume furthermore that $\dlt_{0} > 0$ is sufficiently small depending on $\omg$, $\eta_{0}$ and $\varpibg_{f}^{-1} \abs{\ebg}$. Consider an $(\omg, \dlt, \Lmb)$-admissible background solution $(\Omgbg, \rbg, \phibg, \ebg)$ in $\widetilde{\calD}$ with $\rbg_{\EH} = \sup_{\EH} \rbg$.
On the characteristic initial hypersurface $\uC_{in} \cup C_{out}$, consider an initial data set $(\Omg \restriction_{\uC_{in} \cup C_{out}}, r \restriction_{\uC_{in} \cup C_{out}}, \phi \restriction_{\uC_{in} \cup C_{out}}, \e)$ which obeys the following difference estimates:
\begin{align} 
\sup_{\uC_{in}} \rbg_{\EH} \Abs{\frac{1}{\rd_{U} r} \rd_{U} \phi - \frac{1}{\rd_{U} \rbg} \rd_{U} \phibg} \leq & \Lmb^{-3/2} \eps_{0}, \label{eq:L-st-ch:ini-duphi} \\
\sup_{\uC_{in}}\Abs{\frac{\rd_{U} \rbg}{\rd_{U} r} - 1} \leq & \Lmb^{-1} \eps_{0}, \label{eq:L-st-ch:ini-dur} \\
\sup_{C_{out}} \left(\frac{\rbg}{\rbg_{\EH}}\right)^{\omg} \Abs{\frac{1-\mu}{\rd_{V} r} \rd_{V} (r \phi) - \frac{1-\mubg}{\rd_{V} \rbg} \rd_{V} (\rbg \phibg)} \leq & \eps_{0}, \label{eq:L-st-ch:ini-dvrphi} \\
\sup_{C_{out}} \Abs{\left( \frac{\rd_{V} r}{1-\mu}\right)^{-1}\frac{\rd_{V} \rbg}{1-\mubg} - 1} \leq & \eps_{0}, \label{eq:L-st-ch:ini-dvr} \\
\Abs{\phi - \phibg} \restriction_{\uC_{in} \cap C_{out}} \leq & \Lmb^{-1} \eps_{0}, \label{eq:L-st-ch:ini-phi}\\
\rbg_{\EH}^{-1} \Abs{r - \rbg} \restriction_{\uC_{in} \cap C_{out}} \leq & \eps_{0}, \label{eq:L-st-ch:ini-r} \\
\rbg_{\EH}^{-1} \Abs{\varpi - \varpibg} \restriction_{\uC_{in} \cap C_{out}} \leq & \eps_{0}, \label{eq:L-st-ch:ini-m}\\
\rbg_{\EH}^{-1} \Abs{\e - \ebg} \leq & \eps_{0}. \label{eq:L-st-ch:ini-e}
\end{align}
Then for $\eps_{0} > 0$ sufficiently small (depending on $\dlt_{0}$, $\varpibg_{f}^{-1} \abs{\ebg}$, $\eta_{0}$, $\omg$ and $\Lmb$), the maximal globally hyperbolic future development $(\Omg, r, \phi, \e)$ of $(\Omg \restriction_{\uC_{in} \cup C_{out}}, r \restriction_{\uC_{in} \cup C_{out}}, \phi \restriction_{\uC_{in} \cup C_{out}}, \e)$ has a complete future null infinity $\NI$ and an event horizon $\EH$. Moreover, for some constant $C > 0$ that depends only on $\varpibg_{f}^{-1} \abs{\ebg}$, $\eta_{0}$, $\omg$ and $\Lmb$, the following hold:
\begin{itemize}
\item Denoting $r_{\EH}=\sup_{\EH} r$, we have
\begin{equation}\label{eq:L-st-ch:r}
\Abs{r_{\EH} - \rbg_{\EH}} \leq  C \rbg_{\EH} \eps_{0}.
\end{equation}
\item Denoting $\calD=J^-(\NI)\cap \widetilde{\calD}$ for the $(\omg,\de,\Lmb)$-admissible background solution, as well as $\calD=J^-(\NI)$ within the maximal globally hyperbolic future development of $(\Omg \restriction_{\uC_{in} \cup C_{out}}, r \restriction_{\uC_{in} \cup C_{out}}, \phi \restriction_{\uC_{in} \cup C_{out}}, \e)$, we have
\begin{equation}\label{eq:L-st-ch:L}
	\abs{\mathfrak{L} \restriction_{\calD} - \overline{\mathfrak{L}} \restriction_{\calD}} \leq  C \rbg_{\EH}^{3} \eps_{0},	
\end{equation}
where $\mathfrak{L} \restriction_{\calD}$ is as in \eqref{eq:L-calD} and $\overline{\mathfrak{L}} \restriction_{\calD}$ is as in \eqref{eq:L-calD} but with $(r,\phi,\varpi,\e)$ replaced by $(\rbg, \phibg, \varpibg, \ebg)$.
\end{itemize}

\end{theorem}

\begin{remark} \label{rem:L-stability-char}
Some remarks concerning Theorem~\ref{thm:L-st-ch} are in order.
\begin{enumerate}
\item Note that the bounds in Theorem~\ref{thm:L-st-ch} are dimensionless, and invariant under a simultaneous coordinate transform (i.e., same coordinate transform $(U, V) \mapsto (U'(U), V'(V))$ for the two solutions).
\item In contrast to Theorem~\ref{thm:L-stability}, where the constant $C_{\overline{\Tht}}$ depends on the background solution in a highly nontrivial fashion, Theorem~\ref{thm:L-st-ch} has the feature that the constants in the conclusion \eqref{eq:L-st-ch:r}--\eqref{eq:L-st-ch:L} have more clear dependence on parameters. This point will be useful in the proof of Theorem~\ref{thm:instability} in Section~\ref{sec:instability}.
\item We will show that any solution arising from an $\omg_0$-admissible initial data set, after restricting to a region with sufficiently large $U$ and $V$, is an $(\omg, \dlt, \Lmb)$-admissible solution for any $\omg<\omg_0$ and appropriate $\dlt$ and $\Lmb$; see Lemma~\ref{lem:dlt-adm-exist}. However, the choice of $U$ and $V$ for this depends highly on the solution, and is part of the reason for the nontrivial dependence of $C_{\overline{\Tht}}$ in Theorem~\ref{thm:L-stability} on the background solution.

\item The reason that we need a parameter $\Lmb$ in the definition of $(\omg, \dlt, \Lmb)$-admissible solutions (as opposed to using only $\dlt$ to capture smallness) is that the estimates we proved in Corollary~\ref{cor:DR-large-r} do \underline{not} imply decay of $\left(\f{1-\mu}{\rd_V r} \rd_V\right)^2 \psi$ towards future timelike infinity along constant-$r$ curves. In order to obtain smallness of $\left(\f{1-\mu}{\rd_V r} \rd_V\right)^2 \psi$ as in \eqref{eq:dlt-adm-dvdvrphi}, we need to take $\Lambda$ large so that $C_{out}$ lies within a large-$r$ region (see the proof of Lemma~\ref{lem:dlt-adm-exist}).

\item The relation \eqref{eq:L-st-ch:dlt} is used in our proof below to control various error terms that grow in $\Lmb$. In practice, it will be used to absorb any logarithmically growing terms in $\Lmb$. To ensure that an arbitrary maximal globally hyperbolic future development of an admissible data set approaches an $(\omg, \dlt, \Lmb)$-admissible solution with \eqref{eq:L-st-ch:dlt} in a neighborhood of timelike infinity, which is crucial in application of Theorem~\ref{thm:L-st-ch}, it is important that $\eta_{0} > 0$ can be chosen arbitrarily small.
This requires us to keep a very careful track of dependence of constants on $\Lmb$ in some part (namely, Theorem~\ref{thm:weak-st} below) of the proof of Theorem~\ref{thm:L-st-ch}.
\item Finally we note that, obviously, $\overline{\mathfrak{L}} \restriction_{\widetilde{\calD}} = \overline{\mathfrak{L}} \restriction_{\calD}$. 
\end{enumerate}
\end{remark}

In the remainder of this subsection, we give a proof of Theorem~\ref{thm:L-stability} assuming Theorem~\ref{thm:L-st-ch}. 
We start with a consequence of Price's law (see Section~\ref{sec:bg}).
\begin{lemma} \label{lem:dlt-adm-exist}
For $\omg_{0} > 2$, consider a maximal globally hyperbolic future development $(\Omgbg, \rbg, \phibg, \ebg)$ of an $\omg_{0}$-admissible Cauchy data. Let $\omg, \eta_{0}$ be positive numbers obeying $2 < \omg < \omg_{0}$ and $\eta_{0} < \frac{1}{200}(\omg_{0} - \omg)$. 
Then given any sufficiently small $\dlt_{0} > 0$, there exists $\Lmb > 100$ (depending on $\omg$, $\omg_0$, $\eta_0$, $\dlt_0$ and the solution) and a neighborhood $\widetilde{\calD}$ of $i^{+}$ of the form
\begin{equation*}
	\widetilde{\calD} = \set{(U, V) \in \PD : U \geq U_{i^{+}}, \ V \geq V_{i^{+}},\, \rbg(U, V) \geq (1 - 2 \dlt) \rbg_{\EH}} 
\end{equation*}
for some $U_{i^{+}}, V_{i^{+}}$ (depending on $\omg$, $\omg_0$, $\eta_0$, $\dlt_0$ and the solution), such that the restriction of $(\Omgbg, \rbg, \phibg, \ebg)$ to $\widetilde{\calD}$ is $(\omg, \dlt, \Lmb)$-admissible with $\dlt = \Lmb^{-100 \eta_{0}} \dlt_{0}$.
\end{lemma}

\begin{proof}
Since only the background solution $(\Omgbg, \rbg, \phibg, \ebg)$ is considered, for notational convenience, {\bf we omit the bar above various quantities in this proof}. For $\omg$, $\eta_{0}$ and $\dlt_{0}$ given as in the statement of the lemma, we need to choose $\Lmb > 100$, $U_{i^{+}}$ and $V_{i^{+}}$. Notice that once $\Lmb$ is chosen, $\de$ is fixed according to $\dlt = \Lmb^{-100 \eta_{0}} \dlt_{0}$, which is then a number smaller than $\dlt_{0}$.

\pfstep{Step~1: Handling the interior of the black hole and the $r$ value on $\EH$} 
Definition~\ref{def:dlt-adm} requires that the solution is $C^2$ in the set $\widetilde{\calD}$ as defined in the statement of the proposition. In particular, we need to control the region in the \emph{interior} of the black hole with $V\geq V_{i^{+}}$ and $r(U, V) \geq (1 - 2 \dlt) r_{\EH}$. While the present paper deals only with the exterior region, because of the Price's law result in Theorem~\ref{thm:DR-full}, the stability theorem in our companion paper \cite{LO.interior} for the interior region is applicable. By the $C^0$-stability theorem of the interior of black hole in \cite{LO.interior}, for $V_{i^{+}}$ sufficiently large and $U'_{int} > U_{\EH}$ sufficiently close to $U_{\EH}$, $r(U,V)$ can be made arbitrarily close to the corresponding $r$-value on Reissner--Nordstr\"om spacetime with parameters $(M,\e)=(\varpi_{\EH},\e)$ when $u\in [U_{\EH}, U'_{int}]$ and $V\geq V_{i^{+}}$. In particular, for any $\de>0$ sufficiently small (depending on $\varpi_{\EH}^{-1} |\e|$), there exists $V_{i^{+}}$ sufficiently large (depending on $\de$ and the solution) such that 
\begin{itemize}
\item On $\EH$, $(1-\de)r_{\EH}\leq r \leq r_{\EH}$.
\item In the interior region, the solution remains $C^2$ in  $\set{(U,V):V\geq V_{i^{+}},\, U\geq U_{\EH},\,r(U, V) \geq (1 - 2 \dlt) r_{\EH}}$ with $ \{r(U,V) =(1 - 2 \dlt) r_{\EH}\}$ being a spacelike hypersurface. 
\end{itemize}
We choose $\Lmb$ large enough so that $\de=\Lmb^{-100 \eta_{0}} \dlt_{0}$ is small enough to guarantee that the above holds.

\pfstep{Step~2: Fixing $\Lambda$} 
Let $\calX_{0}$, $\underline{\calN}$ and $\calN$ be defined as in Section~\ref{subsec:DR-full}. By Corollary~\ref{cor:DR-large-r}, there exists a finite constant $B > 0$ such that
\begin{equation*}
	r^{\omg + 1} \Abs{\left(\frac{1-\mu}{\rd_{V} r} \rd_{V} \right)^{2} \psi} \leq r^{\omg-\omg_{0}} B \leq (\Lmb r_{\EH})^{-(\omg_{0} - \omg)} B \quad \hbox{ in } \calN \cap \set{r \geq \Lmb r_{\EH}}.
\end{equation*}
We fix $\Lmb > 100$ to be sufficiently large so that, in addition to the requirement is Step~1, we have that for every $U \in \bbR$,
\begin{equation*} 
	\sup_{C_{U} \cap \set{r \geq \Lmb r_{\EH}}} \frac{r^{\omg+1}}{r_{\EH}^{\omg}} \Abs{\left(\frac{1-\mu}{\rd_{V} r} \rd_{V} \right)^{2} \psi} \leq \Lmb^{- 100 \eta_{0}} \dlt_{0} = \dlt.
\end{equation*}

\pfstep{Step~3: Choosing $U_{i^{+}}$ depending on $V_{i^{+}}$} For any $V_{i^{+}}$, we choose $U_{i^{+}}$ so that $(U_{i^{+}}, V_{i^{+}}) \in \set{r = \Lmb r_{\EH}}$. Since $\Lmb$ is large, note that $\set{r = \Lmb r_{\EH}}$ is a timelike curve in the black hole exterior, and the point $(U_{i^{+}}, V_{i^{+}})$ exists if $V_{i^{+}}$ is sufficiently large. Moreover, $U_{i^{+}} \to \infty$ as $V_{i^{+}} \to \infty$. 

\pfstep{Step~4: Fixing $V_{i^{+}}$} It remains to determine $V_{i^{+}}$. Recall that according to Step~1, we already need to choose $V_{i^+}$ large to show that $\widetilde{\calD}$ is a regular region in the maximal globally hyperbolic future development. It remains to choose $V_{i^{+}}$ possibly even larger in order to guarantee that \eqref{eq:dlt-adm:m}, \eqref{eq:dlt-adm:phi}--\eqref{eq:dlt-adm:dvrphi} and \eqref{eq:dlt-adm-duduphi} hold. By Corollary~\ref{cor:DR-small-r}, for $V_{i^{+}}$ large enough, we have
\begin{equation*} 
	\sup_{V \geq V_{i^{+}}} \sup_{\uC_{V} \cap \set{r \leq \Lmb r_{\EH}}} r_{\EH}^{2} \Abs{\left(\frac{1}{\rd_{U} r} \rd_{U} \right) \phi}^{2} \leq \Lmb^{-3/2} \dlt.
\end{equation*}
Moreover, by Corollary~\ref{cor:DR-final} (with $\omg'$ chosen so that $\omg < \omg' < \omg_{0}$) and Proposition~\ref{prop:bg-uv}, we may ensure that the decay assumptions \eqref{eq:dlt-adm:phi}--\eqref{eq:dlt-adm:dvrphi} and the mass difference bound \eqref{eq:dlt-adm:m} hold in $J^{-} (\NI) \cap \set{(U, V) : U \geq U_{i^{+}}, \ V \geq V_{i^{+}}}$ by taking $V_{i^{+}}$ (and thus also $U_{i^{+}}$) sufficiently large. Therefore, it follows that the solution restricted to $\widetilde{\calD}$ (defined in the statement of the lemma) is $(\omg, \dlt, \Lmb)$-admissible. \qedhere
\end{proof}

We now combine Theorem~\ref{thm:L-st-ch} with the Cauchy and large-$r$ stability results (Propositions~\ref{prop:cauchy-st}, \ref{prop:extr-st-cauchy} and \ref{prop:extr-st}) to prove Theorem~\ref{thm:L-stability}.

\begin{proof}[Proof of Theorem~\ref{thm:L-stability} assuming Theorem~\ref{thm:L-st-ch}]
Let $\omg_{0} > 2$ and consider $\omg_{0}$-admissible Cauchy data sets $\Tht, \overline{\Tht}$ such that $d_{1, \omg_{0}}^{+}(\Tht, \overline{\Tht}) < \eps$. Our aim is to show that $\abs{\mathfrak{L} - \overline{\mathfrak{L}}} \leq C_{\overline{\Tht}} \eps$ for a sufficiently small $\eps > 0$ (depending on $\overline{\Tht}$). (Observe that the conclusion of the theorem is obvious if $\eps$ is not small, i.e., for $\eps$ larger than a fixed constant depending on $\overline{\Tht}$.)

Let $2 < \omg < \omg_{0}$ and $\eta_{0} < \frac{1}{200}(\omg_{0} - \omg)$ be fixed. Let $\dlt_{0} > 0$ be a constant to be specified below. We consider a subset (to be defined precisely below as $\cup_{i=1}^4 \calR_i$) of the maximal globally hyperbolic future development of $\overline{\Tht}$ which contains the connected component of the exterior region with $\rd_V \rbg>0$ and $\rd_U \rbg < 0$. Using the initial-data-normalized coordinates $(U, V)$, we define the following regions:
\begin{figure}[h]
\begin{center}
\def\svgwidth{300px}
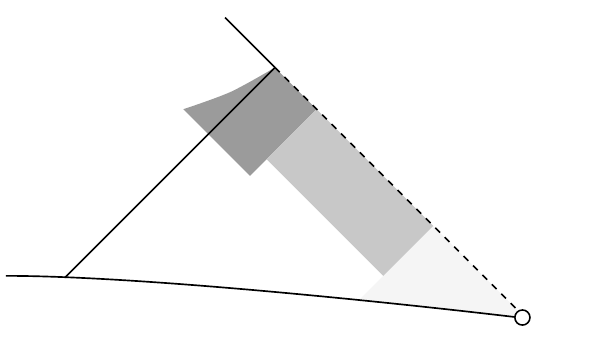 
\caption{Regions $\calR_{1}, \ldots, \calR_{4}$} \label{fig:L-stability}
\end{center}
\end{figure}

\begin{itemize}
\item $\calR_{4}$: We define $\calR_{4}$ to be a neighborhood of $i^{+}$ as in Lemma~\ref{lem:dlt-adm-exist}, i.e.,
\begin{equation*}
\calR_{4} = \widetilde{\calD} = \set{(U, V) \in \PD : U \geq U_{i^{+}}, \ V \geq V_{i^{+}}, \ \rbg(U, V) \geq (1 - 2 \dlt) \rbg_{\EH}},
\end{equation*}
so that the restriction of the background solution to $\calR_{4}$ is $(\omg, \dlt, \Lmb)$-admissible with $\dlt = \Lmb^{-100 \eta_{0}} \dlt_{0}$. We fix $\dlt_{0} > 0$ so that Theorem~\ref{thm:L-st-ch} is applicable. We denote the future endpoint of the initial incoming curve (i.e., $\uC_{V_{i^{+}}} \cap \set{\rbg = (1 - 2 \dlt) \rbg_{\EH}}$) by $(U_{int}, V_{i^{+}})$. 

\item $\calR_{3}$: We define $\calR_{3}$ to be a characteristic rectangle as in Proposition~\ref{prop:extr-st}, i.e.,
\begin{equation*}
	\calR_{3} = \calR = \set{(U, V) \in \PD : U_{0} \leq U \leq U_{i^{+}}, \ V \geq V_{0}},
\end{equation*}
for some $(U_{0}, V_{0}) \in \PD$. We fix $V_{0}$ here, and defer the choice of $U_{0}$ until later.

We choose $U_{\ast} = U_{0}$ so that \eqref{eq:extr-st:hyp:duups} holds vacuously. By Proposition~\ref{prop:bg-large-r}, Theorem~\ref{thm:DR-full} and Corollary~\ref{cor:DR-large-r}, there exists $B \geq 1$ such that \eqref{eq:extr-st:hyp:psi} and \eqref{eq:extr-st:hyp:dvpsi} hold on $\calN \cap \set{U_{0} \leq U \leq U_{i^{+}}}$ for any $\rbg_{0} \geq 1$. Moreover, by Remark~\ref{rem:DR-large-r:dUr}, \eqref{eq:extr-st:hyp:dur} holds on any $\uC_{V_{0}} \cap \set{U \leq U_{i^{+}}}$ inside $\calN$ for some $B$ depending on the background solution and $U_{i^{+}}$. Fixing a large enough $V_{0}> V_{i^{+}}$ (depending on $B$ and $\omg$) so that $\rbg(U_{i^{+}}, V_{0})$ is sufficiently large, we may ensure that \eqref{eq:extr-st:hyp:r-m} and \eqref{eq:extr-st:hyp:m-e} hold with $\rbg_{0} = \rbg(U_{i^{+}}, V_{0})$ and an appropriate $\dlt_{\calR} > 0$ so that Proposition~\ref{prop:extr-st} is applicable on $\calR_{2}$.

\item $\calR_{2}$: We define $\calR_{2}$ to be as in Proposition~\ref{prop:extr-st-cauchy}, i.e.,
\begin{equation*}
	\calR_{2} = \underline{\calR} = \set{(U, V) \in \PD : U \leq U_{0}}.
\end{equation*}
Observe that $\calR_{2}$ is the domain of dependence of $\Sgm_{0} \cap \calR_{2}$. Taking $U_{0}$ sufficiently negative, we may ensure that $\calR_{2} \subseteq \underline{\calN}$, so that by Proposition~\ref{prop:bg-large-r}, \eqref{eq:extr-st-cauchy:hyp:dur}--\eqref{eq:extr-st-cauchy:hyp:dvpsi} hold in $\calR_{2}$ for any $\rbg_{0} \geq 1$ for some $\underline{B} \geq 0$. Choosing $U_{0}$ even more negative if necessary, we may ensure that \eqref{eq:extr-st-cauchy:hyp:r-m} and \eqref{eq:extr-st-cauchy:hyp:m-e} hold with $\rbg_{0} = \rbg(U_{0}, \underline{V}(U_{0}))$ and a small enough $\dlt_{\underline{\calR}} > 0$ (depending on $\omg, \underline{B}$) so that Proposition~\ref{prop:extr-st-cauchy} is applicable.

\item $\calR_{1}$: Finally, we define $\calR_{1}$ to be the domain of dependence of $\Sgm_{0, \ast} = \Sgm_{0} \cap \set{(U, V) \in \PD : U \leq U_{int}, \ V \leq V_{0}}$. 
\end{itemize}

Notice that some of the $\calR_i$ defined above are not disjoint, but it is clear that their union covers the connected component of the exterior region of the maximal globally hyperbolic future development of $\overline{\Tht}$ with $\rd_U \rbg<0$ and $\rd_V \rbg>0$.

We split $\mathfrak{L} - \overline{\mathfrak{L}}$ according to the above decomposition as follows. Given $U_{1} < U_{2}$, let
\begin{equation*}
	\mathfrak{L}(U_{1}, U_{2}) = \int_{U_{1}}^{U_{2}} 2 M(U) \Phi(U) \Gmm(U) \, \ud U, \quad
	\overline{\mathfrak{L}}(U_{1}, U_{2}) = \int_{U_{1}}^{U_{2}} 2 \overline{M}(U) \overline{\Phi}(U) \overline{\Gmm}(U) \, \ud U.
\end{equation*}
Then we decompose
\begin{align*}
	\mathfrak{L} - \overline{\mathfrak{L}}
	& = (\mathfrak{L} - \overline{\mathfrak{L}})(-\infty, U_{0}) 
	+ (\mathfrak{L} - \overline{\mathfrak{L}})(U_{0}, U_{i^{+}}) 
	+ (\mathfrak{L} \restriction_{\calD} - \overline{\mathfrak{L}}\restriction_{\calD}),
\end{align*}
where we have used the notation\footnote{Note that in general, in the perturbed solution, the event horizon does not lie inside $\{U=U_{\EH}\}$. We therefore do not compare the solutions in $\calD$ using the $(U,V)$ coordinate system.} as in Theorem~\ref{thm:L-st-ch}, that $\calD$ denotes both the exterior region of $(\omg,\de,\Lmb)$-admissible background in $\calR_4$ and the exterior region of the perturbed solution restricted to $U\geq U_{i^+}$ and $V\geq V_{i^+}$.
We will treat each term on the RHS in order, using Propositions~\ref{prop:cauchy-st}, \ref{prop:extr-st-cauchy} and \ref{prop:extr-st}, as well as Theorem~\ref{thm:L-st-ch}.

In what follows, we allow the constant $C > 0$ to depend on the background solution (in particular, on $\omg$, $B$ or $\underline{B}$ defined above). We also take $\eps$ smaller and smaller during the argument (again depending on the background solution), so that the stated results can be applied.
 
By Lemma~\ref{lem:cauchy-to-char}, Proposition~\ref{prop:extr-st-cauchy} is applicable in $\calR_{2}$ with $\eps_{2} \leq C \eps$. Hence, we obtain
\begin{equation} \label{eq:L-stability:pf:R2}
 \abs{(\mathfrak{L} - \overline{\mathfrak{L}})(-\infty, U_{0})} \leq C \eps. 
\end{equation}
Moreover, in combination with Proposition~\ref{prop:cauchy-st} and Corollary~\ref{cor:cauchy-st:psidf-phidf} in $\calR_{1}$, Proposition~\ref{prop:extr-st} is applicable in $\calR_{3}$ with $\eps_{3} \leq C \eps$. It follows that
\begin{equation} \label{eq:L-stability:pf:R3}
 \abs{(\mathfrak{L} - \overline{\mathfrak{L}})(U_{0}, U_{i^{+}})} \leq C \eps. 
\end{equation}
Next, combining Proposition~\ref{prop:extr-st} with Proposition~\ref{prop:cauchy-st} and Corollary~\ref{cor:cauchy-st:psidf-phidf} in $\calR_{1}$, we see that Theorem~\ref{thm:L-st-ch} is applicable in $\calR_{4}$ with $\eps_{0} \leq C \eps$. Therefore,
\begin{equation} \label{eq:L-stability:pf:R4}
\abs{\mathfrak{L} \restriction_{\calD} - \overline{\mathfrak{L}} \restriction_{\calD} } \leq C \eps. 
\end{equation}
Putting together \eqref{eq:L-stability:pf:R2}--\eqref{eq:L-stability:pf:R4}, the desired estimate for $\mathfrak{L} - \overline{\mathfrak{L}}$ follows. \qedhere
\end{proof}

The remainder of this section is devoted to the proof of Theorem~\ref{thm:L-st-ch}. 

\subsection{Coordinate choice and a weak stability result} \label{subsec:weak-st}
To ensure that the difference between the two solutions \emph{decays}, which is needed for stability of the quantity $\mathfrak{L} \restriction_{\calD}$ asserted in Theorem~\ref{thm:L-st-ch}, it is desirable to compare the solutions in a coordinate system that is normalized in the future (more precisely, along the event horizon and null infinity) of each solution. On the other hand, the assumptions \eqref{eq:L-st-ch:ini-duphi}--\eqref{eq:L-st-ch:ini-e} on the difference of initial data are made in a different coordinate system (essentially inherited from the coordinates normalized on the Cauchy initial hypersurface $\Sgm_{0}$ in the context of the whole proof of Theorem~\ref{thm:L-stability}). Relating the initial data differences in these alternate coordinate systems turns out to be a rather delicate issue\footnote{In previous works (in particular, see Holzegel \cite{Hol.biaxial}) this problem was not as serious, since there existed a coordinate-invariant way to measure the difference from the background solution.}, as it requires us to carefully control the constants (especially their dependence on $\Lmb = \frac{\rbg(1,1)}{\rbg_{\EH}}$) that arise in the proof.

We take a modular approach to this issue. First, we will establish a weak stability result (Theorem~\ref{thm:weak-st}), which is not enough to establish Theorem~\ref{thm:L-st-ch}, but nevertheless is sufficiently strong to provide a nice control on the transformation from the original coordinates to the future-normalized coordinates. Then, working solely in the future-normalized coordinates, and without the need for a careful control of constants anymore, we will carry out the $r^{p}$-weighted energy estimate argument \`a la Dafermos--Rodnianski (Section~\ref{subsec:rp-weight}) and an integration along characteristics argument from \cite{LO1} (Section~\ref{subsec:int-char}), from which Theorem~\ref{thm:L-st-ch} will follow.

The goal of this subsection is to formulate the above-mentioned weak stability result (Theorem~\ref{thm:weak-st}), whose proof will span Sections~\ref{subsec:btstrp}--\ref{subsec:btstrp-pf}. 

\subsubsection*{Coordinate choice}
To compare the two solutions $(g, \phi, F)$ and $(\gbg, \phibg, \Fbg)$ so that the decay towards timelike infinity of the difference is apparent, we utilize the future-normalized null coordinates adapted to each solution. More precisely, for the perturbed solution $(g, \phi, F)$, we use the double null  coordinates $(u, v) = (u(U), v(V))$, which are normalized (asymptotically) along $\NI$ and $\EH$, respectively, as follows\footnote{The fact that the maximal globally hyperbolic future development $(g, \phi, F)$ possesses complete null infinity $\NI$ and an event horizon $\EH$ follows from general results about maximal globally hyperbolic future developments of admissible data, see \cite{LO.interior}.}:
\begin{itemize}
\item At the point $(U, V) = (1, 1)$, we have $(u, v) = (1, 1)$. 
\item For $U > 1$, $u(U)$ is normalized along $\NI$ by the requirement
\begin{equation}  \label{eq:coord-u}
	\lim_{V \to V_{\NI}} \frac{\rd_{u} r}{1-\mu}(u, V) = -1.
\end{equation}
\item For $V > 1$, $v(V)$ is normalized along $\EH$ by the requirement
\begin{equation} \label{eq:coord-v}
	\frac{\rd_{v} r}{1-\mu}(U_{\EH}, v) = 1.
\end{equation}
\end{itemize}
For the background solution $(\gbg, \phibg, \Fbg)$, we use the double null coordinates $(\ubg, \vbg) = (\ubg(U), \vbg(V))$, which is defined in the same fashion with respect to $(\gbg, \phibg, \Fbg)$:
\begin{itemize}
\item At the point $(U, V) = (1, 1)$, we have $(\ubg, \vbg) = (1, 1)$. 
\item For $U > 1$, $\ubg(U)$ is normalized along $\NI$ by the requirement
\begin{equation}  \label{eq:coord-u-bg}
	\lim_{V \to \overline{V}_{\NI}} \frac{\rd_{\ubg} \rbg}{1-\mubg}(\ubg, V) = -1.
\end{equation}
\item For $V > 1$, $\vbg(V)$ is normalized along $\EH$ by the requirement
\begin{equation} \label{eq:coord-v-bg}
	\frac{\rd_{\vbg} \rbg}{1-\mubg}(\overline{U}_{\EH}, \vbg) = 1.
\end{equation}
\end{itemize}
Note that the coordinates $(\ubg, \vbg)$ here coincide with the future-normalized coordinates $(u, v)$ in Section~\ref{sec:bg} up to an overall translation $(\ubg, \vbg) \mapsto (\ubg + \ubg_{\ast}, \vbg + \vbg_{\ast})$ for an appropriate $(\ubg_{\ast}, \vbg_{\ast}) \in \bbR^{2}$. We also remind the reader that $(u, v)$ and $(\ubg, \vbg)$ only cover the black hole exterior region in $(g, \phi, F)$ and $(\gbg, \phibg, \Fbg)$, respectively.

\subsubsection*{Some conventions and definitions}
In the remainder of this section, {\bf we set $\rbg_{\EH} = 1$, using the homogeneity symmetry $(r, \phi, \varpi, \e)\mapsto (a r, \phi, a \varpi, a \bfe)$}. Note that our gauge normalization conditions require us to concurrently make the coordinate transform $(u, v) \mapsto (a u, a v)$. We remark that we have formulated Definition~\ref{def:dlt-adm} and Theorem~\ref{thm:instability} carefully so that the parameters $\omg, \dlt, \Lmb, \eps_{0}$ and the implicit constants are \emph{independent} of these transformations.

Since the numbers\footnote{We note that upon fixing $\rbg_{\EH}=1$, prescribing $|\ebg|$ is equivalent to prescribing $\varpi_f^{-1}|\ebg|$ due to \eqref{r.varpi.poly} and Corollary~\ref{varpi.same.limit}.} $|\ebg|$, $\omg > 2$ and $0 < \eta_{0} < \min\{2\omg-4, 1\}$ are fixed for the purpose of Theorem~\ref{thm:L-st-ch}, we henceforth {\bf suppress the dependence of constants on $|\ebg|$, $\omg$ and $\eta_{0}$,} for simplicity of notation.

In what follows,  the solution $(g, \phi, F)$ in the $(u, v)$ coordinates will be compared with the background solution $(\gbg, \phibg, \Fbg)$ in the $(\ubg, \vbg)$ coordinates. That is, we consider differences such as $r(u,v) - \rbg(\ubg, \vbg)$ with the same numerical values $(u, v) = (\ubg, \vbg)$. {\bf For this reason, from now on, we omit the bar in $\ubg, \vbg$ and simply write $(\gbg, \phibg, \Fbg) (u, v)$ for $(\gbg, \phibg, \Fbg) (\ubg, \vbg)$ with $(u, v) = (\ubg, \vbg)$.}

With the preceding conventions, we introduce the quantities
\begin{gather*}
	\dvr = \rd_{v} r, \quad \dur = \rd_{u} r, \quad
	\kpp = \frac{\dvr}{1-\mu}, \quad \gmm = \frac{\dur}{1-\mu}, \\
	\dvrbg = \rd_{v} \rbg, \quad \durbg = \rd_{u} \rbg, \quad
	\kppbg = \frac{\dvrbg}{1-\mubg}, \quad \gmmbg = \frac{\durbg}{1-\mubg}.
\end{gather*}
By the above coordinate choice, we have $\kpp = \kppbg = 1$ on $\EH$, which corresponds to $\set{u = \infty}$, and $\gmm = \gmmbg = -1$ on $\NI$, which corresponds to $\set{v = \infty}$.

\subsubsection*{Statement of weak stability}
We introduce the initial difference size $\eps^{2}$ in the above coordinates as follows (recall that $\rbg_{\EH} = 1$):
\begin{equation} \label{eq:eps-def}
\begin{aligned}
	\eps^{2} 
= 	& \int_{C_{out}} \kpp^{-1} (\rd_{v} (\phi - \phibg))^{2} r^{2} \, \ud v 
	+ \int_{C_{out}} \kpp^{-1} \left( \rd_{v} (r (\phi - \phibg)) \right)^{2} \log(1+(r/\rbg_{\EH})) \, \ud v \\
	& + \int_{\uC_{in}} \frac{1}{(-\dur)} (\rd_{u} (\phi - \phibg))^{2} r^{2} \, \ud u  
	+ r \abs{\phi - \phibg}^{2} \restriction_{C_{out} \cap \uC_{in}}		\\
	& + \abs{\varpi - \varpibg}^{2} \restriction_{C_{out} \cap \uC_{in}}
	+ \Abs{r - \rbg}^{2} \restriction_{C_{out} \cap \uC_{in}} 
	+ \abs{\e - \ebg}^{2} .
\end{aligned}
\end{equation}
Note that $\eps^{2}$ is dependent on the coordinate choice for the two solutions $(g, \phi, F)$ and $(\gbg, \phibg, \Fbg)$. One of the consequences of the weak stability theorem (Theorem~\ref{thm:weak-st}) will be a bound for the quantity $\eps^{2}$ in terms of the initial difference $\eps_{0}^{2}$ in the original coordinates $(U, V)$; see \eqref{eq:gauge:eps} below.

Let $R_{0} = r(1, 1)$.  Given $u \in [1, \infty)$, let $v_{R_{0}}(u) \in [1, \infty)$ be defined by the condition $r(u, v_{R_{0}}(u)) = R_{0}$ (existence of such a $v_{R_{0}}(u)$ is guaranteed by the condition $r(1, 1) = R_{0}$ and monotonicity of $r$). In other words, $(u, v_{R_{0}}(u))$ is a parametrization of the curve $\gmm_{R_{0}}$. We now introduce the foliation $(\Gmm_{\tau})_{\tau \in [1, \infty)}$ of the domain $\calD$, which is defined as follows:
\begin{align*}
	\Gmm_{\tau} = & \Gmm_{\tau}^{(in)} \cup \Gmm_{\tau}^{(out)}, \\
	\Gmm_{\tau}^{(out)} = & \set{(u, v) \in \calD : u = \tau, \ v \in [v_{R_{0}}(\tau), \infty)}, \\
	\Gmm_{\tau}^{(in)} = & \set{(u, v) \in \calD : u \in [\tau, \infty), \ v  = v_{R_{0}}(\tau)}.
\end{align*}
We also introduce the definition
\begin{align*}
	\calD(\tau_{1}, \tau_{2}) =& \bigcup_{\tau \in (\tau_{1}, \tau_{2})} \Gmm_{\tau}.
\end{align*}
With the above coordinate choice, $\EH$ and $\NI$ are represented by the ``curves at infinity'' $\set{(\infty, v) : v \in [1, \infty)}$ and $\set{(u, \infty) : u \in [1, \infty)}$ for both the background and perturbed solutions. Accordingly, we define
\begin{align*}
	\EH(\tau_{1}, \tau_{2}) =& \set{(\infty, v) : v \in [v_{R_{0}}(\tau_{1}), v_{R_{0}}(\tau_{2})]} , \\
	\NI(\tau_{1}, \tau_{2}) =& \set{(u, \infty) : u \in [\tau_{1}, \tau_{2}]} .
\end{align*}

\begin{figure}[h]
\begin{center}
\def\svgwidth{250px}
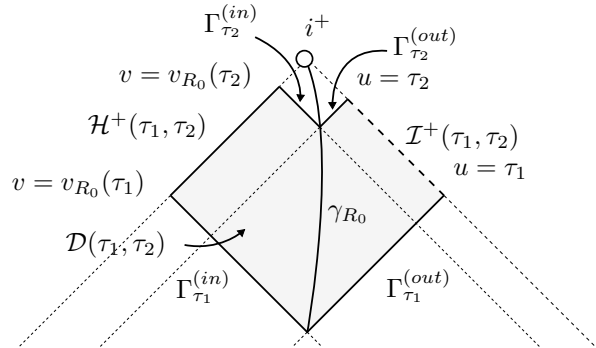 
\caption{Foliation $\Gmm_{\tau} = \Gmm_{\tau}^{(in)} \cup \Gmm_{\tau}^{(out)}$} \label{fig:energy}
\end{center}
\end{figure}

We are now ready to state our weak stability result. Recall the definition of $\eps_{0}$ from Theorem~\ref{thm:L-st-ch}. 
\begin{theorem} [Weak stability] \label{thm:weak-st}
Assume that the hypotheses of Theorem~\ref{thm:L-st-ch} hold for sufficiently small $\dlt_{0} > 0$ (depending on $\abs{\ebg}$, $\eta_0$ and $\omg$) and $\eps_{0} > 0$ (depending on $\dlt_0$, $\abs{\ebg}$, $\eta_0$, $\omg$ and $\Lmb$) with the normalization condition $\rbg_{\EH}=1$. Then there exists a constant $C > 0$ (depending on $\abs{\ebg}$, $\eta_0$ and $\omg$) such that the following statements hold.
\begin{enumerate}
\item The initial difference size $\eps^{2}$ in the coordinates $(u, v)$ and $(\ubg, \vbg)$ is bounded by $\eps_{0}^{2}$:
\begin{align} \label{eq:gauge:eps} 
	\eps^{2} \leq & C \eps_{0}^{2}. 
\end{align}
Moreover, the following pointwise estimate for $\rd_{v}(r \phi)$ on $C_{out}$ holds:
\begin{equation} \label{eq:gauge:dvrphi}	
	\sup_{C_{out}} r^{\omg} \Abs{\frac{1}{\dvr} \rd_{v}(r \phi - \rbg \phibg)}
	\leq C \eps_{0}.		
\end{equation}
\item Boundedness of nondegenerate energy and integrated decay of local energy hold for the difference $\phi - \phibg$. That is, for any $1 \leq \tau_{1} \leq \tau_{2} < \infty$, we have
\begin{equation} \label{eq:energy}
\begin{aligned}
	& \int_{\Gmm^{(in)}_{\tau_{2}}} \frac{1}{(-\dur)} (\rd_{u} (\phi - \phibg))^{2} r^{2} \, \ud u
	+ \int_{\Gmm^{(out)}_{\tau_{2}}} \kpp^{-1} (\rd_{v} (\phi - \phibg))^{2} r^{2} \, \ud v  \\
	& \int_{\NI(\tau_{1}, \tau_{2})} \frac{1}{(-\dur)} (\rd_{u} (\phi - \phibg))^{2} r^{2} \, \ud u
	+ \int_{\EH(\tau_{1}, \tau_{2})} \kpp^{-1} (\rd_{v} (\phi - \phibg))^{2} r^{2} \, \ud v  \\
	& + \iint_{\calD(\tau_{1}, \tau_{2})} \frac{1}{r^{1+\eta_{0}}} \left( \frac{1}{(-\dur)} (\rd_{u} (\phi - \phibg))^{2} + \kpp^{-1} (-\dur) (\rd_{v} (\phi - \phibg))^{2} + \frac{(-\dur)}{r^{2}} (\phi - \phibg)^{2}\right) r^{2} \, \ud u \ud v\\
	& \leq  C \left( \int_{\Gmm^{(in)}_{\tau_{1}}} \frac{1}{(-\dur)} (\rd_{u} (\phi - \phibg))^{2} r^{2} \, \ud u
	+ \int_{\Gmm^{(out)}_{\tau_{1}}} \kpp^{-1} (\rd_{v} (\phi - \phibg))^{2} r^{2} \, \ud v \right) 
		+ C (\log^{2} \Lmb) \tau_{1}^{-2 \omg+1+\eta_{0}} \dlt^{2} \eps^{2}.
\end{aligned}
\end{equation}
Furthermore, $\rd_{v} (\phi - \phibg)$ enjoys an improved integrated decay estimate:
\begin{equation} \label{eq:energy-dvphi}
\iint_{\calD(1, \infty) \cap \set{r \geq 20}} \frac{1}{r} \dvr^{-1} (\rd_{v} (\phi - \phibg))^{2} r^{2} \, \ud u \ud v
\leq C \eps^{2},
\end{equation}
and the difference $\phi - \phibg$ obeys the following pointwise estimate:
\begin{equation} \label{eq:energy-phidf-pt}
	\sup_{\calD(1, \infty)} r \abs{\phi - \phibg}^{2} \leq C \eps^{2}.
\end{equation}

\item The following geometric bounds hold for the perturbed solution:
\begin{align} 
	1/4 \leq \abs{\varpi} \leq 1, \quad \abs{\e} \leq & 1, \label{eq:weak-st:me-bnd} \\
	1 \leq \kpp \leq & 2, \label{eq:weak-st:kpp-bnd} \\
	1/2 \leq \dvr \leq & 2 \quad \hbox{ in } \set{r \geq 20}. \label{eq:weak-st:dvr-large-r} \\
	0 < -\dur \leq & C e^{-c_{(\dur)} (u - v - C_{\gmm_{20}})} \quad \hbox{ in } \set{r \leq 20}, \label{eq:weak-st:dur-small-r} \\
	1/4 \leq -\dur \leq & 1 \quad \hbox{ in } \set{r \geq 20}, \label{eq:weak-st:dur-large-r} \\
	1/2 \leq 1-\mu \leq & 1 \quad \hbox{ in } \set{r \geq 20}. \label{eq:weak-st:mu-large-r} 
\end{align}
where $c_{(\dur)} > 0$ is a constant depending on $\ebg$, and $C_{\gmm_{20}}$ is defined by $r(1+C_{\gmm_{20}}, 1) = 20$. 

\item
On $\gmm_{20} = \set{r = 20}$ and $\gmm_{R_{0}} = \set{r = R_{0}}$, we have
\begin{align} 
	\sup_{\gmm_{20}} \abs{u - v - C_{\gmm_{20}}} \leq  & \frac{1}{10}, \label{eq:weak-st:v-u}  \\
	\sup_{\gmm_{R_{0}}} \abs{u - v} \leq & \frac{1}{10}. \label{eq:weak-st:v-u-R0} 
\end{align}
Moreover, we have
\begin{equation} \label{eq:weak-st:v-u-r}
	\frac{1}{4} (v - u + C_{\gmm_{20}} + 20) \leq r(u, v) \leq 2 (v - u + C_{\gmm_{20}} + 20) \quad \hbox{ in } \calD \cap \set{r \geq 20}.
\end{equation}

\item Finally, we have the following difference bounds:
\begin{align}
	\Abs{\varpi - \varpibg} \leq & C \eps, \label{eq:weak-st:mdf} \\
	\Abs{\log \kpp - \log \kppbg} \leq & C \eps, 		\label{eq:weak-st:kppdf} \\
	\Abs{r - \rbg} \leq & C \max\set{\log \Lmb, \log r} \eps, \label{eq:weak-st:rdf} \\
	\Abs{\dvr - \dvrbg} \leq & C \log \Lmb \eps, 		\label{eq:weak-st:dvrdf} \\
	\Abs{\dur - \durbg} \leq & \left\{
	\begin{array}{ll}
	\displaystyle{C \log \Lmb e^{- \frac{1}{2} c_{(\dur)} (u - v - C_{\gmm_{20}})}} \eps & \hbox{ in } \set{r \leq 20}, \\
	\displaystyle{C \log \Lmb r^{-1}\eps}  & \hbox{ in } \set{r \geq 20}.
	\end{array} \right. 					\label{eq:weak-st:durdf} 
\end{align}
\end{enumerate}
\end{theorem}

As it can be inferred from the statement, the proof of Theorem~\ref{thm:weak-st} relies on establishing $L^{2}$ estimates for the scalar field difference $\phi - \phibg$. These estimates are natural to use for controlling the differences $\varpi - \varpibg$, $\log \kpp - \log \kppbg$ and $\log (-\gmm) - \log (-\gmmbg)$, in view of the equations they satisfy. Such a control leads, in turn, to the geometric difference bounds \eqref{eq:weak-st:rdf}--\eqref{eq:weak-st:durdf} and the control of the coordinate transformations $(U, V) \to (u, v), \, (\ubg, \vbg)$ asserted in \eqref{eq:gauge:eps}, \eqref{eq:gauge:dvrphi}. Furthermore, the $L^{2}$ estimate \eqref{eq:energy} sets stage for the decay argument of Dafermos--Rodnianski, which is the next step of the proof of Theorem~\ref{thm:L-st-ch}.

\subsubsection*{Outline of Sections~\ref{subsec:btstrp}--\ref{subsec:weak-st-pf}}
Sections~\ref{subsec:btstrp}--\ref{subsec:weak-st-pf} are devoted to the proof of Theorem~\ref{thm:weak-st}. For the convenience of the reader, we provide below a short outline of these subsections.
\begin{itemize}
\item {\bf Section~\ref{subsec:btstrp}.} We set up the main bootstrap argument for the proof of Theorem~\ref{thm:weak-st}. In particular, we define the bootstrap domain $\calD_{(t_{B})}$ and the corresponding coordinates $(u_{(t_{B})}, v_{(t_{B})})$ for the perturbed solution, which approach $\calD$ and the future-normalized coordinates $(u, v)$, respectively, as $t_{B} \to \infty$. The bootstrap assumptions include pointwise and integrated bounds for $\varpi - \varpibg$, $\log \kpp - \log \kppbg$ and $\log (-\gmm) - \log(-\gmmbg)$.
\item {\bf Section~\ref{subsec:bg-geom}.} We prove bounds on the background geometric quantities, such as $\rbg$, $\dvrbg$ and $\durbg$. We also show that $\abs{\log \kpp - \log \kppbg}$ and $\abs{\log (-\gmm) - \log (-\gmmbg)}$ are appropriately small on the future boundary of the bootstrap domain.
\item {\bf Section~\ref{subsec:geom}.} We prove bounds on the difference of geometric quantities, such as $r - \rbg$, $\dvr - \dvrbg$ and $\dur - \durbg$, from the bootstrap assumptions. Moreover, we show that $\kppbg$ and $-\gmmbg$ on the future boundary of $\calD_{(t_{B})}$ asymptote to $1$ as $t_{B} \to \infty$, which justifies the set-up in Section~\ref{subsec:btstrp}.
\item {\bf Section~\ref{subsec:en-id}.} We perform some vector field multiplier computation and derive energy identities that we use later. We also prove $L^{2}$-Hardy-type inequalities on null curves.
\item {\bf Section~\ref{subsec:en-pf}.} Using the identities in Section~\ref{subsec:en-id}, as well as the geometric difference bounds in Section~\ref{subsec:geom}, we show uniform boundedness of nondegenerate energy and integrated decay of local energy of $\phi - \phibg$.
\item {\bf Section~\ref{subsec:btstrp-pf}.} From the energy estimates established in Section~\ref{subsec:en-pf}, we close the bootstrap assumptions on $\varpi - \varpibg$, $\log \kpp - \log \kppbg$ and $\log (-\gmm) - \log(-\gmmbg)$.
\item {\bf Section~\ref{subsec:gauge}.} We use the geometric difference bounds in Section~\ref{subsec:geom} to estimate the coordinate change from $(U, V)$ to $(u_{(t_{B})}, v_{(t_{B})})$ and $(\ubg, \vbg)$. In particular, we estimate the initial difference size $\eps^{2}$ in the new coordinates in terms of $\eps_{0}^{2}$ in the original coordinates.
\item {\bf Section~\ref{subsec:weak-st-pf}.} Finally, we put together the results established so far and complete the proof of Theorem~\ref{thm:weak-st}.
\end{itemize}

\subsection{Setting up a bootstrap argument for weak stability} \label{subsec:btstrp}
The proof of Theorem~\ref{thm:weak-st} proceeds by a bootstrap argument. 
The purpose of this subsection is to set up this argument and provide the necessary tools to close it. 

\subsubsection*{Coordinate choice and bootstrap domain}
Recall the notation $R_{0} = r(1, 1)$. On the perturbed spacetime $(\M, g)$, consider the constant-$r$ curve
$\gmm_{R_{0}} = \set{(U, V) : r(U, V) = R_{0}}$
and the smooth future pointing time-like vector field $T$ (restricted to $\gmm_{R_{0}}$):
\begin{equation} \label{eq:vf-T-0}
	T  = - \frac{1-\mu}{\rd_{U} r} \rd_{U} + \frac{1-\mu}{\rd_{V} r} \rd_{V},
\end{equation}
which is tangent to $\gmm_{R_{0}}$.
For $t \geq 0$, we denote by $\gmm_{R_{0}}(t)$ (abusing the notation a bit for convenience) the parametrization of $\gmm_{R_{0}}$ so that $T = \dot{\gmm}_{R_{0}}$ and $\gmm_{R_{0}}(0)$ is the point $(1, 1)$ in the $(U, V)$ coordinates. We use the notation $(U \circ \gmm_{R_{0}}, V \circ \gmm_{R_{0}})(t)$ for the representation of this curve in the $(U,V)$ coordinates.

Given $t_{B} > 0$, we define the bootstrap domain $\calD_{(t_{B})}$ to be the causal past of the point $\gmm_{R_{0}}(t_{B})$ in $\calD$. We construct the coordinate $u_{(t_{B})} = u_{(t_{B})}(U)$ by normalizing $-(1-\mu)^{-1} \rd_{u_{(t_{B})}} r = 1$ along the final incoming curve $\uC_{(t_{B}) in}^{f} = \set{(U, V) : 1 \leq U \leq U \circ \gmm_{R_{0}}(t_{B}), \ V = V \circ \gmm_{R_{0}}(t_{B})}$ as follows:
\begin{align}
\frac{\ud u_{(t_{B})}}{\ud U}(U) =& \frac{- \rd_{U} r}{1-\mu}(U, V \circ \gmm_{R_{0}}(t_{B}))  \quad \hbox{ for } 1 < U \leq U\circ \gmm_{R_{0}}(t_{B}), \quad u_{(t_{B})}(1) = 1. \label{eq:u-tB-cond}
\end{align}
Similarly, we construct the coordinate $v_{(t_{B})} = v_{(t_{B})} (V)$ by normalizing $(1-\mu)^{-1} \rd_{v_{(t_{B})}} r = 1$ along the final outgoing curve $C_{(t_{B}) out}^{f} = \set{(U, V) : U = U\circ \gmm_{R_{0}}(t_{B}), \ 1\leq V \leq V\circ \gmm_{R_{0}}(t_{B})}$:
\begin{align}
\frac{\ud v_{(t_{B})}}{\ud V}(V) =& \frac{\rd_{V} r}{1-\mu}(U \circ \gmm_{R_{0}}(t_{B}), V)  \quad \hbox{ for } 1 < V \leq V\circ \gmm_{R_{0}}(t_{B}), \quad v_{(t_{B})}(1) = 1. \label{eq:v-tB-cond}
\end{align}

Let $(u_{(t_{B})} \circ \gmm_{R_{0}}, v_{(t_{B})} \circ \gmm_{R_{0}})(t)$ denote the parametrization of the curve $\gmm_{R_{0}}(t)$ in the $(u_{(t_{B})}, v_{(t_{B})})$ coordinates. Let $(u_{(t_{B}) f}, v_{(t_{B}) f})$ be the $(u_{(t_{B})}, v_{(t_{B})})$ coordinate of the final point $\gmm_{R_{0}}(t_{B})$, i.e.,
\begin{equation*}
u_{(t_{B}) f} = u_{(t_{B})} \circ \gmm_{R_{0}}(t_{B}), \quad
v_{(t_{B}) f} = v_{(t_{B})} \circ \gmm_{R_{0}}(t_{B}).
\end{equation*}
If $R_{0}$ is large enough (say greater than $2 \sup_{C_{out}} \varpi$, which will indeed be the case), then the point $(U \circ \gmm_{R_{0}} , V \circ \gmm_{R_{0}})(t_{B})$ is well away from the event horizon, and therefore we have $u_{(t_{B}) f}, v_{(t_{B}) f} < \infty$. 

The bootstrap domain $\calD_{(t_{B})}$ is the following box in the $(u_{(t_{B})}, v_{(t_{B})})$-coordinates:
\begin{equation}
	\calD_{(t_{B})} = \set{(u_{(t_{B})}, v_{(t_{B})}) : 1 \leq u_{(t_{B})} \leq u_{(t_{B}) f}, \ 1 \leq v_{(t_{B})} \leq v_{(t_{B}) f} }.
\end{equation}
\begin{figure}[h]
\begin{center}
\def\svgwidth{250px}
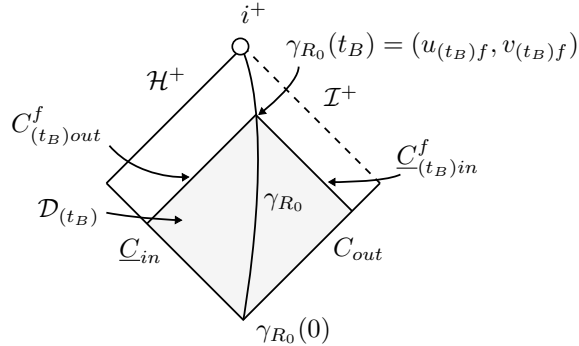 
\caption{Bootstrap domain $\calD_{(t_{B})}$. On $C_{(t_{B}) out}^{f}$ and $\uC_{(t_{B}) in}^{f}$, we normalize $\kpp = 1$ and $-\gmm = 1$, respectively.} \label{fig:btstrp-dom}
\end{center}
\end{figure}

In what follows, we compare the solution $(g, \phi, F)$ in the coordinates $(u_{(t_{B})}, v_{(t_{B})})$  on $\calD_{(t_{B})}$ with the background solution $(\gbg, \phibg, \Fbg)$ in the coordinates $(\ubg, \vbg)$ on $\calD_{(t_{B})}$ (viewed here as a box in the $(\ubg, \vbg)$-coordinates).

Note that $(\calD_{(t_{B})}, u_{(t_{B})}, v_{(t_{B})})$ has been chosen to (at least formally) approach $(\calD(1, \infty), u, v)$ as $t_{B} \to \infty$. We also remark that the normalization conditions \eqref{eq:coord-u-bg}, \eqref{eq:coord-v-bg} for $(\ubg, \vbg)$ are approximately satisfied, in the sense that $\kppbg \restriction_{C_{(t_{B}) out}^{f}}$ and $-\gmmbg \restriction_{\uC_{(t_{B}) in}^{f}}$ approach $1$ as $t_{B} \to \infty$; see Proposition~\ref{prop:dlt-tB-decay} below.

\subsubsection*{Some conventions and definitions}
Before we state the bootstrap assumptions, we fix some conventions (in effect for until Section~\ref{subsec:weak-st-pf}) to make the exposition simpler.
\begin{itemize}
\item To simplify the notation, we omit the $t_{B}$ dependence of $(u_{(t_{B})}$, $v_{(t_{B})})$ and write $(u, v) = (u_{(t_{B})}, v_{(t_{B})})$. Similarly, we omit the bar in $(\ubg, \vbg)$ and write $(u, v) = (\ubg, \vbg)$. Such conventions make sense, since we will consider the difference between $(r, \phi, \varpi)(u_{(t_{B})}, v_{(t_{B})})$ and $(\rbg, \phibg, \varpibg)(\ubg, \vbg)$ at points where $(u_{(t_{B})}, v_{(t_{B})})$ and $(\ubg, \vbg)$ have the same numerical values, which we naturally denote by $(u, v)$. 
\item We use the notation
\begin{gather*}
	\dvr = \rd_{v_{(t_{B})}} r, \quad \dur = \rd_{u_{(t_{B})}} r, \quad
	\kpp = \frac{\dvr}{1-\mu}, \quad \gmm = \frac{\dur}{1-\mu},
\end{gather*}
to denote quantities in the $(u, v) = (u_{(t_{B})}, v_{(t_{B})})$ coordinates. 
\item We use tilde to refer to differences; for instance, we write
\begin{gather*}
	\rdf (u, v)= (r - \rbg)(u, v), \quad
	\phidf (u, v)= (\phi - \phibg)(u, v), \quad
	\varpidf (u, v) = (\varpi - \varpibg)(u, v), \\
	\dvrdf (u, v)= (\dvr - \dvrbg)(u, v), \quad 
	\durdf (u, v) = (\dur - \durbg)(u, v).
\end{gather*}
\end{itemize}

We introduce the initial difference size $\eps_{(t_{B})}^{2}$ in the above coordinates as follows:
\begin{equation} \label{eq:eps-tB}
\begin{aligned}
	\eps_{(t_{B})}^{2} 
= 	& \int_{C_{out} \cap \calD_{(t_{B})}} \kpp^{-1} (\rd_{v} \phidf)^{2} r^{2} \, \ud v 
	+ \int_{C_{out} \cap \calD_{(t_{B})}} \kpp^{-1} \left( \rd_{v} (r \phidf) \right)^{2} \log(1+r) \, \ud v  \\
	& + \int_{\uC_{in} \cap \calD_{(t_{B})}} \frac{1}{(-\dur)} (\rd_{u} \phidf)^{2} r^{2} \, \ud u  
	+ r \abs{\phidf}^{2} \restriction_{C_{out} \cap \uC_{in}}
	+ \abs{\varpidf}^{2} \restriction_{C_{out} \cap \uC_{in}}
	+ \Abs{\rdf}^{2} \restriction_{C_{out} \cap \uC_{in}} 
	+ \abs{\e - \ebg}^{2} .
\end{aligned}
\end{equation}
Observe that (formally at this point) $\eps_{(t_{B})}^{2} \to \eps^{2}$ as $t_{B} \to \infty$.

\subsubsection*{Bootstrap assumptions}
In order to state our bootstrap assumptions, we need to define another (small) quantity $\dlt_{(t_B)}$, which is used to capture the error coming from the fact that the future-normalized gauge for the background solution is fixed at the future event horizon and future null infinity, instead of the future boundary of $\calD_{(t_B)}$. We will show later in Proposition~\ref{prop:dlt-tB-decay} that indeed $\dlt_{(t_{B})} \to 0$ as $t_{(B)}\to \infty$. 

Let
\begin{equation} \label{eq:dlt-tB}
\begin{aligned}
\dlt_{(t_{B})}^{2}
= & \dlt^{2} (1 + t_{B})^{-1} + \sup_{C_{(t_{B}) out}^{f}} \abs{\log \kppbg}
 + \int_{C_{(t_{B}) out}^{f}} \abs{\log \kppbg} \, \ud v \\
& + \sup_{\uC_{(t_{B}) in}^{f}} \rbg \abs{\log (-\gmmbg)} 
+ \int_{\uC_{(t_{B}) in}^{f}} \abs{\log (-\gmmbg)} \, \ud u.
\end{aligned}
\end{equation}

On $\calD_{(t_{B})}$, our bootstrap assumptions are as follows, where $A$ is a bootstrap constant:
\begin{itemize}
\item {\bf Bootstrap assumption on $\ep_{(t_B)}$}:
\begin{align}
\ep_{(t_B)}\leq & 10 A \left(\ep_0 + \de_{(t_B)}^4 \right). \label{eq:en-btstrp-eps} 
\end{align}

\item {\bf Bootstrap assumptions on pointwise difference bounds}:
\begin{align} 
	\abs{\varpidf} \leq& 10 A \left( \eps_{(t_{B})} + \dlt_{(t_{B})}^{2} \right),  \label{eq:en-btstrp-m} \\
	\abs{\log \kpp - \log \kppbg} \leq& 10 A \left( \eps_{(t_{B})} + \dlt_{(t_{B})}^{2} \right), \label{eq:en-btstrp-kpp}\\
	\abs{\log (-\gmm) - \log (-\gmmbg)} \leq& 10 \frac{A}{r} \left( \eps_{(t_{B})} + \dlt_{(t_{B})}^{2} \right)  \quad \hbox{ for } r \geq 20. \label{eq:en-btstrp-gmm}
\end{align}

\item {\bf Bootstrap assumptions on integrated difference bounds}:
\begin{align} 
	\int_{1}^{v_{(t_{B}) f}} \sup_{\uC_{v'} \cap \calD_{(t_{B})}} \Abs{\log \kpp - \log \kppbg}  \, \frac{\ud v'}{(v' + \Lmb)^{\eta_{0}}}
	\leq & 10 A \left( \eps_{(t_{B})} + \dlt_{(t_{B})}^{2} \right), \label{eq:en-btstrp-kpp-int}\\
	\int_{0}^{t_{B}} \left( \Abs{\log \kpp - \log \kppbg} + \Abs{\log (-\gmm) - \log (-\gmmbg)} \right)(u \circ \gmm_{R_{0}}, v \circ \gmm_{R_{0}}) \, \ud t
	\leq & 10 A \left( \eps_{(t_{B})} + \dlt_{(t_{B})}^{2} \right),	\label{eq:en-btstrp-int} \\
	\int_{0}^{t_{B, 20}} \left( \Abs{\log \kpp - \log \kppbg} + \Abs{\log (-\gmm) - \log (-\gmmbg)} \right)(u \circ \gmm_{20}, v \circ \gmm_{20}) \, \ud t
	\leq & 10 A \left( \eps_{(t_{B})} + \dlt_{(t_{B})}^{2} \right).	\label{eq:en-btstrp-int-20} 
\end{align}
Here, we use the parametrization $t \mapsto \gmm_{20}(t)$ of $\gmm_{20} = \set{r = 20}$ such that $\dot{\gmm}_{20}(t) = T$ and $\gmm_{20}(0) \in \uC_{in} \cup C_{out}$. We use the notation $t_{B, 20}$ for the $t$-value when $\gmm_{20}$ exits $\calD_{(t_{B})}$.

\item {\bf Auxiliary bootstrap assumption controlling\footnote{Recall that $(u_{(t_B)f}, v_{(t_B)f})$ is defined implicitly using the perturbed solution by $r(u_{(t_B)f}, v_{(t_B)f})=1$. This bootstrap assumption, via the bounds in Section~\ref{sec:bg}, gives an estimate on the $u_{(t_B)f}$ and $v_{(t_B)f}$} $(u_{(t_B)f}, v_{(t_B)f})$}:
\begin{equation}\label{eq:en-btstrp-R-Lmb}
|\rbg(u_{(t_B)f}, v_{(t_B)f}) - \Lambda |\leq 1.
\end{equation}
\end{itemize}

\begin{remark} \label{rem:kpp-anomaly}
Observe the \emph{anomalous loss of $r$} in the pointwise bound \eqref{eq:en-btstrp-kpp}, i.e., whereas $\log (-\gmm) - \log (-\gmmbg)$ exhibits an $r^{-1}$ decay in $\set{r \geq 20}$, $\log \kpp - \log \kppbg$ does not. Such an anomaly arises because $\log \kpp - \log \kppbg$ is controlled by integrating from $C_{out}^{f}(t_{B})$ in the direction of increasing $r$. The purpose of the integrated bound \eqref{eq:en-btstrp-kpp} is to recover this loss (up to a small power loss $(v' + \Lmb)^{\eta_{0}}$). This recovery of loss turns out to be crucial in closing the bootstrap argument; see Section~\ref{subsec:en-pf}.
\end{remark}

Our goal until Section \ref{subsec:btstrp-pf} is to establish the following result, which improves all the bootstrap assumptions except for \eqref{eq:en-btstrp-eps}.
\begin{proposition} [Improving bootstrap assumptions on the geometric quantities] \label{prop:en-btstrp}
Assume that the hypotheses of Theorem~\ref{thm:L-st-ch} hold with the normalization condition $\rbg_{\EH}=1$. For $t_{B} > 0$, assume in addition that the bootstrap assumptions \eqref{eq:en-btstrp-eps}--\eqref{eq:en-btstrp-R-Lmb} 
 hold on $\calD_{(t_{B})}$, with some large enough constant $A > 0$ (depending on $|\ebg|$, $\eta_{0}$, $\omg$). Then for sufficiently small $\dlt_{0} > 0$ (depending on $A$, $\abs{\ebg}$, $\eta_{0}$, $\omg$) and $\eps_0 > 0$ (depending on $A$, $\dlt_0$, $\abs{\ebg}$, $\eta_{0}$, $\omg$, $\Lmb$), the bootstrap assumptions \eqref{eq:en-btstrp-m}--\eqref{eq:en-btstrp-R-Lmb} on $\calD_{(t_{B})}$ improve as follows:
\begin{align} 
	\abs{\varpidf} \leq& A \left( \eps_{(t_{B})} + \dlt_{(t_{B})}^{2} \right),  \label{eq:en-btstrp-m:imp} \\
	\abs{\log \kpp - \log \kppbg} \leq&  A \left( \eps_{(t_{B})} + \dlt_{(t_{B})}^{2} \right), \label{eq:en-btstrp-kpp:imp}\\
	\abs{\log (-\gmm) - \log (-\gmmbg)} \leq&  \frac{A}{r} \left( \eps_{(t_{B})} + \dlt_{(t_{B})}^{2} \right)  \quad \hbox{ for } r \geq 20, \label{eq:en-btstrp-gmm:imp}	\\
	\int_{1}^{v_{(t_{B}) f}} \sup_{\uC_{v'} \cap \calD_{(t_{B})}} \Abs{\log \kpp - \log \kppbg} \, \frac{\ud v'}{(v' + \Lmb)^{\eta_{0}}}
	\leq &  A \left( \eps_{(t_{B})} + \dlt_{(t_{B})}^{2} \right), \label{eq:en-btstrp-kpp-int:imp}\\
	\int_{0}^{t_{B}} \left( \Abs{\log \kpp - \log \kppbg} + \Abs{\log (-\gmm) - \log (-\gmmbg)} \right)(u \circ \gmm_{R_{0}}, v \circ \gmm_{R_{0}}) \, \ud t
	\leq &  A \left( \eps_{(t_{B})} + \dlt_{(t_{B})}^{2} \right),	\label{eq:en-btstrp-int:imp} \\
	\int_{0}^{t_{B, 20}} \left( \Abs{\log \kpp - \log \kppbg} + \Abs{\log (-\gmm) - \log (-\gmmbg)} \right)(u \circ \gmm_{20}, v \circ \gmm_{20}) \, \ud t
	\leq & A \left( \eps_{(t_{B})} + \dlt_{(t_{B})}^{2} \right),	\label{eq:en-btstrp-int-20:imp} \\
	|\rbg(u_{(t_B)f}, v_{(t_B)f}) - \Lambda |\leq & \f 12. \label{eq:en-btstrp-R-Lmb:imp}
\end{align}
Moreover, the following geometric difference bounds hold on $\calD_{(t_{B})}$ for some $C>0$ depending on $|\ebg|$, $\eta_0$, $\omg$:
\begin{align}
\Abs{\rdf} \leq & C A \max\set{\log \Lmb, \log r} \left( \eps_{(t_{B})} + \dlt_{(t_{B})}^{2}\right), \label{eq:en-btstrp:rdf} \\
	\Abs{\dvrdf} \leq & C A (\log \Lmb) \left( \eps_{(t_{B})} + \dlt_{(t_{B})}^{2} \right), \label{eq:en-btstrp:dvrdf} \\
	\Abs{\durdf} \leq & \left\{
	\begin{array}{ll}
	\displaystyle{C A (\log \Lmb) e^{- \frac{1}{2} c_{(\dur)}(u - v - C_{\gmm_{20}})}} (\eps_{(t_{B})} + \dlt_{(t_{B})}^{2})  & \hbox{ for } r \leq 20, \\
	\displaystyle{C A (\log \Lmb) r^{-1} (\eps_{(t_{B})} + \dlt_{(t_{B})}^{2})}  & \hbox{ for } r \geq 20, 
	\end{array} \right.	\label{eq:en-btstrp:durdf} \\
	\Abs{\log (-\dur) - \log (-\durbg)} 
\leq & 	C A \log \Lmb (1+\abs{u - v - C_{\gmm_{20}}}) \left( \eps_{(t_{B})} + \dlt_{(t_{B})}^{2}\right) \quad \hbox{ for } r \leq 20.	\label{eq:en-btstrp:logdurdf} 
\end{align}
where $c_{(\dur)} > 0$ is specified in Lemma~\ref{lem:df-small-r} below, and $C_{\gmm_{20}}$ is defined by the relation $r(C_{\gmm_{20}} +1, 1) = 20$.
\end{proposition}

\subsubsection*{Estimates for coordinate transformation in the bootstrap argument}
In carrying out the bootstrap argument, an important point is to understand how $\eps^{2}_{(t_{B})}$ is related with $\eps_{0}^{2}$. This is precisely the content of the following proposition, which can be viewed as an estimate for the coordinate transformation between the initial-data-normalized coordinates and the future-normalized coordinates. Note that the following proposition also improves the remaining bootstrap assumption \eqref{eq:en-btstrp-eps}.
\begin{proposition} [Coordinate transformation in the bootstrap argument] \label{prop:g-btstrp}
Assume that the hypotheses of Proposition~\ref{prop:en-btstrp} hold. Then for sufficiently small $\dlt_{0} > 0$ (depending on $A$, $\abs{\ebg}$, $\eta_0$ and $\omg$) and $\eps_{0} > 0$ (depending on $A$, $\dlt_0$, $|\ebg|$, $\eta_0$, $\omg$ and $\Lmb$), we have
\begin{equation} \label{eq:g-btstrp}
	\eps_{(t_{B})}^{2} \leq C \eps_{0}^{2} + C A^{2} \dlt_{0}^{2} \left(\eps_{(t_{B})} + \dlt_{(t_{B})}^{2} \right)^{2}.
\end{equation}
Furthermore, we have
\begin{equation} \label{eq:g-btstrp-dvrphi}
	\sup_{C_{out} \cap \calD_{(t_{B})}} r^{\omg} \Abs{\dvr^{-1} \rd_{v} (r \phi - \rbg \phibg)} \leq C \eps_{0} + C A \dlt_{0} (\eps_{(t_{B})} + \dlt_{(t_{B})}^{2}).
\end{equation}
\end{proposition}

This proposition is proved in Section~\ref{subsec:gauge}, after Proposition~\ref{prop:en-btstrp} is proven.

Beginning in the next subsection until Section~\ref{subsec:gauge}, we will carry out the main part of the bootstrap argument and prove Propositions~\ref{prop:en-btstrp} and \ref{prop:g-btstrp} (see the end of Section~\ref{subsec:weak-st} for an outline of the argument). \textbf{Before we proceed, let us emphasize a few points about the constants involved in the argument:}
\begin{itemize}
\item We will continue to use $C$ to denote an arbitrary constant depending only on $\omg$, $\eta_0$ and $|\ebg|$.
\item The main small parameters that we can choose are $\ep_0$ and $\de_0$. Moreover, $\ep_0$ can be chosen to be much smaller than than $\de_0$.
\item Since the choice of $\de_0$ can depend on $A$, we will assume that $A \de_0$ is much smaller than any constant $C$ by taking $\de_0$ sufficiently small.
\item Importantly, by choosing $\ep_0$ and $\de_0$ small, we can also make the parameters $\de$, $\de_{(t_B)}$ and $\ep_{(t_B)}$ small as follows:
\begin{itemize}
\item $\de$ can be made small since $\de\leq \Lmb^{-100\eta_0} \de_0$ and $\Lmb>100$ according to the assumptions in Theorem~\ref{thm:L-st-ch}.
\item $\dlt_{(t_B)}$ can be made small as it is controlled by $\de$, as will be shown in Proposition~\ref{prop:dlt-tB-bnd} (immediately below in the beginning of Section~\ref{subsec:bg-geom}).
\item $\ep_{(t_B)}$ can be made small because of the bootstrap assumption \eqref{eq:en-btstrp-eps} and the preceding bounds.
\end{itemize}
\end{itemize}

\subsection{Geometric bounds for the background solution} \label{subsec:bg-geom}
In this subsection, we provide estimates for various geometric quantities associated to an $(\omg, \dlt, \Lmb)$-admissible background solution. We continue to work in the bootstrap setting as described in Section~\ref{subsec:btstrp}. Namely, we have a fixed $(\omg, \dlt, \Lmb)$-admissible background solution and we will be estimating the geometric quantities for this background solution in the domain $\calD_{(t_B)}$, which is defined by a perturbed solution. We will also use the conventions for various constants discussed at the end of the last subsection.

We begin with a rough bound for $\de_{(t_B)}$. As mentioned at the end of the last subsection, this bound allows us to take $\de_{(t_B)}$ to be small by taking $\de$ (or $\de_0$) to be small.
\begin{proposition}\label{prop:dlt-tB-bnd}
The following estimate holds:
\begin{equation*}
\de_{(t_B)}\leq C\de.
\end{equation*}
\end{proposition}
\begin{proof}
Since $t_{(B)}\geq 1$, by definition, it suffices to prove
\begin{align} 
	\sup_{C_{(t_{B}) out}^{f}} \abs{\log \kppbg} + \int_{C_{(t_{B}) out}^{f}} \abs{\log \kppbg} \, \ud v 
\leq & C \dlt^{2} ,		\label{eq:dlt-tB-bnd:pf-1}\\
 	\sup_{\uC_{(t_{B}) in}^{f}} \rbg \abs{\log (-\gmmbg)} + \int_{\uC_{(t_{B}) in}^{f}} \abs{\log (-\gmmbg)} \, \ud u
\leq & C \dlt^{2} .		\label{eq:dlt-tB-bnd:pf-2} 
\end{align}
The bound \eqref{eq:dlt-tB-bnd:pf-1} follows from \eqref{eq:bg-uv:kpp-decay}. Notice that for the integral we can divide it into $\rbg\leq 10$ and $\rbg\geq 10$, and the integral over $\rbg\leq 10$ is clearly bounded by $C \dlt^2$. For $\rbg\geq 10$, by \eqref{eq:en-btstrp-R-Lmb}, the $v$-length of $C^{f}_{(t_{B}) out} \cap \set{\rbg \geq 10}$ is bounded by $C \Lmb$; it is compensated by either the $v$ decay or the factor $\Lmb^{-1}$ in \eqref{eq:bg-uv:kpp-decay}. 

By \eqref{eq:en-btstrp-R-Lmb}, $\uC_{(t_{B}) in}^{f} \subseteq \{\rbg \geq 10\}$. Hence, the claim \eqref{eq:dlt-tB-bnd:pf-2} follows from \eqref{eq:bg-uv:gmm-decay}. \qedhere
\end{proof}

The following is the main result for the geometric quantities associated to an $(\omg, \dlt, \Lmb)$-admissible background solution:
\begin{proposition} \label{prop:bg-geom}
Let $(\Omgbg, \rbg, \phibg, \ebg)$ in $\widetilde{\calD}$ be an $(\omg, \dlt, \Lmb)$-admissible background solution as in the hypothesis of Theorem~\ref{thm:L-st-ch}, and we normalize $\rbg_{\EH} = 1$.
Let $(u, v)$ be a double null coordinate system for this solution characterized by \eqref{eq:coord-u-bg} and \eqref{eq:coord-v-bg}. 
Then the following statements hold.
\begin{enumerate}
\item We have
\begin{align}
	\rbg \geq & 1 - C \dlt \qquad \hbox{ in } \calD, \label{eq:rbg-bnd}   \\
	0 < 1-\mubg \leq & 1 + C \dlt \qquad \hbox{ in } \calD, \label{eq:mubg-bnd} \\
	1 \leq \kppbg \leq & 1 + C \dlt^{2} \qquad  \hbox{ in } \calD, \label{eq:kppbg-bnd} \\
	- 1 \leq \gmmbg \leq & - 1 + C \dlt^{2} \qquad  \hbox{ in } \calD \cap \set{\rbg \geq 10}. \label{eq:gmmbg-bnd}  
\end{align}
Furthermore, there exist $c_{(\durbg)}, c'_{(\durbg)} > 0$ depending on $|\ebg|$ so that
\begin{equation} \label{eq:durbg-bnd} 
	C^{-1} \min \set{e^{- c'_{(\durbg)}(u - v - \overline{C}_{\gmm_{30}})}, 1} < - \durbg \leq C \min \set{e^{- c_{(\durbg)} (u - v - \overline{C}_{\gmm_{30}})}, 1} \qquad  \hbox{ in } \calD
\end{equation}
where $\overline{C}_{\gmm_{30}}$ is defined by $\rbg(1 + \overline{C}_{\gmm_{30}}, 1) = 30$.
\item We have
\begin{align} 
	\abs{\log(-\gmmbg)} \leq & C u^{-2 \omg + 2} \rbg^{-2} \dlt^{2} \qquad \hbox{ in } \calD \cap \set{\rbg \geq 10} ,\label{eq:gmmbg-decay} \\
	\abs{\log \kppbg} \leq & C \left( v^{-2 \omg} + u^{-2 \omg + 1} \Lmb^{-1} \right) \dlt^{2} \qquad \hbox{ in } \calD \cap \set{\rbg \geq 10} .\label{eq:kppbg-decay} 
\end{align}
\end{enumerate}
\end{proposition}
\begin{remark} [Coordinate invariance of bounds on $\phibg$]
Before we begin the proof, we remind the reader that the assumptions \eqref{eq:dlt-adm:phi}--\eqref{eq:dlt-adm-dvdvrphi} are invariant under the coordinate transformation $(U, V) \mapsto (\ubg, \vbg)$, since the vector fields $\frac{1}{\rd_{\ubg} \rbg} \rd_{\ubg}$ and $\frac{1-\mubg}{\rd_{\vbg} \rbg} \rd_{\vbg}$ are coordinate invariant. It follows that \eqref{eq:dlt-adm:phi}--\eqref{eq:dlt-adm-dvdvrphi} hold with $(U, V)$ replaced by $(\ubg, \vbg)$; the latter coordinates coincide with $(u, v)$ in the notation of Proposition~\ref{prop:bg-geom}.
\end{remark}

\begin{proof}
By Definition~\ref{def:dlt-adm}, Proposition~\ref{prop:bg-uv} for $\de_0$ sufficiently small. The bound \eqref{eq:rbg-bnd} follows from the bound $\rbg \geq 1-\dlt$ on $\EH \cap \widetilde{\calD}$ and the monotonicity properties of $\rbg$ in $\calD$. The upper bound in \eqref{eq:mubg-bnd} is proved by combining \eqref{eq:dlt-adm:m} with \eqref{eq:rbg-bnd}, whereas the lower bound follows from \eqref{eq:bg-uv:mu-large-r} and \eqref{eq:bg-uv:mu-small-r}. The upper bounds in \eqref{eq:kppbg-bnd} and \eqref{eq:gmmbg-bnd} trivially follow from \eqref{eq:bg-uv:kpp-decay} and \eqref{eq:bg-uv:gmm-decay} in Proposition~\ref{prop:bg-uv}, respectively, whereas the lower bounds are consequences of the future normalization $\overline{\kappa}=1$ at $\EH$ and $\overline{\gamma}=-1$ at $\NI$ together with the monotonicity of $\overline{\kappa}$ and $\overline{\gamma}$ that follows from \eqref{eq:EMSF-ray}. The bound \eqref{eq:durbg-bnd} is a result of combining \eqref{eq:bg-uv:dur-large-r} and \eqref{eq:bg-uv:dur-small-r}. Finally, \eqref{eq:gmmbg-decay}--\eqref{eq:kppbg-decay}  follow from \eqref{eq:bg-uv:gmm-decay}--\eqref{eq:bg-uv:kpp-decay}. \qedhere
\end{proof}

\subsection{Preliminary geometric difference bounds} \label{subsec:geom}
In this subsection, we estimate the difference of geometric quantities using the bootstrap assumptions. In particular, we obtain estimates for $\dvrdf$, $\durdf$ and $\rdf$ that are claimed in Proposition~\ref{prop:en-btstrp}.

From this subsection until Section~\ref{subsec:gauge}, {\bf we assume that the hypotheses of Proposition~\ref{prop:en-btstrp} hold.} In particular, we assume throughout that $\rbg_{\EH}=1$. We will also use the conventions about the constants $C$, $\de_0$, $\eps_0$, $\de$, $\de_{(t_B)}$ and $\eps_{(t_B)}$ as described in the end of Section~\ref{subsec:btstrp}. In particular, in what follows {\bf we freely adjust $\eps_{(t_{B})}$ and $\dlt_{0}$ to be small enough for our purposes}.

Since $t_{B}$ is fixed throughout the argument, \textbf{we drop the $t_{B}$-dependence of $\eps_{(t_{B})}$ and simply write $\eps = \eps_{(t_{B})}$}; this convention will be in effect until Section~\ref{subsec:en-pf}.

We start with a collection of simple bounds concerning the geometry of the perturbed solution $(g, \phi, F)$.
\begin{lemma} \label{lem:geom-prelim}
The following statements hold. 
\begin{enumerate}
\item In the whole bootstrap domain $\calD_{(t_{B})}$, we have
\begin{align} 
	\frac{1}{4} < \varpi \leq 1 & \quad \hbox{ in } \calD_{(t_{B})}, \label{eq:m-bnd}  \\
		\abs{\e} \leq 1 &, 		\label{eq:e-bnd}  \\
	1 \leq \kpp \leq 2 & \quad \hbox{ in } \calD_{(t_{B})}. \label{eq:kpp-bnd}  
\end{align}
\item Restricted to the region away from the event horizon, we have 
\begin{align} 
	\frac{1}{2} \leq 1-\mu \leq 1 & \quad \hbox{ in } \calD_{(t_{B})} \cap \set{r \geq 10}, \label{eq:mu-large-r}  \\
	 \frac{1}{2} \leq - \gmm \leq 1 & \quad \hbox{ in } \calD_{(t_{B})} \cap \set{r \geq 20}, \label{eq:gmm-large-r}  \\
	 \frac{1}{4} \leq - \dur \leq 1 & \quad \hbox{ in } \calD_{(t_{B})} \cap \set{r \geq 20}. \label{eq:dur-large-r}  
\end{align}
\item Finally, we have
\begin{align}  
	\abs{R_{0} - \Lmb} \leq & \eps.  \label{eq:R0-Lmb}
\end{align}
\end{enumerate}
\end{lemma}
\begin{proof}
Throughout the proof, we will take $\eps_0$, $\de_0$ small as necessary (and thus also $\ep$ and $\dlt_{(t_B)}$; see end of Section~\ref{subsec:btstrp}). For \eqref{eq:m-bnd}, observe that $\abs{\varpi - \varpibg}$ can be made arbitrarily small by \eqref{eq:en-btstrp-m}. Since $\rbg_{\EH}=1$ by our normalization, it follows from \eqref{r.varpi.poly} that $\f 12<\overline{\varpi}_{\EH}<1$ (with strict inequalities\footnote{Precisely how far away is $\overline{\varpi}_{\EH}$ from $\f 12$ and $1$ depends on $|\ebg|$.}). Hence, by Corollary~\ref{varpi.same.limit} and \eqref{eq:dlt-adm:m}, \eqref{eq:m-bnd} follows. For \eqref{eq:e-bnd}, we simply use \eqref{eq:subextremality} and $\abs{\e - \ebg} \leq \eps$. Next, the upper bound in \eqref{eq:kpp-bnd} follows from \eqref{eq:kppbg-bnd} and \eqref{eq:en-btstrp-kpp} by writing $\kpp = e^{\log (\kpp/\kppbg)} \kppbg$; whereas the lower bound follows simply from $\kappa=1$ on $\{u=u_{(t_B)f}\}$ and the monotonicity of $\kappa$. 
For \eqref{eq:mu-large-r}, recall that $1 - \mu = 1 - \frac{2 \varpi}{r} + \frac{\e^{2}}{r^{2}}$; therefore, the desired bound follows from \eqref{eq:m-bnd} and \eqref{eq:e-bnd}.
As in the case of \eqref{eq:kpp-bnd}, \eqref{eq:gmm-large-r} follows from \eqref{eq:gmmbg-bnd} and \eqref{eq:en-btstrp-gmm}, as well as monotonicity of $(-\gamma)$. Then the bound \eqref{eq:dur-large-r} follows from the preceding bounds, since $\dur = (1-\mu) \gmm$.
Finally, \eqref{eq:R0-Lmb} follows from the definition \eqref{eq:eps-tB}, since $\rdf (1, 1) = r(1, 1) - \rbg(1, 1) = R_{0} - \Lmb$. \qedhere
\end{proof}

Also, the following elementary analytic lemma will be useful.
\begin{lemma} \label{lem:r-alp-df}
Given any $\alp \in \bbR$, there exists $C_{\alp} > 0$ such that
\begin{equation} \label{eq:r-alp-df}
	\Abs{r^{\alp} - (\rbg)^{\alp} - \alp r^{\alp-1} \rdf} \leq C_{\alp} r^{\alp-2} \abs{\rdf}^{2} \quad \hbox{ if } \abs{\rdf} \leq \frac{1}{2} r.
\end{equation}
\end{lemma}
\begin{proof}
The bound \eqref{eq:r-alp-df} simply follows by writing
\begin{equation*}
r^{\alp} - (\rbg)^{\alp} = r^{\alp} \left( 1 - \left( 1 - \tfrac{\rdf}{r} \right)^{\alp} \right),
\end{equation*}
and performing Taylor expansion in $\rdf / r$. \qedhere
\end{proof}

The integral assumption \eqref{eq:en-btstrp-int} leads to the following bounds on the curve $\gmm_{R_{0}}$.
\begin{lemma} \label{lem:df-const-r}
On the curve $\gmm_{R_{0}} \cap \calD_{(t_{B})}$, we have
\begin{align} 
	\sup_{\gmm_{R_{0}} \cap \calD_{(t_{B})}} \abs{\rdf} 
\leq &	C A \left( \eps + \dlt_{(t_{B})}^{2} \right) . \label{eq:df-const-r:rdf} 
\end{align}
\end{lemma}
\begin{proof}
Recall the vector field $T$ from \eqref{eq:vf-T-0}, which obeys $T r = 0$ so that $T \, \rdf = - T \rbg$. Since $T \restriction_{\gmm_{R_{0}}(t)} = \dot{\gmm}_{R_{0}}(t)$, we have
\begin{align*}
	\frac{\ud}{\ud t} \rdf(u \circ \gmm_{R_{0}}, v \circ \gmm_{R_{0}}) 
=& 	(T \, \rdf) (u \circ \gmm_{R_{0}}, v \circ \gmm_{R_{0}})  \\
= & (1-\mu) \left( \frac{\gmmbg}{\gmm} - \frac{\kppbg}{\kpp} \right) (u \circ \gmm_{R_{0}}, v \circ \gmm_{R_{0}}).
\end{align*}
Integrating in $t$, using \eqref{eq:mu-large-r} and recalling the bootstrap assumptions \eqref{eq:en-btstrp-kpp}, \eqref{eq:en-btstrp-gmm}, \eqref{eq:en-btstrp-int}, and using $|e^{\vartheta}-1|\leq |\vartheta| e^{|\vartheta|}$, it follows that
\begin{align*}
	\sup_{\gmm_{R_{0}} \cap \calD_{(t_{B})}} \abs{\rdf}
	\leq & \abs{\rdf} (1, 1) + \int_{0}^{t_{B}} \Abs{\frac{\gmmbg}{\gmm} - \frac{\kppbg}{\kpp}}  
						(u \circ \gmm_{R_{0}}, v \circ \gmm_{R_{0}})(t) \, \ud t \\
	\leq & \eps + C A (\eps + \dlt_{(t_{B})}^{2}) e^{C A (\eps + \dlt_{(t_{B})}^{2})},		
\end{align*}
which implies \eqref{eq:df-const-r:rdf}.  \qedhere
\end{proof}

Using Lemma~\ref{lem:df-const-r} as a starting point, we now turn to the task of deriving bounds for $\rdf$, $\dvrdf$ and $\durdf$ on the whole domain $\calD_{(t_{B})}$. We first consider the region $\set{r \geq 20}$ away from the event horizon.

\begin{lemma} \label{lem:df-large-r}
In the region $\calD_{(t_{B})} \cap \set{r \geq 20}$, the following difference bounds hold:
\begin{align} 
	\abs{\rdf} 
\leq & 	C A \max\set{\log \Lmb, \log r} \left( \eps + \dlt_{(t_{B})}^{2} \right),				\label{eq:df-large-r:rdf} \\
	\abs{\durdf} 
\leq & 	C A (\log \Lmb) r^{-1} \left( \eps +  \dlt_{(t_{B})}^{2}\right)	,			\label{eq:df-large-r:durdf} \\
	\abs{\dvrdf} 
\leq & 	C A (\log \Lmb) \left( \eps + \dlt_{(t_{B})}^{2}\right) .				\label{eq:df-large-r:dvrdf} 
\end{align}
Moreover, we have the splitting $\dvrdf = \dvrdf_{g} + \dvrdf_{b}$, where $\dvrdf_{g}$ obeys an analogous estimate as \eqref{eq:df-large-r:durdf}:
\begin{equation} \label{eq:df-large-r:dvr-g} 
	\abs{\dvrdf_{g}} 
\leq  	C A (\log \Lmb) r^{-1} \left( \eps + \dlt_{(t_{B})}^{2}\right) ,				
\end{equation}
and $\dvrdf_{b}$ obeys an improved integrated estimate:
\begin{equation}\label{eq:df-large-r:dvr-b}
	\int_{1}^{v} \sup_{\uC_{v'} \cap \calD_{(t_{B})}} \abs{\dvrdf_{b}} \, \frac{\ud v'}{(v' + \Lmb)^{\eta_{0}}}
\leq  	C A (\log \Lmb) \left( \eps + \dlt_{(t_{B})}^{2}\right) .				
\end{equation}
Finally, on the curves $\gmm_{20} = \set{(u, v) : r(u, v) = 20}$ and $\gmm_{R_{0}} = \set{(u, v) : r(u, v) = R_{0}}$, we have
\begin{align} 
	\sup_{\gmm_{20} \cap \calD_{(t_{B})}} \abs{u - v - C_{\gmm_{20}}} 
\leq & 	C A (\eps + \dlt_{(t_{B})}^{2}) + C \dlt^{2}, 		\label{eq:df-large-r:v-u} \\
	\sup_{\gmm_{R_{0}} \cap \calD_{(t_{B})}} \abs{u - v}
\leq &	C A (\eps + \dlt_{(t_{B})}^{2}) + C \dlt^{2}, \label{eq:df-large-r:v-u-R0}
\end{align}
where $C_{\gmm_{20}}$ is defined by the relation $r(C_{\gmm_{20}} + 1, 1) = 20$.
\end{lemma}
\begin{proof}
The most delicate bound is \eqref{eq:df-large-r:rdf}, since we cannot assume any bound on $\durdf$ or $\dvrdf$. For its proof, we analyze a nonlinear ODE that $\rdf$ solves, which can be done using only the bootstrap assumptions. As we will see, the rest of the argument is more routine.

\pfstep{Step~1: Proof of \eqref{eq:df-large-r:rdf}}
Recall that $\dur = (1-\mu) \gmm$. Thus, on the one hand, we have
\begin{equation} \label{eq:df-large-r:gmm-dur}
	\abs{\gmm} = \frac{- \rd_{u} r}{1-\mu} \leq 2 (- \rd_{u} r)
\end{equation}
in $\calD_{(t_{B})} \cap \set{r \geq 20}$, thanks to \eqref{eq:mu-large-r}. On the other hand, taking the difference between $(1-\mu) \gmm$ and $(1-\mubg) \gmmbg$, we obtain the ODE
\begin{equation} \label{eq:durdf}
	\rd_{u} \rdf = \durdf = \left( - \frac{2 \varpidf}{r} + \frac{\e^{2} - \ebg^{2}}{r^{2}} \right) \gmm + (1-\mubg) \left( 1- \frac{\gmmbg}{\gmm}\right) \gmm+ \frac{2(\varpibg - \frac{\ebg^{2}}{r})}{r^{2}} \gmm \, \rdf  + \calE[\rd_{u} \rdf],
\end{equation}
where
\begin{equation*}
	\calE[\rd_{u} \rdf] = \left( - 2 \varpibg (r^{-1} - \rbg^{-1} + r^{-2} \rdf) + \ebg^{2} (r^{-2} - \rbg^{-2} + 2 r^{-3} \rdf) \right) \gmm.
\end{equation*}

Given any $(u, v) \in \calD_{(t_{B})} \cap \set{r \geq 20}$, we will control $\rdf$ by integrating the ODE \eqref{eq:durdf} from $(u_{R_{0}}(v), v)$ to $(u, v)$, where $u_{R_{0}}(v)$ denotes the $u$-coordinate such that $r(u_{R_{0}}(v), v) = R_{0}$, so that $(u_{R_{0}}(v), v) \in \gmm_{R_{0}}$. The existence and uniqueness of such a $u_{R_{0}}(v)$ is guaranteed by monotonicity of $r$.  At $(u_{R_{0}}(v), v)$, by \eqref{eq:df-const-r:rdf} we have the initial estimate
\begin{equation*}
\abs{\rdf}(u_{R_{0}}(v), v) \leq C A\left(\eps + \dlt_{(t_{B})}^{2} \right).
\end{equation*}

For the forcing term, which is the main contribution, we have
\begin{equation*}
\Abs{\int_{u_{R_{0}}(v)}^{u} \Abs{\left( - \frac{2 \varpidf}{r} + \frac{\e^{2} - \ebg^{2}}{r^{2}} \right) + (1-\mubg) \left( 1- \frac{\gmmbg}{\gmm}\right)} \Abs{\gmm}(u', v) \, \ud u'}
\leq C A \max \set{\log \Lmb, \log r(u, v)} \left( \eps + \dlt_{(t_{B})}^{2} \right).
\end{equation*}
This estimate can be proved by first using \eqref{eq:df-large-r:gmm-dur} to make a change of variables $u' \mapsto r(u')$, then using \eqref{eq:eps-tB} and the bootstrap assumptions \eqref{eq:en-btstrp-m}, \eqref{eq:en-btstrp-gmm}. 

The linear term in \eqref{eq:durdf} (the third term on the RHS) can be dealt with Gr\"onwall's inequality, since the integral $\int_{u_{R_{0}}(v)}^{u} \frac{2 (\varpibg - \frac{\ebg^{2})}{r}}{r^{2}} \abs{\gmm} \, \ud u'$ is uniformly bounded. 

For the remaining term $\calE [\rd_{u} \rdf]$, we employ a bootstrap argument in $u$ (starting from $u = u_{R_{0}}(v)$), where the bootstrap assumption is $ \sup_{u' \in [u_{R_{0}}(v), u]} r^{-1}\abs{\rdf} \leq \veps$ (where $0 < \veps < 1/2$ is to be specified below).
Using Lemma~\ref{lem:r-alp-df}, we have
\begin{equation}\label{eq:durdf-err}
	\Abs{\calE[\rd_{u} \rdf]} \leq C \left( \frac{\varpibg}{r^{3}} + \frac{\ebg^{2}}{r^{4}} \right) \abs{\gmm} \abs{\rdf}^{2}.
\end{equation}
Then by \eqref{eq:df-large-r:gmm-dur}, we may estimate
\begin{equation*}
	\Abs{\int_{u_{R_{0}(v)}}^{u} \Abs{\calE[\rd_{u} \rdf]}(u', v)  \, \ud u'} \leq C \sup_{u' \in [u_{R_{0}}(v), u]} r^{-1} \abs{\rdf}^{2}(u',v).
\end{equation*}
where we abuse the notation a bit and write $u' \in [u_{R_{0}}(v), u]$ for points $u'$ between $u_{R_{0}}(v)$ and $u$ (despite the fact that $u$ may be smaller than $u_{R_{0}}(v)$).

In sum, we arrive at
\begin{equation*}
	\abs{\rdf}(u, v) \leq C A \max \set{\log \Lmb, \log r(u, v)} \left( \eps + \dlt_{(t_{B})}^{2} \right) + C \sup_{u' \in [u_{R_{0}}(v), u]} r^{-1}\abs{\rdf}^{2}(u',v).
\end{equation*}
Choosing $\veps$ in the bootstrap assumption small enough, we may absorb the last term into the LHS. Then bootstrap assumption can be improved by taking $A \log \Lmb (\eps+ \dlt_{(t_{B})}^{2})$ sufficiently small; moreover, it holds initially thanks to \eqref{eq:df-const-r:rdf}. The desired estimate \eqref{eq:df-large-r:rdf} follows.

\pfstep{Step~2: Proof of \eqref{eq:df-large-r:durdf}}
The bound \eqref{eq:df-large-r:durdf} is obtained using the same equation \eqref{eq:durdf}. We proceed similarly as before, except that we use the pointwise bound \eqref{eq:gmm-large-r} in place of \eqref{eq:df-large-r:gmm-dur}, and insert \eqref{eq:df-large-r:rdf} to estimate $\abs{\rdf}$. We remark that the factor $\log \Lmb$ is inherited from \eqref{eq:df-large-r:rdf} through the linear term in $\rdf$. 

\pfstep{Step~3: Proofs of \eqref{eq:df-large-r:dvrdf}--\eqref{eq:df-large-r:dvr-b}}
Next, we turn to the bounds concerning $\dvrdf$. As in the case of $\durdf$, $\rd_{v} \rdf = \dvrdf$ obeys the following equation:
\begin{equation} \label{eq:dvrdf-pre}
	\dvrdf = \rd_{v} \rdf = \left( - \frac{2 \varpidf}{r} + \frac{\e^{2} - \ebg^{2}}{r^{2}} \right) \kpp + (1-\mubg) \left( 1- \frac{\kppbg}{\kpp}\right) \kpp+ \frac{2(\varpibg - \frac{\ebg^{2}}{r})}{r^{2}} \kpp \, \rdf  + \calE[\rd_{v} \rdf],
\end{equation}
where
\begin{equation}\label{eq:dvrdf-pre-err}
	\Abs{\calE[\rd_{v} \rdf]} \leq C \left( \frac{\varpibg}{r^{3}} + \frac{\ebg^{2}}{r^{4}} \right) \abs{\kpp} \abs{\rdf}^{2}.
\end{equation}
The bound \eqref{eq:df-large-r:dvrdf} follows by a similar argument as before, where we use \eqref{eq:en-btstrp-kpp} and \eqref{eq:kpp-bnd} instead of \eqref{eq:en-btstrp-gmm} and \eqref{eq:gmm-large-r}, respectively. 

Observe that, in contrast to \eqref{eq:df-large-r:durdf}, we do \emph{not} gain the factor $r^{-1}$, since $\left(1 - \frac{\kppbg}{\kpp}\right)$ is only \emph{bounded} according to our bootstrap assumption \eqref{eq:en-btstrp-kpp}. We separate out the bad term by defining
\begin{equation*}
	\dvrdf_{b} = (1-\mubg) \left(1 - \frac{\kppbg}{\kpp}\right) \kpp
\end{equation*}
and let $\dvrdf_{g} = \dvrdf - \dvrdf_{b}$. Then by the same proof as before, \eqref{eq:df-large-r:dvr-g} follows. On the other hand, \eqref{eq:df-large-r:dvr-b} follows from the simple pointwise inequality (note $|e^\vartheta-1|\leq |\vartheta|e^{|\vartheta|}$)
\begin{equation*}
	\Abs{(1-\mubg) \left(1 - \frac{\kppbg}{\kpp}\right) \kpp}
	\leq \abs{1-\mubg} \abs{\log \kppbg - \log \kpp} e^{\abs{\log \kppbg - \log \kpp}} \abs{\kpp},
\end{equation*}
and the bounds \eqref{eq:en-btstrp-kpp}, \eqref{eq:en-btstrp-kpp-int}, \eqref{eq:mubg-bnd} and \eqref{eq:kpp-bnd}.

\pfstep{Step~4: Proof of \eqref{eq:df-large-r:v-u} and \eqref{eq:df-large-r:v-u-R0}}
These bounds are proved in a similar manner as \eqref{eq:df-const-r:rdf}. In what follows, we only handle the case \eqref{eq:df-large-r:v-u} and leave the similar proof of \eqref{eq:df-large-r:v-u-R0} to the reader. 

Recall the vector field $T$ from \eqref{eq:vf-T-0}. As before, we denote by $\gmm_{20}(t)$ the parametrization of $\gmm_{20}$ so that $\dot{\gmm_{20}}(t) = T(\gmm_{20}(t))$ and $\gmm_{20}(0) \in \uC_{in}$. Note that
\begin{equation*}
	- T(u - v - C_{\gmm_{20}}) = \frac{1}{\gmm} + \frac{1}{\kpp}
	= \left(\frac{\gmmbg}{\gmm} -1 \right) \frac{1}{\gmmbg} + \left(\frac{\kppbg}{\kpp} -1 \right) \frac{1}{\kppbg} + \left(1 + \frac{1}{\gmmbg}\right) + \left(\frac{1}{\kppbg} - 1\right).
\end{equation*}
By Proposition~\ref{prop:bg-geom}, \eqref{eq:en-btstrp-kpp}, \eqref{eq:en-btstrp-gmm} and \eqref{eq:en-btstrp-int-20}, it follows that
\begin{align*}
	\sup_{\gmm_{20} \cap \calD_{(t_{B})}} \abs{u - v - C_{\gmm_{20}}}
	\leq & \int_{0}^{t_{B, 20}} \abs{T (u - v - C_{\gmm_{20}})} (u \circ \gmm_{20}, v \circ \gmm_{20})(t) \, \ud t \\ 
	\leq & C A \left( \eps + \dlt_{(t_{B})}^{2} \right) e^{C A (\eps + \dlt_{(t_{B})}^{2})} + C e^{C \dlt^{2}} \int_{0}^{t_{B, 20}} \abs{\log (-\gmmbg)} + \abs{\log \kppbg} \, \ud t.
\end{align*}
Since $T(u) = - \gmm^{-1}$ and $T(v) = \kpp^{-1}$ are bounded from above and below on $\gmm_{20}$, it follow from \eqref{eq:gmmbg-decay} and \eqref{eq:kppbg-decay} that
\begin{equation*}
\int_{0}^{t_{B, 20}} \abs{\log (-\gmmbg)} + \abs{\log \kppbg} \, \ud t \leq C \dlt^{2},
\end{equation*}
which completes the proof. \qedhere
\end{proof}

As a quick consequence of \eqref{eq:df-large-r:v-u} in the preceding lemma, we obtain the following relation between $v-u$ and $r$ in the region $r \geq 20$.
\begin{corollary} \label{cor:v-u-r}
We have
\begin{equation} \label{eq:v-u-r}
	\frac{1}{4} (v - u + C_{\gmm_{20}} + 20) \leq r(u, v) \leq 2 (v - u + C_{\gmm_{20}} + 20) \quad \hbox{ in } \calD_{(t_{B})} \cap \set{r \geq 20}.
\end{equation}
Moreover, it holds that
\begin{gather} 
C^{-1} \Lmb \leq C_{\gmm_{20}} \leq  C \Lmb, \label{eq:C-gmm-Lmb} \\
\abs{C_{\gmm_{20}} - \overline{C}_{\gmm_{30}}} \leq C. \label{eq:C-gmm}
\end{gather}
\end{corollary}
\begin{proof}
Let $v'_{20}(u) = v_{20}(u)$ (characterized by $r(u, v_{20}(u)) = 20$) if $u \geq 1 + C_{\gmm_{20}}$ and $v'_{20}(u) = 1$ if $u \leq 1 + C_{\gmm_{20}}$. Note that
\begin{equation*}
	r(u, v) - 20= \int_{\min \set{1+C_{\gmm_{20}}, u}}^{1+C_{\gmm_{20}}} -\dur (u', 1) \, \ud u' + \int_{v'_{20}(u)}^{v} \dvr (u, v') \, \ud v'
\end{equation*}
By \eqref{eq:kpp-bnd}, \eqref{eq:mu-large-r} and \eqref{eq:dur-large-r}, we have $\frac{1}{2} \leq \dvr \leq 2$ and $\frac{1}{4} \leq - \dur \leq 1$ in $\set{r \geq 20}$. Moreover, by \eqref{eq:df-large-r:v-u} and taking $\eps$, $\de$ sufficiently small, we have $\abs{u - v_{20}(u) - C_{\gmm_{20}}} \leq 1$ when $u \geq 1 + C_{\gmm_{20}}$. Then \eqref{eq:v-u-r} follows by combining these facts. 

By \eqref{eq:R0-Lmb} and \eqref{eq:v-u-r} with $u = v = 1$, we obtain \eqref{eq:C-gmm-Lmb}. Moreover, \eqref{eq:C-gmm} follows by plugging in $v = 1$, $u = 1 + \overline{C}_{\gmm_{30}}$ (where $\rbg(u, v) = 30$) and using \eqref{eq:df-large-r:rdf}.  \qedhere
\end{proof}

Another corollary of Lemma~\ref{lem:df-large-r} is that we can improve the bootstrap assumption \eqref{eq:en-btstrp-R-Lmb}:

\begin{corollary}\label{cor:rbg-uf-vf-Lmb}
We have
$$|\rbg(u_{(t_B)f}, v_{(t_B)f})-\Lambda|\leq \f 12.$$
\end{corollary}
\begin{proof}
This follows from \eqref{eq:R0-Lmb}, \eqref{eq:df-large-r:rdf} and that by definition $r(u_{(t_B)f}, v_{(t_B)f})=R_0$. \qedhere
\end{proof}

Using the bounds in Lemma~\ref{lem:df-large-r} on the curve $\gmm_{20}$ (on which $r = 20$) as a starting point, we now consider the region $\set{r \leq 20}$. Here we need to crucially exploit the red-shift effect of the event horizon.
\begin{lemma} \label{lem:df-small-r}
In the region $\calD_{(t_{B})} \cap \set{r \leq 20}$, the following difference bounds hold:
\begin{align} 
	\sup_{\calD_{(t_{B})} \cap \set{r \leq 20}} \abs{\rdf} 
\leq & 	C A \log \Lmb \left( \eps + \dlt_{(t_{B})}^{2}\right),				\label{eq:df-small-r:rdf} \\
	\sup_{\calD_{(t_{B})} \cap \set{r \leq 20}} \abs{\dvrdf} 
\leq & 	C A \log \Lmb \left( \eps + \dlt_{(t_{B})}^{2} \right),				\label{eq:df-small-r:dvrdf} \\
	\sup_{\calD_{(t_{B})} \cap \set{r \leq 20}} (1+\abs{u - v - C_{\gmm_{20}}})^{-1} \abs{\log (-\dur) - \log (-\durbg)} 
\leq & 	C A \log \Lmb  \left( \eps + \dlt_{(t_{B})}^{2}\right).				\label{eq:df-small-r:logdurdf} 
\end{align}
Finally, there exists $c_{(\dur)} > 0$, which depends on $c_{(\durbg)}$ in \eqref{eq:durbg-bnd} and $|\ebg|$, $\eta_0$, $\omg$, such that
\begin{equation} \label{eq:df-small-r:dur} 
	\sup_{\calD_{(t_{B})} \cap \set{r \leq 20}} e^{c_{(\dur)} (u - v - C_{\gmm_{20}})} \abs{\dur} 
\leq  	C. 					
\end{equation}
\end{lemma}

\begin{proof}
As before, the key estimate is \eqref{eq:df-small-r:rdf}, and all other bounds in this lemma rely on it. To ensure uniform boundedness in $u$ and $v$, we utilize the \emph{red-shift effect} enjoyed by $\rdf$ along the event horizon towards the past; see \eqref{eq:df-small-r:redshift} below, which is essentially positivity of surface gravity of the event horizon of the background solution.

\pfstep{Step~1: Proof of \eqref{eq:df-small-r:rdf}}
Our starting point is the equation for $\rd_{v} \rdf$. This equation was written down in \eqref{eq:dvrdf-pre}; however, in this case it turns out to be a little more convenient to re-express $\rd_{v} \rdf$ as follows:
\begin{equation} \label{eq:dvrdf}
	\rd_{v} \rdf = \left( - \frac{2 \varpidf}{r} + \frac{\e^{2} - \ebg^{2}}{r^{2}} \right) \kpp + (1-\mubg) \left( 1- \frac{\kppbg}{\kpp}\right) \kpp+ \frac{2(\varpibg - \frac{\ebg^{2}}{\rbg})}{\rbg^{2}} \kpp \, \rdf  + \overline{\calE}[\rd_{v} \rdf],
\end{equation}
where
\begin{equation*}
	\overline{\calE}[\rd_{v} \rdf] = \left( - 2 \varpibg (r^{-1} - \rbg^{-1} + \rbg^{-2} \rdf) + \ebg^{2} (r^{-2} - \rbg^{-2} + 2 \rbg^{-3} \rdf)  \right) \kpp.
\end{equation*}
In order to use the method of an integrating factor, we rewrite \eqref{eq:dvrdf} as
\begin{equation*}
	\rd_{v} (e^{I} \rdf)(u, v) = e^{I} F(u ,v), 
\end{equation*}
where
\begin{align*}
	\rd_{v} I = & - \frac{2 \varpibg - \frac{\ebg^{2}}{\rbg}}{\rbg^{2}} \kpp, \\
	F =& \left( - \frac{2 \varpidf}{r} + \frac{\e^{2} - \ebg^{2}}{r^{2}} \right) \kpp + (1-\mubg) \left( 1- \frac{\kppbg}{\kpp}\right) \kpp + \overline{\calE}[\rd_{v} \rdf].
\end{align*}
Therefore, $\rdf(u, v)$ can be represented by the formula
\begin{equation} \label{eq:dvrdf-rf}
	\rdf(u, v) 
	= e^{I(u, v_{20}(u)) - I(u, v)} \, \rdf(u, v_{20}(u)) - \int_{v}^{v_{20}(u)} e^{I(u, v') - I(u, v)} F(u, v') \, \ud v'.
\end{equation}

We first bound the exponent $I(u, v') - I(u, v)$, where we need to use the red-shift effect of the (background) event horizon. For $\dlt$ sufficiently small, we claim that there exists a constant $\overline{c}_{\EH} > 0$  such that
\begin{equation} \label{eq:df-small-r:redshift}
\frac{2 (\varpibg - \frac{\ebg^{2}}{\rbg})}{\rbg^{2}} \geq \overline{c}_{\EH} > 0 \quad \hbox{ in } \calD \cap \set{\rbg \leq 30}.
\end{equation}
We first verify this statement on the exact Reissner--Nordstr\"om spacetime (i.e., when $\dlt = 0$). By the relation $\rbg_{\EH} = \varpibg_{f} + \sqrt{\varpibg_{f}^{2} - \overline{\e}^{2}}$ (where $\rbg_{\EH} = 1$ by our normalization), when $r\in [r_{\EH}, 30 r_{\EH}]$, we have
\begin{equation*}
	\frac{2 (\varpibg - \frac{\ebg^{2}}{\rbg})}{\rbg^{2}} 
\geq \frac{2}{30^{3} \rbg_{\EH}^{3}} (\varpibg_{f} \rbg_{\EH} - \ebg^{2})
= \frac{2}{30^{3} \rbg_{\EH}^{2}} \sqrt{\varpibg_{f}^{2} - \ebg^{2}},
\end{equation*}
where the last term is positive thanks to the subextremality assumption in Definition~\ref{def:dlt-adm}. Then in view of Definition~\ref{def:dlt-adm}, the desired lower bound \eqref{eq:df-small-r:redshift} follows by taking $\dlt$ sufficiently small compared to $\sqrt{\varpibg_{f}^{2}-\abs{\ebg}^{2}}$.
We note that the constant $\overline{c}_{\EH}$ can be chosen to depend only on $1-\frac{\abs{\ebg}}{\varpibg_{f}}$.

By \eqref{eq:df-large-r:rdf}, \eqref{eq:en-btstrp-eps}, Proposition~\ref{prop:dlt-tB-bnd}, the smallness of $\de_0=\Lmb^{100\eta_0}\de$ and monotonicity of $\rbg$, note that 
\begin{equation*}
\calD_{(t_{B})} \cap \set{(u, v) : r(u, v) \leq 20} \subseteq \set{(u, v) : \rbg(u, v) \leq 30}.
\end{equation*}
Therefore, for $v' > v$ we have the one-sided bound
\begin{equation} \label{eq:dvrdf-exp}
	I(u, v') - I(u, v)
	= - \int_{v}^{v'} \frac{2 \varpibg - \frac{\ebg^{2}}{\rbg}}{\rbg^{2}}(u, v'') \, \ud v''
	\leq - \overline{c}_{\EH} (v' - v).
\end{equation}

We now estimate each term in \eqref{eq:dvrdf-rf}. For the first term on the RHS, we simply neglect the exponential and use \eqref{eq:df-large-r:rdf} to estimate
\begin{equation} \label{eq:df-small-r:rdf-ini}
\abs{e^{I(u, v_{20}(u)) - I(u, v)} \rdf(u, v_{20}(u))} \leq \abs{\rdf(u, v_{20}(u))} \leq C A \log \Lmb (\eps + \dlt_{(t_{B})}^{2}).
\end{equation}
For the contribution of the forcing term, we crucially use \eqref{eq:dvrdf-exp} to bound
\begin{equation*}
\Abs{\int_{v}^{v_{20}(u)} e^{I(u, v') - I(u, v)} F(u, v') \, \ud v'}
\leq \sup_{v' \in [v, v_{20}(u)]} \abs{F(u, v')}.
\end{equation*}
To treat the contribution of $F(u, v')$, which contains nonlinear terms in $\rdf$, we use a bootstrap argument. Our bootstrap assumption is $\sup_{v' \in [v, v_{20}(u)]} \abs{\rdf}(u, v') \leq \veps$ (where $0 < \veps < \frac{1}{100}$ is to be specified below). By Lemma~\ref{lem:r-alp-df}, but with the roles of $r$ and $\rbg$ reversed, we have
\begin{equation*}
	\Abs{\overline{\calE}[\rd_{v} \rdf]} \leq C \left( \frac{\varpibg}{\rbg^{3}} + \frac{\ebg^{2}}{\rbg^{4}} \right) \abs{\kpp} \abs{\rdf}^{2}.
\end{equation*}
Furthermore, by the bootstrap assumption and the lower bound $\rbg \geq 1- C \dlt$, we have (for small $\dlt_{0}$) $r = \rbg + \rdf \geq \frac{1}{2}$.
Then by $\abs{\ebg} < 1$, \eqref{eq:dlt-adm:m}, \eqref{eq:mubg-bnd}, \eqref{eq:kpp-bnd}, \eqref{eq:eps-tB}, \eqref{eq:en-btstrp-m} and \eqref{eq:en-btstrp-kpp}, we have
\begin{equation*}
\sup_{v' \in [v, v_{20}(u)]} \abs{F(u, v')} \leq C A \log \Lmb \left( \eps + \dlt_{(t_{B})}^{2} \right) + C \veps \sup_{v' \in [v, v_{20}(u)]} \abs{\rdf(u, v')}.
\end{equation*}

In sum, we arrive at the bound
\begin{equation*}
	\abs{\rdf(u, v)}
	\leq C A \log \Lmb \left( \eps + \dlt_{(t_{B})}^{2} \right) + C \veps \sup_{v' \in [v, v_{20}(u)]} \abs{\rdf(u, v')}.
\end{equation*}
Choosing $\veps > 0$ small enough, we can absorb the last term into the LHS. The bootstrap assumption can be improved by taking $A \log \Lmb (\eps + \dlt_{(t_{B})}^{2})$ sufficiently small; moreover, it holds initially thanks to \eqref{eq:df-small-r:rdf-ini}. The desired estimate \eqref{eq:df-small-r:rdf} now follows. 

\pfstep{Step~2: Proof of \eqref{eq:df-small-r:dvrdf}} 
This bound follows immediately by plugging in \eqref{eq:df-small-r:rdf} into \eqref{eq:dvrdf} and estimating the remaining terms as in the previous step. 

\pfstep{Step~3: Proof of \eqref{eq:df-small-r:logdurdf}} Recall that $\log (-\dur)$ obeys the equation
\begin{equation*}
	\rd_{v} (\log (-\dur)) = \frac{2 (\varpi - \frac{\e^{2}}{r})}{r^{2}} \kpp.
\end{equation*}
Therefore, we may write
\begin{align*}
	\rd_{v} (\log (-\dur) - \log(-\durbg))
	= & 2 \left(\frac{\varpi- \varpibg}{r^{2}} - \frac{\e^{2} - \ebg^{2}}{r^{3}} \right) \kpp + \frac{2 (\varpibg - \frac{\ebg^{2}}{r})}{r^{2}} \left( 1 - \frac{\kppbg}{\kpp} \right) \kpp \\
	& + 2 \varpibg \kppbg (r^{-2} - \rbg^{-2}) - 2 \ebg^{2} \kppbg (r^{-3} - \rbg^{-3}).
\end{align*}
Proceeding as in Step~1 (in particular, the pointwise estimate for $F(u, v')$), it follows that
\begin{equation*}
	\abs{\rd_{v} (\log (-\dur) - \log(-\durbg))} \leq C A \log \Lmb \left( \eps + \dlt_{(t_{B})}^{2} \right).
\end{equation*}
Integrating the above inequality to $v$ from $v= v_{20}(u)$, we arrive at 
\begin{equation*}
	\abs{\log (-\dur) - \log(-\durbg)}(u, v)
	\leq \abs{\log (-\dur) - \log(-\durbg)}(u, v_{20}(u)) + C A \log \Lmb \left( \eps + \dlt_{(t_{B})}^{2} \right) \abs{v_{20}(u) - v}.
\end{equation*}
Since $-\dur, -\durbg$ are uniformly bounded away from $0$ (by a universal constant), the first term can be bounded by
\begin{equation*}
\abs{\log (-\dur) - \log(-\durbg)}(u, v_{20}(u)) \leq C \abs{\durdf}(u, v_{20}(u)) \leq  C A \log \Lmb \left( \eps + \dlt_{(t_{B})}^{2}\right),
\end{equation*}
which is acceptable. On the other hand, since
\begin{equation*} 
	\abs{v_{20}(u) - u + C_{\gmm_{20}}} \leq C A \left( \eps + \dlt_{(t_{B})}^{2} \right)  + C\de^2
\end{equation*}
thanks to \eqref{eq:df-large-r:v-u}, the second term contributes the factor $(1 + \abs{u - v - C_{\gmm_{20}}})$ in \eqref{eq:df-small-r:logdurdf}.

\pfstep{Step~4: Proof of \eqref{eq:df-small-r:dur}} 
We begin by writing
\begin{equation*}
	(-\dur) 
	= \frac{\dur}{\durbg}  (-\durbg)
	= e^{\log(-\dur) - \log(-\durbg)} (-\durbg).
\end{equation*}
By \eqref{eq:durbg-bnd}, \eqref{eq:C-gmm} and \eqref{eq:df-small-r:logdurdf}, it follows that
\begin{equation*}
	(-\dur) \leq e^{C A \log \Lmb (\eps + \dlt_{(t_{B})}^{2}) (1 + \abs{u - v - C_{\gmm_{20}}})}e^{ - c_{(\durbg)} (u - v - C_{\gmm_{20}})} e^{ c_{(\durbg)} C},
\end{equation*}
where $c_{(\durbg)} > 0$ is the constant from \eqref{eq:durbg-bnd}. Taking $A \log \Lmb (\eps + \dlt_{(t_{B})}^{2})$ in the exponent sufficiently small (compared to $c_{(\durbg)}$), the desired bound follows with $c_{(\dur)} = \frac{1}{2} c_{(\durbg)}$. \qedhere
\end{proof}

As a corollary of Lemma~\ref{lem:df-small-r}, we obtain the following obvious bounds on $r$ and $1-\mu$.
\begin{corollary} \label{cor:small-r}
In the region $\calD_{(t_{B})}$, we have
\begin{gather} 
	r \geq \frac{1}{2}, \label{eq:r-bnd} \\
	0 \leq 1-\mu \leq 20. \label{eq:mu-small-r}
\end{gather}
\end{corollary}
\begin{proof}
For \eqref{eq:r-bnd}, we use \eqref{eq:rbg-bnd} and \eqref{eq:df-small-r:rdf}. Nonnegativity of $1-\mu$ follows from \eqref{no.trapped.in.1} together with the positivity of $\Omg^2$. For the upper bound in \eqref{eq:mu-small-r}, we bound $1-\mu \leq 1 + \frac{2 \e^{2}}{r^{3}} \leq 20$ by \eqref{eq:e-bnd} and \eqref{eq:r-bnd}. \qedhere
\end{proof}

Another corollary of (the proof of) Lemma~\ref{lem:df-small-r} is a bound which, as we will see later, concerns the strength of the red-shift effect of the event horizon of the perturbed solution.
\begin{corollary} \label{cor:red-shift-bnd}
There exists a positive constant $c_{\EH}$, which depends on $|\ebg|$, $\eta_0$ and $\omg$, such that
\begin{equation} \label{eq:red-shift-bnd}
	\frac{2(\varpi - \frac{\e^{2}}{r})}{r^{2}} \geq c_{\EH} > 0 \quad \hbox{ in } \calD_{(t_{B})} \cap \set{r \leq 20}.
\end{equation}
\end{corollary}
\begin{proof}
The lower bound \eqref{eq:red-shift-bnd} follows from \eqref{eq:df-small-r:redshift}, using the difference estimates $\abs{\varpidf} \leq A (\eps + \dlt_{(t_{B})}^{2})$, $\abs{\e - \ebg} \leq \eps$ and \eqref{eq:df-small-r:rdf}.  \qedhere
\end{proof}

Before we end this subsection, let us note that we have already obtained strong enough estimates for the geometry of the perturbed spacetime to show that $\de_{(t_B)}\to 0$ as $t_B\to \infty$. This is given by the following proposition, which improves over the estimates in Proposition~\ref{prop:dlt-tB-bnd}.

\begin{proposition}\label{prop:dlt-tB-decay}
There exists $c_{\Lmb} > 0$ depending on $\Lmb > 0$ so that
\begin{equation} \label{eq:dlt-tB-decay}
\dlt_{(t_{B})}^{2} \leq C \dlt^{2} (1+ c_{\Lmb} t_{B})^{-1}.
\end{equation}
\end{proposition}
\begin{proof}

By Proposition~\ref{prop:dlt-tB-bnd}, it suffices to prove the proposition for $t_{B}\geq C \Lambda$. Moreover, since $Tt_{(B)}=1$ (and that by \eqref{eq:vf-T-0}, $T=-\f{1-\mu}{\rd_U r}\rd_U+\f{1-\mu}{\rd_V} \rd_V=-\f{1-\mu}{\rd_u r}\rd_u+\f{1-\mu}{\rd_v} \rd_v$), by \eqref{eq:kpp-bnd}, \eqref{eq:gmm-large-r} and \eqref{eq:R0-Lmb} (to ensure that $\gamma_{R_0}\subset\{r\geq 20\}$), we have 
\begin{equation}\label{eq:tB-uB-vB}
C^{-1} t_{(B)}\leq u_{(t_B)f},\, v_{(t_B)f}\leq Ct_{(B)}.
\end{equation}
We can therefore assume that $u_{(t_B)f}\geq C \Lambda$. For $C_{30}$ defined by the intersection of $\{\rbg=30\}$ and $\{v=1\}$, we have by \eqref{eq:dlt-adm:m}, \eqref{eq:mubg-bnd} and \eqref{eq:gmmbg-bnd} that $C_{30}\leq 2\Lambda$. Hence, we can assume that $u_{(t_B)f}\geq C_{30}+1$.

Under these assumptions, we claim that
\begin{align} 
	\sup_{C_{(t_{B}) out}^{f}} \abs{\log \kppbg} + \int_{C_{(t_{B}) out}^{f}} \abs{\log \kppbg} \, \ud v 
\leq & C \dlt^{2} (1+ c_{\Lmb} t_{B})^{-2 \omg + 1},		\label{eq:dlt-tB-decay:pf-1}\\
 	\sup_{\uC_{(t_{B}) in}^{f}} \rbg \abs{\log (-\gmmbg)} + \int_{\uC_{(t_{B}) in}^{f}} \abs{\log (-\gmmbg)} \, \ud u
\leq & C \dlt^{2} (1+ c_{\Lmb} t_{B})^{-1} .		\label{eq:dlt-tB-decay:pf-2} 
\end{align}

To prove the claim \eqref{eq:dlt-tB-decay:pf-1}, we first note that by \eqref{eq:bg-uv:kpp-decay},
$$\sup_{C_{(t_{B}) out}^{f}\cap\{\rbg\leq 30\}} \abs{\log \kppbg} + \int_{C_{(t_{B}) out}^{f}\cap\{\rbg\leq 30\}} \abs{\log \kppbg} \, \ud v \leq C \de^2 (u_{(t_B)f}-C_{30})^{-2\omg},$$
$$\sup_{C_{(t_{B}) out}^{f}\cap\{\rbg\geq 30\}} \abs{\log \kppbg} + \int_{C_{(t_{B}) out}^{f}\cap\{\rbg\geq 30\}} \abs{\log \kppbg} \, \ud v \leq C \de^2 \left(\left(v_{30}(u_{(t_B)f})\right)^{-2\omg+1} + u_{(t_B)f}^{-2\omg}\right),$$
where $v_{30}(u_{(t_B)f})$ is defined by\footnote{Note that such a $v_{30}(u_{(t_B)f})$ exists since $u_{(t_B)f}\geq C_{30}+1$.} $\rbg(u_{(t_B)f}, v_{10}(u_{(t_B)f}))=30$.
In deriving the bound for the integral term when $\rbg\geq 30$, we use the observation similar to that in Proposition~\ref{prop:dlt-tB-bnd} that the $v$-length of $C^{f}_{(t_{B}) out} \cap \set{\rbg \geq 30}$ is bounded by $C \Lmb$; and that this can be compensated by either the $v$ decay or the factor $\Lmb^{-1}$ in \eqref{eq:bg-uv:kpp-decay}. The claim \eqref{eq:dlt-tB-decay:pf-1} is then a consequence of \eqref{eq:tB-uB-vB}, the fact $C^{-1}\Lambda\leq C_{30}\leq C\Lambda$ (which follows from \eqref{eq:dlt-adm:m}, \eqref{eq:mubg-bnd} and \eqref{eq:gmmbg-bnd}) and the fact $v_{30}(u_{(t_B)f})\geq C^{-1} (u_{(t_B)f}-C_{\gamma_{20}})$ (which follows from \eqref{eq:v-u-r}).

The claim \eqref{eq:dlt-tB-decay:pf-2} follows from \eqref{eq:bg-uv:gmm-decay} and the observation that either $u \geq \frac{1}{2} t_{B}$ or $\rbg \geq \Lmb + C^{-1} t_{B}$ on $\uC^{f}_{(t_{B}) in}$. This final observation follows from a bound for $\overline{\lambda}$ away from $0$, which is in turn a consequence of \eqref{eq:dlt-adm:m}, \eqref{eq:kppbg-bnd} and $\Lambda>100$. \qedhere
\end{proof}

\subsection{Vector field multipliers and $L^{2}$-Hardy-type estimates} \label{subsec:en-id}
Equipped with the bounds on the difference of geometric quantities, we now analyze the difference $\phidf (u, v) = \phi(u,v) - \phibg(u, v)$ of the scalar field. Clearly, $\phidf$ solves the equation
\begin{equation} \label{eq:wave4phidf}
	\rd_{u} \rd_{v} \phidf + \frac{\dvr}{r} \rd_{u} \phidf + \frac{\dur}{r} \rd_{v} \phidf 
	= - \left( \frac{\dvr}{r} - \frac{\dvrbg}{\rbg} \right) \rd_{u} \phibg - \left( \frac{\dur}{r} - \frac{\durbg}{\rbg} \right) \rd_{v} \phibg.
\end{equation}

The purpose of this subsection is to derive some basic energy and $L^{2}$ identities for analysis of \eqref{eq:wave4phidf}.
We prove four types of energy identities, all by multiplying appropriate vector fields:
\begin{itemize}
\item {\bf Almost conserved energy (Lemma~\ref{lem:consv-en}).} 
This estimate gives control of the energy on null curves, with degenerating weight near $\EH$. We use the vector field
\begin{equation} \label{eq:vf-T}
	T = \frac{1-\mu}{\dvr} \rd_{v} + \frac{1-\mu}{(-\dur)} \rd_{u} ,
\end{equation}
which we already used in Section~\ref{subsec:geom}.

\item {\bf Integrated local energy decay (or Morawetz) estimate (Lemma~\ref{lem:morawetz}).}
This estimate gives the spacetime term in \eqref{eq:energy}. We use the vector field
\begin{equation} \label{eq:vf-X}
	X = f(r) \left( - \frac{1-\mu}{\dvr} \rd_{v} + \frac{1-\mu}{(-\dur)} \rd_{u}\right)
\end{equation}
for an appropriate choice of $f$. This vector field also leads to control of 
$\iint r^{3- \eta_{0}} \frac{\dvr (-\dur)}{1-\mu} (\rd_{v} \phidf)^{2} (\rd_{u} \phidf)^{2} \, \ud u \ud v$, which is important for controlling nonlinear terms in $\phidf$ in the energy estimate (Section~\ref{subsec:en-pf}); see Remark~\ref{rem:int-morawetz} below.

To control the zeroth order term in the spacetime term in \eqref{eq:energy} in a large-$r$ region, we use 
\begin{equation} \label{eq:vf-Xrad}
	X_{(rad)} = \frac{1}{\dvr} \rd_{v} - \frac{1}{(-\dur)} \rd_{u},
\end{equation}
which approaches $X$ with $f = -1$ as $r \to \infty$. Note that, when restricted to the Minkowski space $(\varpi = \e = 0)$, $X_{(rad)}$ coincides with the radial vector field $2 \rd_{r}$ used by Morawetz \cite{Morawetz}. 

Furthermore, to obtain a refined integrated local energy decay estimate for $\rd_{v} \phidf$ near $\EH$ (see Remark~\ref{rem:iled-small-r}), we use the vector field
\begin{equation} \label{eq:vf-Xi}
	\Xi = e^{-c_{\Xi} (u - v - C_{\gmm_{20}})} \rd_{v}
\end{equation}
for some $c_{\Xi} > 0$ to be chosen. Unlike all other vector fields we consider, $\Xi$ is \emph{coordinate-dependent} and \emph{irregular} across $\EH$ in general.

\item {\bf Red-shift estimate (Lemma~\ref{lem:red-shift}).}
This estimate provides control of the nondegenerate energy near $\EH$ in \eqref{eq:energy}. We use the vector field
\begin{equation} \label{eq:vf-Y}
	Y = \frac{\chi_{\EH}(r)}{(-\dur)} \rd_{u}
\end{equation}
where $\chi_{\EH}(r)$ is a suitable cutoff near $\EH$.

\item {\bf $\log r$-weighted energy (or Dafermos--Rodnianski) estimate (Lemma~\ref{lem:DR}).} This estimate is a variant of the $r^{p}$-weighted energy estimate of Dafermos--Rodnianski. It is used to prove \eqref{eq:energy-dvphi}, which in turn is important for closing the bootstrap assumptions for $\log (-\gmm) - \log (-\gmmbg)$. We use the vector field
\begin{equation} \label{eq:vf-Z}
	Z = \frac{\log(1+r)}{\dvr} \rd_{v} .
\end{equation}
\end{itemize}

Also, we prove $L^{2}$-Hardy-type estimates on null curves (Lemma~\ref{lem:hardy-type} and Lemma~\ref{lem:hardy-opt}). Lemma~\ref{lem:hardy-type} is used to relate $\rd_{v} (r \phidf)$ with $\rd_{v} \phidf$, and can be thought of as a variant of the $L^{2}$ Hardy inequality. 
Lemma~\ref{lem:hardy-opt} is the optimal $L^{2}$ Hardy inequality, which is used to control lower order terms involving $\phidf$ using $\rd_{v} \phidf$.

\subsubsection*{Energy identities}
Let $\calO$ be an open set such that $\overline{\calO} \subseteq \calD_{(t_{B})}$. In the following six lemmas (Lemmas~\ref{lem:consv-en}--\ref{lem:DR}), we let $\phidf \in C^{2}(\overline{\calO})$ be a solution to the inhomogeneous wave equation
\begin{equation} \label{eq:wave-inhom}
	\rd_{u} \rd_{v} \phidf + \frac{\dvr}{r} \rd_{u} \phidf + \frac{\dur}{r} \rd_{v} \phidf = F.
\end{equation}
with some forcing term $F \in C^{0}(\overline{\calO})$.

We start with the vector field $T$ in \eqref{eq:vf-T}
\begin{lemma} [Almost conserved energy] \label{lem:consv-en}
We have 
\begin{equation} \label{eq:en-id-T}
\begin{aligned}
&\hskip-2em
	\rd_{u} \left( \frac{1}{2} \frac{1-\mu}{\dvr} (\rd_{v} \phidf)^{2} r^{2}\right) 
	+ \rd_{v} \left( \frac{1}{2} \frac{1-\mu}{(-\dur)} (\rd_{u} \phidf)^{2} r^{2} \right)   \\
= & F \left( \frac{1-\mu}{\dvr} \rd_{v} \phidf + \frac{1-\mu}{(-\dur)} \rd_{u} \phidf \right) r^{2} \\
	& + \frac{1}{2}  \frac{1-\mu}{\dvr (-\dur)}  r^{3} \left( (\rd_{u} \phibg)^{2} + 2 \rd_{u} \phibg \rd_{u} \phidf \right) (\rd_{v} \phidf)^{2} 
	 - \frac{1}{2} \frac{1-\mu}{\dvr (-\dur)}  r^{3} \left( (\rd_{v} \phibg)^{2} + 2 \rd_{v} \phibg \rd_{v} \phidf \right)  (\rd_{u} \phidf)^{2} .
\end{aligned}
\end{equation}
\end{lemma}
\begin{proof}
We multiply \eqref{eq:wave-inhom} by $r^{2} T \phidf$ and compute:
\begin{align*}
& \hskip-2em 
	\left(\rd_{u} \rd_{v} \phidf + \frac{\dvr}{r} \rd_{u} \phidf + \frac{\dur}{r} \rd_{v} \phidf \right) r^{2} \left( \frac{1-\mu}{\dvr} \rd_{v} \phidf + \frac{1-\mu}{(-\dur)} \rd_{u} \phidf \right) \\
	= & \rd_{u} \left( \frac{1}{2} \frac{1-\mu}{\dvr} r^{2} (\rd_{v} \phidf)^{2}\right) 
	+ \rd_{v} \left( \frac{1}{2} \frac{1-\mu}{(-\dur)} r^{2} (\rd_{u} \phidf)^{2} \right) \\
	& - \frac{1}{2} \rd_{u}\left( \frac{1-\mu}{\dvr} r^{2} \right) (\rd_{v} \phidf)^{2} 
	- \frac{1}{2} \rd_{v} \left( \frac{1-\mu}{(-\dur)} r^{2} \right) (\rd_{u} \phidf)^{2} \\
	& + \dvr r \frac{1-\mu}{\dvr} \rd_{u} \phidf \rd_{v} \phidf + \dvr r \frac{1-\mu}{(-\dur)} (\rd_{u} \phidf)^{2} 
	 - (-\dur) r \frac{1-\mu}{\dvr} (\rd_{v} \phidf)^{2} - (-\dur) r \frac{1-\mu}{(-\dur)} \rd_{u} \phidf \rd_{v} \phidf \\
	= & \rd_{u} \left( \frac{1}{2} \frac{1-\mu}{\dvr} r^{2} (\rd_{v} \phidf)^{2}\right) 
	+ \rd_{v} \left( \frac{1}{2} \frac{1-\mu}{(-\dur)} r^{2} (\rd_{u} \phidf)^{2} \right) \\
	& - \frac{1}{2} \rd_{u}\left( \frac{1-\mu}{\dvr} \right) r^{2} (\rd_{v} \phidf)^{2} 
	+ \frac{1}{2} \rd_{v} \left( \frac{1-\mu}{\dur}\right) r^{2}  (\rd_{u} \phidf)^{2}.
\end{align*}
Recalling the equations \eqref{eq:EMSF-ray} for $\rd_{u} \frac{1-\mu}{\dvr} = \rd_{u} \kpp^{-1}$ and $\rd_{v} \frac{1-\mu}{\dur} = \rd_{u} \gmm^{-1}$, we see that the last line is equal to
\begin{align*}
& \hskip-2em
	\frac{1}{2}  \frac{1-\mu}{\dvr \dur}  r^{3} (\rd_{u} \phi)^{2} (\rd_{v} \phidf)^{2} 
	- \frac{1}{2} \frac{1-\mu}{\dvr \dur}  r^{3} (\rd_{v} \phi)^{2}  (\rd_{u} \phidf)^{2} \\
	= & \frac{1}{2}  \frac{1-\mu}{\dvr \dur}  r^{3} \left( (\rd_{u} \phibg)^{2} + 2 \rd_{u} \phibg \rd_{u} \phidf \right) (\rd_{v} \phidf)^{2} 
	- \frac{1}{2} \frac{1-\mu}{\dvr \dur}  r^{3} \left( (\rd_{v} \phibg)^{2} + 2 \rd_{v} \phibg \rd_{v} \phidf \right)  (\rd_{u} \phidf)^{2} ,
\end{align*}
which completes the proof. \qedhere
%
%
\end{proof}

Next, we consider the vector field $X$ in \eqref{eq:vf-X}.
\begin{lemma} [Morawetz] \label{lem:morawetz}
Let $f$ be a $C^{2}$ function on $(0, \infty)$. Denoting $f=f(r)$, $f'=\f{d}{dr}f(r)$ and $f''=\f{d^2}{dr^2}f(r)$, we have
\begin{equation} \label{eq:en-id-X0}
\begin{aligned}
& \hskip-2em
	- \frac{1}{2} f' \left( \frac{1-\mu}{\dvr} (-\dur) (\rd_{v} \phidf)^{2}  + \frac{1-\mu}{(-\dur)} \dvr (\rd_{u} \phidf)^{2}  \right) r^{2} 
	- 2 f \frac{1-\mu}{r}  \rd_{u} \phidf \rd_{v} \phidf \, r^{2} \\
& \hskip-2em
	+ \frac{1}{2} f \frac{1-\mu}{\dvr (-\dur)}  r^{3} (\rd_{v} \phi)^{2} (\rd_{u} \phidf)^{2} 
	+ \frac{1}{2} f \frac{1-\mu}{\dvr (-\dur)}  r^{3} (\rd_{u} \phi)^{2} (\rd_{v} \phidf)^{2} \\
	= & \rd_{u} \left( \frac{f}{2} \frac{1-\mu}{\dvr} (\rd_{v} \phidf)^{2} r^{2} \right) 
		- \rd_{v} \left( \frac{f}{2} \frac{1-\mu}{(-\dur)} (\rd_{u} \phidf)^{2} r^{2} \right) \\
	& + F f \left( - \frac{1-\mu}{\dvr} \rd_{v} \phidf + \frac{1-\mu}{(-\dur)} \rd_{u} \phidf \right) r^{2}  .
\end{aligned}
\end{equation}
Moreover, we also have
\begin{equation} \label{eq:en-id-X}
\begin{aligned}
& \hskip-2em
	- \frac{1}{2} f' \left( \frac{1-\mu}{\dvr} (-\dur) (\rd_{v} \phidf)^{2} + \frac{1-\mu}{(-\dur)} \dvr (\rd_{u} \phidf)^{2} \right) r^{2} 
	+ f \frac{4 \varpi - 2 \frac{\e^{2}}{r}}{r^{2}} \rd_{u} \phidf \rd_{v} \phidf \, r^{2} \\
& \hskip-2em
	+ \left( f'' r \dvr (-\dur) 
	+ (f' r - f) \frac{2 (\varpi - \frac{\e^{2}}{r})}{r^{2}} \frac{\dvr (-\dur)}{1-\mu} \right) \phidf^{2} \\
& \hskip-2em
	+ \frac{1}{2} f \frac{1-\mu}{\dvr (-\dur)}  r^{3} (\rd_{v} \phi)^{2} (\rd_{u} \phidf)^{2} 
	+ \frac{1}{2} f \frac{1-\mu}{\dvr (-\dur)}  r^{3} (\rd_{u} \phi)^{2} (\rd_{v} \phidf)^{2} \\
	= & 
	\rd_{u} \left( \frac{f}{2} \frac{1-\mu}{\dvr} (\rd_{v} \phidf)^{2} r^{2} + \frac{1}{2} \rd_{v}(f r \phidf^{2}) - f' \dvr r \phidf^{2} \right) \\
	& - \rd_{v} \left( \frac{f}{2} \frac{1-\mu}{(-\dur)} (\rd_{u} \phidf)^{2} r^{2} - \frac{1}{2} \rd_{u}(f r \phidf^{2}) - f' (-\dur) r \phidf^{2} \right) \\
	& + F f \left( - \frac{1-\mu}{\dvr} \rd_{v} \phidf + \frac{1-\mu}{(-\dur)} \rd_{u} \phidf - 2 \frac{\phidf}{r} \right) r^{2} .
\end{aligned}
\end{equation}
\end{lemma}

\begin{proof}
We multiply \eqref{eq:wave-inhom} by $r^{2} X\phidf$ and compute:
\begin{align*}
& \hskip-2em
	\left(\rd_{u} \rd_{v} \phidf + \frac{\dvr}{r} \rd_{u} \phidf + \frac{\dur}{r} \rd_{v} \phidf \right) r^{2} f \left( - \frac{1-\mu}{\dvr} \rd_{v} \phidf + \frac{1-\mu}{(-\dur)} \rd_{u} \phidf \right) \\
	= & - \rd_{u} \left( \frac{f}{2} \frac{1-\mu}{\dvr} r^{2} (\rd_{v} \phidf)^{2}\right) 
		- \rd_{v} \left( \frac{f}{2} \frac{1-\mu}{\dur} r^{2} (\rd_{u} \phidf)^{2} \right) \\
	& + \frac{1}{2} \rd_{u}\left( f \frac{1-\mu}{\dvr} r^{2} \right) (\rd_{v} \phidf)^{2} 
	- \frac{1}{2} \rd_{v} \left( f \frac{1-\mu}{(-\dur)} r^{2} \right) (\rd_{u} \phidf)^{2} \\
	& - \dvr f \frac{1-\mu}{\dvr} r \rd_{u} \phidf \rd_{v} \phidf + \dvr f \frac{1-\mu}{(-\dur)} r (\rd_{u} \phidf)^{2} 
	 + (-\dur) f \frac{1-\mu}{\dvr} r (\rd_{v} \phidf)^{2} - (-\dur) f \frac{1-\mu}{(-\dur)} r \rd_{u} \phidf \rd_{v} \phidf \\
	= & \rd_{u} \left( - \frac{f}{2} \frac{1-\mu}{\dvr} r^{2} (\rd_{v} \phidf)^{2}\right) + \rd_{v} \left( \frac{f}{2} \frac{1-\mu}{(-\dur)} r^{2} (\rd_{u} \phidf)^{2} \right)  \\
	& - \frac{1}{2} f' \left( \frac{1-\mu}{\dvr} (-\dur) r^{2} (\rd_{v} \phidf)^{2} + \frac{1-\mu}{(-\dur)} \dvr r^{2} (\rd_{u} \phidf)^{2} \right) \\
	& + \frac{1}{2} f \frac{1-\mu}{\dvr (-\dur)}  r^{3} (\rd_{v} \phi)^{2} (\rd_{u} \phidf)^{2} 
	+ \frac{1}{2} f \frac{1-\mu}{\dvr (-\dur)}  r^{3} (\rd_{u} \phi)^{2} (\rd_{v} \phidf)^{2} \\
	& - 2 f (1-\mu) r \rd_{u} \phidf \rd_{v} \phidf.
\end{align*}

This proves the first identity. To prove the second identity, we use \eqref{eq:wave-inhom} and the equation \eqref{eq:EMSF-r-phi-m} for $\rd_{u} \rd_{v} r$ to rewrite the term $2 f r \rd_{u} \phidf \rd_{v} \phidf$ as follows:
\begin{equation} \label{eq:morawetz-pf}
\begin{aligned}
2 f r \rd_{u} \phidf \rd_{v} \phidf
 = & 	\rd_{u} \left(2 f r \phidf \rd_{v} \phidf\right)
 	- 2 (\rd_{u} f) r \phidf \rd_{v} \phidf - 2 f \phidf \rd_{u} (r \rd_{v} \phidf) \\
 = & 	\rd_{u} \left(2 f r \phidf \rd_{v} \phidf\right)
 	- 2 f r \phidf \left(\rd_{u} \rd_{v} \phidf + \frac{\dvr}{r} \rd_{u} \phidf  + \frac{\dur}{r} \rd_{v} \phidf \right)
	- f' r \dur  \rd_{v} \phidf^{2} + f \dvr \rd_{u} \phidf^{2} \\
 = & 	\rd_{u} \left( 2 f r \phidf \rd_{v} \phidf\right) 
 	- 2 f r \phidf \left(\rd_{u} \rd_{v} \phidf + \frac{\dvr}{r} \rd_{u} \phidf  + \frac{\dur}{r} \rd_{v} \phidf \right) 	 \\
&	+ \rd_{u} \left(f \dvr \phidf^{2} \right)
	- \rd_{v} \left((f' r \dur) \phidf^{2} \right) 
	+ \rd_{v} (f' r \dur)  \phidf^{2} 
	- \rd_{u}(f \dvr) \phidf^{2} \\
 = & 	\rd_{u} \rd_{v} \left( f r \phidf^{2} \right)
 	- \rd_{u} \left( f' r \dvr \phidf^{2} \right) - \rd_{v} \left(f' r \dur \phidf^{2} \right) \\
&	- 2 f r \phidf \left(\rd_{u} \rd_{v} \phidf + \frac{\dvr}{r} \rd_{u} \phidf  + \frac{\dur}{r} \rd_{v} \phidf \right) 
	+ f'' r \dvr \dur \phidf^{2} + (f' r - f) \frac{2 (\varpi - \frac{\e^{2}}{r})}{r^{2}} \frac{\dvr \dur}{1-\mu} \phidf^{2} .
\end{aligned}
\end{equation}
Finally, splitting the term
$\rd_{u} \rd_{v} ( f r \phidf^{2} ) = \frac{1}{2} \rd_{u} (\rd_{v} (f r \phidf^{2})) + \frac{1}{2} \rd_{v} (\rd_{u} (f r \phidf^{2}))$,  \eqref{eq:en-id-X} follows. \qedhere
\end{proof}

We consider now the variant $X_{(rad)}$ of $X$ (with $f = -1$) defined in \eqref{eq:vf-Xrad}.
\begin{lemma} \label{lem:morawetz-1}
We have
\begin{equation} \label{eq:en-id-Xrad}
\begin{aligned}
& \hskip-2em
	\frac{1}{2} \rd_{u} \left( \dvr \phidf^{2} + 2 \phidf \rd_{v} \phidf r + \frac{1}{\dvr} (\rd_{v} \phidf)^{2} r^{2} \right) \\
& \hskip-2em
	- \frac{1}{2} \rd_{v} \left( (-\dur) \phidf^{2} - 2 r \phidf \rd_{u} \phidf + \frac{1}{(-\dur)} (\rd_{u} \phidf)^{2} r^{2} \right) 
	+ \frac{2 (\varpi - \frac{\e^{2}}{r})}{r^{2}} \frac{\dvr (-\dur)}{1-\mu} \phidf^{2} \\
	= & F \left( \frac{1}{\dvr} \rd_{v} \phidf - \frac{1}{(-\dur)} \rd_{u} \phidf + 2 \frac{\phidf}{r} \right) r^{2} 
	 + \frac{(\varpi - \frac{\e^{2}}{r})}{r^{2}} \left( \frac{(-\dur)}{1-\mu} \frac{1}{\dvr} (\rd_{v} \phidf)^{2} + \frac{\dvr}{1-\mu} \frac{1}{(-\dur)} (\rd_{u} \phidf)^{2} \right) r^{2}.
\end{aligned}
\end{equation}
\end{lemma}

\begin{proof}
As in Lemma~\ref{lem:morawetz}, we first compute
\begin{align*}
& \hskip-2em
	\left(\rd_{u} \rd_{v} \phidf + \frac{\dvr}{r} \rd_{u} \phidf + \frac{\dur}{r} \rd_{v} \phidf \right) r^{2} \left( \frac{1}{\dvr} \rd_{v} \phidf - \frac{1}{(-\dur)} \rd_{u} \phidf \right) \\
	= & \rd_{u} \left( \frac{1}{2} \frac{1}{\dvr} r^{2} (\rd_{v} \phidf)^{2}\right) 
		- \rd_{v} \left( \frac{1}{2} \frac{1}{(-\dur)} r^{2} (\rd_{u} \phidf)^{2} \right) 
	 - \frac{1}{2} \rd_{u} \left( \frac{1}{\dvr} r^{2} \right) (\rd_{v} \phidf)^{2} 
	+ \frac{1}{2} \rd_{v} \left( \frac{1}{(-\dur)} r^{2} \right) (\rd_{u} \phidf)^{2} \\
	& + \dvr \frac{1}{\dvr} r \rd_{u} \phidf \rd_{v} \phidf - \dvr \frac{1}{(-\dur)} r (\rd_{u} \phidf)^{2} 
	 - (-\dur) \frac{1}{\dvr} r (\rd_{v} \phidf)^{2} + (-\dur) \frac{1}{(-\dur)} r \rd_{u} \phidf \rd_{v} \phidf \\
	= & \rd_{u} \left( \frac{1}{2} \frac{1}{\dvr} r^{2} (\rd_{v} \phidf)^{2}\right) 
		- \rd_{v} \left( \frac{1}{2} \frac{1}{(-\dur)} r^{2} (\rd_{u} \phidf)^{2} \right) 
		+ 2 r \rd_{u} \phidf \rd_{v} \phidf \\
	& - \frac{1}{2} (-\rd_{u} \rd_{v} r) \left( \frac{1}{\dvr^{2}} (\rd_{v} \phidf)^{2} + \frac{1}{(-\dur)^{2}} (\rd_{u} \phidf)^{2} \right) r^{2} .
\end{align*}
Then for the term $2 r \rd_{u} \phidf \rd_{v} \phidf$, we use \eqref{eq:morawetz-pf} with $f = 1$ and rearrange terms. \qedhere
\end{proof}

We also record a computation concerning $\Xi$ defined in \eqref{eq:vf-Xi}.
\begin{lemma} \label{lem:en-id-Xi}
For any constant $c_{\Xi} > 0$, we have
\begin{equation} \label{eq:en-id-Xi}
\begin{aligned}
& \hskip-2em
	\rd_{u} \left(\frac{1}{2} e^{- c_{\Xi}(u - v - C_{\gmm_{20}})} r^{2} (\rd_{v} \phidf)^{2} \right)
	+ \frac{c_{\Xi}}{2} e^{- c_{\Xi}(u - v - C_{\gmm_{20}})} r^{2} (\rd_{v} \phidf)^{2} \\
	= & e^{- c_{\Xi}(u - v - C_{\gmm_{20}})} F  \rd_{v} \phidf r^{2}
	+ e^{- c_{\Xi}(u - v - C_{\gmm_{20}})} \dvr r \rd_{u} \phidf \rd_{v} \phidf.
\end{aligned}
\end{equation}
\end{lemma}
\begin{proof}
Let $h(u, v) = e^{-c_{\Xi}(u - v - C_{\gmm_{20}})}$. We compute
\begin{align*}
	\left(\rd_{u} \rd_{v} \phidf + \frac{\dvr}{r} \rd_{u} \phidf + \frac{\dur}{r} \rd_{v} \phidf \right) h r^{2} \rd_{v} \phidf 
	= \rd_{u} \left(\frac{1}{2} h r^{2} (\rd_{v} \phidf)^{2} \right) 
	- \frac{1}{2} (\rd_{u} h) r^{2} (\rd_{v} \phidf)^{2}
	+ \dvr r h \rd_{u} \phidf \rd_{v} \phidf.
\end{align*}
Since $- \rd_{u} h = c_{\Xi} e^{- c_{\Xi} (u - v - C_{\gmm_{20}})}$, the lemma follows. \qedhere
\end{proof}

Next, we consider the vector field $Y$ in \eqref{eq:vf-Y}.
\begin{lemma}[Red-shift] \label{lem:red-shift}
Let $\chi_{\EH}$ be a $C^{1}$ function on $(0, \infty)$. We have
\begin{equation} \label{eq:en-id-Y}
\begin{aligned}
& \hskip-2em
	 \rd_{v} \left( \frac{1}{2}\frac{\chi_{\EH}(r)}{(-\dur)} (\rd_{u} \phidf)^{2}  r^{2}  \right) 
	 + \chi_{\EH}(r) \frac{(\varpi - \frac{\e^{2}}{r}) \kpp}{r^{2}}  \frac{1}{(-\dur)} (\rd_{u} \phidf)^{2} r^{2} \\
	= & F \left( \frac{\chi_{\EH}(r)}{(-\dur)} \rd_{u} \phidf \right) r^{2} \\
	& + \chi_{\EH}' (r)\frac{\dvr}{2}  \frac{1}{(-\dur)} (\rd_{u} \phidf)^{2} r^{2}
	+ \frac{\chi_{\EH}(r)}{r^{2}} \rd_{v} \phidf \rd_{u} \phidf \, r^{2}.
\end{aligned}
\end{equation}
\end{lemma}

\begin{proof}
We multiply \eqref{eq:wave-inhom} by $r^{2} Y \phidf$ and compute:
\begin{align*}
& \hskip-2em
	\left(\rd_{u} \rd_{v} \phidf + \frac{\dvr}{r} \rd_{u} \phidf + \frac{\dur}{r} \rd_{v} \phidf \right) r^{2} \left( \frac{\chi_{\EH}(r)}{(-\dur)} \rd_{u} \phidf \right) \\
	= & \rd_{v} \left( \frac{1}{2} r^{2} \frac{\chi_{\EH}(r)}{(-\dur)} (\rd_{u} \phidf)^{2} \right)
		- \frac{1}{2} \rd_{v} \left(r^{2} \frac{\chi_{\EH}(r)}{(-\dur)} \right) (\rd_{u} \phidf)^{2} \\
	& + \dvr r \frac{\chi_{\EH}(r)}{(-\dur)} (\rd_{u} \phidf)^{2} + \dur r \frac{\chi_{\EH}(r)}{(-\dur)} \rd_{v} \phidf \rd_{u} \phidf \\
	= & \rd_{v} \left( \frac{1}{2} r^{2} \frac{\chi_{\EH}(r)}{(-\dur)} (\rd_{u} \phidf)^{2} \right) + \left( \varpi - \frac{\e^{2}}{r}\right) \frac{\kpp}{(-\dur)} \chi_{\EH}(r) (\rd_{u} \phidf)^{2} \\
	& - \frac{1}{2} \dvr  \frac{\chi_{\EH}'(r)}{(-\dur)} (\rd_{u} \phidf)^{2} r^{2}
	- r \chi_{\EH}(r) \rd_{v} \phidf \rd_{u} \phidf.
\end{align*}
In the last equality, we have used the equation \eqref{eq:EMSF-r-phi-m} for $\rd_{u} \rd_{v} r$. \qedhere
\end{proof}

Finally, we consider the vector field $Z$ in \eqref{eq:vf-Z}. 
\begin{lemma}[$\log r$-weighted energy] \label{lem:DR}
We have
\begin{equation} \label{eq:en-id-Z}
\begin{aligned}
& \hskip-2em
	\rd_{u} \left( \frac{1}{2} \frac{\log(1+r)}{\dvr} (\rd_{v} (r \phidf))^{2} \right)
	+ \rd_{v} \left( \frac{1}{2} \frac{\log(1+r)}{\dvr} r (- \rd_{u} \rd_{v} r) \phidf^{2} \right)
	+ \frac{1}{2} \frac{(-\dur)}{1+r} \frac{1}{\dvr} (\rd_{v} (r \phidf))^{2} \\
	= & F \frac{r \log(1+r)}{\dvr} \rd_{v}(r \phidf) 
	+ \frac{\log(1+r)}{2} \frac{-\rd_{u} \rd_{v} r}{\dvr} \frac{1}{\dvr} (\rd_{v} (r \phidf))^{2}
	+ \frac{1}{2} \rd_{v} \left(\frac{\log(1+r)}{r} \frac{(- \rd_{u} \rd_{v} r)}{\dvr} \right) r^{2} \phidf^{2}.
\end{aligned}
\end{equation}

\end{lemma}
\begin{proof}
Let $g(r) = \log (1+r)$.
We first rewrite \eqref{eq:wave-inhom} in the form
\begin{equation*}
	\rd_{u} \rd_{v} (r \phidf) - (\rd_{u} \rd_{v} r) \phidf = r F.
\end{equation*}
We multiply this equation by $\frac{g(r)}{\dvr} \rd_{v}(r \phidf)$. Then
\begin{align*}
& \hskip-2em
	\left(\rd_{u} \rd_{v} (r \phidf) - (\rd_{u} \rd_{v} r) \phidf \right) \frac{g(r)}{\dvr} \rd_{v}(r \phidf) \\
= & \rd_{u} \left( \frac{1}{2} \frac{g(r)}{\dvr} (\rd_{v} (r \phidf))^{2} \right) - \frac{1}{2} \rd_{u} \left( \frac{g(r)}{\dvr} \right) (\rd_{v} (r \phidf))^{2} \\
& + \rd_{v} \left(- \frac{1}{2} \frac{g(r)}{\dvr} r (\rd_{u} \rd_{v} r) \phidf^{2} \right) + \frac{1}{2} \rd_{v} \left(\frac{g(r)}{r} \frac{\rd_{u} \rd_{v} r}{\dvr} \right) r^{2} \phidf^{2}.
\end{align*}
Expanding the second term on the RHS (by the product rule) and rearranging terms, we arrive at \eqref{eq:en-id-Z}. \qedhere
\end{proof}

\subsubsection*{$L^{2}$-Hardy-type identities}
Next, we state and prove an $L^{2}$ identity on outgoing null curves, which relates $\rd(r f)$ with $\rd f$. 
\begin{lemma} \label{lem:hardy-type}
Let $\alp \in \bbR$ and $1 \leq v_{1} < v_{2}$. For any $C^{1}$ function $f$ on $C_{u} \cap \set{v \in [v_{1}, v_{2}]}$, we have
\begin{equation} \label{eq:hardy-alp}
\begin{aligned}
& \hskip-2em
	\alp \int_{v_{1}}^{v_{2}} r^{\alp} \dvr f^{2} (u, v) \, \ud v
	+ \int_{v_{1}}^{v_{2}} \frac{r^{\alp}}{\dvr} \left(\rd_{v} (r f)\right)^{2} (u, v) \, \ud v 
	+ r^{1+\alp} f^{2}(u, v_{1}) \\
	= & \int_{v_{1}}^{v_{2}} \frac{r^{2+\alp}}{\dvr} \left(\rd_{v} f \right)^{2} (u, v) \, \ud v 
		 + r^{1+\alp} f^{2}(u, v_{2}).
\end{aligned}
\end{equation}
Similarly, for any $C^{1}$ function $f$ on $\uC_{v} \cap \set{u \in [u_{1}, u_{2}]}$, we have
\begin{equation} \label{eq:hardy-alp-u}
\begin{aligned}
& \hskip-2em
	\alp \int_{u_{1}}^{u_{2}} r^{\alp} (-\dur) f^{2} (u, v) \, \ud u
	+ \int_{u_{1}}^{u_{2}} \frac{r^{\alp}}{(-\dur)} \left(\rd_{u} (r f)\right)^{2} (u, v) \, \ud u 
	+ r^{1+\alp} f^{2}(u_{2}, v) \\
	= & \int_{u_{1}}^{u_{2}} \frac{r^{2+\alp}}{(-\dur)} \left(\rd_{u} f \right)^{2} (u, v) \, \ud u 
		 + r^{1+\alp} f^{2}(u_{1}, v).
\end{aligned}
\end{equation}
\end{lemma}
\begin{proof}
We only prove \eqref{eq:hardy-alp}, as the proof of \eqref{eq:hardy-alp-u} is entirely analogous.
We fix $u$, and omit the $u$-dependence of all functions. 
We compute
\begin{align*}
	\int_{v_{1}}^{v_{2}} \frac{r^{\alp}}{\dvr} \left(\rd_{v} (r f)\right)^{2} (v) \, \ud v
	= & \int_{v_{1}}^{v_{2}} \frac{r^{\alp}}{\dvr} \left(\dvr f + r \rd_{v} f \right)^{2} (v) \, \ud v \\
	= & \int_{v_{1}}^{v_{2}} \frac{r^{2+\alp}}{\dvr} (\rd_{v} f )^{2} (v) \, \ud v
	+ \int_{v_{1}}^{v_{2}} r^{1+\alp} \rd_{v} f^{2} (v) \, \ud v 
	+ \int_{v_{1}}^{v_{2}} r^{\alp} \dvr f^{2} (v)\, \ud v \\
	= & \int_{v_{1}}^{v_{2}} \frac{r^{2+\alp}}{\dvr} (\rd_{v} f )^{2} (v) \, \ud v - \alp \int_{v_{1}}^{v_{2}} r^{\alp} \dvr f^{2} (v) \, \ud v \\
	& + r^{1+\alp}(v_{2}) f^{2}(v_{2}) - r^{1+\alp}(v_{1}) f^{2}(v_{1}).		\qedhere
\end{align*}
\end{proof}

In fact, essentially the same computation proves one-dimensional Hardy-type inequalities.
\begin{lemma} \label{lem:hardy-opt}
Let $\alp \in \bbR$, $(u, v) \in \calD_{(t_{B})}$, $1 \leq v_{1} < v_{2}$ and $1 \leq u_{1} < u_{2}$. For any $C^{1}$ function $f$ on $C_{u} \cap \set{v \in [v_{1}, v_{2}]}$, we have
\begin{equation} \label{eq:hardy-v}
\begin{aligned}
& \hskip-2em
	\frac{(\alp +1)^{2}}{4} \int_{v_{1}}^{v_{2}} r^{\alp} \dvr f^{2} (u, v) \, \ud v
	+ \int_{v_{1}}^{v_{2}} \frac{r^{\alp}}{\dvr} \left(r \rd_{v} f + \frac{\alp + 1}{2} \dvr f \right)^{2} (u, v) \, \ud v 
	+ \frac{\alp + 1}{2} r^{1+\alp} f^{2}(u, v_{1}) \\
	= & \int_{v_{1}}^{v_{2}} \frac{r^{2+\alp}}{\dvr} \left(\rd_{v} f \right)^{2} (u, v) \, \ud v 
		 + \frac{\alp + 1}{2} r^{1+\alp} f^{2}(u, v_{2}).
\end{aligned}
\end{equation}
Similarly, for any $C^{1}$ function $f$ on $\uC_{v} \cap \set{u \in [u_{1}, u_{2}]}$, we have
\begin{equation} \label{eq:hardy-u}
\begin{aligned}
& \hskip-2em
	\frac{(\alp +1)^{2}}{4} \int_{u_{1}}^{u_{2}} r^{\alp} (-\dur) f^{2} (u, v) \, \ud u
	+ \int_{u_{1}}^{u_{2}} \frac{r^{\alp}}{(-\dur)} \left(r \rd_{u} f + \frac{\alp + 1}{2} (-\dur) f \right)^{2} (u, v) \, \ud u 
	  + \frac{\alp + 1}{2} r^{1+\alp} f^{2}(u_{2}, v) \\
	= & \int_{u_{1}}^{u_{2}} \frac{r^{2+\alp}}{(-\dur)} \left(\rd_{u} f \right)^{2} (u, v) \, \ud u 
		+ \frac{\alp + 1}{2} r^{1+\alp} f^{2}(u_{1}, v).
\end{aligned}
\end{equation}
\end{lemma}
\begin{proof}
Equation \eqref{eq:hardy-v} is proved by exactly the same computation as Lemma~\ref{lem:hardy-type}, except we replace $\dvr f + r \rd_{v} f$ by $\frac{\alp+1}{2} \dvr f + r \rd_{v} f$. Moreover, \eqref{eq:hardy-u} is proved by simply switching $u$ and $v$. We omit the details. \qedhere
\end{proof}

\subsection{Proof of the energy estimates} \label{subsec:en-pf}
We now put together the identities proved in Section~\ref{subsec:en-id} and prove energy estimates. The main result we establish here is Proposition~\ref{prop:energy-goal} below. We also prove the improved integrated local energy decay for $\rd_{v} \phidf$, which quickly follows from Lemma~\ref{lem:DR} once Proposition~\ref{prop:energy-goal} is proved.

\subsubsection*{Statement of energy estimates}
To properly state the energy estimates, we need to define the restriction $(\Gmm_{(t_{B}) \tau})_{\tau \in [1, \infty)}$ of the foliation $(\Gmm_{\tau})_{\tau \in [1, \infty)}$ to $\calD_{(t_{B})}$. More precisely, we define
\begin{align*}
	\calD_{(t_{B})}(\tau_{1}, \tau_{2}) =& \cup_{\tau \in (\tau_{1}, \tau_{2})} \Gmm_{(t_{B}) \tau}, \qquad 
	\Gmm_{(t_{B}) \tau} =  \Gmm_{(t_{B}) \tau}^{(in)} \cup \Gmm_{(t_{B}) \tau}^{(out)} , \\
	\Gmm_{(t_{B}) \tau}^{(in)} =& \Gmm_{\tau}^{(in)} \cap \calD_{(t_{B})} = \set{(u, v) \in \calD_{(t_{B})} : v = v_{R_{0}}(\tau), u \in [\tau, \infty)} , \\
	\Gmm_{(t_{B}) \tau}^{(out)} =& \Gmm_{\tau}^{(out)} \cap \calD_{(t_{B})} = \set{(u, v) \in \calD_{(t_{B})} : u = \tau, v \in [v_{R_{0}}(\tau), \infty)} .
\end{align*}
Moreover, we define the restricted null curves
\begin{equation*}
	C_{(t_{B}) u}(\tau_{1}, \tau_{2}) = C_{u} \cap \overline{\calD_{(t_{B})}(\tau_{1}, \tau_{2})}, \quad
	\uC_{(t_{B}) v}(\tau_{1}, \tau_{2}) = \uC_{v} \cap \overline{\calD_{(t_{B})}(\tau_{1}, \tau_{2})}.
\end{equation*}
In particular, we write
\begin{align*}
	\EH_{(t_{B})} (\tau_{1}, \tau_{2}) 
	= & C_{(t_{B}) u_{(t_{B}) f}}(\tau_{1}, \tau_{2})
	= \set{(u_{(t_{B}) f}, v) : v \in [v_{R_{0}}(\tau_{1}), v_{R_{0}}(\tau_{2})]}, \\
	\NI_{(t_{B})} (\tau_{1}, \tau_{2}) 
	= & \uC_{(t_{B}) v_{(t_{B}) f}}(\tau_{1}, \tau_{2})
	= \set{(u, v_{(t_{B}) f}) : u \in [\tau_{1}, \tau_{2}]}.
\end{align*}
Despite the notation, note that $\EH_{(t_{B})}(\tau_{1}, \tau_{2})$ [resp. $\NI_{(t_{B})}(\tau_{1}, \tau_{2})$] does \underline{not} lie in the event horizon [resp. null infinity] of the background or the perturbed solution. Nevertheless, it tends to $\EH(\tau_{1}, \tau_{2})$ [resp. $\NI(\tau_{1}, \tau_{2})$] as $t_{B} \to \infty$, which justifies our notation.

With the proper definitions in place, we now make the following simplifying convention: Since $t_{B}$ is fixed throughout the argument, we drop the $t_{B}$-dependence of the above objects and simply write $\calD(\tau_{1}, \tau_{2}) = \calD_{(t_{B})}(\tau_{1}, \tau_{2})$, $\Gmm_{\tau} = \Gmm_{(t_{B}) \tau}$ etc. We also write $u_{f} = u_{(t_{B}) f}$ and $v_{f} = v_{(t_{B}) f}$ for the final $u = u_{(t_{B})}$ and $v = v_{(t_{B})}$ coordinates in $\calD_{(t_{B})}$.

\begin{proposition} \label{prop:energy-goal}
For every $1 \leq \tau_{1} \leq \tau_{2}$, we have
\begin{equation} \label{eq:en-goal}
\begin{aligned}
& \hskip-2em
	\int_{\Gmm^{(in)}_{\tau_{2}}} \frac{1}{(-\dur)} (\rd_{u} \phidf)^{2} r^{2} \, \ud u
	+ \int_{\Gmm^{(out)}_{\tau_{2}}} \kpp^{-1} (\rd_{v} \phidf)^{2} r^{2} \, \ud v  \\
& \hskip-2em
	+ \int_{\NI(\tau_{1}, \tau_{2})} \frac{1}{(-\dur)} (\rd_{u} \phidf)^{2} r^{2} \, \ud u
	+ \int_{\EH(\tau_{1}, \tau_{2})} \kpp^{-1} (\rd_{v} \phidf)^{2} r^{2} \, \ud v  \\
& \hskip-2em
	+ \iint_{\calD(\tau_{1}, \tau_{2})} \frac{1}{r^{1+\eta_{0}}} \left( \frac{1}{(-\dur)} (\rd_{u} \phidf)^{2} + \kpp^{-1} (-\dur) (\rd_{v} \phidf)^{2} + \frac{\phidf^{2}}{r^{2}} \right) r^{2} \, \ud u \ud v \\
	\leq & C \left( \int_{\Gmm^{(in)}_{\tau_{1}}} \frac{1}{(-\dur)} (\rd_{u} \phidf)^{2} r^{2} \, \ud u
	+ \int_{\Gmm^{(out)}_{\tau_{1}}} \kpp^{-1} (\rd_{v} \phidf)^{2} r^{2} \, \ud v \right) \\
	& + C A^{2} (\log^{2} \Lmb) \tau_{1}^{-2\omg+1+\eta_{0}} \dlt^{2} \left( \eps + \dlt_{(t_{B})}^{2} \right)^{2}
	+ C r \phidf^{2} (\tau_{1}, v_{f}) .
\end{aligned}
\end{equation}
Moreover, we also have
\begin{equation} \label{eq:en-null-goal}
\begin{aligned}
& \hskip-2em
	\sup_{v \in [v_{R_{0}}(\tau_{1}), v_{R_0}(\tau_{2})] }\int_{\uC_{v}(\tau_{1}, \tau_{2})} \frac{1}{(-\dur)} (\rd_{u} \phidf)^{2} r^{2} \, \ud u
	+ \sup_{u \in [\tau_{1}, \tau_{2}] }\int_{C_{u}(\tau_{1}, \tau_{2})} \kpp^{-1} (\rd_{v} \phidf)^{2} r^{2} \, \ud v  \\
	\leq & C \left( \int_{\Gmm^{(in)}_{\tau_{1}}} \frac{1}{(-\dur)} (\rd_{u} \phidf)^{2} r^{2} \, \ud u
	+ \int_{\Gmm^{(out)}_{\tau_{1}}} \kpp^{-1} (\rd_{v} \phidf)^{2} r^{2} \, \ud v \right) \\
	& + C A^{2} (\log^{2} \Lmb) \tau_{1}^{-2\omg+1+\eta_{0}} \dlt^{2} \left( \eps + \dlt_{(t_{B})}^{2} \right)^{2}
	+ C r \phidf^{2} (\tau_{1}, v_{f}) .
\end{aligned}
\end{equation}
\end{proposition}
\begin{remark}  \label{rem:zeroth}
The zeroth order term $r \phidf^{2} (\tau_{1}, v_{f})$ at the future endpoint of $\Gmm_{\tau_{1}}^{(out)}$ is needed to control lower order terms. Note that this term vanishes as $t_{B} \to 0$, since both $\phi$ and $\phibg$ decay with the rate $r^{-1}$. In the course of our bootstrap argument, this term is initially controlled by $C \eps^{2}$; see Corollary~\ref{cor:zeroth} below.
\end{remark}

Our goal is to simultaneously prove \eqref{eq:en-goal} and \eqref{eq:en-null-goal}. For the proof of \eqref{eq:en-goal}, and also for later use, it is convenient to introduce (for any $1 \leq \tau_{1} \leq \tau_{2}$)
\begin{align}
	E[\phidf](\tau_{1}, \tau_{2}) 
	= & \sup_{\tau' \in [\tau_{1}, \tau_{2}]} \left( \int_{\Gmm^{(in)}_{\tau'}} \frac{1}{(-\dur)} (\rd_{u} \phidf)^{2} r^{2} \, \ud u 
	+ \int_{\Gmm^{(out)}_{\tau'}} \kpp^{-1} (\rd_{v} \phidf)^{2} r^{2} \, \ud v  \right) \notag \\
	&+ \int_{\NI(\tau_{1}, \tau_{2})} \frac{1}{(-\dur)} (\rd_{u} \phidf)^{2} r^{2} \, \ud u
	+ \int_{\EH(\tau_{1}, \tau_{2})} \kpp^{-1} (\rd_{v} \phidf)^{2} r^{2} \, \ud v  \notag \\
	& + \iint_{\calD(\tau_{1}, \tau_{2})} \frac{1}{r^{1+\eta_{0}}} 
		\left( \frac{1}{(-\dur)} (\rd_{u} \phidf)^{2} + \kpp^{-1} (-\dur) (\rd_{v} \phidf)^{2} + (-\dur) \frac{\phidf^{2}}{r^{2}} \right) r^{2} \, \ud u \ud v \notag \\
	& + \iint_{\calD(\tau_{1}, \tau_{2}) \cap \set{r \leq 20}} e^{- \frac{1}{2} c_{(\dur)} (u - v - C_{\gmm_{20}})} (\rd_{v} \phidf)^{2} r^{2} \, \ud u \ud v 		\label{eq:iled-small-r} \\
	& + \iint_{\calD(\tau_{1}, \tau_{2})} \frac{1-\mu}{\dvr (- \dur)} r^{3-\eta_{0}} (\rd_{u} \phidf)^{2} (\rd_{v} \phidf)^{2} \, \ud u \ud v.	\label{eq:int-morawetz}
\end{align}
where $c_{(\dur)} > 0$ is as in \eqref{eq:df-small-r:dur}.

Roughly speaking, the proofs of \eqref{eq:en-goal} and \eqref{eq:en-null-goal} consist of two steps. First, we use the computation in the previous subsection to show that $E[\phidf](\tau_{1}, \tau_{2})$ and the LHS of \eqref{eq:en-null-goal} can be controlled by the nondegenerate energy at $\tau_{1}$ up to some error term $\calE$ (Lemma~\ref{lem:en-ctrl} and Lemma~\ref{lem:en-null-ctrl}). Second, we prove that the error term $\calE$ is controlled (Lemma~\ref{lem:en-err-near}--Lemma~\ref{lem:en-err-nonlin}).

\begin{remark}[Improved integrated local energy decay of $\rd_{v} \phidf$ near $\EH$] \label{rem:iled-small-r} 
The term \eqref{eq:iled-small-r} in $E[\phidf](\tau_{1}, \tau_{2})$ gives an improved control of $\rd_{v} \phidf$ near $\EH$, and it plays an important role in handling error terms near $\EH$ (see Lemma~\ref{lem:en-err-nonlin} below). This term is obtained using the irregular vector field $\Xi$ defined in \eqref{eq:vf-Xi}.
\end{remark}

\begin{remark}[Interaction Morawetz term] \label{rem:int-morawetz} 
The weighted spacetime integral for $(\rd_{v} \phidf)^{2} (\rd_{u} \phidf)^{2}$ in $E[\phidf](\tau_{1}, \tau_{2})$ (see \eqref{eq:int-morawetz}) arises from the Morawetz vector field (see \eqref{eq:en-id-X}), and it is useful for treating error terms in the $T$-estimate that are nonlinear in $\phidf$ (see Lemma~\ref{lem:en-err-nonlin} below). We dub this the \emph{interaction Morawetz term}, in view of its resemblance to the interaction potential estimate in the $(1+1)$-dimensional conservation laws literature.
\end{remark}

\subsubsection*{Combination of the vector field identities}

As outlined above, we begin by controlling $E[\phidf](\tau_{1}, \tau_{2})$ using the identities in Section~\ref{subsec:en-id}, up to an error term $\calE(\tau_{1}, \tau_{2})$ that we decompose carefully for the next step. 

\begin{lemma} \label{lem:en-ctrl}
For every $1 \leq \tau_{1} \leq \tau_{2}$, we have
\begin{equation} \label{eq:en-ctrl}
	E[\phidf](\tau_{1}, \tau_{2}) 
	\leq C \left( \int_{\Gmm^{(in)}_{\tau_{1}}} \frac{1}{(-\dur)} (\rd_{u} \phidf)^{2} r^{2} \, \ud u
	+ \int_{\Gmm^{(out)}_{\tau_{1}}} \kpp^{-1} (\rd_{v} \phidf)^{2} r^{2} \, \ud v + r \phidf^{2} (\tau_{1}, v_{f}) + \calE(\tau_{1}, \tau_{2}) \right).
\end{equation}
The error $\calE (\tau_{1}, \tau_{2}) \geq 0$ admits a decomposition of the form
\begin{equation*}
	\calE(\tau_{1}, \tau_{2}) = \calE_{near}(\tau_{1}, \tau_{2}) + \calE_{far}(\tau_{1}, \tau_{2}) + \calE_{nonlin}(\tau_{1}, \tau_{2}),
\end{equation*}
where
\begin{align*}
	\calE_{near}(\tau_{1}, \tau_{2})
	= & \iint_{\calD(\tau_{1}, \tau_{2}) \cap \set{r \leq 20}} 
		\left( \Abs{\frac{\dvr}{r} - \frac{\dvrbg}{\rbg}} \Abs{\rd_{u} \phibg} 
		+  \Abs{\frac{\dur}{r} - \frac{\durbg}{\rbg}} \Abs{\rd_{v} \phibg} \right) 
		\left( \frac{1}{(-\dur)} \abs{\rd_{u} \phidf} + \abs{\rd_{v} \phidf} + \abs{\phidf} \right)  \, \ud u \ud v, \\
	\calE_{far}(\tau_{1}, \tau_{2})
	= & \iint_{\calD(\tau_{1}, \tau_{2}) \cap \set{r \geq 20}} 
		\left( \Abs{\frac{\dvr}{r} - \frac{\dvrbg}{\rbg}} \Abs{\rd_{u} \phibg} 
		+  \Abs{\frac{\dur}{r} - \frac{\durbg}{\rbg}} \Abs{\rd_{v} \phibg} \right) 
		\left( \abs{\rd_{u} \phidf} + \abs{\rd_{v} \phidf} + \frac{\abs{\phidf}}{r}\right) r^{2} \, \ud u \ud v, \\
	\calE_{nonlin}(\tau_{1}, \tau_{2})
	= & \iint_{\calD(\tau_{1}, \tau_{2})} \frac{1-\mu}{\dvr (-\dur)}  r^{3} \left( (\rd_{u} \phibg)^{2} + \abs{\rd_{u} \phibg} \abs{\rd_{u} \phidf} \right) (\rd_{v} \phidf)^{2} \, \ud u \ud v \\
	& + \iint_{\calD(\tau_{1}, \tau_{2})} \frac{1-\mu}{\dvr (-\dur)}  r^{3} \left( (\rd_{v} \phibg)^{2} + \abs{\rd_{v} \phibg} \abs{\rd_{v} \phidf} \right)  (\rd_{u} \phidf)^{2} \, \ud u \ud v .
\end{align*}
\end{lemma}

\begin{proof}
To simplify the notation, in this proof we omit the line and area elements $\ud u$, $\ud v$ and $\ud u \, \ud v$. In view of \eqref{eq:wave4phidf}, we apply the results in Section~\ref{subsec:en-id} with
\begin{equation*}
	F = - \left( \frac{\dvr}{r} - \frac{\dvrbg}{\rbg} \right) \rd_{u} \phibg - \left( \frac{\dur}{r} - \frac{\durbg}{\rbg} \right) \rd_{v} \phibg.
\end{equation*}

We proceed in several steps. 

\pfstep{Step~1: A weaker energy estimate}
We begin by proving a weaker version of \eqref{eq:en-ctrl}, with sub-optimal weights and no control of $\phidf^{2}$.
More precisely, for any $1 \leq \tau_{1} \leq \tau_{2}$ we define 
\begin{equation*}
\begin{aligned}
	E_{w}[\phidf](\tau_{1}, \tau_{2})
	= & \sup_{\tau' \in [\tau_{1}, \tau_{2}]} \left( \int_{\Gmm^{(in)}_{\tau'}} \frac{1}{(-\dur)} (\rd_{u} \phidf)^{2} r^{2} 
	+ \int_{\Gmm^{(out)}_{\tau'}} \kpp^{-1} (\rd_{v} \phidf)^{2} r^{2}  \right) \\
	& + \int_{\NI(\tau_{1}, \tau_{2})} \frac{1}{(-\dur)} (\rd_{u} \phidf)^{2} r^{2} 
	+ \int_{\EH(\tau_{1}, \tau_{2})} \kpp^{-1} (\rd_{v} \phidf)^{2} r^{2}  \\
	& + \iint_{\calD(\tau_{1}, \tau_{2})} \frac{1}{r^{4}} 
		\left( \frac{1}{(-\dur)} (\rd_{u} \phidf)^{2} + \kpp^{-1} (-\dur) (\rd_{v} \phidf)^{2} \right) r^{2} \\
	& + \iint_{\calD(\tau_{1}, \tau_{2})} \frac{1-\mu}{\dvr (-\dur)} \left((\rd_{v} \phi)^{2} (\rd_{u} \phidf)^{2} + (\rd_{u} \phi)^{2} (\rd_{v} \phidf)^{2} \right) .
\end{aligned}
\end{equation*}
We claim that for every $1 \leq \tau_{1} \leq \tau_{2}$ we have
\begin{equation} \label{eq:en-w-ctrl}
	E_{w}[\phidf](\tau_{1}, \tau_{2}) 
	\leq C \left( \int_{\Gmm^{(in)}_{\tau_{1}}} \frac{1}{-\dur} (\rd_{u} \phidf)^{2} r^{2} 
	+ \int_{\Gmm^{(out)}_{\tau_{1}}} \kpp^{-1} (\rd_{v} \phidf)^{2} r^{2}  + \calE(\tau_{1}, \tau_{2}) \right) .
\end{equation}

Let $X_{0}$ be given by \eqref{eq:vf-X} with $f(r) = r^{-3}$, i.e.,
\begin{equation*}
	X_{0} = r^{-3} \left( - \frac{1-\mu}{\dvr} \rd_{v} + \frac{1-\mu}{(-\dur)} \rd_{u}\right).
\end{equation*}
Also, let $Y$ be given by \eqref{eq:vf-Y} with $\chi_{\EH}(r)$ a smooth non-negative function, with $0\leq \chi\leq 1$, which equals $1$ on $\set{r \leq 10}$ and $0$ on $\set{r \geq 20}$. To prove \eqref{eq:en-w-ctrl}, the idea is to use a multiplier of the form 
\begin{equation*}
a T + b X_{0} + c Y. 
\end{equation*}
That is, we sum up $a \times \eqref{eq:en-id-T} + b \times \eqref{eq:en-id-X0} + c \times \eqref{eq:en-id-Y}$, integrate over $\calD(\tau_{1}, \tau_{2})$, obtain boundary terms from the exact derivatives and choose $a, b, c \geq 0$ appropriately.

We start by computing the contribution of the LHSs of these identities. For the $T$ vector field, we observe that
\begin{equation} \label{eq:en-w-ctrl-T}
\begin{aligned}
& \hskip-2em
	\int_{\Gmm^{(in)}_{\tau_{2}}} \frac{1}{2} \frac{1-\mu}{(-\dur)} (\rd_{u} \phidf)^{2} r^{2} 
	+ \int_{\Gmm^{(out)}_{\tau_{2}}} \frac{1}{2} \kpp^{-1} (\rd_{v} \phidf)^{2} r^{2} \\
& \hskip-2em
	+ \int_{\NI(\tau_{1}, \tau_{2})} \frac{1}{2} \frac{1-\mu}{(-\dur)} (\rd_{u} \phidf)^{2} r^{2} 
	+ \int_{\EH(\tau_{1}, \tau_{2})} \frac{1}{2} \kpp^{-1} (\rd_{v} \phidf)^{2} r^{2}  \\
	\leq & \iint_{\calD(\tau_{1}, \tau_{2})} (\hbox{LHS of \eqref{eq:en-id-T}})
	+\int_{\Gmm^{(in)}_{\tau_{1}}} \frac{1}{2} \frac{1-\mu}{(-\dur)} (\rd_{u} \phidf)^{2} r^{2} 
	+\int_{\Gmm^{(out)}_{\tau_{1}}} \frac{1}{2} \kpp^{-1} (\rd_{v} \phidf)^{2} r^{2} .
\end{aligned}
\end{equation}
Next, we consider $X_{0}$. On the LHS of \eqref{eq:en-id-X0}, since $f \geq 0$ and $- f' \geq 0$, the only term with non-consistent sign is $2 f r^{-1} (1-\mu)  \rd_{u} \phidf \rd_{v} \phidf r^{2}$. For this term, we apply Cauchy-Schwarz and estimate
\begin{align*}
	\Abs{2 r^{-4} (1-\mu) \rd_{u} \phidf \rd_{v} \phidf r^{2}}
	\leq r^{-4} \left( \frac{1-\mu}{\dvr} (-\dur) (\rd_{v} \phidf)^{2} 
		+ \frac{1-\mu}{(-\dur)} \dvr (\rd_{u} \phidf)^{2} \right) r^{2}.
\end{align*}
On the other hand, since $- \frac{1}{2} f' = \frac{3}{2} r^{-4}$, the first term on the LHS of \eqref{eq:en-id-X0} equals
\begin{equation*}
	\frac{3}{2} r^{-4} \left( \frac{1-\mu}{\dvr} (-\dur) (\rd_{v} \phidf)^{2} 
		+ \frac{1-\mu}{(-\dur)} \dvr (\rd_{u} \phidf)^{2} \right) r^{2}.
\end{equation*}
Therefore, we see that
\begin{equation} \label{eq:en-w-ctrl-X}
\begin{aligned}
& \hskip-2em
	\frac{1}{2} \iint_{\calD(\tau_{1}, \tau_{2})} r^{-4} \left( \frac{1-\mu}{\dvr} (-\dur) (\rd_{v} \phidf)^{2} 
		+ \frac{1-\mu}{(-\dur)} \dvr (\rd_{u} \phidf)^{2} \right) r^{2} \\
& \hskip-2em
	+ \frac{1}{2} \iint_{\calD(\tau_{1}, \tau_{2})} \frac{1-\mu}{\dvr (-\dur)} \left((\rd_{v} \phi)^{2} (\rd_{u} \phidf)^{2} + (\rd_{u} \phi)^{2} (\rd_{v} \phidf)^{2} \right) \\
	\leq & \iint_{\calD(\tau_{1}, \tau_{2})} (\hbox{LHS of \eqref{eq:en-id-X0}}).
\end{aligned}\end{equation}
Finally, by \eqref{eq:red-shift-bnd} and our choice of $\chi_{\EH}$, we have
\begin{equation} \label{eq:en-w-ctrl-Y}
\begin{aligned}
& \hskip-2em
	\int_{\Gmm_{\tau_{2}}^{(in)} \cap \set{r \leq 10}} \frac{1}{2} \frac{1}{(-\dur)} (\rd_{u} \phidf)^{2} r^{2}
	+  \f{c_{\EH}}{2} \iint_{\calD(\tau_{1}, \tau_{2}) \cap \set{r \leq 10}} \frac{1}{(-\dur)} (\rd_{u} \phidf)^{2} r^{2} \\
	\leq & \iint_{\calD(\tau_{1}, \tau_{2})} (\hbox{LHS of \eqref{eq:en-id-Y}})
	+ \int_{\Gmm_{\tau_{1}}^{(in)} \cap \set{r \leq 20}} \frac{1}{2} \frac{1}{(-\dur)} (\rd_{u} \phidf)^{2} r^{2}.
\end{aligned}
\end{equation}

We now bound the contribution of the RHSs, using $\calE(\tau_{1}, \tau_{2})$ as well as the preceding lower bounds. For the $T$ vector field, we simply have the bound
\begin{equation} \label{eq:en-w-ctrl-T1}
	\Abs{\iint_{\calD(\tau_{1}, \tau_{2})} (\hbox{RHS of \eqref{eq:en-id-T}})}
	\leq C \calE(\tau_{1}, \tau_{2}),
\end{equation}
where the contribution of $F$ is controlled by $\calE_{near} + \calE_{far}$, and the rest by $\calE_{nonlin}$. For $X_{0}$, the first line of the RHS of \eqref{eq:en-id-X0} contributes boundary integrals, which can be controlled using \eqref{eq:en-w-ctrl-T} and the initial energy for some $C_{X_0}>0$:
\begin{equation} \label{eq:en-w-ctrl-X1}
\begin{aligned}
& \hskip-2em
	\Abs{\iint_{\calD(\tau_{1}, \tau_{2})} \rd_{u} \left( \frac{f}{2} \frac{1-\mu}{\dvr} (\rd_{v} \phidf)^{2} r^{2} \right) 
		- \iint_{\calD(\tau_{1}, \tau_{2})} \rd_{v} \left( \frac{f}{2} \frac{1-\mu}{(-\dur)} (\rd_{u} \phidf)^{2} r^{2} \right) } \\
	\leq & C_{X_{0}} \left(\iint_{\calD(\tau_{1}, \tau_{2})} (\hbox{LHS of \eqref{eq:en-id-T}})
	+\int_{\Gmm^{(in)}_{\tau_{1}}} \frac{1}{2} \frac{1-\mu}{(-\dur)} (\rd_{u} \phidf)^{2} r^{2} 
	+\int_{\Gmm^{(out)}_{\tau_{1}}} \frac{1}{2} \kpp^{-1} (\rd_{v} \phidf)^{2} r^{2} \right).
\end{aligned}\end{equation}
On the other hand, for the second line of the RHS of \eqref{eq:en-id-X0} we have
\begin{equation} \label{eq:en-w-ctrl-X2}
	\Abs{\iint_{\calD(\tau_{1}, \tau_{2})} F f \left( - \frac{1-\mu}{\dvr} \rd_{v} \phidf + \frac{1-\mu}{(-\dur)} \rd_{u} \phidf \right) r^{2} }
	\leq C ( \calE_{near}(\tau_{1}, \tau_{2}) + \calE_{far}(\tau_{1}, \tau_{2})).
\end{equation}
Finally, we turn to $Y$. For the first line on the RHS of \eqref{eq:en-id-Y}, we have
\begin{equation} \label{eq:en-w-ctrl-Y1}
	\Abs{\iint_{\calD(\tau_{1}, \tau_{2})}  F \left( \frac{\chi_{\EH}(r)}{(-\dur)} \rd_{u} \phidf \right) r^{2} }
	\leq C \calE_{near}(\tau_{1}, \tau_{2}) .
\end{equation}
On the other hand, in view of \eqref{eq:en-w-ctrl-X}, the remainder of the RHS of \eqref{eq:en-id-Y} can be bounded as follows for some $C_Y>0$:
\begin{equation} \label{eq:en-w-ctrl-Y2}
	\Abs{\iint_{\calD(\tau_{1}, \tau_{2})} 
	\chi'_{\EH} \frac{\dvr}{2 r^{2}} \frac{1}{(-\dur)} (\rd_{u} \phidf)^{2} r^{2}
	+ \frac{\chi_{\EH}}{r^{2}} \rd_{v} \phidf \rd_{u} \phidf \, r^{2} } 
	\leq C_{Y} \iint_{\calD(\tau_{1}, \tau_{2})} (\hbox{LHS of \eqref{eq:en-id-X0}}).
\end{equation}

In conclusion, we see that if we choose $c =1$, $b = 1+ c C_{Y}$ and $a = 1 + b C_{X_{0}}$, then \eqref{eq:en-w-ctrl} follows.

\pfstep{Step~2: Control of the zeroth order term}
The goal of this step is to get a preliminary control of the zeroth order term.
The key claim is as follows: For any $\veps > 0$ and a sufficiently large number $R_{1} > 40$, we have
\begin{equation} \label{eq:en-z-ctrl}
\begin{aligned}
& \hskip-2em
	\iint_{\calD(\tau_{1}, \tau_{2}) \cap \set{R_{1} \leq r \leq 2 R_{1}}} \frac{1}{r} \left( (\rd_{v} \phidf)^{2} + (\rd_{u} \phidf)^{2} + \frac{1}{r^{2}} \phidf^{2}\right) r^{2} \\
	\leq & C_{R_{1}, \veps} \left( E_{w}[\phidf](\tau_{1}, \tau_{2}) + r \phidf^{2}(\tau_{1}, v_{f}) \right) + C \calE (\tau_{1}, \tau_{2}) \\
	&  + C \frac{\veps^{1-\eta_{0}}}{R_{1}^{1-\eta_{0}}} \iint_{\calD(\tau_{1}, \tau_{2}) \cap \set{r \geq \veps^{-1} R_{1}}} \frac{1}{r^{1+\eta_{0}}} 
	\left(  \frac{1}{(-\dur)} (\rd_{u} \phidf)^{2} + \kpp^{-1} (-\dur) (\rd_{v} \phidf)^{2} \right) r^{2}.
\end{aligned}
\end{equation}
Note that the last term on the RHS is multiplied by an arbitrarily small factor $\veps^{1-\eta_{0}}$, and is moreover integrated only over $r \geq \veps^{-1} R_{1}$, which is well-separated from the region $\set{R_{1} \leq r \leq 2R_{1}}$ controlled by the LHS. These features allow us to handle this term in the last step below.

To prove \eqref{eq:en-z-ctrl}, we use the vector field $X_{(rad)}$ in \eqref{eq:vf-Xrad}. Let $\chi_{R_{1}}(r)$ be a smooth non-decreasing function such that 
\begin{equation*}
\chi_{R_{1}}(r) = 0 \hbox{ in } \set{r \leq \frac{R_{1}}{2}}, \quad
\chi_{R_{1}}(r) = 1 \hbox{ in } \set{r \geq 4 R_{1}}, \quad
\chi'_{R_{1}} \geq \frac{1}{10 r} \hbox{ in } \set{R_{1} \leq r \leq 2 R_{1}}, \quad
\abs{\chi'_{R_{1}}} \leq \frac{10}{r}.
\end{equation*}
We multiply \eqref{eq:en-id-Xrad} by $\chi_{R_{1}}(r)$ and integrate by parts. Concerning the LHS, we claim that
\begin{equation} \label{eq:en-z-ctrl-Xrad}
\begin{aligned}
& \hskip-2em
	\iint_{\calD(\tau_{1}, \tau_{2})} \frac{1}{2} \chi_{R_{1}}'(r) \dvr (-\dur) \phidf^{2} + \iint \frac{2(\varpi - \frac{\e^{2}}{r})}{r^{2}} \frac{\dvr(-\dur)}{1-\mu} \phidf^{2} \\
	\leq & \iint_{\calD(\tau_{1}, \tau_{2})} \chi_{R_{1}} (\hbox{LHS of \eqref{eq:en-id-Xrad}}) + C ( E_{w}[\phidf](\tau_{1}, \tau_{2}) + r \phidf^{2}(\tau_{1}, v_{f}) ),
\end{aligned}
\end{equation}
We show how to treat the first term on the LHS of \eqref{eq:en-id-Xrad} in detail. We have
\begin{align*}
& \hskip-2em
	\iint_{\calD(\tau_{1}, \tau_{2})} \chi_{R_{1}}(r)  \frac{1}{2} \rd_{u} \left( \dvr \phidf^{2} + 2 \phidf \rd_{v} \phidf r + \frac{1}{\dvr} (\rd_{v} \phidf)^{2} r^{2} \right) \\
	= & 	\iint_{\calD(\tau_{1}, \tau_{2})} \chi'_{R_{1}}(r)  \frac{(-\dur)}{2} \left( \dvr \phidf^{2} + 2 \phidf \rd_{v} \phidf r + \frac{1}{\dvr} (\rd_{v} \phidf)^{2} r^{2} \right) 
	 - \f 12 \int_{\Gmm_{\tau_{1}}^{(out)}} \chi_{R_{1}}(r) \left( \dvr \phidf^{2} + 2 \phidf \rd_{v} \phidf r + \frac{1}{\dvr} (\rd_{v} \phidf)^{2} r^{2} \right) \\
	& + \f 12 \int_{\EH(\tau_{1}, \tau_{2}) } \chi_{R_{1}}(r) \left( \dvr \phidf^{2} + 2 \phidf \rd_{v} \phidf r + \frac{1}{\dvr} (\rd_{v} \phidf)^{2} r^{2} \right) 
	 + \f 12 \int_{\Gmm_{\tau_{2}}^{(out)}} \chi_{R_{1}}(r) \left( \dvr \phidf^{2} + 2 \phidf \rd_{v} \phidf r + \frac{1}{\dvr} (\rd_{v} \phidf)^{2} r^{2} \right) \\
	\geq  & \frac{1}{4} \iint_{\calD(\tau_{1}, \tau_{2})} \chi'_{R_{1}}(r) \dvr (-\dur) \phidf^{2}  \\
	& - \frac{1}{2} \iint_{\calD(\tau_{1}, \tau_{2})} \chi'_{R_{1}}(r) \frac{(-\dur)}{\dvr} (\rd_{v} \phidf)^{2} r^{2} 
	- \f 12 \int_{\Gmm_{\tau_{1}}^{(out)}} \left( \dvr \phidf^{2} + 2 \phidf \rd_{v} \phidf r + \frac{1}{\dvr} (\rd_{v} \phidf)^{2} r^{2} \right).
\end{align*}
In the last inequality, we used the simple inequality $2 \phidf \rd_{v} \phidf r \geq - \dvr \phidf^{2} - \dvr^{-1} (\rd_{v} \phidf)^{2} r^{2}$ for the terms $\int_{\EH(\tau_{1}, \tau_{2})}(\cdots)$ and $\int_{\Gmm^{(out)}_{\tau_{2}}}(\cdots)$ (ensuring that both have a good sign), and the slight variant  $2 \phidf \rd_{v} \phidf r \geq - \frac{1}{2} \dvr \phidf^{2} - 2 \dvr^{-1} (\rd_{v} \phidf)^{2} r^{2}$ for the term $\iint_{\calD(\tau_{1}, \tau_{2})} \chi'_{R_{1}}(r) (\cdots)$. The last line is then bounded from below by $- C (E_{w}[\phidf](\tau_{1}, \tau_{2}) + r \phidf^{2}(\tau_{1}, v_{f}))$, where we use \eqref{eq:mu-large-r} and Lemma~\ref{lem:hardy-opt} (with $\alp = 0$) on $\Gmm^{(out)}_{\tau_{1}}$ to control $\int_{\Gmm^{(out)}_{\tau_{1}}} \dvr \phidf^{2}$.

The second term on the LHS of \eqref{eq:en-id-Xrad} is treated similarly, and we obtain
\begin{align*}
& \hskip-2em
	- \iint_{\calD(\tau_{1}, \tau_{2})} \chi_{R_{1}}(r)  \frac{1}{2} \rd_{v} \left( (-\dur) \phidf^{2} - 2 \phidf \rd_{u} \phidf r + \frac{1}{(-\dur)} (\rd_{v} \phidf)^{2} r^{2} \right)  \\
	\geq & \frac{1}{4} \iint_{\calD(\tau_{1}, \tau_{2})} \chi'_{R_{1}}(r) \dvr (-\dur) \phidf^{2}  
		- C (E_{w}[\phidf](\tau_{1}, \tau_{2}) + r \phidf^{2}(\tau_{1}, v_{f})).
\end{align*}
We omit the repetitive details, except to note that Lemma~\ref{lem:hardy-opt} (with $\alp = 0$) is applied both on $\Gmm^{(in)}_{\tau_{1}}$ (to get a term $r \phidf^{2}(\tau_1, v_{R_{0}(\tau_{1})})$) and $\Gmm^{(out)}_{\tau_{1}}$ (to bound the previous term by $r \phidf^{2}(\tau_1,v_f)$ and other acceptable terms). The claim \eqref{eq:en-z-ctrl-Xrad} follows.

We now estimate the contribution of the RHS of \eqref{eq:en-id-Xrad}. Clearly, we have
\begin{equation} \label{eq:en-z-ctrl-Xrad1}
	\Abs{\iint_{\calD(\tau_{1}, \tau_{2})} \chi_{R_{1}} F \left( \frac{1}{\dvr} \rd_{v} \phidf - \frac{1}{(-\dur)} \rd_{u} \phidf + 2 \frac{\phidf}{r} \right) r^{2} }
	\leq C \calE_{far}(\tau_{1}, \tau_{2}).
\end{equation}
For the last term on the RHS of \eqref{eq:en-id-Xrad}, we divide the integral into $\set{r \leq \veps^{-1} R_{1}}$ and $\set{r \geq \veps^{-1} R_{1}}$. For the first integral, by \eqref{eq:m-bnd}, \eqref{eq:e-bnd}, \eqref{eq:kpp-bnd}, \eqref{eq:mu-large-r} and \eqref{eq:gmm-large-r} we have
\begin{equation} \label{eq:en-z-ctrl-Xrad2}
	\Abs{\iint_{\calD(\tau_{1}, \tau_{2}) \cap \set{r \leq \veps^{-1} R_{1}}} \frac{(\varpi - \frac{\e^{2}}{r})}{r^{2}} \left( \frac{(-\dur)}{1-\mu} \frac{1}{\dvr} (\rd_{v} \phidf)^{2} + \frac{\dvr}{1-\mu} \frac{1}{(-\dur)} (\rd_{u} \phidf)^{2} \right) r^{2}}
	\leq C_{\veps^{-1} R_{1}} E_{w}[\phidf](\tau_{1}, \tau_{2}).
\end{equation}
In view of \eqref{eq:m-bnd}, \eqref{eq:e-bnd}, note that $\abs{r^{-2} (\varpi - \frac{\e^{2}}{r})} \leq C r^{-2}$, which decays faster than $r^{-1-\eta_{0}}$. Again by \eqref{eq:kpp-bnd} and \eqref{eq:mu-large-r}, it follows that
\begin{equation} \label{eq:en-z-ctrl-Xrad3}
\begin{aligned} 
& \hskip-2em
	\Abs{\iint_{\calD(\tau_{1}, \tau_{2}) \cap \set{r \geq \veps^{-1} R_{1}}} \frac{(\varpi - \frac{\e^{2}}{r})}{r^{2}} \left( \frac{(-\dur)}{1-\mu} \frac{1}{\dvr} (\rd_{v} \phidf)^{2} + \frac{\dvr}{1-\mu} \frac{1}{(-\dur)} (\rd_{u} \phidf)^{2} \right) r^{2}} \\
	\leq & C \frac{\veps^{1-\eta_{0}}}{R_{1}^{1-\eta_{0}}} \iint_{\calD(\tau_{1}, \tau_{2}) \cap \set{r \geq \veps^{-1} R_{1}}} \frac{1}{r^{1+\eta_{0}}} 
	\left(  \frac{1}{(-\dur)} (\rd_{u} \phidf)^{2} + \kpp^{-1} (-\dur) (\rd_{v} \phidf)^{2} \right) r^{2}.
\end{aligned}
\end{equation}
Putting together \eqref{eq:en-z-ctrl-Xrad}--\eqref{eq:en-z-ctrl-Xrad3}, we obtain 
\begin{equation*}
\iint_{\calD(\tau_{1}, \tau_{2}) \cap \set{R_{1} \leq r \leq 2 R_{1}}} r \phidf^{2} \leq (\hbox{RHS of \eqref{eq:en-z-ctrl})},
\end{equation*}
which proves a part of \eqref{eq:en-z-ctrl}. Observing that the rest of the LHS of \eqref{eq:en-z-ctrl} is obviously bounded by $C_{R_{1}} E_{w}[\phidf](\tau_{1}, \tau_{2})$, the full claim \eqref{eq:en-z-ctrl} follows.

\pfstep{Step~3: Optimal integrated local energy decay}
We upgrade the non-optimal (as $r \to \infty$) weights in the spacetime integrals in \eqref{eq:en-w-ctrl} and \eqref{eq:en-z-ctrl} to the (almost) optimal ones. We claim that there exists $R_{1} = R_{1}(\eta_{0}) > 40$ sufficiently large so that
\begin{equation} \label{eq:en-far-ctrl}
\begin{aligned}
& \hskip-2em
\iint_{\calD(\tau_{1}, \tau_{2}) \cap \set{r \geq R_{1}}} \frac{1}{r^{1+\eta_{0}}} 
		 \left( \frac{1}{(-\dur)} (\rd_{u} \phidf)^{2} + \kpp^{-1} (-\dur) (\rd_{v} \phidf)^{2} + (-\dur) \frac{\phidf^{2}}{r^{2}} \right) r^{2} \\
& \hskip-2em
+ \iint_{\calD(\tau_{1}, \tau_{2}) \cap \set{r \geq R_{1}}} \frac{1-\mu}{\dvr (-\dur)} r^{3-\eta_{0}} \left((\rd_{v} \phi)^{2} (\rd_{u} \phidf)^{2} + (\rd_{u} \phi)^{2} (\rd_{v} \phidf)^{2} \right) \\
\leq & C_{\eta_{0}} \left( E_{w}[\phidf](\tau_{1}, \tau_{2}) + r \phidf^{2}(\tau_{1}, v_{f}) + \calE(\tau_{1}, \tau_{2}) \right).
\end{aligned}
\end{equation}
We use the identity \eqref{eq:en-id-X} with $f = r^{-\eta_{0}}$, whose LHS has a good sign in the region where $r$ is sufficiently large. Thus we multiply this identity by $\chi_{R_{1}}(r)$ for a suitably large $R_{1}$, integrate by parts over $\calD(\tau_{1}, \tau_{2})$, and control the error terms using \eqref{eq:en-w-ctrl}, \eqref{eq:en-z-ctrl} and $\calE$.

In this step, we deviate from our convention and track the dependence of implicit constants on $\eta_{0}$ for logical clarity.
We first consider the contribution of the LHS of \eqref{eq:en-id-X}. We claim that for $R_{1}$ sufficiently large (depending only on $\eta_{0}$),
\begin{equation} \label{eq:en-far-ctrl-X}
\begin{aligned}
& \hskip-2em
\frac{1}{2} \iint_{\calD(\tau_{1}, \tau_{2}) \cap \set{r \geq R_{1}}} 
	\left( -\frac{1}{2} f'(r)  \left( \frac{1-\mu}{(-\dur)} \dvr (\rd_{u} \phidf)^{2} + \frac{1-\mu}{\dvr} (-\dur) (\rd_{v} \phidf)^{2} \right) r^{2} 
	+ f''(r) r \dvr (-\dur) \phidf^{2} \right) \\
& \hskip-2em
+ \frac{1}{2} \iint_{\calD(\tau_{1}, \tau_{2}) \cap \set{r \geq R_{1}}}
		f(r) \frac{1-\mu}{\dvr (-\dur)} r^{3} \left((\rd_{v} \phi)^{2} (\rd_{u} \phidf)^{2} + (\rd_{u} \phi)^{2} (\rd_{v} \phidf)^{2} \right) \\
\leq &
\iint_{\calD(\tau_{1}, \tau_{2})} \chi_{R_{1}} (\hbox{LHS of \eqref{eq:en-id-X}}).
\end{aligned}
\end{equation}
Since $f = r^{-\eta_{0}}$, we have $f' = -\eta_{0} r^{-1-\eta_{0}}$ and $f'' = \eta_{0} (1+\eta_{0}) r^{-2-\eta_{0}}$. 
Thus we have
\begin{align*}
	\Abs{f \frac{4 \varpi - 2 \frac{\e^{2}}{r}}{r^{2}} \rd_{u} \phidf \rd_{v} \phidf r^{2}}
	\leq & \Abs{\frac{2 \varpi - \frac{\e^{2}}{r}}{1-\mu}} \eta_{0}^{-1} r^{-1} \left(-\frac{1}{2} f'(r)  \left( \frac{1-\mu}{\dvr} (-\dur) (\rd_{v} \phidf)^{2} + \frac{1-\mu}{\dvr} (-\dur) (\rd_{v} \phidf)^{2} \right) r^{2} \right), \\
	\Abs{(f' r - f) \frac{2(\varpi - \frac{\e^{2}}{r})}{r^{2}} \frac{\dvr (-\dur)}{1-\mu} \phidf^{2}}
	\leq & \Abs{\frac{2(\varpi - \frac{\e^{2}}{r})}{1-\mu}} \eta_{0}^{-1} r^{-1} \left( f''(r) r \dvr (-\dur) \phidf^{2} \right).
\end{align*}
In the region $\set{r > 20}$, \eqref{eq:m-bnd}, \eqref{eq:e-bnd} and \eqref{eq:mu-large-r} imply
\begin{equation*}
\Abs{\frac{2 \varpi - \frac{\e^{2}}{r}}{(1-\mu) r^{2}}} + \Abs{\frac{2(\varpi - \frac{\e^{2}}{r})}{1-\mu}} \leq C.
\end{equation*}
Note that $\mathrm{supp}\,\chi_{R_{1}} \subseteq \set{r \geq \frac{R_{1}}{2}}$. Choosing $R_{1}$ sufficiently large depending on $\eta_{0}$ and the preceding constant, \eqref{eq:en-far-ctrl-X} follows.

We now bound the contribution of the RHS of \eqref{eq:en-id-X}. For the first term, we claim that
\begin{equation} \label{eq:en-far-ctrl-X1}
\begin{aligned}
 & \hskip-2em
\Abs{\iint_{\calD(\tau_{1}, \tau_{2})} \chi_{R_{1}}(r) \rd_{u} \left( \frac{f}{2} \frac{1-\mu}{\dvr} (\rd_{v} \phidf)^{2} r^{2} + \frac{1}{2} \rd_{v}(f r \phidf^{2}) - f' \dvr r \phidf^{2} \right)} \\
\leq & C \left( E_{w}[\phidf](\tau_{1}, \tau_{2}) + r \phidf^{2}(\tau_{1}, v_{f})
		+ \iint_{\calD(\tau_{1}, \tau_{2}) \cap \set{\frac{1}{2} R_{1} \leq r \leq 4 R_{1}}} \frac{1}{r} \left( (\rd_{v} \phidf)^{2} + (\rd_{u} \phidf)^{2} + \frac{1}{r^{2}} \phidf^{2}\right) r^{2} \right).
\end{aligned}
\end{equation}
To see this, we first perform integration by parts for the integral on the LHS of \eqref{eq:en-far-ctrl-X1}. The resulting boundary integrals can be estimated using $E_{w}[\phidf](\tau_{1}, \tau_{2})$ and Hardy's inequality (Lemma~\ref{lem:hardy-opt} with $\alp = 0$). Note that when using Lemma~\ref{lem:hardy-opt}, we need to apply \eqref{eq:hardy-v} on both the constant-$u$ boundary hypersurfaces in the future and the past, and then use \eqref{eq:hardy-u} to bound $r \phidf^{2}(\tau_{2}, v_{f})$ by $E_{w}[\phidf](\tau_{1}, \tau_{2})+r \phidf^{2}(\tau_{1}, v_{f})$. On the other hand, the resulting spacetime integral has $\chi_{R_{1}}'$ in the integrand. Since
\begin{equation*}
	\mathrm{supp} \, \chi_{R_{1}}' \subset \set{\frac{R_{1}}{2} \leq r \leq 4 R_{1}}, \quad
	\abs{\chi'_{R_{1}}} \leq \frac{10}{r},
\end{equation*}
this term can be bounded by the last term on the RHS of \eqref{eq:en-far-ctrl-X1}.

The second term on the RHS of \eqref{eq:en-id-X} can be handled similarly as the preceding case; we have
\begin{equation} \label{eq:en-far-ctrl-X2}
\begin{aligned}
 & \hskip-2em
\Abs{\iint_{\calD(\tau_{1}, \tau_{2})} \chi_{R_{1}}(r) \rd_{v} \left( \frac{f}{2} \frac{1-\mu}{(-\dur)} (\rd_{u} \phidf)^{2} r^{2} - \frac{1}{2} \rd_{u}(f r \phidf^{2}) - f' (-\dur) r \phidf^{2} \right)} \\
\leq & C \left( E_{w}[\phidf](\tau_{1}, \tau_{2}) + r \phidf^{2}(\tau_{1}, v_{f})
		+ \iint_{\calD(\tau_{1}, \tau_{2}) \cap \set{\frac{1}{2} R_{1} \leq r \leq 4 R_{1}}} \frac{1}{r} \left( (\rd_{v} \phidf)^{2} + (\rd_{u} \phidf)^{2} + \frac{1}{r^{2}} \phidf^{2}\right) r^{2} \right).
\end{aligned}
\end{equation}
We omit the proof. 

For the last term on the RHS of \eqref{eq:en-id-X}, we clearly have
\begin{equation} \label{eq:en-far-ctrl-X3}
\Abs{\iint_{\calD(\tau_{1}, \tau_{2})} \chi_{R_{1}} F f \left( - \frac{1-\mu}{\dvr} \rd_{v} \phidf + \frac{1-\mu}{(-\dur)} \rd_{u} \phidf - 2 \frac{\phidf}{r} \right) r^{2}}
\leq C \calE_{far}(\tau_{1}, \tau_{2}).
\end{equation}

Finally, we are ready to put together the estimates so far and establish \eqref{eq:en-far-ctrl}. By \eqref{eq:en-far-ctrl-X}--\eqref{eq:en-far-ctrl-X3}, where we recall the choice $f = r^{-\eta_{0}}$ and the geometric bounds \eqref{eq:kpp-bnd}--\eqref{eq:gmm-large-r}, we have
\begin{align*}
	(\hbox{LHS of \eqref{eq:en-far-ctrl}})
	\leq & C \eta_{0}^{-1} \left( E_{w}[\phidf](\tau_{1}, \tau_{2}) + r \phidf^{2}(\tau_{1}, v_{f}) + \calE(\tau_{1}, \tau_{2}) \right) \\
	& + C \eta_{0}^{-1} \iint_{\calD(\tau_{1}, \tau_{2}) \cap \set{\frac{1}{2} R_{1} \leq r \leq 4 R_{1}}} \frac{1}{r} \left( (\rd_{v} \phidf)^{2} + (\rd_{u} \phidf)^{2} + \frac{1}{r^{2}} \phidf^{2}\right) r^{2} .
\end{align*}
Using \eqref{eq:en-z-ctrl} (for $R_{1}/2$, $R_{1}$ and $2 R_{1}$), for any $\veps > 0$ we have
\begin{align*}
	(\hbox{LHS of \eqref{eq:en-far-ctrl}})
	\leq & C_{R_{1}, \veps} \eta_{0}^{-1} \left( E_{w}[\phidf](\tau_{1}, \tau_{2}) + r \phidf^{2}(\tau_{1}, v_{f}) \right) + C \eta_{0}^{-1} \calE (\tau_{1}, \tau_{2}) \\
	&  + C \eta_{0}^{-{1} }\frac{\veps^{1-\eta_{0}}}{R_{1}^{1-\eta_{0}}} \iint_{\calD(\tau_{1}, \tau_{2}) \cap \set{r \geq (2\veps)^{-1} R_{1}}} \frac{1}{r^{1+\eta_{0}}} 
	\left(  \frac{1}{(-\dur)} (\rd_{u} \phidf)^{2} + \kpp^{-1} (-\dur) (\rd_{v} \phidf)^{2} \right) r^{2}.
\end{align*}
Observe that, by taking $\veps > 0$ sufficiently small depending on $R_{1} = R_{1}(\eta_{0})$ and $\eta_{0}$, the last term on the RHS can be absorbed into the LHS. This concludes the proof of \eqref{eq:en-far-ctrl}.

From now on, we again {\bf suppress the dependence of constants on $\eta_{0}$}, as we had done before.

\pfstep{Step~4: Completion of the proof}
We now have all we need to finish the proof of \eqref{eq:en-ctrl}. In view of \eqref{eq:en-w-ctrl} and \eqref{eq:en-far-ctrl}, it remains achieve the following tasks:
\begin{itemize}
\item Bound the spacetime integral of $\phidf$ for $r\leq R_1$;
\item Bound the improved integrated local energy decay for $\rd_{v} \phidf$ near $\EH$ in \eqref{eq:iled-small-r}; and
\item Bound the interaction Morawetz term \eqref{eq:int-morawetz}.
\end{itemize}

To control the zeroth order term for $r\leq R_1$, we use Hardy's inequality (Lemma~\ref{lem:hardy-opt} with $\alp = 0$) and \eqref{eq:r-bnd} to estimate
\begin{align*}
	&\iint_{\calD(\tau_{1}, \tau_{2}) \cap \set{r \leq R}} (-\dur) r^{-(1+\eta_{0})} \phidf^{2} \\
	\leq & C\iint_{\calD(\tau_{1}, \tau_{2}) \cap \set{r \leq R}} (-\dur) \phidf^{2} \\
	\leq & C \left( \iint_{\calD(\tau_{1}, \tau_{2}) \cap \set{r \leq R}} \frac{1}{(-\dur)} (\rd_{u} \phidf)^{2} r^{2}
		+ R \int_{v_{R}(\tau_{1})}^{v_{R}(\tau_{2})} \phidf^{2}(u_{R}(v), v) \, \ud v  + R \int_{\Gamma_{\tau_{1}}^{(out)}\cap\{r \leq R\}} \phidf^{2} \right).
\end{align*}
After averaging this inequality over $R \in [R_{1}, R_{1}+1]$, we may apply \eqref{eq:en-w-ctrl} and \eqref{eq:en-far-ctrl} together with Hardy's inequality (Lemma~\ref{lem:hardy-opt} with $\alp = 0$) on $\Gamma_{\tau_{1}}^{(out)}$ to conclude that
\begin{equation*}
\iint_{\calD(\tau_{1}, \tau_{2}) \cap \set{r \leq R_{1}}} (-\dur) r^{-(1+\eta_{0})} \phidf^{2} 
\leq C (\hbox{RHS of \eqref{eq:en-ctrl}}).
\end{equation*}
Notice here that we can allow $C$ to depend on $R_{1}$, since according to Step~3, $R_{1}$ depends only on $\eta_{0}$.

Next, we bound \eqref{eq:iled-small-r} by the RHS of \eqref{eq:en-ctrl}. We use the irregular vector field $\Xi$ defined in \eqref{eq:vf-Xi}. 
For $R \leq 20$, we integrate by parts the identity \eqref{eq:en-id-Xi} with $c_{\Xi} = \frac{1}{2} c_{(\dur)}$ over $\calD(\tau_{1}, \tau_{2}) \cap \set{r \leq R}$. Then
\begin{align*}
& \hskip-2em
	\frac{1}{2} \int_{\EH(\tau_{1}, \tau_{2}) \cap \set{r \leq R}} e^{-\frac{1}{2} c_{(\dur)} (u - v - C_{\gmm_{20}})} r^{2} (\rd_{v} \phidf)^{2}
	+ \frac{c_{(\dur)}}{4} \iint_{\calD(\tau_{1}, \tau_{2}) \cap \set{r \leq R}} e^{-\frac{1}{2} c_{(\dur)} (u - v - C_{\gmm_{20}})} r^{2} (\rd_{v} \phidf)^{2} \\
	\leq & C \left( \iint_{\calD(\tau_{1}, \tau_{2}) \cap \set{r \leq R}} \abs{F \rd_{v} \phidf} r^{2} 
				+ \iint_{\calD(\tau_{1}, \tau_{2}) \cap \set{r \leq R}} \dvr r \abs{\rd_{u} \phidf} \abs{\rd_{v} \phidf} 
				+ \int_{v_{R}(\tau_{1})}^{v_{R}(\tau_{2})} (\rd_{v} \phidf)^{2}(u_{R}(v), v) \, \ud v \right) \\
	\leq & C \left( \calE_{near}(\tau_{1}, \tau_{2})
				+ E_{w}[\phidf](\tau_{1}, \tau_{2})
				+ \int_{v_{R}(\tau_{1})}^{v_{R}(\tau_{2})} (\rd_{v} \phidf)^{2}(u_{R}(v), v) \, \ud v \right).
\end{align*}
In the first inequality, we used \eqref{eq:df-large-r:v-u} to bound $e^{-\frac{1}{2} c_{(\dur)} (u - v - C_{\gmm_{20}})} \leq C$.
Note that the boundary terms on $\Gmm^{(out)}_{\tau_{1}}$ and $\Gmm^{(out)}_{\tau_{2}}$ do not arise, since $R \leq 20 < R_{0}$.
Averaging this inequality over $R \in [10, 20]$, we may bound the remaining integral by $E_{w}[\phidf](\tau_{1}, \tau_{2})$, and conclude
\begin{equation*}
	\frac{c_{(\dur)}}{4} \iint_{\calD(\tau_{1}, \tau_{2}) \cap \set{r \leq 10}} e^{-\frac{1}{2} c_{(\dur)} (u - v - C_{\gmm_{20}})} r^{2} (\rd_{v} \phidf)^{2} 
	\leq C \left( \calE_{near}(\tau_{1}, \tau_{2}) + E_{w}[\phidf](\tau_{1}, \tau_{2}) \right).
\end{equation*}
The same integral over $\calD(\tau_{1}, \tau_{2}) \cap \set{10 \leq r \leq 20}$ is bounded by $E_{w}[\phidf](\tau_{1},\tau_{2})$ thanks to \eqref{eq:df-large-r:v-u}. Then by \eqref{eq:en-w-ctrl}, we see that $\eqref{eq:iled-small-r} \leq C (\hbox{RHS of \eqref{eq:en-ctrl}})$ as desired.

Finally, to control the interaction Morawetz term, note that \eqref{eq:en-w-ctrl} and \eqref{eq:en-far-ctrl} imply
\begin{equation} \label{eq:en-ctrl:int-morawetz}
	\iint_{\calD(\tau_{1}, \tau_{2})} \frac{1-\mu}{\dvr (-\dur)}  r^{3-\eta_{0}} \left( (\rd_{v} \phi)^{2} (\rd_{u} \phidf)^{2} + (\rd_{u} \phi)^{2} (\rd_{v} \phidf)^{2} \right)  \leq C (\hbox{RHS of \eqref{eq:en-ctrl}}).
\end{equation}
We expand $\rd_{v} \phi = \rd_{v} \phibg + \rd_{v} \phidf$, $\rd_{u} \phi = \rd_{u} \phibg + \rd_{u} \phidf$ and bound the LHS from below as follows:
\begin{align*}
(\hbox{LHS of \eqref{eq:en-ctrl:int-morawetz}}) 
\geq & \iint_{\calD(\tau_{1}, \tau_{2})} \frac{1-\mu}{\dvr (-\dur)}  r^{3-\eta_{0}} (\rd_{v} \phidf)^{2} (\rd_{u} \phidf)^{2} \\
	& - C \iint_{\calD(\tau_{1}, \tau_{2})}  \frac{1-\mu}{\dvr (-\dur)}  r^{3-\eta_{0}} \left( (\rd_{v} \phibg)^{2} (\rd_{u} \phidf)^{2} + (\rd_{u} \phibg)^{2} (\rd_{v} \phidf)^{2} \right).
\end{align*}
Estimating the last term by $\calE_{nonlin} (\tau_1,\tau_2)$, we obtain the desired control of the interaction Morawetz term.

This completes the proof of \eqref{eq:en-ctrl}. \qedhere
\end{proof}

We now state the analogous control of the LHS of \eqref{eq:en-null-goal}.
\begin{lemma} \label{lem:en-null-ctrl}
For every $1 \leq \tau_{1} \leq \tau_{2}$, we have
\begin{equation} \label{eq:en-null-ctrl}
\begin{aligned}
& \hskip-2em
	\sup_{v \in [v_{R_{0}}(\tau_{1}), v_{R_0}(\tau_{2})] }\int_{\uC_{v}(\tau_{1}, \tau_{2})} \frac{1}{(-\dur)} (\rd_{u} \phidf)^{2} r^{2} \, \ud u
	+ \sup_{u \in [\tau_{1}, \tau_{2}] }\int_{C_{u}(\tau_{1}, \tau_{2})} \kpp^{-1} (\rd_{v} \phidf)^{2} r^{2} \, \ud v  \\
	\leq & C \left( \int_{\Gmm^{(in)}_{\tau_{1}}} \frac{1}{(-\dur)} (\rd_{u} \phidf)^{2} r^{2} \, \ud u
	+ \int_{\Gmm^{(out)}_{\tau_{1}}} \kpp^{-1} (\rd_{v} \phidf)^{2} r^{2} \, \ud v + r \phidf^{2} (\tau_{1}, v_{f}) + \calE(\tau_{1}, \tau_{2}) \right),
\end{aligned}
\end{equation}
where the error $\calE(\tau_{1}, \tau_{2})$ is defined as in Lemma~\ref{lem:en-ctrl}.
\end{lemma}
This lemma is a straightforward corollary of the proof of Lemma~\ref{lem:en-ctrl}; we simply repeat the proof of \eqref{eq:en-w-ctrl} above, but with $\calD(\tau_{1}, \tau_{2})$ replaced by rectangles with sides belonging to $\Gmm^{(in)}_{\tau_{1}}$, $\Gmm^{(out)}_{\tau_{1}}$, $C_{u}$ and $\uC_{v}$. We omit the details.

\subsubsection*{Control of the error terms}
We now turn to the control of the error term $\calE(\tau_{1}, \tau_{2})$, which was defined in Lemma~\ref{lem:en-ctrl}.
We begin with $\calE_{near}(\tau_{1}, \tau_{2})$.
\begin{lemma} \label{lem:en-err-near}
For any $\veps > 0$, there exists $C_{\veps} > 0$ so that for any $1 \leq \tau_{1} \leq \tau_{2}$, we have
\begin{equation} \label{eq:en-err-near}
	\calE_{near}(\tau_{1}, \tau_{2}) \leq C_{\veps} A^{2} (\log^{2} \Lmb) \tau_{1}^{-2 \omg + 1} \dlt^{2} \left( \eps + \dlt_{(t_{B})}^{2} \right)^{2} + \veps E[\phidf](\tau_{1}, \tau_{2}).
\end{equation}
\end{lemma}
\begin{proof}
By Cauchy--Schwarz, \eqref{eq:kppbg-bnd}, \eqref{eq:durbg-bnd}, \eqref{eq:kpp-bnd} and \eqref{eq:r-bnd}, we have
\begin{align*}
	\calE_{near}(\tau_{1}, \tau_{2})
	\leq \veps E[\phidf](\tau_{1}, \tau_{2}) 
		+ C_{\veps} (I_{1} + I_{2}), 
\end{align*}
where
\begin{align*}
	I_{1} = & \iint_{\calD(\tau_{1}, \tau_{2}) \cap \set{r \leq 20}} \left(\frac{\dvr}{r} - \frac{\dvrbg}{\rbg} \right)^{2} \frac{\durbg}{\dur}(-\durbg) \frac{1}{(-\durbg)^{2}} (\rd_{u} \phibg)^{2} \, \ud u \ud v , \\
	I_{2} = & \iint_{\calD(\tau_{1}, \tau_{2}) \cap \set{r \leq 20}} \left(\frac{1}{r} - \frac{1}{\rbg} \frac{\durbg}{\dur} \right)^{2} (-\dur) \frac{1}{\kppbg^{2}}(\rd_{v} \phibg)^{2} \, \ud u \ud v .
\end{align*}

We start with $I_{1}$. By \eqref{eq:dlt-adm:duphi}, we estimate
\begin{equation*}
	I_{1} \leq \dlt^{2} \iint_{\calD(\tau_{1}, \tau_{2}) \cap \set{r \leq 20}} v^{-2 \omg} \left(\frac{\dvr}{r} - \frac{\dvrbg}{\rbg} \right)^{2} \frac{\durbg}{\dur} (-\durbg)  \, \ud u \ud v.
\end{equation*}
By \eqref{eq:mubg-bnd}, \eqref{eq:durbg-bnd},  \eqref{eq:r-alp-df}, \eqref{eq:df-small-r:rdf}, \eqref{eq:df-small-r:dvrdf} and \eqref{eq:df-small-r:logdurdf}, we have
\begin{equation*}
	\left(\frac{\dvr}{r} - \frac{\dvrbg}{\rbg} \right)^{2} \frac{\durbg}{\dur} (-\durbg) 
	= \left(\frac{\dvrdf}{r} + \kppbg (1-\mubg) \left(\frac{1}{r} - \frac{1}{\rbg} \right) \right)^{2} \frac{\durbg}{\dur} (-\durbg) 
	\leq C A^{2} \log^{2} \Lmb \left( \eps + \dlt_{(t_{B})}^{2} \right)^{2} e^{- \f{c_{(\durbg)}}{2} (u - v - C_{\gmm_{20}})}.
\end{equation*}
By \eqref{eq:df-large-r:v-u} and \eqref{eq:df-large-r:v-u-R0} (which implies $\abs{v_{R_{0}}(\tau) - \tau} \leq C A (\eps + \dlt_{(t_{B})}^{2}) + C \dlt^{2}$, so that $\tau$ and $v_{R_{0}}(\tau)$ are comparable for $\tau \geq 1$), we see that
\begin{equation} \label{eq:en-err-near:pf-1}
	\iint_{\calD(\tau_{1}, \tau_{2}) \cap \set{r \leq 20}} v^{-2 \omg} e^{-c (u - v - C_{\gmm_{20}})} \, \ud u \ud v
	\leq C \int_{v_{R_{0}}(\tau_{1})}^{v_{R_{0}}(\tau_{2})} v^{-2 \omg} \, \ud v
	\leq C \tau_{1}^{-2\omg + 1}.
\end{equation}
Therefore, $I_{1}$ is bounded by the first term on the RHS of \eqref{eq:en-err-near}, which is acceptable.

Next, we treat $I_{2}$. By \eqref{eq:dlt-adm:dvphi}, we have
\begin{align*}
	I_{2} \leq \dlt^{2} \iint_{\calD(\tau_{1}, \tau_{2}) \cap \set{r \leq 20}} v^{-2 \omg} \left( \frac{1}{r} - \frac{1}{\rbg} \frac{\durbg}{\dur} \right)^{2} (-\dur) \, \ud u \ud v .
\end{align*}
By \eqref{eq:r-alp-df}, \eqref{eq:df-small-r:rdf}, \eqref{eq:df-small-r:logdurdf} and \eqref{eq:df-small-r:dur}, we have
\begin{equation*}
	\left( \frac{1}{r} - \frac{1}{\rbg} \frac{\durbg}{\dur} \right)^{2} (-\dur)
	\leq \left( \left(\frac{1}{r} - \frac{1}{\rbg} \right) + \frac{1}{\rbg} \left( 1 - \frac{\durbg}{\dur} \right)\right)^{2} (-\dur)
	\leq C A^{2} \log^{2} \Lmb \left( \eps + \dlt_{(t_{B})}^{2} \right)^{2} e^{- \frac{1}{2} c_{(\dur)} (u - v - C_{\gmm_{20}})}.
\end{equation*}
Then again by \eqref{eq:en-err-near:pf-1}, $I_{2}$ is bounded by the first term on the RHS of \eqref{eq:en-err-near}. This completes the proof. \qedhere
\end{proof}

Next, we handle $\calE_{far}(\tau_{1}, \tau_{2})$. It is for this error term that we need the integrated local energy decay with a (near) optimal $r$-weight. Furthermore, we need to use the decomposition $\dvrdf = \dvrdf_{g} + \dvrdf_{b}$ in Lemma~\ref{lem:df-large-r}, which allows us to overcome the anomalous loss of $r$-decay in \eqref{eq:en-btstrp-kpp}. 
\begin{lemma} \label{lem:en-err-far}
For any $\veps > 0$, there exists $C_{\veps} > 0$ so that for any $1 \leq \tau_{1} \leq \tau_{2}$, we have
\begin{equation} \label{eq:en-err-far}
	\calE_{far}(\tau_{1}, \tau_{2}) \leq C_{\veps} A^{2} (\log^{2} \Lmb) \tau_{1}^{-2 \omg + 1 + \eta_{0}} \dlt^{2} \left( \eps + \dlt_{(t_{B})}^{2} \right)^{2} + \veps E[\phidf](\tau_{1}, \tau_{2}).
\end{equation}
\end{lemma}
\begin{proof}
By Cauchy--Schwarz, \eqref{eq:kppbg-bnd}, \eqref{eq:durbg-bnd}, \eqref{eq:kpp-bnd} and \eqref{eq:dur-large-r}, we have
\begin{equation} \label{eq:en-err-far:pf0}
	\calE_{far}(\tau_{1}, \tau_{2}) \leq \veps E[\phidf](\tau_{1}, \tau_{2}) + C_{\veps}(I_{1} + I_{2}),
\end{equation}
where
\begin{align*}
	I_{1} = & \iint_{\calD(\tau_{1}, \tau_{2}) \cap \set{r \geq 20}} r^{1+\eta_{0}} \Abs{\frac{\dvr}{r} - \frac{\dvrbg}{\rbg}}^{2} \frac{1}{(-\durbg)^{2}}\abs{\rd_{u} \phibg}^{2} r^{2} \, \ud u \ud v , \\
	I_{2} = & \iint_{\calD(\tau_{1}, \tau_{2}) \cap \set{r \geq 20}} r^{1+\eta_{0}} \Abs{\frac{\dur}{r} - \frac{\durbg}{\rbg}}^{2} \frac{1}{\kppbg^{2}}\abs{\rd_{v} \phibg}^{2} r^{2} \, \ud u \ud v .
\end{align*}
We start with $I_{1}$, which is more involved. By \eqref{eq:dlt-adm:duphi} and \eqref{eq:df-large-r:rdf} (which ensures that $r, \rbg$ are comparable), we have
\begin{equation} \label{eq:en-err-far:pf1}
	I_{1} \leq C \Lmb^{-1} \dlt^{2} \iint_{\calD(\tau_{1}, \tau_{2}) \cap \set{r \geq 20}} u^{-2\omg} r^{1+\eta_{0}} \Abs{\frac{\dvr}{r} - \frac{\dvrbg}{\rbg}}^{2} \, \ud u \ud v .
\end{equation}
Recalling the decomposition $\dvrdf = \dvrdf_{g} + \dvrdf_{b}$ in Lemma~\ref{lem:df-large-r}, we split
\begin{equation*}
\Abs{\frac{\dvr}{r} - \frac{\dvrbg}{\rbg}}^{2}
\leq 3 \left( \frac{\dvrdf_{b}^{2}}{r^{2}} + \frac{\dvrdf_{g}^{2}}{r^{2}} + \dvrbg^{2} \left(\frac{1}{r} - \frac{1}{\rbg} \right)^{2} \right).
\end{equation*}
By \eqref{eq:mubg-bnd} and \eqref{eq:kppbg-bnd} we have $\dvrbg \leq C$. Similarly, by \eqref{eq:kpp-bnd} and \eqref{eq:mu-large-r} we have $C^{-1} \leq \rd_v r \leq C$ in $\set{r \geq 20}$. Using \eqref{eq:r-alp-df}, \eqref{eq:df-large-r:rdf} and \eqref{eq:df-large-r:dvr-g}, we may estimate
\begin{equation} \label{eq:en-err-dvr-g}
\begin{aligned}
& \hskip-2em 
	\iint_{\calD(\tau_{1}, \tau_{2}) \cap \set{r \geq 20}} u^{-2\omg} r^{1+\eta_{0}} \left( \frac{\dvrdf_{g}^{2}}{r^{2}} + \dvrbg^{2} \left(\frac{1}{r} - \frac{1}{\rbg} \right)^{2} \right) \, \ud u \ud v  \\
	\leq & C A^{2} \log^{2} \Lmb \left( \eps + \dlt_{(t_{B})}^{2} \right)^{2} 
		\iint_{\calD(\tau_{1}, \tau_{2}) \cap \set{r \geq 20}} u^{-2\omg}  r^{-3+\eta_{0}} \log^2 r(\rd_{v} r) \, \ud u \ud v \\
	\leq & C A^{2} (\log^{2} \Lmb) \tau_{1}^{-2\omg + 1} \left( \eps + \dlt_{(t_{B})}^{2} \right)^{2} .
\end{aligned}
\end{equation}
In the last inequality, we first performed the $v$-integral (which is finite and uniformly bounded in $u$ since $r^{-3+\eta_{0}}\log^2 r$ is integrable in $r$ for $r \geq 20$), and then integrated $u^{-2 \omg}$ over $u \geq \tau_{1}$. 

It remains to treat the contribution of $\dvrdf_{b}$. We claim that
\begin{equation} \label{eq:en-err-dvr-b}
	\iint_{\calD(\tau_{1}, \tau_{2}) \cap \set{r \geq 20}} u^{-2\omg} r^{1 + \eta_{0}} \frac{\dvrdf_{b}^{2}}{r^{2}} \, \ud u \ud v
	\leq C A^{2} (\log^{2} \Lmb) \tau_{1}^{-2 \omg + 1 + \eta_{0}} \left( \eps + \dlt_{(t_{B})}^{2} \right)^{2}.
\end{equation}

In the proof of \eqref{eq:en-err-dvr-b}, the bound \eqref{eq:df-large-r:dvrdf} without $r$-decay is (barely) insufficient. Instead, we need to exploit \eqref{eq:df-large-r:dvr-b}.
Using the shorthand $\uC_{v}(\tau_{1}, \tau_{2}) = \uC_{v} \cap \calD(\tau_{1}, \tau_{2}) \cap \set{r \geq 20}$ in this part of the proof, we first estimate
\begin{align*}
	(\hbox{LHS of \eqref{eq:en-err-dvr-b}})
	\leq & \sup_{\calD(\tau_{1}, \tau_{2}) \cap \set{r \geq 20}} \abs{\dvrdf_{b}} 
	\int_{v_{R_{0}}(\tau_{1})}^{v_{R_{0}}(\tau_{2})} \left( \sup_{\uC_{v}(\tau_{1}, \tau_{2})} \abs{\dvrdf_{b}} \int_{\uC_{v}(\tau_{1}, \tau_{2})} u^{-2 \omg} r^{-1+\eta_{0}} \, \ud u \right) \, \ud v.
\end{align*}
By \eqref{eq:v-u-r}, we either have $v + C_{\gmm_{20}} + 20 \leq 2 u$ or $v + C_{\gmm_{20}} + 20 \leq 8 r(u, v)$. Thus, we see that
\begin{align*}
\int_{\uC_{v}(\tau_{1}, \tau_{2})} u^{-2 \omg} r^{-1+\eta_{0}} \, \ud u
\leq & C \int_{\uC_{v}(\tau_{1}, \tau_{2})} u^{-2 \omg + \eta_{0}} r^{-1+2\eta_{0}} (v + C_{\gmm_{20}} + 20)^{-\eta_{0}} \, \ud u \\
\leq & C (v + C_{\gmm_{20}} + 20)^{-\eta_{0}} \int_{\tau_{1}}^{\infty} u^{-2 \omg + \eta_{0}} \, \ud u\leq C (v + \Lmb)^{-\eta_{0}} \tau_{1}^{-2 \omg + 1 + \eta_{0}}.
\end{align*}
In the last inequality, we used \eqref{eq:C-gmm-Lmb}. With the factor $(v + \Lmb)^{-\eta_{0}}$, we are able to apply \eqref{eq:df-large-r:dvr-b} to bound the $v$-integral above. Combined with \eqref{eq:df-large-r:dvrdf} and \eqref{eq:df-large-r:dvr-g}, which show that $\dvrdf_{b} = \dvrdf - \dvrdf_{g}$ obeys the same pointwise estimate as $\dvrdf$, we obtain \eqref{eq:en-err-dvr-b} as desired.

Finally we bound $I_{2}$, which is simpler. By \eqref{eq:dlt-adm:dvphi} and \eqref{eq:df-large-r:rdf}, we have
\begin{equation*}
	I_{2} \leq C \dlt^{2} \iint_{\calD(\tau_{1}, \tau_{2}) \cap \set{r \geq 20}} u^{-2\omg} r^{1+\eta_{0}} \Abs{\frac{\durdf}{r} + \durbg \left(\frac{1}{r} - \frac{1}{\rbg}\right)}^{2} \, \ud u \ud v .
\end{equation*}
Since $\durdf$ obeys an estimate analogous to $\dvrdf_{g}$ (see \eqref{eq:df-large-r:durdf}), we may bound the integral by exactly the same argument as \eqref{eq:en-err-dvr-g}.

In conclusion, we have proved
\begin{equation} \label{eq:en-err-far:pf2}
\begin{aligned}
& \hskip-2em
\iint_{\calD(\tau_{1}, \tau_{2}) \cap \set{r \geq 20}} r^{1+\eta_{0}} \left( 
	\Abs{\frac{\dvr}{r} - \frac{\dvrbg}{\rbg}}^{2} \frac{1}{(-\durbg)^{2}}\abs{\rd_{u} \phibg}^{2} 
	+ \Abs{\frac{\dur}{r} - \frac{\durbg}{\rbg}}^{2} \frac{1}{\kppbg^{2}}\abs{\rd_{v} \phibg}^{2} \right) r^{2} \, \ud u \ud v \\
\leq & C A^{2} (\log^{2} \Lmb) \tau_{1}^{-2 \omg + 1 + \eta_{0}} \dlt^{2} \left( \eps + \dlt_{(t_{B})}^{2} \right)^{2}.
\end{aligned}
\end{equation}
Combined with \eqref{eq:en-err-far:pf0}, this completes the proof of \eqref{eq:en-err-far}. \qedhere
\end{proof}
 
Finally, we treat $\calE_{nonlin}(\tau_{1}, \tau_{2})$. Here we crucially use the improved integrated local energy decay \eqref{eq:iled-small-r} and the interaction Morawetz terms \eqref{eq:int-morawetz} in $E[\phidf]$. With these ideas we avoid the use of pointwise decay bounds for $\rd_{u} \phidf$, $\rd_{v} \phidf$, which would have complicated the argument.

\begin{lemma} \label{lem:en-err-nonlin}
For any $\veps > 0$, there exists $C_{\veps} > 0$ so that for any $1 \leq \tau_{1} \leq \tau_{2}$, we have
\begin{equation} \label{eq:en-err-nonlin}
	\calE_{nonlin}(\tau_{1}, \tau_{2}) \leq ( C_{\veps} \dlt^{2}  + \veps ) E[\phidf](\tau_{1}, \tau_{2}).
\end{equation}
\end{lemma}
\begin{proof}
We begin by splitting
\begin{equation*}
	\calE_{nonlin}(\tau_{1}, \tau_{2}) = I_{(uv)} + I_{(vu)} + II_{(uv)} + II_{(vu)},
\end{equation*}
where
\begin{align*}
	I_{(uv)} = & \iint_{\calD(\tau_{1}, \tau_{2})} \frac{1-\mu}{\dvr(-\dur)} r^{3} (\rd_{u} \phibg)^{2}  (\rd_{v} \phidf)^{2} \, \ud u \ud v, \\
	II_{(uv)} = & \iint_{\calD(\tau_{1}, \tau_{2})} \frac{1-\mu}{\dvr(-\dur)} r^{3} \abs{\rd_{u} \phibg} \abs{\rd_{u} \phidf}  (\rd_{v} \phidf)^{2} \, \ud u \ud v,
\end{align*}
and $I_{(vu)}$, $II_{(vu)}$ are defined by switching $u$ and $v$. 

\pfstep{Step~1: Bound for $I_{(uv)}$}
Here we bound $I_{(uv)}$. In view of later steps, we establish the following slightly stronger estimate:
\begin{equation} \label{eq:en-err-Iuv}
	\iint_{\calD(\tau_{1}, \tau_{2})} \frac{1-\mu}{\dvr(-\dur)} r^{3+\eta_{0}} (\rd_{u} \phibg)^{2}  (\rd_{v} \phidf)^{2} \, \ud u \ud v
	\leq C \dlt^{2} E[\phidf](\tau_{1}, \tau_{2}).
\end{equation}
We split the integral on the LHS into $\iint_{\calD(\tau_{1}, \tau_{2}) \cap \set{r \leq 20}} (\cdots)$ and $\iint_{\calD(\tau_{1}, \tau_{2}) \cap \set{r \geq 20}} (\cdots)$. For the former, we use \eqref{eq:dlt-adm:duphi} and \eqref{eq:kpp-bnd} to first estimate
\begin{align*}
	\iint_{\calD(\tau_{1}, \tau_{2}) \cap \set{r \leq 20}} \frac{1-\mu}{\dvr(-\dur)} r^{3+\eta_{0}} (\rd_{u} \phibg)^{2}  (\rd_{v} \phidf)^{2} \, \ud u \ud v
	\leq C \dlt^{2} \iint_{\calD(\tau_{1}, \tau_{2}) \cap \set{r \leq 20}} \left(\frac{\durbg}{\dur}\right)^{2} (-\dur) (\rd_{v} \phidf)^{2} \, \ud u \ud v.
\end{align*}
Using \eqref{eq:df-small-r:logdurdf} and \eqref{eq:df-small-r:dur} (taking $\eps, \dlt$ smaller if necessary), we have the pointwise bound 
\begin{equation*}
\left( \frac{\durbg}{\dur} \right)^{2} (-\dur) \leq e^{-\frac{1}{2} c_{(\durbg)}(u - v - C_{\gmm_{20}})} \quad \hbox{ in }\set{r \leq 20}. 
\end{equation*}
Therefore, the RHS is bounded by $C \dlt^{2} E[\phidf](\tau_{1}, \tau_{2})$ thanks to \eqref{eq:iled-small-r}.

It remains to estimate the integral $\iint_{\calD(\tau_{1}, \tau_{2}) \cap \set{r \geq 20}} (\cdots)$. By \eqref{eq:dlt-adm:duphi}, \eqref{eq:durbg-bnd}, \eqref{eq:dur-large-r} and \eqref{eq:df-large-r:rdf} (which ensures that $r, \rbg$ are comparable), we have
\begin{align*}
	\iint_{\calD(\tau_{1}, \tau_{2}) \cap \set{r \geq 20}} \frac{1-\mu}{\dvr(-\dur)} r^{3+\eta_{0}} (\rd_{u} \phibg)^{2}  (\rd_{v} \phidf)^{2} \, \ud u \ud v
	\leq & C \Lmb^{-1} \dlt^{2} \iint_{\calD(\tau_{1}, \tau_{2}) \cap \set{r \geq 20}} u^{-2 \omg} \kpp^{-1} r^{1+\eta_{0}} (\rd_{v} \phidf)^{2} \, \ud u \ud v \\
	\leq & C \Lmb^{-1} \dlt^{2} E[\phidf](\tau_{1}, \tau_{2}) \int_{\tau_{1}}^{\infty} u^{-2 \omg} \, \ud u.
\end{align*}
The last line is bounded by $C \Lmb^{-1} \tau_{1}^{-2\omg+1} \dlt^{2} E[\phidf](\tau_{1}, \tau_{2}) $, which is much stronger than what we need. (Note that this proof works for all $\eta_{0} \leq 1$.)

\pfstep{Step~2: Bound for $I_{(uv)}$}
Next, we bound $I_{(vu)}$. As in the previous step, we establish the following stronger estimate:
\begin{equation} \label{eq:en-err-Ivu} 
	\iint_{\calD(\tau_{1}, \tau_{2})} \frac{1-\mu}{\dvr(-\dur)} r^{3+\eta_{0}} (\rd_{v} \phibg)^{2}  (\rd_{u} \phidf)^{2} \, \ud u \ud v
	\leq C \dlt^{2} E[\phidf](\tau_{1}, \tau_{2}).
\end{equation}
In this case, we simply use \eqref{eq:dlt-adm:dvphi}, \eqref{eq:kpp-bnd} and \eqref{eq:df-large-r:rdf} (the latter ensures that $r, \rbg$ are comparable) to bound $r^{4} (\rd_{v} \phidf)^{2} \leq C \dlt^{2}$. Then by \eqref{eq:kpp-bnd} again, we have
\begin{equation*}
	(\hbox{LHS of \eqref{eq:en-err-Ivu}})
	\leq C \dlt^{2} \iint_{\calD(\tau_{1}, \tau_{2})} r^{-1+\eta_{0}} \f{1}{(-\nu)} (\rd_{u} \phidf)^{2} r^{2} \, \ud u \ud v
	\leq C \dlt^{2} E[\phidf](\tau_{1}, \tau_{2}).
\end{equation*}

\pfstep{Step~3: Nonlinear contribution of $\rd \phidf$}
Finally, we treat the terms $II_{(uv)}$ and $II_{(vu)}$, which are nonlinear in $\rd \phidf$. Given $\veps > 0$, we use Cauchy--Schwarz to simply estimate
\begin{equation*}
	II_{(uv)} 
	\leq C_{\veps} \iint_{\calD(\tau_{1}, \tau_{2})} \frac{1-\mu}{\dvr (-\dur)} r^{3+\eta_{0}} (\rd_{u} \phibg)^{2} (\rd_{v} \phidf)^{2} \, \ud u \ud v 		+ \veps \iint_{\calD(\tau_{1}, \tau_{2})} \frac{1-\mu}{\dvr (-\dur)} r^{3-\eta_{0}} (\rd_{u} \phidf)^{2} (\rd_{v} \phidf)^{2} \, \ud u \ud v .
\end{equation*}
The last term is bounded by $\veps E[\phidf](\tau_{1}, \tau_{2})$ via \eqref{eq:int-morawetz}, and the first term is bounded by $C_{\veps} \dlt^{2} E[\phidf](\tau_{1}, \tau_{2})$ by the preceding two steps. Handling $II_{(vu)}$ similarly, the proof of the lemma is complete. \qedhere
\end{proof}

We are ready to complete the proof of Proposition~\ref{prop:energy-goal}.
\begin{proof}[Proof of Proposition~\ref{prop:energy-goal}]
We begin with \eqref{eq:en-goal}. By Lemma~\ref{lem:en-ctrl} and Lemma~\ref{lem:en-err-near}--Lemma~\ref{lem:en-err-nonlin}, we obtain
\begin{align*}
	E[\phidf](\tau_{1}, \tau_{2})
	\leq & C \left( \int_{\Gmm_{\tau_{1}}^{(in)}} \frac{1}{(-\dur)} (\rd_{u} \phidf)^{2} r^{2} \, \ud u 
			+ \int_{\Gmm_{\tau_{1}}^{(out)}} \kpp^{-1} (\rd_{v} \phidf)^{2} r^{2} \, \ud v
			+ r \phidf^{2} (\tau_{1}, v_{f}) \right) \\
		& + C_{\veps} A^{2} (\log^{2} \Lmb) \tau_{1}^{-2 \omg + 1 + \eta_{0}} \dlt^{2} \left(\eps + \dlt_{(t_{B})}^{2} \right)^{2} 
		+ C (C_{\veps} \dlt + \veps) E[\phidf](\tau_{1}, \tau_{2}).
\end{align*}
Choosing $\veps$ small enough and then taking $\dlt$ small as well accordingly, the last term on the RHS can be absorbed into the LHS. Therefore,
\begin{equation} \label{eq:en-final}
	E[\phidf](\tau_{1}, \tau_{2}) \leq C (\hbox{RHS of \eqref{eq:en-goal}}),
\end{equation}
which proves \eqref{eq:en-goal}. For \eqref{eq:en-null-goal}, we simply use Lemma~\ref{lem:en-null-ctrl} together with Lemma~\ref{lem:en-err-near}--Lemma~\ref{lem:en-err-nonlin} with $\veps = 1$ and apply \eqref{eq:en-final} that we just closed. \qedhere
\end{proof}

We record here a pointwise bound for $r \phidf^{2}$ in terms of $\eps^{2}$. 
\begin{corollary} \label{cor:zeroth}
We have
\begin{equation} \label{eq:phidf-pt}
	\sup_{\calD_{(t_B)}} r \phidf^{2} \leq C \eps^{2}.
\end{equation}
\end{corollary}
Among other things, this corollary allows us to control the last terms on the RHSs of \eqref{eq:en-goal} and \eqref{eq:en-null-goal}. It is also useful for controlling certain nonlinear error terms that arise in the proof of Proposition~\ref{prop:en-dvphi-goal} and Section~\ref{subsec:rp-weight} below.
\begin{proof}
We start by observing that
\begin{equation*}
	\sup_{C_{out}\cap\calD_{(t_B)}} r \phidf^{2} \leq C \eps^{2}.
\end{equation*}
This bound follows simply from Lemma~\ref{lem:hardy-type} (with $\alp = 0$), the definition \eqref{eq:eps-tB} of $\eps = \eps_{(t_{B})}$ (in particular, the term with $\rd_{v} (r \phidf)$) and \eqref{eq:mu-large-r}. Then by Proposition~\ref{prop:energy-goal} and \eqref{eq:hardy-u} (with $\alp = 0$), \eqref{eq:phidf-pt} follows. \qedhere
\end{proof}

\subsubsection*{Improved integrated local energy decay estimate for $\rd_{v} \phidf$}
Next, we turn to the proof of improved integrated local energy decay estimate for $\rd_{v} \phidf$. Our goal is to establish the following proposition.
\begin{proposition} \label{prop:en-dvphi-goal}
We have
\begin{equation} \label{eq:en-dvphi-goal}
	\iint_{\calD_{(t_B)} \cap \set{r \geq 20}} \frac{1}{r} \dvr^{-1}(\rd_{v} \phidf)^{2} r^{2} \, \ud u \ud v  \\
	\leq C \eps^{2}
	+ C A^{2} \log^{2} \Lmb \dlt^{2} \left( \eps + \dlt_{(t_{B})}^{2} \right)^{2} .
\end{equation}
\end{proposition}

We first establish an integrated energy estimate for $\rd_{v} (r \phidf)$, using the multiplier $Z$ in \eqref{eq:vf-Z}.
\begin{lemma} \label{lem:en-dvrphi}
For any $1 \leq \tau_{1} \leq \tau_{2}$, we have
\begin{equation}
\begin{aligned}
& \hskip-2em
	\iint_{\calD(\tau_{1}, \tau_{2}) \cap \set{r \geq 40}} \frac{1}{r} \dvr^{-1} (\rd_{v} (r \phidf))^{2} \, \ud u \ud v \\
	\leq & C  \int_{\Gmm^{(out)}_{\tau_{1}}} \kpp^{-1} (\rd_{v} (r \phidf))^{2} \log(1 + r) \, \ud v 
		+ C E[\phidf](\tau_{1}, \tau_{2}) + C r \phidf^{2}(\tau_{1}, v_{f}) \\
		& + C A^{2} (\log^{2} \Lmb) \tau_{1}^{-2 \omg + 1 + 2 \eta_{0}} \dlt^{2} \left( \eps + \dlt_{(t_{B})}^{2} \right)^{2} .
\end{aligned}
\end{equation}
\end{lemma}

\begin{proof}
As before, we omit the line and area elements in the integrals, and denote the RHS of \eqref{eq:wave4phidf} by $F$.
Let $\chi_{40}(r)$ be a smooth non-negative function such that $\chi_{40}(r) = 0$ in $\set{r \leq 20}$ and $\chi_{40}(r) = 1$ in $\set{r \geq 40}$.
We multiply the identity \eqref{eq:en-id-Z} by $\chi_{40}(r)$ and integrate by parts over $\calD(\tau_{1}, \tau_{2})$. We claim that
\begin{equation} \label{eq:en-dvrphi-Z}
\begin{aligned}
	\frac{1}{2} \iint_{\calD(\tau_{1}, \tau_{2})} \chi_{40} \frac{(-\dur)}{1+r} (\rd_{v} (r \phidf))^{2} 
	\leq & \iint_{\calD(\tau_{1}, \tau_{2})} \chi_{40} (\hbox{LHS of \eqref{eq:en-id-Z}}) 
		+ E[\phidf](\tau_{1}, \tau_{2}) + r \phidf^{2}(\tau_{1}, v_{f}) \\
		&+ C  \int_{\Gmm^{(out)}_{\tau_{1}}} \kpp^{-1} (\rd_{v} (r \phidf))^{2} \log(1 + r) .
\end{aligned}
\end{equation}
Indeed, observe that all future boundary integrals are non-negative (we use $-\rd_{u} \rd_{v} r \geq 0$ on the support of $\chi_{40}$), so they can be thrown away. The past boundary integral on $\Gmm^{(out)}_{\tau_{1}}$ is controlled by the last term on the RHS above.
On the other hand, the past boundary integral $\Gmm^{(in)}_{\tau_{1}}$ is controlled by $E[\phidf](\tau_{1}, \tau_{2})$ and Hardy's inequality (Lemma~\ref{lem:hardy-opt} with $\alp = 0$); the latter inequality necessitates the term $r \phidf^{2}(\tau_{1}, v_{f})$ above. Note that we use here $(-\rd_u \rd_v r)\leq C(-\nu)$, which is a consequence of the equation and Lemma~\ref{lem:geom-prelim}. The spacetime term involving $\chi'_{40}$ (which arises from integration by parts) can also be bounded by $E[\phidf](\tau_{1}, \tau_{2})$, since $\mathrm{supp} \, \chi'_{40} \subseteq \set{20 \leq r \leq 40}$.

Concerning the contribution of the RHS of \eqref{eq:en-id-Z}, we claim that
\begin{align} 
\iint_{\calD(\tau_{1}, \tau_{2}) \cap \set{r \geq 20}} \Abs{F \frac{r \log (1+r)}{\dvr} \rd_{v} (r \phidf)} 
\leq & E[\phidf](\tau_{1}, \tau_{2}) +  C A^{2} (\log^{2} \Lmb) \tau_{1}^{-2 \omg + 1 + 2 \eta_{0}} \dlt^{2} \left( \eps + \dlt_{(t_{B})}^{2} \right)^{2}	,		\label{eq:en-dvrphi-Z1} \\
\iint_{\calD(\tau_{1}, \tau_{2}) \cap \set{r \geq 20}} \Abs{\frac{\rd_{u} \rd_{v} r}{\dvr} \frac{\log (1+r)}{\dvr} (\rd_{v} (r \phidf))^{2}} 
\leq & C E[\phidf](\tau_{1}, \tau_{2})	,	\label{eq:en-dvrphi-Z2} \\
\iint_{\calD(\tau_{1}, \tau_{2}) \cap \set{r \geq 20}}\Abs{\rd_{v} \left(\frac{\log (1+r)}{r} \frac{\rd_{u} \rd_{v} r}{\dvr}\right) r^{2} \phidf^{2} } 
\leq & C E[\phidf](\tau_{1}, \tau_{2}) .			\label{eq:en-dvrphi-Z3}
\end{align}

We prove these claims in order. To prove \eqref{eq:en-dvrphi-Z1} we essentially proceed as in Lemma~\ref{lem:en-err-far}, making use of the room in the $r$-weight in the argument. More precisely, using \eqref{eq:kpp-bnd}, \eqref{eq:mu-large-r}, \eqref{eq:dur-large-r}, Cauchy--Schwarz and the trivial estimate $\log(1+r) \leq C r^{\eta_{0} / 2}$, we have
\begin{equation*}
	(\hbox{LHS of \eqref{eq:en-dvrphi-Z1}})
	\leq E[\phidf](\tau_{1}, \tau_{2})
	+ C \iint_{\calD(\tau_{1}, \tau_{2}) \cap \set{r \geq 20}} r^{1+2\eta_{0}} F^{2} r^{2}.
\end{equation*}
Using \eqref{eq:en-err-far:pf2} (with $\eta_0$ replaced by $2\eta_0$), the second term on the RHS can be bounded by the last term in \eqref{eq:en-dvrphi-Z1}, as desired.

Next, we prove \eqref{eq:en-dvrphi-Z2}. By \eqref{eq:m-bnd}, \eqref{eq:e-bnd}, \eqref{eq:kpp-bnd}, \eqref{eq:mu-large-r} and \eqref{eq:dur-large-r}, we have, when $r\geq 20$,
\begin{equation*}
	\Abs{\frac{\rd_{u}\rd_{v} r}{\dvr} \log (1+r)} \leq \frac{C \log (1+r)}{r^{2}}.
\end{equation*}
Since this weight decays faster than $r^{-1-\eta_{0}}$, \eqref{eq:en-dvrphi-Z2} follows thanks to the integrated local energy decay term in $E[\phidf](\tau_{1}, \tau_{2})$.

Finally, we establish \eqref{eq:en-dvrphi-Z3}. Note first that
\begin{align*}
\rd_{v} \left(\frac{\rd_{u} \rd_{v} r}{\dvr} \right)
= \rd_{v} \left( \frac{2 (\varpi - \frac{\e^{2}}{r})}{r^{2}} \gmm \right)
= \left( - \frac{4 \varpi}{r^{3}} + \frac{6 \e^{2}}{r^{4}} \right) \gmm \dvr + \gmm \kpp^{-1} (\rd_{v} \phi)^{2} + \frac{2 (\varpi - \frac{\e^{2}}{r})}{r^{2}} \gmm \dvr^{-1} r (\rd_{v} \phi)^{2}.
\end{align*}
Using the Leibniz rule, $\phi = \phibg + \phidf$ and the bounds \eqref{eq:m-bnd}, \eqref{eq:e-bnd}, \eqref{eq:kpp-bnd}, \eqref{eq:mu-large-r}, \eqref{eq:gmm-large-r}, it is straightforward to check that
\begin{equation*}
	\Abs{\rd_{v} \left(\frac{\log (1+r)}{r} \frac{\rd_{u} \rd_{v} r}{\dvr}\right) r^{2}}
	\leq \frac{C \log(1+r)}{r^{2}} + C r \log(1+r) (\rd_{v} \phibg)^{2}  + C r \log(1+r) (\rd_{v} \phidf)^{2}
\end{equation*}
in $\set{r \geq 20}$. The first term decays faster than $r^{-1-\eta_{0}}$, so its contribution is bounded by $E[\phidf](\tau_{1}, \tau_{2})$. By \eqref{eq:dlt-adm:dvphi} and \eqref{eq:kppbg-bnd}, the second term  bounded by $C \dlt^{2} r^{-3} \log(1+r)$, which again decays faster than $r^{-1-\eta_{0}}$; thus its contribution is bounded by $E[\phidf](\tau_{1}, \tau_{2})$ as well. Thus, it only remains to consider the contribution of the last term. Using \eqref{eq:phidf-pt}, we have
\begin{align*}
	\iint_{\calD(\tau_{1}, \tau_{2}) \cap \set{r \geq 20}} r \log (1+r) (\rd_{v} \phidf)^{2} \phidf^{2} 
	\leq & C \eps^{2} \iint_{\calD(\tau_{1}, \tau_{2}) \cap \set{r \geq 20}} \frac{\log (1+r)}{r^{2}} (\rd_{v} \phidf)^{2} r^{2},
\end{align*}
which can again be bounded by $E[\phidf](\tau_{1}, \tau_{2})$, as desired. \qedhere
\end{proof}

Using Lemma~\ref{lem:hardy-type}, Lemma~\ref{lem:en-dvrphi} and \eqref{eq:en-final}, we may now prove \eqref{eq:en-dvphi-goal} (and hence Proposition~\ref{prop:en-dvphi-goal}).
\begin{proof}[Proof of \eqref{eq:en-dvphi-goal}]
By Lemma~\ref{lem:hardy-type} with $\alp = - 1$, we have, for any $v_0\leq v_f$,
\begin{equation*}
	\int_{v_{0}}^{v_f} r \dvr^{-1} (\rd_{v} \phidf)^{2} (u, v) \, \ud v
	\leq C \left( \phidf^{2}(u, v_{0}) + \int_{v_{0}}^{v_f} \frac{1}{r} \dvr^{-1}(\rd_{v} (r \phidf))^{2} (u, v) \, \ud v \right).
\end{equation*}
Let $R \in [40,50]$. Define $v'_{R}(u) = v_{R}(u)$ if there exists $v_{R}(u)$ such that $r(u, v_{R}(u)) = R$ in $\calD_{(t_B)}$, and $v'_{R}(u) = 1$ otherwise\footnote{Notice that by our choice of $R_0$, $\{r=R\}$ must intersect $C^f_{(t_B)out}$ instead of $\uC^f_{(t_B)in}$.}.  Taking $v_{0} = v'_{R}(u)$ and integrating over $u\in [1,u_f]$, it follows that
\begin{equation*}
	\iint_{\calD_{(t_B)} \cap \set{r \geq R}} r \dvr^{-1} (\rd_{v} \phidf)^{2} (u, v) \, \ud v
	\leq C \left( \int_{1}^{u_{f}} \phidf^{2}(u, v'_{R}(u)) \, \ud u + \int_{\calD_{(t_B)} \cap \set{r \geq R}} \frac{1}{r} \dvr^{-1}(\rd_{v} (r \phidf))^{2} \, \ud u \ud v \right).
\end{equation*}
We bound the second term by Lemma~\ref{lem:en-dvrphi}. After averaging in $R \in [40, 50]$, we can control the first term by $E[\phidf](1,u_f)$ and obtain
\begin{equation*}
\begin{split}
	&\iint_{\calD_{(t_B)} \cap \set{r \geq 50}} r \dvr^{-1} (\rd_{v} \phidf)^{2} (u, v) \, \ud v\\
	\leq & C \left(  \int_{\Gmm^{(out)}_{1}} \kpp^{-1} (\rd_{v} (r \phidf))^{2} \log(1 + r) \, \ud v 
		+ E[\phidf](1,u_f)  
		+ r \phidf^{2}(1, v_{f}) \right)+CA^2 (\log^2\Lambda) \de^2 \left(\ep+\de_{(t_B)}^2\right)^2.
\end{split}
\end{equation*}
On the other hand, we clearly have
\begin{equation*}
	\iint_{\calD_{(t_B)} \cap \set{20 \leq r \leq 50}} r \dvr^{-1} (\rd_{v} \phidf)^{2} (u, v) \, \ud v
	\leq C E[\phidf](1,u_f).
\end{equation*}
Combining these bounds with \eqref{eq:en-final} and \eqref{eq:phidf-pt}, we obtain \eqref{eq:en-dvphi-goal}. \qedhere
\end{proof}

\subsection{Closing the bootstrap assumptions: Proof of Proposition~\ref{prop:en-btstrp}}
\label{subsec:btstrp-pf}
Here we prove Proposition~\ref{prop:en-btstrp}. This proposition closes all the bootstrap assumptions except \eqref{eq:en-btstrp-eps}, which is handled in the next subsection. 

There are three key ingredients for the proof of Proposition~\ref{prop:en-btstrp}. The first two ingredients are the geometric difference bounds and the energy-type bounds for $\phidf$, which were established in Sections~\ref{subsec:geom} and \ref{subsec:en-pf}, respectively.
The final ingredient is the following energy-type bounds for the background solution.
\begin{lemma} \label{lem:en-bg}
The following statements hold.
\begin{enumerate}
\item Let $\calE$ be a function on $\calD_{(t_{B})}$ obeying
\begin{equation*}
	\abs{\calE} \leq B,
\end{equation*}
for some $B > 0$. For any $R \geq 20$, the following bound holds:
\begin{equation} \label{eq:en-bg-dvphi}
	\sup_{u \in [1, u_{f}]} \int_{C_{u} \cap \calD_{(t_{B})} \cap \set{r \geq R}} \Abs{\calE \frac{r^{2}}{\dvr} (\rd_{v} \phibg)^{2}} \, \ud v
	+ \iint_{\calD_{(t_{B})} \cap \set{r \geq R}} \Abs{\calE \frac{r}{\dvr} (\rd_{v} \phibg)^{2} } \, \ud u \ud v
	\leq C B \dlt^{2}.
\end{equation}
\item Let $\underline{\calE}$ be a function on $\calD_{(t_{B})}$ obeying
\begin{equation*}
	\abs{\underline{\calE}} \leq \underline{B} \max \set{e^{\frac{c_{(\dur)}}{2} (u - v- C_{\gmm_{20})}}, 1},
\end{equation*}
for some $\underline{B} > 0$ and the constants $c_{(\dur)}, C_{\gmm_{20}}$ in Lemma~\ref{lem:df-small-r}. Then the following bound holds:
\begin{equation} \label{eq:en-bg-duphi}
	\sup_{v \in [1, v_{f}]} \int_{\uC_{v} \cap \calD_{(t_{B})}} 
					\Abs{ \underline{\calE} \frac{r^{2}}{\dur}(\rd_{u} \phibg)^{2}} \, \ud v
	+ \iint_{\calD_{(t_{B})}} 
					\Abs{ \underline{\calE} \frac{r^{1-\eta_{0}}}{\dur}(\rd_{u} \phibg)^{2}}  \, \ud u \ud v
	\leq C \underline{B} \dlt^{2} .
\end{equation}
\end{enumerate}
\end{lemma}

\begin{proof}
In this proof, we omit the line and area elements to simplify the notation. For \eqref{eq:en-bg-dvphi}, we first estimate
\begin{equation*}
\Abs{\calE \frac{r^{2}}{\dvr} (\rd_{v} \phibg)^{2}}
\leq B r^{2} \frac{\kppbg^{2}}{\kpp (1-\mu) } \kppbg^{-2} (\rd_{v} \phibg)^{2}
\leq C B r^{2} \kppbg^{-2} (\rd_{v} \phibg)^{2},
\end{equation*}
where we used \eqref{eq:kppbg-bnd}, \eqref{eq:kpp-bnd} and \eqref{eq:mu-large-r} for the last inequality. Then \eqref{eq:en-bg-dvphi} is a straightforward consequence of \eqref{eq:dlt-adm:dvphi}, where we remind the reader that $C^{-1} \leq r/\rbg \leq C$ by \eqref{eq:df-large-r:rdf}.
We omit the details.

The proof of \eqref{eq:en-bg-duphi} is a bit more involved. We first treat the region $\set{r \leq 20}$. 
In this case, we start with the pointwise bound
\begin{align*}
	\Abs{\underline{\calE} \frac{r^{2}}{\dur} (\rd_{u} \phibg)^{2}}
	\leq & C \underline{B} e^{\frac{c_{(\dur)}}{2} (u - v - C_{\gmm_{20}})}\frac{\durbg^{2}}{\dur^{2}} (-\dur) \frac{1}{\durbg^{2}}(\rd_{u} \phibg)^{2} \\
	\leq & C \underline{B} e^{- \frac{c_{(\dur)}}{2} (u - v - C_{\gmm_{20}})} e^{C A \log \Lmb (\eps + \dlt_{(t_{B})}^{2}) (1 + \abs{u - v - C_{\gmm_{20}}})} \frac{1}{\durbg^{2}}(\rd_{u} \phibg)^{2},
\end{align*}
where we used \eqref{eq:df-small-r:logdurdf} and \eqref{eq:df-small-r:dur}.
Taking $\eps, \dlt$ small enough so that $C A \log \Lmb (\eps + \dlt_{(t_{B})}^{2}) \leq \frac{1}{4} c_{(\dur)}$, and applying \eqref{eq:dlt-adm:duphi}, we arrive at
\begin{equation*}
\Abs{\underline{\calE} \frac{r^{2}}{\dur} (\rd_{u} \phibg)^{2}}
\leq C \underline{B} e^{- \frac{c_{(\dur)}}{4} (u - v - C_{\gmm_{20}})} v^{-2 \omg} \dlt^{2}.
\end{equation*}
Then as a consequence, we obtain \eqref{eq:en-bg-duphi} with the integrals (in order) replaced by $\int_{\uC_{v} \cap \calD_{(t_{B})} \cap \set{r \leq 20}}$ and $\iint_{\calD_{(t_{B})} \cap \set{r \leq 20}}$.

It remains to handle the region $\set{r \geq 20}$. By \eqref{eq:df-large-r:v-u}, note that $u - v - C_{\gmm_{20}} \leq C$ in this region. For the first term in \eqref{eq:en-bg-duphi}, we apply \eqref{eq:en-btstrp-m}, \eqref{eq:mu-large-r}, \eqref{eq:dur-large-r},  \eqref{eq:df-large-r:rdf} and \eqref{eq:df-large-r:durdf} to estimate
\begin{align*}
\int_{\uC_{v} \cap \calD_{(t_{B})} \cap \set{r \geq 20}} 
	\Abs{ \underline{\calE} \frac{r^{2}}{\dur}(\rd_{u} \phibg)^{2}} 
\leq C \underline{B} \int_{\uC_{v} \cap \calD_{(t_{B})} \cap \set{r \geq 20}} 
	\frac{\rbg^{2}}{(-\durbg)}(\rd_{u} \phibg)^{2} 
\leq C \underline{B} (\varpibg(1, v) - \varpibg(u_{f}, v)).
\end{align*}
In the last inequality, we used the equation for $\rd_{u} \varpibg$. By \eqref{eq:m-bnd}, the last expression is bounded by $C \underline{B} \dlt^{2}$, which is acceptable.

Finally, for the second term in \eqref{eq:en-bg-duphi}, we use \eqref{eq:dlt-adm:duphi}, \eqref{eq:dur-large-r}, \eqref{eq:df-large-r:rdf} and \eqref{eq:df-large-r:durdf} to estimate
\begin{align*}
\iint_{\calD_{(t_{B})} \cap \set{r \geq 20}}
\Abs{ \underline{\calE} \frac{r^{1-\eta_{0}}}{\dur}(\rd_{u} \phibg)^{2}} 
\leq & C \underline{B} \iint_{\calD_{(t_{B})} \cap \set{r \geq 20}} \frac{r^{1-\eta_{0}}}{\durbg^{2}} (\rd_{u} \phibg)^{2} \\
\leq & C \underline{B} \dlt^{2} \iint_{\calD_{(t_{B})} \cap \set{r \geq 20}} r^{-1-\eta_{0}} u^{-2 \omg}.
\end{align*}
Integrating in $v$ first (note that $\rd_{v} r \geq C^{-1}$ in this region) and then integrating in $u$, we see that the integral on the last line is finite. This completes the proof. \qedhere
\end{proof}

We now prove Proposition~\ref{prop:en-btstrp} in a sequence of lemmas. We start with the modified mass difference on the initial outgoing curve $C_{out} \cap \calD_{(t_{B})}$.
\begin{lemma} \label{lem:btstrp-pf:m-ini}
On $C_{out} \cap \calD_{(t_{B})}$, we have
\begin{equation} \label{eq:btstrp-pf:m-ini}
	\sup_{C_{out} \cap \calD_{(t_{B})}} \abs{\varpidf} \leq \abs{\varpidf}\restriction_{C_{out} \cap \uC_{in}} + C (\dlt + \eps) \eps + C A (\log \Lmb) \dlt^{2} \left( \eps + \dlt_{(t_{B})}^{2} \right).
\end{equation}
\end{lemma}
\begin{proof}
We use take the difference of the $\rd_{v} \varpi$ equation, which we write as follows:
\begin{align*}
	\rd_{v} (\varpi - \varpibg) 
	= \frac{1}{2} \frac{r^{2}}{\kpp} (\rd_{v} \phidf)^{2}
		+\frac{r^{2}}{\kpp} \rd_{v} \phibg \rd_{v} \phidf 
		+ \frac{1}{2} \left(\frac{r^{2}}{\kpp} - \frac{\rbg^{2}}{\kppbg} \right) (\rd_{v} \phibg)^{2}.
\end{align*}
For the contribution of each term, we claim the following estimates, which clearly imply \eqref{eq:btstrp-pf:m-ini}:
\begin{align} 
	\int_{1}^{v_{f}} \Abs{\frac{r^{2}}{\kpp} (\rd_{v} \phidf)^{2}} (1, v) \, \ud v \leq & C \eps^{2}, 		\label{eq:btstrp-pf:m-ini1} \\
	\int_{1}^{v_{f}} \Abs{\frac{r^{2}}{\kpp} \rd_{v} \phibg \rd_{v} \phidf} (1, v)  \, \ud v \leq & C \dlt \eps, \label{eq:btstrp-pf:m-ini2} \\
	\int_{1}^{v_{f}} \Abs{\left(1 - \frac{\rbg^{2}}{r^{2}} \frac{\kpp}{\kppbg} \right) \frac{r^{2}}{\kpp} (\rd_{v} \phibg)^{2}} (1, v)  \, \ud v \leq & C A (\log \Lmb) \dlt^{2} \left( \eps + \dlt_{(t_{B})}^{2} \right).		\label{eq:btstrp-pf:m-ini3}
\end{align}
Indeed, \eqref{eq:btstrp-pf:m-ini1} follows immediately from the definition of $\eps = \eps_{(t_{B})}$ in \eqref{eq:eps-tB}. For \eqref{eq:btstrp-pf:m-ini2}, we first apply Schwarz's inequality to estimate
\begin{equation*}
	(\hbox{LHS of \eqref{eq:btstrp-pf:m-ini2}})
	\leq \left( \int_{1}^{v_{f}} \frac{r^{2}}{\kpp^{2}} (\rd_{v} \phibg)^{2}  \, \ud v \right)^{1/2} 
			\left( \int_{1}^{v_{f}} \frac{r^{2}}{\kpp^{2}} (\rd_{v} \phidf)^{2}  \, \ud v \right)^{1/2}.
\end{equation*}
By \eqref{eq:mu-large-r} and \eqref{eq:en-bg-dvphi} with $\calE =1$, the first factor can be bounded by $C \dlt$. Estimating the second factor by $C \eps$ using \eqref{eq:btstrp-pf:m-ini1}, \eqref{eq:btstrp-pf:m-ini2} follows. Finally, \eqref{eq:btstrp-pf:m-ini3} is a consequence of \eqref{eq:mu-large-r} and \eqref{eq:en-bg-dvphi} with $R = R_{0}$ and
\begin{equation*}
	\calE = \Abs{1 - \frac{\rbg^{2}}{\kppbg} \frac{\kpp}{r^{2}}} \leq C A \log \Lmb \left( \eps + \dlt_{(t_{B})}^{2} \right)  \quad \hbox{ in } \calD_{(t_{B})} \cap \set{r \geq R_{0}}.
\end{equation*}
This pointwise inequality follows from \eqref{eq:kppbg-bnd}, \eqref{eq:kpp-bnd}, \eqref{eq:en-btstrp-kpp}, \eqref{eq:r-alp-df} and \eqref{eq:df-large-r:rdf}. We omit the details. \qedhere
\end{proof}

Next, we propagate the modified mass difference bound to the whole domain $\calD_{(t_{B})}$.
\begin{lemma} \label{lem:btstrp-pf:m}
In $\calD_{(t_{B})}$, we have
\begin{equation} \label{eq:btstrp-pf:m}
\sup_{\calD_{(t_{B})}} \abs{\varpidf} 
\leq \abs{\varpidf}\restriction_{C_{out} \cap \calD_{(t_{B})}} + C (\dlt + \eps) \eps + C A (\log \Lmb) \dlt^{2} \left( \eps + \dlt_{(t_{B})}^{2} \right).
\end{equation}
\end{lemma}

\begin{proof}
This time we take the difference of the $\rd_{u} \varpi$ equation, which we write as follows:
\begin{align*}
	\rd_{u} (\varpi - \varpibg) 
	= & \frac{1}{2} (1-\mu) \frac{r^{2}}{\dur} (\rd_{u} \phidf)^{2}
		+(1-\mu) \frac{r^{2}}{\dur} \rd_{u} \phibg \rd_{u} \phidf \\
		& + \frac{1}{2} \left((1-\mu) \frac{r^{2}}{\dur} - (1-\mubg) \frac{\rbg^{2}}{\durbg} \right)(\rd_{u} \phibg)^{2}.
\end{align*}
As before, it suffices establish the following estimates for the contribution of each term: 
\begin{align} 
	\sup_{v \in [1, v_{f}]} \int_{\uC_{v} \cap \calD_{(t_{B})}} \Abs{\frac{(1-\mu) r^{2}}{\dur} (\rd_{u} \phidf)^{2}} \, \ud u \leq & C \eps^{2}, 		\label{eq:btstrp-pf:m1} \\
	\sup_{v \in [1, v_{f}]} \int_{\uC_{v} \cap \calD_{(t_{B})}} \Abs{\frac{(1-\mu) r^{2}}{\dur} \rd_{u} \phibg \rd_{u} \phidf} \, \ud u \leq & C \dlt \eps, \label{eq:btstrp-pf:m2} \\
	\sup_{v \in [1, v_{f}]} \int_{\uC_{v} \cap \calD_{(t_{B})}} \Abs{\left((1-\mu) - (1-\mubg) \frac{\rbg^{2}}{r^{2}} \frac{\dur}{\durbg} \right) \frac{r^{2}}{\dur} (\rd_{u} \phibg)^{2}}  \, \ud u \leq & C A (\log \Lmb) \dlt^{2} \left( \eps + \dlt_{(t_{B})}^{2} \right).		\label{eq:btstrp-pf:m3}
\end{align}
Estimate \eqref{eq:btstrp-pf:m1} follows directly from Proposition~\ref{prop:energy-goal}, Corollary~\ref{cor:zeroth} and \eqref{eq:eps-tB}. For \eqref{eq:btstrp-pf:m2}, we apply Schwarz's inequality, \eqref{eq:m-bnd}, \eqref{eq:e-bnd}, \eqref{eq:en-bg-duphi} (with $\underline{\calE} = 1$) and \eqref{eq:btstrp-pf:m1}.
Finally, \eqref{eq:btstrp-pf:m3} follows from \eqref{eq:en-bg-duphi} with 
\begin{align*}
	\underline{\calE}
	= & \Abs{(1-\mu) - (1-\mubg) \frac{\rbg^{2}}{r^{2}} \frac{\dur}{\durbg}} \\
	\leq & \Abs{\mu - \mubg} + \Abs{(1-\mubg) \left(1 - \frac{\rbg^{2}}{r^{2}} \frac{\dur}{\durbg}\right)} \\
	\leq & \Abs{\mu - \mubg} + \Abs{(1-\mubg) \left(\frac{r^{2} - \rbg^{2}}{r^{2}}\right)} + \Abs{(1-\mubg) \frac{\rbg^{2}}{r^{2}}\left(1 - \frac{\dur}{\durbg}\right)} \\
	\leq & 
	\left\{
	\begin{array}{ll}
	C A \log \Lmb \left( \eps + \dlt_{(t_{B})}^{2} \right) (1 + \abs{u - v - C_{\gmm_{20}}}) e^{C A \log \Lmb \left( \eps + \dlt_{(t_{B})}^{2} \right) (u - v - C_{\gmm_{20}})}	& \hbox{ in } \calD_{(t_{B})} \cap \set{r \leq 20}, \\ 
	C r^{-1} A \log \Lmb \left( \eps + \dlt_{(t_{B})}^{2} \right) 
	& \hbox{ in } \calD_{(t_{B})} \cap \set{r \geq 20},
	\end{array}
	\right.
\end{align*}
where we take $\eps, \dlt$ small enough so that $(1 + \abs{u - v - C_{\gmm_{20}}}) e^{C A \log \Lmb \left( \eps + \dlt_{(t_{B})}^{2} \right) (u - v - C_{\gmm_{20}})} < C e^{\frac{1}{2} c_{(\dur)}(u - v - C_{\gmm_{20}})}$. The last inequality is a straightforward consequence of \eqref{eq:eps-tB}, \eqref{eq:en-btstrp-m}, \eqref{eq:mubg-bnd}, \eqref{eq:r-alp-df}, Lemma~\ref{lem:df-large-r} and Lemma~\ref{lem:df-small-r}; we omit the details. \qedhere
\end{proof}

We now consider the difference of $\log \kpp$.
\begin{lemma} \label{lem:btstrp-pf:kpp}
In $\calD_{(t_{B})}$, we have
\begin{align} 
	\sup_{\calD_{(t_{B})}} \abs{\log \kpp - \log \kppbg} 
	\leq & C \left( \dlt_{(t_{B})}^{2} + (\dlt + \eps) \eps + A (\log \Lmb) \dlt^{2} \left( \eps + \dlt_{(t_{B})}^{2} \right)  \right), \label{eq:btstrp-pf:kpp} \\
	\int_{1}^{v_{f}} \sup_{\uC_{v'} \cap \calD_{(t_{B})}} \abs{\log \kpp - \log \kppbg} \, \frac{\ud v'}{(v' + \Lmb)^{\eta_{0}}}
	\leq & C \left( \dlt_{(t_{B})}^{2} + (\dlt + \eps) \eps + A (\log \Lmb) \dlt^{2} \left( \eps + \dlt_{(t_{B})}^{2} \right) \right).	\label{eq:btstrp-pf:kpp-int}
\end{align}
Moreover, for every $R \geq 20$, we have
\begin{equation} 
	\int_{v_{R, i}}^{v_{R, f}} \Abs{\log \kpp - \log \kppbg}(u_{R}(v), v) \, \ud v
	\leq C \dlt_{(t_{B})}^{2} + C R^{\eta_{0}} \left( (\dlt + \eps) \eps + A (\log \Lmb) \dlt^{2} \left( \eps + \dlt_{(t_{B})}^{2} \right) \right), \label{eq:btstrp-pf:kpp-int-R}
\end{equation}
where $v_{R, i}$ and $v_{R, f}$ are the $v$-coordinates of the points at the intersection $\gmm_{R} \cap (C_{out} \cup \uC_{in})$ and $\gmm_{R} \cap (C^{f}_{out} \cup \uC^{f}_{in})$, respectively.
\end{lemma}

\begin{proof}
We take the difference of the equation for $\rd_{u} \log \kpp$ and integrate from $C_{(t_{B}) out}^{f}$. The difference equation can be written as follows:
\begin{equation} \label{eq:dulogkppdf}
	\rd_{u} (\log \kpp - \log \kppbg)
	= \frac{r}{\dur} (\rd_{u} \phidf)^{2} + 2 \frac{r}{\dur} \rd_{u} \phibg \rd_{u} \phidf 
	+ \left(1 - \frac{\rbg}{r} \frac{\dur}{\durbg} \right) \frac{r}{\dur} (\rd_{u} \phibg)^{2}.
\end{equation} 
We claim that:
\begin{align}
	\sup_{v \in [1, v_{f}]} \int_{\uC_{v} \cap \calD_{(t_{B})}} \Abs{\frac{r}{\dur} (\rd_{u} \phidf)^{2}}  \, \ud u
	+ \iint_{\calD_{(t_{B})}} \Abs{\frac{r^{1-\eta_{0}}}{\dur} (\rd_{u} \phidf)^{2} } \, \ud u \ud v
	\leq & C \left(\eps+ \dlt_{(t_{B})}^{2} \right)^{2}, 	\label{eq:btstrp-pf:kpp1} \\
	\sup_{v \in [1, v_{f}]} \int_{\uC_{v} \cap \calD_{(t_{B})}} \Abs{\frac{r}{\dur} \rd_{u} \phibg \rd_{u} \phidf} \, \ud u
	+ \iint_{\calD_{(t_{B})}} \Abs{\frac{r^{1-\eta_{0}}}{\dur} \rd_{u} \phibg \rd_{u} \phidf } \, \ud u \ud v
	\leq & C \dlt \left(\eps+ \dlt_{(t_{B})}^{2} \right), 		\label{eq:btstrp-pf:kpp2} \\
	\sup_{v \in [1, v_{f}]} \int_{\uC_{v} \cap \calD_{(t_{B})}} 
					\Abs{\underline{\calE} \frac{r}{\dur}(\rd_{u} \phibg)^{2}} \, \ud u
	+ \iint_{\calD_{(t_{B})}} 
					\Abs{\underline{\calE} \frac{r^{1-\eta_{0}}}{\dur}(\rd_{u} \phibg)^{2}}  \, \ud u \ud v
	\leq & C A (\log \Lmb) \dlt^{2} \left( \eps + \dlt_{(t_{B})}^{2} \right),		\label{eq:btstrp-pf:kpp3}
\end{align}
where
\begin{equation*}
	\underline{\calE} = \Abs{1 - \frac{\rbg}{r} \frac{\dur}{\durbg}}.
\end{equation*}

We first show how to establish the lemma once \eqref{eq:btstrp-pf:kpp1}--\eqref{eq:btstrp-pf:kpp3} are proved. The pointwise estimate \eqref{eq:btstrp-pf:kpp} is proved by integrating \eqref{eq:dulogkppdf} on $u' \in [u, u_{f}]$, bounding the boundary term at $u_{f}$ by $\dlt_{(t_{B})}$ using \eqref{eq:dlt-tB}, and using the boundary integral terms in \eqref{eq:btstrp-pf:kpp1}--\eqref{eq:btstrp-pf:kpp3} to control the integral. 

For \eqref{eq:btstrp-pf:kpp-int}, we start by writing
\begin{align*}
\int_{1}^{v_{f}} \sup_{\uC_{v} \cap \calD_{(t_{B})}} \Abs{\log \kpp - \log \kppbg} \, \frac{\ud v}{(v + \Lmb)^{\eta_{0}}}
\leq	& \int_{1}^{v_{f}} \int_{1}^{u_{f}} \Abs{\rd_{u} (\log \kpp - \log \kppbg)} (u, v) \, \ud u \, \frac{\ud v}{(v + \Lmb)^{\eta_{0}}} \\
	&  + \int_{1}^{v_{f}} \Abs{\log \kpp - \log \kppbg}(u_{f}, v) \, \frac{\ud v}{(v + \Lmb)^{\eta_{0}}} \\
\leq	& \int_{1}^{v_{f}} \int_{1}^{u_{f}} r^{-\eta_{0}} \Abs{\rd_{u} (\log \kpp - \log \kppbg)} (u, v) \, \ud u \ud v \\
	&  + \int_{1}^{v_{f}} r^{-\eta_{0}} \Abs{\log \kppbg}(u_{f}, v) \, \ud v.
\end{align*}
In the last inequality, we used \eqref{eq:v-u-r} and \eqref{eq:C-gmm-Lmb}, which implies $v + \Lmb \geq C^{-1} r(1, v)$, the monotonicity $r(1, v) \geq r(u, v)$, and also $\log\kpp=1$ on $C_{out}^f$ by definition. In the last two lines, the first term is bounded using \eqref{eq:dulogkppdf} and the spacetime integral terms in \eqref{eq:btstrp-pf:kpp1}--\eqref{eq:btstrp-pf:kpp3}, whereas the second term is bounded by $\dlt_{(t_{B})}^{2}$ thanks to \eqref{eq:dlt-tB}.

Lastly, for \eqref{eq:btstrp-pf:kpp-int-R}, we first estimate
\begin{align*}
\int_{v_{R, i}}^{v_{R, f}} \Abs{\log \kpp - \log \kppbg}(u_{R}(v), v) \, \ud v
\leq & \int_{v_{R_{i}}}^{v_{R, f}} \int_{u_{R}(v)}^{u_{f}} \Abs{\rd_{u} (\log \kpp - \log \kppbg)}(u, v) \, \ud u \ud v \\
& + \int_{v_{R_{i}}}^{v_{R, f}} \Abs{\log \kpp - \log \kppbg}(u_{f}, v)  \ud v \\
\leq & R^{\eta_{0}} \iint_{\calD_{(t_{B})} \cap \set{r \leq R}} r^{-\eta_{0}} \Abs{\rd_{u} (\log \kpp - \log \kppbg)} \, \ud u \ud v \\
& + \int_{C_{out}^{f}} \Abs{\log \kppbg}  \ud v.
\end{align*}
In the last expression, the first term can be bounded in the same way as the proof of \eqref{eq:btstrp-pf:kpp-int} using \eqref{eq:dulogkppdf}--\eqref{eq:btstrp-pf:kpp3}, and the last term by \eqref{eq:dlt-tB}.

It remains to verify the claims \eqref{eq:btstrp-pf:kpp1}--\eqref{eq:btstrp-pf:kpp3}. As before, \eqref{eq:btstrp-pf:kpp1} follows from Proposition~\ref{prop:energy-goal} and \eqref{eq:phidf-pt}, and \eqref{eq:btstrp-pf:kpp2} can be proved using Schwarz's inequality, \eqref{eq:en-bg-duphi} with $\underline{\calE} = 1$ and \eqref{eq:btstrp-pf:kpp1}. Finally, \eqref{eq:btstrp-pf:kpp3} follows from \eqref{eq:en-bg-duphi} and the pointwise estimate
\begin{equation*}
	\Abs{1 - \frac{\rbg}{r} \frac{\dur}{\durbg}} \leq C A \log \Lmb \left( \eps + \dlt_{(t_{B})}^{2} \right)
\end{equation*}
which is prove exactly as in the proof of \eqref{eq:btstrp-pf:m3} above. \qedhere
\end{proof}

Next we consider the difference of $\log(-\gamma)$, which is handled similarly as $\log \kpp$.
\begin{lemma} \label{lem:btstrp-pf:gmm}
We have
\begin{equation}  \label{eq:btstrp-pf:gmm}
	\abs{\log (-\gmm) - \log (-\gmmbg)} \leq C r^{-1} \left( \dlt_{(t_{B})}^{2} + (\dlt + \eps) \eps + A (\log \Lmb) \dlt^{2} \left( \eps + \dlt_{(t_{B})}^{2} \right) \right)  \quad \hbox{ in } \calD_{(t_{B})}.
\end{equation}
Moreover, for every $R \geq 20$, we have
\begin{equation} \label{eq:btstrp-pf:gmm-int}
	\int_{u_{R, i}}^{u_{R, f}} \Abs{\log \gmm - \log \gmmbg}(u, v_{R}(u)) \, \ud u
	\leq C \left( \dlt_{(t_{B})}^{2} + (\dlt + \eps) \eps + A (\log \Lmb) \dlt^{2} \left( \eps + \dlt_{(t_{B})}^{2} \right) \right),
\end{equation}
where $u_{R, i}$ and $u_{R, f}$ are the $u$-coordinates of the points at the intersection $\gmm_{R} \cap (C_{out} \cup \uC_{in})$ and $\gmm_{R} \cap (C^{f}_{out} \cup \uC^{f}_{in})$, respectively.
\end{lemma}
\begin{proof}
We take the difference of the $\rd_{v} \log (-\gmm)$ equation, which we write as
\begin{equation*}
	\rd_{v} (\log (-\gmm) - \log (-\gmmbg)) 
	= \frac{r}{\dvr} (\rd_{v} \phidf)^{2} + 2 \frac{r}{\dvr} \rd_{v} \phibg \rd_{v} \phidf + \left(1 - \frac{\rbg}{r} \frac{\dvr}{\dvrbg} \right) \frac{r}{\dvr} (\rd_{v} \phibg)^{2}.
\end{equation*}
For any $R \geq 20$, we claim that
\begin{align}
	\sup_{u \in [1, u_{f}]} R \int_{C_{u} \cap \calD_{(t_{B})} \cap \set{r \geq R}} \Abs{\frac{r}{\dvr} (\rd_{v} \phidf)^{2}}  \, \ud v
	+ \iint_{\calD_{(t_{B})} \cap \set{r \geq 20}} \Abs{\frac{r}{\dvr} (\rd_{v} \phidf)^{2} } \, \ud u \ud v
	\leq & C \left( \eps + \dlt_{(t_{B})}^{2} \right)^{2} ,	 \label{eq:btstrp-pf:gmm1}\\
	\sup_{u \in [1, u_{f}]} R \int_{C_{u} \cap \calD_{(t_{B})} \cap \set{r \geq R}} \Abs{\frac{r}{\dvr} \rd_{v} \phibg \rd_{v} \phidf}  \, \ud v
	+ \iint_{\calD_{(t_{B})} \cap \set{r \geq 20}} \Abs{\frac{r}{\dvr} \rd_{v} \phibg \rd_{v} \phidf } \, \ud u \ud v
	\leq & C \dlt \left( \eps + \dlt_{(t_{B})}^{2} \right) ,		 \label{eq:btstrp-pf:gmm2} \\
	\sup_{u \in [1, u_{f}]} R \int_{C_{u} \cap \calD_{(t_{B})} \cap \set{r \geq R}} 
					\Abs{\underline{\calE} \frac{r}{\dvr}(\rd_{v} \phibg)^{2}} \, \ud v
	+ \iint_{\calD_{(t_{B})} \cap \set{r \geq 20}} 
					\Abs{\underline{\calE} \frac{r}{\dvr}(\rd_{v} \phibg)^{2}} \, \ud u \ud v
	\leq & C A (\log \Lmb) \dlt^{2} \left( \eps + \dlt_{(t_{B})}^{2} \right), \label{eq:btstrp-pf:gmm3}
\end{align}
where
\begin{equation*}
	\underline{\calE} = \Abs{1 - \frac{\rbg}{r} \frac{\dvr}{\dvrbg}}.
\end{equation*}
As in the proof of Lemma~\ref{lem:btstrp-pf:kpp}, the desired bounds \eqref{eq:btstrp-pf:gmm} and \eqref{eq:btstrp-pf:gmm-int} follow from the claimed estimates. 

We now prove the claimed estimates. The estimate \eqref{eq:btstrp-pf:gmm1} follows from Proposition~\ref{prop:en-dvphi-goal}. Next, \eqref{eq:btstrp-pf:gmm2} can be proven by first applying Schwarz inequality and then using \eqref{eq:btstrp-pf:gmm1} and Lemma~\ref{lem:en-bg}. Finally, \eqref{eq:btstrp-pf:gmm3} follows from Lemma~\ref{lem:en-bg} and the pointwise bound
\begin{equation*}
	\Abs{1 - \frac{\rbg}{r} \frac{\dvr}{\dvrbg}}
	\leq  C A \log \Lmb \left( \eps + \dlt_{(t_{B})}^{2} \right),
\end{equation*}
in $\calD_{(t_{B})} \cap \set{r \geq 20}$, which in turn follows from \eqref{eq:kppbg-bnd}, \eqref{eq:mubg-bnd}, \eqref{eq:kpp-bnd}, \eqref{eq:mu-large-r}, \eqref{eq:r-alp-df}, \eqref{eq:df-large-r:rdf} and \eqref{eq:df-large-r:dvrdf}. \qedhere
\end{proof}

We are now ready to complete the proof of Proposition~\ref{prop:en-btstrp}.
\begin{proof}[Proof of Proposition~\ref{prop:en-btstrp}]
The geometric bounds \eqref{eq:en-btstrp:rdf}--\eqref{eq:en-btstrp:logdurdf} follow from Lemma~\ref{lem:df-large-r} and Lemma~\ref{lem:df-small-r}. On the other hand, 
the improved bootstrap assumptions \eqref{eq:en-btstrp-m:imp}--\eqref{eq:en-btstrp-R-Lmb:imp} follow from Lemma~\ref{lem:btstrp-pf:m-ini}--Lemma~\ref{lem:btstrp-pf:gmm} and Corollary~\ref{cor:rbg-uf-vf-Lmb}, once we fix a sufficiently large number $A > 0$, a small enough number $\dlt_{0} > 0$ (depending on $A, \abs{\ebg}$, $\eta_{0}$, $\omg$), and assume that $\eps_{0} > 0$ is sufficiently small (depending on $A, \dlt_0, \abs{\ebg}, \eta_{0}, \omg, \Lmb$). In particular, to handle the contribution of $\abs{\log \kpp - \log \kppbg}$ in \eqref{eq:en-btstrp-int:imp} and \eqref{eq:en-btstrp-int-20:imp}, we use \eqref{eq:btstrp-pf:kpp-int-R} along with the observation that
\begin{equation*}
	\int_{v_{R, i}}^{v_{R, f}} \abs{\log \kpp - \log \kppbg}(u_{R}(v), v) \, \ud v
	\geq \frac{1}{2} \int_{\gmm_{R} \cap \calD_{(t_{B})} \cap  \set{v \leq v_{f}}} \abs{\log \kpp - \log \kppbg} (t)  \, \ud t
\end{equation*}
for any $R \geq 20$. Indeed, note that $\dot{v}(t) = T(v)(t) = \kpp^{-1}(\gmm_{R}(t))$ is bounded below by $\frac{1}{2}$ thanks to \eqref{eq:kpp-bnd}. Similarly, for the contribution of $\abs{\log \gmm - \log \gmmbg}$ in \eqref{eq:en-btstrp-int:imp} and \eqref{eq:en-btstrp-int-20:imp}, we use the lower bound $\dot{u}(t) = T(u)(t) = - \gmm^{-1}(\gmm_{R}(t)) \geq 1$, which follows from \eqref{eq:gmm-large-r}, to replace the line element $\ud u$ in \eqref{eq:btstrp-pf:gmm-int} by $\ud t$. \qedhere
\end{proof}

\subsection{Estimates for the coordinate transformation: Proof of Proposition~\ref{prop:g-btstrp}} \label{subsec:gauge}
Here we establish Proposition~\ref{prop:g-btstrp}, which relates the initial difference $\eps_{(t_{B})}^{2}$ in the coordinates $(u_{(t_{B})}, v_{(t_{B})})$ and $(\ubg, \vbg)$ with $\eps_{0}^{2}$ in the original $(U, V)$ coordinates.
This proposition improves the remaining bootstrap assumption \eqref{eq:en-btstrp-eps}.

\begin{proof}[Proof of Proposition~\ref{prop:g-btstrp}]
In this proof, to ease the notation, we will omit the subscript $(t_{B})$ in $u_{(t_{B})}$, $v_{(t_{B})}$ and $\eps_{(t_{B})}$. We also remind the reader that we work under the normalization $\rbg_{\EH} = 1$.

\pfstep{Step~1} We consider the data on the incoming initial null curve $\uC_{in}$, on which $V = v = \vbg = 1$. Our goal in this step is to prove that, when $\eps_{0}, \dlt_{0}$ are sufficiently small, 
\begin{equation} \label{eq:g-btstrp:du}
	\int_{\uC_{in} \cap \calD_{(t_{B})}} (-\dur)^{-1} (\rd_{u} \phidf)^{2} r^{2} \, \ud u
	\leq C \eps_{0}^{2} + C A^{2} \dlt_{0}^{2} \left( \eps + \dlt_{(t_{B})}^{2} \right)^{2}.
\end{equation}

We claim that 
\begin{equation} \label{eq:g-btstrp:du-key}
	\Abs{\frac{1}{\dur} \rd_{u} \phi - \frac{1}{\durbg} \rd_{u} \phibg}(u, 1) \leq C \Lmb^{-3/2} \left( (1 + (\log \Lmb) \dlt) \eps_{0} + A (\log^{2} \Lmb) \dlt \left( \eps + \dlt_{(t_{B})}^{2} \right) \right).
\end{equation}
This will be proved in Step~2; here, we assume \eqref{eq:g-btstrp:du-key} and prove \eqref{eq:g-btstrp:du}. We begin by writing
\begin{align*}
	\int_{\uC_{in} \cap \calD_{(t_{B})}} (- \dur)^{-1} (\rd_{u} \phidf)^{2} r^{2} \, \ud u
	& =  \int_{\uC_{in} \cap \calD_{(t_{B})}} \left( \frac{1}{\dur} \rd_{u} \phidf \right)^{2} r^{2} (- \dur) \, \ud u \\
	& \leq  2\int_{\uC_{in} \cap \calD_{(t_{B})}} \Abs{\frac{1}{\dur} \rd_{u} \phi - \frac{1}{\durbg} \rd_{u} \phibg}^{2} r^{2}(- \dur) \, \ud u \\ 
	& \phantom{\leq}
		 + 2\int_{\uC_{in} \cap \calD_{(t_{B})}} \Abs{ 1 - \frac{\durbg}{\dur}}^{2} \Abs{\frac{1}{\durbg} \rd_{u} \phibg}^{2} r^{2}(- \dur) \, \ud u \\
	& = :  I + II.
\end{align*}

The first term is bounded using \eqref{eq:g-btstrp:du-key} and $\de \Lambda^{100\eta_0}\leq \de_0$ as 
\begin{align*}
	\abs{I} 
	\leq & C \left( \Lmb^{-3/2} (1 + (\log \Lmb) \dlt) \eps_{0} + A \Lmb^{-3/2} (\log \Lmb) \dlt \left( \eps + \dlt_{(t_{B})}^{2} \right) \right)^{2} \int_{\uC_{in} \cap \calD_{(t_{B})}} r^{2} (- \rd_{u} r) \, \ud u \\
	\leq & C (R_{0} / \Lmb)^{3} \left( (1 + \dlt_{0}) \eps_{0}^{2} + A^{2} \dlt_{0}^{2} \left( \eps + \dlt_{(t_{B})}^{2} \right)^{2} \right),
\end{align*}
which is acceptable by \eqref{eq:R0-Lmb}.

To treat the second term, we split the integral into $\int_{\uC_{in} \cap \calD_{(t_{B})} \cap \set{r \leq 20}} (\cdots)$ and $\int_{\uC_{in} \cap \calD_{(t_{B})} \cap \set{r \geq 20}} (\cdots)$. For the former integral, we first use \eqref{eq:dlt-adm:duphi} to estimate
\begin{align*}
	\Abs{\int_{\uC_{in} \cap \calD_{(t_{B})} \cap \set{r \leq 20}} \Abs{ 1 - \frac{\durbg}{\dur}}^{2} \Abs{\frac{1}{\durbg} \rd_{u} \phibg}^{2} r^{2}(- \dur) \, \ud u }
	\leq & C \dlt^{2} \int_{1 + C_{\gmm_{20}}}^{u_{(t_{B}) f}} \Abs{1 - \frac{\durbg}{\dur}}^{2} (- \dur) \, \ud u.
\end{align*}
By the simple inequality $\abs{1 - e^{\vtht}} \leq \vtht e^{\vtht}$ $(\vtht > 0)$, \eqref{eq:df-small-r:logdurdf} and \eqref{eq:df-small-r:dur}, we have
\begin{align} 
	\int_{1 + C_{\gmm_{20}}}^{u_{(t_{B}) f}} \Abs{1 - \frac{\durbg}{\dur}}^{2} (- \dur) \, \ud u
	\leq & \int_{1 + C_{\gmm_{20}}}^{u_{(t_{B}) f}} \abs{\log \dur - \log \durbg}^{2} e^{2 \abs{\log \durbg - \log \dur}} (- \dur) \, \ud u \notag \\
	\leq & C A^{2} \log^{2} \Lmb \left( \eps + \dlt_{(t_{B})}^{2} \right)^{2} \int_{1 + C_{\gmm_{20}}}^{\infty} (1 + \abs{u - 1 - C_{\gmm_{20}}})^{2} e^{\frac{c_{(\dur)}}{2} (u - 1 - C_{\gmm_{20}})} \, \ud u	\notag \\
	\leq & C A^{2} \log^{2} \Lmb \left( \eps + \dlt_{(t_{B})}^{2} \right)^{2}, \label{eq:g-btstrp:du-1}
\end{align}
which is acceptable since $\de\log\Lmb\leq C\de_0$. On the other hand, for $\int_{\uC_{in} \cap \calD_{(t_{B})} \cap \set{r \geq 20}}(\cdots)$ we use \eqref{eq:dlt-adm:duphi}, \eqref{eq:dur-large-r} and \eqref{eq:df-large-r:durdf} and $\de\log\Lmb\leq C\de_0$ to bound
\begin{align*}
	\Abs{\int_{\uC_{in} \cap \calD_{(t_{B})} \cap \set{r \geq 20}} \Abs{ 1 - \frac{\durbg}{\dur}}^{2} \Abs{\frac{1}{\durbg} \rd_{u} \phibg}^{2} r^{2}(- \dur) \, \ud u }
	\leq & C A \Lmb^{-1} (\log^{2} \Lmb) \dlt^{2} \left( \eps + \dlt_{(t_{B})}^{2} \right)^{2} 
			\int_{1}^{1 + C_{\gmm_{20}}} ( - \rd_{u} r) \, \ud u \\
	\leq & C A (R_{0} / \Lmb) \dlt_{0}^{2} \left( \eps + \dlt_{(t_{B})}^{2} \right)^{2},
\end{align*}
which is acceptable by \eqref{eq:R0-Lmb}.

\pfstep{Step~2}
We now turn to the proof of \eqref{eq:g-btstrp:du-key}. It is convenient introduce a new coordinate $\bfU$ so that $\rd_{\bfU} \rbg = -1$ on $\uC_{in}$. More precisely, we set
\begin{equation*}
	\frac{\ud \bfU}{\ud U} (U) = - \rd_{U} \rbg(U, 1), \quad \bfU(1) = 1.
\end{equation*}
By the chain rule, we have
\begin{equation*}
	(\rd_{\bfU} r) (\bfU(U), 1) = \left( \frac{\ud \bfU}{\ud U} (U) \right)^{-1} \rd_{U} r(U, 1) = - \frac{\rd_{U} r}{\rd_{U} \rbg} (U, 1).
\end{equation*}
Let $U(u)$ be the inverse of the function $U \mapsto u(U)$. Abusing the notation a bit, we denote by $\Ubg(u)$ the inverse of the function $U \mapsto \overline{u}(U)$. We define
\begin{equation*}
	\bfU(u) := \bfU(U(u)), \quad 
	\bfUbg(u) := \bfU(\Ubg(u)).
\end{equation*}
Given some functions $f, \overline{f}$, let $g (u) := f(\bfU(u))$ and $\overline{g}(u) := \overline{f}(\bfUbg(u))$. Then
\begin{align*}
	(g - \overline{g})(u) 
= & 	f(\bfU(u)) - \overline{f}(\bfUbg(u)) \\
= &	(f - \overline{f})(\bfU(u)) + (\bfU - \bfUbg)(u) \int_{0}^{1} \rd_{\bfU} \overline{f}(\sgm \bfU(u) + (1-\sgm) \bfUbg(u)) \, \ud \sgm.
\end{align*}
We apply the above formula with $f (\bfU)= \rd_{\bfU} \phi(\bfU, 1)$ and $\overline{f} (\bfU)= \rd_{\bfU} \phibg(\bfU, 1)$. Using coordinate invariance of $\frac{1}{\rd_{u} r} \rd_{u}$, we see that 
\begin{align*}
g (u) =& f(\bfU(u)) = \frac{1}{\rd_{u} r} \rd_{u} \phi (u) = \frac{1}{\rd_{U} r} \rd_{U} \phi(U(u)), \\
\overline{g} (u) =& \frac{1}{\rd_{u} \rbg} \rd_{u} \phibg (u), \quad
\overline{f}(\bfU(u)) = \frac{1}{\rd_{U} \rbg} \rd_{U} \phibg(U(u)).
\end{align*}
Moreover, we have
\begin{equation*}
	\rd_{\bfU} \overline{f} (\bfU') = \rd_{\bfU}^{2} \phibg (\bfU') = \left( \frac{1}{-\rd_{U} \rbg} \rd_{U} \right)^{2} \phibg (\bfU').
\end{equation*}
Therefore,
\begin{align*}
	\Abs{ \frac{1}{\rd_{u} r} \rd_{u} \phi - \frac{1}{\rd_{u} \rbg} \rd_{u} \phibg }(u,1)
	\leq &\Abs{ \frac{1}{\rd_{U} r} \rd_{U} \phi - \frac{1}{\rd_{U} \rbg} \rd_{U} \phibg } (U(u),1) \\
		&+ \abs{\bfU - \bfUbg}(u) \, \sup_{\bfU' \hbox{ between } \bfU(u), \bfUbg(u)} \Abs{\left( \frac{1}{-\rd_{U} \rbg} \rd_{U} \right)^{2} \phibg}(\bfU',1).
\end{align*}
For $1 \leq u \leq u_{(t_{B}) f}$, the first term is at most of size $\Lmb^{-3/2} \eps_{0}$ by the hypothesis \eqref{eq:L-st-ch:ini-duphi}. By the definition of $(\omg, \dlt, \Lmb)$-admissibility, the second term is bounded by
\begin{equation*}
	\Lmb^{-3/2} \dlt \sup_{u : 1 \leq u \leq u_{(t_{B}) f}} \abs{\bfU - \bfUbg}(u),
\end{equation*}
as long as we can guarantee that $(\bfU(u), 1) \in \uC_{in}$.
Therefore, to conclude the proof of \eqref{eq:g-btstrp:du-key}, it only remains to prove that 
\begin{equation} \label{eq:g-btstrp:bfU}
\abs{\bfU - \bfUbg}(u) \leq C A \log^2 \Lmb \left( \eps + \dlt_{(t_{B})}^{2} \right)  + C \eps_{0} .
\end{equation}
Indeed, this estimate also guarantees that $(\bfU(u), 1) \in \uC_{in}$. To see this, note that by Definition~\ref{def:dlt-adm}, it suffices to show that $\rbg(\bfU(u),1)\geq (1-2\de) \rbg_{\EH}$. On the one hand, by Definition~\ref{def:dlt-adm} and the monotonicity of $\rbg$, $\rbg\geq (1-\de)\rbg_{\EH}$ in the exterior region. Hence, $\rbg(\overline{\bfU}(u),1)\geq (1-\de)\rbg_{\EH}$. On the other hand, by choosing $\dlt_0$ small and then $\eps_{0}$ small (depending on $\dlt_0$) and then using \eqref{eq:en-btstrp-eps} and Proposition~\ref{prop:dlt-tB-bnd}, the RHS of \eqref{eq:g-btstrp:bfU} is $\leq \f{\de}{2}$. Hence, since $\rd_{\bfU} \rbg=-1$ and $(\overline{\bfU}(u),1)\in \uC_{in} \cap \calD$, we have 
$$\rbg(\bfU(u),1)= \rbg(\overline{\bfU}(u),1)-(\bfU(u)-\overline{\bfU}(u)) \geq (1-\de-\f{\de}{2})\rbg_{\EH} \geq (1-2\de)\rbg_{\EH},$$
as desired.

To prove \eqref{eq:g-btstrp:bfU}, note that 
\begin{align*}
	\frac{\ud}{\ud u} (\bfU - \overline{\bfU})(u)
	 = & \frac{\rd_{u} r(u, 1)}{\rd_{\bfU} r(\bfU(u), 1)} - \frac{\rd_{u} \rbg(u, 1)}{-1} \\
	 = & \rd_{u} r(u, 1) \left[ \left(\frac{\rd_{u} \rbg}{\rd_{u} r} (u,1) - 1 \right) 
	 + \left( 1 - \frac{\rd_{U} \rbg}{\rd_{U} r}(U(u), 1) \right) \right],
\end{align*}
and $(\bfU - \overline{\bfU})(1) = 0$. Therefore, the proof of \eqref{eq:g-btstrp:bfU} is reduced to showing the bounds
\begin{align}
	\int_{1}^{u_{(t_{B}) f}} (- \rd_{u} r)(u, 1) \Abs{1 - \frac{\rd_{U} \rbg}{\rd_{U} r}}(U(u), 1)  \, \ud u \leq & C \eps_{0},
	\label{eq:g-btstrp:bfU-1} \\
	\int_{1}^{u_{(t_{B}) f}} (- \dur) \Abs{\frac{\durbg}{\dur} - 1}(u, 1)  \, \ud u \leq & C A \log^2 \Lmb \left( \eps + \dlt_{(t_{B})}^{2} \right). \label{eq:g-btstrp:bfU-2} 
\end{align}
To prove \eqref{eq:g-btstrp:bfU-1}, we use \eqref{eq:L-st-ch:ini-dur}, the bound $\int_{1}^{u_{(t_{B}) f}} (-\rd_{u} r) \, \ud u \leq R_{0}$ and \eqref{eq:R0-Lmb}. For \eqref{eq:g-btstrp:bfU-2}, we first split the integral into $\int_{1}^{1 + C_{\gmm_{20}}} (\cdots)$ (where $r \geq 20$) and $\int_{1 + C_{\gmm_{20}}}^{u_{(t_{B}) f}} (\cdot)$ (where $r \leq 20$). The first integral is handled as in the proof of \eqref{eq:g-btstrp:bfU-1}, where we use \eqref{eq:dur-large-r} and \eqref{eq:df-large-r:durdf} in place of \eqref{eq:L-st-ch:ini-dur}. Finally, the second integral is treated as in \eqref{eq:g-btstrp:du-1} using \eqref{eq:df-small-r:logdurdf} and \eqref{eq:df-small-r:dur}; we omit the similar details.

\pfstep{Step~3} Next, we consider the data on the outgoing initial null curve $C_{out}$, on which $U = u = \ubg = 1$. Throughout this step (and the next one), it is useful to recall that $\inf_{C_{out} \cap \calD_{(t_{B})}} r = R_{0} > 20$ by \eqref{eq:R0-Lmb}. Moreover, $C^{-1} \leq r / \rbg \leq C$ by \eqref{eq:df-large-r:rdf}; therefore, the weights $r^{\omg}$ and $\rbg^{\omg}$ are comparable.

Our goal is to prove the following bounds:
\begin{align} 
	\sup_{C_{out} \cap \calD_{(t_{B})}} r^{\omg} \Abs{\dvr^{-1} \rd_{v} (r \phi - \rbg \phibg)}
	\leq & C \eps_{0} + C A \dlt_{0} \left( \eps + \dlt_{(t_{B})}^{2} \right), 		\label{eq:g-btstrp:dvrphi-pt} \\
	\int_{C_{out} \cap \calD_{(t_{B})}} \kpp^{-1} (\rd_{v} (r \phidf))^{2} \log (1 + r) \, \ud v
	\leq & C \eps_{0}^{2} + C A^{2} \dlt_{0}^{2} \left( \eps + \dlt_{(t_{B})}^{2} \right)^{2}, \label{eq:g-btstrp:dvrphi-int} \\
	\int_{C_{out} \cap \calD_{(t_{B})}} \kpp^{-1} (\rd_{v} \phidf)^{2} r^{2} \, \ud v
	\leq & C \eps_{0}^{2} + C A^{2} \dlt_{0}^{2} \left( \eps + \dlt_{(t_{B})}^{2} \right)^{2}. \label{eq:g-btstrp:dvphi} 
\end{align}
Note that \eqref{eq:g-btstrp:dvrphi-pt} coincides with \eqref{eq:g-btstrp-dvrphi}. Moreover, combined with \eqref{eq:L-st-ch:ini-r}--\eqref{eq:L-st-ch:ini-e} and \eqref{eq:g-btstrp:du}, it is straightforward to prove \eqref{eq:g-btstrp} from \eqref{eq:g-btstrp:dvrphi-int} and \eqref{eq:g-btstrp:dvphi}. 

The key claim for proving these bounds is:
\begin{equation} \label{eq:g-btstrp:dv-key}
	\rbg^{\omg}\Abs{\frac{1}{\kpp} \rd_{v} (r \phi) - \frac{1}{\kppbg} \rd_{v} (\rbg \phibg)}(1, v) \leq C \left((1 + \dlt) \eps_{0} + A \dlt \left( \eps + \dlt_{(t_{B})}^{2} \right)\right).
\end{equation}
We will prove this claim in Step 4; here, we assume \eqref{eq:g-btstrp:dv-key} and establish \eqref{eq:g-btstrp:dvrphi-pt}--\eqref{eq:g-btstrp:dvphi}.

For \eqref{eq:g-btstrp:dvrphi-pt}, we first use \eqref{eq:mu-large-r} to estimate
\begin{align*}
	\abs{\dvr^{-1} \rd_{v} (r \phi - \rbg \phibg)}
	\leq & 2 \Abs{\kpp^{-1} \rd_{v} (r \phi - \rbg \phibg)} \\
	\leq & 2 \Abs{\kpp^{-1} \rd_{v} (r \phi) - \kppbg^{-1} \rd_{v}(\rbg \phibg)} + 2 \Abs{\frac{\kppbg}{\kpp} - 1} \Abs{\kppbg^{-1} \rd_{v} (\rbg \phibg)}.
\end{align*}
Then the first term is bounded by \eqref{eq:g-btstrp:dv-key}, whereas the second term is bounded by \eqref{eq:dlt-adm:dvrphi} and \eqref{eq:en-btstrp-kpp}. 

Next, for \eqref{eq:g-btstrp:dvrphi-int}, we again use \eqref{eq:mu-large-r} to estimate
\begin{align*}
\int_{C_{out} \cap \calD_{(t_{B})}} \kpp^{-1} (\rd_{v} (r \phidf))^{2} \log (1 + r) \, \ud v
\leq 2 \left(\sup_{C_{out} \cap \calD_{(t_{B})}} r^{2} \Abs{\kpp^{-1} (\rd_{v} (r \phidf))}^{2} \right) \int_{1}^{v_{(t_{B})f}} \log (1+r) r^{-2} \rd_{v} r (1, v) \, \ud v .
\end{align*}
The integral is bounded by $\Lmb^{-1}\log \Lmb$, so it only remains to prove 
\begin{equation*}
	r \abs{\kpp^{-1} \rd_{v} (r \phidf)} \leq C \eps_0+ C A \log \Lmb \left( \eps + \dlt_{(t_{B})}^{2} \right).
\end{equation*}
To prove this bound, we first use \eqref{eq:kpp-bnd} and \eqref{eq:mu-large-r} to estimate
\begin{equation*}
	r \kpp^{-1} \abs{\rd_{v} (r \phidf)} 
	\leq r \dvr^{-1} \abs{\rd_{v} (r \phi - \rbg \phibg)} 
		+ r \abs{\dvrdf} \abs{\phibg} 
		+ r \abs{\rdf} \abs{\rd_{v} \phibg}
\end{equation*}
and then invoke \eqref{eq:dlt-adm:phi}, \eqref{eq:dlt-adm:dvphi}, \eqref{eq:df-large-r:rdf}, \eqref{eq:df-large-r:dvrdf} and \eqref{eq:g-btstrp:dvrphi-pt}.

Finally, for \eqref{eq:g-btstrp:dvphi}, we apply \eqref{eq:mu-large-r} and Lemma~\ref{lem:hardy-type} with $\alp = 0$ to estimate
\begin{align*}
	\int_{C_{out} \cap \calD_{(t_{B})}} \kpp^{-1} (\rd_{v} \phidf)^{2} r^{2} \, \ud v
	\leq \int_{1}^{v_{(t_{B}) f}} \frac{1}{\dvr} (\rd_{v} \phidf)^{2} r^{2} \, \ud v 
	\leq \int_{1}^{v_{(t_{B}) f}} \frac{1}{\dvr} (\rd_{v} (r \phidf))^{2}  \, \ud v 
		+ R_{0} \phidf^{2}(1, 1).
\end{align*}
On the far RHS, the second term is bounded by $C \eps_{0}^{2}$ thanks to \eqref{eq:L-st-ch:ini-phi} and \eqref{eq:R0-Lmb}, whereas the first term can be handled using \eqref{eq:g-btstrp:dvrphi-int} and the bound $\lmb^{-1} (1, v)\leq 2 \kpp^{-1}(1,v)$, which follows from \eqref{eq:mu-large-r}.

\pfstep{Step~4} 
It remains to establish \eqref{eq:g-btstrp:dv-key}; the argument here is similar to Step~2. We introduce a new coordinate $\bfV$ so that $(1-\mubg)^{-1} \rd_{\bfV} \rbg = 1$ on $C_{out}$, by setting
\begin{equation*}
	\frac{\ud \bfV}{\ud V}(V) = \frac{\rd_{V} \rbg}{1-\mubg} (1, V), \quad \bfV(1) = 1.
\end{equation*}
By the chain rule, we have
\begin{equation*}
	\frac{\rd_{\bfV} r}{1-\mu} (1, \bfV(V)) = \left(\frac{\rd_{V} \rbg}{1-\mubg}\right)^{-1} \frac{\rd_{V} r}{1-\mu} (1, V).
\end{equation*}
As before, we denote by $V(v)$ and $\Vbg(v)$ the inverses of the functions $V \mapsto v(V)$ and $V \mapsto \vbg(V)$, respectively. We furthermore define
\begin{equation*}
	\bfV(v) := \bfV(V(v)), \quad
	\bfVbg(v) := \bfV(\Vbg(v)).
\end{equation*}
Proceeding as in Step~2, we obtain the pointwise estimate
\begin{align*}
	\Abs{ \frac{1}{\kpp} \rd_{v} (r \phi) - \frac{1}{\kppbg} \rd_{v} (\rbg \phibg) }(1, v)
	\leq &\Abs{ \frac{1-\mu}{\rd_{V} r} \rd_{V} (r \phi) - \frac{1-\mubg}{\rd_{V} \rbg} \rd_{V} (\rbg\phibg) } (1, V(v)) \\
		&+ \abs{\bfV - \bfVbg}(v) \, \sup_{\bfV' \hbox{ between } \bfV(v), \bfVbg(v)} \Abs{\rd_{\bfV}^{2} (\rbg \phibg)}(1, \bfV')
\end{align*}
The first term is bounded by $\rbg^{-\omg} \eps_{0}$ by \eqref{eq:L-st-ch:ini-dvrphi}
To bound the second term, we need to use 
\begin{equation} \label{eq:g-btstrp:bfV}
\abs{\bfV - \bfVbg}(v) \leq C \rbg(1, v) \left( \eps_{0} + A \left(\eps + \dlt_{(t_{B})}^{2} \right) \right),
\end{equation}
which is proved below. Taking $C (\eps_{0} + A (\eps + \dlt_{(t_{B})}^{2} ) ) < \frac{1}{2}$ and using \eqref{eq:dlt-adm-dvdvrphi}, it follows that
\begin{align*}
\abs{\bfV - \bfVbg}(v) \sup_{\bfV' \hbox{ between } \bfV(v), \bfVbg(v)} \Abs{\rd_{\bfV}^{2} (\rbg \phibg)}(\bfV')
\leq & C \abs{\bfV - \bfVbg}(v) \, \rbg(1, v)^{-(\omg+1)} \dlt \\
\leq & C \rbg(1, v)^{-\omg} \dlt \left(\eps_{0} + A \left(\eps + \dlt_{(t_{B})}^{2} \right) \right),
\end{align*}
which is acceptable.

To conclude the proof, we need to verify \eqref{eq:g-btstrp:bfV}. By definition, we have
\begin{align*}
	\frac{\ud}{\ud v} (\bfV - \bfVbg)(v)
	 = & \frac{\kpp(1, v)}{(1-\mu)^{-1} \rd_{\bfV} r(1, \bfV(v))} - \frac{\kppbg(1, v)}{-1} \\
	 = & \kpp(1, v) \left(1 - \frac{\kppbg}{\kpp}(1, v)\right) + \kpp(1, v) \left( \left(\frac{\rd_{V} r}{1-\mu}\right)^{-1} \frac{\rd_{V} \rbg}{1-\mubg} (1, V(v)) - 1\right),
\end{align*}
and $(\bfV - \overline{\bfV})(1) = 0$. Then \eqref{eq:g-btstrp:bfV} can be proved by first integrating the above identity in $v$ from $v=1$, then applying \eqref{eq:L-st-ch:ini-dvr}, \eqref{eq:en-btstrp-kpp} and the bound $\kpp \leq 2 \rd_{v} r$ (which follows from \eqref{eq:mu-large-r}). We omit the straightforward details. \qedhere
\end{proof}

\subsection{Proof of weak stability (Theorem~\ref{thm:weak-st})} \label{subsec:weak-st-pf}
With Propositions~\ref{prop:en-btstrp} and \ref{prop:g-btstrp} in hand, we may now complete the proof of Theorem~\ref{thm:weak-st}.

\begin{proof}[Proof of Theorem~\ref{thm:weak-st}]
We first outline how to initiate and continue the bootstrap argument described in Section~\ref{subsec:btstrp}. As $t_{B} \to 0$, note that 
\begin{equation} \label{eq:weak-st-pf:ini-1}
\limsup_{t_{B} \to 0} \eps_{(t_{B})} \leq \eps_{0}, \quad
\liminf_{t_{B} \to 0} \eps_{(t_{B})} \geq \abs{\varpi - \varpibg}(1,1), \quad
\liminf_{t_{B} \to 0} \dlt_{(t_{B})} \geq \dlt.
\end{equation}
by \eqref{eq:eps-tB} and \eqref{eq:dlt-tB}. On the other hand, by \eqref{eq:EMSF-ray}--\eqref{eq:EMSF-r-phi-m}, we have
\begin{align*}
	\sup_{\calD_{(t_{B})}} \abs{\varpi - \varpi(1,1)} \leq & C_{\calD_{(t_{B})}} \left( \sup_{\calD_{(t_{B})}} r^{3} - \inf_{\calD_{(t_{B})}} r^{3} \right), \\
	\sup_{\calD_{(t_{B})}} \abs{\log \kpp_{(t_{B})} - \log (\kpp_{(t_{B})})(u_{(t_{B}) f}, \cdot)} 
	\leq & C_{\calD_{(t_{B})}} \left( \sup_{\calD_{(t_{B})}} r^{2} - \inf_{\calD_{(t_{B})}} r^{2} \right), \\
	\sup_{\calD_{(t_{B})}} \abs{\log (-\gmm_{(t_{B})}) - \log (-\gmm_{(t_{B})})(\cdot, v_{(t_{B}) f})} 
	\leq & C_{\calD_{(t_{B})}} \left( \sup_{\calD_{(t_{B})}} r^{2} - \inf_{\calD_{(t_{B})}} r^{2} \right),
\end{align*}
for sufficiently small $t_{B} > 0$, where $C_{\calD_{(t_{B})}}$ depends on the coordinate-independent quantities $\frac{1}{\dvr_{(t_{B})}} \rd_{v_{(t_{B})}} \phi$ and $\frac{1}{\dur_{(t_{B})}} \rd_{u_{(t_{B})}} \phi$ on $\calD_{(t_{B})}$ (moreover, it is non-decreasing in $t_{B}$). Indeed, since these inequalities are coordinate-independent, they can be easily proved in the $(U, V)$ coordinates. 

Note that the RHSs of the preceding inequalities vanish as $t_{B} \to 0$. Furthermore, the normalization conditions \eqref{eq:u-tB-cond}--\eqref{eq:v-tB-cond} imply that $\log \kpp_{(t_{B})} (u_{(t_{B}) f}, \cdot) = \log -\gmm_{(t_{B})} (\cdot, v_{(t_{B})f}) = 0$. It follows that
\begin{equation} \label{eq:weak-st-pf:ini-2}
	\limsup_{t_{B} \to 0}  \sup_{\calD_{(t_{B})}} \left( \abs{\varpi - \varpi(1,1)} + \abs{\log \kpp_{(t_{B})}} + \abs{\log -\gmm_{(t_{B})}} \right) = 0.
\end{equation}
By a similar argument for the background solution, we obtain
\begin{equation} \label{eq:weak-st-pf:ini-3}
	\limsup_{t_{B} \to 0} \sup_{\calD_{(t_{B})}} \abs{\varpibg - \varpibg(1,1)} = 0, \quad
	\limsup_{t_{B} \to 0}  \sup_{\calD_{(t_{B})}} \left(\abs{\log \kppbg} + \abs{\log (-\gmmbg)} \right) \leq C \dlt^{2},
\end{equation}
where we used Proposition~\ref{prop:bg-geom} to bound $\log \kppbg(u_{(t_{B}) f}, \cdot)$ and $\log (-\gmmbg)(\cdot, v_{(t_{B}) f})$ by $C \dlt^{2}$ (here, $C$ is as in \eqref{eq:dlt-tB-decay}). Next, notice that by Definition~\ref{def:dlt-adm} and \eqref{eq:L-st-ch:ini-r}, for $\ep_0\leq \f 14$,
\begin{equation} \label{eq:weak-st-pf:ini-4}
 \limsup_{t_{B} \to 0} |r(u_{(t_B)f}, v_{(t_B)f})-\Lambda| \leq \f 14.
\end{equation}

Putting together \eqref{eq:weak-st-pf:ini-1}, \eqref{eq:weak-st-pf:ini-2}, \eqref{eq:weak-st-pf:ini-3} and \eqref{eq:weak-st-pf:ini-4}, and using \eqref{eq:L-st-ch:ini-m}, we see that the bootstrap assumptions \eqref{eq:en-btstrp-eps}--\eqref{eq:en-btstrp-R-Lmb} hold for small enough $t_{B}$ (where $A \geq 1$).

Taking $A$ sufficiently large and $\eps_{0}$, $\dlt_{0}$ sufficiently small, we may apply Propositions~\ref{prop:en-btstrp} and \ref{prop:g-btstrp}, so that the improved bootstrap assumptions \eqref{eq:en-btstrp-m:imp}--\eqref{eq:en-btstrp-R-Lmb:imp} and \eqref{eq:g-btstrp} hold. Since both sides of \eqref{eq:en-btstrp-eps}--\eqref{eq:en-btstrp-R-Lmb} are continuous in $t_{B}$, we may continue the bootstrap argument for arbitrarily large $t_{B}$.

To conclude the proof, observe that as $t_{B} \to \infty$, we have the convergences (uniform on compact subsets) $u_{(t_{B})} (U) \to u(U)$ and $v_{(t_{B})}(V) \to v(V)$ (these are clear from the decay bounds for $(r, \phi, \varpi)$ in Section~\ref{sec:bg}). Note also that $\dlt_{(t_{B})} \to 0$ and $r \phidf^{2}(u, v_{(t_{B}) f}) \leq 2 (r \phi^{2} + r \phibg^{2})(u, v_{(t_{B}) f}) \to 0$ as $t_{B} \to \infty$ (the former by \eqref{eq:dlt-tB-decay} and the latter by the uniform boundedness of $r \phi$ and $\rbg \phibg$ in Section~\ref{sec:bg}). Finally, by the equations \eqref{eq:EMSF-ray}--\eqref{eq:EMSF-r-phi-m} and the normalization conditions \eqref{eq:u-tB-cond}--\eqref{eq:v-tB-cond}, we can also pass to the limit $t_{B} \to \infty$ of Propositions~\ref{prop:en-btstrp}, \ref{prop:g-btstrp}, \ref{prop:en-dvphi-goal}, \ref{prop:energy-goal} and Corollary~\ref{cor:zeroth} (as well as the bounds in Section~\ref{subsec:geom}), which proves Theorem~\ref{thm:weak-st}.  \qedhere
\end{proof}

\subsection{Energy decay in the $\Gmm_{\tau}$ foliation: $r^{p}$-weighted energy estimates} \label{subsec:rp-weight}
Starting from the weak stability result (Theorem~\ref{thm:weak-st}), we now prove decay of nondegenerate energy on $\Gmm_{\tau}$. We work under the same assumptions and conventions as in Theorem~\ref{thm:weak-st}. In particular, we normalize $\rbg_{\EH} = 1$ by scaling. Moreover, from now on till the end of the section, unless otherwise stated, we will \textbf{use $C$ to denote a constant depending} on $\omg$, $\eta_{0}$, $|\ebg|$, $\Lmb$ (consistent with the conventions in Theorem~\ref{thm:weak-st}). Note that this is different from the convention in the proof of Theorem~\ref{thm:weak-st}, where we only allowed $C$ to depend on $\omg$, $\eta_{0}$, $|\ebg|$, but not $\Lmb$. \textbf{We will also take $\dlt_{0}$ and $\eps_{0}$ sufficiently small depending on ($|\ebg|$, $\eta_0$, $\omg$) and ($\dlt_{0}$, $|\ebg|$, $\eta_0$, $\omg$, $\Lmb$) respectively as necessary without further comment.}

The main result of this subsection is as follows:
\begin{proposition} \label{prop:en-decay}
The following estimate holds for $\tau \geq 1$:
\begin{equation} \label{eq:en-decay}
	\int_{\Gmm_{\tau}^{(in)}} \frac{1}{(-\dur)} (\rd_{u} \phidf)^{2} r^{2} \, \ud u
	+ \int_{\Gmm_{\tau}^{(out)}}  \frac{1}{\dvr} (\rd_{v} \phidf)^{2} r^{2} \, \ud v
	\leq C  \tau^{-3 + \eta_{0}} \eps_{0}^{2}.
\end{equation}
\end{proposition} 
As a quick corollary of Proposition~\ref{prop:en-decay}, Hardy's inequality (Lemma~\ref{lem:hardy-opt} with $\alp = 0$) and the (qualitative) fact that $r \phidf^{2}(u, v) \to 0$ as $v \to \infty$, we obtain the following pointwise decay estimate for $\phidf$:
\begin{corollary} \label{cor:phidf-decay-pre}
The following estimate holds for $(u, v) \in \set{u \geq 1, \ r(u, v) \geq R_{0}}$:
\begin{equation} \label{eq:phidf-decay-pre}
	r^{\frac{1}{2}} \abs{\phidf} (u, v) \leq C u^{-\frac{3}{2}+\frac{\eta_{0}}{2}} \eps_{0} .
\end{equation}
\end{corollary}
Note that the decay rate in $u$ of the RHS is \emph{integrable}. In the next subsection, we propagate this decay rate to the radiation field difference by integration along characteristics.

As in Section~\ref{subsec:extr-st}, we work with $\psi = r \phi$ and its difference $\psidf = r \phi - \rbg \phibg$. Recall that $\psidf$ solves
\begin{equation*}
	\rd_{u} \rd_{v} \psidf - \frac{\rd_{u} \rd_{v} r}{r} \psidf = G,
\end{equation*}
where
\begin{equation*}
	G = \left(\frac{\rd_{u} \rd_{v} r}{r} -\frac{\rd_{u} \rd_{v} \rbg}{\rbg} \right) \psibg.
\end{equation*}

We begin with some preliminary estimates, which essentially identifies $\psidf$ with $r \phidf$.
\begin{lemma} \label{lem:psidf-phidf}
The following bounds hold for $(u, v) \in \calD \cap \set{r \geq 20}$:
\begin{equation} \label{eq:psidf-phidf}
	\abs{\psidf - r \phidf} \leq C u^{-\omg+1}  \min \set{u^{-1 + \frac{\eta_{0}}{2}}, r^{-1+\frac{\eta_{0}}{2}}} \dlt \eps_{0} ,
\end{equation}
\begin{equation} \label{eq:psidf-phidf-dv}
	\abs{\rd_{v} \psidf - \rd_{v} (r \phidf)} \leq C u^{-\omg+1} r^{-1} \min \set{u^{-1 + \frac{\eta_{0}}{2}}, r^{-1+\frac{\eta_{0}}{2}}} \dlt \eps_{0} .
\end{equation}
Moreover, for $(u, v) \in \calD$, we have 
\begin{equation} \label{eq:psidf-phidf-du}
\Abs{\frac{1}{(-\dur)}\rd_{u} \psidf - \frac{1}{(-\dur)} \rd_{u} (r \phidf)} \leq 
\left\{
\begin{array}{ll}
C v^{-\omg} e^{\frac{1}{2} c_{(\dur)} (u - v - C_{\gmm_{20}}) } \dlt \eps_{0} & \hbox{ in } \calD \cap \set{r \leq 20}, \\
C u^{-\omg} r^{-1+\frac{\eta_{0}}{2}} \dlt \eps_{0} & \hbox{ in } \calD \cap \set{r \geq 20}.
\end{array}
\right. 
\end{equation}
\end{lemma}
\begin{proof}
Note that $\psidf - r \phidf = \rdf \phibg$. Therefore, \eqref{eq:psidf-phidf}, \eqref{eq:psidf-phidf-dv} and \eqref{eq:psidf-phidf-du} are straightforwad consequences of \eqref{eq:dlt-adm:phi}, \eqref{eq:dlt-adm:duphi}, \eqref{eq:dlt-adm:dvphi}, \eqref{eq:gauge:eps}, \eqref{eq:weak-st:rdf}, \eqref{eq:weak-st:dvrdf} and \eqref{eq:weak-st:durdf}, combined with the crude bound $\log r \leq C r^{\eta_{0}/2}$. We omit the straightforward details. \qedhere
\end{proof}

To prove Proposition~\ref{prop:en-decay}, we adapt the method of $r^{p}$-weighted energy estimates of Dafermos--Rodnianski \cite{DRNM} to our problem. We rely on the following energy identity concerning the multiplier $\frac{r^{p}}{\dvr} \rd_{v}$:
\begin{lemma} \label{lem:en-id-rp}
For any $p \in \bbR$, we have
\begin{equation} \label{eq:en-id-rp}
\begin{aligned}
& \hskip-2em
	\rd_{u} \left( \frac{1}{2} \frac{r^{p}}{\dvr} (\rd_{v} \psidf)^{2} \right) 
	+ \rd_{v} \left(\frac{1}{2} \frac{(-\rd_{u} \rd_{v} r)}{r} \frac{r^{p}}{\dvr} \psidf^{2} \right) \\
& \hskip-2em
	+ \frac{p}{2} r^{p-1} \frac{(-\dur)}{\dvr} (\rd_{v} \psidf)^{2} 
	+  \frac{3-p}{2} r^{p-4} r^{2} (-\rd_{u} \rd_{v} r)  \psidf^{2} \\
= &	G \frac{r^{p}}{\dvr} \rd_{v} \psidf 
	+ \frac{1}{2} r^{p-2}  r^{2} \frac{(-\rd_{u} \rd_{v} r)}{\dvr^{2}} (\rd_{v} \psidf)^{2} 
	+ \frac{1}{2} r^{p-3} \rd_{v} \left(r^{2} \frac{(-\rd_{u} \rd_{v} r)}{\dvr} \right) \psidf^{2}.
\end{aligned}
\end{equation}
\end{lemma}

\begin{proof}
We compute
\begin{align*}
\left( \rd_{u} \rd_{v} \psidf - \frac{\rd_{u} \rd_{v} r}{r} \psidf \right) \frac{r^{p}}{\dvr} \rd_{v} \psidf
= & \rd_{u} \left( \frac{1}{2} \frac{r^{p}}{\dvr} (\rd_{v} \psidf)^{2} \right) 
	+ \rd_{v} \left(\frac{1}{2} \frac{(-\rd_{u} \rd_{v} r)}{r} \frac{r^{p}}{\dvr} \psidf^{2} \right) \\
& - \frac{1}{2} \rd_{u} \left(\frac{r^{p}}{\dvr} \right) (\rd_{v} \psidf)^{2} 
	- \frac{1}{2} \rd_{v} \left(r^{p-1} \frac{(-\rd_{u} \rd_{v} r)}{\dvr} \right) \psidf^{2} \\
= & \rd_{u} \left( \frac{1}{2} \frac{r^{p}}{\dvr} (\rd_{v} \psidf)^{2} \right) 
	+ \rd_{v} \left(\frac{1}{2} \frac{(-\rd_{u} \rd_{v} r)}{r} \frac{r^{p}}{\dvr} \psidf^{2} \right) \\
& + \frac{p}{2} r^{p-1} \frac{(-\dur)}{\dvr} (\rd_{v} \psidf)^{2} 
 + \frac{1}{2} r^{p} \frac{\rd_{u} \rd_{v} r}{\dvr^{2}} (\rd_{v} \psidf)^{2} \\
& +  \frac{3-p}{2} r^{p-4} r^{2} (-\rd_{u} \rd_{v} r)  \psidf^{2} 
- \frac{1}{2} r^{p-3} \rd_{v} \left(r^{2} \frac{(-\rd_{u} \rd_{v} r)}{\dvr} \right) \psidf^{2}. \qedhere
\end{align*}
\end{proof}

In view of the terms on the LHS of \eqref{eq:en-id-rp}, we introduce the $r^{p}$-weighted energy flux on $\Gmm_{\tau}$ as follows:
\begin{equation}\label{eq:Ep-def}
	E_{p}[\psidf](\tau) = \int_{\Gmm^{(in)}_{\tau}} \frac{1}{(-\dur)} (\rd_{u} \psidf)^{2} \, \ud u
					+ \int_{\Gmm^{(out)}_{\tau}} \frac{r^{p}}{\dvr} (\rd_{v} \psidf)^{2} \, \ud v
					+ \sup_{\Gmm^{(in)}_{\tau}} r^{-1} \psidf^{2}.
\end{equation}
For any $R \geq R_{0}+1$, we also introduce the nondegenerate energy flux on $\Gmm_{\tau} \cap \set{r \leq R}$ (augmented with a zeroth order term) as follows:
\begin{equation}\label{eq:dotE-R-def}
	\dot{E}_{\leq R}[\phidf](\tau) = \int_{\Gmm^{(in)}_{\tau}} \frac{1}{(-\dur)} (\rd_{u} \phidf)^{2} r^{2} \, \ud u + \int_{\Gmm^{(out)}_{\tau}  \cap \set{r \leq R}} \frac{1}{\kpp} (\rd_{v} \phidf)^{2} r^{2} \, \ud v + \int_{\Gmm_{\tau}^{(out)} \cap \set{R-1 \leq r \leq R}} \dvr r \phidf^{2} \, \ud v.
\end{equation}
We write $\dot{E}[\phidf](\tau) = \lim_{R \to \infty} \dot{E}_{\leq R}[\phidf](\tau)$ for the nondegenerate energy flux on the whole leaf $\Gmm_{\tau}$, i.e.,
\begin{equation}\label{eq:dotE-def}
	\dot{E}[\phidf](\tau) = \int_{\Gmm^{(in)}_{\tau}} \frac{1}{(-\dur)} (\rd_{u} \phidf)^{2} r^{2} \, \ud u 
					+ \int_{\Gmm^{(out)}_{\tau}} \frac{1}{\kpp} (\rd_{v} \phidf)^{2} r^{2} \, \ud v .
\end{equation}

In the following lemma, we record some basic properties of $E_{p}[\psidf]$, $\dot{E}_{\leq R}[\phidf]$ and $\dot{E}[\phidf]$.
\begin{lemma} \label{lem:Ep}
Let $R \geq R_{0} + 1$ and $\tau \geq 1$. For $p \geq -1$, we have
\begin{equation} 	
	E_{p}[\psidf](\tau) 
	\leq  \int_{C_{\tau} \cap \set{r \geq R}} \frac{r^{p}}{\dvr} (\rd_{v} \psidf)^{2} \, \ud v + C_{R, p} \left( \dot{E}_{\leq R}[\phidf](\tau) + \tau^{-2 \omg + \eta_{0}} \dlt^{2} \eps_{0}^{2} \right). 				\label{eq:Ep-dotE}  
\end{equation}
For $p \geq 0$, we have
\begin{equation} 
	\dot{E}[\phidf](\tau) \leq C E_{p}[\psidf](\tau) + C \tau^{-2 \omg + \eta_{0}} \dlt^{2} \eps_{0}^{2}. \label{eq:dotE-Ep}
\end{equation}
Finally, let $p < 2 \omg - 1$. Then on $\Gmm_{1}$ we have
\begin{equation} \label{eq:Ep-eps0}
	E_{p}[\psidf](1) \leq C \eps_{0}^{2}.
\end{equation}
\end{lemma}
\begin{proof}
We prove \eqref{eq:Ep-dotE}, \eqref{eq:dotE-Ep} and \eqref{eq:Ep-eps0} in order.

\pfstep{Step~1: Proof of \eqref{eq:Ep-dotE}}
Since $C_{\tau} \cap \set{r \geq R_{0}} = \Gmm^{(out)}_{\tau}$, the proof of \eqref{eq:Ep-dotE} reduces to showing
\begin{align} 
	\int_{\Gmm^{(in)}_{\tau}} \frac{1}{(-\dur)} (\rd_{u} \psidf)^{2} \, \ud u 
\leq & C \left( \dot{E}_{\leq R}[\phidf](\tau) + \tau^{-2 \omg} \dlt^{2} \eps_{0}^{2} \right), \label{eq:Ep-dotE:pf1} \\
	\int_{\Gmm^{(out)}_{\tau} \cap \set{r \leq R}} \frac{r^{p}}{\dvr} (\rd_{v} \psidf)^{2} \, \ud v
\leq & C R^{p} \left( \dot{E}_{\leq R}[\phidf](\tau) + \tau^{-2 \omg + \eta_{0}} \dlt^{2} \eps_{0}^{2} \right), \label{eq:Ep-dotE:pf2}  \\
	\sup_{\Gmm^{(in)}_{\tau}} r^{-1} \psidf^{2}
	\leq & C \left( \dot{E}_{\leq R}[\phidf](\tau) + \tau^{-2 \omg + \eta_{0}} \dlt^{2} \eps_{0}^{2} \right). \label{eq:Ep-dotE:pf3} 
\end{align}
To establish \eqref{eq:Ep-dotE:pf1}--\eqref{eq:Ep-dotE:pf3}, we first note that
\begin{equation} \label{eq:Ep-dotE:pf4} 
	\int_{\Gmm^{(in)}_{\tau}} (-\dur) \phidf^{2} \, \ud u 
	+ \int_{\Gmm^{(out)}_{\tau}} \dvr \phidf^{2} \, \ud v
	+ \sup_{\Gmm^{(in)}_{\tau}} r \phidf^{2}
	\leq C \dot{E}_{\leq R}[\phidf](\tau).
\end{equation}
Indeed, by Hardy's inequality (Lemma~\ref{lem:hardy-opt} with $\alp = 0$), for any $R' \in [R-1, R]$ we have
\begin{equation*}
\frac{1}{4} \int_{\Gmm^{(out)}_{\tau} \cap \set{r \leq R'}} \dvr \phidf^{2} \, \ud v + \frac{1}{2} r \phidf^{2}(\tau, v_{R_{0}}(\tau))
\leq \int_{\Gmm^{(out)}_{\tau} \cap \set{r \leq R'}} \frac{1}{\dvr} (\rd_{v} \phidf)^{2} r^{2} \, \ud v + \frac{1}{2} r \phidf^{2}(\tau, v_{R'}(\tau)).
\end{equation*}
Averaging in $R' \in [R-1, R]$, the last term on the RHS may be bounded by $\int_{\Gmm_{\tau}^{(out)} \cap \set{R_{1} \leq r \leq R}} \dvr r \phidf^{2} \, \ud v$. Then applying Hardy's inequality on $\Gmm^{(in)}_{\tau}$ as well, \eqref{eq:Ep-dotE:pf4} follows.

To prove \eqref{eq:Ep-dotE:pf1}, we first estimate
\begin{align*}
	\frac{1}{(-\dur)}(\rd_{u} \psidf)^{2} 
	\leq & \frac{C}{(-\dur)} (\rd_{u} \phidf)^{2} r^{2} + C (-\dur) \phidf^{2} +  \frac{C}{(-\dur)} (\rd_{u} (\psidf - r \phidf))^{2}  .
\end{align*}
By \eqref{eq:weak-st:dur-small-r}, \eqref{eq:weak-st:dur-large-r}, \eqref{eq:weak-st:v-u} and \eqref{eq:psidf-phidf-du}, the integral of the last term over $\Gmm^{(in)}_{\tau}$ is bounded by $C v_{R_{0}}(\tau)^{-2 \omg} \dlt^{2} \eps_{0}^{2}$; by \eqref{eq:weak-st:v-u-R0}, we also have $v_{R_{0}}(\tau)^{- 2\omg} \leq C \tau^{-2 \omg}$, so the contribution of the last term is acceptable. On the other hand, the integral of the remaining two terms can be bounded by $\dot{E}_{\leq R}[\phidf](\tau)$ thanks to its definition and \eqref{eq:Ep-dotE:pf4}.

Next, to prove \eqref{eq:Ep-dotE:pf2}, we start with the pointwise bound
\begin{align*}
	\frac{r^{p}}{\dvr} (\rd_{v} \psidf)^{2}
	\leq \frac{C R^{p}}{\kpp} (\rd_{v} \phidf)^{2} r^{2} + C R^{p} \dvr \phidf^{2} + C R^{p} (\rd_{v} (\psidf - r \phidf))^{2} 
\end{align*}
in $\Gmm_{\tau}^{(out)} \cap \set{r \leq R}$. The integral of the last term on the RHS is bounded by $C R^{p} \tau^{-2\omg + \eta_{0}} \dlt^{2} \eps_{0}^{2}$ by \eqref{eq:weak-st:dvr-large-r} and \eqref{eq:psidf-phidf-dv}. On the other hand, the integral of the first two terms is bounded by $C R^{p} \dot{E}_{\leq R} [\phidf](\tau)$ by its definition and \eqref{eq:Ep-dotE:pf4}.

Finally, \eqref{eq:Ep-dotE:pf3} follows from \eqref{eq:psidf-phidf} and \eqref{eq:Ep-dotE:pf4}. This completes the proof of \eqref{eq:Ep-dotE}.

\pfstep{Step~2: Proof of \eqref{eq:dotE-Ep}}
Clearly, $E_{0}[\psidf](\tau) \leq E_{p}[\psidf](\tau)$ for any $p \geq 0$; therefore, we may just consider the case $p = 0$. By Lemma~\ref{lem:psidf-phidf}, \eqref{eq:weak-st:dvr-large-r}, \eqref{eq:weak-st:dur-small-r}, \eqref{eq:weak-st:dur-large-r} and \eqref{eq:weak-st:mu-large-r}, we have
\begin{equation*}
	\int_{\Gmm^{(in)}_{\tau}} \frac{1}{(-\dur)} (\rd_{u} (r \phidf))^{2} \, \ud u
	+ \int_{\Gmm^{(out)}_{\tau}} \frac{1}{\kpp} (\rd_{v} (r \phidf))^{2} \, \ud v
	+ \sup_{\Gmm^{(in)}_{\tau}} r \phidf^{2}
	\leq C E_{0}[\psidf](\tau) + C \tau^{-2 \omg + \eta_{0}} \dlt^{2} \eps_{0}^{2}. 
\end{equation*}
Then applying Lemma~\ref{lem:hardy-type} with $\alp = 0$, \eqref{eq:dotE-Ep} follows.

\pfstep{Step~3: Proof of \eqref{eq:Ep-eps0}}
We apply \eqref{eq:Ep-dotE} with $\tau = 1$ and $R = R_{0} + 1$. Note that $\dot{E}_{\leq R_{0} + 1}[\phidf](1) \leq C \eps^{2} \leq C \eps_{0}^{2}$ by \eqref{eq:eps-def}, \eqref{eq:gauge:eps} and \eqref{eq:energy-phidf-pt}. Moreover, $\int_{C_{\tau} \cap \set{r \geq R_{0}+1}} r^{p} \dvr^{-1} (\rd_{v} \psidf)^{2} \, \ud v \leq C \eps_{0}^{2}$ for $p < 2 \omg - 1$ by \eqref{eq:gauge:dvrphi}. \qedhere
\end{proof}

In order to prove $r^{p}$-weighted energy estimates (see Lemma~\ref{lem:en-rp} below), we need to bound the RHS of \eqref{eq:en-id-rp}. We formulate here a technical estimate, which treats the contribution of $G$:
\begin{lemma} \label{lem:en-rp-G}
For $p \leq 3$, we have
\begin{equation} \label{eq:en-rp-G}
\iint_{\calD(\tau_{1}, \tau_{2}) \cap \set{r \geq 20}} r^{p+1} \abs{G}^{2} \, \ud u \ud v \leq C \tau_{1}^{-2 \omg + \max \set{1, p}} \dlt^{2}\eps_{0}^{2}.
\end{equation}
\end{lemma}
We remark that \eqref{eq:en-rp-G} is far from optimal (with respect to other estimates we have), but nevertheless sufficient for our purposes.
\begin{proof}
To prove \eqref{eq:en-rp-G}, it suffices to verify it in the cases $p = 1$ and $p = 3$. Indeed, the case $p \leq 1$ obviously follows from the case $p = 1$, and the case $1 \leq p \leq 3$ follows by interpolation; in this context, this simply means that we split the integration domain into $\set{r \leq \tau_{1}}$ and $\set{r \geq \tau_{1}}$, then use the case $p = 1$ for the former and $p = 3$ for the latter.

We begin with the identity
\begin{align*}
	r^{p+1} \abs{G}^{2}
	= r^{p+1} \left(\frac{1}{r^{3}} (r^{2} \rd_{u} \rd_{v} r - \rbg^{2} \rd_{u} \rd_{v} \rbg) + \left(\frac{1}{r^{3}} - \frac{1}{\rbg^{3}}\right) \rbg^{2} \rd_{u} \rd_{v} \rbg \right)^{2} \psibg^{2}.
\end{align*}
Recall that $r^{2} \rd_{u} \rd_{v} r = 2 (\varpi - \frac{\e^{2}}{r}) \kpp \dur$. Hence by \eqref{eq:gauge:eps}, \eqref{eq:weak-st:mdf}, \eqref{eq:weak-st:kppdf}, \eqref{eq:weak-st:rdf}, \eqref{eq:weak-st:durdf}, \eqref{eq:r-alp-df} and Proposition~\ref{prop:bg-geom}, we have
\begin{equation*}
	r^{-3} \abs{r^{2} \rd_{u} \rd_{v} r - \rbg^{2} \rd_{u} \rd_{v} \rbg} \leq C r^{-3} \eps_{0}, \quad
	(r^{-3} - \rbg^{-3}) \abs{\rbg^{2} \rd_{u} \rd_{v} \rbg} \leq C r^{-4} \eps_{0}.
\end{equation*}
Combined with \eqref{eq:dlt-adm:phi}, we arrive at
\begin{equation*}
	r^{p+1} \abs{G}^{2} \leq C r^{p-3} u^{-2 \omg + 2} \min\set{u^{-2}, r^{-2}}\dlt^{2} \eps_{0}^{2} \quad \hbox{ in } \calD \cap \set{r \geq 20}.
\end{equation*}
When $p = 1$, we use the decay $u^{-2}$ inside the minimum, whereas in the case $p = 3$ we use $r^{-2}$. Integrating this bound over $\calD(\tau_{1}, \tau_{2}) \cap \set{r \geq 20}$ and using \eqref{eq:weak-st:dvr-large-r}, the desired estimate \eqref{eq:en-rp-G} follows. \qedhere
\end{proof}

We are now ready to state and prove the main $r^{p}$-weighted energy estimates for $\psidf$.
\begin{lemma} \label{lem:en-rp}
Let $0 < p < 3$. For any $1 \leq \tau_{1} \leq \tau_{2}$, we have
\begin{equation} \label{eq:en-rp}
	E_{p}[\psidf](\tau_{2}) + \int_{\tau_{1}}^{\tau_{2}} E_{p-1}[\psidf] (\tau) \, \ud \tau
	\leq C_{p} \left( E_{p}[\psidf](\tau_{1}) + \tau_{1}^{-2\omg + \max\set{1+\eta_{0}, p}} \dlt^{2} \eps_{0}^{2} \right).
\end{equation}
\end{lemma}

\begin{proof}
In this proof, we omit the line and area elements $\ud u, \ud v, \ud u \ud v$ to simplify the notation.

\pfstep{Step~1}
Let $R_{2} > 20 R_{0}$ (to be specified below), and let $\chi_{R_{2}}(r)$ be a smooth non-negative function such that $\chi_{R_{2}}(r) = 0$ in $\set{r \leq R_{2}/2}$, $\chi_{R_{2}}(r) = 1$ in $\set{r \geq R_{2}}$ and $\abs{\chi'_{R_{2}} (r)} \leq \frac{10}{r}$. We multiply \eqref{eq:en-id-rp} by $\chi_{R_{2}}$, and integrate by parts over $\calD(\tau_{1}, \tau_{2})$.

For the contribution of the LHS of \eqref{eq:en-id-rp}, we claim that
\begin{equation} \label{eq:en-rp0}
\begin{aligned}
& \hskip-2em
	\frac{1}{2} \int_{\Gmm^{(out)}_{\tau_{2}} \cap \set{r \geq R_{2}}} \frac{r^{p}}{\dvr} (\rd_{v} \psidf)^{2} 
	+ \frac{p}{8} \iint_{\calD(\tau_{1}, \tau_{2})} \chi_{R_{2}} \frac{r^{p-1}}{\dvr} (\rd_{v} \psidf)^{2} 
	+ \frac{3-p}{2} c_{(\varpi)} \iint_{\calD(\tau_{1}, \tau_{2})} \chi_{R_{2}} r^{p-4} \psidf ^{2} \\
\leq & \iint_{\calD(\tau_{1}, \tau_{2})} \chi_{R_{2}} (\hbox{LHS of \eqref{eq:en-id-rp}})  
	+  \frac{1}{2} \int_{\Gmm^{(out)}_{\tau_{1}} \cap \set{r \geq R_{2}/2}} \frac{r^{p}}{\dvr} (\rd_{v} \psidf)^{2} + C R_{2}^{p+\eta_{0}} \left( \dot{E}[\phidf](\tau_{1}) + \tau_{1}^{-2 \omg + 1 + \eta_{0}} \dlt^{2} \eps_{0}^{2}\right),
\end{aligned}
\end{equation}
where $c_{(\varpi)}$ is a positive \emph{universal} constant specified in \eqref{eq:c-varpi} below. On the other hand, for the contribution of the RHS of \eqref{eq:en-id-rp}, we claim that
\begin{align} 
	\iint_{\calD(\tau_{1}, \tau_{2})} \chi_{R_{2}} \Abs{G \frac{r^{p}}{\dvr} \rd_{v} \psidf} 
\leq & C_{p} \tau_{1}^{-2 \omg + \max \set{1, p}} \dlt^{2} \eps_{0}^{2}
	+ \frac{p}{16} \iint_{\calD(\tau_{1}, \tau_{2})} \chi_{R_{2}} \frac{r^{p-1}}{\dvr} (\rd_{v} \psidf)^{2} , \label{eq:en-rp1} \\
	\iint_{\calD(\tau_{1}, \tau_{2})} \chi_{R_{2}} \Abs{r^{p-2} r^{2} \frac{(- \rd_{u} \rd_{v} r)}{\dvr^{2}} (\rd_{v} \psidf)^{2} }
	\leq & \frac{C}{R_{2}} \iint_{\calD(\tau_{1}, \tau_{2})} \chi_{R_{2}} \frac{r^{p-1}}{\dvr} (\rd_{v} \psidf)^{2} , \label{eq:en-rp2}  \\
	\iint_{\calD(\tau_{1}, \tau_{2})} \chi_{R_{2}} \Abs{r^{p-3} \rd_{v} \left( r^{2} \frac{(- \rd_{u} \rd_{v} r)}{\dvr} \right) \psidf^{2} }
	\leq & C \left(\frac{1}{R_{2}} + \eps_{0}^{2} \right) \iint_{\calD(\tau_{1}, \tau_{2})} \chi_{R_{2}} r^{p-4} \psidf^{2} \notag \\
	& + \frac{C \eps_{0}^{2}}{R_{2}} \iint_{\calD(\tau_{1}, \tau_{2})} \chi_{R_{2}} \frac{r^{p-1}}{\dvr} (\rd_{v} \psidf)^{2} . \label{eq:en-rp3} 
\end{align}
In the remainder of this step, we prove \eqref{eq:en-rp} assuming these claims. Putting together \eqref{eq:en-rp0}--\eqref{eq:en-rp3}, we obtain
\begin{align*}
	(\hbox{LHS of \eqref{eq:en-rp0}}) 
	\leq & C \int_{\Gmm^{(out)}_{\tau_{1}} \cap \set{r \geq R_{2}/2}} \frac{r^{p}}{\dvr} (\rd_{v} \psidf)^{2} + C R_{2}^{p+\eta_0} \left(\dot{E}[\phidf](\tau_{1}) + \tau_{1}^{-2 \omg + \max \set{1, p}} \dlt^{2} \eps_{0}^{2} \right) \\
	& + \left( \frac{1}{2} + \frac{C}{R_{2}} + C \eps_{0}^{2} \right) \times (\hbox{LHS of \eqref{eq:en-rp0}}).
\end{align*}
Choosing $R_{2}$ sufficiently large and $\eps_{0}^{2}$ sufficiently small (depending on $p$, $\eta_{0}$, $\ebg$ and $\Lmb$), the last line can be absorbed into the LHS. At this point, we fix $R_{2}$ and from now on suppress the dependence of constants on $R_{2} = R_{2}(p, \eta_{0}, \ebg, \Lmb)$. By \eqref{eq:gauge:eps}, \eqref{eq:energy}, \eqref{eq:Ep-dotE} and the preceding estimate, we furthermore have
\begin{align*}
	E_{p}[\psidf](\tau_{2}) + \int_{\tau_{1}}^{\tau_{2}} E_{p-1}[\psidf](\tau) \, \ud \tau
	\leq & C_{p} (\hbox{LHS of \eqref{eq:en-rp0}}) 
		+ C_{p} \left( \dot{E}_{\leq R_{2}}[\phidf](\tau_{2}) + \int_{\tau_{1}}^{\tau_{2}} \dot{E}_{\leq R_{2}}[\phidf](\tau) \, \ud \tau \right) \\
		& + C_{p} \tau^{-2 \omg + \eta_{0} + 1} \dlt^{2} \eps_{0}^{2} \\
	\leq & C_{p} \left( \int_{\Gmm_{\tau_{1}}^{(out)}} \frac{r^{p}}{\dvr} (\rd_{v} \psidf)^{2}  + \dot{E}[\phidf](\tau_{1}) 
		+ \tau_{1}^{-2 \omg + \max \set{1+\eta_{0}, p}} \dlt^{2} \eps_{0}^{2} \right).
\end{align*}
Then \eqref{eq:en-rp} follows from an application of \eqref{eq:dotE-Ep} to the last line.

\pfstep{Step~2} It remains to verify the claims \eqref{eq:en-rp0}--\eqref{eq:en-rp3}. 
For \eqref{eq:en-rp0}, we start by noting that there exist universal constants $0 < c_{(\varpi)} < C$ such that
\begin{equation} \label{eq:c-varpi}
	c_{(\varpi)} \leq r^{2} \frac{(-\rd_{u} \rd_{v} r)}{\dvr} \leq C \quad \hbox{ in } \calD \cap \set{r \geq 20},
\end{equation}
by \eqref{eq:weak-st:me-bnd}, \eqref{eq:weak-st:dvr-large-r} and \eqref{eq:weak-st:dur-large-r}. Therefore, using \eqref{eq:weak-st:dvr-large-r} and \eqref{eq:weak-st:dur-large-r}, we have
\begin{align*}
(\hbox{LHS of \eqref{eq:en-rp0}})
\leq & \iint_{\calD(\tau_{1}, \tau_{2})} \chi_{R_{2}} (\hbox{LHS of \eqref{eq:en-id-rp}}) 
+ \frac{1}{2} \int_{\Gmm^{(out)}_{\tau_{1}}} \frac{r^{p}}{\dvr} (\rd_{v} \psidf)^{2} \\
& +\frac{1}{2} \iint_{\calD(\tau_{1}, \tau_{2})} \Abs{\chi_{R_{2}}' (-\dur) \frac{r^{p}}{\dvr} (\rd_{v} \psidf)^{2} }
 + \frac{1}{2} \iint_{\calD(\tau_{1}, \tau_{2})} \Abs{\chi_{R_{2}}' r^{p} \frac{\rd_{u} \rd_{v} r}{r} \psidf^{2}}.
\end{align*}
Using \eqref{eq:weak-st:me-bnd}, \eqref{eq:weak-st:kpp-bnd}, \eqref{eq:weak-st:dur-large-r}, \eqref{eq:weak-st:mu-large-r}, \eqref{eq:psidf-phidf}, \eqref{eq:psidf-phidf-dv} and the fact that $\mathrm{supp} \, \chi'_{R_{2}} \subseteq \set{R_{2}/2 \leq r \leq R_{2}}$, we may show that 
\begin{align*}
& \hskip-2em
 \iint_{\calD(\tau_{1}, \tau_{2})} \Abs{\chi_{R_{2}}' (-\dur) \frac{r^{p}}{\dvr} (\rd_{v} \psidf)^{2} }
 + \iint_{\calD(\tau_{1}, \tau_{2})} \Abs{\chi_{R_{2}}' r^{p} \frac{\rd_{u} \rd_{v} r}{r} \psidf^{2}} \\
 \leq & C R_{2}^{p+\eta_{0}} \left( \tau_{1}^{-2 \omg + \eta_{0} + 1} \dlt^{2} \eps_{0}^{2} 
 					+ \iint_{\calD(\tau_{1}, \tau_{2}) \cap \set{R_{2} /2 \leq r \leq R_{2}}} r^{-1-\eta_{0}} (-\dur) \left( \kpp^{-1} (\rd_{v} \phidf)^{2} + r^{-2} \phidf^{2} \right) r^{2} \right).
\end{align*}
The spacetime integral on the last line is bounded by $C \dot{E}[\phidf](\tau_{1})+C\tau_{1}^{-2 \omg + \eta_{0} + 1} \dlt^{2} \eps_{0}^{2}$ thanks to \eqref{eq:gauge:eps} and \eqref{eq:energy}, which proves \eqref{eq:en-rp0}.

Next, \eqref{eq:en-rp1} is a clear consequence of the Cauchy--Schwarz inequality, \eqref{eq:weak-st:dvr-large-r} and \eqref{eq:en-rp-G}.

For \eqref{eq:en-rp2}, we use \eqref{eq:c-varpi} and the fact that $\mathrm{supp} \, \chi_{R_{2}} \subseteq \set{r \geq R_{2} / 2}$.

Finally, we prove \eqref{eq:en-rp3}. We start with the pointwise estimate
\begin{equation}  \label{eq:en-rp3-pf}
\Abs{\rd_{v} \left( r^{2} \frac{\rd_{u} \rd_{v} r}{\dvr} \right)} 
\leq C r^{-2}  + C r^{-1} \eps_{0}^{2} + C (\rd_{v} \psidf)^{2}.
\end{equation}
Indeed, we compute
\begin{align*}
\rd_{v} \left( r^{2} \frac{\rd_{u} \rd_{v} r}{\dvr} \right)
= 2 \rd_{v} \left( \left( \varpi - \frac{\e^{2}}{r} \right) \gmm \right)
=  \gmm \kpp^{-1} r^{2} (\rd_{v} \phi)^{2} 
+ 2 \frac{\e^{2}}{r^{2}} \dvr \gmm + 2 \left( \varpi - \frac{\e^{2}}{r} \right) \gmm \dvr^{-1} r (\rd_{v} \phi)^{2}.
\end{align*}
Hence, using \eqref{eq:weak-st:me-bnd}, \eqref{eq:weak-st:kpp-bnd}, \eqref{eq:weak-st:dvr-large-r}, \eqref{eq:weak-st:dur-large-r} and \eqref{eq:weak-st:mu-large-r}, we may estimate
\begin{align*}
	\Abs{\rd_{v} \left( r^{2} \frac{\rd_{u} \rd_{v} r}{\dvr} \right)} 
	\leq & C r^{-2} + C r^{2} (\rd_{v} \phibg)^{2} + C r^{2} (\rd_{v} \phidf)^{2} \\
	\leq & C r^{-2} + C r^{2} (\rd_{v} \phibg)^{2} + C \phidf^{2} + C (\rd_{v} (r \phidf - \psidf) )^{2} + C (\rd_{v} \psidf )^{2} .
\end{align*}
Applying \eqref{eq:dlt-adm:dvphi}, \eqref{eq:gauge:eps}, \eqref{eq:energy-phidf-pt} and \eqref{eq:psidf-phidf-dv}, we obtain \eqref{eq:en-rp3-pf}.

The contribution of the first two terms on the RHS of \eqref{eq:en-rp3-pf} in \eqref{eq:en-rp3} is bounded by the first term on the RHS of \eqref{eq:en-rp3}. For the contribution of the last term, we use \eqref{eq:energy-phidf-pt} and \eqref{eq:psidf-phidf} to estimate
\begin{equation*}
	\iint_{\calD (\tau_{1}, \tau_{2})} \chi_{R_{2}} \Abs{r^{p-3} (\rd_{v} \psidf)^{2} \psidf^{2} } 
\leq 	C \eps_{0}^{2} \iint_{\calD (\tau_{1}, \tau_{2})} \chi_{R_{2}} r^{p-2} (\rd_{v} \psidf)^{2}
\leq 	\frac{C \eps_{0}^{2}}{R_{2}} \iint_{\calD (\tau_{1}, \tau_{2})} \chi_{R_{2}} r^{p-1} (\rd_{v} \psidf)^{2},
\end{equation*}
which concludes the proof. \qedhere
\end{proof}

Using the results established so far, we may now establish \eqref{eq:en-decay} by a pigeonhole argument of Dafermos--Rodnianski \cite{DRNM}.
\begin{proof}[Proof of Proposition~\ref{prop:en-decay}]
We use a hierarchy of $r^{p}$-weighted energy estimates \eqref{eq:en-rp}, where $p = 3-\eta_{0}, 2-\eta_{0}$ and $1-\eta_{0}$. For these values of $p$, observe that
\begin{equation*}
	\tau_{1}^{-2 \omg + \max \set{1+\eta_{0}, p}} \leq \tau_{1}^{-(3 - \eta_{0}) + p}.
\end{equation*}
Therefore, for any $1 \leq \tau_{1} \leq \tau_{2}$ and $p = 3 - \eta_{0}, 2 - \eta_{0}, 1-\eta_{0}$, we have
\begin{equation} \label{eq:en-rp-final}
	E_{p}[\psidf](\tau_{2}) + \int_{\tau_{1}}^{\tau_{2}} E_{p-1}[\psidf](\tau) \, \ud \tau
	\leq C E_{p}[\psidf](\tau_{1}) + C \tau_{1}^{-(3 - \eta_{0}) + p} \dlt^{2} \eps_{0}^{2}.
\end{equation}
Note that we have suppressed the dependence of $C$ on $p$, which due to our choice of $p$ depends only on $\eta_0$.
In the case $p = 3 - \eta_{0}$, by \eqref{eq:Ep-eps0} and \eqref{eq:en-rp-final}, we have $E_{3-\eta_{0}}[\psidf](\tau) \leq C \eps_{0}^{2}$. Moreover, by the pigeonhole principle, for every $\tau \geq 2$ there exists $\overline{\tau} \in [\tau/2, \tau]$ such that
\begin{equation*}
	E_{p-1}[\psidf](\overline{\tau}) \leq \frac{2}{\tau} \int_{\tau/2}^{\tau} E_{p-1}[\psidf](\tau') \, \ud \tau'.
\end{equation*}
Applying \eqref{eq:en-rp-final} with $p-1$ and $p$, we have
\begin{align*}
	E_{p-1}[\psidf](\tau)
	\leq & C E_{p-1}[\psidf](\overline{\tau}) + C \overline{\tau}^{- (3 + \eta_{0}) + (p-1)} \dlt^{2} \eps_{0}^{2} \\
	\leq & \frac{C}{\tau} E_{p}[\psidf](\tau/2) + C \tau^{- (3 + \eta_{0}) + (p-1)} \dlt^{2} \eps_{0}^{2}.
\end{align*}
Plugging in the bound $E_{3-\eta_{0}}[\psidf](\tau) \leq C \eps_{0}^{2}$ with $p = 3 - \eta_{0}$, we obtain $E_{2-\eta_{0}}[\psidf](\tau) \leq C \tau^{-1} \eps_{0}^{2}$. Then plugging in the preceding bound with $p = 2$, we obtain $E_{1-\eta_{0}}[\psidf](\tau) \leq C \tau^{-2} \eps_{0}^{2}$. 

Applying \eqref{eq:en-rp-final} again with $p = 2-\eta_{0}$ and $1-\eta_{0}$, but now using the decay of $E_{p}[\psidf](\tau_{1})$ derived just now, we obtain
\begin{equation} \label{eq:en-decay:pf1}
	\tau_{1}^{2} \int_{\tau_{1}}^{\tau_{2}} E_{-\eta_{0}}[\psidf](\tau) \, \ud \tau 
	+ \tau_{1} \int_{\tau_{1}}^{\tau_{2}} E_{1-\eta_{0}}[\psidf](\tau) \, \ud \tau \leq C \eps_{0}^{2}. 
\end{equation}
Observe that
\begin{align*}
	E_{0}[\psidf](\tau)
	\leq & \int_{\Gmm^{(in)}_{\tau}} \frac{1}{(-\dur)} (\rd_{u} \psidf)^{2} \, \ud u
	+ \sup_{\Gmm^{(in)}_{\tau}} r^{-1} \psidf^{2} \\
	& + \tau_{1}^{\eta_{0}} \int_{\Gmm^{(out)}_{\tau} \cap \set{r \leq \tau_{1}}} \frac{r^{-\eta_{0}}}{\dvr} (\rd_{v} \psidf)^{2} \, \ud v
	+ \tau_{1}^{-1+\eta_{0}}\int_{\Gmm^{(out)}_{\tau} \cap \set{r \geq \tau_{1}}} \frac{r^{1-\eta_{0}}}{\dvr} (\rd_{v} \psidf)^{2} \, \ud v \\
	\leq & (1 + \tau_{1}^{\eta_{0}}) E_{-\eta_{0}}[\psidf](\tau) + \tau_{1}^{-1+\eta_{0}} E_{1-\eta_{0}}[\psidf](\tau).
\end{align*}
Therefore, by \eqref{eq:dotE-Ep} and \eqref{eq:en-decay:pf1}, we arrive at
\begin{align*}
	\int_{\tau_{1}}^{\tau_{2}} \dot{E}[\phidf](\tau) \, \ud \tau
	\leq & C \int_{\tau_{1}}^{\tau_{2}} E_{0}[\psidf](\tau) \, \ud \tau + C \tau_{1}^{-2 \omg + 1 + \eta_{0}} \dlt^{2} \eps_{0}^{2} \\
	\leq & C \tau_{1}^{\eta_{0}} \int_{\tau_{1}}^{\tau_{2}} E_{-\eta_{0}}[\psidf](\tau) \, \ud \tau 
	+ C \tau_{1}^{-1+\eta_{0}} \int_{\tau_{1}}^{\tau_{2}} E_{1-\eta_{0}}[\psidf](\tau) \, \ud \tau 
	+ C \tau_{1}^{-2 + \eta_{0}} \dlt^{2} \eps_{0}^{2} \\	
	\leq & C \tau_{1}^{-2+\eta_{0}} \eps_{0}^{2}.
\end{align*}
By another pigeonhole argument, for every $\tau \geq 2$ we may find $\overline{\tau} \in [\tau/2, \tau]$ such that
\begin{equation*}
	\dot{E}[\phidf](\overline{\tau})
	\leq \frac{2}{\tau} \int_{\tau/2}^{\tau} \dot{E} [\phidf] (\tau') \, \ud \tau'
	\leq C \tau^{-3 + \eta_{0}} \eps_{0}^{2}.
\end{equation*}
Then by the energy inequality \eqref{eq:energy}, we obtain $\dot{E}[\phidf] (\tau) \leq C \tau^{-3+\eta_{0}} \eps_{0}^{2}$, which is the desired decay. \qedhere
\end{proof}

\subsection{Decay of the radiation field difference} \label{subsec:int-char}
Here we use an integration along characteristic argument (similar to \cite{LO1} and Section~\ref{sec.contra.unif}) to upgrade \eqref{eq:phidf-decay-pre} to decay of the radiation field difference $\Phidf = \Phi - \Phibg$. We work under the same assumptions and conventions as in the previous subsection.

We first state the main result of this subsection.
\begin{proposition} \label{prop:psidf-decay}
There exists $R_{3} > R_{0}$ such that for $(u, v) \in \calD \cap \set{r \geq R_{3}}$, we have
\begin{equation} \label{eq:psidf-decay}
	\abs{\psidf}(u, v) \leq C u^{-\bt} \eps_{0},
\end{equation}
for some $\bt = \bt(\omg, \eta_{0}) > 1$ specified in \eqref{eq:psidf-decay-bt}.
In particular, $\Phidf = \Phi - \Phibg$ obeys the decay estimate
\begin{equation} \label{eq:Phidf-decay}
	\abs{\Phidf}(u) \leq C u^{-\bt} \eps_{0}.
\end{equation}
\end{proposition}

As in the previous section, we work with $\psi = r \phi$ and the corresponding difference $\psidf$. We start with the following preliminary decay estimate for $\psidf$, which is a quick consequence of \eqref{eq:phidf-decay-pre}.
\begin{lemma} \label{lem:psidf-decay-pre}
The following estimate holds for $(u, v) \in \set{u \geq 1, \ r(u, v) \geq R_{0}}$:
\begin{equation} \label{eq:psidf-decay-pre}
	r^{-1/2} \abs{\psidf} \leq C u^{-\frac{1}{2}(3-\eta_{0})}  \eps_{0}.
\end{equation}
\end{lemma}
To prove \eqref{eq:psidf-decay-pre}, we simply apply \eqref{eq:phidf-decay-pre} and \eqref{eq:psidf-phidf}; we omit the straightforward details.

Another ingredient we need is an estimate for the contribution of the term $G$ in the wave equation for $\psidf$.
\begin{lemma} \label{lem:int-char-G}
For $(u, v) \in \calD \cap \set{r \geq 20}$, we have
\begin{equation} \label{eq:int-char-G}
	\int_{1}^{u} \abs{G} (u', v) \, \ud u' \leq C u^{-\frac{1}{2} (4-\eta_{0})} r^{-\frac{1}{2}(2+\eta_{0})}  (u, v) \dlt \eps_{0}.
\end{equation}
\end{lemma}
\begin{proof}
By \eqref{eq:gauge:eps}, \eqref{eq:weak-st:me-bnd}, \eqref{eq:weak-st:kpp-bnd}, \eqref{eq:weak-st:dur-large-r}, \eqref{eq:weak-st:mdf}, \eqref{eq:weak-st:kppdf}, \eqref{eq:weak-st:rdf}, \eqref{eq:weak-st:durdf} and \eqref{eq:r-alp-df}, we have
\begin{equation*}
	\Abs{\frac{\rd_{u} \rd_{v} r}{r} -  \frac{\rd_{u} \rd_{v} \rbg}{\rbg}} \leq C r^{-3} \eps_{0} \quad \hbox{ in } \calD \cap \set{r \geq 20}.
\end{equation*}
Combined with \eqref{eq:dlt-adm:phi} and \eqref{eq:weak-st:rdf} (which ensures that $\rbg$ and $r$ are comparable), we have
\begin{equation*}
	\int_{1}^{u} \abs{G}(u', v) \, \ud u'
	\leq C \eps_{0} \int_{1}^{u} r^{-3} \abs{\psibg} (u', v) \, \ud u' 
	\leq C \dlt \eps_{0} \int_{1}^{u} r^{-2} (u')^{-\omg + 1} \min \set{(u')^{-1}, r^{-1}} \, \ud u'.
\end{equation*}
We split the integral on the RHS and estimate it using \eqref{eq:weak-st:dur-large-r} (for the second integral) as follows:
\begin{align*}
& \hskip-2em
\int_{1}^{u} r^{-2} (u')^{-\omg + 1} \min \set{(u')^{-1}, r^{-1}} \, \ud u' \\
\leq & \int_{1}^{\min \set{\frac{1}{2}(v+C_{\gmm_{20}}+20), u}} r^{-3} (u')^{-\omg + 1} \, \ud u' 
 + \int_{\frac{1}{2}(v+C_{\gmm_{20}}+20)}^{\max \set{\frac{1}{2}(v+C_{\gmm_{20}}+20), u}} r^{-2 - \frac{1}{2} \eta_{0}} (u')^{-\omg + \frac{1}{2}\eta_{0}}  \, \ud u' \\
\leq & C r^{-3}(\min\set{\tfrac{1}{2}(v+C_{\gmm_{20}}+20), u}, v) + C (v+C_{\gmm_{20}}+20)^{-\frac{1}{2} (2\omg -\eta_{0})} r^{-\frac{1}{2}(2+\eta_{0})} (u,v).
\end{align*}
On the one hand, $r(\min\set{\tfrac{1}{2}(v+C_{\gmm_{20}}+20), u}, v) \geq r(u, v)$ by monotonicity of $r$. On the other hand, by \eqref{eq:weak-st:v-u-r}, we have $r(\min\set{\tfrac{1}{2}(v+C_{\gmm_{20}}+20), u}, v) \geq C^{-1} (v+C_{\gmm_{20}} + 20)$ and $(v+C_{\gmm_{20}} + 20) \geq C^{-1} u$ in $\calD \cap \set{r \geq 20}$. Since $\omg > 2$, we have
\begin{equation*}
\int_{1}^{u} r^{-2} (u')^{-\omg + 1} \min \set{(u')^{-1}, r^{-1}} \, \ud u' 
\leq C u^{-\frac{1}{2} (4-\eta_{0})} r^{-\frac{1}{2} (2+\eta_{0})}(u, v),
\end{equation*}
which completes the proof of \eqref{eq:int-char-G}. \qedhere
\end{proof}

We are now prepared to give a proof of Proposition~\ref{prop:psidf-decay}.
\begin{proof}[Proof of Proposition~\ref{prop:psidf-decay}]
Let
\begin{equation} \label{eq:psidf-decay-bt}
	\bt = \min \set{\tfrac{1}{2}(3 - \eta_{0}), \tfrac{1}{2}(2 \omg - 2 -  \eta_{0})}.
\end{equation}
By the hypothesis $0 < \eta_{0} < \min \set{2 \omg - 4, 1}$, we have $\bt > 1$.

For any $\tau > 1$, we introduce the quantity
\begin{equation*}
	\calB(\tau) = \sup_{\calD(1, \tau) \cap \set{r \geq R_{3}}} u^{\bt} \abs{\psidf},
\end{equation*}
where $R_{3} \geq R_{0}$ is specified below. Our goal is to show that $\calB(\tau) \leq C \eps_{0}$ with a constant $C$ that depends only on $\omg, \eta_{0}, \ebg, \Lmb$ (in particular, independent of $\tau$).

For $(u,v)\in \calD(1, \tau) \cap \set{r \geq R_{3}}$, we claim that
\begin{equation} \label{eq:psidf-decay:pf1}
	\abs{\rd_{v} \psidf}(u, v) \leq C u^{-\bt} r^{-\frac{1}{2}(2 + \eta_{0})}(u, v) + C u^{-\bt} r^{-2} \calB(\tau).
\end{equation}
To prove \eqref{eq:psidf-decay:pf1}, we first use the wave equation for $\psi$ to estimate
\begin{align*}
	\abs{\rd_{v} \psidf}(u,v)
	\leq & \abs{\rd_{v} \psidf}(1,v) + \int_{1}^{u} \abs{G} (u', v) \, \ud u' \\
	& + \int_{1}^{\min\set{\frac{1}{2}(v + C_{\gmm_{20}}+20), u}} \frac{2 (\varpi - \frac{\e^{2}}{r})}{r^{3}} \kpp (-\dur) \abs{\psidf} (u', v) \, \ud u' \\
	& + \int_{\frac{1}{2}(v + C_{\gmm_{20}}+20)}^{\max\set{\frac{1}{2}(v + C_{\gmm_{20}}+20), u}} \frac{2 (\varpi - \frac{\e^{2}}{r})}{r^{3}} \kpp (-\dur) \abs{\psidf} (u', v) \, \ud u' .
\end{align*}
We bound each term on the RHS. For the first term, we use \eqref{eq:gauge:dvrphi} and the simple fact that $r(1, v) \geq C^{-1} \max\set{u, r(u, v)}$ (which follows from \eqref{eq:weak-st:v-u-r}) to estimate
\begin{equation*}
	\abs{\rd_{v} \psidf}(1, v) \leq C r^{-\omg}(1, v) \eps_{0} \leq C u^{-\omg+1+\frac{1}{2}\eta_{0}} r^{-\frac{1}{2}(2+\eta_{0})}(u, v) \eps_{0},
\end{equation*}
which is acceptable. For the second term, we apply Lemma~\ref{lem:int-char-G}. For the third term, we use Lemma~\ref{lem:psidf-decay-pre} together with \eqref{eq:weak-st:me-bnd}, \eqref{eq:weak-st:kpp-bnd}, \eqref{eq:weak-st:dur-large-r} to estimate
\begin{align*}
& \hskip-2em
\int_{1}^{\min\set{\frac{1}{2}(v + C_{\gmm_{20}}+20), u}} \frac{2 (\varpi - \frac{\e^{2}}{r})}{r^{3}} \kpp (-\dur) \abs{\psidf} (u', v) \, \ud u' \\
\leq & C \eps_{0} \int_{1}^{\min\set{\frac{1}{2}(v + C_{\gmm_{20}}+20), u}} \frac{1}{(u')^{-\frac{1}{2}(3-\eta_{0})}}\frac{(\varpi - \frac{\e^{2}}{r})}{r^{5/2}} \kpp (-\dur) (u', v) \, \ud u' \\
\leq & C r^{-\frac{5}{2}}(\min\set{\tfrac{1}{2}(v + C_{\gmm_{20}}+20), u}, v) \eps_{0}
\leq C u^{-(3-\eta_{0})/2} r^{-(2+\eta_{0})/2}(u,v) \eps_{0},
\end{align*}
where we used \eqref{eq:weak-st:v-u-r} and the monotonicity of $r$ to get $r(\min\set{\tfrac{1}{2}(v + C_{\gmm_{20}}+20), u}, v)\geq C^{-1}\max\{u, r(u,v)\}$ in the final inequality.
For the last term, we estimate
\begin{align*}
& \hskip-2em
\int_{\frac{1}{2}(v + C_{\gmm_{20}}+20)}^{\max\set{\frac{1}{2}(v + C_{\gmm_{20}}+20), u}} \frac{2 (\varpi - \frac{\e^{2}}{r})}{r^{3}} \kpp (-\dur) \abs{\psidf} (u', v) \, \ud u' \\
\leq & C (v + C_{\gmm_{20}}+20)^{-\bt} \calB(\tau) \int_{\frac{1}{2}(v + C_{\gmm_{20}}+ 20)}^{\max\set{\frac{1}{2}(v + C_{\gmm_{20}}+20), u}} \frac{(\varpi - \frac{\e^{2}}{r})}{r^{3}} \kpp (-\dur)  (u', v) \, \ud u' \\
\leq & C u^{-\bt} r^{-2}(u, v) \calB(\tau) .
\end{align*}
In the last inequality, we used the fact that $v+C_{\gmm_{20}} + 20 \geq u$ in $\calD \cap \set{r \geq 20}$ by \eqref{eq:weak-st:v-u-r}.

For any $(u, v) \in \calD(0, \tau) \cap \set{r \geq R_{3}}$, let $v_{R_{3}}(u)$ be defined by $r(u, v_{R_{3}}(u)) = R_{3}$. Since $R_{3} \geq R_{0}$, $(u, v_{R_{3}}(u)) \in \calD(0, \tau) \cap \set{r \geq R_{3}}$. Integrating \eqref{eq:psidf-decay:pf1} from $v_{R_{3}}(u)$ to $v$, and using \eqref{eq:weak-st:dvr-large-r}, we obtain
\begin{align*}
	\abs{\psidf} (u,v) 
\leq & \abs{\psidf}(u, v_{R_{3}}(u)) + \int_{v_{R_{3}}(u)}^{v} \abs{\rd_{v} \psidf}(u, v') \, \ud v' \\
\leq & C u^{-\frac{1}{2}(3-\eta_{0})} R_{3}^{1/2} \eps_{0} + C u^{-\bt} R_{3}^{-\eta_{0}/2} \eps_{0} + C u^{-\bt} R_{3}^{-1} \calB(\tau).
\end{align*}
In the last inequality, we used Lemma~\ref{lem:psidf-decay-pre} for the first term. By the definition of $\calB(\tau)$, we have
\begin{equation*}
	\calB(\tau) \leq C (R_{3}^{1/2} + R_{3}^{-\eta_{0}/2}) \eps_{0} + C R_{3}^{-1} \calB(\tau).
\end{equation*}
Choosing $R_{3}$ large enough (possibly depending only on $\omg, \eta_{0}, \ebg, \Lmb$ by our conventions), we may absorb the last term into the LHS. Since $\tau > 1$ is arbitrary, \eqref{eq:psidf-decay} follows. \qedhere

\end{proof}

\subsection{Proof of Theorem~\ref{thm:L-st-ch}} \label{subsec:main-st-pf} 
We may finally complete the proof of Theorem~\ref{thm:L-st-ch}.
\begin{proof}[Proof of Theorem~\ref{thm:L-st-ch}]
Note that \eqref{eq:L-st-ch:r} is an immediate consequence of \eqref{eq:weak-st:rdf}; therefore, it only remains to establish \eqref{eq:L-st-ch:L}.
In this proof, we abbreviate $\mathfrak{L} = \mathfrak{L} \restriction_{\calD}$ and $\overline{\mathfrak{L}} = \overline{\mathfrak{L}} \restriction_{\calD}$.
In the future-normalized coordinates $(u, v)$, the quantities $\mathfrak{L}$ and $\overline{\mathfrak{L}}$ can be expressed as
\begin{equation*}
	\mathfrak{L} = - \int_{1}^{\infty} 2 M(u) \Phi(u) \, \ud u, \quad
	\overline{\mathfrak{L}} = - \int_{1}^{\infty} 2 \overline{M}(u) \overline{\Phi}(u) \, \ud u,
\end{equation*}
and thus
\begin{align*}
	\mathfrak{L} - \overline{\mathfrak{L}} 
	=& - \int_{1}^{\infty} 2 M(u) \Phi(u) - 2 \overline{M}(u) \Phibg(u) \, \ud u \\
	=& - \int_{1}^{\infty} 2 M(u) (\Phi - \Phibg)(u)  \, \ud u - \int_{1}^{\infty} 2 (M - \overline{M})(u) \Phibg (u)  \, \ud u.
\end{align*}
By $\varpibg_{f} \leq \rbg_{\EH} = 1$, \eqref{eq:dlt-adm:m}, \eqref{eq:gauge:eps} and \eqref{eq:weak-st:mdf}, we have
\begin{equation*}
\abs{M - \overline{M}}(u) \leq C \eps_{0}, \quad \abs{M}(u) \leq C.
\end{equation*}
On the other hand, by \eqref{eq:dlt-adm:phi} and \eqref{eq:Phidf-decay}, we have
\begin{equation*}
\abs{\Phi - \Phibg}(u) \leq C u^{-\bt} \eps_{0}, \quad
\abs{\Phibg}(u) \leq C u^{-\omg + 1} \dlt
\end{equation*}
where $\bt, \omg-1 > 1$. As a consequence, $M(u) (\Phi - \Phibg)(u)$ and $(M - \overline{M})(u) \Phibg (u)$ are integrable over $u \in [1, \infty)$, with integrals bounded above by $C\ep_0$, and the desired conclusion follows. \qedhere
\end{proof}

\section{Instability of no backscattering: Proof of Theorem~\ref{thm:instability}} \label{sec:instability}
The goal of this section is to establish Theorem~\ref{thm:instability} concerning instability of the condition $\mathfrak{L}_{(\omg_{0}) \infty} = 0$ when $\omg_0\geq 3$. The section is structured as follows:
\begin{itemize}
\item \textbf{Section~\ref{sec:instability.ideas}.} We briefly discuss the ideas of the proof.
\item \textbf{Section~\ref{subsec:const-eq}.} We discuss some properties of the Einstein--Maxwell--(real)--scalar--field constraint equations in spherical symmetry, and solve the constraint equations near the asymptotically flat end. 
\item \textbf{Section~\ref{subsec:inst-key}.} We carry out the proof of Theorem~\ref{thm:instability}.
\end{itemize}

\subsection{Ideas of the proof}\label{sec:instability.ideas}
Given an initial data set $\overline{\Tht}$ with $\mathfrak{L}_{(\omg_{0}) \infty} = 0$ for $\omg_0\geq 3$, the goal is to construct perturbations $\Tht_{\eps}$ (for small $\eps$) of the initial data $\overline{\Tht}$, such that $\mathfrak{L}[\Tht_\ep] -\mathfrak{L}[\overline{\Tht}] \geq c \eps$ for some $c > 0$. The main idea of the proof is based on \cite[Theorem~1.5]{LO.instab}, except that now we also need to control all the contributions coming from both the background and the perturbed geometry. More precisely, the perturbations that we construct have the following features:
\begin{enumerate}
\item We choose the perturbation to be small, smooth and essentially compactly supported (more precisely, only the $f, h$ components have a small polynomial tail, which is forced by the constraint equation); therefore it is small in every topology $d_{k, \omg}^{+}$ (as long as the background initial data set is also $C^k$). 
\item Moreover, we take the $\phi, \dot{\phi}$ components of the perturbation to be outgoing and place it in a region where $r$ is sufficiently large; as a consequence, the leading order effect of the perturbation is localized in a large-$r$ region, which is close to the Minkowski spacetime (by asymptotic flatness) and thus explicit computation can be performed. 
\item Finally, we use the stability theorems established in Sections~\ref{sec:extr} and \ref{sec:L-stability} to show that the contribution of the remainder is small. 
\end{enumerate}

A conceptual difference from \cite[Proof of Theorem~1.5]{LO.instab}, however, is that here we \emph{directly} perturb the Cauchy data, whereas in \cite{LO.instab} the perturbation was constructed from a characteristic data set; see Section~\ref{subsec:const-eq}. When applied to the linear setting of \cite{LO.instab}, we remark that this feature trivializes the ``proof'' of compact support of the perturbation on $\Sgm_{0}$. On the other hand, in order to be able to perturb the Cauchy data, we will need to solve the constraint equations. This will be achieved in Section~\ref{subsec:const-eq}.

\subsection{Constraint equations on a Cauchy hypersurface} \label{subsec:const-eq}
In general, prescription of a Cauchy data set for the Einstein equation is not a straightforward task due to the presence of the constraint equations, which by themselves form a nontrivial system of underdetermined PDEs. Fortunately, the constraint equations \eqref{ham.con} and \eqref{mom.con} reduce to a system of ODEs under spherical symmetry, and at least in the regime we are interested in, it is not difficult to identify the part (which we call the \emph{reduced initial data set}) of the Cauchy data set which can (essentially) be freely prescribed. What we present below relies on using the modified mass and requires $(1-\f{2\varpi}{r}+\f{\e^2}{r^2})$ to be non-zero. It works near an asymptotically flat end, which is sufficient for our purposes in Section~\ref{subsec:inst-key}.

Let $\Tht = (r, f, h, \ell, \phi, \dot{\phi}, \e)$ be an admissible Cauchy data set on $\Sgm_{0}$. We begin by introducing a double null coordinate system $(U, V)$ on a (local) globally hyperbolic future development of $\Sgm_{0}$ normalized by the condition $- \frac{\ud U}{\ud \rho} = \frac{\ud V}{\ud \rho} = 1$ on $\Sgm_{0}$. Consider now the following subset of $\Sgm_{0}$:
\begin{equation*}
\Sgm_{0}^{+} = \set{\rho \in \Sgm_{0} : \forall \rho' \geq \rho, \ r(\rho') \geq 10 \max\set{\varpi_{i}, \abs{\e}}, \ \rd_{U} r \restriction_{\Sgm_{0}}(\rho') < 0, \  \rd_{V} r \restriction_{\Sgm_{0}}(\rho') > 0}.
\end{equation*}
Clearly, $\Sgm_{0}^{+}$ is a half-open interval with a finite endpoint on the left, which we denote by $\rho_{\ast}^{+}$. 

A \emph{reduced initial data set} on $\Sgm_{0}^{+}$ corresponding to $\Tht$ consists of four numbers $r_{\ast}, \psi_{\ast}, \varpi_{\ast}, \e \in \bbR$ and four functions $\dvrd, \durd, \alpd, \btd : \Sgm_{0}^{+} \to \bbR$ given by
\begin{equation} \label{eq:reduced-data-set}
\begin{gathered}
	r_{\ast} = r (\rho^{+}_{\ast}), \quad
	\psi_{\ast} = \psi (\rho^{+}_{\ast}), \quad
	\varpi_{\ast} = \varpi(\rho^{+}_{\ast}), \\
	\dvrd (\rho) = \rd_{V} r \restriction_{\Sgm_{0}^{+}} (\rho), \quad 
	\durd (\rho) = \rd_{U} r \restriction_{\Sgm_{0}^{+}} (\rho), \quad 
	\alpd (\rho) = \rd_{V} \psi \restriction_{\Sgm_{0}^{+}} (\rho), \quad 
	\btd (\rho) = \rd_{U} \psi \restriction_{\Sgm_{0}^{+}} (\rho).
\end{gathered}
\end{equation}
Observe that these numbers and functions obey the following conditions:
\begin{equation} \label{eq:reduced-data-cond}
	r_{\ast} \geq 10 \max \set{\lim_{\rho \to \infty} \varpi(\rho), \abs{\e}}, \quad
	r_{\ast} > 0, \quad \varpi_{\ast} > 0, \quad
	\durd < 0, \quad
	\dvrd > 0.
\end{equation}

It turns out that the conditions in \eqref{eq:reduced-data-cond} are enough for the converse direction to hold. More precisely, let us call a collection $(r_{\ast}, \psi_{\ast}, \varpi_{\ast}, \e, \dvrd, \durd, \alpd, \btd)$ of four numbers and four functions on $\Sgm_{0}^{+}$ a \emph{reduced initial data set} on $\Sgm_{0}^{+}$ if they satisfy \eqref{eq:reduced-data-cond}. Then we have the following statement:
\begin{proposition} \label{prop:reduced-data}
Given a reduced initial data set $(r_{\ast}, \psi_{\ast}, \varpi_{\ast}, \e, \dvrd, \durd, \alpd, \btd)$ on a half-open interval $\Sgm_{0}^{+} = [\rho_{\ast}^{+}, \infty)$ satisfying \eqref{eq:reduced-data-cond} (with $\varpi$ as in \eqref{eq:reduced-data:varpi} below) and
\begin{equation*}
	\dvrd, \durd \in C^{k+1}(\Sgm_{0}^{+}), \quad
	\alpd, \btd \in C^{k}(\Sgm_{0}^{+}) \qquad \hbox{ for some nonnegative } k \in \bbZ,
\end{equation*}
there uniquely exists a corresponding Cauchy data set $(r, f, h, \ell, \phi, \dot{\phi}, \e)$ on $\Sgm_{0}^{+}$ such that
\begin{equation*}
	r \in C^{k+2}(\Sgm_{0}^{+}), \quad
	\ell, f, \phi \in C^{k+1}(\Sgm_{0}^{+}), \quad
	h, \dot{\phi} \in C^{k}(\Sgm_{0}^{+}).
\end{equation*}
\end{proposition}
\begin{remark} \label{rem:reduced-data}
To have $\rd_{\rho} f, \ell \in C^{k}(\Sgm_{0}^{+})$, we need continuity of one higher order derivatives of $\dvrd, \durd$, which are at the level of $\rd_{\rho} r, \ell$. In our application below, however, this discrepancy does not arise thanks to the particular form of our perturbation; see Lemma~\ref{lem:inst-cont}.
\end{remark}
\begin{proof}
The idea is to show that a reduced initial data set gives rise to an initial data set of the system \eqref{eq:EMSF-r-phi-m} for $r, \phi$ and $\varpi$. Then by the equivalence of \eqref{eq:EMSF-r-phi-m} and \eqref{eq:EMSF-wave}--\eqref{eq:EMSF-ray-orig}, we may define the desired Cauchy data set $(r, f, h, \ell, \phi, \dot{\phi}, \e)$ on $\Sgm_{0}^{+}$.

\pfstep{Step~1: Prescribing $(r, f, h, \ell, \phi, \dot{\phi}, \e)$} We begin by prescribing $\e$ as the given real number.

Next, in view of \eqref{eq:reduced-data-set}, the coordinate normalization conditions $- \frac{\ud U}{\ud \rho} = \frac{\ud V}{\ud \rho} = 1$ on $\Sgm_{0}$ and Lemma~\ref{lem:cauchy-to-char}, we define $r$ and $\psi = r \phi$ on $\Sgm_{0}^{+}$ by the equations
\begin{align}
	\rd_{\rho} r
	= & \rd_{V} r - \rd_{U} r = \underline{\dvr} - \underline{\dur} , \label{eq:reduced-data:r}\\
	\rd_{\rho} \psi
	= & \rd_{V} \psi - \rd_{U} \psi = \underline{\alp} - \underline{\bt} \label{eq:reduced-data:psi}, 
\end{align}
and the initial conditions $r(\rho^{+}_{\ast}) = r_{\ast}$, $\psi(\rho^{+}_{\ast}) = \psi_{\ast}$. Note that the signs conditions in \eqref{eq:reduced-data-cond} guarantee that $r$ is increasing. In particular, $r>0$, as is required. Then we define $\varpi$ by solving the ODE
\begin{equation} \label{eq:reduced-data:varpi}
\left\{
\begin{aligned}
	\rd_{\rho} \varpi 
	= & \rd_{V} \varpi - \rd_{U} \varpi 
	=  \frac{1}{2} \left( 1 - \frac{2 \varpi}{r} + \frac{\e^{2}}{r^{2}} \right) \left( \underline{\dvr} \left( \underline{\dvr}^{-1} \underline{\alp} - r^{-1} \psi \right)^{2} + \left(- \underline{\dur}) ((-\underline{\dur} )^{-1} \underline{\bt} + r^{-1} \psi \right)^{2} \right), \\
	\varpi(\rho^{+}_{\ast}) = & \varpi_{\ast}.
\end{aligned}
\right.
\end{equation}
Clearly, after $r$ and $\psi$ on $\Sgm_{0}$ are determined, the ODE for $\varpi$ is linear and hence a unique solution exists on $\Sgm_{0}^{+}$. Moreover, it is easy to see that with the sign conditions in \eqref{eq:reduced-data-cond}, $\varpi$ is increasing. Hence, the first condition in \eqref{eq:reduced-data-cond}, together with the monotonicity of $r$ and $\varpi$ that we just established, guarantees that $1-\frac{2 \varpi}{r} + \frac{\e^{2}}{r^{2}}>0$. Define $f$ by $f> 0$ and
\begin{equation} \label{eq:reduced-data:f}
	f^{2} = \frac{4 \rd_{V} r (-\rd_{U} r)}{1-\frac{2 \varpi}{r} + \frac{\e^{2}}{r^{2}}}=\frac{4 \underline{\dvr} (-\underline{\dur})}{1-\frac{2 \varpi}{r} + \frac{\e^{2}}{r^{2}}}.
\end{equation}
Notice that by $\underline{\dvr}>0$, $\underline{\dur}<0$ (see \eqref{eq:reduced-data-cond}) and $1-\frac{2 \varpi}{r} + \frac{\e^{2}}{r^{2}}>0$, the RHS is positive. Hence, $f$ is well-defined and is positive as required. In a manner consistent with Lemma~\ref{lem:cauchy-to-char}, we then define $\ell$ and $\dot{\phi}$ by 
\begin{equation}\label{eq:reduced-data:ell-phidot}
\ell= \f{r}{f}(\underline{\dvr}+\underline{\dur}),\quad \dot{\phi}=\f{1}{rf}(\underline{\alp}-r^{-1}\underline{\dvr}\psi+\underline{\bt}-r^{-1}\underline{\dur}\psi).
\end{equation}
These are well-defined since $r>0$, $f>0$.
Finally, define $h$ by
\begin{equation} \label{eq:reduced-data:h}
h=\f{rf^2}{\underline{\dvr} - \underline{\dur}}\left(\f{\rd_\rho \ell}{r^2}-\f{\ell (\underline{\dvr} - \underline{\dur})}{r^3}+r^{-1}\dot{\phi}\left((\underline{\alp}-r^{-1}\underline{\dvr}\psi)-(\underline{\bt}-r^{-1}\underline{\dur}\psi)\right)\right).
\end{equation}
This is well-defined since $\underline{\dvr} - \underline{\dur}>0$ by \eqref{eq:reduced-data-cond}.

\pfstep{Step~2: Checking the constraint equations} We need to check that $(r, f, h, \ell, \phi, \dot{\phi}, \e)$ defined as above indeed satisfies the constraint equations \eqref{ham.con} and \eqref{mom.con}. For convenience, we introduce for the rest of this proof the notation $'=\f{d}{d\rho}$. One computes that
$$R_{\hat{g}}=-\f{4r''}{f^2 r}+ \f{4f'r'}{f^3 r}-\f{2(r')^2}{f^2r^2}+\f 2{r^2},\quad -|\hat{k}|_{\hat{g}}+(\mbox{tr}_{\hat{g}}\hat{k})^2=\f{4h\ell}{f^2 r^2}+\f{2\ell^2}{r^4}.$$
Therefore, \eqref{ham.con} can be given as follows
\begin{equation}\label{ham.con.2}
-\f{2r''}{f^2 r}+ \f{2f'r'}{f^3 r}-\f{(r')^2}{f^2r^2}+\f {1}{r^2}+\f{2h\ell}{f^2 r^2}+\f{\ell^2}{r^4}=\dot{\phi}^2+\f {1}{f^2} (\phi')^2-\f{\e^2}{r^4}.
\end{equation}
One also computes that \eqref{mom.con} is given by
\begin{equation}\label{mom.con.2}
\f{hr'}{rf^2}-\f{\ell'}{r^2}+\f{\ell r'}{r^3}=\dot{\phi}\phi'.
\end{equation}

By \eqref{eq:reduced-data:h}, and expressing $\underline{\alp}$, $\underline{\bt}$, $\underline{\dvr}$ and $\underline{\dur}$ in terms of $r$, $\ell$, $\phi$ and $\dot{\phi}$ using \eqref{eq:reduced-data:r}, \eqref{eq:reduced-data:psi} and \eqref{eq:reduced-data:ell-phidot}, one sees that \eqref{mom.con} is satisfied.

To check \eqref{ham.con.2}, we first express $\varpi'$ using $\varpi=\f r2+\f{\e^2}{2r}+\f 12\f{\ell^2}{r}-\f{r(r')^2}{2f^2}$ (which follows from \eqref{eq:reduced-data:f}) and apply \eqref{mom.con.2}:
\begin{equation*}
\begin{split}
\varpi'=&\f{r'}{2}-\f{\e^2 r'}{2r^2}-\f 12\f{\ell^2 r'}{r^2}+ \f{\ell \ell'}{r}-\left(\f{r(r')^2}{2f^2}\right)'\\
=&\f{r'}{2}-\f{\e^2 r'}{2r^2}-\f 12\f{\ell^2 r'}{r^2}+ \ell r\left(\f{hr'}{rf^2}+\f{\ell r'}{r^3}-\dot{\phi}\phi'\right)-\left(\f{(r')^3}{2f^2}+\f{rr'r''}{f^2}-\f{r(r')^2f'}{f^3}\right)\\
=&\f{r'}{2}-\f{\e^2 r'}{2r^2}+\f 12\f{\ell^2 r'}{r^2}+ \f{h\ell r'}{f^2}-\ell r\dot{\phi}\phi'-\f{(r')^3}{2f^2}-\f{rr'r''}{f^2}+\f{r(r')^2f'}{f^3}.
\end{split}
\end{equation*}
Combining this with \eqref{eq:reduced-data:varpi}, we thus obtain
\begin{equation*}
\begin{split}
0=&\f{r'}{2}-\f{\e^2 r'}{2r^2}+\f 12\f{\ell^2 r'}{r^2}+ \f{h\ell r'}{f^2}-\ell r\dot{\phi}\phi'-\f{(r')^3}{2f^2}-\f{rr'r''}{f^2}+\f{r(r')^2f'}{f^3}\\
&-\f{(r'-\f{f\ell}{r})(r'+\f{f\ell}{r})}{2 f^2} \left( \f {r^2}{2( r'+\f{f\ell}{r})} \left( -\phi'+f\dot{\phi} \right)^{2} - \f{r^2}{2(- r'+\f{f\ell}{r})} \left( \phi'+f\dot{\phi} \right)^{2} \right)\\
=&\f{r'}{2}-\f{\e^2 r'}{2r^2}+\f 12\f{\ell^2 r'}{r^2}+ \f{h\ell r'}{f^2}-\f{(r')^3}{2f^2}-\f{rr'r''}{f^2}+\f{r(r')^2f'}{f^3}-\f{r^2r'}{2}\left(\dot{\phi}^2+\f{1}{f^2}(\phi')^2\right),
\end{split}
\end{equation*}
which is clearly equivalent to \eqref{ham.con.2} since $r,\, r'>0$.

\pfstep{Step~3: Regularity and uniqueness of the Cauchy data set}
The regularity assertions follow directly from the definitions \eqref{eq:reduced-data:r}, \eqref{eq:reduced-data:psi}, \eqref{eq:reduced-data:varpi}, \eqref{eq:reduced-data:f}, \eqref{eq:reduced-data:ell-phidot} and \eqref{eq:reduced-data:h}. Finally, it is easy to see that in order to satisfy the constraint equations, the definition as in Step~1 is the only choice, and thus the prescription of the Cauchy data set is uniquely determined by reduced initial data set.
\qedhere
\end{proof}

\subsection{Construction of Cauchy data perturbation} \label{subsec:inst-key}
Let $\overline{\Tht} = (\rbg, \fbg, \hbg, \ellbg, \phibg, \dphibg, \ebg)$ be an $\omg_{0}$-admissible initial data set on $\Sgm_{0}$ for some $\omg_{0} \geq 3$ such that the corresponding maximal globally hyperbolic future development $(\gbg, \phibg, \Fbg)$ obeys
\begin{equation*}
\overline{\mathfrak{L}}_{(\omg_{0}) \infty} = \overline{\mathfrak{L}}_{(\omg_{0}) 0} + \overline{\mathfrak{L}} = 0.
\end{equation*}
Here, we have used the shorthand $\overline{\mathfrak{L}}_{(\omg_{0}) \infty}=\mathfrak{L}_{(\omg_{0}) \infty}[\overline{\Tht}]$, etc. We will use this convention for the rest of this section. In what follows, we refer to $(\gbg, \phibg, \Fbg)$ as the \emph{background solution}.

Let $\rho_*>1$ be a sufficiently large number such that $\rbg_*=\rbg(\rho_*)> 20\max \set{\varpibg_{i}, \ebg}$. Define $\Sgm_{0}^{+}\subset \Sgm_0$ to be the set
\begin{equation*}
\Sgm_{0}^{+} = \set{\rho \in \Sgm_{0} : \forall \rho' \geq \rho, \ r(\rho') \geq r_*, \ \rd_{U} r \restriction_{\Sgm_{0}}(\rho') < 0, \  \rd_{V} r \restriction_{\Sgm_{0}}(\rho') > 0}.
\end{equation*}

For a large number $R_{pert} > \rbg_*$ to be specified below, let $\rho_{pert}$ be the unique $\rho \in \Sgm_{0}^{+}$ such that $\rbg(\rho_{pert}) = R_{pert}$. We divide
\begin{equation*}
	\Sgm_{0}^{+} = \Sgm_{0, pert}^{+} \cup (\Sgm_{0}^{+} \setminus \Sgm_{0, pert}^{+}), \quad \hbox{where }
	\Sgm_{0, pert}^{+} = \Sgm_{0}^{+} \cap \set{\rho \geq \rho_{pert}}.
\end{equation*}
Let $\chi$ be a smooth function that satisfies the following properties:
\begin{equation} \label{eq:pert-chi}
	\mathrm{supp} \, \chi \subseteq [0, 1], \quad
	\sup \abs{\chi} \leq 1, \quad
	\int_{0}^{1} \chi(\rho) \, \ud \rho = 0 , \quad
	\int_{0}^{1} \int_{0}^{\rho} \chi(\rho') \, \ud \rho' \, \ud \rho\geq 2 c_{0}.
\end{equation}
for some fixed universal constant $c_{0} > 0$ (say $c_{0} = \frac{1}{20}$). We also introduce the notation 
\begin{equation*}
\tilde{\chi}(\rho) = \int_{-\infty}^{\rho} \chi(\rho') \, \ud \rho'.
\end{equation*}
We define a one-parameter family $\Tht_{\eps}$ of admissible initial data sets in the following fashion:
\begin{itemize}
\item On $\Sgm_{0} \setminus \Sgm_{0}^{+}$, let $\Tht_{\eps} = \overline{\Tht}$ for all $\eps$.
\item On $\Sgm_{0}^{+}$, let $\Tht_{\eps}$ be determined by the following reduced data set:
\begin{gather*}
	r_{(\eps)}(\rho^{+}_{\ast}) = \rbg(\rho^{+}_{\ast}), \quad
	\psi_{(\eps)}(\rho^{+}_{\ast}) = \psibg(\rho^{+}_{\ast}), \quad
	\varpi_{(\eps)}(\rho^{+}_{\ast}) = \varpibg(\rho^{+}_{\ast}), \quad
	\e_{(\eps)} = \ebg, \\
	\rd_{U} r_{(\eps)} \restriction_{\Sgm_{0}^{+}}(\rho) = \rd_{U} \rbg \restriction_{\Sgm_{0}^{+}}(\rho), \quad 
	\rd_{V} r_{(\eps)} \restriction_{\Sgm_{0}^{+}}(\rho) = \rd_{V} \rbg \restriction_{\Sgm_{0}^{+}}(\rho), \\
	\rd_{U} \psi_{(\eps)} \restriction_{\Sgm_{0}^{+}}(\rho) = - \eps \chi(\rho - \rho_{pert}) + \rd_{U} \psibg \restriction_{\Sgm_{0}^{+}}(\rho), \quad
	\rd_{V} \psi_{(\eps)} \restriction_{\Sgm_{0}^{+}}(\rho) = \rd_{V} \psibg \restriction_{\Sgm_{0}^{+}}(\rho),
\end{gather*}
where $(\rbg (\rho^{+}_{\ast}), \psibg (\rho^{+}_{\ast}), \varpibg(\rho^{+}_{\ast}), \ebg, \rd_{U} \rbg \restriction_{\Sgm_{0}^{+}}, \rd_{V} \rbg \restriction_{\Sgm_{0}^{+}}, \rd_{U} \psibg \restriction_{\Sgm_{0}^{+}}, \rd_{V} \psibg \restriction_{\Sgm_{0}^{+}})$ is the reduced data set corresponding to $\overline{\Tht}$ on $\Sgm_{0}^{+}$. 
\end{itemize}
By definition, the reduced data set for $\Tht_{\eps}$ only differs with that for $\overline{\Tht}$ in $\Sgm_{0, pert}^{+}$. Moreover, they only differ in the $\rd_{U} \psi_{(\eps)}$ component, by $\eps \chi(\rho - \rho_{pert})$. Since $\rd_{\rho} \psi_{(\eps)}(\rho) = (\rd_{V} - \rd_{U} ) \psi_{(\eps)} \restriction_{\Sgm_{0}^{+}} (\rho)$, we have
\begin{equation} \label{eq:inst:psi-eps}
	(\psi_{(\eps)} - \psibg )(\rho) = \eps \tilde{\chi}(\rho-\rho_{pert}).
\end{equation}
For $\eps>0$ sufficiently small, \eqref{eq:reduced-data-cond} indeed holds: $r(\rho_*) > \max \set{\varpibg_{i}, \ebg}$ holds for sufficiently small $\eps>0$ by standard stability results for ODEs, and the other inequalities in \eqref{eq:reduced-data-cond} are easily verified. We can therefore apply Proposition~\ref{prop:reduced-data} to obtain a unique Cauchy initial data set $\Tht_\eps$ on $\Sgm_0$.

Let us collect a few obvious facts about $\Tht_\eps$:
\begin{lemma}\label{lem:Thtep.facts}
The one-parameter family $\Tht_{\eps}$ satisfies the following properties:
\begin{itemize}
\item $\Tht_{0} = \overline{\Tht}$, 
\item $\mathfrak{L}_{(\omg_{0}) 0}[\Tht_{\eps}] = \mathfrak{L}_{(\omg_{0}) 0}[\overline{\Tht}]$ and $\mathfrak{L}'_{(\omg_{0}) 0}[\Tht_{\eps}] = \mathfrak{L}'_{(\omg_{0}) 0}[\overline{\Tht}]$
for each $\eps \in \bbR$, 
\item $\Tht_{\eps} = \overline{\Tht}$ in $\set{\rho \in \Sgm_{0} : \rho < \rho_{pert}}$, and
\item $\phi_{\eps}=\phibg$, $f_{\eps}\dot{\phi}_{\eps}=\fbg \overline{\dot{\phi}}$ in $\set{\rho \in \Sgm_{0} : \rho > \rho_{pert}+1}$.
\end{itemize}
\end{lemma}
Moreover, the following lemma holds:
\begin{lemma} \label{lem:inst-cont}
For each $\eps \in \bbR$ sufficiently small, $\Tht_{\eps}$ is an $\omg_{0}$-admissible Cauchy data set. Moreover, the mapping $\eps \mapsto \Tht_{\eps}$ is continuous with respect to the metric $d^{+}_{2, \omg}$ for any $\omg > 2$.
\end{lemma}
\begin{proof}
For concreteness, we prove continuity of $\eps \mapsto \Tht_{\eps}$ at $\eps = 0$; as we will see, the general case can be handled by essentially the same proof. Our goal is to show that for $\eps$ sufficiently small, $d_{2, \omg} (\Tht_{\eps}, \overline{\Tht})\leq C_{R_{pert}, \overline{\Tht}} \eps$, where $C_{R_{pert}, \overline{\Tht}}$ may depend both on the background data set $\overline{\Tht}$ and the choice of $R_{pert}$. This clearly implies continuity of $\eps \mapsto \Tht_{\eps}$ at $\eps = 0$. Using moreover that $\phi_{\eps}-\phibg$ and $f_{\eps}\dot{\phi}_{\eps}-\fbg \overline{\dot{\phi}}$ are compactly supported (Lemma~\ref{lem:Thtep.facts}), this implies that $\Tht_{\eps}$ is an $\omg_{0}$-admissible Cauchy data set.

Let $\eps \neq 0$ be fixed, and denote by $(f, r, h, \ell, \phi, \dot{\phi}, \e)$ the components of the Cauchy data set $\Tht_{\eps}$. Note that $\Tht_{\eps} = \overline{\Tht}$ on $\Sgm_{0} \setminus \Sgm_{0, pert}^{+}$ and $(r, f \ell, \e) = (\rbg, \fbg \ellbg, \ebg)$ on $\Sgm_{0}$.
Hence, we have
\begin{align*}
	d_{2, \omg} (\Tht_{\eps}, \overline{\Tht})
	= & \nrm{\rho \log(f/ \fbg)}_{C^{0}(\Sgm_{0, pert}^{+})}
	 + \sum_{k=1}^{2} \left( \nrm{\rho^{1+k} \rd_{\rho}^{k} \log(f/ \fbg)}_{C^{0}(\Sgm_{0, pert}^{+})}
					+ \nrm{\rho^{1+k} \rd_{\rho}^{k-1} (h - \hbg)}_{C^{0}(\Sgm_{0, pert}^{+})} \right) \\
	& + \nrm{\rho^{\omg} (\phi - \phibg)}_{C^{0}(\Sgm_{0, pert}^{+})} 
	+ \sum_{k=1}^{2} \left( \nrm{\rho^{\omg+k} \rd^{k}_{\rho} (\phi - \phibg)}_{C^{0}(\Sgm_{0, pert}^{+})}
					+ \nrm{\rho^{\omg+k} \rd^{k-1}_{\rho}(f \dot{\phi} - \overline{f}\overline{\dot{\phi}})}_{C^{0}(\Sgm_{0, pert}^{+})} \right).
\end{align*}

Note that by \eqref{eq:pert-chi}, $\phi - \phibg, f \dot{\phi} - \fbg \overline{\dot{\phi}}$ are supported in $\set{\rho \in \Sgm_{0, pert}^{+} : \rho_{pert} \leq \rho \leq \rho_{pert} + 1}$; therefore, we need not worry about the weights and clearly have
\begin{equation*}
\nrm{\rho^{\omg} (\phi - \phibg)}_{C^{0}(\Sgm_{0, pert}^{+})} 
+ \sum_{k=1}^{2} \left( \nrm{\rho^{\omg+k} \rd^{k}_{\rho} (\phi - \phibg)}_{C^{0}(\Sgm_{0, pert}^{+})}
					+ \nrm{\rho^{\omg+k} \rd^{k-1}_{\rho}(f \dot{\phi} - \fbg \overline{\dot{\phi}})}_{C^{0}(\Sgm_{0, pert}^{+})} \right)
\leq C_{\omg, R_{pert}} \eps,
\end{equation*}
which is acceptable. We emphasize that this holds for all $\omg>2$, regardless of the value of $\omg_0$ associated to the decay of the scalar field for the background initial data set.

To bound $\log(f/\fbg)$, $h - \hbg$ and their derivatives, it is more convenient to not only consider the definitions in Proposition~\ref{prop:reduced-data}, but also to work with (the restriction of) the local developments $(\Omg, r, \phi)$ and $(\Omgbg, \rbg, \phibg)$ arising from $\Tht_{\eps}$ and $\overline{\Tht}$, respectively, in a neighborhood of $\Sgm_{0, pert}^{+}$. This has the advantage that we can now apply \eqref{eq:EMSF-r-phi-m}. On $\Sgm_{0, pert}^{+}$, we claim that
\begin{align}
	2 \log (f / \fbg)
	= & - \log \left( 1 - 2 \frac{\varpi - \varpibg}{\rbg (1-\mubg)} \right) , \label{eq:inst-cont:key-1} \\
	2 \left(\f{h}{f} - \f{\hbg}{\fbg}\right)
	= & (\rd_{U} + \rd_{V}) \left(\log \Omg^{2} - \log \Omgbg^{2}\right) \notag \\
	= & \left( \frac{4 (\rd_{U} + \rd_{V}) \rbg}{\rbg^{2} (1-\mu)} + \frac{8 (\varpibg - \frac{\ebg^{2}}{\rbg}) (\rd_{U} + \rd_{V}) \rbg}{(1-\mu) (1-\mubg) \rbg^{3}} \right) (\varpi - \varpibg)
	-  (\rd_{U} + \rd_{V}) \log \left( 1 - 2 \frac{\varpi - \varpibg}{\rbg (1-\mubg)} \right). \label{eq:inst-cont:key-2} 
\end{align}
To prove \eqref{eq:inst-cont:key-1}, recall that $(r, \rd_{U} r, \rd_{V} r, \e) = (\rbg, \rd_{U} \rbg, \rd_{V} \rbg, \ebg)$ on $\Sgm_{0}$. Therefore, by \eqref{eq:reduced-data:f}, we have on $\Sgm_{0}$
\begin{align*}
	2 \log(f / \fbg)
	= &\log \left(4 \frac{\rd_{V} r (- \rd_{U} r)}{1-\mu} \right) - \log \left(4 \frac{\rd_{V} \rbg (- \rd_{U} \rbg)}{1-\mubg} \right) \\
	= & - \log \frac{1 - \mu}{1 - \mubg} 
	=- \log \left(1 - 2 \frac{\varpi - \varpibg}{\rbg (1 - \mubg)}  \right) .
\end{align*}
For \eqref{eq:inst-cont:key-2}, we first note that
\begin{equation*}
	\f{h}{f} - \f{\hbg}{\fbg}
	= \Omg^{-1} (\rd_{U} + \rd_{V}) \Omg - \Omgbg^{-1} (\rd_{U} + \rd_{V}) \Omgbg
	= \frac{1}{2} (\rd_{U} + \rd_{V}) \left(\log \Omg^{2} - \log \Omgbg^{2} \right) \quad \hbox{ on } \Sgm_{0}
\end{equation*}
 by Lemma~\ref{lem:cauchy-to-char}. Moreover, we have
 \begin{align*}
(\rd_{U} + \rd_{V}) \left(\log \Omg^{2} - \log \Omgbg^{2} \right)
= & (\rd_{U} + \rd_{V}) \log \frac{\rd_{V} r}{\rd_{V} \rbg} + (\rd_{U} + \rd_{V}) \log \frac{\rd_{U} r}{\rd_{U} \rbg}
- (\rd_{U} + \rd_{V}) \log \frac{1-\mu}{1-\mubg}.
\end{align*}
Note that $(\rd_{V} - \rd_{U}) \log \frac{\rd_{V} r}{\rd_{V} \rbg} = \rd_{\rho} \log \frac{\rd_{V} r}{\rd_{V} \rbg} = 0$ on $\Sgm_{0}$; similarly, $(\rd_{V} - \rd_{U}) \log \frac{\rd_{U} r}{\rd_{U} \rbg} = 0$ on $\Sgm_{0}$. Therefore, by \eqref{eq:EMSF-r-phi-m}, on $\Sgm_{0}$ we have
\begin{align*}
(\rd_{U} + \rd_{V}) \left(\log \Omg^{2} - \log \Omgbg^{2} \right)
= & 2 \rd_{U} \log \frac{\rd_{V} r}{\rd_{V} \rbg} + 2 \rd_{V} \log \frac{\rd_{U} r}{\rd_{U} \rbg} - (\rd_{U} + \rd_{V}) \log \frac{1-\mu}{1-\mubg} \\
= & \frac{4 (\varpi - \frac{\e^{2}}{r})}{r^{2}} \frac{\rd_{U} r}{1-\mu} - \frac{4 (\varpibg - \frac{\ebg^{2}}{\rbg})}{\rbg^{2}} \frac{\rd_{U} \rbg}{1-\mubg} \\
& + \frac{4 (\varpi - \frac{\e^{2}}{r})}{r^{2}} \frac{\rd_{V} r}{1-\mu} - \frac{4 (\varpibg - \frac{\ebg^{2}}{\rbg})}{\rbg^{2}} \frac{\rd_{V} \rbg}{1-\mubg} \\
& - (\rd_{U} + \rd_{V}) \log \left(1 - 2 \frac{\varpi - \varpibg}{\rbg (1-\mubg)}\right).
\end{align*}
Rearranging terms, while recalling that $(r, \rd_{U} r, \rd_{V} r, \e) = (\rbg, \rd_{U} \rbg, \rd_{V} \rbg, \ebg)$ on $\Sgm_{0}$, we obtain \eqref{eq:inst-cont:key-2}. 

After writing the differences as in \eqref{eq:inst-cont:key-1} and \eqref{eq:inst-cont:key-2}, the key is now to estimate the mass difference. By \eqref{eq:reduced-data:varpi}, it can be shown that
\begin{equation*}
	\Abs{(\varpi - \varpibg) \restriction_{\Sgm_{0, pert}^{+}}}
	\leq C_{\overline{\Tht}} \eps,
\end{equation*}
where $C_{\overline{\Tht}}$ is independent of $R_{pert}$. Moreover, in view of \eqref{eq:EMSF-r-phi-m} and the fact that $\rd_{\rho} = \rd_{V} - \rd_{U}$ on $\Sgm_{0}$, it can be seen that
\begin{equation*}
	\Abs{\rd_{\rho}^{k}  (\varpi - \varpibg) \restriction_{\Sgm_{0, pert}^{+}}} 
	+ \Abs{\rd_{\rho}^{k-1}  \left( (\rd_{U} + \rd_{V}) (\varpi - \varpibg) \right)\restriction_{\Sgm_{0, pert}^{+}}}
	\leq C_{k, R_{pert}, \overline{\Tht}} \, \rbg^{- 2 \omg - (k - 1)} \eps
\end{equation*}
for any $1 \leq k \leq 2$. Then by the formulae \eqref{eq:inst-cont:key-1} and \eqref{eq:inst-cont:key-2}, it follows that
\begin{equation*}
\begin{split}
\nrm{\rho \log(f/ \fbg)}_{C^{0}(\Sgm_{0, pert}^{+})}  + \sum_{k=1}^{2} \left( \nrm{\rho^{1+k} \rd_{\rho}^{k} \log(f/ \fbg)}_{C^{0}(\Sgm_{0, pert}^{+})}
					+ \nrm{\rho^{1+k} \rd_{\rho}^{k-1} (h - \hbg)}_{C^{0}(\Sgm_{0, pert}^{+})} \right) 
			 \leq C_{R_{pert}, \overline{\Tht}} \eps,
\end{split}
\end{equation*}
which completes the proof. \qedhere
\end{proof}
\begin{remark}\label{rem:inst-high-reg}
Suppose that $\overline{\Tht}$ is moreover $C^{k}_{\omg_0}$-regular for some $k \geq 2$. Then proceeding as in the previous proof, it can be shown that $\eps \mapsto \Tht_{\eps}$ is continuous with respect to $d^{+}_{k, \omg}$ for any $\omg > 2$. We omit the straightforward details.
\end{remark}
In view of Lemmas~\ref{lem:Thtep.facts}, \ref{lem:inst-cont} and Remark~\ref{rem:inst-high-reg}, in order to prove Theorem~\ref{thm:instability}, it only remains to establish the following statement, which in particular implies that $\mathfrak{L}_{(\omg_0)\infty}[\Tht_\eps]\neq 0$ when $\eps \neq 0$.
\begin{proposition} \label{prop:inst-key}
Let $\mathfrak{L}_{(\eps)}=\mathfrak{L}[\Tht_\eps]$ be defined as in \eqref{eq:L-def} from the maximal globally hyperbolic future development of the Cauchy data set $\Tht_{\eps}$. For $R_{pert}$ sufficiently large and $\eps$ small enough (both depending on the background solution), we have
\begin{equation} \label{eq:inst-key}
\abs{\mathfrak{L}_{(\eps)} - \overline{\mathfrak{L}}} \geq c_{0} \varpibg_{i} \eps ,
\end{equation}
where $c_{0}$ is the universal constant in \eqref{eq:pert-chi}.
\end{proposition}

For the proof of Proposition~\ref{prop:inst-key}, we split the background spacetime into four regions $\calR_{1}, \ldots, \calR_{4}$ in a similar manner as in the proof of Theorem~\ref{thm:L-stability} in Section~\ref{subsec:main-st}. We introduce a parameter $R_{pert}$ to be specified below and define the regions $\calR_{1}, \ldots, \calR_{4}$ as follows:

\begin{figure}[h]
\begin{center}
\def\svgwidth{300px}
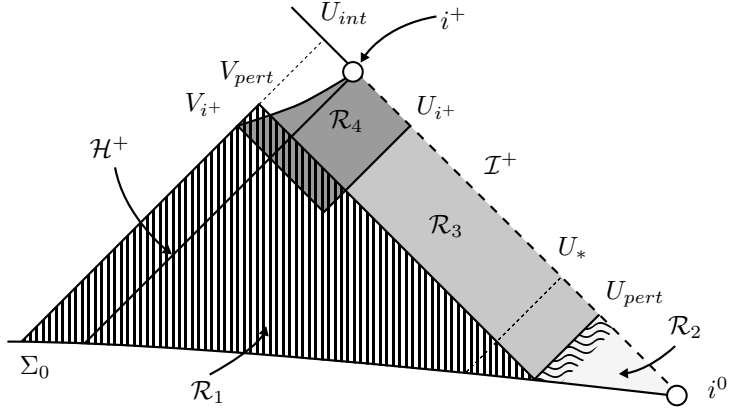 
\caption{Regions $\calR_{1}, \ldots, \calR_{4}$. The support of main term $\eps \tilde{\chi}_{R_{pert}}$ of the perturbation is shaded by horizontal wavy lines.} \label{fig:L-inst-key}
\end{center}
\end{figure}

\begin{itemize}
\item We define $\calR_{4}$ using Lemma~\ref{lem:dlt-adm-exist}. More precisely, let $2 < \omg < \min \set{\omg_{0}, 3}$ and $\eta_{0} < \frac{1}{200} (\min \set{\omg_{0}, 3} - \omg)$ be fixed. By Lemma~\ref{lem:dlt-adm-exist}, for any $\dlt_{0} > 0$ there exist $\Lmb \geq 100$ and $(U_{i^+}, V_{i^+})$ such that the restriction of the background solution to 
\begin{equation*}
\calR_{4} := \set{(U, V) \in \PD: U \geq U_{i^+}, \ V \geq V_{i^+}, \ \rbg(U, V) \geq (1 - 2 \dlt) \rbg_{\EH}}
\end{equation*}
is $(\omg, \dlt, \Lmb)$-admissible with $\dlt = \Lmb^{-100 \eta_{0}} \dlt_{0}$. We fix $\dlt_{0}$ sufficiently small (depending on the background solution) so that Theorem~\ref{thm:L-st-ch} is applicable in $\calR_{4}$.
\item Let $(U_{pert}, V_{pert})$ be the $(U, V)$-coordinates of the point $\rho_{pert}$ on $\Sgm_{0, pert}^{+}$; hence $\rbg(U_{pert}, V_{pert}) = R_{pert}$. Since $R_{pert}$ will later be taken sufficiently large, we may safely assume that $U_{pert} < U_{i^+}$ and $V_{pert} > V_{i^+}$. We define
\begin{equation*}
	\calR_{3} := \set{(U, V) \in \PD : U_{pert} \leq U \leq U_{i^+}, \ V \geq V_{pert}}.
\end{equation*}
\item We also introduce 
\begin{equation*}
\calR_{2} := \set{(U, V) \in \PD : U \leq U_{pert}}.
\end{equation*}
Note that $\calR_{2}$ is precisely the domain of dependence of $\Sgm_{0, pert}^{+}$.
\item Finally, we define
\begin{equation*}
\calR_{1} := \set{(U, V) \in \PD : U \leq U_{int}, \ V \leq V_{pert}},
\end{equation*}
where $(U_{int}, V_{i^{+}})$ is the point in $\uC_{V_{i^{+}}} \cap \set{\rbg =( 1 - 2 \dlt )\rbg_{\EH}}$. Note that we can argue as in Section~\ref{subsec:main-st} to see that there exists $\de>0$ sufficiently small such that we can simultaneously guarantee that such a point $(U_{int}, V_{i^{+}})$ is in the maximal globally hyperbolic development of the initial data and that Theorem~\ref{thm:L-st-ch} is applicable in $\calR_{4}$. 
\end{itemize}

Note that $\calR_{1}$ slightly overlaps with $\calR_{4}$. Note that, since $\Tht_{\eps}$ and $\overline{\Tht}$ coincide in $\set{\rho \in \Sgm_{0} : \rho < \rho_{pert}}$, their developments in the $(U, V)$-coordinates also coincide in $\calR_{1}$. To prove Proposition~\ref{prop:inst-key}, we estimate $\mathfrak{L}_{(\eps)} - \overline{\mathfrak{L}}$ restricted to each of $\calR_{2}$, $\calR_{3}$ and $\calR_{4}$, using the tools developed in Sections~\ref{sec:extr} and \ref{sec:L-stability}.

 We start with some preliminary bounds on the background solution.
\begin{lemma} \label{lem:inst-bg}
Let $\underline{\calN}$ be as in Proposition~\ref{prop:bg-large-r} for the background solution $(\gbg, \phibg, \Fbg)$. 
Then there exists $\underline{B}_{\ast} > 0$, which depends on the solution $(\gbg, \phibg, \Fbg)$ on $\underline{\calN}$, such that the following bounds hold for $(U, V) \in \underline{\calN}$:
\begin{align}
	\underline{B}_{\ast}^{-1} \leq \rd_{V} \rbg(U, V) \leq & \underline{B}_{\ast}, \label{eq:inst-bg:dvr} \\
	\underline{B}_{\ast}^{-1} \leq -\rd_{U} \rbg(U, V) \leq & \underline{B}_{\ast}, \label{eq:inst-bg:dur} \\
	\underline{B}_{\ast}^{-1} \leq \frac{-\rd_{U} \rbg(U, V)}{-\rd_{U} \underline{\ups}(U)} \leq & \underline{B}_{\ast}, \label{eq:inst-bg:duups} \\
	\Abs{\overline{\psi}(U, V)} \leq & \underline{B}_{\ast} \, \underline{\ups}^{-2}(U) \label{eq:inst-bg:psi} \\
	\Abs{\frac{1-\mubg}{\rd_{V} \rbg}\rd_{V} \overline{\psi}(U, V)} \leq & \underline{B}_{\ast} \, \rbg^{-3}(U, V) \label{eq:inst-bg:dvpsi}\\	
	\Abs{\overline{M}(U) - \varpibg_{i}} \leq & \underline{B}_{\ast} \, \underline{\ups}^{-3} (U) , \label{eq:inst-bg:M} \\
	\Abs{\overline{\Gmm}(U) + \frac{1}{2}} \leq & \underline{B}_{\ast} \, \underline{\ups}^{-1} (U) .	\label{eq:inst-bg:Gmm}
\end{align}
\end{lemma}

\begin{proof}
The bounds \eqref{eq:inst-bg:dvr}--\eqref{eq:inst-bg:dvpsi} follow from Proposition~\ref{prop:bg-large-r} (where $\omg_{0} = 3$). To prove \eqref{eq:inst-bg:M} and \eqref{eq:inst-bg:Gmm}, observe that
\begin{align*}
	\Abs{\overline{M} (U) - \varpibg (U, \underline{V}(U))} 
	\leq & \int_{\underline{V}(U)}^{\infty} (1-\mubg) \left( \frac{1}{\rd_{V} r} \rd_{V} \psibg - \frac{1}{\rbg} \psibg \right)^{2} \rd_{V} \rbg (U, V) \, \ud V 
	\leq C \underline{B}_{\ast}^{2} \underline{\ups}^{-3}(U), \\
	\Abs{\log \overline{\Gmm}(U) - \log \frac{\rd_{U} \rbg}{1-\mubg}(U, \underline{V}(U))} 
	\leq & \int_{\underline{V}(U)} \frac{1}{\rbg} \left( \frac{1}{\rd_{V} r} \rd_{V} \psibg - \frac{1}{\rbg} \psibg \right)^{2} \rd_{V} \rbg  (U, V) \, \ud V 
	\leq C \underline{B}_{\ast}^{2} \underline{\ups}^{-4}(U),
\end{align*}
where we used \eqref{eq:inst-bg:psi}, \eqref{eq:inst-bg:dvpsi} and the bound $C^{-1} \leq 1-\overline{\mu} \leq C$ from \eqref{eq:bg-large-r:pf-mu}. By Lemma~\ref{lem:cauchy-to-char}, \eqref{eq:adm-id-af}, \eqref{eq:adm-id-phi} and \eqref{eq:reduced-data:varpi}, we have
\begin{equation*}
	\Abs{\varpibg(\rho) - \varpibg_{i}} \leq C_{\overline{\Tht}} r^{-5}(\rho), \quad
	\Abs{\frac{\rd_{U} \rbg}{1-\mubg}(\rho) + \frac{1}{2}} \leq C_{\overline{\Tht}} r^{-1}(\rho)
\end{equation*}
on $\Sgm_{0} \cap \underline{\calN}$. Combined with the preceding two bounds,  \eqref{eq:inst-bg:M} and \eqref{eq:inst-bg:Gmm} follow. \qedhere
\end{proof}

In $\calR_{2}$, we claim the following bounds.
\begin{lemma} \label{lem:inst-s}
Let $R_{pert} \geq \underline{R}_{\ast}$. For $\underline{R}_{\ast}$ sufficiently large and $\eps$ small enough (both depending on $\underline{B}_{\ast}, \overline{\Tht}$), there exists $C = C(\underline{B}_{\ast}, \underline{R}_{\ast}, \overline{\Tht}) > 0$ such that the following bounds hold in $\calR_{2}$:
\begin{align} 
	\Abs{\varpi - \varpibg }(U, V) \leq & C \eps,  \label{eq:inst-s:mdf} \\
	\rbg^{-1} \Abs{r - \rbg }(U, V) \leq & C \eps,  \label{eq:inst-s:rdf} \\
	\Abs{\log \frac{(-\rd_{U} r)}{1-\mu} - \log \frac{(-\rd_{U} \rbg)}{1-\mubg}} (U, V) \leq & C \eps, \label{eq:inst-s:loggmmdf}\\	
	\Abs{\log \rd_{V} r - \log \rd_{V} \rbg} (U, V) \leq & C \rbg^{-1} (U, V) \eps, \label{eq:inst-s:logdvrdf} \\
	\Abs{\log \frac{\rd_{V} r}{1-\mu} - \log \frac{\rd_{V} \rbg}{1-\mubg}} (U, V) \leq & C \rbg^{-1} (U, V) \eps, \label{eq:inst-s:logkppdf} \\
	\Abs{\frac{1-\mu}{\rd_{V} r} \rd_{V} \psi - \frac{1-\mubg}{\rd_{V} \rbg} \rd_{V} \psibg} (U, V) \leq & C \rbg^{-3} (U, V) \eps, \label{eq:inst-s:dvpsidf} \\
	\Abs{\psi (U, V) - \psibg (U, V) - \eps \tilde{\chi}_{R_{pert}} (U)} \leq & C \underline{\ups}^{-2}(U) \eps , \label{eq:inst-s:psidf} 
\end{align}
where $\tilde{\chi}_{R_{pert}} (U) = \tilde{\chi}(- U + U_{pert})$.
\end{lemma}
This lemma is essentially a corollary of Proposition~\ref{prop:extr-st-cauchy}. In order to obtain the refined bounds stated above, we simply reiterate the wave equations for $r - \rbg$ and $\psi - \psibg$ using the bounds provided by Proposition~\ref{prop:extr-st-cauchy}. We emphasize that the constant $C$ in the lemma is \emph{independent} of $R_{pert}$; it is for this reason that we use the integrated norm of $\psi - \psibg$ in the definition of $\eps_{2}$ in Proposition~\ref{prop:extr-st-cauchy} (see Remark~\ref{rem:extr-st-cauchy}).

\begin{proof}
In this proof, we denote by $C$ a positive constant that depends on $\underline{B}_{\ast}$, $\underline{R}_{\ast}
$ and $\overline{\Tht}$, which may change from line to line. Again, $C$ is independent of $R_{pert}$.

To apply Proposition~\ref{prop:extr-st-cauchy}, we choose $\underline{\calR} = \calR_{2}$, $\omg = 3$, $\varpibg_{0} = \varpibg_{i}$, $\rbg_{0} = \underline{R}_{\ast}$ and $\underline{B} = \underline{B}_{\ast}$; taking $\underline{R}_{\ast}$ large enough, note that the assumptions \eqref{eq:extr-st-cauchy:hyp:r-m}--\eqref{eq:extr-st-cauchy:hyp:dvpsi} are satisfied. By construction, observe that
\begin{align*}
	\eps_{2} = & \sup_{\Sgm_{0, pert}^{+}} \left( \rbg_{0}^{-1} \abs{\varpi - \varpibg} + \abs{\log (1-\mu) - \log (1-\mubg)} \right) \\
	& + \sup_{\Sgm_{0, pert}^{+}} \rbg_{0}^{-1} \abs{\psi - \psibg}  + \rbg_{0}^{-2} \int_{\Sgm_{0, pert}^{+}} \abs{\psi - \psibg} \rd_{\rho} \rbg \, \ud \rho 
	\leq C_{\overline{\Tht}} \eps
\end{align*}
for some $C_{\overline{\Tht}} > 0$ independent of $R_{pert}$. Then the bounds \eqref{eq:inst-s:mdf} and \eqref{eq:inst-s:rdf} follow directly from Proposition~\ref{prop:extr-st-cauchy}. Moreover, making $\sup_{\calR_{2}} \abs{\mubg}$ sufficiently small by choosing $\underline{R}_{\ast}$ large enough, we also obtain
\begin{equation} \label{eq:inst-s:mu}
	\frac{1}{2} \leq 1 - \mubg \leq 2, \quad \frac{1}{2} \leq 1 - \mu \leq 2, \quad \Abs{\log(1 - \mu) - \log(1 - \mubg)} \leq C \rbg^{-1} \eps,
\end{equation}
for sufficiently small $\eps$. Combined with the bound for $\log (-\rd_{U} r) - \log (-\rd_{U} \rbg)$ in Proposition~\ref{prop:extr-st-cauchy}, we obtain \eqref{eq:inst-s:loggmmdf}. Moreover, by \eqref{eq:inst-s:mdf}, \eqref{eq:inst-s:rdf}, \eqref{eq:inst-s:loggmmdf} and \eqref{eq:inst-s:mu}, we have
\begin{align*}
	\Abs{\rd_{U} (\log \rd_{V} r - \log \rd_{V} \rbg)} (U,V)
	\leq & \Abs{\frac{2 (\varpi - \frac{\e^{2}}{r})}{r^{2}} \frac{\rd_{U} r}{1-\mu} - \frac{2 (\varpibg - \frac{\ebg^{2}}{\rbg})}{\rbg^{2}} \frac{\rd_{U} \rbg}{1-\mubg}} (U,V) \\
	\leq & C \rbg^{-2} (- \rd_{U} \rbg) (U, V) \eps.
\end{align*}
Since $\rd_{V} r = \rd_{V} \rbg$ on $\Sgm_{0, pert}^{+}$, \eqref{eq:inst-s:logdvrdf} follows after integrating in $U$ from $(\underline{U}(V), V)$ to $(U, V)$. Furthermore, in combination with \eqref{eq:inst-s:mu}, we obtain \eqref{eq:inst-s:logkppdf} as well. Next, by \eqref{eq:inst-bg:dur}, \eqref{eq:inst-bg:duups}, \eqref{eq:inst-bg:psi}, \eqref{eq:inst-s:mdf}, \eqref{eq:inst-s:rdf}, \eqref{eq:inst-s:loggmmdf}, \eqref{eq:inst-s:logdvrdf}, \eqref{eq:inst-s:mu} and the bound for $\abs{\psi - \psibg}$ from Proposition~\ref{prop:extr-st-cauchy}, we have
\begin{align}
	\Abs{\rd_{U} \rd_{V} (\psi - \psibg)(U, V)} 
	\leq & \Abs{\frac{2 (\varpi - \frac{\e^{2}}{r})}{r^{3}} \frac{\rd_{V} r \rd_{U} r}{1-\mu} \psi - \frac{2 (\varpibg - \frac{\ebg^{2}}{r})}{\rbg^{3}} \frac{\rd_{V} \rbg \rd_{U} \rbg}{1-\mubg} \psibg}(U,V) \notag\\
	\leq & C \rbg^{-3} (U, V) \left( \abs{\psi - \psibg}(U, \underline{V}(U)) + \underline{\ups}^{-2} (- \rd_{U} \underline{\ups}) (U) \eps \right) .\label{eq:inst-s:dudvpsidf}
\end{align}
Since $\rd_{V} \psi = \rd_{V} \psibg$ on $\Sgm_{0, pert}^{+}$, integration in $U$ from $(\underline{U}(V), V)$ to $(U, V)$ yields
\begin{equation}\label{eq:inst-s:dvpsidf.pl}
	\Abs{\rd_{V} (\psi - \psibg)(U, V)} \leq C \rbg^{-3} (U, V) \eps.
\end{equation}
Combined with \eqref{eq:inst-bg:dvpsi} and \eqref{eq:inst-s:logkppdf}, the preceding bound implies \eqref{eq:inst-s:dvpsidf}. Moreover, 
using \eqref{eq:inst-bg:dvr} and integrating \eqref{eq:inst-s:dvpsidf.pl} in $V$ from $(U, \underline{V}(U))$ to $(U, V)$, then noting that \eqref{eq:inst:psi-eps} holds on $\Sgm_{0}$, we obtain \eqref{eq:inst-s:psidf}. 

\qedhere

\end{proof}

We are now ready to complete the proof of Proposition~\ref{prop:inst-key}.
\begin{proof}[Proof of Proposition~\ref{prop:inst-key}]
To ease the notation, we drop the subscripts $\eps$ and $(\eps)$; hence we write $\Tht = \Tht_{\eps}$, $\mathfrak{L} = \mathfrak{L}_{(\eps)}$ etc.

We divide the maximal globally hyperbolic future development of $\overline{\Tht}$ into $\calR_{1}, \ldots, \calR_{4}$ as in Figure~\ref{fig:L-inst-key}. 
The parameter $R_{pert}$ will be chosen at the end of the proof, so that the error in establishing the lower bound \eqref{eq:inst-key} is small, and also Propositions~\ref{prop:extr-st-cauchy} and \ref{prop:extr-st} are applicable in $\calR_{2}$ and $\calR_{3}$, respectively.

In the $(U, V)$ coordinate system, we introduce the notation (cf. the proof of Theorem~\ref{thm:L-stability} in Section~\ref{subsec:main-st})
\begin{equation*}
\mathfrak{L}(U_{1}, U_{2}) = \int_{U_{1}}^{U_{2}} 2 M(U) \Phi(U) \Gmm(U) \, \ud U, \quad
\overline{\mathfrak{L}}(U_{1}, U_{2}) = \int_{U_{1}}^{U_{2}} 2 \overline{M}(U) \Phibg(U) \overline{\Gmm}(U) \, \ud U.
\end{equation*}
In order to estimate the LHS of \eqref{eq:inst-key}, 
we analyze the contribution to $\mathfrak{L} - \overline{\mathfrak{L}}$ from the regions $\calR_{2}$, $\calR_{3}$ and $\calR_{4}$ separately.

\pfstep{Step~1: Region $\calR_{2}$}
This is where we obtain the main lower bound. We begin by writing
\begin{align*}
	(\mathfrak{L} - \overline{\mathfrak{L}})(-\infty, U_{pert})
	= & \int_{-\infty}^{U_{pert}} 2 \overline{M} (\Phi - \Phibg) \overline{\Gmm}(U)\, \ud U \\
	& + \int_{-\infty}^{U_{pert}}  2 ( (M - \overline{M}) \Gmm + \overline{M} (\Gmm - \overline{\Gmm}) ) \Phibg (U)\, \ud U \\
	& + \int_{-\infty}^{U_{pert}} 2 ( (M - \overline{M}) \Gmm + \overline{M} (\Gmm - \overline{\Gmm}) ) (\Phi - \Phibg) (U)\, \ud U .
\end{align*}
Clearly, we have the lower bound
\begin{equation} \label{eq:inst-key:R2-main} 
\Abs{\int_{-\infty}^{U_{pert}} \varpibg_{i} \eps \tilde{\chi}_{R_{pert}} (U)\, \ud U} \geq 2 c_{0} \varpibg_{i} \eps,
\end{equation}
where $c_{0}$ is as in \eqref{eq:pert-chi}. We claim that, for a sufficiently large $R_{pert}$, the following bounds for the error hold:
\begin{align}
\Abs{\int_{-\infty}^{U_{pert}} \left( 2 \overline{M} (\Phi - \Phibg) \overline{\Gmm} + \varpibg_{i} \eps \tilde{\chi}_{R_{pert}} \right) (U) \, \ud U} \leq & C_{1} R_{pert}^{-1} \eps, \label{eq:inst-key:R2-0} \\
\Abs{\int_{-\infty}^{U_{pert}}   ( (M - \overline{M}) \Gmm + \overline{M} (\Gmm - \overline{\Gmm})) \Phibg (U)\, \ud U} \leq & C_{1} R_{pert}^{-1} \eps, \label{eq:inst-key:R2-1} \\
\Abs{\int_{-\infty}^{U_{pert}}   ( (M - \overline{M}) \Gmm + \overline{M} (\Gmm - \overline{\Gmm})) (\Phi - \Phibg) (U)\, \ud U} \leq & C_{1} \eps^{2}. \label{eq:inst-key:R2-2}
\end{align}
for some $C_{1} = C_{1}(\underline{B}_{\ast},\underline{R}_*, \overline{\Tht})$. 

We begin with \eqref{eq:inst-key:R2-0}. We take $R_{pert} \geq \underline{R}_{\ast}$ so that Lemma~\ref{lem:inst-s} is applicable. Moreover, by \eqref{eq:inst-bg:M} and \eqref{eq:inst-bg:Gmm}, for $R_{pert}$ sufficiently large depending on $\underline{B}_{\ast}$ and $\varpibg_{i}$, we have 
\begin{equation} \label{eq:inst-key:M-Gmm-bg}
	\overline{M}(U) \leq 2 \varpibg_{i}, \quad
	\abs{\overline{\Gmm}(U)} \leq 1,
\end{equation}
for $(U, V) \in \calR_{2}$. In combination with \eqref{eq:inst-bg:M}, \eqref{eq:inst-bg:Gmm} and \eqref{eq:inst-s:psidf}, we obtain
\begin{align*}
\Abs{\overline{M} (\Phi - \Phibg) \overline{\Gmm} - \varpibg_{i} (\eps \tilde{\chi}_{R_{pert}}) \left( -\frac{1}{2} \right)}
\leq & \abs{\overline{M} \overline{\Gmm}} \Abs{\Phi - \Phibg - \eps \tilde{\chi}_{R_{pert}}} + \eps \overline{M} \Abs{\overline{\Gmm} + \frac{1}{2}} |\tilde{\chi}_{R_{pert}|}
+ \frac{1}{2} \eps \abs{\overline{M} - \overline{\varpi_{i}}} =|\tilde{\chi}_{R_{pert}}| \\
\leq & C_{\underline{B}_{\ast},\underline{R}_*, \overline{\Tht}} \underline{\ups}^{-2} (-\rd_{U} \underline{\ups}) \eps
+ C_{\underline{B}_{\ast},\underline{R}_*, \overline{\Tht}} \underline{\ups}^{-1} |\tilde{\chi}_{R_{pert}}| \eps.
\end{align*}
Integrating in $U$ from $-\infty$ to $U_{pert}$ and using \eqref{eq:pert-chi}, we arrive at \eqref{eq:inst-key:R2-0}.

Next, we turn to \eqref{eq:inst-key:R2-1}. Taking $\eps$ sufficiently small (depending on $\underline{B}_{\ast}, \overline{\Tht}$), \eqref{eq:inst-s:loggmmdf} and \eqref{eq:inst-key:M-Gmm-bg} imply
\begin{equation} \label{eq:inst-key:Gmm}
	\abs{\Gmm(U)} \leq 2.
\end{equation}
Therefore, in combination with \eqref{eq:inst-bg:dur}, \eqref{eq:inst-bg:duups}, \eqref{eq:inst-bg:psi}, \eqref{eq:inst-s:mdf} and \eqref{eq:inst-key:M-Gmm-bg}, we obtain
\begin{equation*}
	\abs{(M - \overline{M}) \Gmm \overline{\Phi}} + \abs{\overline{M} (\Gmm - \overline{\Gmm}) \overline{\Phi}}
	\leq C_{\underline{B}_{\ast},\underline{R}_*, \overline{\Tht}} \underline{\ups}^{-2}(U) (- \rd_{U} \underline{\ups}) \eps.
\end{equation*}
Integrating in $U$ from $-\infty$ to $U_{pert}$, \eqref{eq:inst-key:R2-1} follows.

Finally, we prove \eqref{eq:inst-key:R2-2}. By \eqref{eq:inst-bg:dur}, \eqref{eq:inst-bg:duups}, \eqref{eq:inst-s:mdf}, \eqref{eq:inst-s:loggmmdf}, \eqref{eq:inst-s:psidf}, \eqref{eq:inst-key:M-Gmm-bg} and \eqref{eq:inst-key:Gmm}, we have
\begin{align*}
	\left(\abs{(M - \overline{M}) \Gmm} + \abs{\overline{M} (\Gmm - \overline{\Gmm})}\right) \abs{\Phi - \overline{\Phi}}
	\leq & C_{\underline{B}_{\ast},\underline{R}_*, \overline{\Tht}} \underline{\ups}^{-2} (-\rd_{U} \underline{\ups}) \eps^{2} + C_{\underline{B}_{\ast},\underline{R}_*, \overline{\Tht}} \abs{\tilde{\chi}_{R_{pert}}} \eps^{2}.
\end{align*}
Integrating in $U$ from $-\infty$ to $U_{pert}$, we obtain \eqref{eq:inst-key:R2-2}.

\pfstep{Step~2: Region $\calR_{3}$}
We fix $U_{\ast}$ so that $C_{U_{\ast}} \subseteq \underline{\calN}$. 
By Theorem~\ref{thm:DR-full}, Corollary~\ref{cor:DR-large-r}, Remark~\ref{rem:DR-large-r:dUr} and Lemma~\ref{lem:inst-bg}, there exists $B_{\ast} \geq 1$ (depending on $B_{\eta_{\NI}}, B_{cpt}, \underline{B}_{\ast}$) so that:
\begin{itemize}
\item \eqref{eq:extr-st:hyp:dur}--\eqref{eq:extr-st:hyp:duups} hold on any $\uC_{V} \cap \set{U \leq U_{i^{+}}} \subseteq \calN$ with $B = B_{\ast}$,
\item \eqref{eq:extr-st:hyp:psi}--\eqref{eq:extr-st:hyp:dvpsi} hold on $\calN$ with $B = B_{\ast}$.
\end{itemize}
Fix $\rbg_{0}$ sufficiently large such that \eqref{eq:extr-st:hyp:r-m}--\eqref{eq:extr-st:hyp:m-e} hold with $\rbg_{0}$, $\varpibg_{0} = \varpibg_{i}$ and a small enough $\dlt_{\calR}$. By choosing $R_{pert}$ sufficiently large, we can guarantee that $\rbg(U_{i^{+}}, V_{pert})\geq \rbg_{0}$. Therefore, Proposition~\ref{prop:extr-st} is applicable in $\calR_{3}$. \textbf{From now on, we will allow our constants to depend on the fixed $\rbg_{0}$.}

Our goal is now to apply Proposition~\ref{prop:extr-st} to show that the contribution to $\Abs{\mathfrak{L} - \overline{\mathfrak{L}}}$from this region is (much) smaller than $c_0\eps$. The source of smallness will come from choosing $R_{pert}$ large, together with exploiting the freedom of choosing $\omg$ in Proposition~\ref{prop:extr-st}. For this purpose, we fix some $\omg\in (2,3)$ and estimate $\eps_{3}$ (as in the statement of Proposition~\ref{prop:extr-st}). Note that the difference of the data on the initial incoming curve $\uC_{in}$ is trivial (since the two solutions coincide on $\calR_{1}$ by the finite speed of propagation), and on the initial outgoing curve $C_{out}$ it is given by Lemma~\ref{lem:inst-s}. Hence
$$\eps_{3} \leq C_{\underline{B}_{\ast}, \underline{R}_{\ast}, \overline{\Tht}, \rbg_{0}, \omg} \eps\left(R_{pert}^{-1} + R_{pert}^{\omg - 3}\right)\leq C_{\underline{B}_{\ast}, \underline{R}_{\ast}, \overline{\Tht}, \rbg_{0}, \omg} R_{pert}^{\omg - 3} \eps.$$
Here, we have the smallness factors $R_{pert}^{-1}$ and $R_{pert}^{\omg - 3}$ since in \eqref{eq:inst-s:logkppdf}, \eqref{eq:inst-s:dvpsidf} (which are the only terms we need to bound on $C_{out}$), there are extra factors of $\inf_{C_{out}}\rbg^{-1}$ or $\inf_{C_{out}}\rbg^{-3+\omg}\rbg_0^{-\omg}$ compared to what is needed in the definition of $\eps_{3}$.

By Proposition~\ref{prop:extr-st}, we obtain 
\begin{equation} \label{eq:inst-key:R3}
\Abs{(\mathfrak{L} - \overline{\mathfrak{L}})(U_{pert}, U_{i^+})}
 \leq C_{2}' R_{pert}^{\omg-3} \eps,
\end{equation}
where $C_{2}' = C_{2}' (\underline{B}_{\ast}, \underline{R}_{\ast}, \overline{\Tht}, \rbg_{0}, \omg, U_{i^{+}})$. Moreover, in $\calR_{3}$ we have
\begin{align} 
	\Abs{\left(\frac{\rd_{V} r}{1-\mu}\right)^{-1}\frac{\rd_{V} \rbg}{1-\mubg}-1} (U, V) \leq & C_{2}'' R_{pert}^{\omg-3}  \eps, \label{eq:inst-key:R3-1} \\
	\Abs{\frac{1-\mu}{\rd_{V} r} \rd_{V} \psi - \frac{1-\mubg}{\rd_{V} \rbg} \rd_{V} \psibg} (U, V) \leq & C_{2}'' R_{pert}^{\omg-3} (\rbg / \rbg_{\EH})^{-\omg} (U, V) \eps,	\label{eq:inst-key:R3-2}
\end{align}
for some $C_{2}'' = C_{2}'' (\underline{B}_{\ast}, \underline{R}_{\ast}, \overline{\Tht}, \rbg_{0}, \rbg_{\EH}, \omg, U_{i^{+}})$.

\pfstep{Step~3: Region $\calR_{4}$}
By definition, Theorem~\ref{thm:L-st-ch} is applicable in $\calR_{4}$. As in the previous step, the difference of the data on the initial incoming curve (which belongs to $\calR_{1}$) is trivial, and on the initial outgoing curve $C_{out}$
its size is controlled by \eqref{eq:inst-key:R3-1}--\eqref{eq:inst-key:R3-2}. Thus we may apply Theorem~\ref{thm:L-st-ch} with $\eps_{0}=C_{2}'' R_{pert}^{\omg-3}  \eps$. By Theorem~\ref{thm:L-st-ch} we obtain\footnote{Recall from Section~\ref{sec:L-stability} that $\mathfrak{L} \restriction_{\calD}$ and $\overline{\mathfrak{L}}\restriction_{\calD}$ exactly denote the contributions to $\mathfrak{L}$ and $\overline{\mathfrak{L}}$ from $\calR_{4}$. Unlike the previous regions, regions $\calR_{4}$ in the solutions arising from $\Tht$ and $\overline{\Tht}$ do not have the same $U$-range.}
\begin{equation} \label{eq:inst-key:R4}
\Abs{\mathfrak{L} \restriction_{\calD}- \overline{\mathfrak{L}}\restriction_{\calD}}
\leq C_{3} R_{pert}^{\omg-3} \eps.
\end{equation}
where $C_{3} = C_{3}(\eta_{0}, \varpibg_{f}^{-1} \abs{\ebg}, \omg, \Lmb, \rbg_{\EH}, C_{2}'')$. 

\pfstep{Step~4: Putting everything together} Choosing $R_{pert}$ large enough (depending on $c_0 \varpibg_{i}$, $C_{1}$, $C_{2}'$, $C_{3}$) and $\eps$ sufficiently small (depending on $c_0 \varpibg_{i}$, $C_{1}$), we obtain \eqref{eq:inst-key} from \eqref{eq:inst-key:R2-main}, \eqref{eq:inst-key:R2-0}, \eqref{eq:inst-key:R2-1}, \eqref{eq:inst-key:R2-2}, \eqref{eq:inst-key:R3} and \eqref{eq:inst-key:R4}. \qedhere
\end{proof}

\clearpage
\appendix
\section{List of symbols} \label{sec:symbols}
For the convenience of the reader, we provide a list of some frequently used symbols.
\vskip.5em
\begin{center}
\small
\begin{tabular}{c  c p{25em}}
{\bfseries Symbol} & {\bfseries Ref.} & {\bfseries Description}  \\
$\Omg, r, \phi, \e$ & \S \ref{sec.SS} & Metric components and matter fields for spherically symmetric solutions to Einstein--Maxwell--(real)--scalar--field system \\
$C_u$, $\uC_{v}$ & \S \ref{sec.SS} & Constant $u$- and $v$-hypersurface \\
$m, \mu, \varpi$ & \S \ref{sec.SS} & $m$ is the Hawking mass; see \eqref{eq:mass.def}. $\mu = \frac{m}{r}$, $\varpi = m - \frac{\e^{2}}{r}$. \\
$r, f, h, \ell, \phi, \dot{\phi}, \e$ & \S \ref{sec:Cauchy} & Initial data set \\
$\omg, \omg_0$ & \S \ref{sec.review.basic} & Parameter for initial decay of scalar field \\
$\Sgm_{0}, \EH, \NI$ & \S \ref{sec.review.basic} & Initial curve, event horizon, null infinity \\
$\varpi_i$, $\varpi_{\EH}$, $\varpi_{f}$ & \S \ref{sec:MGHFD.prelim} and \ref{sec:bg} & Various limits of $\varpi$, see Definition~\ref{def.ADM.psi}, \eqref{r.varpi.limits} and Corollary~\ref{varpi.same.limit}\\
$\Phi(U), M(U), \Gmm(U)$ & \S \ref{sec:L.Li} & Limiting quantities at future null infinity in definition of $\mathfrak L$ \\
$\mathfrak{L}, \mathfrak{L}_{(\omg_{0})0}, \mathfrak{L}_{(\omg_{0})\infty}$ & \S \ref{sec:L.Li} & Quantities capturing lower bound of decay near asymptotically flat end, see Definition~\ref{def.L} \\
$(U, V)$ & \S \ref{subsec:bg-coords} & Initial-data-normalized double null coordinates \\
$(u, v)$ & \S \ref{subsec:bg-coords} & Future-normalized double null coordinates, also used as notation for a general double null coordinate system \\
$\underline{\ups}(U)$ & \S \ref{subsec:bg-coords} & $\underline{\ups}(U)$ is the value of $r$ at $\Sgm_{0} \cap C_{U}$. \\
$\psi$ & \S \ref{subsec:bg-large-r} & $\psi = r \phi$ \\
$\eta_{i^{0}}$, $\underline{\calN}$ & \S \ref{subsec:bg-large-r} & $\underline{\calN}$ is a region near spatial infinity in Proposition~\ref{prop:bg-large-r} defined by the parameter $\eta_{i^{0}}$\\
$\calX_{0}, u_{\calX_{0}}, V_{\calX_{0}}$ & \S \ref{subsec:DR-full} & $\calX_{0}$ is a region in a neighborhood of timelike infinity for the Dafermos--Rodnianski Price's law theorem (Theorem~\ref{thm:DR-full}). $(u_{\calX_{0}}, V_{\calX_{0}}) \in \calX_{0}$ is a point on the past boundary of $\calX_{0}$, usually normalized to $(1, 1)$. \\
$\eta_{\NI}$, $\calN$ & \S \ref{subsec:DR-full} & $\calN$ is a region near null infinity in Corollary~\ref{cor:DR-large-r} defined by the parameter $\eta_{\NI}$\\
$\dur, \dvr, \gmm, \kpp$ & \S \ref{subsec:bg-uv} & $\dur = \rd_{u} r$, $\dvr = \rd_{v} r$, $\gmm = \frac{\rd_{v} r}{1-\mu}$ and $\kpp = \frac{\rd_{v} r}{1-\mu}$ in the future-normalized coordinates \\
$\calX_{1}$ & \S \ref{subsec:bg-uv} & $\calX_{1}\subseteq \calX_{0}$ region in Proposition~\ref{prop:bg-uv} \\
$b_{1}, b_{1}', B_{1}, C_{30}$ & \S \ref{subsec:bg-uv} & Constants in Proposition~\ref{prop:bg-uv} \\
$A_{contra}$ & \S \ref{sec:blowup:ideas} & Constant for the contradiction argument in Section~\ref{sec:blowup}\\
$B_{contra}$, $B_{contra}'$, $C_{contra}$ & \S \ref{subsec:blowup-pf} and \ref{sec.contra.unif} & Large constants in contradiction argument \\
$\de_r, \Lambda_{\de_r}$ & \S \ref{subsec:blowup-pf} & Constants defining the red-shift region for contradiction argument\\
$R_{0, \NI}$, $\ub_{R_{\NI, 0}}, \Vb_{R_{\NI, 0}}$ & \S \ref{sec.contra.unif} & $R_{0, \NI}$ is a large parameter in Lemma~\ref{lem.uv-est}. $(\ub_{R_{\NI, 0}}, \Vb_{R_{\NI, 0}})$ is the point in $\Sgm_{0} \cap \set{r = R_{\NI, 0}}$. \\
$R_{\NI}$, $V_{\NI}$, $u_{\NI, 0}$, $u_{\NI}$ & \S \ref{sec.contra.unif} and \ref{subsec:blowup-const-r} & Various largeness parameters defining regions near null infinity for the contradiction argument\\
$\mathfrak L_{approx}$, $\de_{\mathfrak L}$ & \S \ref{subsec:blowup-const-r} & $\mathfrak L_{approx}$ is an approximation of $\mathfrak L_{(\omg_0)\infty}$, with $\de_{\mathfrak L}$ quantifying the difference\\
$\eps_{1}, \eps_{2}, \eps_{3}$ & \S \ref{subsec:cauchy-st}, \ref{subsec:extr-st} & Initial data difference for Cauchy and large-$r$ stability in Propositions~\ref{prop:cauchy-st}, \ref{prop:extr-st-cauchy}, \ref{prop:extr-st} \\
$\underline{\calR}, \calR$ & \S \ref{subsec:extr-st} & Regions for large-$r$ stability in Propositions~\ref{prop:extr-st-cauchy}, \ref{prop:extr-st} \\
$\underline{B}, B$ & \S \ref{subsec:extr-st} & Size of the background solution in Propositions~\ref{prop:extr-st-cauchy}, \ref{prop:extr-st} \\
$\dlt_{\underline{\calR}}, \dlt_{\calR}$ & \S \ref{subsec:extr-st} & Upper bounds for $\f{\varpibg_0}{\rbg_0}$, $\f{\ebg_0}{\rbg_0}$ in in Propositions~\ref{prop:extr-st-cauchy}, \ref{prop:extr-st}\\
$A_{0}, A_{1}$ & \S \ref{subsec:extr-st} & Large bootstrap constants, chosen in the order $A_{0} \to A_{1}$ \\
$\widetilde{\calD}, \calD$ & \S \ref{subsec:main-st} & $\widetilde{\calD}$ is the domain of an $(\omg, \dlt, \Lmb)$-admissible solution, and $\calD \subseteq \widetilde{\calD}$ is the black hole exterior \\
$\dlt, \dlt_{0}$ & \S \ref{subsec:main-st} & $\dlt_{0}\geq \Lmb^{100\eta_0}\dlt$ is a smallness parameter for Theorem~\ref{thm:L-st-ch} \\
$\eps_{0}, \eps, \eps_{(t_B)}$ & \S \ref{subsec:main-st}, \ref{subsec:weak-st}, \ref{subsec:btstrp} & Initial data differences under difference coordinate normalizations, see Theorem~\ref{thm:L-st-ch}, \eqref{eq:eps-def}, \eqref{eq:eps-tB} 
\end{tabular}
\end{center}

\vskip.5em
\begin{center}
\small
\begin{tabular}{c  c p{25em}}
{\bfseries Symbol} & {\bfseries Ref.} & {\bfseries Description}  \\
$\calR_{1}, \ldots, \calR_{4}$ & \S \ref{subsec:main-st}, \ref{subsec:inst-key} & Regions in the whole exterior region in the proofs of Theorems~\ref{thm:L-stability} and \ref{thm:instability} \\
$\Lmb, R_{0}$ & \S \ref{subsec:main-st}, \ref{subsec:weak-st} & $\Lmb=\rbg(1,1)$, $R_{0}=r(1,1)$ \\
$\Gmm_{\tau}$, $\Gmm_{\tau}^{(in)}$, $\Gmm_{\tau}^{(out)}$, $\calD(\tau_{1}, \tau_{2})$ & \S \ref{subsec:weak-st} & Foliations and region of spacetime in Theorem~\ref{thm:weak-st} \\
$\calD_{(t_{B})}$ & \S \ref{subsec:btstrp} & Bootstrap domain for the proof of the weak stability theorem (Theorem~\ref{thm:weak-st}) \\
$(u_{(t_{B})}, v_{(t_{B})})$ & \S \ref{subsec:btstrp} & Future-normalized double null coordinates during the bootstrap argument in \S \ref{subsec:btstrp}--\S {subsec:weak-st-pf}. We often abbreviate $(u, v) = (u_{(t_{B})}, v_{(t_{B})})$ \\
$\dlt_{(t_{B})}$ & \S \ref{subsec:btstrp} & Constant associated to $\calD_{(t_B)}$ in \eqref{eq:dlt-tB} \\
$A$ & \S \ref{subsec:btstrp} & A large bootstrap constant in the bootstrap argument in \S \ref{subsec:btstrp}--\S \ref{subsec:weak-st-pf} \\
$c_{(\durbg)}, c_{(\durbg)}', \overline{C}_{\gmm_{30}}$ & \S \ref{subsec:bg-geom} & Constants in Proposition~\ref{prop:bg-geom} \\
$c_{(\dur)}, C_{\gmm_{20}}$ & \S \ref{subsec:geom} & Constants in Lemma~\ref{lem:df-large-r} and Lemma~\ref{lem:df-small-r} \\
$\Gmm_{(t_{B}) \tau}$, $\Gmm_{(t_{B}) \tau}^{(in)}$, $\Gmm_{(t_{B}) \tau}^{(out)}$, $\calD_{(t_{B})} (\tau_{1}, \tau_{2})$ & \S \ref{subsec:en-pf} & Foliation in the bootstrap argument in \S \ref{subsec:en-pf} for the energy estimates. We often omit the subscript $(t_{B})$.\\
$R_1$, $R_2$ & \S \ref{subsec:en-pf}, \ref{subsec:rp-weight} & Large $r$ parameters for cutoff functions in energy estimates in Lemma~\ref{lem:en-ctrl} and \ref{lem:en-rp} respectively \\
$G$ & \S \ref{subsec:rp-weight} & Inhomogeneous term in the wave equation for $\psibg$\\
$R_3$ & \S \ref{subsec:int-char} & Large $r$ parameter for proof of decay of difference of radiation field in Proposition~\ref{prop:psidf-decay}\\
$\Sgm_{0}^{+}, \rho_{\ast}^{+}$ & \S \ref{subsec:const-eq} & $\Sgm_{0}^+\subset \Sgm_{0}$ on which we solve the constraint equations, $\rho=\rho_{\ast}^{+}$ denotes the end-point of $\Sgm_{0}^+$ \\
$r_{\ast}, \psi_{\ast}, \varpi_{\ast}, \e, \underline{\dvr}, \underline{\dur}, \underline{\alp}, \underline{\bt}$ & \S \ref{subsec:const-eq} & Reduced initial data set \\
$\chi, \tilde{\chi}, c_{0}$ & \S \ref{subsec:inst-key} & Functions and parameter used to define perturbation \\ 
$R_{pert}, \rho_{pert}, U_{pert}$ & \S \ref{subsec:inst-key} & Parameters for the support of the perturbation, which has $\rho$ value $\approx \rho_{pert}$, $r$ value $\approx R_{pert}$ and $U$ value $\approx U_{pert}$\\
$\tilde{\chi}_{R_{pert}}$ & \S \ref{subsec:inst-key} & Translated version of $\tilde{\chi}$
\end{tabular}
\end{center}

\bibliographystyle{hplain}
\bibliography{SCCExterior}
\end{document}